\documentclass[11pt]{article}

\usepackage{rotating}
\usepackage{titling}

\usepackage{enumitem} % Provides the ability to easily modify the amount of white space in a list
\setlist{noitemsep,topsep=0pt,parsep=0pt} % Modify the amount of white space for all lists

\usepackage[margin=1cm]{caption} % Adds an addition margin on either side of figure captions. Must come before "\usepackage{subfig}" (otherwise there are errors).
\usepackage{subfig}
\usepackage{mathtools} % Adds cases* and dcases envirnoments

\usepackage{tikz}
\usetikzlibrary{arrows,backgrounds,calc,fit,decorations.pathreplacing,decorations.markings,shapes.geometric}

\tikzset{every fit/.append style=text badly centered}

\usepackage{tyson}

\usepackage[bookmarks=true,hypertexnames=false]{hyperref}
\hypersetup{colorlinks=true, citecolor=blue, linkcolor=red, urlcolor=blue}

\newcommand{\Holant}{\operatorname{Holant}}
\newcommand{\PlHolant}{\operatorname{Pl-Holant}}
\newcommand{\holant}[2]{\ensuremath{\Holant\left(#1\mid #2\right)}}
\newcommand{\plholant}[2]{\ensuremath{\PlHolant\left(#1\mid #2\right)}}
\newcommand{\arity}{\operatorname{arity}}
\newcommand{\CSP}{\operatorname{\#CSP}}
\newcommand{\PlCSP}{\operatorname{Pl-\#CSP}}
\newcommand{\vd}{\operatorname{vd}}
\newcommand{\rd}{\operatorname{rd}}
\newcommand{\Sym}{\operatorname{Sym}}

\newcommand{\EQ}{\mathcal{EQ}}
\newcommand{\Stab}[1]{\operatorname{Stab}(\mathscr{#1})}

\newcommand{\tbcolvec}[2]{\ensuremath{\left[\begin{smallmatrix} #1 \\ #2 \end{smallmatrix}\right]}}

\newcommand{\eblock}[1]{\ensuremath{E_{#1}}-block}
\newcommand{\eoeqblock}{\ensuremath{EO}-\ensuremath{Eq}-\ensuremath{4}-block}
\newcommand{\exactone}[1]{\ensuremath{\text{\sc ExactOne}_{#1}}}
\newcommand{\allbutone}[1]{\ensuremath{\text{\sc AllButOne}_{#1}}}

\newcommand{\nondeg}[1]{\ensuremath{\mathcal{#1}_{nd}^{\ge 3}}}

\newcommand{\numP}{{\rm \#P}}
\newcommand{\trans}[4]{\ensuremath{\left[\begin{smallmatrix} #1 & #2 \\ #3 & #4 \end{smallmatrix}\right]}}

\renewcommand{\CC}{\ensuremath{\mathbb{C}}}

\theoremstyle{remark}
\newtheorem{remark}{Remark}

\def\borderColor{blue!60}
\def\scale{0.6}
\def\nodeDist{1.4cm}
\tikzstyle{internal} = [draw, fill, shape=circle]
\tikzstyle{external} = [shape=circle]
\tikzstyle{square}   = [draw, fill, rectangle]
\tikzstyle{triangle} = [draw, fill, regular polygon, regular polygon sides=3, inner sep=3pt]
\tikzstyle{pentagon} = [draw, fill, regular polygon, regular polygon sides=5, inner sep=2pt, minimum size=14pt]

\title{A Holant Dichotomy: Is the FKT Algorithm Universal?}
\author{
 Jin-Yi Cai%
 \thanks{University of Wisconsin--Madison. Supported by NSF CCF-1217549.}\\
 \footnotesize \texttt{jyc@cs.wisc.edu}
 \and
 Zhiguo Fu\footnotemark[1]\
 \thanks{School of Mathematics, Jilin University}\\
 \footnotesize \texttt{zfu8@wisc.edu}
 \and
 Heng Guo\footnotemark[1]\
 \thanks{Also supported by a Simons Award for Graduate Students in Theoretical Computer Science from the Simons Foundation.}\\
 \footnotesize \texttt{hguo@cs.wisc.edu}
 \and
 Tyson Williams\footnotemark[1]\
 \thanks{Also supported by a Cisco Systems Distinguished Graduate Fellowship.}\\
 \footnotesize \texttt{tdw@cs.wisc.edu}
}
\date{} % no date

\begin{document}
\maketitle

\begin{abstract}
 We prove a complexity dichotomy for complex-weighted Holant problems
 with an arbitrary set of symmetric constraint functions on Boolean variables.
 This dichotomy is specifically to answer the question:
 Is the FKT algorithm under a holographic transformation~\cite{Val08} a \emph{universal} strategy to obtain
 polynomial-time algorithms for problems over planar graphs that are intractable in general?
 This dichotomy is a culmination of previous ones,
 including those for Spin Systems~\cite{Kow10},
 Holant~\cite{HL12,CGW13},
 and \#CSP~\cite{GW13}.

 In the study of counting complexity,
 such as \#CSP,
 there are problems which are \#P-hard over general graphs but polynomial-time solvable over planar graphs.
 A recurring theme has been that a holographic reduction to FKT precisely captures these problems.
 Surprisingly,
 for planar Holant,
 we discover new planar tractable problems that are not expressible by a holographic reduction to FKT.
 In particular,
 a straightforward formulation of a dichotomy for planar Holant problems along the above recurring theme is false.

 In previous work,
 an important tool was a dichotomy for \#CSP$^d$,
 which denotes \#CSP where every variable appears a multiple of $d$ times.
 However the very first step in the \#CSP$^d$ dichotomy proof fundamentally violates planarity.
 In fact,
 due to our newly discovered tractable problems,
 the putative form of a planar \#CSP$^d$ dichotomy is false when $d \ge 5$.
 Nevertheless,
 we prove a dichotomy for planar \#CSP$^2$.
 In this case,
 the putative form of the dichotomy is true.
 We manage to prove the planar Holant dichotomy without relying on
 a planar \#CSP$^d$ dichotomy for $d \ge 3$,
 while the dichotomy for planar \#CSP$^2$ plays an essential role.

 As a special case of our new planar tractable problems,
 counting perfect matchings (\#PM) over $k$-uniform hypergraphs is polynomial-time computable
 when the incidence graph is planar and $k \ge 5$.
 The same problem is \#P-hard when $k=3$ or $k=4$,
 which is also a consequence of our dichotomy.
 When $k=2$,
 it becomes \#PM over planar graphs and is tractable again.
 More generally,
 over hypergraphs with specified hyperedge sizes and the same planarity assumption,
 \#PM is polynomial-time computable if the greatest common divisor (gcd) of all hyperedge sizes is at least~$5$.
 It is worth noting that it is the gcd,
 and not a bound on hyperedge sizes,
 that is the criterion for tractability.
\end{abstract}

\section{Introduction}

The Fisher-Kasteleyn-Temperley (FKT) algorithm~\cite{TF61, Kas61, Kas67}
is a classical gem that counts perfect matchings over planar graphs in polynomial time.
This was an important milestone in a decades-long research program by physicists in statistical mechanics to determine what is known as
Exactly Solved Models~\cite{Bax82, Isi25, Ons44, Yan52, YL52, LY52, TF61, Kas61, Kas67, Lie67, LS81, Wel93}.

For four decades,
the FKT algorithm stood as \emph{the} polynomial-time algorithm for any counting problem over planar graphs that is \#P-hard over general graphs.
Then Valiant introduced \emph{matchgates}~\cite{Val02a, Val02b} and \emph{holographic} reductions to the FKT algorithm~\cite{Val08, Val06}.
These reductions differ from classical ones by introducing quantum-like superpositions.
This novel technique extended the reach of the FKT algorithm
and produced polynomial-time algorithms for a number of problems for which only exponential-time algorithms were previously known.

Since the new polynomial-time algorithms appear so exotic and unexpected,
and since they solve problems that appear so close to being \#P-hard,
they challenge our faith in the well-accepted conjecture that P $\neq$ NP.
Quoting Valiant~\cite{Val06}:
``The objects enumerated are sets of polynomial systems
such that the solvability of any one member would give
a polynomial time algorithm for a specific problem.
$\ldots$
the situation with the P = NP question is not dissimilar to that of
other unresolved enumerative conjectures in mathematics.
The possibility that accidental or freak objects in the enumeration exist cannot be discounted
if the objects in the enumeration have not been studied systematically.''
Indeed,
if any ``freak'' object exists in this framework,
it would collapse \#P to P.

Therefore,
over the past 10 to 15 years,
this technique has been intensely studied in order to gain a systematic understanding
to the limit of the trio of holographic reductions, matchgates, and the FKT algorithm~\cite{Val02b, CC07b, CCL09, CL10, Val10, CL11a, LMN13, Mor11, MM13}.
Without settling the P versus \#P question,
the best hope is to achieve a complexity classification.
This program finds its sharpest expression in a complexity dichotomy theorem,
which classifies \emph{every} problem expressible in a framework as either solvable in P or \#P-hard,
with nothing in between.

Out of this work,
a strong theme has emerged.
For a wide variety of problems,
such as those expressible as a \#CSP,
holographic reductions to the FKT algorithm is a \emph{universal} technique
for turning problems that are \#P-hard in general to P-time solvable over planar graphs.
In fact,
a preponderance of evidence suggests the following putative classification
of all counting problems defined by local constraints into \emph{exactly} three categories:
(1) those that are P-time solvable over general graphs;
(2) those that are P-time solvable over planar graphs but \#P-hard over general graphs; and
(3) those that remain \#P-hard over planar graphs.
Moreover,
category (2) consists precisely of those problems that are holographically reducible to the FKT algorithm.
This theme is so strong that it has become an intuitive and trusty guide for us when we investigate unknown problems and plan proof strategies.
In fact,
many of the results in the present paper were proved in this way.
However,
one is still left wondering whether the FKT algorithm is \emph{universal},
or more precisely,
is the combined algorithmic power of the trio sufficient to capture all tractable problems over planar graphs that are intractable in general?

We list some of the supporting evidence for this putative classification.
These date back to the classification of the complexity of the Tutte polynomial~\cite{Ver91, Ver05}.
It has also been an unfailing theme in the classification of spin systems and \#CSP~\cite{Kow10, CLX10, CKW12, GW13}.
However,
these frameworks do not capture all locally specified counting problems.
Some natural problems,
such as counting perfect matchings (\#PM),
are not expressible as a point on the Tutte polynomial or a $\CSP$,
and \#PM is provably not expressible within the special case of vertex assignment
models~\cite{FLS07, DGLRS12, Sch13}.
However, this is the problem for which FKT was designed,
and is the basis of Valiant's matchgates and holographic reductions.

A refined framework,
called Holant problems~\cite{CLX11d},
was proposed to address this issue.
It is an edge assignment model.
It naturally encodes and expresses \#PM as well as Valiant's matchgates and holographic reductions.
Thus,
Holant is the proper framework in which to study the power of holographic algorithms.
It is also more general than $\CSP$ in the sense that
a complete complexity classification for Holant problems implies one for $\CSP$.

In this paper,
we classify for the first time the complexity of Holant problems over planar graphs.
Our result generalizes both the dichotomy for Holant~\cite{HL12, CGW13} and the dichotomy for planar $\CSP$~\cite{CLX10, GW13}.
Surprisingly,
we discover new planar tractable problems that are not expressible by a holographic reduction to matchgates and FKT.
To the best of our knowledge,
this is the first \emph{primitive} extension since FKT to a problem solvable in P over planar instances but \#P-hard in general.
Furthermore,
our dichotomy theorem says that this completes the picture:
there are no more undiscovered extensions for problems expressible in this framework,
unless \#P collapses to P.
In particular,
the putative form of the planar Holant dichotomy is \emph{false}.

Before stating our main theorem,
we give a brief description of the Holant framework~\cite{CLX11d}.
Fix a set of local constraint functions $\mathcal{F}$.
A \emph{signature grid} $\Omega = (G, \pi)$ is a tuple,
where $G = (V,E)$ is a graph,
$\pi$ labels each $v \in V$ with a function $f_v \in \mathcal{F}$
with input variables from the incident edges $E(v)$ at $v$.
Each $f_v$ maps $\{0,1\}^{\deg(v)}$ to $\mathbb{C}$.
We consider all 0-1 edge assignments.
An assignment $\sigma$ for every $e \in E$ gives an evaluation $\prod_{v \in V} f_v(\sigma \mid_{E(v)})$,
where $\sigma \mid_{E(v)}$ denotes the restriction of $\sigma$ to $E(v)$.
The counting problem on the instance $\Omega$ is to compute
\begin{equation} \label{eqn:holant-sum}
 \Holant(\Omega; \mathcal{F}) = \sum_{\sigma : E \to \{0,1\}} \prod_{v \in V} f_v\left(\sigma \mid_{E(v)}\right).
\end{equation}
For example,
\#PM,
the problem of counting perfect matchings in $G$,
corresponds to assigning the \textsc{ExactOne} function at every vertex of $G$.
The Holant problem parameterized by the set $\mathcal{F}$ is denoted by $\Holant(\mathcal{F})$.

At a high level,
we can state our main theorem as follows.
\begin{theorem} \label{thn:main:intro}
 Let $\mathcal{F}$ be a set of complex-valued, symmetric functions on Boolean variables.
 Then there is an effective classification for all possible $\mathcal{F}$,
 according to which,
 $\Holant(\mathcal{F})$ is either
 (1) P-time computable over general graphs, or
 (2) P-time computable over planar graphs but \numP-hard over general graphs, or
 (3) \numP-hard over planar graphs.
\end{theorem}

The complete statement is given in Theorem~\ref{thm:main}.
The classification is explicit.
The tractability criterion is decidable in polynomial time due to~\cite{CL11a, CGW14a}.
Tractable problems over general graphs have been previously studied in~\cite{CGW13}.
The planar tractable class includes both those solvable by holographic reductions to FKT and those newly discovered.
Explicit criteria for these are also proved in this paper.

Let us meet some new tractable problems.
They can be described as orientation problems,
which are Holant problems after a complex-valued holographic transformation.%
\footnote{This transformation is $Z = \trans{1}{1}{i}{-i}$.
It is common that one problem can be transformed to another over $\mathbb{C}$
while one or both problems are specified by
% integer-valued
%%% the above is true, but not for our new problem
%%% notes to ourselves: P_2 = Z (EQ) includes [1\\i]^n + [1\\-i]^n
%%% but M_4 has Z(PM), not real valued.
real-valued constraint functions,
and provably no transformation exists over $\mathbb{R}$.
Thus to study the classification question over complex-valued constraint functions is natural and proper.
For example,
the integer-valued orientation problem studied here is complex weighted if expressed directly as Holant.}
Given a planar graph,
we allow two kinds of vertices.
The first kind can be either a sink or a source while the second kind only allow one incoming edge.
The goal is to compute the number of orientations satisfying these constraints.
This problem can be expressed in the Holant framework under a $Z$-transformation.
It can be shown that this is equivalent to the Holant problem on the edge-vertex incidence graph
where we assign the \textsc{Disequality} function to every edge,
and to each vertex,
we assign either the \textsc{Equality} function or the \textsc{ExactOne} function.
Suppose vertices assigned \textsc{Equality} functions all have degree $k$.
If $k=2$,
then this problem can be solved by FKT.
We show that this problem is \numP-hard if $k=3$ or $k=4$,
but is tractable again if $k \ge 5$.
The algorithm involves a recursive procedure that simplifies the instance until it can be solved by known algorithms,
including FKT.
The algorithm crucially uses global topological properties of a planar graph,
in particular Euler's characteristic formula.
If the graph is not planar,
then this algorithm does not work,
and indeed the problem is \numP-hard over general graphs.

More generally,
we allow vertices of arbitrary degrees to be assigned \textsc{Equality}.
If all the degrees are at most~$2$,
then the problem is tractable by the FKT algorithm.
Otherwise,
the complexity depends on the greatest common divisor (gcd) of the degrees.
The problem is tractable if $\gcd \ge 5$ and \numP-hard if $\gcd \le 4$.
It is worth noting that the criterion for tractability is not a degree lower bound.
Moreover,
the planarity assumption and the degree rigidity pose a formidable challenge in the hardness proofs for $\gcd \le 4$.

If the graph is bipartite with \textsc{Equality} functions assigned on one side
and \textsc{ExactOne} functions on the other,
then this is the problem of \#PM over hypergraphs with planar incidence graphs.
Our results imply that the complexity of this problem depends on the gcd of the hyperedge sizes.
The problem is computable in polynomial time when $\gcd \ge 5$ and is \numP-hard when $\gcd \le 4$
(assuming there are hyperedges of size at least~$3$).
For a formal statement,
see Theorem \ref{thm:hypergraph}.

Most of the reductions in previous Holant dichotomy theorems~\cite{HL12, CGW13} do not hold for planar graphs,
so we are forced to develop new techniques and formulate new proof strategies.
In particular,
an important ingredient in previous proofs is the $\CSP^d$ dichotomy by Huang and Lu~\cite{HL12}.
Here $\CSP^d$ denotes \#CSP where every variable appears a multiple of $d$ times.
The very first step in the \#CSP$^d$ dichotomy proof uses the popular pinning technique.
Multiple copies of an instance are created and vertices are connected across different copies.
But this construction fundamentally violates planarity.
Moreover,
this violation of planarity is unavoidable,
a consequence of the new dichotomy.
Due to our newly discovered tractable problems,
the putative form of a planar \#CSP$^d$ dichotomy is \emph{false} when $d \ge 5$.
Nevertheless,
we prove a dichotomy for planar \#CSP$^2$ for which the putative form is,
luckily for us,
true
(but not obvious in hindsight).
Obtaining a dichotomy for planar \#CSP$^2$ is essential because
it captures a significant fraction of planar Holant problems either directly or through reductions.
We manage to prove the planar Holant dichotomy without appealing to planar $\CSP^d$ for $d \ge 3$.

The proof of the planar \#CSP$^2$ dichotomy comprises the entire Part~II of this paper that starts on page~\pageref{partII:sec:prelim}.
A brief outline of the proof is given in Section~\ref{sec:PlCSP2} of Part~I.
Among the concepts and techniques introduced are some special tractable families of constraint functions specific to the \#CSP$^2$ framework.
We also introduce a \emph{derivative} $\partial$ and its inverse operator \emph{integral} $\int$ to streamline the proof argument.
There is also an application of the theory of \emph{cyclotomic fields}.

We began this project expecting to prove the putative form of the planar Holant dichotomy.
It was determined that a planar \#CSP$^d$ dichotomy in the putative form would be both a more modest,
and thus hopefully more attainable,
intermediate step as well as a good launch station for the final goal.
However after some attempt,
even the planar \#CSP$^d$ dichotomy appeared too difficult to achieve,
and so we scaled back the ambition to prove just a planar \#CSP$^2$ dichotomy.
Luckily,
a successful \#CSP$^2$ dichotomy can carry most of the weight of a full \#CSP$^d$ dichotomy,
\emph{and},
as it turned out,
the putative form of the planar \#CSP$^2$ dichotomy is \emph{true} while that for planar \#CSP$^d$ is not.
Ironically,
many steps of our proof in this paper were guided by the putative form of the complexity classification.
The discovery of the new tractable problems changed the original plan,
but also helped complete the picture.

Coming back to the challenge of the P vs.~NP question posed by Valiant's holographic algorithms,
we venture the opinion that the dichotomy theorem provides a satisfactory answer.
Indeed,
it would be difficult to conceive a world where \#P is P,
and yet all this algebraic theory can somehow maintain a consistent, sharp but faux division where there is none.

\section{Preliminaries} \label{sec:preliminaries}

\subsection{Problems and Definitions}

The framework of Holant problems is defined for functions mapping any $[q]^n \to R$ for a finite $q$ and some commutative semiring $R$.
In this paper,
we investigate complex-weighted Boolean $\Holant$ problems,
that is,
all functions are of the form $[2]^n \to \mathbb{C}$.
For consideration of models of computation,
functions take complex algebraic numbers.

Graphs may have self-loops and parallel edges.
A graph without self-loops or parallel edges is a \emph{simple} graph.
Fix a set of local constraint functions $\mathcal{F}$.
A \emph{signature grid} $\Omega = (G, \pi)$ consists of a graph $G = (V,E)$,
where $\pi$ assigns to each vertex $v \in V$ and its incident edges some $f_v \in \mathcal{F}$ and its input variables.
We say that $\Omega$ is a \emph{planar signature grid} if $G$ is planar,
where the variables of $f_v$ are ordered counterclockwise starting from an edge specified by $\pi$.
The Holant problem on instance $\Omega$ is to evaluate $\Holant(\Omega; \mathcal{F}) = \sum_{\sigma} \prod_{v \in V} f_v(\sigma \mid_{E(v)})$,
a sum over all edge assignments $\sigma: E \to \{0,1\}$,
where $E(v)$ denotes the incident edges of $v$ and $\sigma \mid_{E(v)}$ denotes the restriction of $\sigma$ to $E(v)$.
We write $G$ in place of $\Omega$ when $\pi$ is clear from context.

A function $f_v$ can be represented by listing its values in lexicographical order as in a truth table,
which is a vector in $\mathbb{C}^{2^{\deg(v)}}$,
or as a tensor in $(\mathbb{C}^{2})^{\otimes \deg(v)}$.
A function $f \in \mathcal{F}$ is also called a \emph{signature}.
A symmetric signature $f$ on $n$ Boolean variables can be expressed as $[f_0,f_1,\dotsc,f_n]$,
where $f_w$ is the value of $f$ on inputs of Hamming weight $w$.
In this paper,
we consider symmetric signatures.
An example is the \textsc{Equality} signature $=_n$ of arity $n$.

A Holant problem is parametrized by a set of signatures.

\begin{definition}
 Given a set of signatures $\mathcal{F}$,
 we define the counting problem $\Holant(\mathcal{F})$ as:

 Input: A \emph{signature grid} $\Omega = (G, \pi)$;

 Output: $\Holant(\Omega; \mathcal{F})$.
\end{definition}

\noindent
The problem $\PlHolant(\mathcal{F})$ is defined similarly using a planar signature grid.

A signature $f$ of arity $n$ is \emph{degenerate} if there exist unary signatures $u_j \in \mathbb{C}^2$ ($1 \le j \le n$)
such that $f = u_1 \otimes \cdots \otimes u_n$.
A symmetric degenerate signature has the form $u^{\otimes n}$.
Replacing such signatures by $n$ copies of the corresponding unary signature does not change the Holant value.
Replacing a signature $f \in \mathcal{F}$ by a constant multiple $c f$,
where $c \ne 0$,
does not change the complexity of $\Holant(\mathcal{F})$.
In this paper, we may say we obtain a signature $f$
when in fact we have obtained a signature $cf$ for some $c \ne 0$.
It introduces a global nonzero factor to $\Holant(\Omega; \mathcal{F})$.

We allow $\mathcal{F}$ to be an infinite set.
For $\PlHolant(\mathcal{F})$ to be tractable,
the problem must be computable in polynomial time even when the description of the signatures in the input $\Omega$ are included in the input size.
In contrast,
we say $\PlHolant(\mathcal{F})$ is $\SHARPP$-hard if there exists a finite subset of $\mathcal{F}$ for which the problem is $\SHARPP$-hard.
We say a signature set $\mathcal{F}$ is tractable (resp.~$\SHARPP$-hard)
if the corresponding counting problem $\PlHolant(\mathcal{F})$ is tractable (resp.~$\SHARPP$-hard).
Similarly for a signature $f$,
we say $f$ is tractable (resp.~$\SHARPP$-hard) if $\{f\}$ is.
We follow the usual conventions about polynomial time Turing reduction $\le_T$ and polynomial time Turing equivalence $\equiv_T$.

\subsection{Holographic Reduction}

To introduce the idea of holographic reductions,
it is convenient to consider bipartite graphs.
For a general graph,
we can always transform it into a bipartite graph while preserving the Holant value,
as follows.
For each edge in the graph,
we replace it by a path of length two.
(This operation is called the \emph{2-stretch} of the graph and yields the edge-vertex incidence graph.)
Each new vertex is assigned the binary \textsc{Equality} signature $(=_2) = [1,0,1]$.

We use $\holant{\mathcal{F}}{\mathcal{G}}$ to denote the Holant problem over signature grids with a bipartite graph $H = (U,V,E)$,
where each vertex in $U$ or $V$ is assigned a signature in $\mathcal{F}$ or $\mathcal{G}$,
respectively.
Signatures in $\mathcal{F}$ are considered as row vectors (or covariant tensors);
signatures in $\mathcal{G}$ are considered as column vectors (or contravariant tensors)~\cite{DP91}.
Similarly,
$\plholant{\mathcal{F}}{\mathcal{G}}$ denotes the Holant problem over signature grids with a planar bipartite graph.

For a $2$-by-$2$ matrix $T$ and a signature set $\mathcal{F}$,
define $T \mathcal{F} = \{g \mid \exists f \in \mathcal{F}$ of arity $n,~g = T^{\otimes n} f\}$,
and similarly for $\mathcal{F} T$.
Whenever we write $T^{\otimes n} f$ or $T \mathcal{F}$,
we view the signatures as column vectors;
similarly for $f T^{\otimes n} $ or $\mathcal{F} T$ as row vectors.
In the special case that $T = \left[\begin{smallmatrix} 1 & 1 \\ 1 & -1 \end{smallmatrix}\right]$,
we also define $T \mathcal{F} = \widehat{\mathcal{F}}$.

Let $T$ be an invertible $2$-by-$2$ matrix.
The holographic transformation defined by $T$ is the following operation:
given a signature grid $\Omega = (H, \pi)$ of $\holant{\mathcal{F}}{\mathcal{G}}$,
for the same bipartite graph $H$,
we get a new grid $\Omega' = (H, \pi')$ of $\holant{\mathcal{F} T}{T^{-1} \mathcal{G}}$ by replacing each signature in
$\mathcal{F}$ or $\mathcal{G}$ with the corresponding signature in $\mathcal{F} T$ or $T^{-1} \mathcal{G}$.

\begin{theorem}[Valiant's Holant Theorem~\cite{Val08}]
 If $T \in \mathbb{C}^{2 \times 2}$ is an invertible matrix,
 then we have $\Holant(\Omega; \mathcal{F} \mid \mathcal{G}) = \Holant(\Omega'; \mathcal{F} T \mid T^{-1} \mathcal{G})$.
\end{theorem}

Therefore,
an invertible holographic transformation does not change the complexity of the Holant problem in the bipartite setting.
Furthermore,
there is a special kind of holographic transformation,
the orthogonal transformation,
that preserves the binary equality and thus can be used freely in the standard setting.

\begin{theorem}[Theorem~2.6 in~\cite{CLX11d}] \label{thm:orthogonal}
 If $T \in \mathbf{O}_2(\mathbb{C})$ is an orthogonal matrix (i.e.~$T \transpose{T} = I_2$),
 then $\Holant(\Omega; \mathcal{F}) = \Holant(\Omega'; T \mathcal{F})$.
\end{theorem}

We frequently apply a holographic transformation defined by the matrix $Z = \frac{1}{\sqrt{2}} \trans{1}{1}{i}{-i}$
(or sometimes without the nonzero factor of $\frac{1}{\sqrt{2}}$ since this does not affect the complexity).
This matrix has the property that the binary \textsc{Equality} signature $(=_2) = [1,0,1]$
is transformed to $[1,0,1] Z^{\otimes 2} = [0,1,0] = (\neq_2)$,
the binary \textsc{Disequality} signature.

An important definition involving a holographic transformation is the notion of a signature set being transformable.

\begin{definition} \label{def:prelim:trans}
 We say a signature set $\mathcal{F}$ is $\mathscr{C}$-transformable
 if there exists a $T \in \mathbf{GL}_2(\mathbb{C})$ such that
 $[1,0,1] T^{\otimes 2} \in \mathscr{C}$ and $\mathcal{F} \subseteq T \mathscr{C}$.
\end{definition}

This definition is important because if $\PlHolant(\mathscr{C})$ is tractable,
then $\PlHolant(\mathcal{F})$ is tractable for any $\mathscr{C}$-transformable set $\mathcal{F}$.

\subsection{Counting Constraint Satisfaction Problems}

We can define the framework of counting constraint satisfaction problems (\#CSP) in terms of the Holant framework.
An instance of $\CSP(\mathcal{F})$ has the following bipartite view.
Create a vertex for each variable and each constraint.
Connect a variable vertex to a constraint vertex if the variable appears in the constraint.
This bipartite graph is also known as the \emph{constraint graph}.
Moreover, each variable can be viewed as an \textsc{Equality} function,
as it takes two values.
Under this view,
we see that $\CSP(\mathcal{F}) \equiv_T \holant{\mathcal{EQ}}{\mathcal{F}}$,
where $\mathcal{EQ} = \{{=}_1, {=}_2, {=}_3, \dotsc\}$ is the set of \textsc{Equality} signatures of all arities.
By restricting to planar constraint graphs,
we have the planar \#CSP framework,
which we denote by $\PlCSP$.
The construction above also shows that $\PlCSP(\mathcal{F}) \equiv_T \plholant{\mathcal{EQ}}{\mathcal{F}}$.

For any positive integer $d$,
the problem $\CSP^d(\mathcal{F})$ is the same as $\CSP(\mathcal{F})$ except that every variable appears a multiple of $d$ times.
Thus,
$\PlCSP^d(\mathcal{F}) \equiv_T \plholant{\mathcal{EQ}_d}{\mathcal{F}}$,
where $\mathcal{EQ}_d = \{{=}_d, {=}_{2 d}, {=}_{3 d}, \dotsc\}$ is the set of \textsc{Equality} signatures of arities that are a multiple of $d$.
If $d \in \{1,2\}$,
then we further have
\begin{equation} \label{eqn:prelim:PlCSPd_equiv_Holant}
 \PlCSP^d(\mathcal{F}) \equiv_T \plholant{\mathcal{EQ}_d}{\mathcal{F}} \equiv_T \PlHolant(\mathcal{EQ}_d \cup \mathcal{F}).
\end{equation}
The reduction from left to right in the second equivalence is trivial.
For the other direction,
we take a signature grid for the problem on the right and create a bipartite signature grid for the problem on the left
such that both signature grids have the same Holant value up to an easily computable factor.
If two signatures in $\mathcal{F}$ are assigned to adjacent vertices,
then we subdivide all edges between them and assign the binary \textsc{Equality} signature ${=}_2 \in \mathcal{EQ}_d$ to all new vertices.
Suppose \textsc{Equality} signatures ${=}_n, {=}_m \in \mathcal{EQ}_d$ are assigned to adjacent vertices connected by $k$ edges.
If $n = m = k$,
then we simply remove these two vertices.
The Holant of the resulting signature grid differs from the original by a factor of~$2$.
Otherwise,
we contract all $k$ edges and assign ${=}_{n+m-2k} \in \mathcal{EQ}_d$ to the new vertex.

\subsection{Realization}

One basic notion used throughout the paper is realization.
We say a signature $f$ is \emph{realizable} or \emph{constructible} from a signature set $\mathcal{F}$
if there is a gadget with some dangling edges such that each vertex is assigned a signature from $\mathcal{F}$,
and the resulting graph,
when viewed as a black-box signature with inputs on the dangling edges,
is exactly $f$.
If $f$ is realizable from a set $\mathcal{F}$,
then we can freely add $f$ into $\mathcal{F}$ while preserving the complexity.

\begin{figure}[t]
 \centering
 \begin{tikzpicture}[scale=\scale,transform shape,node distance=\nodeDist,semithick]
  \node[external]  (0)                     {};
  \node[internal]  (1) [below right of=0]  {};
  \node[external]  (2) [below left  of=1]  {};
  \node[internal]  (3) [above       of=1]  {};
  \node[internal]  (4) [right       of=3]  {};
  \node[internal]  (5) [below       of=4]  {};
  \node[internal]  (6) [below right of=5]  {};
  \node[internal]  (7) [right       of=6]  {};
  \node[internal]  (8) [below       of=6]  {};
  \node[internal]  (9) [below       of=8]  {};
  \node[internal] (10) [right       of=9]  {};
  \node[internal] (11) [above right of=6]  {};
  \node[internal] (12) [below left  of=8]  {};
  \node[internal] (13) [left        of=8]  {};
  \node[internal] (14) [below left  of=13] {};
  \node[external] (15) [left        of=14] {};
  \node[internal] (16) [below left  of=5]  {};
  \path let
         \p1 = (15),
         \p2 = (0)
        in
         node[external] (17) at (\x1, \y2) {};
  \path let
         \p1 = (15),
         \p2 = (2)
        in
         node[external] (18) at (\x1, \y2) {};
  \node[external] (19) [right of=7]  {};
  \node[external] (20) [right of=10] {};
  \path (1) edge                             (5)
            edge[bend left]                 (11)
            edge[bend right]                (13)
            edge node[near start] (e1) {}   (17)
            edge node[near start] (e2) {}   (18)
        (3) edge                             (4)
        (4) edge[out=-45,in=45]              (8)
        (5) edge[bend right, looseness=0.5] (13)
            edge[bend right, looseness=0.5]  (6)
        (6) edge[bend left]                  (8)
            edge[bend left]                  (7)
            edge[bend left]                 (14)
        (7) edge node[near start] (e3) {}   (19)
       (10) edge[bend right, looseness=0.5] (12)
            edge[bend left,  looseness=0.5] (11)
            edge node[near start] (e4) {}   (20)
       (12) edge[bend left]                 (16)
       (14) edge node[near start] (e5) {}   (15)
            edge[bend right]                (12)
       (16) edge[bend left,  looseness=0.5]  (9)
            edge[bend right, looseness=0.5]  (3);
  \begin{pgfonlayer}{background}
   \node[draw=\borderColor,thick,rounded corners,fit = (3) (4) (9) (e1) (e2) (e3) (e4) (e5)] {};
%    \node[draw=\borderColor,thick,rounded corners,fit = (3) (4) (9) (e1) (e2) (e3) (e4) (e5),transform shape=false] {};
  \end{pgfonlayer}
 \end{tikzpicture}
 \caption{An $\mathcal{F}$-gate with 5 dangling edges.}
 \label{fig:Fgate}
\end{figure}
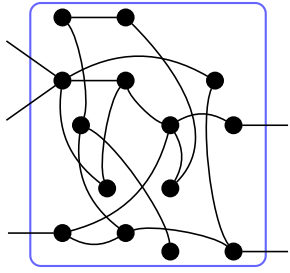

Formally,
such a notion is defined by an $\mathcal{F}$-gate~\cite{CLX10}.
An $\mathcal{F}$-gate is similar to a signature grid $(G, \pi)$ for $\Holant(\mathcal{F})$ except that $G = (V,E,D)$ is a graph with some dangling edges $D$.
The dangling edges define external variables for the $\mathcal{F}$-gate.
(See Figure~\ref{fig:Fgate} for an example.)
We denote the regular edges in $E$ by $1, 2, \dotsc, m$ and the dangling edges in $D$ by $m+1, \dotsc, m+n$.
Then we can define a function $\Gamma$ for this $\mathcal{F}$-gate as
\[
 \Gamma(y_1, \dotsc, y_n) = \sum_{x_1, \dotsc, x_m \in \{0, 1\}} H(x_1, \dotsc, x_m, y_1, \dotsc, y_n),
\]
where $(y_1, \dotsc, y_n) \in \{0, 1\}^n$ is an assignment on the dangling edges
and $H(x_1, \dotsc, x_m, y_1, \dotsc, y_n)$ is the value of the signature grid on an assignment of all edges in $G$,
which is the product of evaluations at all internal vertices.
We also call this function $\Gamma$ the signature of the $\mathcal{F}$-gate.

An $\mathcal{F}$-gate is planar if the underlying graph $G$ is a planar graph,
and the dangling edges,
ordered counterclockwise corresponding to the order of the input variables,
are in the outer face in a planar embedding.
A planar $\mathcal{F}$-gate can be used in a planar signature grid as if it is just a single vertex with the particular signature.

Using the idea of planar $\mathcal{F}$-gates,
we can reduce one planar Holant problem to another.
Suppose $g$ is the signature of some planar $\mathcal{F}$-gate.
Then $\PlHolant(\mathcal{F} \cup \{g\}) \leq_T \PlHolant(\mathcal{F})$.
The reduction is simple.
Given an instance of $\PlHolant(\mathcal{F} \cup \{g\})$,
by replacing every appearance of $g$ by the $\mathcal{F}$-gate,
we get an instance of $\PlHolant(\mathcal{F})$.
Since the signature of the $\mathcal{F}$-gate is $g$,
the Holant values for these two signature grids are identical.

Although our main result is about symmetric signatures,
some of our proofs utilize asymmetric signatures.
When a gadget has an asymmetric signature,
we place a diamond on the edge corresponding to the first input.
The remaining inputs are ordered counterclockwise around the vertex.
(See Figure~\ref{fig:two_gadgets_same_sig} for two examples.)

We note that even for a very simple signature set $\mathcal{F}$,
the signatures for all $\mathcal{F}$-gates can be quite complicated and expressive.

\subsection{Tractable Signature Sets}

We define the sets of signatures that were previously known to be tractable.
All quotations of results and definitions from~\cite{CGW13,GW13,CGW14a},
both in this section and throughout the paper,
refer to the full versions of these papers.

\paragraph{Affine Signatures}

\begin{definition}[Definition~3.1 in~\cite{CLX14}]
 A $k$-ary function $f(x_1, \dotsc, x_k)$ is \emph{affine} if it has the form
 \[
  \lambda \cdot \chi_{A x = 0} \cdot i^{\sum_{j=1}^n \langle \mathbf{v}_j, x \rangle},
 \]
 where $\lambda \in \mathbb{C}$,
 $x = \transpose{(x_1, x_2, \dotsc, x_k, 1)}$,
 $A$ is a matrix over $\mathbb{F}_2$,
 $\mathbf{v}_j$ is a vector over $\mathbb{F}_2$,
 and $\chi$ is a 0-1 indicator function such that $\chi_{Ax = 0}$ is~$1$ iff $A x = 0$.
 Note that the dot product $\langle \mathbf{v}_j, x \rangle$ is calculated over $\mathbb{F}_2$,
 while the summation $\sum_{j=1}^n$ on the exponent of $i = \sqrt{-1}$ is evaluated as a sum mod~$4$ of 0-1 terms.
 We use $\mathscr{A}$ to denote the set of all affine functions.
\end{definition}

Notice that there is no restriction on the number of rows in the matrix $A$.
It is permissible that $A$ is the zero matrix so that $\chi_{A x = 0} = 1$ holds for all $x$.
An equivalent way to express the exponent of $i$ is as a quadratic polynomial where all cross terms have an even coefficient
(cf.~\cite{CCLL10}).

It is known that the set of non-degenerate symmetric signatures in $\mathscr{A}$ is precisely the nonzero signatures
($\lambda \neq 0$) in $\mathscr{F}_1 \union \mathscr{F}_2 \union \mathscr{F}_3$ with arity at least~$2$,
where $\mathscr{F}_1$, $\mathscr{F}_2$, and $\mathscr{F}_3$ are three families of signatures defined as
\begin{align*}
 \mathscr{F}_1 &= \left\{\lambda \left([1,0]^{\otimes k} + i^r [0, 1]^{\otimes k}\right) \st \lambda \in \mathbb{C}, k = 1, 2, \dotsc, r = 0, 1, 2, 3\right\},\\
 \mathscr{F}_2 &= \left\{\lambda \left([1,1]^{\otimes k} + i^r [1,-1]^{\otimes k}\right) \st \lambda \in \mathbb{C}, k = 1, 2, \dotsc, r = 0, 1, 2, 3\right\}, \text{ and}\\
 \mathscr{F}_3 &= \left\{\lambda \left([1,i]^{\otimes k} + i^r [1,-i]^{\otimes k}\right) \st \lambda \in \mathbb{C}, k = 1, 2, \dotsc, r = 0, 1, 2, 3\right\}.
\end{align*}
We explicitly list these signatures up to an arbitrary constant multiple from $\mathbb{C}$:\\
\parbox{0.61\textwidth}{
 \begin{enumerate}
  \item $[1, 0, \dotsc, 0, \pm 1]$; \hfill $(\mathscr{F}_1, r=0,2)$
  \item $[1, 0, \dotsc, 0, \pm i]$; \hfill $(\mathscr{F}_1, r=1,3)$
  \item $[1,  0, 1,  0, \dotsc,   0  \text{ or } 1]$; \hfill $(\mathscr{F}_2, r=0)$
  \item $[1, -i, 1, -i, \dotsc, (-i) \text{ or } 1]$; \hfill $(\mathscr{F}_2, r=1)$
  \item $[0,  1, 0,  1, \dotsc,   0  \text{ or } 1]$; \hfill $(\mathscr{F}_2, r=2)$
  \item $[1,  i, 1,  i, \dotsc,   i  \text{ or } 1]$; \hfill $(\mathscr{F}_2, r=3)$
  \item $[1,  0, -1,  0, 1,  0, -1,  0, \dotsc, 0 \text{ or } 1 \text{ or } (-1)]$; \hfill $(\mathscr{F}_3, r=0)$
  \item $[1,  1, -1, -1, 1,  1, -1, -1, \dotsc,               1 \text{ or } (-1)]$; \hfill $(\mathscr{F}_3, r=1)$
  \item $[0,  1,  0, -1, 0,  1,  0, -1, \dotsc, 0 \text{ or } 1 \text{ or } (-1)]$; \hfill $(\mathscr{F}_3, r=2)$
  \item $[1, -1, -1,  1, 1, -1, -1,  1, \dotsc,               1 \text{ or } (-1)]$. \hfill $(\mathscr{F}_3, r=3)$
 \end{enumerate}}

\paragraph{Product-Type Signatures}

\begin{definition}[Definition~3.3 in~\cite{CLX14}]
 A function is of \emph{product type} if it can be expressed as a product of unary functions,
 binary equality functions $([1,0,1])$, and binary disequality functions $([0,1,0])$.
 We use $\mathscr{P}$ to denote the set of product-type functions.
\end{definition}

An alternate definition for $\mathscr{P}$,
implicit in~\cite{CLX11a},
is the tensor closure of signatures with support on two complementary bit vectors.
It can be shown
(cf.~Lemma~A.1 in the full version of~\cite{HL12})
that if $f$ is a symmetric signature in $\mathscr{P}$,
then $f$ is either degenerate,
binary \textsc{Disequality} $\neq_2$,
or $[a,0,\dotsc,0,b]$ for some $a, b \in \mathbb{C}$.

\paragraph{Matchgate Signatures}

Matchgates were introduced by Valiant~\cite{Val02a, Val02b} to give polynomial-time algorithms for a collection of counting problems over planar graphs.
As the name suggests,
problems expressible by matchgates can be reduced to computing a weighted sum of perfect matchings.
The latter problem is tractable over planar graphs by Kasteleyn's algorithm~\cite{Kas67},
a.k.a.~the FKT algorithm~\cite{TF61,Kas61}.
These counting problems are naturally expressed in the Holant framework using \emph{matchgate signatures}.
We use $\mathscr{M}$ to denote the set of all matchgate signatures;
thus $\PlHolant(\mathscr{M})$ is tractable.
Holographic transformations extend the reach of the FKT algorithm even further,
as stated below.

\begin{theorem} \label{thm:tractable:M}
 Let $\mathcal{F}$ be any set of symmetric, complex-valued signatures in Boolean variables.
 If $\mathcal{F}$ is $\mathscr{M}$-transformable,
 then $\PlHolant(\mathcal{F})$ is computable in polynomial time.
\end{theorem}

Matchgate signatures are characterized by the matchgate identities
(see~\cite{CG14} for the identities and a self-contained proof).
The parity of a matchgate signature is even (resp.~odd) if its support is on entries of even (resp.~odd) Hamming weight.
We explicitly list all the symmetric signatures in $\mathscr{M}$ (see~\cite{CG14}).

\begin{proposition} \label{prop:match:symmetric}
 Let $f$ be a symmetric signature in $\mathscr{M}$.
 Then there exists $a,b \in \mathbb{C}$ and $n \in \N$ such that $f$ takes one of the following forms:\\
 \parbox{0.68\textwidth}{
  \begin{enumerate}
   \item $[a^n, 0, a^{n-1} b, 0, \dotsc, 0, a b^{n-1}, 0, b^n]$       \hfill (of arity $2 n \phantom{{} + 1} \ge 2$);
   \item $[a^n, 0, a^{n-1} b, 0, \dotsc, 0, a b^{n-1}, 0, b^n, 0]$    \hfill (of arity $2 n             + 1  \ge 1$);
   \item $[0, a^n, 0, a^{n-1} b, 0, \dotsc, 0, a b^{n-1}, 0, b^n]$    \hfill (of arity $2 n             + 1  \ge 1$);
   \item $[0, a^n, 0, a^{n-1} b, 0, \dotsc, 0, a b^{n-1}, 0, b^n, 0]$ \hfill (of arity $2 n             + 2  \ge 2$).
  \end{enumerate}
 }\\
 In the last three cases with $n = 0$,
 the signatures are $[1,0]$, $[0,1]$, and $[0,1,0]$.
 Any multiple of these is also a matchgate signature.
\end{proposition}

Roughly speaking,
the symmetric matchgate signatures have~$0$ for every other entry
(which is called the \emph{parity condition}),
and form a geometric progression with the remaining entries.

Another useful way to view the symmetric signature in $\mathscr{M}$ is via a low tensor rank decomposition.
To state these low rank decompositions,
we use the following definition.

\begin{definition} \label{def:sym}
 Let $S_n$ be the symmetric group of degree $n$.
 Then for positive integers $t$ and $n$ with $t \le n$ and unary signatures $v, v_1, \dotsc, v_{n-t}$,
 we define
 \[\Sym_n^t(v; v_1, \dotsc, v_{n-t}) = \sum_{\pi \in S_n} \bigtensor_{k=1}^n u_{\pi(k)},\]
 where the ordered sequence
 $(u_1, u_2, \dotsc, u_n) = ( \underbrace{v, \dotsc, v}_{t \text{ copies}}, v_1, \dotsc, v_{n-t} )$.
\end{definition}

\begin{proposition} \label{prop:match:symmetric:decomp}
 Let $f$ be a symmetric signature in $\mathscr{M}$ of arity $n$.
 Then there exist $a,b,\lambda \in \mathbb{C}$ such that $f$ takes one of the following forms:
 \begin{enumerate}
  \item $[a, b]^{\otimes n} + [a, -b]^{\otimes n}
  = \begin{cases*}
     2 [a^n, 0, a^{n-2} b^2, 0, \dotsc, 0,   b^n]        & $n$ is even,\\
     2 [a^n, 0, a^{n-2} b^2, 0, \dotsc, 0, a b^{n-1}, 0] & $n$ is odd;
    \end{cases*}$
  \item $[a, b]^{\otimes n} - [a, -b]^{\otimes n}
  = \begin{cases*}
     2 [0, a^{n-1} b, 0, a^{n-3} b^3, 0, \dotsc, 0, a b^{n-1}, 0] & $n$ is even,\\
     2 [0, a^{n-1} b, 0, a^{n-3} b^3, 0, \dotsc, 0,   b^n]        & $n$ is odd;
    \end{cases*}$
  \item $\lambda \Sym_n^{n-1}([1, 0]; [0, 1])
  = [0,\lambda ,0,\dotsc,0]$;
  \item $\lambda \Sym_n^{n-1}([0, 1]; [1, 0])
  = [0,\dotsc,0,\lambda ,0]$.
 \end{enumerate}
\end{proposition}

The understanding of matchgates was further developed in~\cite{CL11a},
which characterized,
for every symmetric signature,
the set of holographic transformations under which the transformed signature becomes a matchgate signature.

\paragraph{Vanishing Signatures}

Vanishing signatures were first introduced in~\cite{GLV13} in the parity setting to denote signatures for which the Holant value is always~$0$ modulo~$2$.

\begin{definition}
 A set of signatures $\mathcal{F}$ is called \emph{vanishing} if the value $\Holant_\Omega(\mathcal{F})$ is~$0$ for every signature grid $\Omega$.
 A signature $f$ is called \emph{vanishing} if the singleton set $\{f\}$ is vanishing.
\end{definition}

A Holant problem defined only by vanishing signatures is trivially tractable by definition.
Question is how to determine which sets of signatures are vanishing?
We introduce the following definitions to answer this question.

\begin{definition}[Definition~4.4 in~\cite{CGW13}]
 A nonzero symmetric signature $f$ of arity $n$ has \emph{positive vanishing degree} $k \ge 1$, which is denoted by $\vd^+(f) = k$,
 if $k \le n$ is the largest positive integer such that there exists $n-k$ unary signatures $v_1, \dotsc, v_{n-k}$ satisfying
 \begin{align*}
  f = \Sym_n^{k}([1,i]; v_1, \dots, v_{n-k}).
 \end{align*}
 If $f$ cannot be expressed as such a symmetrization form, we define $\vd^+(f) = 0$.
 If $f$ is the all zero signature, define $\vd^+(f) = n + 1$.
 
 We define \emph{negative vanishing degree} $\vd^-$ similarly, using $-i$ instead of $i$.
\end{definition}

\begin{definition}[Definition~4.5 in~\cite{CGW13}]
 For $\sigma \in \{+, -\}$,
 we define $\mathscr{V}^\sigma = \{f \st 2 \vd^\sigma(f) > \arity(f)\}$.
\end{definition}
Furthermore,
we let $\mathscr{V} = \mathscr{V}^+ \cup \mathscr{V}^-$.
The fact that $\mathscr{V}$ is closed under orthogonal transformations follows directly from the next lemma.

\begin{lemma} \label{lem:van:orth}
 For a symmetric signature $f$ of arity $n$, $\sigma \in \{+,-\}$, and an orthogonal matrix $T \in \mathbb{C}^{2 \times 2}$,
 either $\vd^\sigma(f) = \vd^\sigma(T^{\otimes n} f)$ or $\vd^\sigma(f) = \vd^{-\sigma}(T^{\otimes n} f)$.
\end{lemma}

The following characterization of vanishing signature sets holds.

\begin{theorem}[Theorem~4.13 in~\cite{CGW13}] \label{thm:van}
 Let $\mathcal{F}$ be a set of symmetric signatures.
 Then $\mathcal{F}$ is vanishing
 if and only if
 $\mathcal{F} \subseteq \mathscr{V}^{+}$ or $\mathcal{F} \subseteq \mathscr{V}^{-}$.
\end{theorem}

To prove this theorem,
two more definitions were made,
which complement the previous two definitions because of Corollary~\ref{cor:van:degree}.

\begin{definition}[Definition~4.7 in~\cite{CGW13}]
 A symmetric signature $f = [f_0, f_1, \dotsc, f_n]$ of arity $n$ is in $\mathscr{R}^{+}_t$ for a nonnegative integer $t \ge 0$
 if $t > n$ or for any $0 \le k \le n - t$, $f_k, \dotsc, f_{k+t}$ satisfy the recurrence relation
 \begin{align}
  \binom{t}{t} i^t f_{k+t} + \binom{t}{t-1} i^{t-1} f_{k+t-1} + \dotsb + \binom{t}{0} i^0 f_k = 0. \label{eqn:recurrence}
 \end{align}
 We define $\mathscr{R}^{-}_t$ similarly but with $-i$ in place of $i$ in~(\ref{eqn:recurrence}).
\end{definition}

\begin{definition}[Definition~4.8 in~\cite{CGW13}] \label{def:Rt}
 For a nonzero symmetric signature $f$ of arity $n$,
 it is of \emph{positive} (resp.~\emph{negative}) \emph{recurrence degree} $t \le n$, denoted by $\rd^+(f) = t$ (resp.~$\rd^-(f) = t$),
 if and only if $f \in \mathscr{R}^+_{t+1} - \mathscr{R}^+_{t}$ (resp.~$f \in \mathscr{R}^-_{t+1} - \mathscr{R}^-_{t}$).
 If $f$ is the all zero signature, we define $\rd^+(f) = \rd^-(f)= -1$.
\end{definition}

\begin{corollary}[Corollary~4.16 in~\cite{CGW13}] \label{cor:van:degree}
 If $f$ is a symmetric signature and $\sigma \in \{+,-\}$,
 then $\vd^\sigma(f) + \rd^\sigma(f) = \arity(f)$.
\end{corollary}

An observation was made in Section~4.3 of~\cite{CGW13} that we utilize.
We state it here as a lemma.

\begin{lemma} \label{lem:prelim:vanishing_form_in_Z_basis}
 Suppose $f$ is a symmetric signature of arity $n$.
 Let $\hat{f} = (Z^{-1})^{\otimes n} f$.
 If $\rd^+(f) = d$,
 then $\hat{f} = [\hat{f}_0, \hat{f}_1, \dotsc, \hat{f}_d, 0, \dotsc, 0]$ and $\hat{f}_d \ne 0$.
 Also $f \in \mathscr{R}_d^+$ iff 
 all nonzero entries of $\hat{f}$ are among the first $d$ entries in its symmetric signature notation.
 
 Similarly,
 if $\rd^-(f) = d$,
 then $\hat{f} = [0, \dotsc, 0, \hat{f}_{n-d}, \dotsc, \hat{f}_n]$ and $\hat{f}_{n-d} \ne 0$.
 Also $f\in\mathscr{R}_d^-$ iff
 all nonzero entries of $\hat{f}$ are among the last $d$ entries in its symmetric signature notation.
\end{lemma}

The following lemma is a reduction involving binary signatures in the $Z$ basis.
It is used in Section~\ref{sec:mixing:vanishing} to help determine what binary signatures can mix with vanishing signatures.
The original statement is for general graphs,
but the proof clearly holds for planar graphs as well.

\begin{lemma}[Lemma A.1 in \cite{CGW13}] \label{lem:simple_interpolation:van:bin}
 Let $x \in \mathbb{C}$.
 If $x \ne 0$,
 then for any set $\mathcal{F}$ containing $[x,1,0]$,
 we have
 \[
  \plholant{{\neq}_2}{\mathcal{F} \union \{[v,1,0]\}}
  \le_T
  \plholant{{\neq}_2}{\mathcal{F}}
 \]
 for any $v \in \mathbb{C}$.
\end{lemma}

\subsection{Some Known Dichotomies}

Here we list several known dichotomies.
The first is the dichotomy for Holant.

\begin{theorem}[Theorem~5.1 in~\cite{CGW13}] \label{thm:Holant:set}
 Let $\mathcal{F}$ be any set of symmetric, complex-valued signatures in Boolean variables.
 Then $\Holant(\mathcal{F})$ is $\SHARPP$-hard unless $\mathcal{F}$ satisfies one of the following conditions, in which case the problem is in $\P$:
 \begin{enumerate}
  \item All non-degenerate signatures in $\mathcal{F}$ are of arity at most~2;
  \item $\mathcal{F}$ is $\mathscr{A}$-transformable;
  \item $\mathcal{F}$ is $\mathscr{P}$-transformable;
  \item $\mathcal{F} \subseteq \mathscr{V}^\sigma \union \{f \in \mathscr{R}_2^\sigma \st \arity(f) = 2\}$ for $\sigma \in \{+,-\}$;
  \item All non-degenerate signatures in $\mathcal{F}$ are in $\mathscr{R}_2^\sigma$ for $\sigma \in \{+,-\}$.
 \end{enumerate}
\end{theorem}

%To extend this result to planar graphs,
%we use several dichotomy theorems for planar Holant problems.
We also use several dichotomy theorems for planar Holant problems with additional restrictions.
The first of these is a dichotomy theorem for a single signature of small arity.
It is a combination of Theorem~V.1 in~\cite{CLX10} and Theorem~14 in~\cite{GW13} for arity~$3$ and~$4$, respectively.
This theorem forms the base case of an inductive proof of Theorem~\ref{thm:dic:single},
our single signature dichotomy.

\begin{theorem} \label{thm:PlHolant:arity34}
 If $f$ is a non-degenerate, symmetric, complex-valued signature of arity~$3$ or $4$ in Boolean variables,
 then $\PlHolant(f)$ is $\SHARPP$-hard unless
 $f$ satisfies one of the following conditions,
 in which case,
 the problem is computable in polynomial time:
 \begin{enumerate}
  \item $\Holant(f)$ is tractable (i.e. $f$ is $\mathscr{A}$-transformable, $\mathscr{P}$-transformable, or vanishing);
  \item $f$ is $\mathscr{M}$-transformable.
 \end{enumerate}
\end{theorem}

We also state a corollary of this result,
which shows that counting weighted matchings
in $4$-regular planar graphs is \numP-hard.
This is easier to apply than Theorem \ref{thm:PlHolant:arity34}.

\begin{corollary}[Lemma~5.5 in~\cite{GW13}] \label{cor:arity4:double_root}
 Let $v \in \mathbb{C}$.
 If $v \ne 0$,
 then $\PlHolant([v,1,0,0,0])$ is \numP-hard.
\end{corollary}

Next is a dichotomy theorem about counting complex weighted graph homomorphisms over degree prescribed graphs.

\begin{theorem}[Theorem~3 in~\cite{CK12}] \label{thm:k-reg_homomorphism}
 Let $S \subseteq \mathbb{Z}^+$ containing some $r \ge 3$,
 let $\mathcal{G} = \{=_k \st k \in S\}$,
 and let $d = \gcd(S)$.
 Further suppose that $f_0, f_1, f_2 \in \mathbb{C}$.
 Then $\plholant{[f_0, f_1, f_2]}{\mathcal{G}}$ is $\SHARPP$-hard unless
 one of the following conditions holds:
 \begin{enumerate}
  \item $f_0 f_2 = f_1^2$;
  \item $f_0 = f_2 = 0$;
  \item $f_1 = 0$;
  \item $f_0 f_2 = -f_1^2$ and $f_0^d = -f_2^d\neq0$;
  \item $f_0^d = f_2^d\neq0$.
 \end{enumerate}
 In all exceptional cases,
 the problem is computable in polynomial time.
\end{theorem}

Theorem~\ref{thm:k-reg_homomorphism} is the original statement as in~\cite{CK12}.
It is explicit and easy to apply.
Conceptually, it can be restated as Theorem~\ref{thm:k-reg_homomorphism}$'$,
which supports the putative form of the $\PlCSP^d$ dichotomy.

%%% TDW: this command makes the next "special theorem" have the same numbering as the previous theorem but with a prime in mathmode ($'$).
\theoremstyle{plain}
\newtheorem*{specialtheorem}{Theorem {\thetheorem}$'$}

%%% TDW: to cite this theorem, cite the previous theorem and add a prime in mathmode ($'$).

\begin{specialtheorem}[Theorem~3 in~\cite{CK12}]
 Let $S \subseteq \mathbb{Z}^+$ contain $k \ge 3$,
 let $\mathcal{G} = \{=_k \st k \in S\}$,
 and let $d = \gcd(S)$.
 Further suppose that $f$ is a non-degenerate, symmetric, complex-valued binary signature in Boolean variables.
 Then $\plholant{f}{\mathcal{G}}$ is \numP-hard unless
 $f$ satisfies one of the following conditions,
 in which case,
 the problem is computable in polynomial time:
 \begin{enumerate}
  \item there exists $T \in \mathcal{T}_{4d}$ such that $T^{\otimes 2} f \in \mathscr{A}$;
  \item                                                 $f               \in \mathscr{P}$;
  \item there exists $T \in \mathcal{T}_{2d}$ such that $T^{\otimes 2} f \in \widehat{\mathscr{M}}$.
 \end{enumerate}
\end{specialtheorem}

Lastly,
we quote the $\PlCSP$ dichotomy.
It also supports the putative form of a dichotomy,
which states that holographic algorithms using matchgates followed by the FKT algorithm is a universal strategy.

\begin{theorem}[Theorem~19 in~\cite{GW13}] \label{thm:PlCSP}
 Let $\mathcal{F}$ be any set of symmetric, complex-valued signatures in Boolean variables.
 Then $\PlCSP(\mathcal{F})$ is $\SHARPP$-hard unless
 $\mathcal{F} \subseteq \mathscr{A}$,
 $\mathcal{F} \subseteq \mathscr{P}$, or
 $\mathcal{F} \subseteq \widehat{\mathscr{M}}$,
 in which case the problem is computable in polynomial time.
\end{theorem}

\subsection{Redundant Signature Matrices and Related Hardness Results}
\label{sec:redundant}

\begin{definition}[Definition~6.1 in~\cite{CGW13}]
 A 4-by-4 matrix is \emph{redundant} if its middle two rows and middle two columns are the same.
\end{definition}

An example of a redundant matrix is the signature matrix of a symmetric arity~$4$ signature.

\begin{definition}[Definition~6.2 in~\cite{CGW13}]
 The \emph{signature matrix} of a symmetric arity~$4$ signature $f = [f_0, f_1, f_2, f_3, f_4]$ is
 \begin{align*}
  M_f =
  \begin{bmatrix}
   f_0 & f_1 & f_1 & f_2\\
   f_1 & f_2 & f_2 & f_3\\
   f_1 & f_2 & f_2 & f_3\\
   f_2 & f_3 & f_3 & f_4
  \end{bmatrix}.
 \end{align*}
 This definition extends to an asymmetric signature $g$ as
 \begin{align*}
  M_g =
  \begin{bmatrix}
   g^{0000} & g^{0010} & g^{0001} & g^{0011}\\
   g^{0100} & g^{0110} & g^{0101} & g^{0111}\\
   g^{1000} & g^{1010} & g^{1001} & g^{1011}\\
   g^{1100} & g^{1110} & g^{1101} & g^{1111}
  \end{bmatrix},
 \end{align*}
 where $g^{wxyz}$ is the output of $g$ on input $wxyz$.
 When we present $g$ as an $\mathcal{F}$-gate, we order the four external edges ABCD counterclockwise.
 In $M_g$,
 the row index bits are ordered AB and the column index bits are ordered DC,
 in reverse order.
 This is for convenience so that the signature matrix of the linking of two arity~4 $\mathcal{F}$-gates
 is the matrix product of the signature matrices of the two $\mathcal{F}$-gates.

 If $M_g$ is redundant, we also define the \emph{compressed signature matrix} of $g$ as
 \[
  \widetilde{M_g}
  =
  \begin{bmatrix}
   1 & 0 & 0 & 0\\
   0 & \frac{1}{2} & \frac{1}{2} & 0\\
   0 & 0 & 0 & 1
  \end{bmatrix}
  M_g
  \begin{bmatrix}
   1 & 0 & 0\\
   0 & 1 & 0\\
   0 & 1 & 0\\
   0 & 0 & 1
  \end{bmatrix}.
 \]
\end{definition}

\begin{lemma} [Corollary 3.8 in \cite{GW13}] \label{lem:arity4:nonsingular_compressed_hard}
 Let $f$ be an arity $4$ signature with complex weights.
 If $M_f$ is redundant and $\widetilde{M_f}$ is nonsingular,
 then $\PlHolant(f)$ is $\numP$-hard.
\end{lemma}

Furthermore,
by combining Lemma~\ref{lem:arity4:nonsingular_compressed_hard} with Lemma~6.8 in~\cite{CGW13},
we obtain the planar version of Corollary~6.9 in~\cite{CGW13}.

\begin{corollary} \label{cor:prelim:nonsingular_compressed_hard_trans}
 Let $f$ be an arity~$4$ signature with complex weights.
 If there exists a nonsingular matrix $T \in \mathbb{C}^{2 \times 2}$ such that $\hat{f} = T^{\otimes 4} f$,
 where $M_{\hat{f}}$ is redundant and $\widetilde{M_{\hat{f}}}$ is nonsingular,
 then $\PlHolant(f)$ is $\SHARPP$-hard.
\end{corollary}

\begin{figure}[t]
 \centering
 \def\capWidth{6cm}
 \captionsetup[subfigure]{width=\capWidth}
 \tikzstyle{entry} = [internal, inner sep=2pt]
 \subfloat[A counterclockwise rotation]{
  \begin{tikzpicture}[scale=\scale,transform shape,node distance=1.7 * \nodeDist,semithick]
   \node[internal]  (0)                    {};
   \node[external]  (1) [above  left of=0] {};
   \node[external]  (2) [above right of=0] {};
   \node[external]  (3) [below  left of=0] {};
   \node[external]  (4) [below right of=0] {};
   \node[external]  (5) [      right of=0] {};
   \node[external]  (6) [      right of=5] {};
   \node[internal]  (7) [      right of=6] {};
   \node[external]  (8) [above  left of=7] {};
   \node[external]  (9) [above right of=7] {};
   \node[external] (10) [below  left of=7] {};
   \node[external] (11) [below right of=7] {};
   \path (0) edge[postaction={decorate, decoration={
                                         markings,
                                         mark=at position 0.25 with {\arrow[>=diamond,white] {>}; },
                                         mark=at position 0.25 with {\arrow[>=open diamond]  {>}; },
                                         mark=at position 0.65 with {\arrow[>=diamond,white] {>}; },
                                         mark=at position 0.65 with {\arrow[>=open diamond]  {>}; } } }] (1)
             edge (2)
             edge (3)
             edge (4)
    (5.west) edge[->, >=stealth] (6.east)
         (7) edge[postaction={decorate, decoration={
                                         markings,
                                         mark=at position 0.65 with {\arrow[>=diamond,white] {>}; },
                                         mark=at position 0.65 with {\arrow[>=open diamond]  {>}; } } }] (8)
             edge (9)
             edge [postaction={decorate, decoration={
                                         markings,
                                         mark=at position 0.25 with {\arrow[>=diamond,white] {>}; },
                                         mark=at position 0.25 with {\arrow[>=open diamond]  {>}; } } }] (10)
             edge (11);
   \begin{pgfonlayer}{background}
    \node[draw=\borderColor,thick,rounded corners,fit = (0),inner sep=16pt] {};
%     \node[draw=\borderColor,thick,rounded corners,fit = (0),inner sep=10pt,transform shape=false] {};
    \node[draw=\borderColor,thick,rounded corners,fit = (7),inner sep=16pt] {};
%     \node[draw=\borderColor,thick,rounded corners,fit = (7),inner sep=10pt,transform shape=false] {};
   \end{pgfonlayer}
  \end{tikzpicture}}
 \qquad
 \subfloat[Movement of signature matrix entries]{
  \makebox[\capWidth][c]{
   \begin{tikzpicture}[scale=\scale,transform shape,>=stealth,node distance=\nodeDist,semithick]
    \node[entry] (11)               {};
    \node[entry] (12) [right of=11] {};
    \node[entry] (13) [right of=12] {};
    \node[entry] (14) [right of=13] {};
    \node[entry] (21) [below of=11] {};
    \node[entry] (22) [right of=21] {};
    \node[entry] (23) [right of=22] {};
    \node[entry] (24) [right of=23] {};
    \node[entry] (31) [below of=21] {};
    \node[entry] (32) [right of=31] {};
    \node[entry] (33) [right of=32] {};
    \node[entry] (34) [right of=33] {};
    \node[entry] (41) [below of=31] {};
    \node[entry] (42) [right of=41] {};
    \node[entry] (43) [right of=42] {};
    \node[entry] (44) [right of=43] {};
    \node[external] (nw) [above left  of=11] {};
    \node[external] (ne) [above right of=14] {};
    \node[external] (sw) [below left  of=41] {};
    \node[external] (se) [below right of=44] {};
    \path (13) edge[<-, dotted]                (12)
          (12) edge[<-, dotted]                (21)
          (21) edge[<-, dotted]                (31)
          (31) edge[<-, dotted,out=65,in=-155] (13)
          (42) edge[<-, dashed]                (43)
          (43) edge[<-, dashed]                (34)
          (34) edge[<-, dashed]                (24)
          (24) edge[<-, dashed,out=-115,in=25] (42)
          (14) edge[<-, very thick]            (22)
          (22) edge[<-, very thick]            (41)
          (41) edge[<-, very thick]            (33)
          (33) edge[<-, very thick]            (14)
          (23) edge[<->]                      (32);
    \path (nw.west) edge (sw.west)
          (ne.east) edge (se.east)
          (nw.west) edge (nw.east)
          (sw.west) edge (sw.east)
          (ne.west) edge (ne.east)
          (se.west) edge (se.east);
   \end{tikzpicture}}}
 \caption{The movement of the entries in the signature matrix of a quaternary signature under a counterclockwise rotation of the input edges.
  Entires of Hamming weight~$1$ are in the dotted cycle,
  entires of Hamming weight~$2$ are in the two solid cycles (one has length~$4$ and the other one is a swap),
  and entries of Hamming weight~$3$ are in the dashed cycle.}
 \label{fig:rotate_asymmetric_signature}
\end{figure}
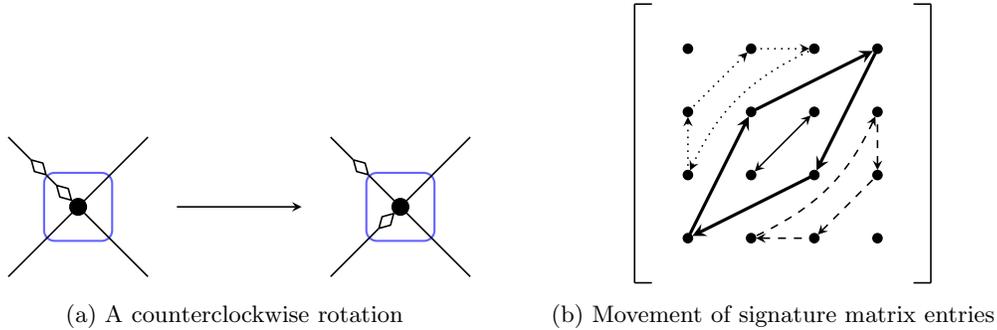

In the course of working with symmetric signature,
we sometimes construct gadgets with signatures that are not symmetric.
The power of Lemma~\ref{lem:arity4:nonsingular_compressed_hard} and Corollary~\ref{cor:prelim:nonsingular_compressed_hard_trans}
is that they apply to such signatures provided the corresponding signature matrix is redundant.
Sometimes one can apply a rotation to obtain a signature with a redundant signature matrix (see Figure~\ref{fig:rotate_asymmetric_signature}).

\section{\texorpdfstring{$\mathscr{A}$}{A}-, \texorpdfstring{$\mathscr{P}$}{P}-, and \texorpdfstring{$\mathscr{M}$}{M}-transformable Signatures} \label{sec:APMtrans}

In this section,
we investigate the properties of $\mathscr{A}$-, $\mathscr{P}$-, and $\mathscr{M}$-transformable signatures.
Throughout,
we define $\alpha = \frac{1+i}{\sqrt{2}} = \sqrt{i} = e^{\frac{\pi i}{4}}$
and use $\mathbf{O}_2(\mathbb{C})$ to denote the group of $2$-by-$2$ orthogonal matrices over $\mathbb{C}$.
While the main results in this section assume that the signatures involved are symmetric,
we note that some of the lemmas also hold without this assumption.

\subsection{Characterization of \texorpdfstring{$\mathscr{A}$}{A}- and \texorpdfstring{$\mathscr{P}$}{P}-transformable Signatures}

$\mathscr{A}$- and $\mathscr{M}$-transformable signatures have been well studied in previous work \cite{CGW13,CGW14a}.
We summarize some useful notions and lemmas here.
The three sets $\mathscr{A}_1$,
$\mathscr{A}_2$,
and $\mathscr{A}_3$ capture all symmetric $\mathscr{A}$-transformable signatures.

\begin{definition} \label{def:A123}
 A symmetric signature $f$ of arity $n$ is in, respectively, $\mathscr{A}_1$, or $\mathscr{A}_2$, or $\mathscr{A}_3$
 if there exist an $H \in \mathbf{O}_2(\mathbb{C})$ and a nonzero constant $c \in \mathbb{C}$ such that $f$ has the form, respectively,
 $c H^{\otimes n} \left(\tbcolvec{1}{1}^{\otimes n}      + \beta \tbcolvec{1}{-1}^{\otimes n}\right)$, or
 $c H^{\otimes n} \left(\tbcolvec{1}{i}^{\otimes n}      +       \tbcolvec{1}{-i}^{\otimes n}\right)$, or
 $c H^{\otimes n} \left(\tbcolvec{1}{\alpha}^{\otimes n} +   i^r \tbcolvec{1}{-\alpha}^{\otimes n}\right)$,
 where $\beta = \alpha^{tn+2r}$, $r \in \{0,1,2,3\}$, and $t \in \{0,1\}$.
\end{definition}

For $k \in \{1,2,3\}$,
when such an orthogonal $H$ exists,
we say that $f \in \mathscr{A}_k$ with transformation $H$.
If $f \in \mathscr{A}_k$ with $I_2$,
then we say $f$ is in the canonical form of $\mathscr{A}_k$.

The following lemma characterizes the signatures in $\mathscr{A}_2$.

\begin{lemma}[Lemma~8.8 in~\cite{CGW13}] \label{lem:single:A2}
 Let $f$ be a symmetric signature of arity $n$.
 Then $f \in \mathscr{A}_2$ if and only if
    $f = c \left(\left[\begin{smallmatrix} 1 \\  i \end{smallmatrix}\right]^{\otimes n}
         + \beta \left[\begin{smallmatrix} 1 \\ -i \end{smallmatrix}\right]^{\otimes n}\right)$
 for some nonzero constants $c, \beta \in \mathbb{C}$.
\end{lemma}

Membership in these three sets characterize the $\mathscr{A}$-transformable signatures.

\begin{lemma}[Lemma~8.10 in~\cite{CGW13}] \label{lem:cha:affine}
 Let $f$ be a non-degenerate symmetric signature.
 Then $f$ is $\mathscr{A}$-transformable if and only if $f \in \mathscr{A}_1 \union \mathscr{A}_2 \union \mathscr{A}_3$.
\end{lemma}

There is a similar characterization for $\mathscr{P}$-transformable signatures.

\begin{definition}\label{definition:P1}
 A symmetric signature $f$ of arity $n$ is in $\mathscr{P}_1$ if
 there exist an $H \in \mathbf{O}_2(\mathbb{C})$ and a nonzero $c \in \mathbb{C}$ such that
 $f = c H^{\otimes n} \left(\tbcolvec{1}{ 1}^{\otimes n}
                    + \beta \tbcolvec{1}{-1}^{\otimes n}\right)$,
 where $\beta \neq 0$.
\end{definition}

We define $\mathscr{P}_2 = \mathscr{A}_2$.
For $k \in \{1,2\}$,
when such an $H$ exists,
we say that $f \in \mathscr{P}_k$ with transformation $H$.
If $f \in \mathscr{P}_k$ with $I_2$,
then we say $f$ is in the canonical form of $\mathscr{P}_k$.

\begin{lemma}[Lemma~8.13 in~\cite{CGW13}] \label{lem:cha:product}
 Let $f$ be a non-degenerate symmetric signature.
 Then $f$ is $\mathscr{P}$-transformable if and only if $f \in \mathscr{P}_1 \union \mathscr{P}_2$.
\end{lemma}

\subsection{Characterization of \texorpdfstring{$\mathscr{M}$}{M}-transformable Signatures}

Now we develop a similar theory for the $\mathscr{M}$-transformable signatures.
Recall from Definition~\ref{def:prelim:trans} that for a signature set $\mathcal{F}$ to be $\mathscr{M}$-transformable,
it must be that there exists a $T \in \mathbf{GL}_2(\mathbb{C})$ such that $[1,0,1] T^{\otimes 2} \in \mathscr{M}$.
Since $[1,0,1]$ is symmetric,
$[1,0,1] T^{\otimes 2}$ is also symmetric.
However,
it is unnecessary to consider all binary signatures in $\mathscr{M}$.
We can normalize via \emph{right} multiplication by elements in
\[
 \Stab{M} = \{T \in \mathbf{GL}_2(\mathbb{C}) \st T \mathscr{M} \subseteq \mathscr{M}\},
\]
the stabilizer group of $\mathscr{M}$.
Technically this set is the left stabilizer group of $\mathscr{M}$,
but it is easy to see that the left and right stabilizer groups of $\mathscr{M}$ coincide
and that they are generated by nonzero scalar multiples of matrices of the form $\left[\begin{smallmatrix} 1 & 0 \\ 0 & \nu \end{smallmatrix}\right]$
for any nonzero $\nu \in \mathbb{C}$ and $X = \left[\begin{smallmatrix} 0 & 1 \\ 1 & 0 \end{smallmatrix}\right]$.

After this normalization,
it is enough to consider cases~\ref{prop:matrix_cha:case_101} and~\ref{prop:matrix_cha:case_010} in the following proposition.

\begin{proposition}[Proposition~8.1 in~\cite{CGW13}] \label{prop:matrix_cha}
 Let $T \in \mathbb{C}^{2 \times 2}$ be a matrix.
 Then the following hold:
 \begin{enumerate}
  \item $[1,0,1] T^{\otimes 2} = [1,0,1]$ if and only if $T \in \mathbf{O}_2(\mathbb{C})$; \label{prop:matrix_cha:case_101}
  \item $[1,0,1] T^{\otimes 2} = [1,0,i]$ if and only if there exists an $H \in \mathbf{O}_2(\mathbb{C})$
  such that $T = H \left[\begin{smallmatrix} 1 & 0 \\ 0 & \alpha \end{smallmatrix}\right]$; \label{prop:matrix_cha:case_10a}
  \item $[1,0,1] T^{\otimes 2} = [0,1,0]$ if and only if there exists an $H \in \mathbf{O}_2(\mathbb{C})$
  such that $T = \frac{1}{\sqrt{2}} H \left[\begin{smallmatrix} 1 & 1 \\ i & -i \end{smallmatrix}\right]$. \label{prop:matrix_cha:case_010}
 \end{enumerate}
\end{proposition}

\begin{lemma} \label{lem:match:trans}
 Let $\mathcal{F}$ be a set of signatures.
 Then $\mathcal{F}$ is $\mathscr{M}$-transformable if and only if
 $\mathcal{F} \subseteq \left[\begin{smallmatrix} 1 & 1 \\ i & -i \end{smallmatrix}\right] \mathscr{M}$ or
 there exists an $H \in \mathbf{SO}_2(\mathbb{C})$ such that $\mathcal{F} \subseteq H \mathscr{M}$.
\end{lemma}

\begin{proof}
 Sufficiency is easily verified by checking that $=_2$ is transformed into $\mathscr{M}$ in both cases.
 In particular,
 $H$ leaves $=_2$ unchanged.
 
 If $\mathcal{F}$ is $\mathscr{M}$-transformable,
 then by definition,
 there exists a matrix $T$ such that $(=_2) T^{\otimes 2} \in \mathscr{M}$ and $\mathcal{F} \subseteq T \mathscr{M}$.
 The non-degenerate binary signatures in $\mathscr{M}$ are either $[0,1,0]$ or of the form $[1,0,\nu]$,
 up to a scalar.
 However,
 notice that $[1,0,1] = [1,0,\nu]
 \left[\begin{smallmatrix} 1 & 0 \\ 0 & \nu^{-\frac{1}{2}} \end{smallmatrix}\right]^{\otimes 2}$
 and $\left[\begin{smallmatrix} 1 & 0 \\ 0 & \nu^{-\frac{1}{2}} \end{smallmatrix}\right] \in \Stab{M}$.
 Thus,
 we only need to consider $[1,0,1]$ and $[0,1,0]$.
 Now we apply Proposition~\ref{prop:matrix_cha}.
 \begin{enumerate}
  \item If $(=_2) T^{\otimes 2} = [1,0,1]$, 
  then by case~\ref{prop:matrix_cha:case_101} of Proposition~\ref{prop:matrix_cha}, 
  we have $T \in \mathbf{O}_2(\mathbb{C})$.
  If $T \in \mathbf{SO}_2(\mathbb{C})$,
  then we are done with $H = T$.
  Otherwise,
  $T \in \mathbf{O}_2(\mathbb{C}) - \mathbf{SO}_2(\mathbb{C})$.
  We want to find an $H \in \mathbf{SO}_2(\mathbb{C})$ such that $\mathcal{F} \subseteq H \mathscr{M}$.
  Let $H = T \trans{1}{0}{0}{-1} \in \mathbf{SO}_2(\mathbb{C})$.
  Then
  \[
   \mathcal{F}
   \subseteq T \mathscr{M}
   = T \begin{bmatrix} 1 & 0 \\ 0 & -1 \end{bmatrix} \mathscr{M}
   = H \mathscr{M}
  \]
  since $\trans{1}{0}{0}{-1} \in \Stab{M}$.
  
  \item If $(=_2) T^{\otimes 2} = [0,1,0]$, 
  then by case~\ref{prop:matrix_cha:case_010} of Proposition~\ref{prop:matrix_cha}, 
  there exists an $H \in \mathbf{O}_2(\mathbb{C})$ such that $T = \frac{1}{\sqrt{2}} H \trans{1}{1}{i}{-i}$.
  Therefore $\mathcal{F} \subseteq H \trans{1}{1}{i}{-i} \mathscr{M}$.
  Furthermore,
  if $H = \trans{a}{b}{-b}{a} \in \mathbf{SO}_2(\mathbb{C})$,
  then $a^2 + b^2 = 1$ and
  \[
   \mathcal{F}
   \subseteq H \begin{bmatrix} 1 & 1 \\ i & -i \end{bmatrix} \mathscr{M}
   = \begin{bmatrix} 1 & 1 \\ i & -i \end{bmatrix} \begin{bmatrix} a + b i & 0 \\ 0 & a - b i \end{bmatrix} \mathscr{M}
   = \begin{bmatrix} 1 & 1 \\ i & -i \end{bmatrix} \mathscr{M}
  \]
  since $H \trans{1}{1}{i}{-i} = \trans{1}{1}{i}{-i} \trans{a+bi}{0}{0}{a-bi}$ and $\trans{a+bi}{0}{0}{a-bi} \in \Stab{M}$.
  Otherwise,
  $H = \trans{a}{b}{b}{-a} \in \mathbf{O}_2(\mathbb{C}) - \mathbf{SO}_2(\mathbb{C})$,
  so $a^2 + b^2 = 1$ and
  \[
   \mathcal{F}
   \subseteq H \begin{bmatrix} 1 & 1 \\ i & -i \end{bmatrix} \mathscr{M}
   = \begin{bmatrix} 1 & 1 \\ i & -i \end{bmatrix} \begin{bmatrix} 0 & a - b i \\ a + b i & 0 \end{bmatrix} \mathscr{M}
   = \begin{bmatrix} 1 & 1 \\ i & -i \end{bmatrix} \mathscr{M}
  \]
  since $H \trans{1}{1}{i}{-i} = \trans{1}{1}{i}{-i} \trans{0}{a-bi}{a+bi}{0}$ and $\trans{0}{a-bi}{a+bi}{0} \in \Stab{M}$.
  \qedhere
 \end{enumerate}
\end{proof}

We use four sets to characterize the $\mathscr{M}$-transformable signatures.
The notation $\Sym$ is from Definition~\ref{def:sym}.

\begin{definition} \label{def:trans:M}
 A symmetric signature $f$ of arity $n$ is in, respectively, $\mathscr{M}_1$, or $\mathscr{M}_2$, or $\mathscr{M}_3$, or $\mathscr{M}_4$
 if there exist an $H \in \mathbf{O}_2(\mathbb{C})$ and nonzero constants $c, \gamma \in \mathbb{C}$ such that $f$ has the form, respectively,
 $c H^{\otimes n} \left(\left[\begin{smallmatrix} 1 \\  1 \end{smallmatrix}\right]^{\otimes n}
                \pm i^n \left[\begin{smallmatrix} 1 \\ -1 \end{smallmatrix}\right]^{\otimes n}\right)$, or
 $c H^{\otimes n} \left(\left[\begin{smallmatrix} 1 \\  \gamma \end{smallmatrix}\right]^{\otimes n}
                    \pm \left[\begin{smallmatrix} 1 \\ -\gamma \end{smallmatrix}\right]^{\otimes n}\right)$, or
 $c H^{\otimes n} \Sym_n^{n-1}(\left[\begin{smallmatrix} 1 \\ 0 \end{smallmatrix}\right]; \left[\begin{smallmatrix} 0 \\  1 \end{smallmatrix}\right])$, or
 $c H^{\otimes n} \Sym_n^{n-1} (\left[\begin{smallmatrix} 1 \\ i \end{smallmatrix}\right];$ $\left[\begin{smallmatrix} 1 \\ -i \end{smallmatrix}\right])$.
\end{definition}

For $k \in \{1,2,3,4\}$,
when such an $H$ exists,
we say that $f \in \mathscr{M}_k$ with transformation $H$.
If $f \in \mathscr{M}_k$ with $I_2$,
then we say $f$ is in the canonical form of $\mathscr{M}_k$.

Notice that
$\{\left[\begin{smallmatrix} 1 \\  i \end{smallmatrix}\right],
   \left[\begin{smallmatrix} 1 \\ -i \end{smallmatrix}\right]\}$
is set-wise invariant under any transformation in $\mathbf{O}_2(\mathbb{C})$ up to nonzero constants.
Using this fact,
the following lemma gives a characterization of $\mathscr{M}_4$.
It says that any signature in $\mathscr{M}_4$ is essentially in canonical form.

\begin{lemma} \label{lem:trans:M4}
 Let $f$ be a symmetric signature of arity $n$.
 Then $f \in \mathscr{M}_4$ if and only if
 $f = c \Sym_n^{n-1}(\left[\begin{smallmatrix} 1 \\  i \end{smallmatrix}\right]; \left[\begin{smallmatrix} 1 \\ -i \end{smallmatrix}\right])$ or
 $f = c \Sym_n^{n-1}(\left[\begin{smallmatrix} 1 \\ -i \end{smallmatrix}\right]; \left[\begin{smallmatrix} 1 \\  i \end{smallmatrix}\right])$
 for some nonzero constant $c \in \mathbb{C}$.
\end{lemma}

\begin{proof}
 Suppose $f \in \mathscr{M}_4$,
 so that $f = c H^{\otimes n} \Sym_n^{n-1}(\left[\begin{smallmatrix} 1 \\ i \end{smallmatrix}\right]; \left[\begin{smallmatrix} 1 \\ -i \end{smallmatrix}\right])$.
 If $H \in \mathbf{SO}_2(\mathbb{C})$,
 then $H = \left[\begin{smallmatrix} a & b \\ -b & a\end{smallmatrix}\right]$ for some $a,b \in \mathbb{C}$ such that $a^2 + b^2 = 1$.
 Since $H \left[\begin{smallmatrix} 1 \\ i \end{smallmatrix}\right] = (a + b i) \left[\begin{smallmatrix} 1 \\ i \end{smallmatrix}\right]$
 and $H \left[\begin{smallmatrix} 1 \\ -i \end{smallmatrix}\right] = (a - b i) \left[\begin{smallmatrix} 1 \\ -i \end{smallmatrix}\right]$,
 it follows that
 $f = c (a + b i)^{n-1} (a - b i) \Sym_n^{n-1}(\left[\begin{smallmatrix} 1 \\ i \end{smallmatrix}\right]; \left[\begin{smallmatrix} 1 \\ -i \end{smallmatrix}\right])$.
 Otherwise,
 $H \in \mathbf{O}_2(\mathbb{C}) - \mathbf{SO}_2(\mathbb{C})$,
 so $H = \left[\begin{smallmatrix} a & b \\ b & -a\end{smallmatrix}\right]$ for some $a,b \in \mathbb{C}$ such that $a^2 + b^2 = 1$.
 Then
 $f = c (a + b i) (a - b i)^{n-1} \Sym_n^{n-1}(\left[\begin{smallmatrix} 1 \\ -i \end{smallmatrix}\right]; \left[\begin{smallmatrix} 1 \\ i \end{smallmatrix}\right])$.
 
 Now suppose
 $f = c \Sym_n^{n-1}(\left[\begin{smallmatrix} 1 \\ i \end{smallmatrix}\right]; \left[\begin{smallmatrix} 1 \\ -i \end{smallmatrix}\right])$
 or $f = c \Sym_n^{n-1}(\left[\begin{smallmatrix} 1 \\ -i \end{smallmatrix}\right]; \left[\begin{smallmatrix} 1 \\ i \end{smallmatrix}\right])$.
 The first case is already in the standard form of $\mathscr{M}_4$.
 In the second case,
 we pick $H = \left[\begin{smallmatrix} 1 & 0 \\ 0 & -1 \end{smallmatrix}\right] \in \mathbf{O}_2(\mathbb{C})$.
 Then $H^{\otimes n} f$ is in the standard form of $\mathscr{M}_4$.
\end{proof}

We further split $\mathscr{M}_4$ into $\mathscr{M}_4^\pm$ for future use.
Define $\mathscr{M}_4^\pm=\{f|f=c \Sym_n^{n-1}(\left[\begin{smallmatrix} 1 \\ \pm i \end{smallmatrix}\right];
\left[\begin{smallmatrix} 1 \\ \mp i \end{smallmatrix}\right])\}$.
In other words, $\mathscr{M}_4^+$ contains signatures of the form $Z^{\otimes n}[0,1,0,\dots,0]$
and $\mathscr{M}_4^-$ contains signatures of the form $Z^{\otimes n}[0,\dots,0,1,0]$ up to a scalar,
where $Z = \trans{1}{1}{i}{-i}$.
We will denote $[0,1,0,\dots,0]$ of arity $k$ by \exactone{k},
and $[0,\dots,0,1,0]$ of arity $k$ by \allbutone{k}.
Note that these are precisely the \textsc{Perfect Matching}
signatures and corresponding reversals.

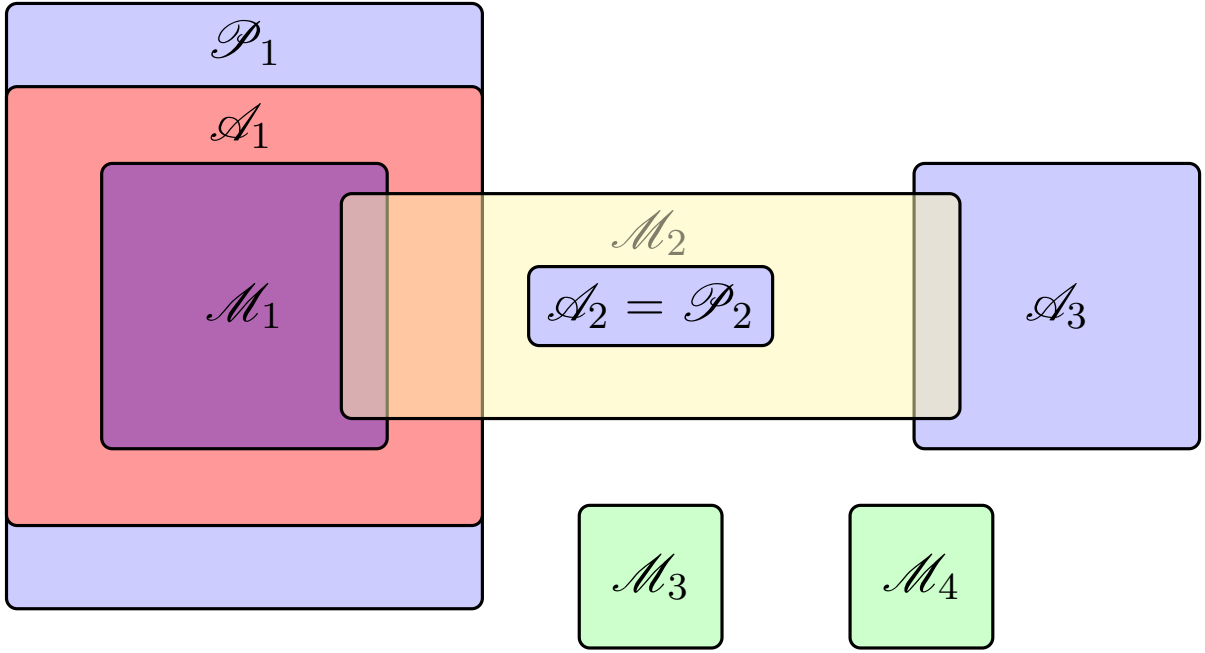
\begin{figure}[t]
 \centering
 \begin{tikzpicture}[scale=1.8, transform shape]
  \def\P1Height{120pt}
  \node[draw, rectangle, rounded corners, very thick, fill=blue!20, fill opacity=1, text opacity=1, minimum height=\P1Height, minimum width=100pt]
  at (0,0)
  (P1)
  {
   \begin{minipage}[t][\P1Height]{16pt}
    $\mathscr{P}_1$
   \end{minipage}
  };
  
  \def\A1Height{85pt}
  \node[draw, rectangle, rounded corners, very thick, fill=red!40, fill opacity=1, text opacity=1, minimum height=\A1Height, minimum width=100pt]
  at (0,0)
  (A1)
  {
   \begin{minipage}[t][\A1Height]{15pt}
    $\mathscr{A}_1$
   \end{minipage}
  };
  
  \node[draw, rectangle, rounded corners, very thick, fill=violet!60, fill opacity=1, text opacity=1, minimum height=60pt, minimum width=60pt]
  at (0,0)
  (A1)
  {$\mathscr{M}_1$};
  
  \node[draw, rectangle, rounded corners, very thick, fill=blue!20, fill opacity=1, text opacity=1, minimum height=60pt, minimum width=60pt]
  at (6,0)
  (A1)
  {$\mathscr{A}_3$};
  
  \def\M2Height{40pt}
  \node[draw, rectangle, rounded corners, very thick, fill=yellow!40, fill opacity=0.5, text opacity=1, minimum height=\M2Height, minimum width=130pt]
  at (3,0)
  (A1)
  {
   \begin{minipage}[t][\M2Height]{17pt}
    $\mathscr{M}_2$
   \end{minipage}
  };
  
  \node[draw, rectangle, rounded corners, very thick, fill=blue!20, fill opacity=1, text opacity=1, minimum height=15pt, minimum width=50pt]
  at (3,0)
  (A2P2)
  {$\mathscr{A}_2 = \mathscr{P}_2$};
  
  \node[draw, rectangle, rounded corners, very thick, fill=green!20, fill opacity=1, text opacity=1, minimum height=30pt, minimum width=30pt]
  at (3,-2)
  (A1)
  {$\mathscr{M}_3$};
  
  \node[draw, rectangle, rounded corners, very thick, fill=green!20, fill opacity=1, text opacity=1, minimum height=30pt, minimum width=30pt]
  at (5,-2)
  (A1)
  {$\mathscr{M}_4$};
 \end{tikzpicture}
 \caption{Relationships among $\mathscr{A}_1$, $\mathscr{A}_2$, $\mathscr{A}_3$, $\mathscr{P}_1$, $\mathscr{P}_2$, $\mathscr{M}_1$, $\mathscr{M}_2$, $\mathscr{M}_3$, and $\mathscr{M}_4$.
 Note that $\mathscr{P}_1 \intersect \mathscr{M}_2 \subseteq \mathscr{A}_1$.}
 \label{fig:Mtrans:venn_diagram}
\end{figure}

Notice that $\mathscr{M}_1 \subset \mathscr{A}_1 \subset \mathscr{P}_1$ and $\mathscr{A}_2 = \mathscr{P}_2 \subset \mathscr{M}_2$.
See Figure~\ref{fig:Mtrans:venn_diagram} for a visual description of the relationships among sets.

Next we show that $\mathscr{M}_k$ for $k=1,2,3,4$ captures all $\mathscr{M}$-transformable signatures.

\begin{lemma} \label{lem:cha:Mtrans}
 Let $f$ be a non-degenerate symmetric signature.
 Then $f$ is $\mathscr{M}$-transformable if and only if $f \in \mathscr{M}_1 \union \mathscr{M}_2 \union \mathscr{M}_3 \union \mathscr{M}_4$.
\end{lemma}

\begin{proof}
 Assume that $f$ is $\mathscr{M}$-transformable of arity $n$.
 By applying Lemma~\ref{lem:match:trans} to $\{f\}$,
 we have $f \in \left[\begin{smallmatrix} 1 & 1 \\ i & -i \end{smallmatrix}\right] \mathscr{M}$ or
 there exists an $H \in \mathbf{SO}_2(\mathbb{C})$ such that $f \in H \mathscr{M}$.
 Proposition~\ref{prop:match:symmetric:decomp} lists the symmetric signatures in $\mathscr{M}$.
 Since we are only interested in non-degenerate signatures,
 we only consider $a$, $b$, and $\lambda$ that are nonzero.
 Now we consider the possible cases.
 \begin{enumerate}
  \item Suppose $f \in \trans{1}{1}{i}{-i} \mathscr{M}$.
  \begin{itemize}
   \item Further suppose
   $f = \trans{1}{1}{i}{-i}^{\otimes n} \left(\tbcolvec{a}{b}^{\otimes n} \pm \tbcolvec{a}{-b}^{\otimes n}\right)$
   for some nonzero $a,b \in \mathbb{C}$.
   Let $T = \frac{1-i}{2} \trans{u}{v}{v}{-u}$,
   where $u = a + b i$ and $v = i (a - b i)$.
   Then $f = T^{\otimes n} \left(\tbcolvec{1}{1}^{\otimes n} \pm i^n \tbcolvec{1}{-1}^{\otimes n}\right)$.
   Since $T \in \mathbf{O}_2(\mathbb{C})$ up to a nonzero factor of $\sqrt{2 a b}$,
   we have $f \in \mathscr{M}_1$.
   
   \item Further suppose
   $f = \lambda \trans{1}{1}{i}{-i}^{\otimes n} \Sym_n^{n-1}(\tbcolvec{1}{0}; \tbcolvec{0}{1})$
   for some nonzero $\lambda \in \mathbb{C}$.
   Then we have $f = \lambda \Sym_n^{n-1}(\tbcolvec{1}{i}; \tbcolvec{1}{-i})$,
   so $f \in \mathscr{M}_4$.
   
   \item Further suppose
   $f = \lambda \trans{1}{1}{i}{-i}^{\otimes n} \Sym_n^{n-1}(\tbcolvec{0}{1}; \tbcolvec{1}{0})$
   for some nonzero $\lambda \in \mathbb{C}$.
   Then we have $f = \lambda \Sym_n^{n-1}(\tbcolvec{1}{-i}; \tbcolvec{1}{i})$,
   so $f \in \mathscr{M}_4$ by Lemma~\ref{lem:trans:M4}.
  \end{itemize}
  
  \item Suppose $f \in H \mathscr{M}$.
  \begin{itemize}
   \item Further suppose
   $f = H^{\otimes n} \left(\tbcolvec{a}{b}^{\otimes n} \pm \tbcolvec{a}{-b}^{\otimes n}\right)$
   for some nonzero $a,b \in \mathbb{C}$.
   Then we have
   $f = a^n H^{\otimes n} \left(\tbcolvec{1}{\gamma}^{\otimes n} \pm \tbcolvec{1}{-\gamma}^{\otimes n}\right)$,
   where $\gamma = \frac{b}{a}$,
   so $f \in \mathscr{M}_2$.
   
   \item Further suppose
   $f = \lambda H^{\otimes n} \Sym_n^{n-1}(\tbcolvec{1}{0}; \tbcolvec{0}{1})$
   for some nonzero $\lambda \in \mathbb{C}$.
   Then $f \in \mathscr{M}_3$.
   
   \item Further suppose
   $f = \lambda H^{\otimes n} \Sym_n^{n-1}(\tbcolvec{0}{1}; \tbcolvec{1}{0})$
   for some nonzero $\lambda \in \mathbb{C}$.
   Let $H' = H \trans{0}{1}{1}{0} \in \mathbf{O}_2(\mathbb{C})$.
   Then we have $f = \lambda H'^{\otimes n} \Sym_n^{n-1}(\tbcolvec{1}{0}; \tbcolvec{0}{1})$,
   so $f \in \mathscr{M}_3$.
  \end{itemize}
 \end{enumerate}

 Conversely, if there exists a matrix $H \in \mathbf{O}_2(\mathbb{C})$ such that
 $H^{\otimes n} f$ is in one of the canonical forms of $\mathscr{M}_1$, $\mathscr{M}_2$, $\mathscr{M}_3$, or $\mathscr{M}_4$,
 then one can directly check that $f$ is $\mathscr{M}$-transformable by Definition \ref{def:prelim:trans}.
 In fact, the transformations that we applied above are all invertible,
 except for $\mathscr{M}_1$,
 if the given orthogonal transformation is of the form $\trans{u}{-v}{v}{u}$,
 we do $\trans{1}{0}{0}{-1}$ first followed by $\trans{u}{v}{v}{-u}$.
\end{proof}

Furthermore, we show that a nontrivial signature $f$
in the set $\mathscr{M}_3$ is not $\mathscr{A}$- or $\mathscr{P}$-transformable.
Moreover, the only transformation to make $f$ in $\mathscr{M}$ is very restricted.
This is for future use.

\begin{lemma} \label{lem:M3:transformation}
  Let $f\in\mathscr{M}_3$ be a non-degenerate signature of arity $n \ge 3$
  with $H\in\mathbf{O}_2(\mathbb{C})$.
  Then $f$ is not $\mathscr{A}$- or $\mathscr{P}$-transformable.
  Moreover, $f$ is $\mathscr{M}$-transformable with only $HD$ or $H\trans{0}{1}{1}{0}D$
  for some diagonal matrix $D$.
\end{lemma}

\begin{proof}
  Suppose $f=[f_0,f_1,\dots,f_n]$.
  If $f$ is $\mathscr{A}$- or $\mathscr{P}$-transformable,
  then $f$ has to satisfy a second order recurrence relation
  that $a f_i+ bf_{i+1} +c f_{i+2}=0$,
  for $a,b,c\in\mathbb{C}$ such that not all $a,b,c$ are $0$ and $b^2-4ac\neq 0$.
  In other words, the second order recurrence relation has to have distinct eigenvalues.
  This is due to Lemma 6.15 or Lemma 7.2 in \cite{CGW14a}.
  Moreover, this property is preserved by holographic transformations (cf.\ Lemma 6.2 in \cite{CGW14a}).
  However, $f$ is in $\mathscr{M}_3$.
  Hence $f=H^{\otimes n}\exactone{n}$ for some $H\in\mathbf{O}_2(\mathbb{C})$
  up to a nonzer factor.
  On the other hand,
  $\exactone{n}$ does not satisfy a second recurrence with distinct eigenvalues if $n \ge 3$,
  a contradiction.

  Moreover, notice that the only signatures in $\mathscr{M}$
  that do not satisfy such second order recurrence relations
  are \exactone{k} and \allbutone{k} functions.
  If $f$ is $\mathscr{M}$-transformable,
  then there exists a transformation $T$ such that $f=T^{\otimes n}g$ for some $g\in\mathscr{M}$
  and $[1,0,1]T^{\otimes 2}\in\mathscr{M}$.
  Hence $g=\exactone{n}$ or $\allbutone{n}$.
  On the other hand $f=H^{\otimes n}\exactone{n}$ up to a nonzer factor.
  Therefore $(T^{-1}H)^{\otimes n}\exactone{n}=\exactone{n}$ or $\allbutone{n}$ up to a nonzer factor.

  Let $J = T^{-1} H = \trans{x}{y}{z}{w}$ and let $h = J^{\otimes n} \exactone{n}$.  
  As $\exactone{n}=\Sym_n^{n-1}(\tbcolvec{1}{0}; \tbcolvec{0}{1})$,
  $h=(\trans{x}{y}{z}{w})^{\otimes n}\exactone{n}=\Sym_n^{n-1}(\tbcolvec{x}{z}; \tbcolvec{y}{w})$.
  The first and last entries of $h$ are $x^{n-1}y$ and $z^{n-1}w$.
  As $h=\exactone{n}$ or $\allbutone{n}$,
  we have that $x^{n-1}y=z^{n-1}w=0$.
  It is easy to see that $x$ and $z$, or $y$ and $w$ cannot be both $0$.
  Then $x=w=0$ or $y=z=0$.
  This implies that $J=D$ or $J=D\trans{0}{1}{1}{0}$
  for some diagonal matrix $D$.
  Thus $T=HJ^{-1}=HD^{-1}$ or $H\trans{0}{1}{1}{0} D^{-1}$.
\end{proof}

Let $g = [x,y,0, \dotsc, 0,z]$ have arity $n \ge 3$,
where $xyz \neq 0$.
As an example of the theory developed in this section,
we discuss the signature $Z^{\otimes n} g$
%$Zg = Z^{\otimes n}[x,y,0, \dotsc, 0,z]$
in the following lemma, which will be used in
Lemma~\ref{lem:a1000b} in  the proof of the single signature
dichotomy  Theorem~\ref{thm:dic:single}.

\begin{lemma} \label{lem:Zg:not-transformable}
Let $n \ge 3$, $g = [x,y,0, \dotsc, 0,z]$ have arity $n$
and $xyz \neq 0$.
Then the signature $Z^{\otimes n}g$
is neither $\mathscr{A}$-, $\mathscr{P}$-, $\mathscr{M}$-transformable,
nor vanishing.
\end{lemma}

\begin{remark}
By Theorem~\ref{thm:PlHolant:arity34}, for arity $n = 3$ or $4$,
Lemma~\ref{lem:Zg:not-transformable} implies that $\PlHolant(Z^{\otimes n}g)$
is $\SHARPP$-hard. After we have proved Theorem~\ref{thm:dic:single},
this lemma will imply that $\PlHolant(Z^{\otimes n}g)$
is $\SHARPP$-hard for all $n \ge 3$.
\end{remark}

\begin{proof}
That $Z^{\otimes n}g$ is not vanishing follows from Lemma~\ref{lem:prelim:vanishing_form_in_Z_basis} combined with Corollary~\ref{cor:van:degree}
and Theorem~\ref{thm:van}.
To show that $Z^{\otimes n}g$ is not 
$\mathscr{A}$-, $\mathscr{P}$-, $\mathscr{M}$-transformable,
we only need to show that $Z^{\otimes n}g \not \in 
\mathscr{P}_1 \cup \mathscr{M}_2 \cup \mathscr{A}_3 \cup \mathscr{M}_3 \cup \mathscr{M}_4$
by Lemma~\ref{lem:cha:affine},~\ref{lem:cha:product} and \ref{lem:cha:Mtrans},
and the fact that $\mathscr{M}_1 \subset \mathscr{A}_1 \subset \mathscr{P}_1$ and $\mathscr{A}_2 = \mathscr{P}_2 \subset \mathscr{M}_2$. 
See Figure~\ref{fig:Mtrans:venn_diagram}.

We first show that $Z^{\otimes n}g \not \in
\mathscr{P}_1 \cup \mathscr{M}_2 \cup \mathscr{A}_3$.
We say a signature $f=[f_0, f_1, \dotsc,f_n]$
satisfies a second order recurrence of type $\langle a,b,c \rangle$
if $af_k-bf_{k+1}+cf_{k+2}=0$  for $1\leq k\leq n-2$,
for some $a, b$ and $c$ not all zero.
Suppose $Z^{\otimes n}g$ is a nonzer constant multiple of 
$H f \in \mathscr{P}_1 \cup \mathscr{M}_2 \cup \mathscr{A}_3$
in the forms given in Definitions~\ref{definition:P1}, \ref{def:trans:M} 
and \ref{def:A123}, then
 $f$, and hence also $(Z^{-1})^{\otimes n}f$,
satisfies a second order recurrence.
%%% do n't use that "with distinct eigenvalues"
%%% with distinct eigenvalues.
%%% also would allow \infty in M_2 case. actually P_2 case.
%%% so omit that mention.
We have $H^{-1}Z = Z D$ or $Z D \trans{0}{1}{1}{0}$ 
 for some non-singular diagonal $D$ since $H 
\in \mathbf{O}_2(\mathbb{C})$.
Thus $f  = Z^{\otimes n} g'$ for some $g' =[x',y',0, \dotsc, 0,z']$
or $[x',0, \dotsc, 0,y',z']$,
with $x'y'z' \not =0$.
We assume the former; the proof is similar for the latter.
%%% could be [1 1\\-i i] [a 0\\0 b] 
%% = [1 1\\i -i] [0 1\\1 0] [a 0\\0 b] = .. = Z [b 0 \\0 a] [0 1\\1 0]

However, for $n \ge 4$, $g'$
% =[x',y',0, \dotsc, 0,z']$ 
does not satisfy any second order recurrence.
For a contradiction suppose $g'$ does.
By $x'y'z' \not =0$,  $ay' - b 0 + c 0 =0$ gives $a=0$,
$a x' - b y' + c 0 =0$ gives $b=0$,
and $a 0 - b 0 + c z' =0$ gives $c=0$; but $a, b,c$ cannot be all zero.

Next suppose $n=3$, and we show that $g' = (Z^{-1})^{\otimes n} f$
is still impossible.
For $\mathscr{P}_1$, $f = 
\tbcolvec{1}{ 1}^{\otimes 3}
                    + \beta \tbcolvec{1}{-1}^{\otimes 3}$.
It is easy to check that $(Z^{-1})^{\otimes n}f$ satisfies a second order recurrence
with its two eigenvalues sum to zero. However
$g' = [x', y', 0, z']$ has type $\langle y'z', x'z', -y'^2 \rangle$,
the sum of its two eigenvalues is $-x'z'/y'^2 \not =0$.

For $\mathscr{M}_2$, $f = \left[\begin{smallmatrix} 1 \\  \gamma \end{smallmatrix}\right]^{\otimes 3}
                    \pm \left[\begin{smallmatrix} 1 \\ -\gamma \end{smallmatrix}\right]^{\otimes 3}$.
In $(Z^{-1})^{\otimes n} f$,  $Z^{-1} \trans{1}{1}{\gamma}{-\gamma}$ has the form
$\trans{u}{v}{v}{u}$,
and $(Z^{-1})^{\otimes n} f = \left[\begin{smallmatrix} u \\  v \end{smallmatrix}\right]^{\otimes 3}
                    \pm \left[\begin{smallmatrix} v \\ u \end{smallmatrix}\right]^{\otimes 3}$.
 Thus the weight 1 and weight 2 entries of $(Z^{-1})^{\otimes n} f$ are either equal
or negative of each other.  If $g' = (Z^{-1})^{\otimes n} f$ this would imply
$y' =0$, a contradiction.

For $\mathscr{A}_3$, $f = \tbcolvec{1}{\alpha}^{\otimes n} +   i^r \tbcolvec{1}{-\alpha}^{\otimes n}$.
$Z^{-1} \trans{1}{1}{\alpha}{-\alpha}
= \trans{u}{v}{v}{u}$, with $u = 1 - \alpha i$ and $v = 1 + \alpha i$.
The weight 2 entry of $(Z^{-1})^{\otimes n} f$ is
$uv^2 +   i^r vu^2 = (uv) (v + i^r u)$. This is nonzer
for all $r$.
%%% directly ck: u v \not =0
%%% v + i^r u = (1+i^r) + \al i (1-i^r). 
%%% if r= 0, 2, it's 2 or  it's  2 *\al i
%%% if r \neq 0, 2, can div out (1+i^r) , it's 1 + \al i ( \pm i) 
However $g' = [x', y', 0, z']$ has this property.

It remains to show that $Z^{\otimes n}g \not \in \mathscr{M}_3 \cup \mathscr{M}_4$.
If $Z^{\otimes n}g \in \mathscr{M}_3$, then
$Z^{\otimes n} g = cH f$
for some $H \in \mathbf{O}_2(\mathbb{C})$
and $f = \Sym_n^{n-1}(\left[\begin{smallmatrix} 1 \\ 0 \end{smallmatrix}\right]; \left[\begin{smallmatrix} 0 \\  1 \end{smallmatrix}\right])$.
Again $f = (cH)^{-1} Z^{\otimes n} g = Z^{\otimes n} g'$ for some $g'$ 
having the same or its reversal form
as $g$. Then $g' = (Z^{-1})^{\otimes n} f$ is the signature
$[n, n-2,  \dotsc, -(n-2), -n]$.
%If $n$ is odd, there are no zero entries, while $g'$ has.
%If $n$ is even, then 
The weight 1 entry and weight $n-1$ entry
have the same absolute value. By the form of $g'$ this is a contradiction.

Finally if $Z^{\otimes n}g \in \mathscr{M}_4$, then by Lemma~\ref{lem:trans:M4},
$Z^{\otimes n} g = c Z^{\otimes n}f$, for 
some nonzero constant $c \in \mathbb{C}$, and
$f = \Sym_n^{n-1}(\left[\begin{smallmatrix} 1 \\  0 \end{smallmatrix}\right]; \left[\begin{smallmatrix} 0 \\ 1 \end{smallmatrix}\right])$ or
its reversal
 $\Sym_n^{n-1}(\left[\begin{smallmatrix} 0 \\ 1 \end{smallmatrix}\right]; \left[\begin{smallmatrix} 1 \\  0 \end{smallmatrix}\right])$.
In either case, after canceling out $Z$, the weight 0 entry is
$0$ in the expression but not so in $g$; a contradiction.
\end{proof}

\section{Mixing with Vanishing Signatures} \label{sec:mixing:vanishing}

In this section,
we prove some hardness results for vanishing signature sets when augmented by other signatures.
We first consider the mixing of vanishing signatures with unary and binary signatures.
Over general graphs,
these cases are handled by Lemma~7.1 and Lemma~7.2 in~\cite{CGW13}.
One can check that the hardness in Lemma~7.1 in~\cite{CGW13} holds for planar graphs.
We state the planar version of Lemma~7.1 in~\cite{CGW13}
and provide a proof for completeness.
Specifically,
the reduction to obtain the signature $f''$ is planar 
and $\PlHolant(f'')$ is \numP-hard by Theorem~\ref{thm:PlHolant:arity34}.
%More directly,
%we now know that counting Eulerian orientations is $\SHARPP$-hard over planar $4$-regular graphs~\cite[Theorem~3.7]{GW13}.

\begin{lemma} \label{lem:van:deg}
 Let $f \in \mathscr{V}^\sigma$ be a symmetric signature of arity $n$ with $\rd^\sigma(f)=d\ge 2$ where $\sigma \in \{+,-\}$.
 Suppose $v = u^{\otimes m}$ is a symmetric degenerate signature for some unary signature $u$ and some integer $m \ge 1$.
 If $u$ is not a multiple of $[1, \sigma i]$,
 then $\PlHolant(f,v)$ is \numP-hard.
\end{lemma}

\begin{proof}
  We consider $\sigma = +$ since the other case is similar.
  Since $f \in \mathscr{V}^+$,
  we have $n > 2d \ge 4$.
  Under a holographic transformation by $Z$,
  we have
  \[
   \PlHolant(f,v) \equiv \plholant{{\neq}_2}{\hat{f},[a,b]^{\otimes m}},
  \]
  where $\hat{f} = \left(Z^{-1}\right)^{\otimes n} f$ and $[a,b]^{\otimes m} = \left(Z^{-1}\right)^{\otimes m} v$ with $b \neq 0$
  since $u$ is not a multiple of $[1,i]$.
  Moreover,
  $\hat{f} = [\hat{f}_0, \hat{f}_1, \dotsc, \hat{f}_d, 0, \dotsc, 0]$ with 
  $\hat{f}_d\neq 0$ by Lemma~\ref{lem:prelim:vanishing_form_in_Z_basis}.

  We get $\widehat{f'} = [\hat{f}_{d-2}, \hat{f}_{d-1}, \hat{f}_d, 0, \dotsc, 0]$ of arity $n-2d+4$
  by $d-2$ self-loops via $\neq_2$ on $\hat{f}$.
  This is on the right side.
  With two more self-loops,
  we get $[1,0]^{\otimes n-2d}$,
  also on the right.

  We claim that we can use $[1,0]^{\otimes n-2d}$ and $[a,b]^{\otimes m}$
  to create $[a,b]^{\otimes n-2d}$.
  Let $t=\gcd(m,n-2d)$.
  If $n-2d>m$,
  then we connect $[a,b]^{\otimes m}$ to $[1,0]^{\otimes n-2d}$ via $\neq_2$
  to get $[1,0]^{\otimes n-2d-m}$ up to a nonzero factor $b \neq 0$.
  We repeat this process until we get a tensor power $[1,0]^{\otimes \ell}$ for some $\ell \le m$.
  We can do a similar construction if $m>n-2d$.
  Repeat this process, which is a subtractive Euclidean algorithm.
  Halt upon getting both $[1,0]^{\otimes t}$ and $[a,b]^{\otimes t}$.
  Then we combine $\tfrac{n-2d}{t}$ copies of $[a,b]^{\otimes t}$ to get $[a,b]^{\otimes n-2d}$.

  Now connecting $[a,b]^{\otimes n-2d}$ back to $\widehat{f'}$ via $\neq_2$,
  gives $\widehat{f''}=[\widehat{f''}_0, \widehat{f''}_1, \widehat{f''}_2, 0, 0]$ of arity~$4$.
  Moreover,
  $\widehat{f''}_2 = b^{n-2d} \hat{f}_d \neq 0$.
  Notice that $\PlHolant({\neq}_2 \mid [\widehat{f''}_{0}, \widehat{f''}_{1}, \widehat{f''}_{2}, 0, 0]) \equiv \plholant{{\neq}_2}{[0,0,1,0,0]}$,
  the Eulerian Orientation problem over planar $4$-regular graphs,
  which is \numP-hard by Corollary~\ref{cor:prelim:nonsingular_compressed_hard_trans}
  (or more directly by~\cite[Theorem~3.7]{GW13}).
  Thus,
  $\PlHolant(f,v)$ is \numP-hard.
\end{proof}

Next come binary signatures.
The statement of Lemma~7.2 in~\cite{CGW13} must be modified to rule out a planar tractable case
(which is proved \numP-hard for general graphs in Lemma~7.2 in~\cite{CGW13}).
Excluding this planar tractable case, 
there is one more nonplanar reduction in the proof of Lemma~7.2 in~\cite{CGW13}.
This reduction is used to show that $\holant{{\ne}_2}{\{[t,1,0,0,0], [c,0,1]\}}$ is \numP-hard 
when $c \ne 0$ (since the gadget in Figure~12a of~\cite{CGW13} is nonplanar).
In the following lemma, we first show that this problem
$\holant{{\ne}_2}{\{[t,1,0,0,0], [c,0,1]\}}$ 
remains \numP-hard even restricted to planar graphs provided $t \ne 0$.
If $t=0$, then all signatures belong to $\mathscr{M}$ and the problem is tractable.

\begin{lemma} \label{lem:van:binary:planar_fix}
 Let $c, t \in \mathbb{C}$.
 If $c t \ne 0$,
 then $\plholant{{\ne}_2}{[t,1,0,0,0], [c,0,1]}$ is \numP-hard.
\end{lemma}

\begin{proof}
 By connecting two copies of $\ne_2$ to either side of $[c,0,1]$,
 we get the signature $[1,0,c]$ on the left.
 Clearly $\plholant{[1,0,c]}{[t,1,0,0,0]} \le_T \plholant{{\ne}_2}{[t,1,0,0,0], [c,0,1]}$.
 Then under a holographic transformation by $T^{-1}$,
 where $T = \left[\begin{smallmatrix} 1 & 0 \\ 0 & \sqrt{c} \end{smallmatrix}\right]$,
 we have
 \begin{align*}
  \plholant{[1,0,c]}{[t,1,0,0,0]}
  &\equiv \plholant{[1,0,c] (T^{-1})^{\otimes 2}}{T^{\otimes 4} [t,1,0,0,0]}\\
  &\equiv \plholant{[1,0,1]}{[t,\sqrt{c},0,0,0]}\\
  &\equiv \PlHolant([t,\sqrt{c},0,0,0]).
 \end{align*}
 The last problem is \numP-hard by Corollary~\ref{cor:arity4:double_root} after dividing by $\sqrt{c}$.
\end{proof}

Next we prove the planar version of Lemma~7.2 in~\cite{CGW13} 
using Lemma~\ref{lem:van:binary:planar_fix}.
We have to rule out the planar tractable case $f \in \mathscr{M}_4^\pm$.
Also note that if $f \in \mathscr{V}^\pm$ is a symmetric non-degenerate signature,
then $f$ has arity at least $3$.
This is because a unary signature is degenerate,
and if a binary symmetric signature $f$ is vanishing,
then its vanishing degree is greater than~$1$, 
hence at least~$2$, and therefore $f$ is also degenerate.
In the following lemma,
we explicitly state this condition $\arity(f) \ge 3$.

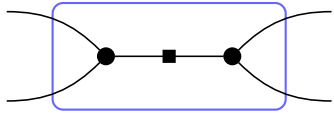
\begin{figure}[t]
 \centering
 \begin{tikzpicture}[scale=\scale,transform shape,node distance=\nodeDist,semithick]
  \node[internal]  (0)                    {};
  \node[external]  (1) [above left  of=0] {};
  \node[external]  (2) [below left  of=0] {};
  \node[external]  (3) [left        of=1] {};
  \node[external]  (4) [left        of=2] {};
  \node[square]    (5) [right       of=0] {};
  \node[internal]  (6) [right       of=5] {};
  \node[external]  (7) [above right of=6] {};
  \node[external]  (8) [below right of=6] {};
  \node[external]  (9) [right       of=7] {};
  \node[external] (10) [right       of=8] {};
  \path (0) edge[out= 135, in=   0]  (3)
            edge[out=-135, in=   0]  (4)
            edge                     (6)
        (6) edge[out=  45, in= 180]  (9)
            edge[out= -45, in= 180] (10);
  \begin{pgfonlayer}{background}
   \node[draw=\borderColor,thick,rounded corners,fit = (1) (2) (7) (8),inner sep=0pt] {};
%    \node[draw=\borderColor,thick,rounded corners,fit = (1) (2) (7) (8),inner sep=0pt,transform shape=false] {};
  \end{pgfonlayer}
 \end{tikzpicture}
 \caption{Circle vertices are assigned $[t,1,0,0]$ and the square vertex is assigned $\neq_2$.}
 \label{fig:gadget:arity3_to_arity4}
\end{figure}

\begin{lemma} \label{lem:van:bin}
 Let $f \in \mathscr{V}^\sigma$ be a symmetric non-degenerate signature of arity $n\ge 3$
 for some $\sigma \in \{+,-\}$.
 Suppose $h$ is a non-degenerate binary signature.
 If $f\not \in \mathscr{M}_4^\sigma$ and $h \notin \mathscr{R}^\sigma_2$,
 then $\PlHolant(f,h)$ is \numP-hard.
\end{lemma}

\begin{proof}
 We consider $\sigma = +$ since the other case is similar.
 Under a $Z$ transformation,
 \[
  \PlHolant(f,h)
  \equiv
  \plholant{{\neq}_2}{\hat{f},\hat{h}},
 \]
 where $\hat{f}=\left(Z^{-1}\right)^{\otimes n}f$ and $\hat{h}=\left(Z^{-1}\right)^{\otimes 2}h$.
 Since $h\not\in\mathscr{R}_2^+$,
 we may assume that $\hat{h}=[a,b,1]$ by Lemma~\ref{lem:prelim:vanishing_form_in_Z_basis} with a nonzero entry $\hat{h}_2$.
 Moreover since $h$ is non-degenerate,
 so is $\hat{h}$, and $b^2\neq a$.

 We prove the lemma by induction on the arity of $f$ (or equivalently $\hat{f}$).
 There are two base cases,
 $n = 3$ and $n = 4$.
 However, the arity $3$ case is easily reduced to the arity $4$ case.
 We show this first, 
 and then show that the lemma holds in the arity~$4$ case.

 Assume $n = 3$.
 Since $f \in \mathscr{V}^+$,
 we have $\hat{f}=[t,1,0,0]$ for some $t\neq 0$,
 by Lemma~\ref{lem:prelim:vanishing_form_in_Z_basis} and $f\not\in\mathscr{M}_4^+$.
 Consider the gadget in Figure~\ref{fig:gadget:arity3_to_arity4}.
 We assign $\hat{f}$ to the circle vertices and $\neq_2$ to the square vertex.
 Let $\hat{f}'$ be the signature of the resulting gadget.
 The signature $\hat{f}'$ may not seem symmetric by construction,
 but it is not hard to verify that indeed $\hat{f}'=[2t,1,0,0,0]$.
 The crucial observation is that it takes the same value~$0$ on inputs $1010$ and $1100$,
 where bits are ordered counterclockwise,
 starting from an arbitrary edge.
 This finishes our reduction to $n=4$.

 Now we consider the base case of $n = 4$.
 Since $f \in \mathscr{V}^+$,
 we have $\vd^+(f) > 2$ and $\rd^+(f) < 2$.
 As $f$ is not degenerate, $\rd^+(f)\not\in\{-1,0\}$.
 It implies that $\rd^+(f)=1$ and 
 by Lemma~\ref{lem:prelim:vanishing_form_in_Z_basis}, $\hat{f} = [t,1,0,0,0]$.

 Our next goal is to show that we can realize a signature of the form $[c,0,1]$ with $c \neq 0$.
 Then $\plholant{{\neq}_2}{[t,1,0,0,0],[c,0,1]}\le\PlHolant(f,h)$.
 Moreover, $t\neq 0$ since $f\notin\mathscr{M}_4^+$.
 Then by Lemma~\ref{lem:van:binary:planar_fix},
 $\plholant{{\neq}_2}{[t,1,0,0,0],[c,0,1]}$ is \numP-hard.

 If $b=0$,
 then $\hat{h}$ is what we want since in this case $a = a - b^2 \ne 0$.
 
 Otherwise $b \neq 0$.
 By connecting $\hat{h}$ to $\hat{f}$ via $\neq_2$,
 we get $[t + 2 b, 1, 0]$.
%%% direct check
%%% F_0 = f0 g2 + 2 f1 g1 + f2 g0. note f2 = 0...
%%% F_1 = f1 g2 + 2 f2 g1 + f3 g0. note f2 = f3 = 0. etc.
 If $t \neq -2b$,
 then by Lemma~\ref{lem:simple_interpolation:van:bin},
 we can interpolate any binary signature of the form $[v,1,0]$.
 Otherwise $t = - 2 b$.
 Then we connect two copies of $\hat{h}$ via $\neq_2$,
 and get $\widehat{h}' = [2 a b, a + b^2, 2 b]$.
%%% direct check
 By connecting this $\widehat{h}'$ to $\hat{f}$ via $\neq_2$,
 we get $[2 (a - b^2), 2 b, 0]$,
 using $t = -2 b$.
%%% direct check
 Since $a \neq b^2$ and $b \neq 0$,
 we can once again interpolate any $[v,1,0]$ by Lemma~\ref{lem:simple_interpolation:van:bin}.

\begin{figure}[t]
 \centering
 \begin{tikzpicture}[scale=\scale,transform shape,node distance=\nodeDist,semithick]
  \node[external] (0)                    {};
  \node[internal] (1) [right       of=0] {};
  \node[square]   (2) [right       of=1] {};
  \node[triangle] (3) [right       of=2] {};
  \node[square]   (4) [right       of=3] {};
  \node[internal] (5) [right       of=4] {};
  \node[external] (6) [right       of=5] {};
  \node[external] (7) [above right of=3] {};
  \node[external] (8) [below right of=3] {};
  \path (0) edge            node[near end]   (e1) {} (1)
        (1) edge                                     (2)
        (2) edge                                     (3)
        (3) edge                                     (4)
        (4) edge                                     (5)
        (5) edge            node[near start] (e2) {} (6);
  \path (3) edge[opacity=0] node[near start] (e3) {} (7)
            edge[opacity=0] node[near start] (e4) {} (8);
  \begin{pgfonlayer}{background}
   \node[draw=\borderColor,thick,rounded corners,fit = (e1) (e2) (e3) (e4)] {};
%    \node[draw=\borderColor,thick,rounded corners,fit = (e1) (e2) (e3) (e4),inner sep=2pt,transform shape=false] {};
  \end{pgfonlayer}
 \end{tikzpicture}
 \caption{A sequence of binary gadgets that forms another binary gadget.
 The circles are assigned $[v, 1, 0]$,
 the squares are assigned $\neq_2$,
 and the triangle is assigned $[a,b,1]$.}
 \label{fig:gadget:long_edge}
\end{figure}
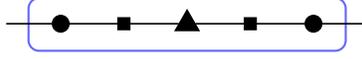

 Hence,
 we have the signature $[v,1,0]$,
 where $v \in \mathbb{C}$ is for us to choose.
 We construct the gadget in Figure~\ref{fig:gadget:long_edge} with the circles assigned $[v,1,0]$,
 the squares assigned $\neq_2$,
 and the triangle assigned $[a,b,1]$.
 The resulting gadget has signature $[a + 2 b v + v^2, b + v, 1]$,
 which can be verified by the matrix product
 \[
  \begin{bmatrix} v & 1 \\ 1 & 0 \end{bmatrix}
  \begin{bmatrix} 0 & 1 \\ 1 & 0 \end{bmatrix}
  \begin{bmatrix} a & b \\ b & 1 \end{bmatrix}
  \begin{bmatrix} 0 & 1 \\ 1 & 0 \end{bmatrix}
  \begin{bmatrix} v & 1 \\ 1 & 0 \end{bmatrix}
  =
  \begin{bmatrix} a + 2 b v + v^2 & b + v \\ b + v & 1 \end{bmatrix}.
 \]
 By setting $v = -b$,
 we get $[c,0,1]$,
 where $c = a - b^2 \neq 0$.

 Now we do the induction step.
 Assume $n \geq 5$.
 Since $f$ is non-degenerate, $\rd^+(f) \ge 1$.
 If $\rd^+(f) = 1$,
 then $\hat{f}=[t,1,0,\dots,0]$ for some $t\neq 0$.
 We connect $\hat{h}$ to $\hat{f}$ via $\neq_2$,
 getting $[t+2b,1,0,\dots,0]$ of arity $n-2\ge 3$.
 If $t+2b\neq 0$, then we are done by induction hypothesis.
 Otherwise $t=-2b$,
 and we connect two $\hat{h}$ together via $\neq_2$.
 The signature is $\hat{h}':=[2ab,b^2+a,2b]$.
 Connect $\hat{h}'$ to $\hat{f}$ via $\neq_2$.
 We get $[-4b^2+2(b^2+a),2b,0,\dots,0]=[2(a-b^2),2b,0,\dots,0]$.
 If $b=0$, then $t=0$. Contradiction.
 Hence $b\neq 0$, and $a-b^2\neq 0$ for $b$ is not degenerate.
 Then we can apply induction hypothesis on $[2(a-b^2),2b,0,\dots,0]$.

 The case left is that $\rd^+(f) = d \ge 2$.
 Then $\hat{f}=[\hat{f}_0,\hat{f}_1,\dots,\hat{f}_{d},0,\dots,0]$
 with $\hat{f}_d\neq 0$ by Lemma~\ref{lem:prelim:vanishing_form_in_Z_basis}.
 We do a self-loop of $\hat{f}$ via $\neq_2$,
 getting $\hat{f}'':=[\hat{f}_1,\dots,\hat{f}_{d},0,\dots,0]$ of arity $n-2\ge 3$.
 Since $d\ge 2$, $\hat{f}''$ is non-degenerate and $f''=Z^{\otimes (n-2)}\hat{f}''\in\mathscr{V}^+$.
 If $f''\not\in\mathscr{M}_4^+$,
 then apply the induction hypothesis and we are done.
 Otherwise $d=2$ and we may assume $\hat{f}=[\hat{f}_0,0,1,0,\dots,0]$
 since $\hat{f}_2\neq 0$.

 In this case, we connect $\hat{h}$ to $\hat{f}$ via $\neq_2$,
 getting $\hat{f}''':=[a+\hat{f}_0,2b,1,0,\dots,0]$ of arity $n-2\ge 3$.
 If $n\ge 7$, then we can apply the induction hypothesis.
 If $n=6$, then $\hat{f}'''=[a+\hat{f}_0,2b,1,0,0]$ of arity $4$.
 Notice that $\plholant{[0,1,0]}{[a+\hat{f}_0,2b,1,0,0]}$ is equivalent 
 to $\plholant{[0,1,0]}{[0,0,1,0,0]}$,
 which is counting Eulerian orientations in $4$-regular planar graphs.
 Then $\plholant{{\neq}_2}{\hat{f}'''}$ is \numP-hard by Corollary~\ref{cor:prelim:nonsingular_compressed_hard_trans}.

 The only case left now is when $n=5$ and $\hat{f}=[\hat{f}_0,0,1,0,0,0]$.
 We do two self-loops on $\hat{f}$ via $\neq_2$ to get $[1,0]$.
 Then connect $[1,0]$ to $\hat{h}$ via $\neq_2$ and get $[b,1]$.
 At last, connect $[b,1]$ to $\hat{f}$ via $\neq_2$,
 resulting in $[\hat{f}_0,b,1,0,0]$.
 Similar to the case above,
 this is counting Eulerian orientations in $4$-regular planar graphs,
 and is \numP-hard by Corollary~\ref{cor:prelim:nonsingular_compressed_hard_trans}.
\end{proof}

If $f \in \mathscr{M}_4^\pm$,
there is an additional case for the binary signature.

\begin{lemma} \label{lem:van:M4:bin}
 Let $f \in \mathscr{M}_4^\sigma$ be a symmetric non-degenerate signature with $\sigma \in \{+,-\}$ of arity $k \ge 3$.
 Suppose $h$ is a non-degenerate binary signature such that
 $h \notin \mathscr{R}^\sigma_2$ and $h$ is not a multiple of $Z^{\otimes 2}[a,0,1]$ for any $a \neq 0$.
 Then $\PlHolant(f,h)$ is \numP-hard.
\end{lemma}

\begin{proof}
 We assume $f \in \mathscr{M}_4^+$ since the other case is similar.
 Suppose $h = Z^{\otimes 2}[a,b,c]$ for some $a,b,c \in \mathbb{C}$.
 Since $h \notin \mathscr{R}^+_2$,
 we have $c \neq 0$,
 so we assume $c=1$.
 Moreover $b \neq 0$. 
 This is because,
 if $b=0$ then either $h$ is degenerate or 
 is a multiple of $Z^{\otimes 2} [a,0,1]$ for some $a\neq 0$.
 Either case is a contradiction.
 Then under a holographic transformation by $Z$,
 the problem becomes $\plholant{{\neq}_2}{\exactone{k},[a,b,1]}$.
 If we connect two copies of $\exactone{k}$ via $\neq_2$,
 we get $\exactone{2k-2}$.
 Hence we may assume that $k \ge 5$.
 Then we connect $[a,b,1]$ to $\exactone{k}$ via $\neq_2$,
 and get $[2b,1,0,\dots,0]$ of arity $k-2 \ge 3$.
 Since $b \neq 0$,
 $\PlHolant(f,h)$ is \numP-hard by Lemma~\ref{lem:van:bin}.
\end{proof}

Next we consider mixing signatures from $\mathscr{V}^+$ and $\mathscr{V}^-$.
This is a planar version of Lemma~7.3 in~\cite{CGW13}.
However,
for planar graphs,
there is a tractable case when one signature is in $\mathscr{M}_4^+$ and the other is in $\mathscr{M}_4^-$.
This case was shown to be \numP-hard over general graphs by Lemma~6.12 in~\cite{CGW13} using a nonplanar reduction.
One can check that the rest of the proof of Lemma~7.3 in~\cite{CGW13} holds for planar graphs.
For completeness we include a proof.

\begin{lemma} \label{lem:van:plus_and_minus}
 Let $f \in \mathscr{V}^+$ and $g \in \mathscr{V}^-$ be symmetric non-degenerate signatures of arities $\ge 3$ respectively.
 If $f \notin \mathscr{M}_4^+$ or $g \notin \mathscr{M}_4^-$
 then $\PlHolant(f,g)$ is \numP-hard.
\end{lemma}

\begin{proof}
 Let $\rd^+(f) = d$,
 $\rd^-(g) = d'$,
 $\arity(f) = n$ and $\arity(g) = n'$,
 then $2 d < n$ and $2 d' < n'$.
 Under a holographic transformation by $Z = \trans{1}{1}{i}{-i}$,
 we have
 \[
  \plholant{{=}_2}{f,g}
  \equiv_T
  \plholant{{\neq}_2}{\hat{f},\hat{g}},
 \]
 where $\hat{f}:=(Z^{-1})^{\otimes n} f = [\hat{f}_0, \dots,\hat{f}_d, 0, \dotsc, 0]$
 and $\hat{g}:=(Z^{-1})^{\otimes n'} g = [0, \dotsc, 0, \hat{g}_{d'},\dots,\hat{g}_0]$
 due to Lemma~\ref{lem:prelim:vanishing_form_in_Z_basis}.
 Moreover $\hat{f}_d \not = 0$ and $\hat{g}_{d'} \neq 0$.

 If $d\ge 2$,
 we can do $d'$ many self-loops of $\neq_2$ on $\hat{g}$,
 getting $\hat{g}':=[0,\dots,0,\hat{g}_{d'}]$ of arity $n'-2d'\ge 1$.
 Thus $g':=Z^{\otimes (n'-2d')}\hat{g}'=[1,-i]^{\otimes (n'-2d')}$ up to a nonzero constant.
 We apply Lemma~\ref{lem:van:deg} to derive that $\PlHolant(f,g)$ is \numP-hard.
 If $d'\ge 2$,
 we can similarly get $[1,i]^{\otimes (n-2d)}$ and apply Lemma~\ref{lem:van:deg}.
 Thus we can assume that $d = d' = 1$.

 So up to nonzero constants,
 we have $\hat{f} = [a, 1, 0, \dotsc, 0]$ and $\hat{g} = [0, \dotsc, 0, 1, b]$ for some $a, b \in \mathbb{C}$.
 We can assume that $f \notin \mathscr{M}_4^+$ and $a \neq 0$.
 The case of $b\neq 0$ is similar.
 We show that it is always possible to get two such signatures of the same arity $\min\{n, n'\}$.
 Suppose $n > n'$.
 We form a loop from $\hat{f}$ via $\neq_2$.
 It is easy to see that this signature is the degenerate signature $2 [1,0]^{\otimes (n-2)}$.
 Similarly,
 we can form a loop from $\hat{g}$ and can get $2 [0, 1]^{\otimes (n'-2)}$.
 Thus we have both $[1,0]^{\otimes (n-2)}$ and $[0, 1]^{\otimes (n'-2)}$.
 We can connect all $n'-2$ edges of the second to the first,
 connected by $\neq_2$.
 This gives $[1,0]^{\otimes (n-n')}$.
 We can continue subtracting the smaller arity from the larger one.
 We continue this process in a subtractive version of the Euclidean algorithm,
 and end up with both $[1,0]^{\otimes t}$ and $[0, 1]^{\otimes t}$,
 where $t = \gcd(n-2, n'-2) = \gcd(n-n', n'-2)$.
 In particular, $t \mid n-n'$ and by taking $\frac{n-n'}{t}$ copies of $[0, 1]^{\otimes t}$,
 we can get $[0, 1]^{\otimes (n-n')}$.
 Connecting this back to $\hat{f}$ via $\neq_2$,
 we get a symmetric signature of arity $n'$ consisting of the first $n' + 1$ entries of $\hat{f}$.
 A similar proof works when $n' > n$.

 Thus we may assume $n = n'$.
 Connecting $[0, 1]^{\otimes (n-2)}$ to $\hat{f} = [a, 1, 0, \dotsc, 0]$ via $\neq_2$ we get $\hat{h} = [a,1,0]$.
 Recall that $a\neq 0$.
 Translating this back by $Z$,
 we have a binary signature $h\notin\mathscr{R}^-_2$ and $h$ is not a multiple of $Z^{\otimes 2}[c,0,1]$ for any $c\neq 0$.
 Since  $g \in \mathscr{V}^-$,
 by Lemma~\ref{lem:van:bin} or Lemma~\ref{lem:van:M4:bin},
 $\PlHolant(g,h)$ is \numP-hard.
 Hence $\PlHolant(f,g)$ is also \numP-hard.
\end{proof}

When signatures in both $\mathscr{M}_4^+$ and $\mathscr{M}_4^-$ appear,
we show that the only degenerate signatures that mix must also be vanishing.

\begin{lemma} \label{lem:M4+:M4-:deg}
 Let $f \in\mathscr{M}_4^+$ and $g \in\mathscr{M}_4^-$ be two non-degenerate signatures of arity $\ge 3$.
 Let $v = u^{\otimes m}$ be a degenerate signature for some unary signature $u$ and some integer $m \ge 1$.
 If $u$ is not a multiple of $[1,\pm i]$,
 then $\PlHolant(f,g,v)$ is \numP-hard.
\end{lemma}

\begin{proof}
 Suppose $f$ is of arity $n$ and $g$ of arity $\ell$.
 Under a holographic transformation by $Z$,
 we have
 \[
  \PlHolant(f,g,v) \equiv \plholant{{\neq}_2}{\exactone{n},\allbutone{\ell},[a,b]^{\otimes m}},
 \]
 where $ab \neq 0$.
 Notice that $v$ is transformed to $(Z^{-1} u)^{\otimes m} = [a,b]^{\otimes m}$.
 We have $ab \neq 0$ since $u$ is not a multiple of $[1,\pm i]$.
 First we get $[1,0]^{\otimes n-2}$ by a self-loop via $\neq_2$ on $\exactone{n}$.
 By the same subtractive Euclidean argument as in Lemma~\ref{lem:van:deg},
 we can realize $[a,b]^{\otimes n-2}$ by $[1,0]^{\otimes n-2}$ and $[a,b]^{\otimes m}$.
 Connecting $[a,b]^{\otimes n-2}$ to \exactone{n} via $\neq_2$
 we get a binary signature $h=[(n-2)ab^{n-3},b^{n-2},0]$.
 After transforming back,
 we have
 \[
  \PlHolant(g, Z^{\otimes 2} h) \le_T \PlHolant(f,g,v).
 \]
 However $Z^{\otimes 2}h\notin\mathscr{R}_2^-$ by Lemma~\ref{lem:prelim:vanishing_form_in_Z_basis}
 and it is not a multiple of $Z^{\otimes 2}[c,0,1]$ for any $c\neq 0$.
 Hence $\PlHolant(f,g,v)$ is \numP-hard by Lemma~\ref{lem:van:M4:bin},
 where $(g, Z^{\otimes 2} h)$ plays the role of ``$(f,h)$'' in Lemma~\ref{lem:van:M4:bin}
 and $\sigma = -$.
\end{proof}

We also consider the mixing of vanishing signatures with those in $\mathscr{P}_2$.

\begin{lemma} \label{lem:vanishing:P2}
 Let $f \in \mathscr{V} \setminus \mathscr{M}_4$ and $g \in \mathscr{P}_2$ be two non-degenerate signatures with arities $m$ and $n$ respectively.
 If $m, n \ge 3$,
 then $\PlHolant(f,g)$ is \numP-hard.
\end{lemma}

\begin{proof}
 We claim that it suffices to consider $f \in \mathscr{V}^+ \setminus \mathscr{M}_4$ and $g = \tbcolvec{1}{i}^{\otimes n} + \tbcolvec{1}{-i}^{\otimes n}$.
 By Lemma~\ref{lem:single:A2},
 we know that $g = \tbcolvec{1}{i}^{\otimes n} + \beta \tbcolvec{1}{-i}^{\otimes n}$ for some $\beta \ne 0$ up to a nonzero scalar.
 Under a holographic transformation by $T = Z \trans{1}{0}{0}{\beta^{\frac{1}{n}}} Z^{-1}$,
 which is orthogonal up to a nonzero factor of $\beta^{\frac{1}{n}}$,
 we have $\hat{g} = (T^{-1})^{\otimes n} g = \tbcolvec{1}{i}^{\otimes n} + \tbcolvec{1}{-i}^{\otimes n}$.
 Now $\mathscr{M}_4$ is closed under orthogonal transformations by definition,
 and $\mathscr{V}$ is closed under orthogonal transformations by Lemma~\ref{lem:van:orth}.
 Thus,
 we still have a signature $\hat{f} = (T^{-1})^{\otimes n} f$ such that $\hat{f} \in \mathscr{V} \setminus \mathscr{M}_4$.
 If $\hat{f} \in \mathscr{V}^-$,
 then under a holographic transformation by $D = \trans{1}{0}{0}{-1}$,
 we have $\hat{f} \in \mathscr{V}^+$.
 Furthermore,
 $\hat{g}$ is invariant under $D$.
 This proves the claim.
 
 Now we assume that $f \in \mathscr{V}^+ \setminus \mathscr{M}_4$ and $g = \tbcolvec{1}{i}^{\otimes n} + \tbcolvec{1}{-i}^{\otimes n}$.
 By Corollary~\ref{cor:van:degree},
 we have $\rd^+(f) = d < \frac{m}{2}$.
 Under a holographic transformation by $Z$,
 we have
 \begin{align*}
  \plholant{{=}_2}{f, g}
  &\equiv \plholant{[1,0,1] Z^{\otimes 2}}{Z^{-1} \{f, g\}}\\
  &\equiv \plholant{{\neq}_2}{\hat{f}, {=}_n},
 \end{align*}
 where $\hat{f} = (Z^{-1})^{\otimes m} f$.
 By Lemma~\ref{lem:prelim:vanishing_form_in_Z_basis},
 the support of $\hat{f}$ is on entries with Hamming weight at most $d$
 and includes the entry of Hamming weight exactly $d$.
 Now $f \notin \mathscr{M}_4$,
 so by Lemma~\ref{lem:trans:M4},
 we either have $d = 1$ and $\hat{f} = [\hat{f}_0,1,0,\dotsc,0]$ with $\hat{f}_0 \ne 0$
 or $d \ge 2$ and $\hat{f} = [\hat{f}_0, \hat{f}_1, \dotsc, \hat{f}_{d-1}, 1, 0, \dots, 0]$
 (and up to a nonzero scalar in either case).

 In the first case,
 a self-loop on $\hat{f}$ via $\neq_2$ gives $[1,0]^{\otimes m-2}$ on the right side.
 Let $r = \gcd(n, m-2)$,
 and let $\ell_1, \ell_2$ be two positive integers such that $\ell_1 n - \ell_2 (m-2) = r$.
 We connect $\ell_1$ copies of $=_n$ with $\ell_2$ copies of $[1,0]^{\otimes m-2}$ via $\neq_2$'s to get $[0,1]^{\otimes r}$.
 Since $r \mid m-2$,
 we can also realize $[0,1]^{\otimes m-2}$ by putting $\frac{m-2}{r}$ copies of $[0,1]^{\otimes r}$ together.
 Now connect $[0,1]^{\otimes m-2}$ to $\hat{f}$ via $\neq_2$.
 The resulting signature is $[\hat{f}_0,1,0]$.
 We can also move $=_n$ to the left using $n$ copies of $\neq_2$.
 Hence,
 we have $\PlHolant(=_n \mid [\hat{f}_0,1,0]) \le_T \PlHolant({\neq}_2 \mid \hat{f}, {=}_n)$.
 The former problem is \numP-hard by Theorem~\ref{thm:k-reg_homomorphism} since $\hat{f}_0 \ne 0$,
 so the latter problem is \numP-hard as well.

 In the second case,
 we have $m \ge 5$ since $2 \le d < \frac{m}{2}$.
 Furthermore,
 we may assume that $d=2$,
 since otherwise can we do $d-2$ self-loops on $\hat{f}$ via $\neq_2$.
 With this assumption,
 we do two self-loops on $\hat{f}$ via $\neq_2$ to get $[1,0]^{\otimes m-4}$ on the right side.
 By a similar argument as in the previous case,
 we can construct $[0,1]^{\otimes m-4}$ by using $[1,0]^{\otimes m-4}$ and $=_n$ via $\neq_2$.
 Now connect $[0,1]^{\otimes m-4}$ back to $\hat{f}$ via $\neq_2$.
 We get the arity~$4$ signature $[\hat{f}_0, \hat{f}_1, 1, 0, 0]$.
 Hence, we have $\PlHolant({\neq}_2 \mid [\hat{f}_0, \hat{f}_1, 1, 0, 0]) \le_T \PlHolant({\neq}_2 \mid \hat{f}, {=}_n)$.
 Note that $\PlHolant({\neq}_2 \mid [\hat{f}_0, \hat{f}_1, 1, 0, 0])$
 is equivalent to $\PlHolant({\neq}_2 \mid [0, 0, 1, 0, 0])$,
 counting Eulerian Orientations in planar $4$-regular graphs,
 which is \numP-hard by Corollary~\ref{cor:prelim:nonsingular_compressed_hard_trans}.
 Thus $\PlHolant({\neq}_2 \mid \hat{f}, {=}_n)$ is \numP-hard as well.
\end{proof}

\section{Dichotomy for \texorpdfstring{$\PlCSP^2$}{Pl-Count-CSP2} and Related Lemmas}\label{sec:PlCSP2}

In this section,
we state the dichotomy for $\PlCSP^2$.
We defer the proof to Part II of this paper starting on page \pageref{partII:sec:prelim}.
We provide a sketch of the proof here.
Afterwards,
we discuss several related lemmas,
which are used for the full dichotomy of $\PlHolant$.
Let $\mathcal{T}_k = \left\{\left[\begin{smallmatrix} 1 & 0 \\ 0 & \omega \end{smallmatrix}\right] 
\in \mathbb{C}^{2 \times 2} \st \omega^k = 1\right\}$.

\begin{theorem} \label{thm:PlCSP2}
 Let $\mathcal{F}$ be a set of symmetric signatures.
 Then $\PlCSP^2(\mathcal{F})$ is \numP-hard unless
 $\mathcal{F}$ satisfies one of the following conditions:
 \begin{enumerate}
   \item \label{case:plcsp2:TA}
     there exists $T \in \mathcal{T}_8$ such that $\mathcal{F} \subseteq T \mathscr{A}$;
   \item \label{case:plcsp2:P}
     $\mathcal{F} \subseteq   \mathscr{P}$;
   \item \label{case:plcsp2:TM}
     there exists $T \in \mathcal{T}_4$ such that $\mathcal{F} \subseteq T \widehat{\mathscr{M}}$.
 \end{enumerate}
 In each exceptional case,
 $\PlCSP^2(\mathcal{F})$ is computable in polynomial time.
\end{theorem}

\begin{proof}[Proof Sketch]
 We first define some tractable families of signatures specific to the $\PlCSP^2$ framework.
 Let $\widetilde{\mathscr{A}} = \mathscr{A} \cup \trans{1}{0}{0}{e^{\pi i/4}}\mathscr{A}$ and
 $\widetilde{\mathscr{M}} = \widehat{\mathscr{M}} \cup \trans{1}{0}{0}{i} \widehat{\mathscr{M}}$.
 One can show that $\widetilde{\mathscr{A}}$ covers Case \ref{case:plcsp2:TA} above,
 and $\widetilde{\mathscr{M}}$ covers Case \ref{case:plcsp2:TM}.
 The proof will revolve around these tractable classes.

 The overall plan is to break the proof into two main steps.

 The first step is to prove the dichotomy theorem for $\PlCSP^2(\mathcal{F})$
 when there is at least one nonzero signature of \emph{odd} arity in $\mathcal{F}$.
 In this case, we can make use of a lemma showing
 that we can simulate $\PlCSP(\mathcal{F})$ by $\PlCSP^2(\mathcal{F})$
 if $\mathcal{F}$ includes a unary signature $[a, b]$ with $a b \neq 0$.
 Then we can apply the known dichotomy Theorem~\ref{thm:PlCSP} for $\PlCSP$.
 However this strategy (provably) \emph{cannot} work 
 when every signature in $\mathcal{F}$ satisfies the \emph{parity} constraint.
 In that case we employ other means.
 This first step of the proof is relatively uncomplicated.

The second step is to deal with the case when
all nonzero signatures in $\mathcal{F}$ have even arity. 
This is where the real difficulties lie. 
In this case it is impossible to directly construct \emph{any} unary signature.
So we cannot use that lemma pertaining to a unary signature.
%Lemma~\ref{[1,a]XXX} in this case.
But we prove another lemma which provides a way to simulate
% Lemma~\ref{mixing-P-global} provides a way to simulate
$\PlCSP(\mathcal{F})$ by $\PlCSP^{2}(\mathcal{F})$ in a \emph{global} fashion,
\emph{if} $\mathcal{F}$ includes some tensor power of the form $[a, b]^{\otimes 2}$ where $ab\neq 0$.
Moreover, we have a lucky break (for the complexity of the proof) 
if $\mathcal{F}$ includes a signature that is 
in $\widehat{\mathscr{M}}\setminus(\mathscr{P}\cup\widetilde{\mathscr{A}})$.
In this case, we can construct a special binary signature,
%%% [1,a,1], with a^4 \not = 0, 1
%and then use Lemma~\ref{[1,a,1]-interpolation}
and obtain $[1,1]^{\otimes 2}$ by interpolation.
This proof uses the theory of \emph{cyclotomic fields}.
This simplifies the proof greatly.
For all other cases 
(when $\mathcal{F}$ has only even arity signatures), 
the proof gets going in earnest---%
we will attempt an induction on the arity of signatures.

The lowest arity of this induction will be $2$.
We will try to reduce the arity to $2$ whenever possible; 
however for many cases an arity reduction to $2$ destroys the \#P-hardness at hand. 
Therefore the true basis of this induction proof of $\PlCSP^2$ starts with arity $4$. 
Consequently we will first prove a dichotomy theorem for $\PlCSP^2(f)$, 
where $f$ is a signature of arity $4$.
Several tools will be used. 
These include the rank criterion for redundant signatures, 
Theorem \ref{thm:k-reg_homomorphism} for arity $2$ signatures,
and a trick we call the \emph{Three Stooges} by domain pairing.

However, in the next step we do not attempt a general
$\PlCSP^2$ dichotomy for a \emph{single} signature of even arity. 
This would have been natural at this point, 
but it would have been too difficult.
We will need some additional leverage by proving a conditional
``No-Mixing'' Lemma for pairs of signatures of even arity.
So, seemingly taking a detour,
we prove that for two signatures $f$ and $g$ both of even arity,
that individually belong to some tractable class,
but do not belong to a single tractable class in the conjectured dichotomy
(that is yet to be proved),
the problem $\PlCSP^2(f,g)$ is \numP-hard.
We prove this No-Mixing Lemma for any pair of signatures $f$ and $g$ 
both of even arity, not restricted to arity $4$.
Even though at this point we only have a dichotomy for a single signature of arity 4, 
we prove this No-Mixing Lemma for higher even arity pairs $f$ and $g$ 
by simulating two signatures $f'$ and $g'$ of arity 4 
that belong to different tractable sets, from that of $\PlCSP^2(f,g)$.
After this arity reduction (within the No-Mixing Lemma),
we prove that $\PlCSP^2(f', g')$ is \numP-hard by the dichotomy 
for a \emph{single} signature of arity $4$.
After this, we prove a No-Mixing Lemma for a \emph{set} of signatures $\mathcal{F}$ of even arities, 
which states that \emph{if} $\mathcal{F}$ is contained in the union of all tractable classes, 
then it is still \numP-hard unless it is \emph{entirely} contained in one single tractable class.  
Note that at this point we still only have a \emph{conditional} No-Mixing Lemma 
in the sense that we have to assume every signature in $\mathcal{F}$ belongs to some tractable set.

We then attempt the proof of a $\PlCSP^2$ dichotomy 
for a \emph{single} signature of arbitrary even arity. 
This uses all the previous lemmas, 
in particular the (conditional) No-Mixing Lemma for a set of signatures. 
However, after completing the proof of this $\PlCSP^2$ dichotomy 
for a single signature of even arity, 
the No-Mixing Lemma becomes absolute.

Finally the dichotomy for a single signature of even arity 
is logically extended to a dichotomy theorem for $\PlCSP^2(\mathcal{F})$
where all signatures in $\mathcal{F}$ have even arity.
Together with the first main step 
when $\mathcal{F}$ contains some nonzero signature of odd arity,
this completes the proof of Theorem~\ref{thm:PlCSP2}.
\end{proof}

\subsection{Related Lemmas}

Now we give some consequences of Theorem~\ref{thm:PlCSP2}.
These are cases that can be reduced to $\PlCSP^2$.
We consider signatures in $\mathscr{P}_1$, 
$\mathscr{M}_2 \setminus \mathscr{P}_2$, $\mathscr{A}_3$, or $\mathscr{M}_3$.

We begin with the cases of $\mathscr{P}_1$ and $\mathscr{A}_3$.
The following two lemmas are rephrased from~\cite{CGW13}.
One can check that the reductions in these proofs are planar.

\begin{lemma}[Lemma 8.15 in~\cite{CGW13}] \label{lem:reductions:P1}
  Let $f\in\mathscr{P}_1$ be a non-degenerate signature of arity $n\ge 3$ with an orthogonal transformation $H$
  and $\mathcal{F}$ be a set of signatures containing $f$.
  Let $H_2$ be the $2$-by-$2$ matrix $\tfrac{1}{\sqrt{2}}\trans{1}{1}{1}{-1}$.
  Then $\PlCSP^2(H_2H\mathcal{F})\le_T\PlHolant(\mathcal{F})$.
\end{lemma}

\begin{lemma}[Lemma 8.17 in~\cite{CGW13}] \label{lem:reductions:A3}
  Let $f\in\mathscr{A}_3$ be a non-degenerate signature of arity $n\ge 3$ with an orthogonal transformation $H$
  and $\mathcal{F}$ be a set of signatures containing $f$.
  Let $\alpha=e^{\pi i/4}$ 
  and $Y$ be the $2$-by-$2$ matrix \trans{\alpha}{1}{-\alpha}{1}.
  Then $\PlCSP^2(YH\mathcal{F}\cup\{[1,-i,1]\})\le_T\PlHolant(\mathcal{F})$. 
\end{lemma}

With these reductions, 
we can apply Theorem~\ref{thm:PlCSP2}
to get the following corollaries.
The next one follows directly from Lemma~\ref{lem:reductions:P1} and Theorem~\ref{thm:PlCSP2}
as $H_2$ is orthogonal and every $\PlCSP^2$ tractable case is also tractable for $\PlHolant$.

\begin{corollary}\label{cor:dichotomy:P1}
  Let $\mathcal{F}$ be a set of signatures.
  Suppose there exists $f\in\mathcal{F}$ which is
  a non-degenerate signature of arity $n\ge 3$ in $\mathscr{P}_1$
  with $H\in\mathbf{O}_2(\mathbb{C})$.
  Then $\PlHolant(\mathcal{F})$ is \numP-hard unless $\mathcal{F}$ is $\mathscr{A}$-, $\mathscr{P}$-, or $\mathscr{M}$-transformable,
  in which case $\PlHolant(\mathcal{F})$ is tractable.
\end{corollary}

The proof of this corollary is straightforward. To illustrate the
power of Theorem~\ref{thm:PlCSP2}, we give a short proof here.

\begin{proof}
Let $H' = (H_2 H)^{-1} \in \mathbf{O}_2(\mathbb{C})$.
By Lemma~\ref{lem:reductions:P1} and Theorem~\ref{thm:PlCSP2},
 $\PlHolant(\mathcal{F})$ is \numP-hard
unless either (1) $\mathcal{F} \subseteq H' \mathscr{P}$,
or (2) $\mathcal{F} \subseteq H' T \mathscr{A}$,
or (3) $\mathcal{F} \subseteq H' T' \trans{1}{1}{1}{-1} \mathscr{M}$,
where $T \in \mathcal{T}_8$ and $T' \in \mathcal{T}_4$. In case (1), $\mathcal{F}$ is $\mathscr{P}$-transformable
since $(=_2) H'^{\otimes 2} = (=_2) \in  \mathscr{P}$.
In case (2), $\mathcal{F}$ is $\mathscr{A}$-transformable
since $(=_2) (H'T)^{\otimes 2} = (=_2)T^{\otimes 2}  \in \mathscr{A}$.
In case (3),  $\mathcal{F}$ is $\mathscr{M}$-transformable.
If $T' = \trans{1}{0}{0}{\pm 1}$,
then $T' \in \mathbf{O}_2(\mathbb{C})$.
So $(=_2) (H'T' \trans{1}{1}{1}{-1})^{\otimes 2}  = (=_2) \in \mathscr{M}$.
If $T' = \trans{1}{0}{0}{\pm i}$,
then $T' \trans{1}{1}{1}{-1} = \trans{1}{1}{i}{-i}$,
and $(=_2) (H'T' \trans{1}{1}{1}{-1})^{\otimes 2}  = 2[0,1,0] \in
\mathscr{M}$.
\end{proof}

Corollary~\ref{cor:dichotomy:P1} is useful in Section~\ref{sec:full:dic}.
In Section~\ref{sec:single:dichotomy},
we need the following further specialization.

\begin{corollary} \label{cor:induction:P1}
 Suppose $f$ is a non-degenerate signature of arity $n \ge 5$.
 Let $f'$ be $f$ with a self loop,
 and assume that $f' \in\mathscr{P}_1$ is non-degenerate.
 Then $\PlHolant(f)$ is \numP-hard unless $f$ is $\mathscr{A}$-, $\mathscr{P}$-, or $\mathscr{M}$-transformable,
 in which case $\PlHolant(f)$ is tractable.
\end{corollary}

For the other case of $\mathscr{A}_3$,
some case analysis is required.

\begin{corollary}\label{cor:dichotomy:A3}
  Let $\mathcal{F}$ be a set of signatures.
  Suppose there exists $f\in\mathcal{F}$ which is
  a non-degenerate signature of arity $n\ge 3$ in $\mathscr{A}_3$.
  Then $\PlHolant(\mathcal{F})$ is \numP-hard unless $\mathcal{F}$ is $\mathscr{A}$- or $\mathscr{M}$-transformable,
  in which case $\PlHolant(\mathcal{F})$ is tractable.
\end{corollary}

\begin{proof}
  Assume that $f\in\mathscr{A}_3$ with an orthogonal transformation $H$.
  By Lemma~\ref{lem:reductions:A3},
  we have $\PlCSP^2(Y H \mathcal{F} \cup \{[1,-i,1]\}) \le_T \PlHolant(\mathcal{F})$,
  where $Y = \trans{\alpha}{1}{-\alpha}{1}$ and $\alpha = e^{\pi i/4}$.
  Let $g = [1,-i,1]$ and $\mathcal{F}' = Y H \mathcal{F} \cup \{g\}$.

  We apply Theorem~\ref{thm:PlCSP2} to $\PlCSP^2(\mathcal{F}')$.
  The consequence is that $\PlCSP^2(\mathcal{F}')$ (and hence $\PlHolant\mathcal({F})$) is \numP-hard unless
  $\mathcal{F}'\subseteq\mathscr{P}$, 
  $\mathcal{F}'\subseteq\trans{1}{0}{0}{i^r}\widehat{\mathscr{M}}$ for some integer $0\le r \le 3$,
  or $\mathcal{F}'\subseteq \trans{1}{0}{0}{\alpha^r}\mathscr{A}$ for some integer $0\le r\le 7$ where $\alpha=e^{i\pi/4}$.
  Notice that $g\not\in\mathscr{P}$ and hence
  the first case is impossible.

  Suppose $\mathcal{F}'\subseteq\trans{1}{0}{0}{i^r}\widehat{\mathscr{M}}$ for some integer $0\le r \le 3$.
  Then as $g\not\in\trans{1}{0}{0}{i^r}\widehat{\mathscr{M}}$ for $r=1,3$,
  we have that $YH\mathcal{F}\subseteq\trans{1}{0}{0}{\pm1}\widehat{\mathscr{M}}$.
  Moreover, notice that 
  $\trans{1}{0}{0}{-1}\widehat{\mathscr{M}}
  = \trans{1}{1}{1}{-1} \trans{0}{1}{1}{0} \mathscr{M}
  = \trans{1}{1}{1}{-1} \mathscr{M}
  = \widehat{\mathscr{M}}$.
  Hence $YH\mathcal{F}\subseteq\widehat{\mathscr{M}}$.
  Rewrite $Y$ as $Y=\trans{1}{1}{-1}{1}\trans{\alpha}{0}{0}{1}$.
  We deduce that
  \begin{align*}
   H \mathcal{F}
   &\subseteq
   \tfrac{1}{2}
   \trans{\alpha^{-1}}{0}{0}{1}
   \trans{1}{-1}{1}{1}
   \widehat{\mathscr{M}}
   =
   \tfrac{1}{2}
   \trans{\alpha^{-1}}{0}{0}{1}
   \trans{1}{-1}{1}{1}
   \trans{1}{1}{1}{-1}
   \mathscr{M}\\
   &=
   \trans{\alpha^{-1}}{0}{0}{1}
   \trans{0}{1}{1}{0}
   \mathscr{M}
   =
   \mathscr{M}.
  \end{align*}
  Hence $\mathcal{F}$ is $\mathscr{M}$-transformable in this case.

  The last case is when $\mathcal{F}'\subseteq \trans{1}{0}{0}{\alpha^r}\mathscr{A}$ for some integer $0\le r\le 7$.
  It implies that $r=0,2,4,6$ as $g\in\trans{1}{0}{0}{\alpha^r}\mathscr{A}$
  and $g \not \in\trans{1}{0}{0}{\alpha}\mathscr{A}$.
  That is, $\mathcal{F}'\subseteq \trans{1}{0}{0}{i^l}\mathscr{A}$ for some integer $0\le l\le 3$.
  Notice that $\trans{1}{0}{0}{i^l}\in\Stab{A}$.
  It implies that $YH\mathcal{F}\subseteq\mathscr{A}$.
  Again, rewriting $Y$ as $Y=\trans{1}{1}{-1}{1}\trans{\alpha}{0}{0}{1}$,
  we have 
  \[
   H \mathcal{F}
   \subseteq
   \tfrac{1}{2}
   \trans{\alpha^{-1}}{0}{0}{1}
   \trans{1}{-1}{1}{1}
   \mathscr{A}
   =\tfrac{1}{2}
   \trans{\alpha^{-1}}{0}{0}{1}
   \mathscr{A}.
  \]
  Therefore $\mathcal{F}$ is $\mathscr{A}$-transformable.
  This finishes the proof.
\end{proof}

Again, we specialize Corollary~\ref{cor:dichotomy:A3} to our need.

\begin{corollary}\label{cor:induction:A3}
  Let $f$ be a non-degenerate signature of arity $n\ge5$.
  Let $f'$ be $f$ with a self loop, and $f'$ is non-degenerate and $f'\in\mathscr{A}_3$ with an orthogonal transformation $H$.
  Then $\PlHolant(f)$ is \numP-hard unless $f$ is $\mathscr{A}$- or $\mathscr{M}$-transformable,
  in which case $\PlHolant(f)$ is tractable.
\end{corollary}

The next case is when $f$ is in $\mathscr{M}_2$ but not $\mathscr{P}_2$.

\begin{lemma}  \label{lem:dichotomy:M2}
 Let $\mathcal{F}$ be a set of signatures.
 Suppose there exists $f\in\mathcal{F}$ which is
 a non-degenerate signature of arity $n\ge 3$ in $\mathscr{M}_2\setminus\mathscr{P}_2$.
 Then $\PlHolant\left( \mathcal{F} \right)$ is \numP-hard 
 unless $\mathcal{F}$ is $\mathscr{A}$-, $\mathscr{P}$-, or $\mathscr{M}$-transformable,
 in which case $\PlHolant\left( \mathcal{F} \right)$ is tractable.
\end{lemma}

\begin{proof}
  As $f\in\mathscr{M}_2\setminus\mathscr{P}_2$,
  assume $f= H^{\otimes n} \left(\tbcolvec{1}{\gamma}^{\otimes n} \pm \tbcolvec{1}{-\gamma}^{\otimes n}\right)$,
  where $H$ is an orthogonal $2$-by-$2$ matrix and $\gamma\neq 0, \pm i$.

  We first show that
  \begin{align} \label{eqn:M2:CSP2}    
    \PlCSP^2(T^{-1}\mathcal{F},g)\le_T\PlHolant\left( \{f\}\cup\mathcal{F} \right),
  \end{align}  
  where $T= H\trans{1}{1}{\gamma}{-\gamma}$ and $g= (=_2)T^{\otimes 2} = [1+\gamma^2,1-\gamma^2,1+\gamma^2]$

  Assume that $f = H^{\otimes n} \left(\tbcolvec{1}{\gamma}^{\otimes n} + \tbcolvec{1}{-\gamma}^{\otimes n}\right)$ with the $+$ sign.
  In this case, we do the transformation $T$:
  \begin{align*}
   \plholant{{=}_2}{f,\mathcal{F}}
   &\equiv_T
   \PlHolant
   \left(
    [1,0,1] H^{\otimes 2}
    \trans{1}{1}{\gamma}{-\gamma}^{\otimes 2}
    \enspace \middle| \enspace
    \left(\trans{1}{1}{\gamma}{-\gamma}^{-1}\right)^{\otimes n}
    \left(H^{-1}\right)^{\otimes n} f, T^{-1}\mathcal{F}
   \right)\\
   &\equiv_T
   \plholant{g}{{=}_n,T^{-1}\mathcal{F}}.
  \end{align*}
  By connecting $g$ to $=_n$,
  we get $=_{n-2}$ up to a constant factor of $1+\gamma^2\neq 0$ as $\gamma\neq\pm i$.
  We repeat this process.
  If $n$ is even, then we get $=_2$ eventually, which is on the right hand side.
  If $n$ is odd, then eventually we get $=_3$ and $(=_1)=[1,1]$ on the right.
  Connecting $[1,1]$ to $g$ we get $2[1,1]$ on the left.
  Then connecting $[1,1]$ to $=_3$ we get $=_2$ on the right.
  To summarize, we get that 
  \begin{align} 
   \plholant{g}{{=}_2,{=}_n,T^{-1}\mathcal{F}}
   &\le_T
   \plholant{g}{{=}_n,T^{-1}\mathcal{F}} \notag\\
   &\le_T \PlHolant\left( f,\mathcal{F} \right).\label{eqn:M2:reductions}
  \end{align}
  Next we show that 
  \begin{align} \label{eqn:binary:interpolation}
   \plholant{{=}_2,g}{{=}_2,{=}_n,T^{-1}\mathcal{F}}
   &\le_T
   \plholant{g}{{=}_2,{=}_n,T^{-1}\mathcal{F}}.
  \end{align}
  Let $N=\trans{1+\gamma^2}{1-\gamma^2}{1-\gamma^2}{1+\gamma^2}$ be the signature matrix of $g$.
  If there is a positive integer $k$ and a nonzero constant $c$ such that $N^k=c I_2$,
  where $I_2$ is the $2$-by-$2$ identity matrix,
  then we may directly implement $=_2$ on the left by connecting $k$ copies of $[1+\gamma^2,1-\gamma^2,1+\gamma^2]$ via $=_2$ on the right.
  It implies \eqref{eqn:binary:interpolation} holds.

  Otherwise such $k$ and $c$ do not exist.
  The two eigenvalues of $N$ are $\lambda_1=2$ and $\lambda_2=2\gamma^2$.
  If $\lambda_1=\lambda_2$, then $\gamma^2=1$ and $N=\trans{2}{0}{0}{2}$.
  Contradiction.
  Hence $\lambda_1 \not = \lambda_2$, and $N$ is diagonalizable.
  Let $N=P\trans{\lambda_1}{0}{0}{\lambda_2}P^{-1}$, for some non-singular matrix $P$.
  By connecting $l$ many copies of $N$ on the left via $=_2$ on the right, where $l$ is a positive integer, 
  we can implement $N^l=P\trans{\lambda_1^l}{0}{0}{\lambda^l_2}P^{-1}$.
  Since $N$ does not have finite order up to a scalar,
  for any positive integer $l$, $\left(\lambda_1/\lambda_2\right)^l\neq 1$.

  Consider an instance $\Omega$ of $\plholant{=_2,g}{=_2,=_n,T^{-1}\mathcal{F}}$.
  Suppose that the left $=_2$ appears $t$ times.
  Let $l$ be a positive integer.
  We obtain $\Omega_l$ from $\Omega$ by replacing each occurrence of $=_2$ on the left with $N^l$.

  Since $N^l=P\trans{\lambda_1^l}{0}{0}{\lambda_2^l}P^{-1}$,
  we can view our construction of $\Omega_l$ as replacing $N^l$ by $3$ signatures,
  with matrix $P$, $\trans{\lambda_1^l}{0}{0}{\lambda_2^l}$, and $P^{-1}$, respectively.
  This does not change the Holant value, 

  We stratify the assignments in $\Omega_l$ based on the assignments 
  to the $t$ occurrences of the signature whose matrix 
  is the diagonal matrix $\trans{\lambda_1^l}{0}{0}{\lambda_2^l}$.
  Suppose there are $i$ many times it was assigned  $00$ with function value $\lambda_1^l$,
  and $j$ times $11$ with function value $\lambda_2^l$.
  Clearly $i+j=t$ if the assignment has a nonzero evaluation.
  Let $c_{ij}$ be the sum over all such assignments of the products of evaluations of all signatures
  (including the signatures corresponding to matrices $P$ and $P^{-1}$)
  in $\Omega_l$ except for this diagonal one.
  Then
  \begin{align*}
   \Holant_{\Omega_l} & = \sum_{i + j = t} \left(\lambda_1^l\right)^i \left(\lambda_2^l\right)^j c_{ij}\\
   & = \lambda_2^{lt} \sum_{0\leq i \leq t} \left(\left(\frac{\lambda_1}{\lambda_2}\right)^l\right)^i c_{i,t-i}.
  \end{align*}
  By an oracle of $\plholant{g}{=_2,=_n,T^{-1}\mathcal{F}}$,
  we can get $\Holant_{\Omega_l}$ for any $1\le l\le t+1$.
  Recall that for any positive integer $l$, 
  $\left(\lambda_1/\lambda_2\right)^l\neq 1$.
  This implies that for any two distinct integers $i,j\geq 0$,
  $\left(\lambda_1/\lambda_2\right)^i\neq \left(\lambda_1/\lambda_2\right)^j$.
  Therefore we get a non-singular Vandermonde system.
  We can solve all $c_{ij}$ for $i+j=t$ given $\Holant_{\Omega_l}$ for all $1\le l\le t+1$.
  Then notice that $\sum_{i+j=t}c_{ij}$ is the Holant value of $\Omega_l$ 
  by replacing both $\lambda_1^l$ and $\lambda_2^l$ with $1$,
  which is the instance $\Omega$ as $PI_2P^{-1}=I_2$.
  Therefore we may compute $\Holant_\Omega$ via $t+1$ many oracle calls to $\plholant{g}{{=}_2,{=}_n,T^{-1}\mathcal{F}}$.
  This finishes the reduction in \eqref{eqn:binary:interpolation}.

  In the left hand side of \eqref{eqn:binary:interpolation}
  we have $=_2$ on both sides.
  Therefore we may lift the bipartite restriction.
  Combining it with \eqref{eqn:M2:reductions},
  we get \[\PlHolant\left( =_n,g,T^{-1}\mathcal{F} \right)\le_T \PlHolant\left(f, \mathcal{F} \right).\]
  Notice that given an equality of arity $n\ge 3$,
  we can always construct all equalities of even arity, 
  regardless of the parity of $n$, in the $\PlHolant$ setting.
  Therefore, we have $\PlCSP^2(T^{-1}\mathcal{F},g)\le_T\PlHolant\left( f,\mathcal{F} \right)$.

  To prove \eqref{eqn:M2:CSP2}, there is another case 
  that $f = H^{\otimes n} \left(\tbcolvec{1}{\gamma}^{\otimes n} - \tbcolvec{1}{-\gamma}^{\otimes n}\right)$,
  with the $-$ sign.
  Again we do a $T$ transformation, where $(T^{-1})^{\otimes}f=[1,0,\dots,0,-1]$ has arity $n$:
  \begin{align*}
    \plholant{{=}_2}{f,\mathcal{F}}
    &\equiv_T \plholant{g}{[1,0,\dots,0,-1],T^{-1}\mathcal{F}}.
  \end{align*}
  We then do the same construction as in the previous case of connecting $g$ to $[1,0,\dots,0,-1]$ repeatedly.
  Depending on the parity of $n$, we have two cases.
  \begin{enumerate}
    \item If $n$ is odd, then eventually we get $[1,0,0,-1]$ and $[1,-1]$ on the right as $\gamma\neq\pm i$, 
      and therefore $2 \gamma^2 [1,-1]$, i.e., $[1,-1]$ on the left as $\gamma\neq 0$.
      Then connecting $[1,-1]$ to $[1,0,0,-1]$ we get $=_2$ on the right.
      Thus, for odd $n$,
      \begin{align*} 
       \plholant{g}{{=}_2,[1,0,\dots,0,-1],T^{-1}\mathcal{F}}
       &\le_T
       \plholant{g}{[1,0,\dots,0,-1],T^{-1}\mathcal{F}}\\
       &\le_T
       \PlHolant\left( f,\mathcal{F} \right).
      \end{align*}
      Notice that our previous binary interpolation proof only relies on $g$ and $=_2$.
      Hence we get
      \begin{align*}
       \plholant{g}{{=}_2,[1,0,\dots,0,-1],T^{-1}\mathcal{F}}
       &\ge_T
       \plholant{{=}_2,g}{{=}_2,[1,0,\dots,0,-1],T^{-1}\mathcal{F}}\\
       &\equiv_T
       \PlHolant([1,0,\dots,0,-1],g,T^{-1}\mathcal{F}).
      \end{align*}
      Moreover it is straightforward to construct all even equalities from $[1,0,\dots,0,-1]$ in the normal $\PlHolant$ setting as $n\ge 5$.
      Combining everything together gives us 
      \begin{align*}
        \PlCSP^2(g,T^{-1}\mathcal{F})\le_T\PlHolant\left( f,\mathcal{F} \right).
      \end{align*}
    \item Otherwise $n$ is even.
      By the same construction of connecting $g$ to $[1,0,\dots,0,-1]$ repeatedly,
      we get $[1,0,0,0,-1]$ and $[1,0,-1]$ on the right eventually.
      Then we connect two copies of $g$ via $[1,0,-1]$, resulting in 
      $\trans{1+\gamma^2}{1-\gamma^2}{1-\gamma^2}{1+\gamma^2}
      \trans{1}{0}{0}{-1}
      \trans{1+\gamma^2}{1-\gamma^2}{1-\gamma^2}{1+\gamma^2}
      =4\gamma^2\trans{1}{0}{0}{-1}$ on the left.
      Then connect $[1,0,-1]$ to $[1,0,0,0,-1]$ to get $[1,0,1]$ on the right.
      At last we connect two $[1,0,-1]$'s on the left via $[1,0,1]$ on the right to get $[1,0,1]$ on the left.
      Then it reduces to the previous case.
  \end{enumerate}
  This concludes the proof of \eqref{eqn:M2:CSP2}.

  We apply Theorem~\ref{thm:PlCSP2} to $\PlCSP^2(T^{-1}\mathcal{F},g)$.
  Then we have that $\PlCSP^2(T^{-1}\mathcal{F},g)$ 
  (and hence $\PlHolant(f,\mathcal{F})$) is \numP-hard unless
  $T^{-1}\mathcal{F}\cup \{g\}\subseteq\mathscr{P}$, 
  or $T^{-1}\mathcal{F}\cup \{g\}\subseteq\trans{1}{0}{0}{i^r}\widehat{\mathscr{M}}$ for some integer $0\le r \le 3$,
  or $T^{-1}\mathcal{F}\cup \{g\}\subseteq \trans{1}{0}{0}{\alpha^r}\mathscr{A}$ for some integer $0\le r\le 7$ where $\alpha=e^{i\pi/4}$.
  We have three cases.
  \begin{enumerate}
    \item The first case is that $T^{-1}\mathcal{F}\cup \{g\}\subseteq\mathscr{P}$.
      Recall that $\gamma\neq 0$ or $\pm i$,
      it can be verified that $g\not\in\mathscr{P}$ unless $\gamma^2=1$.
      Hence $\gamma=\pm 1$.
      In either case we have that $\trans{1}{1}{\gamma}{-\gamma}$ is an orthogonal matrix up to a nonzero scalar,
      and hence so is $T$.
      It implies that $\mathcal{F}$ is $\mathscr{P}$-transformable.
    \item Next suppose $T^{-1}\mathcal{F}\cup \{g\}\subseteq\trans{1}{0}{0}{i^r}\widehat{\mathscr{M}}$ for some integer $0\le r \le 3$.
      If $\gamma=\pm1$, then $T$ is an orthogonal matrix as $\trans{1}{1}{\gamma}{-\gamma}$ is, up to a factor of $\tfrac{1}{\sqrt{2}}$.
      Hence $\mathcal{F}$ is $\mathscr{M}$-transformable,
      as $\mathcal{F}\subseteq T\trans{1}{0}{0}{i^r}\trans{1}{1}{1}{-1}\mathscr{M}$
      and $(=_2)\left(T\trans{1}{0}{0}{i^r}\trans{1}{1}{1}{-1}\right)^{\otimes 2}$
      is either $[1,0,1]$ when $r=0,2$, or $[0,1,0]$ when $r=1,3$,
      up to a nonzero factor.

      Otherwise $\gamma^2\neq 1$ and it is straightforward to verify that $g\not\in\trans{1}{0}{0}{i^r}\widehat{\mathscr{M}}$ for $r= 1,3$.
      Hence we may assume that $T^{-1}\mathcal{F}\subseteq\trans{1}{0}{0}{\pm1}\widehat{\mathscr{M}}$.
      Moreover, $\trans{1}{0}{0}{-1}\widehat{\mathscr{M}}
      = \trans{1}{1}{1}{-1} \trans{0}{1}{1}{0} \mathscr{M}
      = \trans{1}{1}{1}{-1} \mathscr{M}
      = \widehat{\mathscr{M}}$.
      Then $T^{-1}\mathcal{F}\subseteq\widehat{\mathscr{M}}$.
      As $T^{-1}=\trans{1}{1}{\gamma}{-\gamma}^{-1}H^{-1}$,
      it implies that
      \begin{align*}
        H^{-1}\mathcal{F}&\subseteq\trans{1}{1}{\gamma}{-\gamma}\widehat{\mathscr{M}}
        =\trans{1}{0}{0}{\gamma}\trans{1}{1}{1}{-1}\trans{1}{1}{1}{-1}\mathscr{M}\\
        &=\trans{1}{0}{0}{\gamma}\mathscr{M}=\mathscr{M}.
      \end{align*}
      Hence $\mathcal{F}\subseteq H\mathscr{M}$
      and $\mathcal{F}$ is $\mathscr{M}$-transformable.
    \item In the last case, $T^{-1}\mathcal{F}\cup \{g\}\subseteq \trans{1}{0}{0}{\alpha^r}\mathscr{A}$ for some integer $0\le r\le 7$.
      If $\gamma=\pm 1$, then $T$ is an orthogonal matrix as $\trans{1}{1}{\gamma}{-\gamma}$ is, up to a factor of $\tfrac{1}{\sqrt{2}}$.
      Hence $\mathcal{F}$ is $\mathscr{A}$-transformable,
      as $\mathcal{F}\subseteq T\trans{1}{0}{0}{\alpha^r}\mathscr{A}$
      and $(=_2)\left(T\trans{1}{0}{0}{\alpha^r}\right)^{\otimes 2}$
      is $[1,0,i^r]\in\mathscr{A}$, up to a nonzero factor.

      Otherwise $\gamma^2\neq 1$ and $g\not\in \trans{1}{0}{0}{\alpha^r}\mathscr{A}$ for any integer $r=1,3,5,7$.
      Hence $T^{-1}\mathcal{F}\cup \{g\}\subseteq \mathscr{A}$
      as $\trans{1}{0}{0}{i^r}\mathscr{A}=\mathscr{A}$ for any integer $0\le r\le 3$.
      If $\tfrac{1+\gamma^2}{1-\gamma^2}\neq \pm i$, 
      then one can check that $g\not\in\mathscr{A}$.
      A contradiction.
      Otherwise $\tfrac{1+\gamma^2}{1-\gamma^2}= \pm i$.
      It implies that $\gamma=\alpha^l$ for some integer $l=1,3,5,7$.
      We may assume $l=1$ as other cases are similar.
      In this case it is possible that $T^{-1}\mathcal{F}\cup \{g\}\subseteq\mathscr{A}$.
      As $T^{-1}=\trans{1}{1}{\gamma}{-\gamma}^{-1}H^{-1}=\trans{1}{1}{\alpha}{-\alpha}^{-1}H^{-1}$,
      it implies that
      \[
       H^{-1} \mathcal{F}
       \subseteq
       \trans{1}{1}{\alpha}{-\alpha}
       \mathscr{A}
       =
       \trans{1}{0}{0}{\alpha}
       \trans{1}{1}{1}{-1}
       \mathscr{A}
       =
       \trans{1}{0}{0}{\alpha}
       \mathscr{A}.
      \]
      Hence, $\mathcal{F}$ is $\mathscr{A}$-transformable, so $\PlHolant(\mathcal{F})$ is tractable.
      This finishes the proof. \qedhere
  \end{enumerate}
\end{proof}

Lemma~\ref{lem:dichotomy:M2} leads to the following specialization.

\begin{corollary} \label{cor:induction:M2}
 Let $f$ be a non-degenerate signature of arity $n\ge5$.
 Let $f'$ be $f$ with a self loop, and $f'$ is non-degenerate and $f'\in\mathscr{M}_2\setminus\mathscr{P}_2$.
 Then $\PlHolant(f)$ is \numP-hard unless $f$ is $\mathscr{A}$-, $\mathscr{P}$-, or $\mathscr{M}$-transformable,
 in which case $\PlHolant(f)$ is tractable.
\end{corollary}

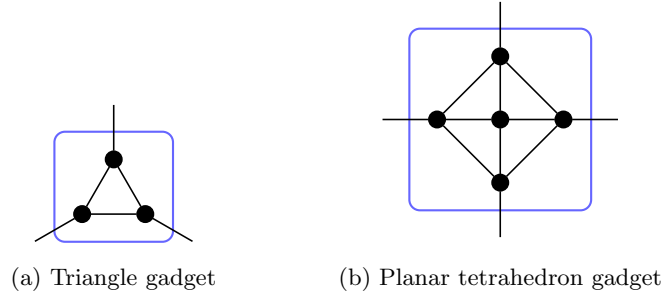
\begin{figure}[t]
 \centering
 \def\capWidth{4.5cm}
 \captionsetup[subfigure]{width=\capWidth}
 \tikzstyle{entry} = [internal, inner sep=2pt]
 \subfloat[Triangle gadget]{
  \makebox[\capWidth][c]{
   \begin{tikzpicture}[scale=\scale,transform shape,node distance=\nodeDist,semithick]
    \node [external] (0)              {};
    \node [internal] (1) [below of=0] {};
    \path (1) ++(-120:\nodeDist) node [internal] (2) {} ++(-150:\nodeDist) node [external] (3) {};
    \path (1) ++( -60:\nodeDist) node [internal] (4) {} ++( -30:\nodeDist) node [external] (5) {};
    \path (0) edge (1)
          (1) edge (2)
              edge (4)
          (2) edge (3)
              edge (4)
          (4) edge (5);
    \begin{pgfonlayer}{background}
     \node[draw=\borderColor,thick,rounded corners,fit = (1) (2) (4),inner sep=12pt] {};
%      \node[draw=\borderColor,thick,rounded corners,fit = (1) (2) (4),inner sep=8pt,transform shape=false] {};
    \end{pgfonlayer}
   \end{tikzpicture}
   \label{subfig:gadget:triangle}
  }
 }
 \quad
 \subfloat[Planar tetrahedron gadget]{
  \makebox[\capWidth][c]{
   \begin{tikzpicture}[scale=\scale,transform shape,node distance=\nodeDist,semithick]
    \node[external] (0)                    {};
    \node[internal] (1) [right       of=0] {};
    \node[internal] (2) [right       of=1] {};
    \node[internal] (3) [above       of=2] {};
    \node[external] (4) [above       of=3] {};
    \node[internal] (5) [below       of=2] {};
    \node[external] (6) [below       of=5] {};
    \node[internal] (7) [right       of=2] {};
    \node[external] (8) [right       of=7] {};
     \path (1) edge (2)
              edge (3)
              edge (5)
          (2) edge (3)
              edge (5)
              edge (7)
          (3) edge (7)
          (5) edge (7);
    \path (0) edge (1)
          (3) edge (4)
          (5) edge (6)
          (7) edge (8);
    \begin{pgfonlayer}{background}
     \node[draw=\borderColor,thick,rounded corners,fit = (1) (3) (5) (7),inner sep=12pt] {};
%      \node[draw=\borderColor,thick,rounded corners,fit = (1) (3) (5) (7),inner sep=8pt,transform shape=false] {};
    \end{pgfonlayer}
   \end{tikzpicture}
   \label{subfig:gadget:planar_tetrahedron}
  }
 }
 \caption{Two gadgets used to create a signature in $\mathscr{M}_2 \setminus \mathscr{P}_2$.}
 \label{fig:label}
\end{figure}

We can reduce the case of $f\in\mathscr{M}_3$ to the previous case.

\begin{lemma} \label{lem:dichotomy:M3}
  Let $\mathcal{F}$ be a set of signatures.
  Suppose there exists $f\in\mathcal{F}$ which is
  a non-degenerate signature of arity $n\ge 3$ in $\mathscr{M}_3$ with $H\in\mathbf{O}_2(\mathbb{C})$.
  Then $\PlHolant(\mathcal{F})$ is \numP-hard unless $\mathcal{F}\subseteq H \mathscr{M}$,
  in which case $\mathcal{F}$ is $\mathscr{M}$-transformable
  and $\PlHolant(\mathcal{F})$ is tractable.
\end{lemma}

\begin{proof}
  We first claim that $\PlHolant(\mathcal{F})$ is \numP-hard 
  unless $\mathcal{F}$ is $\mathscr{A}$-, $\mathscr{P}$-, or $\mathscr{M}$-transformable.

  By the definition of $\mathscr{M}_3$,
  we may assume that $f = \exactone{n}$ is of arity $n$
  after an orthogonal transformation $H$.
  After zero or more self loops,
  we can further assume that either $f = \exactone{3}$
  or $f = \exactone{4}$ depending on the parity of $n$.
 
  Suppose $f = \exactone{3}$.
  Consider the gadget in Figure~\ref{subfig:gadget:triangle}.
  We assign $f$ to all vertices.
  The signature of the resulting gadget is $g = [0,1,0,1]$,
  which is in $\mathscr{M}_2$ and not in $\mathscr{P}_2=\mathscr{A}_2$ by Lemma~\ref{lem:single:A2}.
  Thus, the claim follows from Lemma~\ref{lem:dichotomy:M2}.
  
  Otherwise, $f = \exactone{4}$.
  Consider the gadget in Figure~\ref{subfig:gadget:planar_tetrahedron}.
  We assign $f$ to all vertices. Note that this is a matchgate.
  The signature of the resulting gadget is $[0,2,0,1,0]$,
  which is in $\mathscr{M}_2$ and not in $\mathscr{P}_2 = \mathscr{A}_2$ by Lemma~\ref{lem:single:A2}.
  Thus, the claim follows from Lemma~\ref{lem:dichotomy:M2}.

  However, as $f\in\mathcal{F}$ and $f\in\mathscr{M}_3$,
  $\mathcal{F}$ cannot be $\mathscr{A}$- or $\mathscr{P}$-transformable
  by Lemma~\ref{lem:M3:transformation}.
  Also by Lemma~\ref{lem:M3:transformation}, 
  if $\mathcal{F}$ is $\mathscr{M}$-transformable,
  then $\mathcal{F}\subseteq HD\mathscr{M}$ or $H\trans{0}{1}{1}{0}D\mathscr{M}$
  for some diagonal matrix $D$.
  Notice that $D\in\Stab{M}$ and $\trans{0}{1}{1}{0}D\in\Stab{M}$.
  It implies that $\mathcal{F}\subseteq H\mathscr{M}$.
\end{proof}

Once again, we specialize Lemma~\ref{lem:dichotomy:M3} to our needs.

\begin{corollary} \label{cor:induction:M3}
 Let $f$ be a non-degenerate signature of arity $n\ge5$.
 Let $f'$ be $f$ with a self loop, and $f'$ is non-degenerate and $f'\in\mathscr{M}_3$.
 Then $\PlHolant(f)$ is \numP-hard unless $f$ is $\mathscr{M}$-transformable,
 in which case $\PlHolant(f)$ is tractable.
\end{corollary}

\section{Single Signature Dichotomy} \label{sec:single:dichotomy}

Theorem~\ref{thm:dic:single} is the single signature dichotomy for $\PlHolant$ problems. 

\begin{theorem} \label{thm:dic:single}
 If $f$ is a non-degenerate symmetric signature of arity $n \ge 3$ with complex weights in Boolean variables,
 then $\PlHolant(f)$ is $\SHARPP$-hard unless
 $f \in \mathscr{P}_1 \cup \mathscr{M}_2 \cup \mathscr{A}_3 \cup \mathscr{M}_3 \cup \mathscr{M}_4 \cup \mathscr{V}$,
 in which case the problem is computable in polynomial time.
\end{theorem}

We prove Theorem~\ref{thm:dic:single} by induction on the arity.
Before proceeding to the proof,
we first introduce several lemmas involved in the inductive step.

\subsection{Lemmas applied to Non-Degenerate Signatures in the Inductive Step}

The single signature dichotomy relies on the following key lemma.
The important assumption here is that $f'$ is non-degenerate.

\begin{lemma} \label{lem:selfloop}
 Suppose $f$ is a non-degenerate signature of arity $n \ge 5$.
 Let $f'$ be $f$ with a self loop.
 If $f'\in \mathscr{P}_1 \cup \mathscr{M}_2 \cup \mathscr{A}_3 \cup \mathscr{M}_3 \cup \mathscr{V}$ is non-degenerate,
 then $\PlHolant(f)$ is \numP-hard unless
 $f \in \mathscr{P}_1 \cup \mathscr{M}_2 \cup \mathscr{A}_3 \cup \mathscr{M}_3 \cup \mathscr{V}$.
\end{lemma}

Lemma~\ref{lem:selfloop} depends on several results,
each of which handles a different case.
In fact,
the proof of Lemma~\ref{lem:selfloop} is a straightforward combination of
Corollary~\ref{cor:induction:P1} (for $\mathscr{P}_1$),
Corollary~\ref{cor:induction:A3} (for $\mathscr{A}_3$),
Corollary~\ref{cor:induction:M2} (for $\mathscr{M}_2 \setminus \mathscr{P}_2$),
and Corollary~\ref{cor:induction:M3} (for $\mathscr{M}_3$) 
from Section~\ref{sec:PlCSP2},
as well as Corollary~\ref{cor:induction:P2} (for $\mathscr{P}_2$) 
and Lemma \ref{lem:induction:vanishing} (for $\mathscr{V}$),
which we will prove shortly.
These last two results handle the cases $f' \in \mathscr{P}_2$ and $f' \in \mathscr{V}$ respectively.
First we consider the case of $f'\in\mathscr{P}_2$ and show the following lemma.

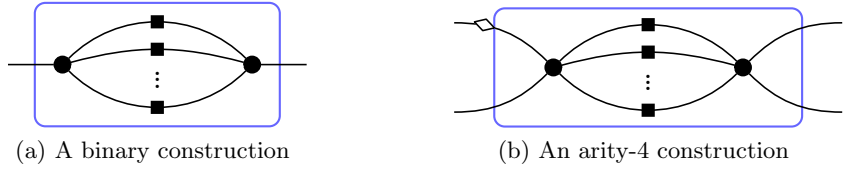
\begin{figure}[t]
 \centering
 \def\capWidth{5.5cm}
 \captionsetup[subfigure]{width=\capWidth}
 \subfloat[A binary construction]{
  \makebox[\capWidth][c]{
   \begin{tikzpicture}[scale=\scale,transform shape,node distance=\nodeDist,semithick]
    \node[external] (0)              {};
    \node[internal] (1) [right of=0] {};
    \node[external] (2) [right of=1] {};
    \node[external] (3) [right of=2] {};
    \node[internal] (4) [right of=3] {};
    \node[external] (5) [right of=4] {};
    \path (0) edge                                                                 (1)
          (1) edge[out= 45, in= 135]        node[square]     (e1) {}               (4)
              edge[out= 15, in= 165]        node[square]          {}               (4)
              edge[out=-10, in=-170, white] node[black]           {\Huge $\vdots$} (4)
              edge[out=-45, in=-135]        node[square]     (e2) {}               (4)
          (4) edge                                                                 (5);
    \begin{pgfonlayer}{background}
     \node[draw=\borderColor,thick,rounded corners,fit = (1) (4) (e1) (e2),inner xsep=12pt,inner ysep=8pt] {};
%      \node[draw=\borderColor,thick,rounded corners,fit = (1) (4) (e1) (e2),inner xsep=8pt,inner ysep=4pt,transform shape=false] {};
    \end{pgfonlayer}
  \end{tikzpicture}} \label{subfig:gadget:van-binary}}
 \qquad
 \subfloat[An arity-4 construction]{
  \makebox[\capWidth][c]{
   \begin{tikzpicture}[scale=\scale,transform shape,node distance=\nodeDist,semithick]
    \node[internal]  (0)                    {};
    \node[external]  (1) [above left  of=0] {};
    \node[external]  (2) [below left  of=0] {};
    \node[external]  (3) [left        of=1] {};
    \node[external]  (4) [left        of=2] {};
    \node[external]  (5) [right       of=0] {};
    \node[external]  (6) [right       of=5] {};
    \node[internal]  (7) [right       of=6] {};
    \node[external]  (8) [above right of=7] {};
    \node[external]  (9) [below right of=7] {};
    \node[external] (10) [right       of=8] {};
    \node[external] (11) [right       of=9] {};
    \path (3) edge[in= 135, out=   0,postaction={decorate, decoration={
                                                 markings,
                                                 mark=at position 0.4 with {\arrow[>=diamond, white] {>}; },
                                                 mark=at position 0.4 with {\arrow[>=open diamond]   {>}; } } }] (0)
          (0) edge[out=-135, in=   0]                                     (4)
              edge[out=  45, in= 135]        node[square] {}              (7)
              edge[out=  15, in= 165]        node[square] {}              (7)
              edge[out= -10, in=-170, white] node[black] {\Huge $\vdots$} (7)
              edge[out= -45, in=-135]        node[square] {}              (7)
          (7) edge[out=  45, in= 180]                                    (10)
              edge[out= -45, in= 180]                                    (11);
    \begin{pgfonlayer}{background}
     \node[draw=\borderColor,thick,rounded corners,fit = (1) (2) (8) (9)] {};
%      \node[draw=\borderColor,thick,rounded corners,fit = (1) (2) (8) (9),inner sep=2pt,transform shape=false] {};
    \end{pgfonlayer}
   \end{tikzpicture}} \label{subfig:gadget:van-arity4}}
 \caption{Two gadgets used. 
 In the normal basis, circles are assigned $f$ and squares are assigned $=_2$.
 In the $Z$ basis, circles are assigned $\hat{f}$ and squares are assigned $\neq_2$.}
\end{figure}

\begin{lemma} \label{lem:induction:P2}
 Let $f$ be a non-degenerate signature of arity $n \ge 5$.
 If $f = Z^{\otimes n} [a,1,0,\dots,0,1,b]$ for some $a,b \in \CC$, where the number of $0$'s is $n-3$.
 Then $\PlHolant(f)$ is \numP-hard.
\end{lemma}

\begin{proof}
  First we use the gadget in Figure~\ref{subfig:gadget:van-arity4},
  where we put $f$ on both vertices.
  Let the resulting signature be $h=Z^{\otimes 4}\hat{h}$.
  It is easier to calculate $\hat{h}$, that is, $h$ in the $Z$ basis.
  Indeed, $\hat{h}$ is not symmetric, but $\hat{h}$ has the following matrix representation as $n\ge 5$:
  \begin{align*}
    M_{\hat{h}}=
    \begin{bmatrix}
      0 & a & a & ab+(n-2) \\
      a & 2 & 2 & b \\
      a & 2 & 2 & b \\
      ab+(n-2) & b & b & 0
    \end{bmatrix}.
  \end{align*}
  Notice that this matrix is redundant,
  and $\det(\widetilde{M_{\hat{h}}})=-4 (n-2) (a b + n-2)$.
  If $a b\neq 2-n$, 
  then by Corollary~\ref{cor:prelim:nonsingular_compressed_hard_trans} 
  $\PlHolant(h)$ is \numP-hard, and so is $\PlHolant(f)$.
  Hence in the following we assume $ab=2-n$.

  Let $f'$ be $f$ with a self loop.
  Then apply the $Z$ transformation as follows:
     \begin{align*}
      \plholant{{=}_2}{f,f'}
      &\equiv_T \plholant{[0,1,0]}{\hat{f},\widehat{f'}}\\
      &\equiv_T \plholant{[0,1,0]}{\hat{f},[1,0,\dots,0,1]},
     \end{align*}
  where $\widehat{f'}=[1,0,\dots,0,1]$ and $\hat{f}=[a,1,0,\dots,0,1,b]$ for some $a,b\in\CC$.
  We get this expression of $\widehat{f'}$ 
  because doing a self loop commutes
  with the operation of holographic transformations.
  
  We connect $\widehat{f'}$ to $\hat{f}$ via $[0,1,0]$,
  getting $[a,2,b]$.
  Then we connect $[a,2,b]$ to $\hat{f}$ via $[0,1,0]$ again, 
  getting $\hat{g}=[ab+4,b,0,\dots,0,a,ab+4]$ of arity $n-2$.

  If $n\geq 7$, then we use the gadget in Figure~\ref{subfig:gadget:van-arity4} again,
  where we put $g$ on both vertices this time.
  We get some signature $h'$, which in $Z$ basis has the following matrix representation as $n-2\ge 5$:
  \begin{align*}
    M_{\widehat{h'}}=
    \begin{bmatrix}
                     0 & a(ab+4) & a(ab+4) & (n-4)ab+(ab+4)^2 \\
               a(ab+4) &     2ab &     2ab & b(ab+4) \\
               a(ab+4) &     2ab &     2ab & b(ab+4) \\
      (n-4)ab+(ab+4)^2 & b(ab+4) & b(ab+4) & 0
    \end{bmatrix}.
  \end{align*}
  Once again this matrix is redundant.
  It can be simplified as $ab=2-n$.
  The compressed matrix is 
  \begin{align*}
    \widetilde{M_{\widehat{h'}}}=
    \begin{bmatrix}
            0 & -2(n-6)a & -6n+28 \\
      -(n-6)a &     8-4n & -(n-6)b\\
     -6n+28  & -2(n-6)b & 0
    \end{bmatrix}.
  \end{align*}
  It is easy to compute that $\det(\widetilde{M_{\widehat{h'}}})=-8(3n-14)(ab(n-6)^2-6n^2+40n-56)
  =8(n-4)(n-2)^2(3n-14)$.
  Since $n\geq 7$, $\det(\widetilde{M_{\widehat{h'}}})>0$.
  Then by Corollary~\ref{cor:prelim:nonsingular_compressed_hard_trans} 
  $\PlHolant(h')$ is \numP-hard, and so is $\PlHolant(f)$.

  The remaining cases are $n=6$ and $n=5$.
  When $n=6$, $ab=2-n=-4$.
  Moreover, $\hat{g}$ is of arity $4$ and $\hat{g}=[ab+4,b,0,a,ab+4]=[0,b,0,a,0]$.
  We do one more self loop on $g$ via $[0,1,0]$ in the $Z$ basis,
  resulting in $\widehat{g'}=[b,0,a]$.
  Connecting $\widehat{g'}$ to $\hat{f}$ via $[0,1,0]$,
  we get $\widehat{g_1}=[a^2,a,0,b,b^2]$.
  Hence $\det(\widetilde{M_{\widehat{g_1}}})=-4a^2b^2=-64\neq 0$.
  Then by Corollary~\ref{cor:prelim:nonsingular_compressed_hard_trans} 
  $\PlHolant(g_1)$ is \numP-hard, and so is $\PlHolant(f)$.  

  At last, $n=5$ and $ab=2-n=-3$.
  We also have $\hat{g}=[ab+4,b,a,ab+4]=[1,b,a,1]$.
  One more self loop on $g$ via $[0,1,0]$ in the $Z$ basis
  results in $\widehat{g''}=[b,a]$.
  Connecting $\widehat{g''}$ to $\hat{f}$ via $[0,1,0]$,
  we get $\widehat{g_2}=[a^2+b,a,0,b,b^2+a]$.
  Hence $\det(\widetilde{M_{\widehat{g_2}}})=-2 (a^3 + 2 a^2 b^2 + b^3)=-2 (a^3 + b^3 + 18)$.
  If $a^3 + b^3 + 18\neq 0$, then we are done by Corollary~\ref{cor:prelim:nonsingular_compressed_hard_trans}.
  Otherwise $a^3 + b^3 = - 18$, 
  and we construct a binary signature $[a,0,b]$ by doing a self-loop on $\widehat{g_2}$ in $Z$ basis.
  Then we construct another unary signature 
  by connecting $\widehat{g''}=[b,a]$ to $[a,0,b]$ via $[0,1,0]$, 
  which gives $[a^2,b^2]$.
  Connecting $[a^2,b^2]$ to $\hat{f}$ via $[0,1,0]$, 
  we have another arity-4 signature $\widehat{g_3}=[ab^2+a^2,b^2,0,a^2,a^2 b +b^2]$.
  We compute $\det(\widetilde{M_{\widehat{g_3}}})=-2 (a^6 + a^5 b^2 + a^2 b^5 + b^6)=-2(a^6+b^6-162)$.
  If $a^6+b^6-162\neq 0$, again we are done by Corollary~\ref{cor:prelim:nonsingular_compressed_hard_trans}.
  Otherwise $a^6+b^6=162$.
  Together with $a^3 + b^3 = - 18$ and $a b=-3$, 
  there is no solution of $a$ and $b$.
  This finishes the proof.
\end{proof}

This lemma essentially handles the case of $f' \in \mathscr{P}_2$ due to the following corollary.

\begin{corollary} \label{cor:induction:P2}
 Suppose $f$ be a non-degenerate signature of arity $n \ge 5$.
 Let $f'$ be $f$ with a self loop.
 If $f' \in \mathscr{P}_2$ is non-degenerate,
 then $\PlHolant(f)$ is \numP-hard.
\end{corollary}

\begin{proof}
  Since $f'\in\mathscr{P}_2$, 
  we have that $f'=Z^{\otimes n-2}[1,0,\dots,0,1]$ 
  up to an orthogonal transformation $H$.
  Since $H$ does not change the complexity, 
  we may assume we are under this transformation.
  Then $f$ is of the form $Z^{\otimes n}[a,1,0,\dots,0,1,b]$.
  The claim follows by Lemma~\ref{lem:induction:P2}.
\end{proof}

The next lemma handles the case when $f'$ is a non-degenerate vanishing signature.
Its proof is partly contained in the proof of Theorem~9.1 in~\cite{CGW13}.
We include this part here for completeness.
As we shall see,
the case of $f' \in \mathscr{M}_4$ is a special case of this result.

\begin{lemma} \label{lem:induction:vanishing}
 Suppose $f$ is a non-degenerate signature of arity $n \ge 5$.
 Let $f'$ be $f$ with a self loop.
 If $f'$ is non-degenerate and vanishing,
 then $\PlHolant(f)$ is \numP-hard unless $\{f,f'\}$ is vanishing,
 in which case $\PlHolant(f)$ is tractable.
\end{lemma}

\begin{proof}
  Since $f'$ is vanishing, $f' \in \mathscr{V}^\sigma$ 
  for some $\sigma \in \{+,-\}$ by Theorem~\ref{thm:van}.
  For simplicity, assume that $f' \in \mathscr{V}^+$.
  The other case is similar.
  Let $\rd^+(f') = d - 1$, where $2 d < n$ and $d \ge 2$ since $f'$ is non-degenerate.
  Then the entries of $f'$ can be expressed as
  \[f_k' = i^k q(k),\]
  where $q(x)$ is a polynomial of degree exactly $d-1$.
  However, notice that if $f'$ satisfies some recurrence relation with characteristic polynomial $t(x)$,
  then $f$ satisfies a recurrence relation with characteristic polynomial $(x^2 + 1) t(x)$.
  In this case, $t(x) = (x-i)^{d}$.
  Then the corresponding characteristic polynomial of $f$ is $(x-i)^{d+1} (x+i)$, and thus the entries of $f$ are
  \[f_k = i^k p(k) + c (-i)^k\]
  for some constant $c$ and a polynomial $p(x)$ of degree at most $d$.
  However, the degree of $p(x)$ is exactly $d$, otherwise the polynomial $q(x)$ for $f'$ would have degree less than $d-1$.
  If $c=0$, then $\{f,f'\}$ is vanishing, the tractable case.
  Now assume $c \neq 0$, and we want to show that $\PlHolant(f)$ is $\SHARPP$-hard.

  Thus, under the transformation $Z = \frac{1}{\sqrt{2}} \left[\begin{smallmatrix} 1 & 1 \\ i & -i \end{smallmatrix}\right]$, we have
  \begin{align*}
   \plholant{{=}_2}{f}
   &\equiv_T \plholant{[1,0,1] Z^{\otimes 2}}{(Z^{-1})^{\otimes n} f}\\
   &\equiv_T \plholant{[0,1,0]}{\hat{f}},
  \end{align*}
  where $\hat{f} = [\hat{f}_0, \hat{f}_1, \dotsc, \hat{f}_d, 0, \dotsc, 0, c]$, with $\hat{f}_d \neq 0$.
  Taking a self loop in the original setting is equivalent to connecting $[0,1,0]$ to a signature after this transformation.
  Thus, doing this once on $\hat{f}$, 
  we get $\widehat{f'} = [\hat{f}_1, \dotsc, \hat{f}_{d}, 0, \dotsc, 0]$,
  corresponding to $f'$ transformed,
  and doing this $d-2$ times on $\hat{f}$,
  we get a signature $\hat{h} = [\hat{f}_{d-2}, \hat{f}_{d-1}, \hat{f}_{d}, 0, \dotsc, 0, 0 / c]$ of arity $n - 2 (d-2) = n - 2 d + 4$.
  The last entry is $c$ when $d=2$ and is 0 when $d>2$.
  As $n > 2 d$, we may do two more self loops and get $[\hat{f}_{d}, 0, \dotsc, 0]$ of arity $k = n - 2 d$.
  Now connect this signature back to $\hat{f}$ via $[0,1,0]$.
  It is the same as getting the last $n - k + 1 = 2 d + 1$ signature entries of $\hat{f}$ up to a nonzero scalar.
  We may repeat this operation zero or more times until the arity $k'$ of the resulting signature is less than or equal to $k$.
  We claim that this signature has the form $\hat{g} = [0, \dotsc, 0, c]$.
  In other words, the $k'+1$ entries of $\hat{g}$ consist of the last $c$ and $k'$ many 0's from the signature $\hat{f}$, 
  all appearing after $\hat{f}_{d}$.
  This is because there are $n - d - 1$ many $0$ entries in the signature $\hat{f}$ after $\hat{f}_{d}$,
  and $n - d - 1 \ge k \ge k'$.

  Having both $[\hat{f}_{d}, 0, \dotsc, 0]$ of arity $k$ and $\hat{g} = [0, \dotsc, 0, c]$ of arity $k'$ in the Z basis is equivalent to
  having both $[1,i]^{\otimes k}$ and $[1,-i]^{\otimes k'}$ in the standard basis.
  If $k > k'$,
  then we can connect $[1,-i]^{\otimes k'}$ to $[1,i]^{\otimes  k}$ and get $[1,i]^{\otimes (k-k')}$.
  Replacing $k$ by $k-k'$,
  we can repeat this process until the new $k \le k'$.
  If the new $k < k'$,
  then we can continue as in the subtractive Euclid algorithm.
  We continue this procedure and eventually we get $[1,i]^{\otimes t}$ and  $[1,-i]^{\otimes  t}$,
  where $t = \gcd(k,k')$
   %, where $k = n - 2 d$ and $k' \le k$,
   %as defined in the previous paragraph.
  Now putting $k / t$ many copies of $[1,-i]^{\otimes  t}$ together,
  we get $[1,-i]^{\otimes k}$.

  In the transformed setting,
  $[1,-i]^{\otimes k}$ is $[0, \dotsc, 0, 1]$ of arity $k$.
  Then we connect this back to $\hat{h}$ via $[0,1,0]$.
  Doing this is the same as forcing $k$ connected edges of $\hat{h}$ be assigned $0$,
  because $[0,1,0]$ flips $[0, \dotsc, 0, 1]$.
  Thus we get a signature of arity $n - 2 d + 4 - k = 4$,
  which is $[\hat{f}_{d-2}, \hat{f}_{d-1}, \hat{f}_{d}, 0, 0]$.
  Note that the last entry is~$0$ (and not $c$),
  because $k \ge 1$ and $\arity(\hat{h})\ge 5$.

  However,
  $\PlHolant([0,1,0] \: | \: [\hat{f}_{d-2}, \hat{f}_{d-1}, \hat{f}_{d}, 0, 0])$
  is equivalent to $\PlHolant([0,1,0] \: | \: [0,0,1,0,0])$ when $\hat{f}_{d} \neq 0$,
  which is transformed back by $Z$ to $\PlHolant([3,0,1,0,3])$.
  This is the Eulerian Orientation problem on planar $4$-regular graphs
  and is \numP-hard by Theorem~\ref{thm:PlHolant:arity34}.
\end{proof}

\subsection{Lemmas applied to Degenerate Signatures in the Inductive Step}

\begin{figure}[p]
 \centering
 \def\capWidth{7.5cm}
 \captionsetup[subfigure]{width=\capWidth}
 \subfloat[$({\ne}_2 \mid \lbrack0,1,0,0,0\rbrack,\lbrack0,0,0,1,0\rbrack,\hat{g})$-gate on right side]{
  \makebox[\capWidth][c]{
   \begin{tikzpicture}[scale=1.2 * \scale,transform shape,node distance=1.5 * \nodeDist,semithick]
    \node[pentagon] (0)                    {};
    \node[internal] (1) [above  left of=0] {};
    \node[triangle] (2) [above right of=0] {};
    \node[triangle] (3) [below  left of=0] {};
    \node[internal] (4) [below right of=0] {};
    \node[external] (5) [above  left of=1] {};
    \node[external] (6) [above right of=2] {};
    \node[external] (7) [below  left of=3] {};
    \node[external] (8) [below right of=4] {};
    \path (0) edge node[square] {} (1)
              edge node[square] {} (2)
              edge node[square] {} (3)
              edge node[square] {} (4)
          (1) edge node[square] {} (2)
              edge[postaction={decorate, decoration={
                                          markings,
                                          mark=at position 0.52 with {\arrow[>=diamond, white] {>}; },
                                          mark=at position 0.52 with {\arrow[>=open diamond]   {>}; } } }] (5)
          (2) edge node[square] {} (4)
              edge (6)
          (3) edge node[square] {} (1)
              edge (7)
          (4) edge node[square] {} (3)
              edge (8);
    \begin{pgfonlayer}{background}
     \node[draw=\borderColor,thick,rounded corners,fit = (1) (4),inner sep=10pt] {};
%      \node[draw=\borderColor,thick,rounded corners,fit = (1) (4),inner sep=7pt,transform shape=false] {};
    \end{pgfonlayer}
  \end{tikzpicture}} \label{subfig:gadget:valid-bipartite-gadet}}
 \qquad
 \subfloat[Simpler construction with the same signature]{
  \makebox[\capWidth][c]{
   \begin{tikzpicture}[scale=1.2 * \scale,transform shape,node distance=1.5 * \nodeDist,semithick]
    \node[diamond,fill,draw,aspect=0.5]  (0) {};
    \node[internal] (1) [above  left of=0] {};
    \node[internal] (2) [above right of=0] {};
    \node[internal] (3) [below  left of=0] {};
    \node[internal] (4) [below right of=0] {};
    \node[external] (5) [above  left of=1] {};
    \node[external] (6) [above right of=2] {};
    \node[external] (7) [below  left of=3] {};
    \node[external] (8) [below right of=4] {};
    \path (0) edge (1)
              edge (2)
              edge (3)
              edge (4)
          (1) edge (2)
              edge[postaction={decorate, decoration={
                                          markings,
                                          mark=at position 0.72 with {\arrow[>=diamond, white] {>}; },
                                          mark=at position 0.72 with {\arrow[>=open diamond]   {>}; } } }] (5)
          (2) edge (4)
              edge node[square, pos=0.3] (e1) {} (6)
          (3) edge (1)
              edge node[square, pos=0.3] (e2) {} (7)
          (4) edge (3)
              edge (8);
    \begin{pgfonlayer}{background}
     \node[draw=\borderColor,thick,rounded corners,fit = (e1) (e2),inner sep=4pt] {};
%      \node[draw=\borderColor,thick,rounded corners,fit = (e1) (e2),inner sep=3pt,transform shape=false] {};
     \node[draw=\borderColor,thick,rounded corners,dashed,fit = (2) (3),inner sep=3pt] {};
%      \node[draw=\borderColor,thick,rounded corners,dashed,fit = (2) (3),inner sep=3pt,transform shape=false] {};
    \end{pgfonlayer}
   \end{tikzpicture}
  }
  \label{subfig:gadget:bipartite-simplified}
 }
 \caption{Two gadgets with the same signature used in Lemma~\ref{lem:PM-InvPM-g}.}
 \label{fig:two_gadgets_same_sig}
\end{figure}
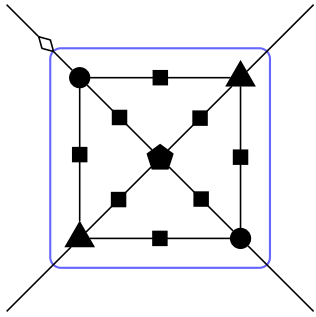
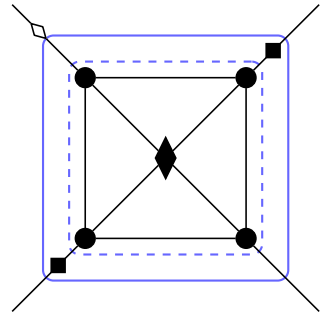

\begin{figure}[p]
 \centering
 \def\capWidth{6.5cm}
 \captionsetup[subfigure]{width=\capWidth}
 \tikzstyle{entry} = [internal, inner sep=2pt]
 \subfloat[Negating the second and fourth inputs]{
  \begin{tikzpicture}[scale=\scale,transform shape,node distance=1.7 * \nodeDist,semithick]
   \node[internal]  (0)                    {};
   \node[external]  (1) [above  left of=0] {};
   \node[external]  (2) [above right of=0] {};
   \node[external]  (3) [below  left of=0] {};
   \node[external]  (4) [below right of=0] {};
   \node[external]  (5) [      right of=0] {};
   \node[external]  (6) [      right of=5] {};
   \node[internal]  (7) [      right of=6] {};
   \node[external]  (8) [above  left of=7] {};
   \node[external]  (9) [above right of=7] {};
   \node[external] (10) [below  left of=7] {};
   \node[external] (11) [below right of=7] {};
   \path (0) edge[postaction={decorate, decoration={
                                         markings,
                                         mark=at position 0.25 with {\arrow[>=diamond,white] {>}; },
                                         mark=at position 0.25 with {\arrow[>=open diamond]  {>}; },
                                         mark=at position 0.75 with {\arrow[>=diamond,white] {>}; },
                                         mark=at position 0.75 with {\arrow[>=open diamond]  {>}; } } }] (1)
             edge node[square, pos=0.3, opacity=0] (e1) {} (2)
             edge node[square, pos=0.3, opacity=0] (e2) {} (3)
             edge (4)
    (5.west) edge[->, >=stealth] (6.east)
         (7) edge[postaction={decorate, decoration={
                                         markings,
                                         mark=at position 0.25 with {\arrow[>=diamond,white] {>}; },
                                         mark=at position 0.25 with {\arrow[>=open diamond]  {>}; },
                                         mark=at position 0.75 with {\arrow[>=diamond,white] {>}; },
                                         mark=at position 0.75 with {\arrow[>=open diamond]  {>}; } } }] (8)
             edge node[square, pos=0.3] (e3) {} (9)
             edge node[square, pos=0.3] (e4) {} (10)
             edge (11);
   \begin{pgfonlayer}{background}
    \node[draw=\borderColor,thick,rounded corners,fit = (e1) (e2),inner sep=6pt] {};
%     \node[draw=\borderColor,thick,rounded corners,fit = (e1) (e2),inner sep=4pt,transform shape=false] {};
    \node[draw=\borderColor,thick,rounded corners,fit = (e3) (e4),inner sep=6pt] {};
%     \node[draw=\borderColor,thick,rounded corners,fit = (e3) (e4),inner sep=4pt,transform shape=false] {};
   \end{pgfonlayer}
  \end{tikzpicture}}
 \qquad
 \subfloat[Movement of even Hamming weight entries]{
  \makebox[\capWidth][c]{
   \begin{tikzpicture}[scale=\scale,transform shape,>=stealth,node distance=\nodeDist,semithick]
    \node[entry] (11)               {};
    \node[entry] (12) [right of=11] {};
    \node[entry] (13) [right of=12] {};
    \node[entry] (14) [right of=13] {};
    \node[entry] (21) [below of=11] {};
    \node[entry] (22) [right of=21] {};
    \node[entry] (23) [right of=22] {};
    \node[entry] (24) [right of=23] {};
    \node[entry] (31) [below of=21] {};
    \node[entry] (32) [right of=31] {};
    \node[entry] (33) [right of=32] {};
    \node[entry] (34) [right of=33] {};
    \node[entry] (41) [below of=31] {};
    \node[entry] (42) [right of=41] {};
    \node[entry] (43) [right of=42] {};
    \node[entry] (44) [right of=43] {};
    \node[external] (nw) [above left  of=11] {};
    \node[external] (ne) [above right of=14] {};
    \node[external] (sw) [below left  of=41] {};
    \node[external] (se) [below right of=44] {};
    \path (11) edge[<->] (23)
          (22) edge[<->] (14)
          (32) edge[<->] (44)
          (41) edge[<->] (33);
    \path (nw.west) edge (sw.west)
          (ne.east) edge (se.east)
          (nw.west) edge (nw.east)
          (sw.west) edge (sw.east)
          (ne.west) edge (ne.east)
          (se.west) edge (se.east);
   \end{tikzpicture}}}
 \caption{The movement of the even Hamming weight entries in the signature matrix of a quaternary signature under the negation of the second and fourth inputs
 (i.e.~the square vertices are assigned $[0,1,0]$).}
 \label{fig:negate_second_fourth}
\end{figure}
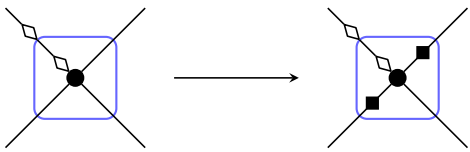
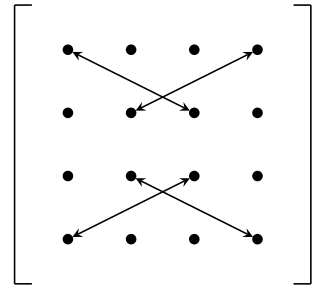

\begin{figure}[p]
 \centering
 \def\capWidth{8cm}
 \subfloat[Gadget that realizes a partial crossover]{
  \makebox[\capWidth][c]{
   \begin{tikzpicture}[scale=\scale,transform shape,node distance=2 * \nodeDist,semithick]
    \node[internal]  (0)                    {};
    \node[external]  (1) [above  left of=0] {};
    \node[internal]  (2) [      right of=0] {};
    \node[external]  (3) [above right of=2] {};
    \node[internal]  (4) [below       of=2] {};
    \node[external]  (5) [below right of=4] {};
    \node[internal]  (6) [       left of=4] {};
    \node[external]  (7) [below  left of=6] {};
    \path (0) edge (1)
              edge[bend  left]           node[square] (e1)               {} (2)
              edge[out=8, in=172, white] node[black]       {\Huge $\vdots$} (2)
              edge[bend right]           node[square]                    {} (2)
          (2) edge (3)
              edge[bend  left]           node[square, pos=0.2]           {} (4)
              edge[bend  left]           node[triangle] (e2)             {} (4)
              edge[bend  left]           node[square, pos=0.8]           {} (4)
              edge[bend right]           node[square]                    {} (4)
          (4) edge (5)
              edge[bend  left]           node[square] (e3)               {} (6)
              edge[out=172, in=8, white] node[black]       {\Huge $\vdots$} (6)
              edge[bend right]           node[square]                    {} (6)
          (6) edge (7)
              edge[bend  left]           node[square, pos=0.2]           {} (0)
              edge[bend  left]           node[triangle] (e4)             {} (0)
              edge[bend  left]           node[square, pos=0.8]           {} (0)
              edge[bend right]           node[square]                    {} (0);
    \begin{pgfonlayer}{background}
     \node[draw=\borderColor,thick,rounded corners,fit = (e1) (e2) (e3) (e4),inner sep=8pt] {};
%      \node[draw=\borderColor,thick,rounded corners,fit = (e1) (e2) (e3) (e4),inner sep=6pt,transform shape=false] {};
    \end{pgfonlayer}
   \end{tikzpicture}
  }
  \label{subfig:gadget:partial-crossover}
 }
 \subfloat[Gadget with a useful signature matrix]{
  \makebox[\capWidth][c]{
   \begin{tikzpicture}[scale=\scale,transform shape,node distance=\nodeDist,semithick]
    \node[internal]  (0)                     {};
    \node[external]  (1) [above  left of=0]  {};
    \node[external]  (2) [below  left of=0]  {};
    \node[external]  (3) [       left of=1]  {};
    \node[external]  (4) [       left of=2]  {};
    \node[external]  (5) [above  left of=3]  {};
    \node[external]  (8) [      right of=0]  {};
    \node[external]  (9) [      right of=8]  {};
    \node[internal] (10) [      right of=9]  {};
    \node[external] (11) [above right of=10] {};
    \node[external] (12) [below right of=10] {};
    \node[external] (13) [      right of=11] {};
    \node[external] (14) [      right of=12] {};
    \path (3) edge[out=   0, in= 135, postaction={decorate, decoration={
                                                             markings,
                                                             mark=at position 0.4   with {\arrow[>=diamond,white] {>}; },
                                                             mark=at position 0.4   with {\arrow[>=open diamond]  {>}; },
                                                             mark=at position 0.999 with {\arrow[>=diamond,white] {>}; },
                                                             mark=at position 1     with {\arrow[>=open diamond]  {>}; } } }] (0)
          (0) edge[out=-135, in=   0]  (4)
         (10) edge[out=  45, in= 180] (13)
              edge[out= -45, in= 180] (14);
    \path (0) edge[out=  45, in= 135] node[square] {} (10)
              edge[out= -45, in=-135, postaction={decorate, decoration={
                                                             markings,
                                                             mark=at position 0.999 with {\arrow[>=diamond,white] {>}; },
                                                             mark=at position 0.999 with {\arrow[>=open diamond]  {>}; } } }] node[square] {} (10);
    \begin{pgfonlayer}{background}
     \node[draw=\borderColor,thick,rounded corners,fit = (1) (2) (11) (12)] {};
%      \node[draw=\borderColor,thick,rounded corners,fit = (1) (2) (11) (12),inner sep=2pt,transform shape=false] {};
    \end{pgfonlayer}
   \end{tikzpicture}
  }
  \label{subfig:gadget:symmetric_signature_matrix}
 }
 \caption{Two quaternary gadgets used in the proof of Lemma~\ref{lem:PM-InvPM-g} and~\ref{lem:01000b}.}
\end{figure}
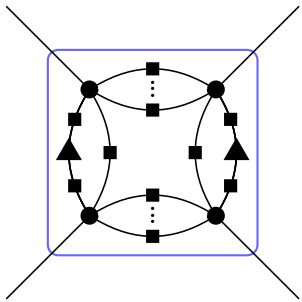
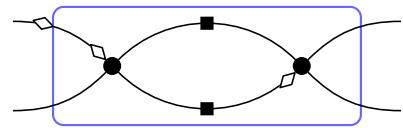

Lemma~\ref{lem:selfloop} does not solve the case when $f'$ is degenerate.
In general, when $f'$ is degenerate, the inductive step is straightforward 
unless $f'$ is also vanishing.
Lemma~\ref{lem:a1000b} and~\ref{lem:01000b} are the two missing pieces to this end.

\begin{lemma} \label{lem:a1000b}
 Let $a, b \in \mathbb{C}$.
 Suppose $f$ is a signature of the form $\trans{1}{1}{i}{-i}^{\otimes n}[a,1,0,\dotsc,0,b]$ with arity $n \ge 3$.
 If $a b \neq 0$,
 then $\PlHolant(f)$ is \numP-hard.
\end{lemma}

\begin{proof}
 We prove by induction on $n$. 
 For $n=3$ or $4$, it follows from Lemma~\ref{lem:Zg:not-transformable} and Theorem~\ref{thm:PlHolant:arity34}
 that $\PlHolant(f)$ is \numP-hard.

 Now assume $n \ge 5$.
 Under a holographic transformation by $Z = \trans{1}{1}{i}{-i}$,
 we have
 \begin{align*}
   \plholant{{=}_2}{f}
   &\equiv_T \plholant{[1,0,1] Z^{\otimes 2}}{(Z^{-1})^{\otimes n} f}\\
   &\equiv_T \plholant{[0,1,0]}{\hat{f}},
 \end{align*}
 where $\hat{f}=[a, 1, 0, \dotsc, 0, b]$.
 Now consider the gadget in Figure~\ref{subfig:gadget:van-binary} with $\hat{f}$ assigned to both vertices.
 This gadget has the binary signature $\widehat{g_1}=[0, ab, 2b]$, which is equivalent to $[0, a, 2]$ since $b \neq 0$.
 Translating back by $Z$ to the original setting, this signature is $g_1 = [a + 1, -i, a - 1]$.
 This can be verified as
 \[
   \begin{bmatrix} 1 & 1 \\ i & -i \end{bmatrix}
   \begin{bmatrix} 0 & a \\ a & 2  \end{bmatrix}
   \transpose{\begin{bmatrix} 1 & 1 \\ i & -i \end{bmatrix}}
   =
   2
   \begin{bmatrix} a + 1 & -i \\ -i & a - 1 \end{bmatrix}.
 \]
 By the form of $\widehat{g_1} =[0, ab, 2b]$ and $b \neq 0$,
 it follows from Lemma~\ref{lem:prelim:vanishing_form_in_Z_basis} 
 that $g_1 \not\in \mathscr{R}^+_2$.
 Moreover, since $a\neq 0$, $g_1$ is non-degenerate.

 Doing a self loop on $f$ yields $f'=Z^{\otimes n-2}[1,0,\dots,0]$.
 Connecting $f'$ back to $f$, we get a binary signature $g_2=Z^{\otimes 2}[0,0,b]$.
 Once again we connect $g_2$ to $f$, 
 the resulting signature is $h=Z^{\otimes n-2}[a,1,0,\dots,0]$ of arity $n-2\ge 3$ 
 up to the constant factor of $b\neq 0$.

 Notice that $h$ is non-degenerate and $h\in\mathscr{V}^+$.
 By Lemma~\ref{lem:van:bin}, 
 $\PlHolant(h, g_1)$ is $\numP$-hard, 
 hence $\PlHolant(f)$ is also $\numP$-hard.  
\end{proof}

\begin{figure}[t]
 \centering
 \captionsetup[subfigure]{labelformat=empty}
 \subfloat[$N_1$]{
  \begin{tikzpicture}[scale=\scale,transform shape,node distance=\nodeDist,semithick]
   \node[external] (0)                    {};
   \node[external] (1) [right       of=0] {};
   \node[internal] (2) [below right of=1] {};
   \node[external] (3) [below left  of=2] {};
   \node[external] (4) [left        of=3] {};
   \node[external] (5) [above right of=2] {};
   \node[external] (6) [right       of=5] {};
   \node[external] (7) [below right of=2] {};
   \node[external] (8) [right       of=7] {};
   \path (0) edge[out=   0, in=135, postaction={decorate, decoration={
                                                           markings,
                                                           mark=at position 0.45  with {\arrow[>=diamond, white] {>}; },
                                                           mark=at position 0.45  with {\arrow[>=open diamond]   {>}; },
                                                           mark=at position 0.999 with {\arrow[>=diamond, white] {>}; },
                                                           mark=at position 1.0   with {\arrow[>=open diamond]   {>}; } } }] (2)
         (2) edge[out=-135, in=  0] (4)
             edge[out=  45, in=180] (6)
             edge[out= -45, in=180] (8);
   \begin{pgfonlayer}{background}
    \node[draw=\borderColor,thick,rounded corners,fit = (1) (3) (5) (7),inner sep=0pt] {};
%     \node[draw=\borderColor,thick,rounded corners,fit = (1) (3) (5) (7),inner sep=0pt,transform shape=false] {};
   \end{pgfonlayer}
  \end{tikzpicture}
 }
 \qquad
 \subfloat[$N_2$]{
  \begin{tikzpicture}[scale=\scale,transform shape,node distance=\nodeDist,semithick]
   \node[external]  (0)                    {};
   \node[external]  (1) [right       of=0] {};
   \node[internal]  (2) [below right of=1] {};
   \node[external]  (3) [below left  of=2] {};
   \node[external]  (4) [left        of=3] {};
   \node[external]  (5) [right       of=2] {};
   \node[internal]  (6) [right       of=5] {};
   \node[external]  (7) [above right of=6] {};
   \node[external]  (8) [right       of=7] {};
   \node[external]  (9) [below right of=6] {};
   \node[external] (10) [right       of=9] {};
   \path (0) edge[out=   0, in=135, postaction={decorate, decoration={
                                                           markings,
                                                           mark=at position 0.45  with {\arrow[>=diamond, white] {>}; },
                                                           mark=at position 0.45  with {\arrow[>=open diamond]   {>}; },
                                                           mark=at position 0.999 with {\arrow[>=diamond, white] {>}; },
                                                           mark=at position 1.0   with {\arrow[>=open diamond]   {>}; } } }] (2)
         (2) edge[out=-135, in=  0]  (4)
             edge[bend left,        postaction={decorate, decoration={
                                                           markings,
                                                           mark=at position 0.999 with {\arrow[>=diamond, white] {>}; },
                                                           mark=at position 1.0   with {\arrow[>=open diamond]   {>}; } } }] node[square] {} (6)
             edge[bend right] node[square] {} (6)
         (6) edge[out=  45, in=180]  (8)
             edge[out= -45, in=180] (10);
   \begin{pgfonlayer}{background}
    \node[draw=\borderColor,thick,rounded corners,fit = (1) (3) (7) (9),inner sep=0pt] {};
%     \node[draw=\borderColor,thick,rounded corners,fit = (1) (3) (7) (9),inner sep=0pt,transform shape=false] {};
   \end{pgfonlayer}
  \end{tikzpicture}
 }
 \qquad
 \subfloat[$N_{s+1}$]{
  \begin{tikzpicture}[scale=\scale,transform shape,node distance=\nodeDist,semithick]
   \node[external]  (0)                     {};
   \node[external]  (1) [above left  of=0]  {};
   \node[external]  (2) [below left  of=0]  {};
   \node[external]  (3) [below left  of=1]  {};
   \node[external]  (4) [below left  of=3]  {};
   \node[external]  (5) [above left  of=3]  {};
   \node[external]  (6) [left        of=4]  {};
   \node[external]  (7) [left        of=5]  {};
   \node[external]  (8) [right       of=0]  {};
   \node[internal]  (9) [right       of=8]  {};
   \node[external] (10) [above right of=9]  {};
   \node[external] (11) [below right of=9]  {};
   \node[external] (12) [right       of=10] {};
   \node[external] (13) [right       of=11] {};
   \path let
          \p1 = (1),
          \p2 = (2)
         in
          node[external] at (\x1, \y1 / 2 + \y2 / 2) {\Huge $N_s$};
   \path let
          \p1 = (0)
         in
          node[external] (14) at (\x1 + 2, \y1 + 10) {};
   \path let
          \p1 = (0)
         in
          node[external] (15) at (\x1 + 2, \y1 - 10) {};
   \path let
          \p1 = (3)
         in
          node[external] (16) at (\x1 - 2, \y1 + 10) {};
   \path let
          \p1 = (3)
         in
          node[external] (17) at (\x1 - 2, \y1 - 10) {};
   \path (7) edge[out=   0, in=135, postaction={decorate, decoration={
                                                           markings,
                                                           mark=at position 0.48  with {\arrow[>=diamond, white] {>}; },
                                                           mark=at position 0.48  with {\arrow[>=open diamond]   {>}; },
                                                           mark=at position 0.999 with {\arrow[>=diamond, white] {>}; },
                                                           mark=at position 1.0   with {\arrow[>=open diamond]   {>}; } } }] (16)
        (17) edge[out=-135, in=  0]  (6)
        (14) edge[out=  35, in=135, postaction={decorate, decoration={
                                                           markings,
                                                           mark=at position 0.999 with {\arrow[>=diamond, white] {>}; },
                                                           mark=at position 1.0   with {\arrow[>=open diamond]   {>}; } } }] node[square, pos=0.45] {} (9)
         (9) edge[out=-135, in=-35] node[square, pos=0.55] {} (15)
             edge[out=  45, in=180] (12)
             edge[out= -45, in=180] (13);
   \begin{pgfonlayer}{background}
    \node[draw=\borderColor,thick,densely dashed,rounded corners,fit = (0) (1.south) (2.north) (3),inner sep=0pt] {};
%     \node[draw=\borderColor,thick,densely dashed,rounded corners,fit = (0) (1.south) (2.north) (3),inner sep=0pt,transform shape=false] {};
    \node[draw=\borderColor,thick,rounded corners,fit = (4) (5) (10) (11),inner sep=0pt] {};
%     \node[draw=\borderColor,thick,rounded corners,fit = (4) (5) (10) (11),inner sep=0pt,transform shape=false] {};
   \end{pgfonlayer}
  \end{tikzpicture}
 }
 \caption{Linear recursive construction used for interpolation in a nonstandard basis.}
 \label{fig:gadget:arity4:linear_interpolation}
\end{figure}
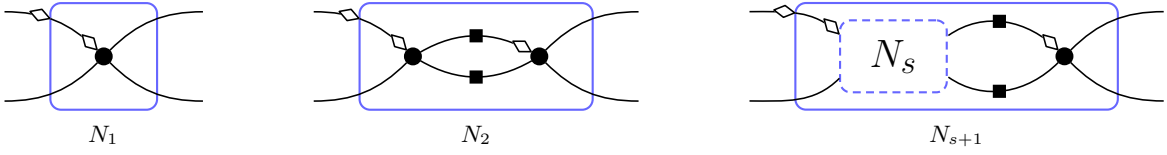

The next case uses the following technical lemma.
It is also applied more than once in Section~\ref{sec:mixing}.

\begin{lemma} \label{lem:PM-InvPM-g}
  Let $\hat{g}$ be the arity~$4$ signature whose matrix is
  \begin{equation}
    M_{\hat{g}} =
    \begin{bmatrix}
      0 & 0 & 0 & 0 \\
      0 & 0 & 1 & 0 \\
      0 & 1 & 0 & 0 \\
      0 & 0 & 0 & 0
    \end{bmatrix}.
    \label{eqn:signature-matrix:partial-crossover}
  \end{equation}
  Then \plholant{{\neq}_2}{[0,1,0,0,0],[0,0,0,1,0],\hat{g}} is \numP-hard.
\end{lemma}

\begin{proof}
  Consider the gadget in Figure~\ref{subfig:gadget:valid-bipartite-gadet}.
  We assign $[0,0,0,1,0]$ to the triangle vertices,
  $[0,1,0,0,0]$ to the circle vertices,
  $\hat{g}$ to the pentagon vertex,
  and $[0,1,0]$ to the square vertices.
Let $\hat{h}$ be the signature of this gadget.
  By adding two more disequality signatures and then grouping appropriately,
  it is clear that the gadget in Figure~\ref{subfig:gadget:bipartite-simplified} 
  has the same signature of the gadget in Figure~\ref{subfig:gadget:valid-bipartite-gadet},
  where the circle vertices are still assigned $[0,1,0,0,0]$,
  the square vertices are still assigned $[0,1,0]$,
  and the diamond vertex is assigned the quaternary equality signature.
  To compute the signature $\hat{h}$,
  first compute the signature $\hat{h}'$ of the inner gadget enclosed by the dashed line,
  which has signature matrix
  \[  M_{\hat{h}'} =
    \begin{bmatrix}
      3 & 0 & 0 & 1 \\
      0 & 1 & 0 & 0 \\
      0 & 0 & 1 & 0 \\
      1 & 0 & 0 & 1
    \end{bmatrix}.
    \text{\quad Then by Figure~\ref{fig:negate_second_fourth},
    the signature matrix of $\hat{h}$ is }
    M_{\hat{h}} =
    \begin{bmatrix}
      0 & 0 & 0 & 1 \\
      0 & 1 & 3 & 0 \\
      0 & 1 & 1 & 0 \\
      1 & 0 & 0 & 0
    \end{bmatrix}. \]
  One more gadget before we finish the proof using interpolation.
  Consider the gadget in Figure~\ref{subfig:gadget:symmetric_signature_matrix}.
  We assign $\hat{h}$ to the circle vertices and $[0,1,0]$ to the square vertices.
  The signature of the resulting gadget is $\hat{r}$ with signature matrix
$M_{\hat{r}}$
(see Figure~\ref{fig:rotate_asymmetric_signature} for the signature of a
 rotated copy of $\hat{h}$ that appears as the second circle vertex in 
Figure~\ref{subfig:gadget:symmetric_signature_matrix}),
where
  \[  M_{\hat{r}} =
    \begin{bmatrix}
      0 & 0 & 0 & 1 \\
      0 & 1 & 3 & 0 \\
      0 & 1 & 1 & 0 \\
      1 & 0 & 0 & 0
    \end{bmatrix}
    \left(
    \begin{bmatrix}
      0 & 1 \\
      1 & 0
    \end{bmatrix}
    \otimes
    \begin{bmatrix}
      0 & 1 \\
      1 & 0
    \end{bmatrix}
    \right)
    \begin{bmatrix}
      0 & 0 & 0 & 1 \\
      0 & 1 & 1 & 0 \\
      0 & 3 & 1 & 0 \\
      1 & 0 & 0 & 0
    \end{bmatrix}
    =
    \begin{bmatrix}
      0 & 0 & 0 & 1 \\
      0 & 6 & 4 & 0 \\
      0 & 4 & 2 & 0 \\
      1 & 0 & 0 & 0
    \end{bmatrix}. \]
  Consider an instance $\Omega$ of $\plholant{{\neq}_2}{\mathcal{F} \union \{\hat{r}'\}}$ with $\hat{r} \in \mathcal{F}$,
  where the signature matrix of $\hat{r}'$ is
  \[ M_{\hat{r}'} =
    \begin{bmatrix}
      0 & 0 & 0 & 1 \\
      0 & 3 & 1 & 0 \\
      0 & 1 & 1 & 0 \\
      1 & 0 & 0 & 0
    \end{bmatrix}. \]
  Suppose that $\hat{r}'$ appears $n$ times in $\Omega$.
  We construct from $\Omega$ a sequence of instances $\Omega_s$ of $\plholant{{\neq}_2}{\mathcal{F}}$ indexed by $s \ge 1$.
  We obtain $\Omega_s$ from $\Omega$ by replacing each occurrence of $\hat{r}'$
  with the gadget $N_s$ in Figure~\ref{fig:gadget:arity4:linear_interpolation}
  with $\hat{r}$ assigned to the circle vertices and $[0,1,0]$ assigned to the square vertices.
  In $\Omega_s$, the edge corresponding to the $i$th significant index bit of $N_s$ connects to
  the same location as the edge corresponding to the $i$th significant index bit of $\hat{r}'$ in $\Omega$.

  We can express the signature matrix of $N_s$ as
  \[ M_{N_s} = X (X M_{\hat{r}})^s = X P \diag\left(1, 4 + 2 \sqrt{3}, 4 - 2 \sqrt{3}, 1\right)^s P^{-1}, \]
  where
  \[ X =
    \begin{bmatrix}
      0 & 0 & 0 & 1 \\
      0 & 0 & 1 & 0 \\
      0 & 1 & 0 & 0 \\
      1 & 0 & 0 & 0
    \end{bmatrix}
    \qquad \text{and} \qquad
    P =
    \begin{bmatrix}
      1 &        0 &         0 & 0 \\
      0 &        1 &         1 & 0 \\
      0 & \sqrt{3} & -\sqrt{3} & 0 \\
      0 &        0 &         0 & 1
    \end{bmatrix}. \] 
  Since $M_{\hat{r}'} = X P \diag\left(1, 1 + \sqrt{3},  1 - \sqrt{3}, 1\right) P^{-1}$,
  we can view our construction of $\Omega_s$ 
  as first replacing $M_{\hat{r}'}$ with $X P \diag\left(1, 1 + \sqrt{3},  1 - \sqrt{3}, 1\right) P^{-1}$,
  which does not change the Holant value, 
  and then replacing the diagonal matrix with the diagonal matrix $\diag\left(1, 4 + 2 \sqrt{3}, 4 - 2 \sqrt{3}, 1\right)^s$.

 We stratify the assignments in $\Omega$ based on the assignments to the $n$ occurrences of the signature whose signature matrix is the diagonal matrix
 \begin{equation} \label{eqn:rPrime_jnf_signature_matrix}
  \begin{bmatrix}
   1 &            0 &            0 & 0 \\
   0 & 1 + \sqrt{3} &            0 & 0 \\
   0 &            0 & 1 - \sqrt{3} & 0 \\
   0 &            0 &            0 & 1
  \end{bmatrix}.
 \end{equation}
 We only need to consider the assignments that assign
 \begin{itemize}
  \item $i$ many times the bit patterns $0000$ or $1111$,
  \item $j$ many times the bit pattern  $0110$, and
  \item $k$ many times the bit pattern  $1001$,
 \end{itemize}
 since any other assignment contributes a factor of~$0$.
 Let $c_{ijk}$ be the sum over all such assignments of the products of evaluations of all signatures
 (including the signatures corresponding to the signature matrices $X$, $P$, and $P^{-1}$)
 in $\Omega$ except for signature corresponding to the signature matrix in~(\ref{eqn:rPrime_jnf_signature_matrix}).
 Then
 \[
  \Holant_\Omega = \sum_{i + j + k = n} \left(1 + \sqrt{3}\right)^j \left(1 - \sqrt{3}\right)^k c_{ijk}
 \]
 and the value of the Holant on $\Omega_s$, for $s \ge 1$, is
 \[
  \Holant_{\Omega_s}
  = \sum_{i + j + k = n} \left(\left(4 + 2 \sqrt{3}\right)^j \left(4 - 2 \sqrt{3}\right)^k\right)^s c_{ijk}
  = \sum_{i + j + k = n} \left(\left(4 + 2 \sqrt{3}\right)^{j-k} 4^k\right)^s c_{ijk}.
 \]
 We argue that this Vandermonde system has full rank,
 which is to say that $\left(4 + 2 \sqrt{3}\right)^{j-k} 4^k \ne \left(4 + 2 \sqrt{3}\right)^{j'-k'} 4^{k'}$ unless $(j,k) = (j',k')$.
 If $\left(4 + 2 \sqrt{3}\right)^{j-k} 4^k = \left(4 + 2 \sqrt{3}\right)^{j'-k'} 4^{k'}$,
 then we have $\left(4 + 2 \sqrt{3}\right)^{j-k - (j'-k')} 4^{k-k'} = 1$.
 Since any nonzero integer power of $4 + 2 \sqrt{3}$ is not rational,
 we must have $j-k = j'-k'$. And in this case, $4^{k-k'} = 1$, and hence
$k = k'$ and $j = j'$.
 
 Therefore,
 we can solve for the unknown $c_{ijk}$'s and obtain the value of $\Holant_\Omega$.
 Then after a counterclockwise rotation of $\hat{r}'$ (c.f.\ Figure~\ref{fig:rotate_asymmetric_signature}),
 we are done by Corollary~\ref{cor:prelim:nonsingular_compressed_hard_trans}.
\end{proof}

With Lemma~\ref{lem:PM-InvPM-g} at hand,
we continue to prove Lemma~\ref{lem:01000b}.

\begin{lemma} \label{lem:01000b}
 Let $b \in \mathbb{C}$.
 Suppose $f$ is a signature of the form $\trans{1}{1}{i}{-i}^{\otimes n}[0,1,0,\dotsc,0,b]$ with arity $n \ge 4$.
 If $b \neq 0$,
 then $\PlHolant(f)$ is \numP-hard.
\end{lemma}

\begin{remark}
For $n=3$, $Z^{\otimes 3} [0,1,0,b]$ is tractable, as
it is $\mathscr{M}$-transformable.
\end{remark}

\begin{proof}
 If $n = 4$,
 then we are done by Corollary~\ref{cor:prelim:nonsingular_compressed_hard_trans}.
 Thus,
 assume that $n \ge 5$.
 
 Under a holographic transformation by $Z = \trans{1}{1}{i}{-i}$,
 we have
 \begin{align*}
   \plholant{{=}_2}{f}
   &\equiv_T \plholant{[1,0,1] Z^{\otimes 2}}{(Z^{-1})^{\otimes n} f}\\
   &\equiv_T \plholant{[0,1,0]}{\hat{f}},
 \end{align*}
 where $\hat{f} = [0, 1, 0, \dotsc, 0, b]$.
 We show how to construct the following three signatures:
 $[0,0,0,1,0]$, $[0,1,0,0,0]$, and $\hat{g}$,
 where $\hat{g}$ is defined by \eqref{eqn:signature-matrix:partial-crossover}.
 Then we are done by Lemma~\ref{lem:PM-InvPM-g}.
 
 Consider the gadget in Figure~\ref{subfig:gadget:van-arity4}.
 We assign $\hat{f}$ to the circle vertices and $[0,1,0]$ to the square vertices.
 The signature of the resulting gadget is $[0,0,0,1,0]$ up to a nonzero factor of $b$.
 
 Taking a $[0,1,0]$ self loop on $[0,0,0,1,0]$ gives $[0,0,1] = [0,1]^{\otimes 2}$.
 We connect this back to $\hat{f}$ through $[0,1,0]$ until the arity of the resulting signature is either~$4$ or~$5$,
 depending on the parity of $n$.
 If $n$ is even,
 then we have $[0,1,0,0,0]$ as desired.
 Otherwise,
 $n$ is odd and we have $[0,1,0,0,0,b/0]$,
 where the last entry is $b$ if $n = 5$ and~$0$ if $n > 5$.
 Connection $[0,1]^{\otimes 2}$ through $[0,1,0]$ to $\hat{f}$ twice more gives $[0,1]$.
 We connect this through $[0,1,0]$ to $[0,1,0,0,0,b/0]$ to get $[0,1,0,0,0]$ as desired.
 
 Taking a $[0,1,0]$ self loop on $[0,1,0,0,0]$ gives $[1,0,0] = [1,0]^{\otimes 2}$.
 Now consider the gadget in Figure~\ref{subfig:gadget:partial-crossover}.
 We assign $\hat{f}$ to the circle vertices,
 $[1,0]^{\otimes 2}$ to the triangle vertices,
 and $[0,1,0]$ to the square vertices.
 Up to a factor of $b^2$,
 the signature of the resulting gadget is $\hat{g}$ with signature matrix 
%$M_g$ given in~(\ref{eqn:signature-matrix:partial-crossover}).
%%% JYC corrected: g --> g-hat
$M_{\hat{g}}$ given in~(\ref{eqn:signature-matrix:partial-crossover}).
 To see this,
 first replace the two copies of the signatures $[1,0]^{\otimes 2}$ assigned to the triangle vertices with two copies of $[1,0]$ each.
 Then notice that $\hat{f}$ simplifies to a weighted equality signature when connected to $[1,0]$ through $[0,1,0]$.
\end{proof}

\subsection{Proof of the Single Signature Dichotomy}

Now we are ready to prove the dichotomy for a single signature.
Recall that $\mathscr{M}_1 \subset \mathscr{A}_1 \subset \mathscr{P}_1$ and $\mathscr{A}_2 = \mathscr{P}_2 \subset \mathscr{M}_2$.
Thus $f \in \mathscr{P}_1 \cup \mathscr{M}_2 \cup \mathscr{A}_3 \cup \mathscr{M}_3 \cup \mathscr{M}_4$ 
if and only if $f$ is $\mathscr{A}$-, $\mathscr{P}$-, or $\mathscr{M}$-transformable by
Lemma~\ref{lem:cha:affine},
Lemma~\ref{lem:cha:product}, or
Lemma~\ref{lem:cha:Mtrans}.

\begin{proof}[Proof of Theorem~\ref{thm:dic:single}]
 The proof is by induction on $n$.
 The base cases of $n=3$ and $n=4$ are proved in Theorem~\ref{thm:PlHolant:arity34}.
 Now assume $n \ge 5$.

 With the signature $f$,
 we form a self loop to get a signature $f'$ of arity at least~$3$.
 In general we use prime to denote the signature with a self loop.
 We consider separately whether or not $f'$ is degenerate.
 \begin{itemize}
  \item Suppose $f' = [a,b]^{\otimes(n-2)}$ is degenerate.
   Then there are three cases to consider.
   \begin{enumerate}
    \item If $a = b = 0,$ then $f'$ is the all zero signature.
     For $f$, this means $f_{k+2} = -f_k$ for $0 \le k \le n-2$,
     so $f \in \mathscr{P}_2$ by Lemma~\ref{lem:single:A2}, and therefore $\PlHolant(f)$ is tractable.
     
    \item If $a^2 + b^2 \ne 0$, then $f'$ is nonzero and $[a,b]$ is not a constant multiple of either $[1,i]$ or $[1,-i]$.
     We may normalize so that $a^2 + b^2 = 1$.
     Then the orthogonal transformation $\left[\begin{smallmatrix} a & b \\ -b & a \end{smallmatrix}\right]$ transforms the column vector $[a,b]$ to $[1,0]$.
     Let $\hat{f}$ be the transformed signature from $f$, 
     and $\widehat{f'} = [1,0]^{\otimes(n-2)}$ the transformed signature from $f'$.

     Since an orthogonal transformation keeps $=_2$ invariant, 
     this transformation commutes with the operation of taking a self loop, i.e., $\widehat{f'} = (\hat{f})'$.
     Here $(\hat{f})'$ is the function obtained from $\hat{f}$ by taking a self loop.
     As $(\hat{f})' = [1,0]^{\otimes(n-2)}$, 
     we have $\hat{f}_0 + \hat{f}_2 = 1$ and 
     for every integer $1 \le k \le n-2$, we have $\hat{f}_k = -\hat{f}_{k+2}$.
     With one or more self loops on $(\hat{f})'$, we eventually obtain either $[1,0]$ when $n$ is odd or $[1,0,0]$ when $n$ is even.
     In either case, we connect $[1,0]$ or $[1,0,0]$ to $\hat{f}$ 
     until we get an arity $4$ signature, which is $\hat{g} = [\hat{f}_0, \hat{f}_1, \hat{f}_2, -\hat{f}_1, -\hat{f}_2]$.
     This is possible because that the parity matches and the arity of $\hat{f}$ is at least $5$.
     We show that $\PlHolant(\hat{g})$ is $\SHARPP$-hard.
     To see this, 
     we first compute $\det(\widetilde{M_g}) = -2 (\hat{f}_0 + \hat{f}_2) (\hat{f}_1^2 + \hat{f}_2^2)=-2(\hat{f}_1^2 + \hat{f}_2^2)$, 
     since $\hat{f}_0 + \hat{f}_2 = 1$.
     Therefore if $\hat{f}_1^2 + \hat{f}_2^2 \ne 0$,
     $\PlHolant(\hat{g})$ is $\numP$-hard by Lemma~\ref{lem:arity4:nonsingular_compressed_hard}.
     Otherwise $\hat{f}_1^2 + \hat{f}_2^2 = 0$, and we assume $\hat{f}_2 = i \hat{f}_1$ since the other case is similar.
     Since $f$ is non-degenerate, $\hat{f}$ is non-degenerate, which implies $\hat{f}_2 \neq 0$.
     We can rewrite $\hat{g}$ as $[1,0]^{\otimes 4} -\hat{f}_2 [1,i]^{\otimes 4}$.
     Under the holographic transformation by $T = \left[\begin{smallmatrix} 1 & (-\hat{f}_2)^{1/4} \\ 0 & i (-\hat{f}_2)^{1/4} \end{smallmatrix}\right]$,
     we have
     \begin{align*}
      \plholant{{=}_2}{\hat{g}}
      &\equiv_T \plholant{[1,0,1] T^{\otimes 2}}{(T^{-1})^{\otimes 4} \hat{g}}\\
      &\equiv_T \plholant{\hat{h}}{{=}_4},
     \end{align*}
     where
     \[\hat{h} = [1,0,1] T^{\otimes 2} = [1,(-\hat{f}_2)^{1/4},0]\]
     and $\hat{g}$ is transformed by $T^{-1}$ into the arity 4 equality $=_4$, since
     \[
      T^{\otimes 4} \left( \begin{bmatrix} 1 \\ 0 \end{bmatrix}^{\otimes 4} + \begin{bmatrix} 0 \\ 1 \end{bmatrix}^{\otimes 4} \right)
      = \begin{bmatrix} 1 \\ 0 \end{bmatrix}^{\otimes 4} - \hat{f}_2 \begin{bmatrix} 1 \\ i \end{bmatrix}^{\otimes 4}
      = \hat{g}.
     \]
     By Theorem~\ref{thm:k-reg_homomorphism}, $\plholant{\hat{h}}{{=}_4}$ is $\SHARPP$-hard as $\hat{f}_2 \ne 0$.

    \item If $a^2 + b^2 = 0$ but $(a,b) \ne (0,0)$, then $[a,b]$ is a nonzero multiple of $[1, \pm i]$.
     Ignoring the constant multiple, we have $f' = [1,i]^{\otimes (n-2)}$ or $[1,-i]^{\otimes (n-2)}$.
     We consider the first case since the other case is similar.

     In the first case,
     the characteristic polynomial of the recurrence relation of $f'$ is $x-i$,
     so that of $f$ is $(x-i) (x^2 + 1) = (x-i)^2 (x+i)$.
     Hence there exist $a_0, a_1$, and $c$ such that
     \[
      f_k = (a_0 + a_1 k) i^k + c (-i)^k
     \]
     for every integer $0 \le k \le n$.
     Let $f^+$ and $f^-$ be two signatures of arity $n$
     such that $f^+_k=(a_0 + a_1 k) i^k$ and $f^-_k=c (-i)^k$ for every $0\le k\le n$.
     Hence $f_k=f^+_k+f^-_k$ and we write $f=f^++f^-$.
     If $a_1 = 0$, then $f'$ is the all zero signature, a contradiction.
     If $c=0$, then $f$ is vanishing, one of the tractable cases.
     Now we assume $a_1 c \neq 0$ and show that $\PlHolant(f)$ is $\SHARPP$-hard.
     Hence $\rd^+ (f^+) = 1$ and $\rd^- (f^-) = 0$.
     Under the holographic transformation $Z = \frac{1}{\sqrt{2}} \left[\begin{smallmatrix} 1 & 1 \\ i & -i \end{smallmatrix}\right]$,
     we have
     \begin{align*}
      \plholant{{=}_2}{f}
      &\equiv_T \plholant{[1,0,1] Z^{\otimes 2}}{(Z^{-1})^{\otimes n} f}\\
      &\equiv_T \plholant{[0,1,0]}{\hat{f}},
     \end{align*}
     where $\hat{f}$ takes the form $[\hat{f_0}, \hat{f_1}, 0, \dotsc, 0, c']$ with $c' = 2^{n/2} c \ne 0$ and $\hat{f_1} \neq 0$,
     since $\hat{f}$ is the $Z^{-1}$-transformation of the sum of $f^+$ and $f^-$,
     with $\rd^+ (f^+) = 1$ and $\rd^- (f^-) = 0$ respectively.
     On the other side, $(=_2) = [1,0,1]$ is transformed into $(\neq_2) = [0,1,0]$.
     Depending on whether $\hat{f_0}=0$ or not, 
     we apply Lemma~\ref{lem:01000b} or
     Lemma~\ref{lem:a1000b} 
     and $\PlHolant(f)$ is $\numP$-hard.
   \end{enumerate}
  \item Suppose $f'$ is non-degenerate.
    By inductive hypothesis, $\PlHolant(f)$ is \numP-hard,
    unless $f' \in \mathscr{P}_1 \cup \mathscr{M}_2 \cup \mathscr{A}_3 \cup \mathscr{M}_3 \cup \mathscr{M}_4\cup\mathscr{V}$.
    Note that $f'$ has arity $n-2 \ge 3$,
    and every signature in $\mathscr{M}_4$ of arity at least $3$ is also in $\mathscr{V}$.
   Hence the exceptional case is equivalent to $f' \in \mathscr{P}_1 \cup \mathscr{M}_2 \cup \mathscr{A}_3 \cup \mathscr{M}_3\cup\mathscr{V}$.
   In this case, we apply Lemma \ref{lem:selfloop} to $f'$ and $f$.
   Hence $\PlHolant(f)$ is \numP-hard, 
   unless $f\in \mathscr{P}_1 \cup \mathscr{M}_2 \cup \mathscr{A}_3 \cup \mathscr{M}_3\cup\mathscr{V}$.
   The exceptional cases imply that
   $f$ is $\mathscr{A}$- or $\mathscr{P}$- or $\mathscr{M}$-transformable or vanishing,
   and $\PlHolant(f)$ is tractable.
   \qedhere
 \end{itemize}
\end{proof}

\section{Mixing \texorpdfstring{$\mathscr{P}_2$}{P2} 
and \texorpdfstring{$\mathscr{M}_4$}{M4}---Equalities and 
Matchgates in the \texorpdfstring{$Z$}{Z} Basis} \label{sec:mixing}

Given a set $\mathcal{F}$ of symmetric signatures,
by Theorem~\ref{thm:dic:single},
$\PlHolant(\mathcal{F})$ is \numP-hard unless
every single non-degenerate signature $f$ of arity at least~$3$ in $\mathcal{F}$ is
in $\mathscr{P}_1 \cup \mathscr{M}_2 \cup \mathscr{A}_3 \cup \mathscr{M}_3 \cup \mathscr{M}_4
\cup \mathscr{V}$.
We have already proved that
the desired full dichotomy holds if $\mathcal{F}$ contains such an $f$ 
in $\mathscr{P}_1$, $\mathscr{A}_3$, $\mathscr{M}_2 \setminus \mathscr{P}_2$,
or $\mathscr{M}_3$ due to Corollary~\ref{cor:dichotomy:P1},
Corollary~\ref{cor:dichotomy:A3},
Lemma~\ref{lem:dichotomy:M2},
or Lemma~\ref{lem:dichotomy:M3},
respectively.

The remaining cases are when  all
non-degenerate signatures of arity at least~$3$ in $\mathcal{F}$ are
contained in $\mathscr{P}_2 \cup
\mathscr{M}_4 \cup \mathscr{V}$.
In this section,
we consider the mixing of $\mathscr{P}_2$ and $\mathscr{M}_4$.
For this,
we do a holographic transformation by $Z$.
Then the problem becomes $\plholant{{\neq}_2}{{=}_k,\exactone{d}}$ with various arities $k$ and $d$.
Recall that $\exactone{d}$ denotes the exact one function $[0,1,0,\dots,0]$ of arity $d$.
These are the signatures  for {\sc Perfect Matching} and they are the
basic components of \emph{Matchgates}.

\emph{A big surprise},
 against the putative form of a complexity classification
for planar counting problems, is that we found 
the complexity of 
 $\plholant{{\neq}_2}{{=}_k,\exactone{d}}$ 
depends on the values of $d$ and $k$, and the problem is
tractable for all large $k$.  This result has the consequence that,
for the first time since Kasteleyn's algorithm, we have discovered
some new \emph{primitive} tractable family  of  counting problems
on planar graphs.
These problems \emph{cannot} be captured by a holographic reduction
to  Kasteleyn's algorithm, or any other known algorithm.
Thus for planar problems the paradigm of holographic algorithms using matchgates
%via a holographic reduction to the FKT 
(i.e., being $\mathscr{M}$-transformable)
\emph{is not universal}.
%: there are other tractable cases.
%Some mixing is possible. 

Let $\mathcal{EO} = \{\exactone{d} \mid d \ge 3\}$.

\subsection{Hardness when \texorpdfstring{$k = 3$ or $4$}{k=3,4}}

We begin with some hardness results.

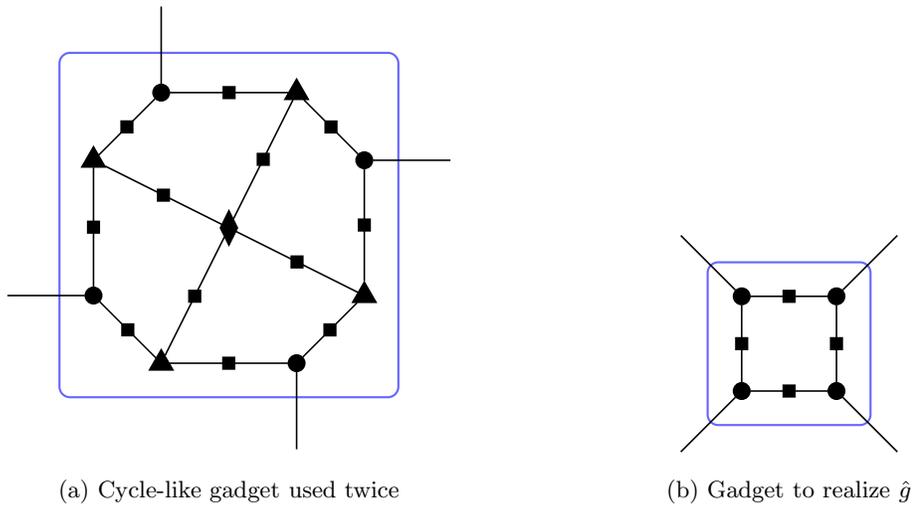
\begin{figure}[htpb]
 \centering
 \def\capWidth{6.5cm}
 \captionsetup[subfigure]{width=\capWidth}
 \subfloat[Cycle-like gadget used twice]{
 \begin{tikzpicture}[scale=\scale,transform shape,node distance=1.5*\nodeDist,semithick]
  \draw    (0,0)    node[diamond,fill,draw,aspect=0.5]   (0) {};
  \draw  (1.5,3)    node[triangle] (p1) {};
  \draw    (3,1.5)  node[internal] (e1) {};
  \draw    (3,-1.5) node[triangle] (p2) {};
  \draw  (1.5,-3)   node[internal] (e2) {};
  \draw (-1.5,-3)   node[triangle] (p3) {};
  \draw   (-3,-1.5) node[internal] (e3) {};
  \draw   (-3,1.5)  node[triangle] (p4) {};
  \draw (-1.5,3)    node[internal] (e4) {};
  
  \node[external, right of=e1] (ex1) {};
  \node[external, below of=e2] (ex2) {};
  \node[external, left  of=e3] (ex3) {};
  \node[external, above of=e4] (ex4) {};    
  
  \path (e1) edge                (ex1);
  \path (e2) edge                (ex2);
  \path (e3) edge                (ex3);
  \path (e4) edge                (ex4);    
  
  \path (p1) edge node[square] {} (e1)
             edge node[square] {} (e4);
  \path (p2) edge node[square] {} (e1)
             edge node[square] {} (e2);
  \path (p3) edge node[square] {} (e2)
             edge node[square] {} (e3);
  \path (p4) edge node[square] {} (e3)
             edge node[square] {} (e4);
  \path (0)  edge node[square] {} (p1)
             edge node[square] {} (p2)
             edge node[square] {} (p3)
             edge node[square] {} (p4);
  \begin{pgfonlayer}{background}
   \node[draw=\borderColor,thick,rounded corners,fit = (e1) (e2) (e3) (e4) (p1) (p2) (p3) (p4),inner sep=16pt] {};
%    \node[draw=\borderColor,thick,rounded corners,fit = (e1) (e2) (e3) (e4) (p1) (p2) (p3) (p4),inner sep=8pt,transform shape=false] {};
  \end{pgfonlayer}               
 \end{tikzpicture}
 \label{subfig:Eq3-PM3:cycle}}
 \qquad
 \subfloat[Gadget to realize $\hat{g}$]{
  \makebox[\capWidth][c]{
   \begin{tikzpicture}[scale=\scale,transform shape,node distance=1.5*\nodeDist,semithick]
    \node[internal]               (e1) {};
    \node[internal, below of=e1]  (e2) {};
    \node[internal, left  of=e2]  (e3) {};
    \node[internal, above of=e3]  (e4) {};
    
    \node[external, above right of=e1] (ex1) {};
    \node[external, below right of=e2] (ex2) {};
    \node[external, below left  of=e3] (ex3) {};
    \node[external, above left  of=e4] (ex4) {};
    
    \path (e1) edge node[square] {} (e2)
               edge node[square] {} (e4);
    \path (e3) edge node[square] {} (e2)
               edge node[square] {} (e4);
    
    \path (e1) edge                (ex1);
    \path (e2) edge                (ex2);
    \path (e3) edge                (ex3);
    \path (e4) edge                (ex4);
    \begin{pgfonlayer}{background}
     \node[draw=\borderColor,thick,rounded corners,fit = (e1) (e2) (e3) (e4),inner sep=16pt] {};
%      \node[draw=\borderColor,thick,rounded corners,fit = (e1) (e2) (e3) (e4),inner sep=8pt,transform shape=false] {};
    \end{pgfonlayer}
   \end{tikzpicture}
   \label{subfig:Eq3-PM3:g}
  }
 }
 \caption{Two gadgets used in the proof of Lemma~\ref{lem:Z:Eq3PM3}.}
 \label{fig:Eq3-PM3}
\end{figure}

\begin{figure}[htpb]
 \centering
 \begin{tikzpicture}[scale=\scale,transform shape,node distance=1.5*\nodeDist,semithick]
  \draw    (1,0)    node[internal] (o1) {};
  \draw   (-1,0)    node[internal] (o2) {};    
  \draw  (1.5,3)    node[triangle] (p1) {};
  \draw    (3,1.5)  node[internal] (e1) {};
  \draw    (3,-1.5) node[triangle] (p2) {};
  \draw  (1.5,-3)   node[internal] (e2) {};
  \draw (-1.5,-3)   node[triangle] (p3) {};
  \draw   (-3,-1.5) node[internal] (e3) {};
  \draw   (-3,1.5)  node[triangle] (p4) {};
  \draw (-1.5,3)    node[internal] (e4) {};
  
  \path (p1) edge node[square] {} (e1)
             edge node[square] {} (e4);
  \path (p2) edge node[square] {} (e1)
             edge node[square] {} (e2);
  \path (p3) edge node[square] {} (e2)
             edge node[square] {} (e3);
  \path (p4) edge node[square] {} (e3)
             edge node[square] {} (e4);
  \path (o1) edge node[square] {} (p1)
             edge node[square] {} (p2)
             edge node[square] {} (o2)
        (o2) edge node[square] {} (p3)
             edge node[square] {} (p4);
  
  \draw  (4.5,4.5)  node[internal] (ex-p1) {};
  \draw  (4.5,-4.5) node[internal] (ex-p2) {};
  \draw (-4.5,-4.5) node[internal] (ex-p3) {};
  \draw (-4.5,4.5)  node[internal] (ex-p4) {};
  
  \path (e1) edge node[square] {} (ex-p1);
  \path (e2) edge node[square] {} (ex-p2);
  \path (e3) edge node[square] {} (ex-p3);
  \path (e4) edge node[square] {} (ex-p4);    
  
  \draw  (4.5,0)    node[triangle] (ex-e1) {};
  \draw    (0,-4.5) node[triangle] (ex-e2) {};
  \draw (-4.5,0)    node[triangle] (ex-e3) {};
  \draw     (0,4.5) node[triangle] (ex-e4) {};
  
  \path (ex-p1) edge node[square] {} (ex-e1)
                edge node[square] {} (ex-e4);
  \path (ex-p2) edge node[square] {} (ex-e2)
                edge node[square] {} (ex-e1);
  \path (ex-p3) edge node[square] {} (ex-e3)
                edge node[square] {} (ex-e2);
  \path (ex-p4) edge node[square] {} (ex-e4)
                edge node[square] {} (ex-e3);    
  
  \draw  (6,0)  node[external] (ex1) {};
  \draw  (0,-6) node[external] (ex2) {};
  \draw (-6,0)  node[external] (ex3) {};
  \draw  (0,6)  node[external] (ex4) {};
  
  \path (ex-e1) edge (ex1);
  \path (ex-e2) edge (ex2);
  \path (ex-e3) edge (ex3);
  \path (ex-e4) edge (ex4);
  \begin{pgfonlayer}{background}
   \node[draw=\borderColor,thick,rounded corners,fit = (ex-p1) (ex-p2) (ex-p3) (ex-p4),inner sep=14pt] {};
%    \node[draw=\borderColor,thick,rounded corners,fit = (ex-p1) (ex-p2) (ex-p3) (ex-p4),inner sep=7pt,transform shape=false] {};
  \end{pgfonlayer}               
 \end{tikzpicture}
 \caption{The whole gadget to realize $[0,0,0,1,0]$.}
 \label{fig:Eq3-PM3:InvPM}
\end{figure}
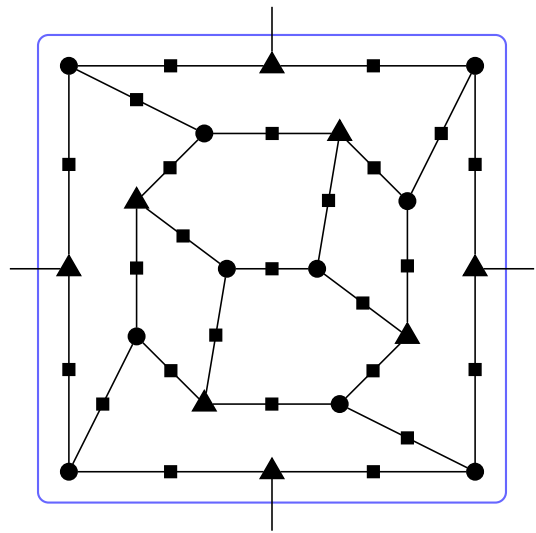

\begin{lemma} \label{lem:Z:Eq3PM3}
 $\plholant{{\neq}_2}{{=}_3,[0,1,0,0]}$ is \numP-hard.
\end{lemma}

\begin{proof}
 By connecting two copies of $[0,1,0,0]$ together via $\neq_2$,
 we have $[0,1,0,0,0]$ on the right.
 Consider the gadget in Figure~\ref{subfig:Eq3-PM3:cycle}.
 We assign $=_3$ to the triangle vertices,
 $[0,1,0,0]$ to the circle vertices,
 $\neq_2$ to the square vertices,
 and $[0,1,0,0,0]$ on the diamond vertex in the middle.
 Let $f$ be the signature of this gadget.

 We claim that the support of $f$ is $\{0011,0110,1100,1001\}$.
 To see this,
 notice that $[0,1,0,0,0]$ in the middle must match exactly one of the half edges,
 which forces the corresponding equality signature to take the value~$0$ and all other equality signatures to take value~$1$.
 The two $[0,1,0,0]$'s adjacent to the equality assigned~$0$ must have~$0$ going out,
 and the other two $[0,1,0,0]$'s have~$1$ going out.

  Now we consider the gadget in Figure~\ref{subfig:Eq3-PM3:cycle} again.
  This time we place $[0,1,0,0]$ on each triangle,
  $=_3$ on each circle,
  $f$ on the middle diamond,
  and again $\neq_2$ on each square.
  Now notice that each support of $f$ makes two $[0,1,0,0]$'s
  that are cyclically adjacent on the outer cycle to become $[0,1,0]$
  and the other two $[1,0,0]$.
  It is easy to see that the support of the resulting signature is $\{0111,1011,1101,1110\}$.
  Therefore it is the reversed \exactone{4} signature $[0,0,0,1,0]$ (namely \allbutone{4}).
  The whole gadget is illustrated in Figure~\ref{fig:Eq3-PM3:InvPM},
  where each circle is assigned $[0,1,0,0]$,
  triangle $=_3$, and square $\neq_2$.
  
  Finally, we build the gadget in Figure~\ref{subfig:Eq3-PM3:g}.
  We place $=_3$ on each circle and $\neq_2$ on each square.
  It is easy to see that there are only two support vectors of the resulting signature,
  which are $0101$ and $1010$.
  Recall the definition~\eqref{eqn:signature-matrix:partial-crossover} of the partial crossover $\hat{g}$.
  This gadget realizes exactly $\hat{g}$.

  By Lemma~\ref{lem:PM-InvPM-g}, 
  $\plholant{{\neq}_2}{[0,1,0,0,0],[0,0,0,1,0],\hat{g}}$ is \numP-hard.
  We have constructed $[0,1,0,0,0]$, $[0,0,0,1,0]$, and $\hat{g}$ on the right side.
  Therefore $\plholant{{\neq}_2}{{=}_3,[0,1,0,0]}$ is \numP-hard.
\end{proof}

\begin{figure}[t]
 \centering
 \subfloat[Step one: Degree~$4$ vertex example]{
  \begin{tikzpicture}[scale=\scale,transform shape,node distance=\nodeDist,semithick]
   \node[internal]  (0)             {};
   \node[external]  (1) [left of=0] {};
   \node[external]  (2) [above of=0] {};
   \node[external]  (3) [below of=0] {};
   \node[external]  (4) [right of=0] {};
   \node[external]  (5) [right of=4] {};
   \node[external]  (6) [right of=5] {};
   \node[external]  (7) [right of=6] {};
   \node[internal]  (8) [right of=7] {};
   \node[internal]  (9) [above right of=8] {};
   \node[internal] (10) [below right of=8] {};
   \node[external] (11) [above of=9] {};
   \node[external] (12) [below of=10] {};
   \node[internal] (13) [below right of=9] {};
   \node[external] (14) [right of=13] {};
   \path (0) edge (1)
             edge (2)
             edge (3)
             edge (4)
         (5.west) edge[->,very thick] (6.east)
         (7) edge (8)
         (8) edge (9)
             edge (10)
         (9) edge (11)
             edge (13)
        (10) edge (12)
             edge (13)
        (13) edge (14);
   \begin{pgfonlayer}{background}
    \node[draw=\borderColor,thick,rounded corners,fit = (0),inner sep=8pt] {};
%     \node[draw=\borderColor,thick,rounded corners,fit = (0),inner sep=4pt,transform shape=false] {};
    \node[draw=\borderColor,thick,rounded corners,fit = (8) (9) (10) (13),inner sep=8pt] {};
%     \node[draw=\borderColor,thick,rounded corners,fit = (8) (9) (10) (13),inner sep=4pt,transform shape=false] {};
   \end{pgfonlayer}
  \end{tikzpicture}
  \label{subfig:CSP_to_4ary_Hol:step1}
 }
 \qquad
 \subfloat[Step two: Contract edges]{
  \begin{tikzpicture}[scale=\scale,transform shape,node distance=\nodeDist,semithick]
   \node[internal]  (0)             {};
   \node[external]  (1) [left of=0] {};
   \node[external]  (2) [right of=0] {};
   \node[triangle]  (3) [below of=0] {};
   \node[internal]  (4) [below of=3] {};
   \node[external]  (5) [left of=4] {};
   \node[external]  (6) [right of=4] {};
   \node[external]  (7) [right of=2] {};
   \node[external]  (8) [below of=7] {};
   \node[external]  (9) [right of=8] {};
   \node[external] (10) [above of=9] {};
   \node[external] (11) [right of=10] {};
   \node[external] (12) [right of=11] {};
   \node[external] (13) [right of=12] {};
   \node[square]   (14) [below of=12] {};
   \node[external] (15) [below of=14] {};
   \node[external] (16) [left of=15] {};
   \node[external] (17) [right of=15] {};
   \path (0) edge[postaction={decorate, decoration={
                                        markings,
                                        mark=at position 0.78 with {\arrow[>=diamond,white] {>}; },
                                        mark=at position 0.78 with {\arrow[>=open diamond]  {>}; } } }] (1)
             edge (2)
             edge (3)
         (3) edge (4)
         (4) edge (5)
             edge (6)
         (8.west) edge[->,very thick] (9.east)
        (14) edge[postaction={decorate, decoration={
                                        markings,
                                        mark=at position 0.6 with {\arrow[>=diamond,white] {>}; },
                                        mark=at position 0.6 with {\arrow[>=open diamond]  {>}; } } }] (11)
             edge (13)
             edge (16)
             edge (17);
   \begin{pgfonlayer}{background}
    \node[draw=\borderColor,thick,rounded corners,fit = (0) (4),inner sep=8pt] {};
%     \node[draw=\borderColor,thick,rounded corners,fit = (0) (4),inner sep=5pt,transform shape=false] {};
    \node[draw=\borderColor,thick,rounded corners,fit = (14),inner sep=12pt] {};
%     \node[draw=\borderColor,thick,rounded corners,fit = (14),inner sep=7pt,transform shape=false] {};
   \end{pgfonlayer}
  \end{tikzpicture}
  \label{subfig:CSP_to_4ary_Hol:step2}
 }
 \caption{A reduction from $\plholant{\EQ}{h}$ to $\PlHolant(g)$ for any binary signature $h$ and a quaternary signature $g$ that depends on $h$.
 The circle vertices are assigned $=_4$ or $=_3$ respectively,
 the triangle vertex is assigned $h$,
 and the square vertex is assigned the signature of the gadget to its left.}
 \label{subfig:CSP_to_4ary_Hol}
\end{figure}
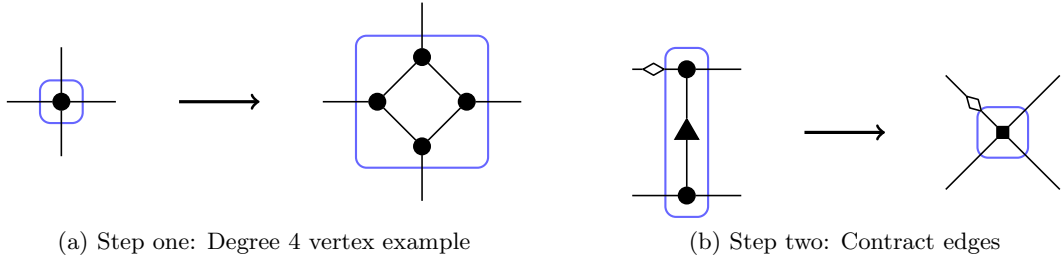

For $k = 4$,
we need the following lemma.

\begin{lemma} \label{lem:arity4:2-spin}
 Let $g$ be the arity~$4$ signature whose matrix is
 \begin{align*}
  M_{g} =
  \begin{bmatrix}
   2 & 0 & 0 & 0 \\
   0 & 1 & 0 & 0 \\
   0 & 0 & 1 & 0 \\
   0 & 0 & 0 & 1
  \end{bmatrix}.
 \end{align*}
 Then $\PlHolant(g)$ is \numP-hard.
\end{lemma}

\begin{proof}
 Let $h = [2,1,1]$.
 We show that $\PlCSP(h) \le_T \PlHolant(g)$ in two steps.
 In each step,
 we begin with a signature grid and end with a new signature grid such that the Holants of both signature grids are the same.
 Then we are done by Theorem~\ref{thm:PlCSP}.
 Or more explicitly,
 since $\PlCSP(h) \equiv \plholant{\EQ}{h}$ by~\eqref{eqn:prelim:PlCSPd_equiv_Holant},
 we are done by Theorem~\ref{thm:k-reg_homomorphism}.
 
 For step one, let $G = (U,V,E)$ be an instance of $\plholant{\EQ}{h}$.
 Fix an embedding of $G$ in the plane.
 This defines a cyclic ordering of the edges incident to each vertex.
 Consider a vertex $u \in U$ of degree $k$.
 It is assigned the signature $=_k$.
 We decompose $u$ into $k$ vertices.
 Then we connect the $k$ edges originally incident to $u$ to these $k$ new vertices so that each vertex is incident to exactly one edge.
 We also connect these $k$ new vertices in a cycle according to
the  cyclic ordering induced on them by their incident edges.
 Each of these vertices has degree~$3$,
 and we assign them $=_3$.
 Clearly the Holant value is unchanged.
 This completes step one.
 An example of this step applied to a vertex of degree~$4$ is given in Figure~\ref{subfig:CSP_to_4ary_Hol:step1}.
 The resulting graph has the following properties:
 (1) it is planar;
 (2) every vertex is either degree~$2$ (in $V$ and assigned $h$) or degree~$3$ (newly created and assigned $=_3$);
 (3) each degree $2$ vertex is connected to two degree~$3$ vertices; and
 (4) each degree $3$ vertex is connected to one degree~$2$ vertex and two other degree~$3$ vertices.
 
 Now step two.
 For every $v \in V$, $v$ has degree~$2$.
 We contract the two edges incident to $v$,
 or equivalently, we replace the two circle vertices and one triangle vertex boxed in Figure~\ref{subfig:CSP_to_4ary_Hol:step2}
 with a single (square) vertex of degree~$4$.
 The resulting graph $G' = (V',E')$ is planar and $4$-regular.
 
 Next we determine what is the signature on $v'\in V'$ after this contraction.
 Clearly the two inputs to each original circle have to be the same.
 Therefore its support is $0000,0110,1001,1111$,
 listed starting from the diamond and going counterclockwise.
 Moreover, due to the triangle assigned $h$ in the middle, the weight on $0000$ is $2$,
 and every other weight is $1$.
 Hence it is exactly the signature $g$,
 with the diamond in Figure~\ref{subfig:CSP_to_4ary_Hol:step2} marking the first input bit.
 This finishes the proof.
\end{proof}

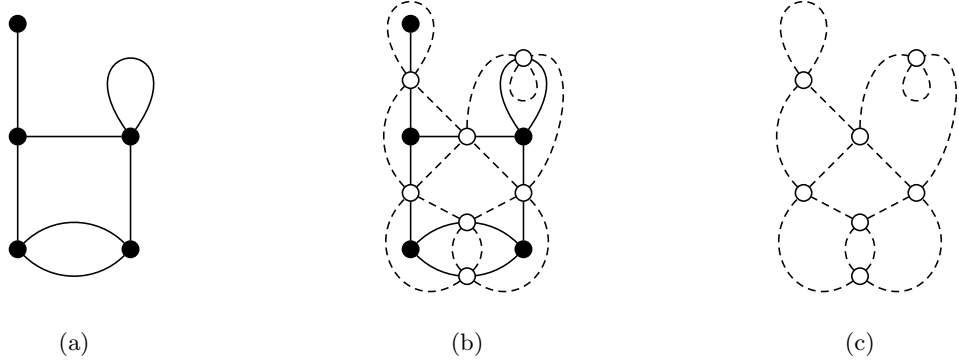
\begin{figure}[t]
 \centering
 \def\medialNodeDist{2.5cm}
 \tikzstyle{open}   = [draw, black, fill=white, shape=circle]
 \tikzstyle{closed} = [draw,        fill,       shape=circle]
 \subfloat[]{
  \begin{tikzpicture}[scale=\scale,transform shape,node distance=\medialNodeDist,semithick]
   \node[closed] (0)              {};
   \node[closed] (1) [right of=0] {};
   \node[closed] (2) [above of=0] {};
   \node[closed] (3) [above of=1] {};
   \node[closed] (4) [above of=2] {};
   \path (0) edge[out=-45, in=-135]               node[external] (m0) {} (1)
             edge[out= 45, in= 135]               node[external] (m1) {} (1)
             edge                                 node[external] (m2) {} (2)
         (1) edge                                 node[external] (m3) {} (3)
         (2) edge                                 node[external] (m4) {} (3)
             edge                                 node[external] (m5) {} (4)
         (3) edge[out=125, in=  55, looseness=30] node[external] (m6) {} (3);
   \path (m0) edge[white, densely dashed, out= 135, in=-135]                (m1)
              edge[white, densely dashed, out=  45, in= -45]                (m1)
              edge[white, densely dashed, out=-145, in=-135, looseness=1.7] (m2)
              edge[white, densely dashed, out= -35, in= -45, looseness=1.7] (m3)
         (m1) edge[white, densely dashed]                                   (m2)
              edge[white, densely dashed]                                   (m3)
         (m2) edge[white, densely dashed]                                   (m4)
              edge[white, densely dashed, out= 135, in=-135]                (m5)
         (m3) edge[white, densely dashed]                                   (m4)
              edge[white, densely dashed, out=  45, in=  15]                (m6)
         (m4) edge[white, densely dashed]                                   (m5)
              edge[white, densely dashed, out=  90, in= 165]                (m6)
         (m5) edge[white, densely dashed, out= 125, in=  55, looseness=30]  (m5)
         (m6) edge[white, densely dashed, out=-125, in= -55, looseness=15]  (m6);
  \end{tikzpicture}
  \label{subfig:planar_graph}
 }
 \qquad
 \qquad
 \subfloat[]{
  \begin{tikzpicture}[scale=\scale,transform shape,node distance=\medialNodeDist,semithick]
   \node[closed] (0)              {};
   \node[closed] (1) [right of=0] {};
   \node[closed] (2) [above of=0] {};
   \node[closed] (3) [above of=1] {};
   \node[closed] (4) [above of=2] {};
   \path (0) edge[out=-45, in=-135]               node[open] (m0) {} (1)
             edge[out= 45, in= 135]               node[open] (m1) {} (1)
             edge                                 node[open] (m2) {} (2)
         (1) edge                                 node[open] (m3) {} (3)
         (2) edge                                 node[open] (m4) {} (3)
             edge                                 node[open] (m5) {} (4)
         (3) edge[out=125, in=  55, looseness=30] node[open] (m6) {} (3);
   \path (m0) edge[densely dashed, out= 135, in=-135]                (m1)
              edge[densely dashed, out=  45, in= -45]                (m1)
              edge[densely dashed, out=-145, in=-135, looseness=1.7] (m2)
              edge[densely dashed, out= -35, in= -45, looseness=1.7] (m3)
         (m1) edge[densely dashed]                                   (m2)
              edge[densely dashed]                                   (m3)
         (m2) edge[densely dashed]                                   (m4)
              edge[densely dashed, out= 135, in=-135]                (m5)
         (m3) edge[densely dashed]                                   (m4)
              edge[densely dashed, out=  45, in=  15]                (m6)
         (m4) edge[densely dashed]                                   (m5)
              edge[densely dashed, out=  90, in= 165]                (m6)
         (m5) edge[densely dashed, out= 125, in=  55, looseness=30]  (m5)
         (m6) edge[densely dashed, out=-125, in= -55, looseness=15]  (m6);
  \end{tikzpicture}
  \label{subfig:superimposed}
 }
 \qquad
 \qquad
 \subfloat[]{
  \begin{tikzpicture}[scale=\scale,transform shape,node distance=\medialNodeDist,semithick]
   \node[external] (0)              {};
   \node[external] (1) [right of=0] {};
   \node[external] (2) [above of=0] {};
   \node[external] (3) [above of=1] {};
   \node[external] (4) [above of=2] {};
   \path (0) edge[white, out=-45, in=-135]               node[open] (m0) {} (1)
             edge[white, out= 45, in= 135]               node[open] (m1) {} (1)
             edge[white]                                 node[open] (m2) {} (2)
         (1) edge[white]                                 node[open] (m3) {} (3)
         (2) edge[white]                                 node[open] (m4) {} (3)
             edge[white]                                 node[open] (m5) {} (4)
         (3) edge[white, out=125, in=  55, looseness=30] node[open] (m6) {} (3);
   \path (m0) edge[densely dashed, out= 135, in=-135]                (m1)
              edge[densely dashed, out=  45, in= -45]                (m1)
              edge[densely dashed, out=-145, in=-135, looseness=1.7] (m2)
              edge[densely dashed, out= -35, in= -45, looseness=1.7] (m3)
         (m1) edge[densely dashed]                                   (m2)
              edge[densely dashed]                                   (m3)
         (m2) edge[densely dashed]                                   (m4)
              edge[densely dashed, out= 135, in=-135]                (m5)
         (m3) edge[densely dashed]                                   (m4)
              edge[densely dashed, out=  45, in=  15]                (m6)
         (m4) edge[densely dashed]                                   (m5)
              edge[densely dashed, out=  90, in= 165]                (m6)
         (m5) edge[densely dashed, out= 125, in=  55, looseness=30]  (m5)
         (m6) edge[densely dashed, out=-125, in= -55, looseness=15]  (m6);
  \end{tikzpicture}
  \label{subfig:medial_graph}
 }
 \caption{A plane graph~\protect\subref{subfig:planar_graph}, its medial graph~\protect\subref{subfig:medial_graph},
 and both graphs superimposed~\protect\subref{subfig:superimposed}.}
 %the \protect command changes the fragile command \subref into a robust one.  See http://www.tex.ac.uk/cgi-bin/texfaq2html?label=extrabrace
 \label{fig:medial_graph_example}
\end{figure}

\begin{remark}
 From the planar embedding of the graph $G$,
 treating $h$ vertices as edges,
 the resulting graph $G'$ is known as the medial graph of $G$.
 The (constructive) definition is usually phrased in the following way.
 The medial graph $G_m$ of plane graph $G$ has a vertex on each edge of $G$
 and two vertices in $G_m$ are joined by an edge for each face of $G$ in which their corresponding edges occur consecutively.
 See Figure~\ref{fig:medial_graph_example} for an example.
 However,
 our construction described in the proof clearly extends to nonplanar graphs as well.
% Furthermore,
% the faces of $G_m$ can be two colored.
% Suppose we use white and black to color the faces and that we color the outer face white.
% Then black faces correspond to vertex of $G$ and white faces correspond to faces of $G$.
% With this coloring of the faces,
% the edge that corresponds to the first input of $g$ must have a white face on its left and a black face on its right.
\end{remark}

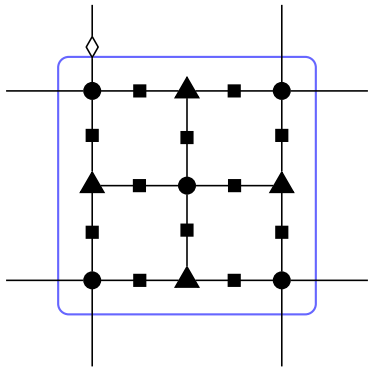
\begin{figure}[htpb]
 \centering
 \begin{tikzpicture}[scale=\scale,transform shape,node distance=1.5 * \nodeDist,semithick]
  \node[internal]                (00) {};
  
  \node[triangle, above of=00]   (01) {};
  \node[triangle, right of=00]   (10) {};
  \node[triangle, left of=00]   (-10) {};
  \node[triangle, below of=00]  (0-1) {};
  
  \node[internal, left of=01]   (-11) {};
  \node[internal, right of=01]   (11) {};
  \node[internal, left of=0-1] (-1-1) {};
  \node[internal, right of=0-1] (1-1) {};
  
  \node[external, above of=-11] (-12) {};
  \node[external, left of=-11]  (-21) {};
  \node[external, above of=11]   (12) {};
  \node[external, right of=11]   (21) {};
  \node[external, below of=-1-1] (-1-2) {};
  \node[external, left of=-1-1]  (-2-1) {};
  \node[external, below of=1-1] (1-2) {};
  \node[external, right of=1-1] (2-1) {};
  
  \path  (00) edge node[square] {} (01)
              edge node[square] {} (10)
              edge node[square] {} (-10)
              edge node[square] {} (0-1)
        (-11) edge node[square] {} (-10)
              edge node[square] {} (01)
              edge [postaction={decorate, decoration={markings,
                    mark=at position 0.6 with {\arrow[>=diamond,white] {>}; }],
                    mark=at position 0.6 with {\arrow[>=open diamond]  {>}; } } }]
                                   (-12)
              edge                 (-21)
         (11) edge node[square] {} (10)
              edge node[square] {} (01)
              edge                 (12)
              edge                 (21)
        (1-1) edge node[square] {} (10)
              edge node[square] {} (0-1)
              edge                 (1-2)
              edge                 (2-1)
       (-1-1) edge node[square] {} (-10)
              edge node[square] {} (0-1)
              edge                 (-1-2)
              edge                 (-2-1);
  \begin{pgfonlayer}{background}
   \node[draw=\borderColor,thick,rounded corners,fit = (11) (1-1) (-1-1) (-11),inner sep=16pt] {};
%    \node[draw=\borderColor,thick,rounded corners,fit = (11) (1-1) (-1-1) (-11),inner sep=9pt,transform shape=false] {};
  \end{pgfonlayer}
 \end{tikzpicture}
 \caption{Grid-like gadget used in the proof of Lemma~\ref{lem:Z:Eq4PM3},
          whose support vectors are $00110011$, $11001100$, and $11111111$.
          Each square is assigned a binary disequality $\neq_2$, 
          circle $=_4$, and triangle $[0,1,0,0]$.}
 \label{fig:Eq4-PM3:grid}
\end{figure}

\begin{figure}[p]
 \centering
 \def\capWidth{6.5cm}
 \captionsetup[subfigure]{width=\capWidth}
 \subfloat[Gadget with signature $g$.
           Each square is assigned a binary disequality $\neq_2$, 
           circle $=_4$, triangle ${[0,1,0,0]}$, and pentagon $f$.]{
  \makebox[\capWidth][c]{
   \begin{tikzpicture}[scale=\scale,transform shape,node distance=1.5 * \nodeDist,semithick]
    \node[internal]                     (e1) {};
    \node[triangle, right of=e1]        (p1) {};
    \node[internal, right of=p1]        (e2) {};
    \node[triangle, right of=e2]        (p2) {};
    \node[internal, below right of=p2]  (e3) {};
    \node[triangle, below of=e3]        (p3) {};
    \node[internal, below of=p3]        (e4) {};
    \node[triangle, below of=e4]        (p4) {};
    \node[internal, below left of=p4]   (e5) {};
    \node[triangle, left of=e5]         (p5) {};
    \node[internal, left of=p5]         (e6) {};
    \node[triangle, left of=e6]         (p6) {};
    \node[internal, above left of=p6]   (e7) {};
    \node[triangle, above of=e7]        (p7) {};
    \node[internal, above of=p7]        (e8) {};
    \node[triangle, above of=e8]        (p8) {};
    
    \node[external, above of=p1]       (ex1) {};
    \node[external, above of=p2]       (ex2) {};
    \node[external, right of=p3]       (ex3) {};
    \node[external, right of=p4]       (ex4) {};
    \node[external, below of=p5]       (ex5) {};
    \node[external, below of=p6]       (ex6) {};
    \node[external, left  of=p7]       (ex7) {};
    \node[external, left  of=p8]       (ex8) {};
    
    \path (p1) edge node[square] {} (e1)
               edge node[square] {} (e2)
               edge                 (ex1);
    \path (p2) edge node[square] {} (e2)
               edge node[square] {} (e3)
               edge                 (ex2);
    \path (p3) edge node[square] {} (e3)
               edge node[square] {} (e4)
               edge                 (ex3);
    \path (p4) edge node[square] {} (e4)
               edge node[square] {} (e5)
               edge                 (ex4);
    \path (p5) edge node[square] {} (e5)
               edge node[square] {} (e6)
               edge                 (ex5);
    \path (p6) edge node[square] {} (e6)
               edge node[square] {} (e7)
               edge                 (ex6);
    \path (p7) edge node[square] {} (e7)
               edge node[square] {} (e8)
               edge                 (ex7);
    \path (p8) edge node[square] {} (e8)
               edge node[square] {} (e1)
               edge [postaction={decorate, decoration={
                     markings,
                     mark=at position 0.58 with {\arrow[>=diamond,white] {>}; }],
                     mark=at position 0.58 with {\arrow[>=open diamond]  {>}; } } }]
                                    (ex8);
    \node[external, below of=e2] (a1) {};
    \node[pentagon, below of=a1, text=white] (f1) {$f_1$};
    \path (f1) edge[bend left] node[square] {} (e1)
               edge            node[square] {} (e1)
               edge[bend left] node[square] {} (e2)
               edge            node[square] {} (e2)
               edge[bend left] node[square] {} (e3)
               edge            node[square] {} (e3)
               edge[bend left] node[square] {} (e4)
               edge[postaction={decorate, decoration={
                    markings,
                    mark=at position 0.14 with {\arrow[>=diamond,white] {>}; }],
                    mark=at position 0.14 with {\arrow[>=open diamond]  {>}; } } }] 
                               node[square] {} (e4);
    \node[external, above of=e6] (a2) {};
    \node[pentagon, above of=a2, text=white] (f2) {$f_2$};
    \path (f2) edge[bend left] node[square] {} (e5)
               edge            node[square] {} (e5)
               edge[bend left] node[square] {} (e6)
               edge            node[square] {} (e6)
               edge[bend left] node[square] {} (e7)
               edge            node[square] {} (e7)
               edge[bend left] node[square] {} (e8)
               edge[postaction={decorate, decoration={
                    markings,
                    mark=at position 0.14 with {\arrow[>=diamond,white] {>}; }],
                    mark=at position 0.14 with {\arrow[>=open diamond]  {>}; } } }]
                               node[square] {} (e8);
    \begin{pgfonlayer}{background}
     \node[draw=\borderColor,thick,rounded corners,fit = (e1) (p2) (e3) (p4) (e5) (p6) (e7) (p8),inner sep=14pt] {};
%      \node[draw=\borderColor,thick,rounded corners,fit = (e1) (p2) (e3) (p4) (e5) (p6) (e7) (p8),inner sep=8pt,transform shape=false] {};
    \end{pgfonlayer}                
   \end{tikzpicture}
  }
  \label{subfig:Eq4-PM3:cycle}
 }
 \qquad
 \subfloat[Support of $g$.
           Each vector is an assignment ordered counterclockwise from the diamond.]{
  \makebox[\capWidth][c]{
   \begin{tabular}[b]{|c|c||c|}
    \hline
    $f_1$    & $f_2$    & $g$\\
    \hline
    \hline
    00000000 & 00000000 & 11111111\\
    \hline
    00110011 & 00000000 & 01111000\\
    \hline
    11001100 & 00000000 & 11110000\\
    \hline
    00000000 & 00110011 & 10000111\\
    \hline
    00110011 & 00110011 & 00000000\\
    \hline
    11001100 & 00110011 & -\\
    \hline
    00000000 & 11001100 & 00001111\\
    \hline
    00110011 & 11001100 & -\\
    \hline
    11001100 & 11001100 & 00000000\\
    \hline
   \end{tabular}
  }
  \label{table:Eq4-PM3}
 }
 \caption{Another gadget used in the proof of Lemma~\ref{lem:Z:Eq4PM3} and a Table listing the support of its signature.}
 \label{fig:Eq4-PM3:cycle}
\end{figure}

\begin{lemma} \label{lem:Z:Eq4PM3}
  $\plholant{{\neq}_2}{{=}_4,[0,1,0,0]}$ is \numP-hard.
\end{lemma}

\begin{proof}
  Consider the gadget in Figure~\ref{fig:Eq4-PM3:grid}.
  We assign binary disequality $\neq_2$ to the square vertices,
  $=_4$ to the circle vertices,
  and $[0,1,0,0]$ to the triangle vertices.
  We show that the support of the resulting signature is the set $\{00110011,11001100,11111111\}$,
  where each vector is the assignment ordered counterclockwise starting from the diamond point.

  We call the equality signature $=_4$ in the middle the origin.
  There are two possible assignments at the origin.
  If it is assigned~$0$, 
  then every adjacent perfect matching signature $[0,1,0,0]$
  is matched to the half edge towards the origin,
  and every equality $=_4$ is forced to be~$1$.
  This gives the support vector $11111111$.

  The other possibility is that the origin is~$1$.
  In this case,
  we can remove the origin leaving the outer cycle,
  with every $[0,1,0,0]$ becoming $[0,1,0]$.
  This is effectively a cycle of four equalities connected by $\neq_2$.
  It is easy to see that there are only two support vectors,
  which are exactly $00110011$ and $11001100$.

  Every pair of half edges at each corner always take the same value.
  We further connect each pair of these edges to different copy of $=_4$ via two copies of $\neq_2$.
  This results in a gadget with signature $f$ whose support is the complement of the original support,
  that is, $\{11001100,00110011,00000000\}$.

  Now consider the gadget in Figure~\ref{subfig:Eq4-PM3:cycle}.
  We assign $\neq_2$ to the square vertices,
  $=_4$ to the circle vertices,
  $[0,1,0,0]$ to the triangle vertices,
  and $f$ to the pentagon vertex.
  Notice that each pair of edges coming out of the pentagon vertex are from the same corner of the gadget in Figure~\ref{fig:Eq4-PM3:grid} used to realize $f$.
  We now study the signature of this gadget.
  
  Notice that if a $=_4$ on the outer cycle is assigned $0$,
  then the two adjacent perfect matchings must match half edges toward that $=_4$,
  and their outgoing edges must be $0$.
  Furthermore, the two $=_4$ one more step away must be $1$.
  A further observation is that any pair of consecutive $=_4$'s cannot be both $0$,
  and if a pair of consecutive $=_4$'s are both $1$,
  then the $[0,1,0,0]$ in the middle must have a $1$ going out.
  In Figure~\ref{subfig:Eq4-PM3:cycle},
  we call the pentagon connecting to four equalities $=_4$ on the upper right $f_1$ and the other one $f_2$.
  Let $g$ be the signature of resulting gadget.
  We further order the external wires of $f_1$, $f_2$, and $g$ counterclockwise,
  each starting from edge marked with a diamond.
  With this notation and these observations,
  we get Table~\ref{table:Eq4-PM3} listing the support of $g$.
  The support of $g$ is $\{11111111, 01111000, 11110000, 10000111, 00000000, 00001111, 00000000\}$,
  and $00000000$ has multiplicity~$2$.

  Next we use a domain pairing argument.
  First we move $=_4$ to the left hand side,
  by contracting four $\neq_2$ into it.
  We apply the domain pairing on the problem $\plholant{=_4}{g}$.
  Specifically, we use $=_4$ as $=_2$, by pairing each pair of edges together.
  We also pair adjacent two outputs of $g$ clockwise, 
  starting from the diamond point.
  Each pair of output wires of $g$ are connected to a pair of wires from $=_4$ on the left hand side.
  Note that $=_4$ enforces that each pair of edges always takes the same value.
  We re-interpret $00$ or $11$ as $0$ or $1$ in the Boolean domain.
  In this way, we can treat $g$ as an arity~$4$ signature $g'$ in the Boolean domain.
  So the reduction is
  \[
   \plholant{{=}_2}{g'}\le_T\plholant{{=}_4}{g}.
  \]
  We get the expression of $g'$ next.
  The two support bit strings $01111000$ and $10000111$ of $g$
  are eliminated as they do not agree on adjacent paired outputs.
  So in the paired (Boolean) domain, the support of $g'$ becomes
  $\{1111,1100,0011,0000\}$ where $0000$ has multiplicity $2$.
  We further rotate $g'$ as a Boolean domain signature
  such that the support is $\{1111,0110,1001,0000\}$.
  Now it is easy to see that the matrix of $g'$, 
  an arity~$4$ signature in the Boolean domain, is
  \begin{align*}
    M_{g'}=
    \begin{bmatrix}
      2 & 0 & 0 & 0 \\
      0 & 1 & 0 & 0 \\
      0 & 0 & 1 & 0 \\
      0 & 0 & 0 & 1 \\      
    \end{bmatrix}.
  \end{align*}
  By Lemma~\ref{lem:arity4:2-spin} $\PlHolant(g')$ is \numP-hard.
  Hence $\plholant{{\neq}_2}{{=}_4,[0,1,0,0]}$ is \numP-hard.
\end{proof}
%%% JYC do we want to point to the rotational diagram? and list that pre-rotation
% matrix 2 0 0 1\\0 0 0 0 \\0 0 0 0\\1 0 0 1 ?

To extend Lemma~\ref{lem:Z:Eq3PM3} and Lemma~\ref{lem:Z:Eq4PM3}
to general \exactone{d} functions,
we show that we can always realize constant functions $[1,0]$ and $[0,1]$ in this setting.

\begin{lemma}  \label{lem:Z:Eq34:PMk:pin}
  For any integer $k\ge 3$ and $d\ge 3$ and any signature set $\mathcal{F}$, 
  \[\plholant{{\neq}_2}{{=}_k,\exactone{d},[0,1],[1,0],\mathcal{F}}\le_T\plholant{{\neq}_2}{{=}_k,\exactone{d},\mathcal{F}}.\]
\end{lemma}

\begin{proof}
  Given an instance $\Omega$ of $\plholant{{\neq}_2}{{=}_k,\exactone{d},[0,1],[1,0],\mathcal{F}}$
  with underlying planar graph $G$,
  if there is any $[1,0]$ on the right hand side,
  then it can be combined with $\neq_2$ as a $[0,1]$ on the left hand side,
  and then contracted into whatever function it is attached to.
  If it is connected to $[1,0]$ or $[0,1]$,
  we either know the Holant is $0$ or remove the two vertices.
  If it is connected to $\exactone{d}$,
  then the contraction gives us $d-1$ many $[1,0]$ pinnings.
  Similarly, if it is connected to $=_k$,
  the whole function decomposes into $k-1$ many $[0,1]$'s.
  These additional pinnings by $[1,0]$'s or $[0,1]$'s can be recursively applied.

  By a similar analysis,
  it is easy to show that the only nontrivial occurrences of $[0,1]$'s
  are those attached to $\exactone{d}$ via $\neq_2$.
  We may therefore assume there is no $[1,0]$ in $\Omega$,
  and the only appearances of $[0,1]$ and $[1,0]$'s are those of $[0,1]$'s applied to $\exactone{d}$ via $\neq_2$.

  We can  construct $=_{\ell k}$ for any integer $\ell \ge 1$,
by $\neq_2$ on the left and $=_{k}$ on the right.
In fact if we connect two copies of $=_{k}$ via $\neq_2$ we get
a signature of arity $2k-2$ with $k-1$ consecutive external wires labeled $+$
and the others labeled $-$. As $k \ge 3$, we can take 2 wires of the $k-1$
wires labeled $-$ and attach to two copies of $=_{k}$ via two $\neq_2$.
This creates a signature of arity $3(k-1) + (k-3)$
with $3(k-1)$ consecutive wires labeled $+$ and the other
 $k-3$ wires labeled $-$. Finally connect $k-3$ pairs of adjacent
$+/-$ labeled wires by $\neq_2$ recursively. This creates a planar gadget
with an equality signature of arity $3(k-1) - (k-3)
= 2k$.  This can be extended to any $=_{\ell k}$ by applying the same
process on any consecutive $k$ wires.

Next we construct $[0,1]^{\otimes r}$ for some integer $r \ge 1$.
We get $[1,0]^{\otimes d-2}$ by a self-loop of $\exactone{d}$ via $\neq_2$,
ignoring the factor 2.
  We pick an integer $\ell$ large enough so that $d-2 < \ell k$.
  Then we connect $[1,0]^{\otimes d-2}$ to $=_{\ell k}$ via $\neq_2$ to get $[0,1]^{\otimes (\ell k-d+2)}$.
  This is what we claim with $r=\ell k-d+2$.

  One more construction we will use is $\exactone{2+\ell(d-2)}$ for any integer $\ell \ge 1$.
  This is realizable by connecting $\ell$ many copies of $\exactone{d}$ sequentially via $\neq_2$.

  Consider the dual graph $G^*$ of $G$.
  Take a spanning tree $T$ of $G^*$,
  with the external face as the root.
  In each face $F$, let $c_F$ be the number of $[0,1]$'s in the face.
  We start from the leaves to recursively move all the pinnings of $[0,1]$ to the external face.
  Suppose we are working on the face $F$ as a leaf of $T$.
  If $c_F=0$ then we just remove the leaf from $T$ and recurse on another leaf.
  Otherwise we remove all $[0,1]$'s in $F$.
  Let $s$ be the smallest integer such that $s r \ge c_F$.
  We replace the $\neq_2$ edge bordering between $F$ and its parent $F'$
  by a sequence of three signatures:
 $\neq_2$, $\exactone{2+\ell (d-2)}$ and $\neq_2$,
  where $\ell$ is a sufficiently large integer such that $\ell (d-2)\ge sr-c_F$.
  From $\exactone{2+\ell (d-2)}$ there are two edges connected to the two adjacent copies of $\neq_2$.
  Of the other $\ell (d-2)$ edges we will put $sr-c_F$ many dangling edges in $F$, 
  and the remaining $\ell (d - 2) - (s r - c_F)$ dangling edges in $F'$.
  Hence there are $sr$ dangling edges in $F$,
  including those $c_F$ many that were connected to $[0,1]$'s before we removed the $[0,1]$'s.
  We put $s$ copies of $[0,1]^{\otimes r}$ inside the face $F$ to pin all of them in a planar way.
  We add $\ell(d-2)-(sr-c_F)$ to $c_{F'}$.
  Remove the leaf $F$ from $T$,
  and recurse.

  After the process,
  all $[0,1]$'s are in the external face of $G$.
  Suppose the number is $p$.
  We put $r$ disjoint copies of $G$ together to form a planar signature grid.
  Apply a total of $p r$ many $[0,1]$'s by $p$ copies of $[0,1]^{\otimes r}$ in a planar way.
  This is now an instance of $\plholant{{\neq}_2}{{=}_k,\exactone{d},\mathcal{F}}$
  and the Holant value is the $r$th power of that of $\Omega$.
  Since the Holant value of $\Omega$ is a nonnegative integer,
  we can take the $r$th root and finish the reduction.
\end{proof}

Once we have constant functions $[0,1]$ and $[1,0]$,
it is easy to construct $\exactone{3}$ from $\exactone{d}$.
Therefore combining Lemma~\ref{lem:Z:Eq34:PMk:pin} with Lemma~\ref{lem:Z:Eq3PM3} and Lemma~\ref{lem:Z:Eq4PM3} 
we get the following corollary.

\begin{corollary} \label{cor:Z:Eq34:PMk}
 If $d \ge 3$ and $k \in \{3,4\}$,
 then $\plholant{{\neq}_2}{{=}_k,\exactone{d}}$ is \numP-hard.
\end{corollary}

\subsection{Tractability when \texorpdfstring{$k \ge 5$}{k>=5}}

On the other hand,
if the arity of the equality signature is at least~$5$,
then the problem is tractable.
In this subsection we will first prove that
the problem \plholant{{\neq}_2}{{=}_k,\mathcal{EO}}
is tractable for $k \ge 6$.
After that we  will  return to $=_5$.

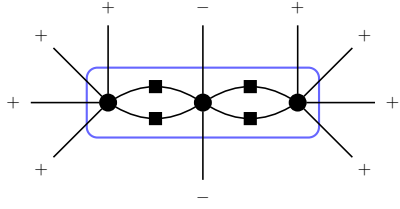
\begin{figure}[t]
 \centering
 \begin{tikzpicture}[scale=\scale,transform shape,node distance=1.5*\nodeDist,semithick]
  \node[internal]  (e1) {};
  \node[internal, right of=e1]  (e2) {};
  \node[internal, right of=e2]  (e3) {};
  \path (e1) edge[bend  left] node[square] (s1) {} (e2)
             edge[bend right] node[square] (s2) {} (e2)
        (e2) edge[bend  left] node[square] (s3) {} (e3)
             edge[bend right] node[square] (s4) {} (e3);
  \node[external, above left of=e1] (ex1-1) {$+$};
  \node[external, below left of=e1] (ex1-2) {$+$};
  \node[external, above      of=e1] (ex1-3) {$+$};
  \node[external,       left of=e1] (ex1-4) {$+$};
  
  \node[external, above right of=e3] (ex3-1) {$+$};
  \node[external, below right of=e3] (ex3-2) {$+$};
  \node[external, above       of=e3] (ex3-3) {$+$};
  \node[external,       right of=e3] (ex3-4) {$+$};
  
  \node[external, above of=e2] (ex2-1) {$-$};
  \node[external, below of=e2] (ex2-2) {$-$};
  
  \path (e1) edge (ex1-1)
             edge (ex1-2)
             edge (ex1-3)
             edge (ex1-4);
  \path (e3) edge (ex3-1)
             edge (ex3-2)
             edge (ex3-3)
             edge (ex3-4);
  \path (e2) edge (ex2-1)
             edge (ex2-2);
  \begin{pgfonlayer}{background}
   \node[draw=\borderColor,thick,rounded corners,fit = (e1) (e3) (s1) (s2) (s3) (s4),inner sep=8pt] {};
%    \node[draw=\borderColor,thick,rounded corners,fit = (e1) (e3) (s1) (s2) (s3) (s4),inner xsep=8pt,inner ysep=4pt,transform shape=false] {};
  \end{pgfonlayer}                              
 \end{tikzpicture}
 \caption{Example \eblock{6}. Circle vertices are assigned $=_6$ and square vertices are assigned $\neq_2$.}
 \label{fig:eblock:example}
\end{figure}

To prove this,
we first do some preprocessing.
Let $G$ be the underlying graph of an instance of $\plholant{{\neq}_2}{{=}_k,\mathcal{EO}}$.
Any self loop on an \exactone{d} by a $\neq_2$ changes it to a $[1,0]^{\otimes(d-2)}$ with factor $2$.
These pinning signatures can be applied recursively.
Any $[1,0]$ is first transformed to
$[0,1]$ via ${\neq}_2$ on LHS 
and then applied either to ${=}_k$ producing $[0,1]^{\otimes(k-1)}$, 
or to \exactone{d} (for some $d$) producing $[1,0]^{\otimes(d-1)}$. 
Similarly, any $[0,1]$ is first transformed to $[1,0]$ via ${\neq}_2$ on LHS 
and then applied either to  ${=}_k$ producing $[1,0]^{\otimes(k-1)}$, 
or to \exactone{d} (for some $d$) producing \exactone{d-1}.
Note that if $d=3$ then \exactone{d-1} is just ${\neq}_2$ on RHS, 
which combined with its adjacent two copies of ${\neq}_2$ of LHS, 
is equivalent to a single ${\neq}_2$ of LHS.
Moreover, whenever an \exactone{d} and another \exactone{\ell} are connected by a $\neq_2$,
we replace it by a single \exactone{d + \ell - 2}, 
shrinking the edge between (and remove the connecting  $\neq_2$).
On the other hand, consider a connected component made of $=_k$ and $\neq_2$.
We call such a component an \eblock{k}.
Notice that each \eblock{k} has either exactly  two or zero  support vectors.
This depends on whether or not there exists a contradiction,
which is formed by an odd cycle of  $=_k$ connected by $\neq_2$.
We say an \eblock{k} is trivial if it has no support.
This is easy to check.
The two support vectors of a nontrivial \eblock{k} are complements of each other.
We mark dangling edges of a nontrivial \eblock{k} by ``$+$'' or ``$-$'' signs.
Dangling edges marked with the same sign take the same value on both support vectors
while dangling edges marked with different signs take opposite values on both support vectors.
Let $n_\pm$ be the number of dangling edges marked $\pm$.
Then it is easy to see that 
\begin{equation} \label{eqn:eblock:mod}
 n_+ \equiv n_- \mod k.
\end{equation}
%%% JYC: proved by induction, first cutting all connections... omitted here.
An example of \eblock{6} is illustrated in Figure~\ref{fig:eblock:example},
with $8$ $+$ signs and $2$ $-$ signs.

After contracting all edges between \exactone{d}'s and forming \eblock{k}'s
we obtain a bipartite graph connected between \exactone{d}'s and \eblock{k}'s
by edges labeled by $=_2$.

A key observation is that a planar (bipartite) graph cannot be simple,
i.e., it must have parallel edges,
if its degrees are large.

\begin{lemma} \label{lem:planar_bipartite}
 Let $G = (L \union R, E)$ be a planar bipartite graph with parts $L$ and $R$.
 If every vertex in $L$ has degree at least~$6$ and every vertex in~$R$ has degree at least~$3$,
 then $G$ is not simple.
\end{lemma}

\begin{proof}
 Suppose $G$ is simple.
 Let $v$, $e$ and $f$ be the total number of vertices, edges, and faces, respectively.
 Let $v_i$ be the number of vertices of degree $i$ in $L$, where $i\ge 6$,
 and $u_j$ be the number of vertices of degree $j$ in $R$, where $j\ge 3$.
 Since $G$ is simple and bipartite,
 each face has at least~$4$ edges.
 Thus,
 \begin{equation} \label{eqn:sum:faces}
  2 e \ge 4 f.
 \end{equation}
 Furthermore,
 it is easy to see that
 \begin{equation} \label{eqn:sum:vertices_edges}
  v = \sum_{i \ge 6} v_i + \sum_{j\ge 3} u_j
  \qquad
  \text{and}
  \qquad
  e = \sum_{i \ge 6} i v_i = \sum_{j \ge 3} j u_j.
 \end{equation}

 Then starting from Euler's characteristic equation for planar graphs,
 we have
 \begin{align*}
  2
  &= v - e + f \\
  &\le v - \frac{e}{2} \tag{By~\eqref{eqn:sum:faces}}\\
  &= \sum_{i\ge 6} v_i + \sum_{j\ge 3} u_j - \frac{1}{6} \sum_{i\ge 6} i v_i - \frac{1}{3}  \sum_{j\ge 3} j u_j \tag{By~\eqref{eqn:sum:vertices_edges}}\\
  &= \sum_{i\ge 6} \frac{6-i}{6} v_i+\sum_{j\ge 3} \frac{3-j}{3} u_j \le 0,
 \end{align*}
 a contradiction.
\end{proof}

Lemma~\ref{lem:planar_bipartite} 
does not give us tractability for the case of $k\ge 6$ yet.
The reason is that given an instance of \plholant{{\neq}_2}{{=}_k,\mathcal{EO}},
after the preprocessing and forming \eblock{k}s to make the graph bipartite,
it is possible to have \eblock{k}s of arity less than $6$,
in which case Lemma~\ref{lem:planar_bipartite} does not apply.
However, for $k\ge 6$ and a nontrivial \eblock{k} of arity $n$ where $n<6$,
by~\eqref{eqn:eblock:mod} and the fact that $0 \le n_+,n_-\le n<k$,
we see that $n_+=n_-$, and $n = n_+ + n_-$ must be even.
Moreover, if $n=2$, then this means that the \eblock{k} is just $\neq_2$,
in which case we can replace it by a single $\neq_2$
connecting signatures from $\mathcal{EO}$ to produce a new \exactone{} signature.
The only problematic case is when $n=4$.
We identify two possibilities of such \eblock{k}s up to a rotation in Figure~\ref{subfig:eblocks:4}.

\begin{figure}[t]
 \centering
 \def\capWidth{6.5cm}
 \captionsetup[subfigure]{width=\capWidth}
 \subfloat[Two different arity~$4$ \eblock{k}s.]{
  \makebox[\capWidth][c]{
   \begin{tikzpicture}[scale=\scale,transform shape,node distance=1.5*\nodeDist,semithick]
    \node[internal] (o1) {};
    \node[external, above left  of=o1] (e1-nw) {$+$};
    \node[external, above right of=o1] (e1-ne) {$+$};
    \node[external, below left  of=o1] (e1-sw) {$-$};
    \node[external, below right of=o1] (e1-se) {$-$};

    \node[external, right of=o1] (aux) {};

    \node[internal, right of=aux] (o2) {};
    \node[external, above left  of=o2] (e2-nw) {$+$};
    \node[external, above right of=o2] (e2-ne) {$-$};
    \node[external, below left  of=o2] (e2-sw) {$-$};
    \node[external, below right of=o2] (e2-se) {$+$};

    \path (o1) edge (e1-nw)
               edge (e1-ne)
               edge (e1-sw)
               edge (e1-se);
    \path (o2) edge (e2-nw)
               edge (e2-ne)
               edge (e2-sw)
               edge (e2-se);
    \begin{pgfonlayer}{background}
     \node[draw=\borderColor,thick,rounded corners,fit = (o1), inner sep=10pt] {};
%      \node[draw=\borderColor,thick,rounded corners,fit = (o1), inner sep=8pt,transform shape=false] {};
     \node[draw=\borderColor,thick,rounded corners,fit = (o2), inner sep=10pt] {};
%      \node[draw=\borderColor,thick,rounded corners,fit = (o2), inner sep=8pt,transform shape=false] {};
    \end{pgfonlayer}               
   \end{tikzpicture}
  }
  \label{subfig:eblocks:4}
 }
 \qquad
 \subfloat[Replace them by parallel $\neq_2$'s.]{
  \makebox[\capWidth][c]{
   \begin{tikzpicture}[scale=\scale,transform shape,node distance=1.5*\nodeDist,semithick]
    \node[external] (o1) {};
    \node[external, above left  of=o1] (e1-nw) {$+$};
    \node[external, above right of=o1] (e1-ne) {$+$};
    \node[external, below left  of=o1] (e1-sw) {$-$};
    \node[external, below right of=o1] (e1-se) {$-$};

    \node[external, right of=o1] (aux) {};

    \node[external, right of=aux] (o2) {};
    \node[external, above left  of=o2] (e2-nw) {$+$};
    \node[external, above right of=o2] (e2-ne) {$-$};
    \node[external, below left  of=o2] (e2-sw) {$-$};
    \node[external, below right of=o2] (e2-se) {$+$};

    \path (e1-nw) edge[out=-45,  in= 45] node[square] (dis1-1) {} (e1-sw)
          (e1-se) edge[out=135, in=-135] node[square] (dis1-2) {} (e1-ne)
          (dis1-1) edge[dotted] (dis1-2);
    \path (e2-nw) edge[out=-45,  in= 45] node[square] (dis2-1) {} (e2-sw)
          (e2-se) edge[out=135, in=-135] node[square] (dis2-2) {} (e2-ne)
          (dis2-1) edge[dotted] (dis2-2);
    \begin{pgfonlayer}{background}
     \node[draw=\borderColor,thick,rounded corners,fit = (dis1-1) (dis1-2), inner sep=8pt] {};
%      \node[draw=\borderColor,thick,rounded corners,fit = (dis1-1) (dis1-2), inner sep=5pt,transform shape=false] {};
     \node[draw=\borderColor,thick,rounded corners,fit = (dis2-1) (dis2-2), inner sep=8pt] {};
%      \node[draw=\borderColor,thick,rounded corners,fit = (dis2-1) (dis2-2), inner sep=5pt,transform shape=false] {};
    \end{pgfonlayer}
   \end{tikzpicture}
  }
  \label{subfig:eblocks:diseq}
 }
 \caption{Arity $4$ \eblock{k}s.}
 \label{fig:eblocks:4}
\end{figure}

Formally we define a \emph{contraction} process on the connected graph of \eblock{k} with dangling edges.
Recursively, for any non-dangling non-loop edge $e$, 
we shrink it to a point, maintaining planarity.
The local cyclic orders of incident edges of the two vertices of $e$
are spliced along $e$ to form the cyclic order of the new vertex.
For any loop we simply remove it.
This contraction process ends in a single point with a cyclic order of the dangling edges.
Figure~\ref{subfig:eblocks:4} depicts the two possibilities of \eblock{k}s of arity 4 up to a rotation.
An \eblock{k} of arity 4 can be viewed as a pair of
$\neq_2$ in parallel, but there is a correlation between them,
namely their support vectors are paired up in a unique way.
%% JYC
%% i still think it cleaner and clearer to prove realizability to
%% for any n_+,  n_- pattern, externally in a simply connected Eblock.
If we replace the contracted \eblock{k} of arity 4 by
two parallel edges as indicated in Fig~\ref{subfig:eblocks:diseq},
one can revert back to a planar realization in the \eblock{k}
as it connects to the rest of the graph. This can be seen by reversing
the contraction process step by step.

We will show in the following lemma how to replace \eblock{k} of arity $4$ by 
some other signatures while keeping track of the Holant value.
We also observe that this tractable set is compatible with binary $\neq_2$
and unary $[1,0]$ or $[0,1]$ signatures.

\begin{lemma} \label{lem:Z:Eqd-PM3}
 For any integer $k \ge 6$,
 \plholant{{\neq}_2}{{=}_k,\mathcal{EO},\neq_2,[1,0],[0,1]} is tractable.
\end{lemma}

\begin{proof}
 Let $\Omega$ be an instance of \plholant{{\neq}_2}{{=}_k,\mathcal{EO},\neq_2,[1,0],[0,1]}.
 Without loss of generality,
 we assume that $\Omega$ is connected.
Any occurrence of $\neq_2$ of the right hand side can be removed as follows:
It is connected to two adjacent copies of $\neq_2$ of the left hand side.
We replace these 3 copies of $\neq_2$   by a single $\neq_2$ 
from the left hand side.

The given signatures have no weight, however the proof below can be 
adapted to the weighted case. For the unweighted case,
 %Since there is no weight in any of the signatures,
 we only need to count the number of satisfying assignments.
 We call an edge pinned if it has the same value in all satisfying assignments,
\emph{if there is any}.
 Clearly any edge incident to a vertex assigned $[1,0]$ or $[0,1]$ is pinned.
 
 When an edge is pinned to a known value,
 we can get a smaller instance of the problem
 $\PlHolant({\neq}_2\mid =_k,\mathcal{EO},\neq_2,[1,0],[0,1])$
 without changing the number of satisfying assignments.
In our algorithm we  may also find a contradiction and simply return~$0$.
 If $e$ is a pinned edge,
 then it is adjacent to another edge $e'$ via $\neq_2$ on the left hand side,
and both  $e$ and $e'$ are pinned.
 We remove $e$, $e'$, and $\neq_2$, and perform the following 
on  $e$  (and on $e'$ as well).
 If the other endpoint of  $e$ is  $u = [1,0]$ or $[0,1]$
we either remove that  $u$ if the pinned value on
$e$ is consistent with $u$, or return 0 otherwise.
 If the other endpoint of  $e$ is $=_k$,
 then all  edges of this $=_k$ are pinned to the same value which we can
recursively apply.
 %We replace them by $k-1$ many $[1,0]$'s or $[0,1]$'s depending on the pinned value.
 If the other endpoint of $e$ is \exactone{d} $\in \mathcal{EO}$,
 then we replace this signature by \exactone{d-1} when the pinned value is~$0$;
 or if the pinned value is~$1$ then 
 the remaining $d-1$ edges of this \exactone{d} are pinned to~$0$
which we recursively apply.
 %which we replace with $[1,0]$'s.
 Notice that we may create an \exactone{2} (i.e.~$\neq_2$) on the right hand side when
we pin $0$ on \exactone{3}. Such $\neq_2$'s are replaced as described at the beginning.
% ,  or an \exactone{1} (i.e.~$[0,1]$) 
% when edges incident to a vertex assigned \exactone{} are pinned to~$0$.
%%% JYC since i don't have $\neq_2$ on the right around, i don't have to
%%% worry about creating \exactone{1}.
 %If the other endpoint of $e$ is $[1,0]$ or $[0,1]$,
 %then we remove them when there is no contradiction,
 %or return $0$ when there is.
 %We do the same for $e'$ as well.
 It is easy to see that all these procedures do not change 
the number of satisfying assignments, and work in polynomial time.

 We claim that there always exists an edge in $\Omega$ that is pinned,
 unless $\Omega$ does not contain $=_k$,
 or does not contain \exactone{d} functions (for some $d \ge 3$),
 or there is a contradiction.
Furthermore if there are  $=_k$ or  \exactone{d} functions (for some $d \ge 3$),
 in polynomial time  we can find a pinned edge with a known value,
or return that  there is a contradiction.
(If there is a contradiction in $\Omega$,
we may still return a purported pinned edge with a known value,
which we can apply and simplify $\Omega$.
The contradiction will eventually be found.)
%
% Moreover, we can decide which is the case and 
% find the pinned edge, if any, in polynomial time.
If  $\Omega$ does not contain $=_k$,
 or does not contain \exactone{d} functions (for some $d \ge 3$),
 then the problem is tractable,
 since $\Omega$ is an instance of $\mathscr{M}$,
 or an instance of $\mathscr{P}$.
 %or there are no satisfying assignments.
 The lemma follows from the claim,
 for we either recurse on a smaller instance or have a tractable instance.

 %Suppose for a contradiction that our claim does not hold.
Suppose  $\Omega$ is an instance where
at least one $=_k$ and at least one \exactone{d} $\in \mathcal{EO}$ appear.
% (for some $d \ge 3$) appear.
% in $\Omega$,
% there is at least one satisfying assignment,
% and no edge in $\Omega$ is pinned.
We assume no $\neq_2$ appears on the right hand side.
If any  $[1,0]$ or $[1,0]$ appear, then we have found a
pinned edge with a known value. 
 Hence we may assume neither $[1,0]$ nor $[1,0]$ appears in $\Omega$.

 If a signature \exactone{d} $\in \mathcal{EO}$
 is connected to itself by a self-loop through a $\neq_2$,
 then there are two choices for the assignment on this pair of edges through the $\neq_2$,
 but the remaining $d-2 \ge 1$ edges are pinned to~$0$.
We can keep track of the factor 2 and have found a pinned edge
with a known value.
 Thus we may assume there are
no self-loops via $\neq_2$  on \exactone{} signatures.

 Next we consider the case that two
separate signatures \exactone{d} and \exactone{\ell} 
from $\mathcal{EO}$
%%% JYC i used the word separate, not distinct, as d might = \ell
are connected by some number of $\neq_2$'s.
 Depending on the number of connecting edges,
 there are three cases:
 \begin{enumerate}
  \item \label{item:exactone} The connection is by a single $\neq_2$.
   We contract the connecting edge, maintaining planarity,
and replace these three signatures by an \exactone{d + \ell - 2} to 
get a new instance $\Omega'$.
   If an edge is pinned in $\Omega'$ then it is also pinned in $\Omega$
to the same value.
%%% technically have to say :  with a known value...
   %If this is the case, we proceed with $\Omega'$.
%%% I don't know what this "if" refers to
We continue with $\Omega'$.
  \item The connection is by two $\neq_2$'s.
   There are two choices for the assignment on these two pairs of edges 
through  $\neq_2$,
   but the remaining $d + \ell - 4 \ge 2$ edges are pinned to $0$.
   %a contradiction.
  \item The connection is by at least three $\neq_2$'s.
   The three $\neq_2$'s cannot be all satisfied,
   so there is no satisfying assignment,
   a contradiction. We return the value 0.
 \end{enumerate}
 Hence, we may assume 
 there is no connection via any number of $\neq_2$'s among \exactone{} signatures.

Define an \eblock{k} as a connected component composed of $=_k$ and $\neq_2$.
All external connecting edges of each \eblock{k} are marked with $+$ or $-$
and this can be found by testing bipartiteness of a \eblock{k}
where we treat $\neq_2$'s as edges. If any \eblock{k} is not
bipartite, we return 0.
 We contract all \eblock{k}s and maintain planarity.
For each \eblock{k}
we contract two vertices that are connected by an edge, one edge at a time,
and remove self loops in this contraction process.
 If a trivial \eblock{k} appears,
 then there is no satisfying assignment,
 we return 0.
 Thus we may assume all \eblock{k}s are nontrivial.
 If there is a nontrivial \eblock{k} of arity~$2$,
 as discussed earlier,
 its signature is $\neq_2$.
 %We replace it with a single vertex assigned $\neq_2$ to form an instance $\Omega'$,
%%% JYC why add a vertex?
We replace it with an edge labeled by $\neq_2$ to form an instance $\Omega'$,
maintaining planarity,
  such that any pinned edge in $\Omega'$ corresponds to a pinned edge in $\Omega$.
This new edge is between \exactone{} signatures and can be dealt with as
described earlier.
% Hence,
% we assume that no \eblock{k} of arity~$2$ occurs,
 So we may assume the arity of any \eblock{k} is at least~$4$.
 Since $k\ge 6$,
 the only possible \eblock{k}s of arity~$4$ are those in Figure~\ref{subfig:eblocks:4} up to a rotation.
 Since there is at least one \exactone{d} signature with $d\ge 3$,
 forming \eblock{k}s does not consume all of $\Omega$.
 
After these steps we may consider
$\Omega$ a bipartite graph, with
 one side consisting of
\eblock{k}s and the other side \exactone{} signatures.
And they are now connected by edges labeled by $=_2$.

 Suppose there are parallel edges between an \eblock{k} and an \exactone{d} signature.
 We show that this always leads to some pinned edges.
 If two parallel edges are marked by the same sign in the \eblock{k},
 then they must be pinned to~$0$.
 If they are marked by different signs,
 then the remaining $d-2 \ge 1$ edges of the \exactone{d} signature must be pinned to~$0$.
 Therefore, we may assume
there are no parallel edges between any \eblock{k} and any \exactone{} signature.

\begin{figure}[t]
 \centering
 \def\capWidth{6.5cm}
 \captionsetup[subfigure]{width=\capWidth}
 \subfloat[An \eoeqblock.
  Triangles are assigned \exactone{} signatures
  and circles are \eblock{k}s of arity~$4$.]{
  \makebox[\capWidth][c]{
   \begin{tikzpicture}[scale=\scale,transform shape,node distance=1.5*\nodeDist,semithick]
    \node[internal, label={10:{$+$}}, label={170:{$+$}}, label={190:{$-$}}, label={-10:{$-$}}] (o1) {};
    \node[triangle, above right of=o1] (p1) {};
    \node[triangle, below right of=o1] (p2) {};    
    \node[internal, below right of=p1, label={10:{$+$}}, label={170:{$-$}}, label={190:{$+$}}, label={-10:{$-$}}] (o2) {};    
    \node[triangle, below left  of=o1] (p3) {};
    \node[triangle, above left  of=o1] (p4) {};
    \node[triangle, above right of=o2] (p5) {};
    \node[triangle, below right of=o2] (p6) {};
    \path (o1) edge (p1)
               edge (p2)
               edge (p3)
               edge (p4);
    \path (o2) edge (p1)
               edge (p2)
               edge (p5)
               edge (p6);
    
    \node[external, above of=p4] (ext1) {};
    \node[external, left  of=p4] (ext2) {};
    \node[external, below of=p3] (ext3) {};
    \node[external, left of=p3]  (ext4) {};
    \node[external, below left of=p3]  (ext5) {};    
    \path (p4) edge (ext1)
               edge (ext2)
          (p3) edge (ext3)
               edge (ext4)
               edge (ext5);
    
    \node[external, above of=p5] (ext7) {};
    \node[external, right of=p5] (ext8) {};
    \node[external, below of=p6] (ext9) {};
    \node[external, right of=p6] (ext10) {};
    \path (p5) edge (ext7)
               edge (ext8)
          (p6) edge (ext9)
               edge (ext10);
    
    \node[external, above right of=p1] (ext11) {};
    \node[external, above left  of=p1] (ext12) {};
    \node[external, below of=p2] (ext13) {};
    \path (p1) edge (ext11)
               edge (ext12)
          (p2) edge (ext13);               
    
    \begin{pgfonlayer}{background}
     \node[draw=\borderColor,thick,rounded corners,fit = (p4) (p6), inner sep=14pt] {};
%      \node[draw=\borderColor,thick,rounded corners,fit = (p4) (p6), inner sep=8pt,transform shape=false] {};
    \end{pgfonlayer}               
   \end{tikzpicture}
  }
  \label{subfig:eoeqblock}
 }
 \qquad
 \subfloat[Break the \eoeqblock\ into three components.
  Squares are assigned $\neq_2$.
  The component in the middle contains a cycle,
  and hence is degenerate.
  The other two are equivalent to \exactone{} signatures.]{
  \makebox[\capWidth][c]{
   \begin{tikzpicture}[scale=\scale,transform shape,node distance=1.5*\nodeDist,semithick]
    \node[external, ] (o1) {};
    \node[triangle, above right of=o1] (p1) {};
    \node[triangle, below right of=o1] (p2) {};    
    \node[external, below right of=p1] (o2) {};    
    \node[triangle, below left  of=o1] (p3) {};
    \node[triangle, above left  of=o1] (p4) {};
    \node[triangle, above right of=o2] (p5) {};
    \node[triangle, below right of=o2] (p6) {};
    
    \path (p1) edge[out=-45,  in= 45] node[square] (dis12-1) {} (p2)
          (p1) edge[out=-135, in=135] node[square] (dis12-2) {} (p2);
    \path (p5) edge[out=-135, in=135] node[square] (dis56) {} (p6);
    \path (p3) edge[out= 45,  in=-45] node[square] (dis34) {} (p4);
    \path (dis12-1) edge[dotted] (dis56);
    \path (dis12-2) edge[dotted] (dis34);
    
    \node[external, above of=p4] (ext1) {};
    \node[external, left  of=p4] (ext2) {};
    \node[external, below of=p3] (ext3) {};
    \node[external, left of=p3]  (ext4) {};
    \node[external, below left of=p3]  (ext5) {};    
    \path (p4) edge (ext1)
               edge (ext2)
          (p3) edge (ext3)
               edge (ext4)
               edge (ext5);
    
    \node[external, above of=p5] (ext7) {};
    \node[external, right of=p5] (ext8) {};
    \node[external, below of=p6] (ext9) {};
    \node[external, right of=p6] (ext10) {};
    \path (p5) edge (ext7)
               edge (ext8)
          (p6) edge (ext9)
               edge (ext10);
    
    \node[external, above right of=p1] (ext11) {};
    \node[external, above left  of=p1] (ext12) {};
    \node[external, below of=p2] (ext13) {};
    \path (p1) edge (ext11)
               edge (ext12)
          (p2) edge (ext13);               
    
    \begin{pgfonlayer}{background}
     \node[draw=\borderColor,thick,rounded corners,fit = (p4) (p3) (dis34), inner sep=10pt] {};
%      \node[draw=\borderColor,thick,rounded corners,fit = (p4) (p3) (dis34), inner sep=5pt,transform shape=false] {};
     \node[draw=\borderColor,thick,rounded corners,fit = (p5) (p6) (dis56), inner sep=10pt] {};
%      \node[draw=\borderColor,thick,rounded corners,fit = (p5) (p6) (dis56), inner sep=5pt,transform shape=false] {};
     \node[draw=\borderColor,thick,rounded corners,fit = (p1) (p2) (dis12-1) (dis12-2), inner sep=10pt] {};
%      \node[draw=\borderColor,thick,rounded corners,fit = (p1) (p2) (dis12-1) (dis12-2), inner sep=5pt,transform shape=false] {};
    \end{pgfonlayer}               
   \end{tikzpicture}
  }
  \label{subfig:eoeqblocks:break}
 }
 \caption{\eoeqblock s}
 \label{fig:eoeqblocks}
\end{figure}

 The next thing we do is to consider \eblock{k}s of arity~$4$ with \exactone{} signatures together.
 Call a connected component consisting of \eblock{k}s of arity~$4$
and \exactone{} an \eoeqblock.
 Figure~\ref{subfig:eoeqblock} illustrates an example.
 Notice that the two possibilities of \eblock{k}s of arity~$4$ 
 can be viewed as two parallel $\neq_2$'s but with some correlation between them.
 This is illustrated in Figure~\ref{subfig:eblocks:diseq}.
 Note that the two dotted lines in Figure~\ref{subfig:eblocks:diseq} represent different correlations.

At this point we would like to replace every arity~$4$ \eblock{k}
 by two parallel $\neq_2$'s. However this replacement destroys the
equivalence of the Holant values, before and after.

\emph{The surprising move of this proof is that we shall do so anyway!}

 Suppose we ignore the correlation for the time being
 and replace every arity~$4$ \eblock{k} 
 by two parallel $\neq_2$'s as in Figure~\ref{subfig:eblocks:diseq}.
%If we contract the \eblock{k}s of arity 4
% one edge at a time as described earlier,
%then it can be seen that 
 This replacement produces a \emph{planar} signature grid $\Omega_1$.
 %Call the resulting grid $\Omega_1$.
 Every edge in $\Omega_1$ corresponds to a unique edge in $\Omega$.
 The set of satisfying assignments of $\Omega_1$
 is a superset of that of $\Omega$.
 Moreover, if there is an edge pinned in $\Omega_1$ to a known value,
 the corresponding edge is also pinned in $\Omega$ to the same value.
Once we find that in $\Omega_1$
 we revert back to work in $\Omega$ and apply the pinning to
the pinned edge.

 All that remains to be shown is that pinning always happens in $\Omega_1$.
 Each \eoeqblock\ splits into
 some number of connected components in $\Omega_1$.
 If any component contains a cycle
 (which must alternate between $\neq_2$, which are
the newly created ones
from the \eblock{k}s of arity 4, and \exactone{d} signatures for $d\ge 3$),
 then any edges not in the cycle but incident to some vertex in the cycle is 
pinned to $0$.
Moreover such edges must exist,
 for \exactone{d} signatures in the cycle are of arity at least $3$.
Note that the cycle has even length, and  there are
exactly two satisfying assignments, which assign exactly one 0 and
 one 1 to the two cycle edges incident to each \exactone{d} signature.  
This produces pinned edges.

 Hence we may assume there are no cycles in these components,
 and every such component forms a tree,
 whose vertices are \exactone{} functions and edges are $\neq_2$'s.
 Suppose there are $n\ge 2$ vertices in such a tree.
 As discussed in item~\ref{item:exactone} above,
 the whole tree is an \exactone{t} function for some arity $t$.
 Since each vertex in the tree has degree at least $3$,
 $t\ge 3n-2(n-1)=n+2\ge 4$.
 We replace these components by \exactone{t}'s.

 Thus, each connected component in the graph underlying $\Omega_1$ 
 is a planar bipartite graph with \eblock{k}s of arity at least~$6$ on one side
 and \exactone{d} signatures of arity at least $3$ on the other.
 By Lemma~\ref{lem:planar_bipartite},
 no component is simple,
 which means that there are parallel edges between some \eblock{k} and some \exactone{d} signature.
 As discussed earlier, there must exist some pinned edge,
and we can find a pinned edge with a known  value in polynomial time.
This finishes the proof.
\end{proof}

Unlike the situation in Lemma~\ref{lem:planar_bipartite},
a planar $(5,3)$-regular bipartite graph \emph{can} be simple.
However,
we show that such graphs must have a special induced subgraph.
% structure.
We call this structure a ``wheel'',
which is pictured in Figure~\ref{fig:Eq5-PM3}.
There is a vertex $v$ of degree~$5$ in the middle,
and all faces adjacent to this vertex are $4$-gons (i.e.~quadrilaterals).
Moreover,
at least four neighbors of $v$ have degree~$3$.
Depending on the degree of the fifth neighbor
(whether it is~$3$ or not),
we have two types of wheel, 
which are pictured in Figure~\ref{subfig:wheel:type1} and Figure~\ref{subfig:wheel:type2}.

\begin{figure}[t]
 \centering
 \subfloat[Type 1]{
  \begin{tikzpicture}[scale=\scale,transform shape,node distance=1.5*\nodeDist,semithick]
   \draw    (0,0)    node[internal]  (o) {};
   \draw    (0,3)    node[internal] (e1) {};
   \draw    (2,3)    node[triangle] (p1) {};
   \draw   (-2,3)    node[triangle] (p2) {};
   \draw    (4,1)    node[internal] (e2) {};
   \draw    (4,-1)   node[triangle] (p3) {};
   \draw   (-4,1)    node[internal] (e3) {};
   \draw   (-4,-1)   node[triangle] (p4) {};
   \draw    (0,-3)   node[triangle] (p5) {};
   \draw    (2,-3)   node[internal] (e4) {};
   \draw   (-2,-3)   node[internal] (e5) {};
   
   \path (o) edge (p1)
             edge (p2)
             edge (p3)
             edge (p4)
             edge (p5);
   \path (e1) edge (p1)
              edge (p2);
   \path (e2) edge (p1)
              edge (p3);
   \path (e3) edge (p2)
              edge (p4);
   \path (e4) edge (p3)
              edge (p5);
   \path (e5) edge (p4)
              edge (p5);
   
   \draw (-0.5,4.5) node[external] (ex-11) {};
   \draw  (0.5,4.5) node[external] (ex-13) {};
   \path (e1) edge node[external,pos=0.9] (inv-11) {} (ex-11)
              edge node[external,pos=0.9] (inv-12) {} (ex-13);
   \path (inv-11) edge[opacity=0] node[external,opacity=1] () {$\cdots$} (inv-12);
   
   \draw  (5.5,0.5) node[external] (ex-21) {};
   \draw  (5.5,1.5) node[external] (ex-23) {};
   \path (e2) edge node[external,pos=0.9] (inv-21) {} (ex-21)
              edge node[external,pos=0.9] (inv-22) {} (ex-23);
   \path (inv-21) edge[dotted] (inv-22);
   
   \draw (-5.5,0.5) node[external] (ex-31) {};
   \draw (-5.5,1.5) node[external] (ex-33) {};
   \path (e3) edge node[external,pos=0.9] (inv-31) {} (ex-31)
              edge node[external,pos=0.9] (inv-32) {} (ex-33);
   \path (inv-31) edge[dotted] (inv-32);
   
   \draw (1.5,-4.5) node[external] (ex-41) {};
   \draw (2.5,-4.5) node[external] (ex-43) {};
   \path (e4) edge node[external,pos=0.9] (inv-41) {} (ex-41)
              edge node[external,pos=0.9] (inv-42) {} (ex-43);
   \path (inv-41) edge[opacity=0] node[external,opacity=1] () {$\cdots$} (inv-42);
   
   \draw (-1.5,-4.5) node[external] (ex-51) {};
   \draw (-2.5,-4.5) node[external] (ex-53) {};
   \path (e5) edge node[external,pos=0.9] (inv-51) {} (ex-51)
              edge node[external,pos=0.9] (inv-52) {} (ex-53);
   \path (inv-51) edge[opacity=0] node[external,opacity=1] () {$\cdots$} (inv-52);
   
   \begin{pgfonlayer}{background}
    \node[draw=\borderColor,thick,rounded corners,fit = (e1) (e2) (e3) (e4) (e5) (p1) (p2) (p3) (p4) (p5),inner sep=8pt] {};
%     \node[draw=\borderColor,thick,rounded corners,fit = (e1) (e2) (e3) (e4) (e5) (p1) (p2) (p3) (p4) (p5),inner sep=6pt,transform shape=false] {};
   \end{pgfonlayer}
  \end{tikzpicture}
  \label{subfig:wheel:type1}
 }
 \qquad
 \subfloat[Type 2]{
  \begin{tikzpicture}[scale=\scale,transform shape,node distance=1.5*\nodeDist,semithick]
   \draw    (0,0)    node[internal]  (o) {};
   \draw    (0,3)    node[internal] (e1) {};
   \draw    (2,3)    node[triangle] (p1) {};
   \draw   (-2,3)    node[triangle] (p2) {};
   \draw    (4,1)    node[internal] (e2) {};
   \draw    (4,-1)   node[triangle] (p3) {};
   \draw   (-4,1)    node[internal] (e3) {};
   \draw   (-4,-1)   node[triangle] (p4) {};
   \draw    (0,-3)   node[triangle] (p5) {};
   \draw    (2,-3)   node[internal] (e4) {};
   \draw   (-2,-3)   node[internal] (e5) {};
   
   \path (o) edge (p1)
             edge (p2)
             edge (p3)
             edge (p4)
             edge (p5);
   \path (e1) edge (p1)
              edge (p2);
   \path (e2) edge (p1)
              edge (p3);
   \path (e3) edge (p2)
              edge (p4);
   \path (e4) edge (p3)
              edge (p5);
   \path (e5) edge (p4)
              edge (p5);
   
   \draw (-0.5,4.5) node[external] (ex-11) {};
   \draw  (0.5,4.5) node[external] (ex-13) {};
   \path (e1) edge node[external,pos=0.9] (inv-11) {} (ex-11)
              edge node[external,pos=0.9] (inv-12) {} (ex-13);
   \path (inv-11) edge[opacity=0] node[external,opacity=1] () {$\cdots$} (inv-12);
   
   \draw  (1.5,4.5) node[external] (ex-p1) {};
   \draw  (2.5,4.5) node[external] (ex-p3) {};
   \path (p1) edge node[external,pos=0.9] (inv-p1) {} (ex-p1)
              edge node[external,pos=0.9] (inv-p2) {} (ex-p3);
   \path (inv-p1) edge[opacity=0] node[external,opacity=1] () {$\cdots$} (inv-p2);
   
   \draw  (5.5,0.5) node[external] (ex-21) {};
   \draw  (5.5,1.5) node[external] (ex-23) {};
   \path (e2) edge node[external,pos=0.9] (inv-21) {} (ex-21)
              edge node[external,pos=0.9] (inv-22) {} (ex-23);
   \path (inv-21) edge[dotted] (inv-22);
   
   \draw (-5.5,0.5) node[external] (ex-31) {};
   \draw (-5.5,1.5) node[external] (ex-33) {};
   \path (e3) edge node[external,pos=0.9] (inv-31) {} (ex-31)
              edge node[external,pos=0.9] (inv-32) {} (ex-33);
   \path (inv-31) edge[dotted] (inv-32);
   
   \draw (1.5,-4.5) node[external] (ex-41) {};
   \draw (2.5,-4.5) node[external] (ex-43) {};
   \path (e4) edge node[external,pos=0.9] (inv-41) {} (ex-41)
              edge node[external,pos=0.9] (inv-42) {} (ex-43);
   \path (inv-41) edge[opacity=0] node[external,opacity=1] () {$\cdots$} (inv-42);
   
   \draw (-1.5,-4.5) node[external] (ex-51) {};
   \draw (-2.5,-4.5) node[external] (ex-53) {};
   \path (e5) edge node[external,pos=0.9] (inv-51) {} (ex-51)
              edge node[external,pos=0.9] (inv-52) {} (ex-53);
   \path (inv-51) edge[opacity=0] node[external,opacity=1] () {$\cdots$} (inv-52);
   
   \begin{pgfonlayer}{background}
    \node[draw=\borderColor,thick,rounded corners,fit = (e1) (e2) (e3) (e4) (e5) (p1) (p2) (p3) (p4) (p5),inner sep=8pt] {};
%     \node[draw=\borderColor,thick,rounded corners,fit = (e1) (e2) (e3) (e4) (e5) (p1) (p2) (p3) (p4) (p5),inner sep=6pt,transform shape=false] {};
   \end{pgfonlayer}
  \end{tikzpicture}
  \label{subfig:wheel:type2}
 }
 \caption{Two types of wheels. Each circle is an \eblock{5} and triangle an \exactone{} signature.}
 \label{fig:Eq5-PM3}
\end{figure}
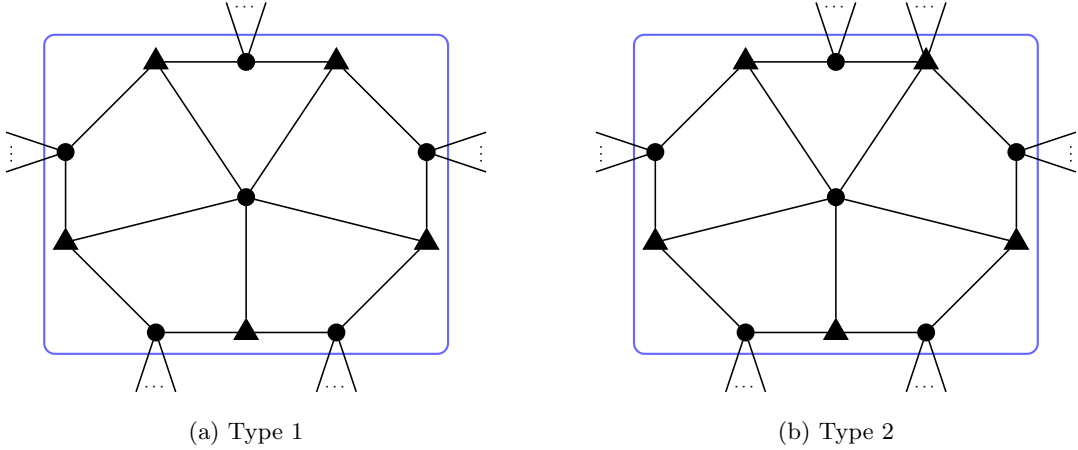

\begin{lemma} \label{lem:planar_bipartite:35}
 Let $G = (L \union R, E)$ be a planar bipartite graph with parts $L$ and $R$.
 Every vertex in $L$ has degree at least\ $5$ and every vertex in\ $R$ has degree at least\ $3$.
 If $G$ is simple, then there exists one of the two wheel structures in Figure~\ref{fig:Eq5-PM3} in $G$.
\end{lemma}

\begin{proof}
  Let $V=L\cup R$ be the set of vertices and
let $F$ be the set of faces.
  We assign a \emph{score} $s_v$ to each vertex $v\in V$.
  We will define $s_v$ so that $\sum_{v\in V}s_v=|V|-|E|+|F|=2>0$.
  The base score is $+1$ for each vertex, which accounts for $|V|$.
  For each $k$-gon face, we assign $\frac{1}{k}$ to each of its vertex.
  This accounts for $|F|$.
  As $G$ is a bipartite and a simple graph,
  $k\ge 4$ and a score coming from a face to a vertex
 is at most $\frac{1}{4}$.

  For $-|E|$, we separate two cases.
  For any edge if one of the two endpoints has degree $3$,
  we give the degree $3$ vertex a score of $-\frac{7}{12}$,
  and the other one $-\frac{5}{12}$.
  This is well defined because all degree $3$ vertices are in $R$.
  If the endpoints are not of degree $3$, we give each endpoint $-\frac{1}{2}$.
  This accounts for $-|E|$.

  Now we claim that $s_v\le 0$ unless $v\in L$ and has degree $5$.
  Suppose $v\in L$ and has degree $d\ge 6$, then
  \begin{align*}
    s_v&\le 1+\frac{d}{4} - \frac{5}{12}d =1-\frac{d}{6}\le 0.
  \end{align*}
  Now suppose $v\in R$ and $v$ has degree $d\ge 4$. 
  Then every edge adjacent to $v$ gives a score $-\frac{1}{2}$.
  Hence,
  \begin{align*}
    s_v&\le 1+\frac{d}{4} - \frac{1}{2}d=1-\frac{d}{4}\le 0.
  \end{align*}
  The remaining case is that $v\in R$ and $v$ has degree $3$.
  Then,
  \begin{align*}
    s_v&\le 1+\frac{d}{4} - \frac{7}{12}d=1-\frac{d}{3}\le 0.
  \end{align*}
  The claim is proved.

  Since the total score is positive, 
  there must exist $v\in L$,  $v$ has degree $5$ and $s_v>0$.
  We then claim that there must exist such a $v$ so that all adjacent faces are $4$-gons.
  Suppose otherwise.
  Then any such $v$ is adjacent to at least one $k$-gon with $k\ge 6$.
  In this case,
  \begin{align*}
    s_v&\le 1+\frac{1}{4}\cdot4+\frac{1}{6}-\frac{5}{12}\cdot 5 = \frac{1}{12}.
  \end{align*}
  Moreover, if $v$ is adjacent to more than one $k$-gon with $k\ge 6$,
  Then 
  \begin{align*}
    s_v&\le 1+\frac{1}{4}\cdot 3+\frac{1}{6}\cdot 2-\frac{5}{12}\cdot 5 = 0,
  \end{align*}
  contrary to the assumption that $s_v >0$.
  Hence $v$ is adjacent to exactly one $k$-gon with $k\ge 6$.
  Call this face $F_v$.

  In $F_v$, $v$ has two neighbors in $R$.
  We match each vertex $v$ that has a positive score to
the vertex on $F_v$ that is the next one in clockwise order from $v$.
By bipartiteness, every such $v$ is matched to a vertex in $R$.
% their own clockwise next one in $f_v$.
  We do this matching in all faces containing at least one positively scored vertex.
It is possible that more than one such $v$ are matched to the same $u \in R$.
  Suppose a vertex $u\in R$ is matched to from $\ell$ different 
such vertices of positive score.
  This means that $u$ is adjacent to at least $\ell$ many $k$-gons with $k\ge 6$.
Then, if  $u$ has degree $3$ then $u$ has score
  \begin{align*}
    s_u\le 1+\frac{1}{4}\cdot (3-\ell)+\frac{1}{6}\cdot\ell-\frac{7}{12}\cdot 3 = -\frac{\ell}{12}.
  \end{align*}
If $u$ has degree $d \ge 4$ then $u$ has score
  \begin{align*}
    s_u\le 1+\frac{1}{4}\cdot (d-\ell)+\frac{1}{6}\cdot\ell-\frac{1}{2}\cdot d \le -\frac{\ell}{12}.
  \end{align*}
  Hence in any case, we have  $s_u\le -\frac{\ell}{12}$.
  It implies that the total 
score of $u$ and all positively scored vertices matched to $u$ is at most $0$.
  However each positively scored vertex is matched to a vertex in $R$.
  Hence the total score cannot be positive.
  This is  a contradiction.

  Therefore there exists $v\in L$ such that $s_v>0$, 
and has degree $5$, and all adjacent faces are $4$-gons.
  We further note that at most one neighbor of $v$ can have
 degree $\ge 4$,
  for otherwise,
  \begin{align*}
    s_v\le 1+\frac{5}{4}-\frac{1}{2}\cdot 2-\frac{5}{12}\cdot 3 = 0.
  \end{align*}
  If all neighbors of $v$ have degree $3$,
  that is a wheel of type 1 as in Figure~\ref{subfig:wheel:type1}.
  If one neighbor of $v$ has degree $\ge 4$,
  that is a wheel of type 2 as in Figure~\ref{subfig:wheel:type2}.
\end{proof}

As we shall see,
either structure in Figure~\ref{fig:Eq5-PM3} leads to pinned edges.

\begin{lemma} \label{lem:Z:Eq5-PM3}
 \plholant{{\neq}_2}{{=}_5,\mathcal{EO},\neq_2,[1,0],[0,1]} is tractable.
\end{lemma}

\begin{proof}
  We proceed as in Lemma~\ref{lem:Z:Eqd-PM3}
  up until the point of getting $\Omega_1$.
  Note that due to \eqref{eqn:eblock:mod}
  the only nontrivial \eblock{5}s of arity $\le 4$ are 
  $\neq_2$ and those in Figure~\ref{subfig:eblocks:4}.
  Moreover, each connected component of $\Omega_1$ is planar and bipartite
  with vertices on one side having degree at least $5$ and those on the other at least $3$.
  We only need to show that there are edges pinned in $\Omega_1$.

  Unlike in Lemma~\ref{lem:Z:Eqd-PM3}, 
  these components do not satisfy the condition of Lemma~\ref{lem:planar_bipartite} 
  but that of Lemma~\ref{lem:planar_bipartite:35}.
  If any such component is not simple,
  then there are pinned edges similar to Lemma~\ref{lem:Z:Eqd-PM3}.
  Otherwise by Lemma~\ref{lem:planar_bipartite:35},
  the wheel structure in Figure~\ref{fig:Eq5-PM3} appears.
  All we need to show is that wheel structures of either type contain pinned edges.

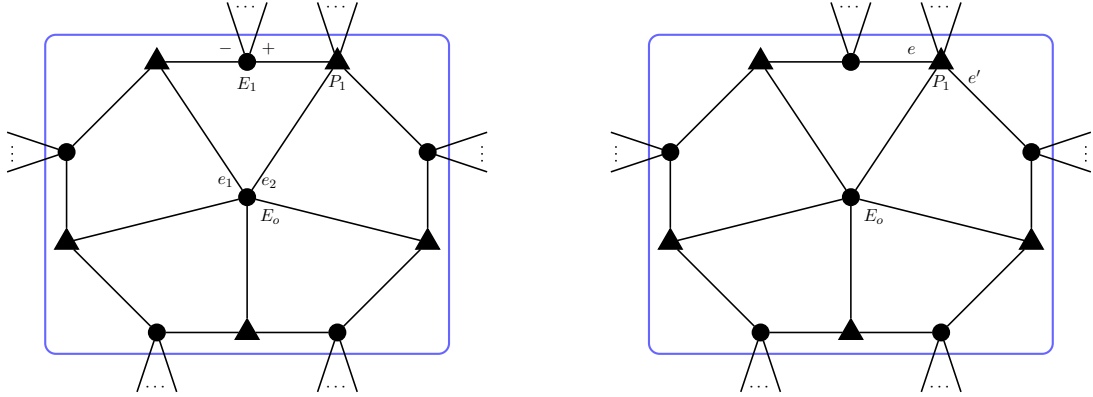
\begin{figure}[t]
 \centering
 \def\capWidth{7cm}
 \captionsetup[subfigure]{width=\capWidth}
 \subfloat[Different signs of an \eblock{5} along the cycle lead to pinning]{
  \makebox[\capWidth][c]{
   \begin{tikzpicture}[scale=\scale,transform shape,node distance=1.5*\nodeDist,semithick]
    \draw    (0,0)    node[internal, label={20:$e_2$}, label={150:$e_1$}, label={-40:$E_o$}] (o) {};
    \draw    (0,3)    node[internal, label={10:$+$}, label={170:$-$}, label={-90:$E_1$}] (e1) {};
    \draw    (2,3)    node[triangle, label={-90:$P_1$}] (p1) {};
    \draw   (-2,3)    node[triangle] (p2) {};
    \draw    (4,1)    node[internal] (e2) {};
    \draw    (4,-1)   node[triangle] (p3) {};
    \draw   (-4,1)    node[internal] (e3) {};
    \draw   (-4,-1)   node[triangle] (p4) {};
    \draw    (0,-3)   node[triangle] (p5) {};
    \draw    (2,-3)   node[internal] (e4) {};
    \draw   (-2,-3)   node[internal] (e5) {};
    
    \path (o) edge (p1)
              edge (p2)
              edge (p3)
              edge (p4)
              edge (p5);
    \path (e1) edge (p1)
               edge (p2);
    \path (e2) edge (p1)
               edge (p3);
    \path (e3) edge (p2)
               edge (p4);
    \path (e4) edge (p3)
               edge (p5);
    \path (e5) edge (p4)
               edge (p5);

    \draw (-0.5,4.5) node[external] (ex-11) {};
    \draw  (0.5,4.5) node[external] (ex-13) {};
    \path (e1) edge node[external,pos=0.9] (inv-11) {} (ex-11)
               edge node[external,pos=0.9] (inv-12) {} (ex-13);
    \path (inv-11) edge[opacity=0] node[external,opacity=1] () {$\cdots$} (inv-12);

    \draw  (1.5,4.5) node[external] (ex-p1) {};
    \draw  (2.5,4.5) node[external] (ex-p3) {};
    \path (p1) edge node[external,pos=0.9] (inv-p1) {} (ex-p1)
               edge node[external,pos=0.9] (inv-p2) {} (ex-p3);
    \path (inv-p1) edge[opacity=0] node[external,opacity=1] () {$\cdots$} (inv-p2);

    \draw  (5.5,0.5) node[external] (ex-21) {};
    \draw  (5.5,1.5) node[external] (ex-23) {};
    \path (e2) edge node[external,pos=0.9] (inv-21) {} (ex-21)
               edge node[external,pos=0.9] (inv-22) {} (ex-23);
    \path (inv-21) edge[dotted] (inv-22);

    \draw (-5.5,0.5) node[external] (ex-31) {};
    \draw (-5.5,1.5) node[external] (ex-33) {};
    \path (e3) edge node[external,pos=0.9] (inv-31) {} (ex-31)
               edge node[external,pos=0.9] (inv-32) {} (ex-33);
    \path (inv-31) edge[dotted] (inv-32);

    \draw (1.5,-4.5) node[external] (ex-41) {};
    \draw (2.5,-4.5) node[external] (ex-43) {};
    \path (e4) edge node[external,pos=0.9] (inv-41) {} (ex-41)
               edge node[external,pos=0.9] (inv-42) {} (ex-43);
    \path (inv-41) edge[opacity=0] node[external,opacity=1] () {$\cdots$} (inv-42);

    \draw (-1.5,-4.5) node[external] (ex-51) {};
    \draw (-2.5,-4.5) node[external] (ex-53) {};
    \path (e5) edge node[external,pos=0.9] (inv-51) {} (ex-51)
               edge node[external,pos=0.9] (inv-52) {} (ex-53);
    \path (inv-51) edge[opacity=0] node[external,opacity=1] () {$\cdots$} (inv-52);

    \begin{pgfonlayer}{background}
     \node[draw=\borderColor,thick,rounded corners,fit = (e1) (e2) (e3) (e4) (e5) (p1) (p2) (p3) (p4) (p5),inner sep=8pt] {};
%      \node[draw=\borderColor,thick,rounded corners,fit = (e1) (e2) (e3) (e4) (e5) (p1) (p2) (p3) (p4) (p5),inner sep=6pt,transform shape=false] {};
    \end{pgfonlayer}
   \end{tikzpicture}
  }
  \label{subfig:wheel:different:sign}
 }
 \qquad
 \subfloat[Edges $e$ and $e'$ are pinned in wheels of type~2]{
  \makebox[\capWidth][c]{
   \begin{tikzpicture}[scale=\scale,transform shape,node distance=1.5*\nodeDist,semithick]
    \draw    (0,0)    node[internal, label={-40:$E_o$}]  (o) {};
    \draw    (0,3)    node[internal] (e1) {};
    \draw    (2,3)    node[triangle, label={170:$e~~$}, label={-10:$~~e'$}, label={-90:$P_1$}] (p1) {};
    \draw   (-2,3)    node[triangle] (p2) {};
    \draw    (4,1)    node[internal] (e2) {};
    \draw    (4,-1)   node[triangle] (p3) {};
    \draw   (-4,1)    node[internal] (e3) {};
    \draw   (-4,-1)   node[triangle] (p4) {};
    \draw    (0,-3)   node[triangle] (p5) {};
    \draw    (2,-3)   node[internal] (e4) {};
    \draw   (-2,-3)   node[internal] (e5) {};
    
    \path (o) edge (p1)
              edge (p2)
              edge (p3)
              edge (p4)
              edge (p5);
    \path (e1) edge (p1)
               edge (p2);
    \path (e2) edge (p1)
               edge (p3);
    \path (e3) edge (p2)
               edge (p4);
    \path (e4) edge (p3)
               edge (p5);
    \path (e5) edge (p4)
               edge (p5);

    \draw (-0.5,4.5) node[external] (ex-11) {};
    \draw  (0.5,4.5) node[external] (ex-13) {};
    \path (e1) edge node[external,pos=0.9] (inv-11) {} (ex-11)
               edge node[external,pos=0.9] (inv-12) {} (ex-13);
    \path (inv-11) edge[opacity=0] node[external,opacity=1] () {$\cdots$} (inv-12);

    \draw  (1.5,4.5) node[external] (ex-p1) {};
    \draw  (2.5,4.5) node[external] (ex-p3) {};
    \path (p1) edge node[external,pos=0.9] (inv-p1) {} (ex-p1)
               edge node[external,pos=0.9] (inv-p2) {} (ex-p3);
    \path (inv-p1) edge[opacity=0] node[external,opacity=1] () {$\cdots$} (inv-p2);

    \draw  (5.5,0.5) node[external] (ex-21) {};
    \draw  (5.5,1.5) node[external] (ex-23) {};
    \path (e2) edge node[external,pos=0.9] (inv-21) {} (ex-21)
               edge node[external,pos=0.9] (inv-22) {} (ex-23);
    \path (inv-21) edge[dotted] (inv-22);

    \draw (-5.5,0.5) node[external] (ex-31) {};
    \draw (-5.5,1.5) node[external] (ex-33) {};
    \path (e3) edge node[external,pos=0.9] (inv-31) {} (ex-31)
               edge node[external,pos=0.9] (inv-32) {} (ex-33);
    \path (inv-31) edge[dotted] (inv-32);

    \draw (1.5,-4.5) node[external] (ex-41) {};
    \draw (2.5,-4.5) node[external] (ex-43) {};
    \path (e4) edge node[external,pos=0.9] (inv-41) {} (ex-41)
               edge node[external,pos=0.9] (inv-42) {} (ex-43);
    \path (inv-41) edge[opacity=0] node[external,opacity=1] () {$\cdots$} (inv-42);

    \draw (-1.5,-4.5) node[external] (ex-51) {};
    \draw (-2.5,-4.5) node[external] (ex-53) {};
    \path (e5) edge node[external,pos=0.9] (inv-51) {} (ex-51)
               edge node[external,pos=0.9] (inv-52) {} (ex-53);
    \path (inv-51) edge[opacity=0] node[external,opacity=1] () {$\cdots$} (inv-52);

    \begin{pgfonlayer}{background}
     \node[draw=\borderColor,thick,rounded corners,fit = (e1) (e2) (e3) (e4) (e5) (p1) (p2) (p3) (p4) (p5),inner sep=8pt] {};
%      \node[draw=\borderColor,thick,rounded corners,fit = (e1) (e2) (e3) (e4) (e5) (p1) (p2) (p3) (p4) (p5),inner sep=6pt,transform shape=false] {};
    \end{pgfonlayer}
   \end{tikzpicture}
  }
  \label{subfig:wheel:type2:pin}
 }
 \caption{Degeneracies in the wheel structure.}
 \label{fig:wheel:pin}
\end{figure}

  First we claim that if a wheel of either type has a \eblock{5},
call it $E_1$, on the outer cycle
  which has different signs on the two edges incident to
it along the cycle,
  then the middle $=_5$, denoted by $E_o$, is pinned.
  This is pictured in Figure~\ref{subfig:wheel:different:sign}.
  It does not matter whether the wheel is type $1$ or $2$,
  or the position of $E_1$ relative to the special triangle $P_1$ in type $2$.
Because $E_o$ is an equality, both $e_1$ and $e_2$, the two
edges incident to $E_o$ that are connected to the two
\exactone{} signatures flanking $E_1$,  must take the
same value.
  If both $e_1$ and $e_2$ are assigned $1$,
  then the two incoming wires of $E_1$ along the cycle have to be both assigned $0$,
  whereas they are marked by different signs.
  This is a contradiction.
  Hence both $e_1$ and $e_2$ are pinned to $0$ as well as all edges of $E_o$.

  We may therefore assume that each \eblock{5} has same signs along the outer cycle, either $++$ or $--$.
  If the wheel is of type $1$,
  then there is no valid assignment such that $E_o$ is assigned $0$ 
  because the cycle has odd length.
  In fact if $E_o$ is assigned $0$,
  then we can remove $E_o$ and its incident edges,
  and effectively the five \exactone{} signatures are now $\neq_2$'s forming a 5-cycle linked by binary equalities.
  Hence $E_o$ and all its edges are pinned to $1$.

  Otherwise the wheel is of type $2$, 
  and each \eblock{5} has signs $++$ or $--$ along the outer cycle.
  We denote by $P_1$ the special \exactone{d} function that has arity $d>3$.
  We claim that the two edges $e$ and $e'$ incident to $P_1$ along the cycle are both pinned to $0$.
  This is illustrated in Figure~\ref{subfig:wheel:type2:pin}.
  As $P_1$ is \exactone{d}, at most one of $e$ and $e'$ is $1$.
  If one of $e$ and $e'$ is $1$, the other is $0$,
  and as $P_1$ is an \exactone{d} function its edge to $E_o$ is also $0$,
  and thus all edges incident to $E_o$ are $0$. 
  As all five neighbors of $E_o$ are \exactone{} functions, 
  the four \exactone{3} functions effectively become $(\neq_2)$ functions along the wheel, 
  and we can remove $E_o$ and its incident edges.
  This becomes the same situation as in the previous case of type 1,
  where effectively a cycle of five binary equalities are linked by five binary disequalities,
  which has no valid assignment.
  It implies that both $e$ and $e'$ are pinned to $0$.
  This finishes the proof.
\end{proof}

\subsection{Lemmas related to \texorpdfstring{$\mathscr{M}_4$}{M4} and \texorpdfstring{$\mathscr{P}_2$}{P2}}

Now we prove some lemmas relating to $\mathscr{M}_4$ and $\mathscr{P}_2$ that are used in the proof of the full dichotomy.

\begin{figure}[t]
 \centering
 \begin{tikzpicture}[scale=\scale,transform shape,node distance=1.5*\nodeDist,semithick]
  \node[internal]                (e1) {};
  \node[internal, below of=e1]   (e2) {};
  \node[external, left  of=e2]   (aux3) {};
  \node[internal, left  of=aux3] (e3) {};
  \node[internal, above of=e3]   (e4) {};
  
  \node[external, above right of=e1] (ex1) {};
  \node[external, below right of=e2] (ex2) {};
  \node[external, below left  of=e3] (ex3) {};
  \node[external, above left  of=e4] (ex4) {};
  
  \path (e1) edge node[square] {} (e2)
             edge[out=160, in= 20] node[square] (mid14-1) {} (e4)
             edge[out=200, in=-20] node[square] (mid14-2) {} (e4);
  \path (e3) edge[out= 20, in=160] node[square] (mid23-1) {} (e2)
             edge[out=-20, in=200] node[square] (mid23-2) {} (e2)
             edge node[square] {} (e4);
  \path (mid14-1) edge[dotted] (mid14-2)
        (mid23-1) edge[dotted] (mid23-2);
  
  \path (e1) edge (ex1);
  \path (e2) edge (ex2);
  \path (e3) edge (ex3);
  \path (e4) edge (ex4);
  
  \begin{pgfonlayer}{background}
   \node[draw=\borderColor,thick,rounded corners,fit = (e1) (e2) (e3) (e4) (mid14-1) (mid23-2),inner sep=8pt] {};
%    \node[draw=\borderColor,thick,rounded corners,fit = (e1) (e2) (e3) (e4) (mid14-1) (mid23-2),inner sep=4pt,transform shape=false] {};
  \end{pgfonlayer}
 \end{tikzpicture}
 \caption{Gadget to realize $\hat{g}$ in Lemma~\ref{lem:Eqd:PM:InvPM}.
          Circle vertices are assigned $=_k$ and square vertices are assigned $\neq_2$.
          The number of parallel edges is $k-2$.}
 \label{fig:PM-InvPM:hatg}
\end{figure}
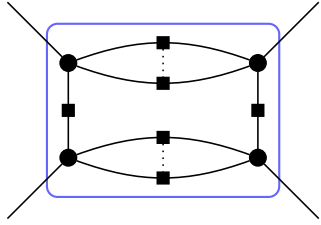

Recall that $\allbutone{d}$ is the signature $[0,\dots,0,1,0]$ of arity $d$,
which is the reverse of $\exactone{d}$.
After a $Z$ transformation,
$\mathscr{M}_4$ contains both $\allbutone{d}$ and $\exactone{d}$.
However, if both appear,
then with any $=_k$ the problem is hard.

\begin{lemma} \label{lem:Eqd:PM:InvPM}
  If integers $d_1, d_2, k \ge 3$,
  then $\plholant{{\neq}_2}{{=}_k,\exactone{d_1},\allbutone{d_2}}$ is \numP-hard.
\end{lemma}

\begin{proof}
  We apply Lemma~\ref{lem:Z:Eq34:PMk:pin} to create constant functions $[1,0]$ and $[0,1]$ first.
  Then we construct $\exactone{4}$ and $\allbutone{4}$.
  With both $[1,0]$ and
$[0,1]$ in hand,
 we may reduce $d_1$ or $d_2$ to $4$ if $d_1> 4$ or $d_2> 4$.
  If either of the two arities is $3$,
  then we connect two copies together via $\neq_2$ to realize an arity $4$ copy.

  Moreover, we use the gadget illustrated in Figure~\ref{fig:PM-InvPM:hatg}
  to create the function $\hat{g}$ in Lemma~\ref{lem:PM-InvPM-g} as an \eblock{k}.
  Then by Lemma~\ref{lem:PM-InvPM-g}, \plholant{{\neq}_2}{{=}_k,\exactone{d_1},\allbutone{d_2}} is \numP-hard.  
\end{proof}

In general signatures in $\mathscr{P}_2$ are non-degenerate
weighted equalities under the $Z$ transformation.
The next several lemmas show that the hardness criterion is the same regardless of the weight.

\begin{lemma}\label{lem:P2:M4+-}
  Let $f\in\mathscr{P}_2$, $g_1\in\mathscr{M}_4^+$, $g_2\in\mathscr{M}_4^-$ be non-degenerate signatures with arity $\ge 3$.
  Then $\PlHolant(f,g_1,g_2)$ is \numP-hard.
\end{lemma}

\begin{proof}
 Suppose the arities of $f$, $g_1$, and $g_2$ are $n$, $m_1$, and $m_2$ respectively.
 Under a holographic transformation by $Z$,
 we have
 \begin{align*}
  \PlHolant(f,g_1,g_2)
  &\equiv \plholant{{\neq}_2}{\left( Z^{-1} \right)^{\otimes n}f, \left( Z^{-1} \right)^{\otimes m_1}g_1,\left( Z^{-1} \right)^{\otimes m_2}g_2}\\
  &\equiv \plholant{{\neq}_2}{\hat{f},
  \exactone{m_1},\allbutone{m_2}},
 \end{align*}
 where $\hat{f}= (Z^{-1})^{\otimes n}f$ which has
the form $[1,0,\dots,0,c]$ up to a nonzero constant, with $c\neq 0$,
 as $f\in\mathscr{P}_2$.
 We do another diagonal transformation by $D=\trans{1}{0}{0}{c^{1/n}}$.
 Then  
 \begin{align*}
  \operatorname{Pl-}&\operatorname{\!Holant}(f,g_1,g_2)\\
  &\equiv
  \PlHolant
  \left(
   (\neq_2) D^{\otimes 2}
   \; \middle| \;
   (D^{-1})^{\otimes n}\hat{f},(D^{-1})^{\otimes m_1}\exactone{m_1}, (D^{-1})^{\otimes m_2}\allbutone{m_2}
  \right)\\
  &\equiv \plholant{{\neq}_2}{{=}_n, \exactone{m_1}, \allbutone{m_2}},
 \end{align*}
 where in the last line we ignored several nonzero factors.
 The lemma follows from Lemma~\ref{lem:Eqd:PM:InvPM}.
\end{proof}

We also need to consider the mixture of $\mathscr{P}_2$ and binary signatures.
%The following lemma also relies on Lemma~\ref{lem:P2:interpolate}.

\begin{lemma}  \label{lem:P2:binary}
  Let $\mathcal{F}$ be a set of symmetric signatures.
  Suppose $\mathcal{F}$ contains a non-degenerate signature $f\in\mathscr{P}_2$ of arity $n\ge 3$ 
  and a binary signature $h$.
  Then $\PlHolant(\mathcal{F})$ is \numP-hard unless $h\in Z\mathscr{P}$,
  or $\PlCSP^2(DZ^{-1}\mathcal{F})\le_T\PlHolant(\mathcal{F})$ for some diagonal transformation $D$.
\end{lemma}

\begin{proof}
  We do a $Z$ transformation and get
  \begin{align*}
    \PlHolant(\mathcal{F})
    &\equiv \PlHolant(\mathcal{F},h,f)\\
    &\equiv \plholant{{\neq}_2}{Z^{-1}\mathcal{F},\left(Z^{-1}\right)^{\otimes 2}h, \hat{f}},
  \end{align*}
  where $\hat{f}= (Z^{-1})^{\otimes n}f = [1,0,\dots,0,t]$ up to a 
  nonzero constant with $t\neq 0$.
  We further do another diagonal transformation of $D_1=\trans{1}{0}{0}{t^{1/n}}$.
  Then  
  \begin{align*}
    \PlHolant(\mathcal{F})
    &\equiv \plholant{(\neq_2) D_1^{\otimes 2}}{(D_1^{-1})^{\otimes n}\hat{f}, (ZD_1)^{-1}\mathcal{F}, \left((Z D_1)^{-1}\right)^{\otimes 2} h}\\
    &\equiv \plholant{{\neq}_2}{{=}_n, (ZD_1)^{-1}\mathcal{F}, \left((Z D_1)^{-1}\right)^{\otimes 2} h}\\
    &\ge_T \plholant{{=}_n}{(ZD_1)^{-1}\mathcal{F}, \left((ZD_1)^{-1}\right)^{\otimes 2} h},
  \end{align*}
  where in the second line we ignore a nonzero factor on $\neq_2$.
  Hence by Theorem~\ref{thm:k-reg_homomorphism},
  $\PlHolant(\mathcal{F})$ is \numP-hard 
  unless $\left( (ZD_1)^{-1}\right)^{\otimes 2}h\in \mathscr{P}$ (cases 1, 2 or 3 in Theorem~\ref{thm:k-reg_homomorphism}) 
  or $\left( (ZD_1)^{-1}\right)^{\otimes 2}h=[a,b,c]$
  for some $a,b,c\in\mathbb{C}$ such that $ac\neq 0$ 
  and $(a/c)^{2n}=1$ (cases 4 or 5 in Theorem~\ref{thm:k-reg_homomorphism}).

  If $\left( (ZD_1)^{-1}\right)^{\otimes 2}h\in \mathscr{P}$,
  then $h\in ZD_1\mathscr{P}=Z\mathscr{P}$ as $D_1\in\Stab{P}$.
  In the latter case,
  we construct $=_{2n}$ on the right by connecting three copies of $=_n$ to one copy of $=_n$ via $\neq_2$.
  We do the same construction again to realize $=_{4n}$ using $=_{2n}$.
  We connect $n-1$ many $[a,b,c]$'s to $=_{2n}$ via $\neq_2$
  to realize a binary weighted equality $[1,0,r]$ with $r=\left( a/c \right)^{n-1}\neq 0$
  ignoring a factor of $c^{n-1}$.
  Note that $r^{2n}=\left(a/c \right)^{2n(n-1)}=1$.
  Then we do another diagonal transformation of $D_2=\trans{1}{0}{0}{r^{1/2}}$ to get
  $\plholant{{\neq}_2}{ (ZD_1D_2)^{-1}\mathcal{F}, =_2, \left( D_2^{-1} \right)^{\otimes 4n}(=_{4n})}$.
  Notice that 
  \begin{align*}
    \left( D_2^{-1} \right)^{\otimes 4n}(=_{4n}) = [1,0,\dots,0,r^{-2n}] = (=_{4n}),
  \end{align*}
  as $r^{2n}=1$.

  Hence we have $=_2$ and $=_{4n}$ on the right.
  With $\neq_2$ on the left,
  we get $=_2$ on the left and therefore equalities of all even arities on the right.
  Let $D=(D_1D_2)^{-1}$.
  Then we have the reduction chain:
  \begin{align*}
    \PlHolant(\mathcal{F}) & \ge_T \plholant{{\neq}_2}{DZ^{-1}\mathcal{F} \union \{=_2,=_{4n}\}}\\
    & \ge_T \plholant{{\neq}_2}{DZ^{-1}\mathcal{F} \union \mathcal{EQ}_2}\\
    & \ge_T \plholant{\mathcal{EQ}_2}{DZ^{-1}\mathcal{F}}.
  \end{align*}
  The last problem is $\PlCSP^2(DZ^{-1}\mathcal{F})$.
  Thus $\PlCSP^2(DZ^{-1}\mathcal{F})\le_T\PlHolant(\mathcal{F})$.
\end{proof}

At last, we strengthen Corollary~\ref{cor:Z:Eq34:PMk},
Lemma~\ref{lem:Z:Eqd-PM3}, and Lemma~\ref{lem:Z:Eq5-PM3} to weighted equalities.
We split the hardness and tractability cases.
For a set $\mathcal{F}$ of signatures,
denote by $\nondeg{F}$ the set of non-degenerate signatures in $\mathcal{F}$ of arity at least~$3$.
Moreover denote by $\mathcal{F}^*$ the signature set that is the same as $\mathcal{F}$
but with each degenerate signature $[a,b]^{\otimes m}$ in $\mathcal{F}$ replaced by the unary $[a,b]$.

Notice that $\mathcal{F}\cap\mathscr{P}_2$
and $\mathcal{F}^*\cap\mathscr{P}_2$ agree on signatures of arity at least $2$,
since signatures in $\mathscr{P}_2$ of arity at least $2$ are non-degenerate.
So $\mathcal{F}\cap\mathscr{P}_2\subseteq \mathcal{F}^*\cap\mathscr{P}_2$,
and the only possible extra elements are some unary $[x,y]$'s 
from $[x,y]^{\otimes m}\in\mathcal{F}$
for some integer $m\ge 2$ and $[x,y]$ is not a multiple of $[1, \pm i]$.
Equivalently the only possible extra elements are unary signatures
of the form $Z[a,b]$ for $ab\neq 0$, i.e.,
\emph{not} of the form a multiple of $Z[1,0]$ or $Z[0,1]$,
when $Z^{-1} \mathcal{F}$ contains some degenerate signatures
of the form $[a,b]^{\otimes m}$ for some integer $m\ge 2$ and $ab\neq 0$.

\begin{lemma}  \label{lem:P2:M4:hardness}
  Let $\mathcal{F}$ be a set of symmetric signatures.
  Let $\nondeg{F}$ be the set of non-degenerate signatures in $\mathcal{F}$ of arity at least $3$. 
  Suppose $\nondeg{F}$ contains $f\in\mathscr{M}_4$ of arity $d\ge 3$.
  Moreover, suppose $\nondeg{F}\cap\mathscr{P}_2$ is nonempty,
  and let $k$ be the greatest common divisor of the arities of signatures in $\mathcal{F}^*\cap\mathscr{P}_2$.
  If $k\le 4$, then $\PlHolant(\mathcal{F})$ is \numP-hard.
\end{lemma}

\begin{proof}
  We may assume that $f\in\mathscr{M}_4^+$.
  Since $\nondeg{F}\cap\mathscr{P}_2$ is nonempty,
  there exists $g\in \nondeg{F}\cap\mathscr{P}_2$.
  By the definition of $\nondeg{F}$, $g$ has arity $n\ge 3$.
  We do a $Z$ transformation,
  \begin{align*}
    \PlHolant(\mathcal{F})
    \equiv
    \plholant{{\neq}_2}{\hat{g},\exactone{d},Z^{-1}\mathcal{F}},
  \end{align*}
  where $\hat{g}= (Z^{-1})^{\otimes n} g$
has the form $[1,0,\dots,0,c]$
% is a weighted equality 
of arity $n$
  for some $c\neq 0$ up to a nonzero factor.
  We further do a diagonal transformation $D=\trans{1}{0}{0}{c^{1/n}}$ and get 
  \begin{align*}
    \PlHolant(\mathcal{F})
    \equiv
    \plholant{{\neq}_2}{{=}_n,\exactone{d},(ZD)^{-1}\mathcal{F}},
  \end{align*}
  where we ignore nonzero factors on $\neq_2$ and \exactone{d}.
  Then by Lemma~\ref{lem:Z:Eq34:PMk:pin},  
  \begin{align*}
    \PlHolant(\mathcal{F})
    \ge_T
    \plholant{{\neq}_2}{{=}_n,\exactone{d},[0,1],[1,0],(ZD)^{-1}\mathcal{F}}.
  \end{align*}
By a weighted equality we mean a signature of the form
$[a, 0, \dots,0,b]$ of some arity $\ge 1$, where $ab \not = 0$. 
Recall that $\mathscr{P}_2$ consists of the $Z$ transformation of
all weighted equalities.
  Let $\mathcal{G}$ be the set of weighted equalities in $(ZD)^{-1}\mathcal{F}$.
  In other words, $\mathcal{G}=(ZD)^{-1}\left(\mathcal{F}\cap\mathscr{P}_2\right)$
  as $(ZD)^{-1}\mathscr{P}_2$ contains all weighted equalities.
  Moreover, up to a nonzero factor, $(=_n)\in\mathcal{G}$.

  Let $k'$ be the gcd of all arities of signatures in $\mathcal{G}$,
or equivalently the gcd of all arities of signatures 
in $\mathcal{F}\cap\mathscr{P}_2$.
  If $k'\neq k$, then the only possibility  is that
  $(ZD)^{-1}\mathcal{F}$ contains a degenerate signature
 $[a,b]^{\otimes m}$ for some $m \ge 2$ with $ab\neq 0$.
  In this case we use pinnings $[1,0]$ or $[0,1]$ to realize $[a,b]$ from $[a,b]^{\otimes m}$
  and put $[a,b]$ in $\mathcal{G}$.
  Hence we may assume that $k'=k$.

  Pick any $g_1,g_2\in\mathcal{G}$ of arities $\ell_1$ and $\ell_2$.
  Let $r=\gcd(\ell_1,\ell_2)$.
  Let $t_1,t_2$ be two positive integers such that $t_1\ell_1-t_2\ell_2=r$.
  Then connecting $t_1$ copies of $g_1$ and $t_2$ copies of $g_2$ via $\neq_2$ in a bipartite and planar way,
  we get a weighted equality signature of arity $r$.

  Apply the same argument repeatedly. 
  Eventually we construct a weighted equality $h$ of arity $k$.
  We further do a diagonal transformation $D_1$ to make it $=_k$,
  that is,
  \begin{align*}
    \PlHolant(\mathcal{F}) & \ge_T\plholant{{\neq}_2}{\mathcal{G},\exactone{d}}\\
    & \ge_T \plholant{{\neq}_2}{h,\exactone{d},\mathcal{G}}\\
    & \ge_T \plholant{(\neq_2)D_1^{\otimes 2}}{{=}_k,\left(D_1^{-1}\right)^{\otimes d}\exactone{d},D_1^{-1}\mathcal{G}}\\
    & \ge_T \plholant{{\neq}_2}{{=}_k,\exactone{d},D_1^{-1}\mathcal{G}},
  \end{align*}
  where in the last line we ignored nonzero factors of $\exactone{d}$ and $\neq_2$.
  If $k=3$ or $4$,
 then the hardness follows from Corollary~\ref{cor:Z:Eq34:PMk}.

  If $k=1$ or $2$, then on the right hand side
  we have $=_k$, which is $=_1$ or $=_2$, and a weighted equality 
  $\left( D_1^{-1} \right)^{\otimes n}(=_n)\in D_1^{-1}\mathcal{G}$.
  Call it $\hat{g}'$.
  We move the $=_k$ to the left hand side via $\neq_2$.
  Then we connect zero or more copies of this $=_k$, 
  which is $=_1$ or $=_2$, to $\hat{g}'$ such that its arity is $3$ or $4$.
  It is possible that $n=3$ or $4$ to begin with, 
  and if so we do nothing.
  We are done by yet another diagonal transformation and Corollary~\ref{cor:Z:Eq34:PMk}.
\end{proof}

\begin{lemma}  \label{lem:P2:M4:tractability}
  Let $\mathcal{F}$ be a set of symmetric signatures.
  Suppose $\mathcal{F}\subseteq Z\mathscr{P}\cup\mathscr{M}_4^\sigma$ for some $\sigma\in\{+,-\}$
  and the greatest common divisor of the arities of all signatures in $\mathcal{F}^*\cap\mathscr{P}_2$ is $k\ge 5$.
  Then $\PlHolant(\mathcal{F})$ can be computed in polynomial time.
\end{lemma}
\begin{proof}
  We may assume that $\sigma=+$ and the case of $\sigma=-$ is similar.
  We do a $Z$ transformation on $\PlHolant(\mathcal{F})$, 
  and get a problem of \plholant{{\neq}_2}{Z^{-1}\mathcal{F}}.

  In this bipartite setting, given $=_{n}$ on the right hand side,
  we can realize $=_{\ell n}$ for any integer $\ell\ge 1$
  as an \eblock{n} on the right.
  The problem \plholant{{\neq}_2}{\mathcal{EQ}_n,\mathcal{EO},\neq_2,[1,0],[0,1]} is tractable
  for any $n\ge 5$
by Lemma~\ref{lem:Z:Eqd-PM3} and Lemma~\ref{lem:Z:Eq5-PM3},
  where $\mathcal{EQ}_n$ denotes the set of all equalities of arity $\ell n$ for all integers $\ell\ge 1$.

  The symmetric signatures in the
 set $Z\mathscr{P}$ consist of $\mathscr{P}_2$, $Z^{\otimes 2}(\neq_2)$, 
  and degenerate signatures.
  If there is any degenerate signature of the form $(Z[a,b])^{\otimes m}\in\mathcal{F}$ with $ab\neq 0$,
  then $Z[a,b]\in\mathcal{F}^*\cap\mathscr{P}_2$.
  This contradicts  $k\ge 5$.
  Hence all degenerate signatures in $\mathcal{F}$ are of the form $(Z[1,0])^{\otimes m}$ or $(Z[0,1])^{\otimes m}$, if any.
Since $\mathcal{F}\subseteq Z\mathscr{P}\cup\mathscr{M}_4^+$,
  after a $Z$ transformation,
  $\PlHolant(\mathcal{F})$ is an instance of \plholant{{\neq}_2}{\mathcal{EQ}_k,\mathcal{EO},\neq_2,[1,0],[0,1]}
  except for the weights on the equalities.
  It can be checked that the tractability results of 
  Lemma~\ref{lem:Z:Eqd-PM3} and Lemma~\ref{lem:Z:Eq5-PM3} also apply to weighted equalities.
  The lemma follows.
\end{proof}

\vspace{.2in}

Let $\mathcal{G} = \{{=}_k \mid k \in S\}$ be a set of \textsc{Equality} signatures,
where $S$ is a set of positive integers containing at least one $r \ge 3$. 
Moreover let $\mathcal{EO}^+:=\{\exactone{d}\mid d\in \mathbb{Z}^+\}=\mathcal{EO}\cup\{\neq_2,[0,1]\}$.
Then $\plholant{\mathcal{G}}{\mathcal{EO}^+}$ is the problem of 
counting perfect matchings over hypergraphs with planar incidence graphs,
where the hyperedge sizes are prescribed by $S$.
In the incidence graph,
vertices assigned signatures in $\mathcal{G}$ on the left represent hyperedges, 
and vertices assigned signatures in $\mathcal{EO}^+$ on the right represent vertices of the hypergraph.
Let $t = \gcd(S)$.
It is stated in the introduction that
this problem is tractable if $t \ge 5$ and \numP-hard if $t \le 4$.
The tractability when $t \ge 5$ follows from Lemma~\ref{lem:Z:Eqd-PM3} and \ref{lem:Z:Eq5-PM3},
since we can reduce $\plholant{\mathcal{G}}{\mathcal{EO}^+}$ 
to $\plholant{{\neq}_2}{{=}_t,\mathcal{EO},\neq_2,[0,1]}$.
The reduction goes as follows.
With $\neq_2$ on the left hand side and $=_t$ on the right hand side,
we can construct all \eblock{t}s and hence all of $\EQ_t$ on the right.
Note that $\mathcal{G} \subseteq \EQ_t$.
Then we move all signatures in $\mathcal{G}$ to the left via $\neq_2$.

The hardness of $\plholant{\mathcal{G}}{\mathcal{EO}^+}$ 
for $t\le 4$ follows from Corollary~\ref{cor:Z:Eq34:PMk}.
The reason is as follows.
We construct $\neq_2$ on the left using the gadget pictured in Figure~\ref{subfig:gadget:van-binary}
with $(=_r) \in \mathcal{G}$ on the left side assigned to circle vertices and
$\neq_2$ on the right side assigned to square vertices.
Then we move $\mathcal{G}$ to the right side via $\neq_2$ on the right side. 
We construct ${=}_t$ on the right side in the same Euclidean process using 
$\mathcal{G}$ of the right side and $\neq_2$ of the left side. 
This gives us a reduction from $\plholant{{\neq}_2}{{=}_t,\mathcal{EO}}$,
which is \numP-hard by Corollary~\ref{cor:Z:Eq34:PMk} if $t=3,4$.
Otherwise $t=1,2$.
Recall that $(=_r) \in\mathcal{G}$ for some $r\ge 3$.
We use $=_t$ to reduce the arity of $=_r$ to $3$ or $4$, if necessary.
Again we are done by Corollary~\ref{cor:Z:Eq34:PMk}.

If we do not assume there is at least one hyperedge of size $\ge 3$
in $\plholant{\mathcal{G}}{\mathcal{EO}^+}$, and $t = \gcd(S) \le 2$, 
then the problem is tractable if and only if $S\subseteq\{1,2\}$.
The tractability is due to Kasteleyn's algorithm, 
as there is no hyperedge.
In summary, we have the following theorem.

\begin{theorem}\label{thm:hypergraph}
  The problem $\plholant{\mathcal{G}}{\mathcal{EO}^+}$ 
  counts perfect matchings over hypergraphs with planar incidence graphs,
  where the hyperedge sizes are prescribed by a set $S$ of positive integers.
  Let $t = \gcd(S)$.
  If $t \ge 5$ or $S\subseteq\{1,2\}$,
  then the problem is computable in polynomial time.
  Otherwise $t \le 4$, $S\not\subseteq\{1,2\}$,
  and the problem is \numP-hard.
\end{theorem}

%\begin{remark}
%If we do not assume there is at least one hyperedge of size $\ge 3$
%in $\plholant{\mathcal{G}}{\mathcal{EO}^+}$, 
%and $t = \gcd(S) \le 2$, then there is the following additional classification.
%
%Suppose $t=2$. If $S = \{2\}$, then this is exactly Kasteleyn's algorithm
%for planar perfect matchings for ordinary (i.e., not hyper-) graphs.
%If $S \not = \{2\}$, then there is some $(=_r) \in \mathcal{G}$ where
%$r > 2$. Connecting $=_r$ to $\exactone{r+2}$ gives us $\neq_2$ on the
%right, and then the same Euclidean process gives us $(=_d) = (=_2)$ on the left,
%and by $=_r$, we get all  $\EQ_2$. Thus this is equivalent to 
%$\PlCSP^2(\mathcal{EO})$ which is \numP-hard by Theorem~\ref{thm:PlCSP2}
%(to be proved in Part II).
%
%Suppose $d=1$. If $S \subseteq \{1, 2\}$, then this is also solvable
%by Kasteleyn's algorithm. 
%%%% each unary =_1 simple gets absorbed to form ExOne_{k-1}.
%If $S \not \subseteq \{1, 2\}$, then by the same Euclidean process we get
%$\PlCSP(\mathcal{EO})$ which is \numP-hard by Theorem~\ref{thm:PlCSP}.
%\end{remark}

\section{Full Dichotomy} \label{sec:full:dic}

We are finally ready to prove our main dichotomy theorem.
Recall that for a set $\mathcal{F}$ of signatures,
$\nondeg{F}$ denotes the set of non-degenerate signatures in $\mathcal{F}$ of arity at least~$3$,
and $\mathcal{F}^*$ denotes $\mathcal{F}$ with all degenerate signatures $[a,b]^{\otimes m}$ 
replaced by unary $[a,b]$.

\begin{theorem} \label{thm:main}
 Let $\mathcal{F}$ be any set of symmetric, complex-valued signatures in Boolean variables.
 Then $\PlHolant(\mathcal{F})$ is \numP-hard unless $\mathcal{F}$ satisfies one of the following conditions:
 \begin{enumerate}
  \item All non-degenerate signatures in $\mathcal{F}$ are of arity at most~2; \label{case:main_tractable:trivial}
  \item $\mathcal{F}$ is $\mathscr{A}$-transformable; \label{case:main_tractable:CSP:A}
  \item $\mathcal{F}$ is $\mathscr{P}$-transformable; \label{case:main_tractable:CSP:P}
  \item $\mathcal{F} \subseteq \mathscr{V}^\sigma \union \{f \in \mathscr{R}_2^\sigma \st \arity(f) = 2\}$ 
    for some $\sigma \in \{+,-\}$; \label{case:main_tractable:vanishing_and_binary}
  \item All non-degenerate signatures in $\mathcal{F}$ are 
    in $\mathscr{R}_2^\sigma$ for some $\sigma \in \{+,-\}$. \label{case:main_tractable:vanishing_and_unary}
  \item $\mathcal{F}$ is $\mathscr{M}$-transformable; \label{case:main_tractable:CSP:M}       
  \item $\mathcal{F}\subseteq Z\mathscr{P}\cup\mathscr{M}_4^\sigma$ for some $\sigma\in\{+,-\}$, and
    the greatest common divisor of the arities of the signatures in $\mathcal{F}^*\cap\mathscr{P}_2$ is at least~$5$.\label{case:Eq:PM}
 \end{enumerate}
 In each exceptional case,
 $\PlHolant(\mathcal{F})$ is computable in polynomial time.
 If $\mathcal{F}$ satisfies condition \ref{case:main_tractable:trivial},
 \ref{case:main_tractable:CSP:A}, \ref{case:main_tractable:CSP:P},
 \ref{case:main_tractable:vanishing_and_binary}, or \ref{case:main_tractable:vanishing_and_unary},
 then $\Holant(\mathcal{F})$ is computable in polynomial time without planarity;
 otherwise $\Holant(\mathcal{F})$ is \numP-hard.
\end{theorem}

\begin{proof}
We may assume that $\mathcal{F}$ contains no identically $0$ signatures.
We note that removing any identically $0$ signature from a set does not affect its
complexity, being either tractable or \numP-hard,
and it does not affect the set $\mathcal{F}$ satisfying any of the
listed conditions in Case~\ref{case:main_tractable:trivial} to \ref{case:Eq:PM}.

  If all non-degenerate signatures in $\mathcal{F}$ are of arity at most~$2$,
  then the problem is tractable case~\ref{case:main_tractable:trivial}.
  Otherwise, there is a non-degenerate signature $f \in \mathcal{F}$ of arity at least~$3$.
  By Theorem~\ref{thm:dic:single},
  $\PlHolant(\mathcal{F})$ is \numP-hard unless
  $f \in \mathscr{P}_1 \cup \mathscr{M}_2 \cup \mathscr{A}_3 \cup \mathscr{M}_3 \cup \mathscr{M}_4$
  or $f$ is vanishing.
  If $f \in \mathscr{P}_1$ or $f\in \mathscr{M}_2 \setminus \mathscr{P}_2$ or $f\in \mathscr{A}_3$ or $f\in \mathscr{M}_3$,
  then we are done by Corollary~\ref{cor:dichotomy:P1} or
  Lemma~\ref{lem:dichotomy:M2} or Corollary~\ref{cor:dichotomy:A3} or
  Lemma~\ref{lem:dichotomy:M3} respectively.
  Therefore, we assume that none of these is the case.
  This implies that $\nondeg{F}$ is nonempty and that each of its signatures 
is in $\mathscr{P}_2$ or in $\mathscr{M}_4$ or vanishing.
  That is,
  \[
    \emptyset \neq \nondeg{F} \subseteq \mathscr{P}_2 \cup \mathscr{M}_4 \cup \mathscr{V}.
  \]
  
  Suppose there exists some $f \in \nondeg{F}$
  which is in $\mathscr{V} \setminus \mathscr{M}_4$.
  We assume $f \in \mathscr{V}^+$ since the other case $\mathscr{V}^-$ is similar.
  In this case, we show that $\PlHolant(\mathcal{F})$ is \numP-hard,
  unless $\mathcal{F}$ is in Case~\ref{case:main_tractable:vanishing_and_binary}
  or Case \ref{case:main_tractable:vanishing_and_unary}.
  Assume that $\PlHolant(\mathcal{F})$ is not \numP-hard.
  We will discuss non-degenerate signatures of arity $\ge 3$, of arity $2$,
  and degenerate signatures separately.
  \begin{enumerate}
    \item For any $g \in \nondeg{F}$,
      we claim that $g \in \mathscr{V}^+$.
      Suppose otherwise, then $g\in\mathscr{P}_2$ or $g\in\mathscr{V}^-$.
      Notice that the latter covers the case where $g\in\mathscr{M}_4$ but $g\not\in\mathscr{V}^+$
      (namely $g\in\mathscr{M}_4^-$).
      If $g\in\mathscr{P}_2$, then $\PlHolant(f,g)$ is \numP-hard by Lemma~\ref{lem:vanishing:P2}.
      If $g\in\mathscr{V}^-$, then $\PlHolant(f,g)$ is \numP-hard by Lemma~\ref{lem:van:plus_and_minus} as $f\not\in\mathscr{M}_4$.
    \item For any non-degenerate binary signature $h\in\mathcal{F}$,
      it must be that $h\in \mathscr{R}_2^+$ as otherwise
      $\PlHolant(f,h)$ is \numP-hard by Lemma~\ref{lem:van:bin}.
    \item 
%If there is no degenerate signature in $\mathcal{F}$,
%      then by the previous two items,
% it is Case~\ref{case:main_tractable:vanishing_and_binary} and we are done.
      If $\rd^+(g) = 1$ for all $g\in\nondeg{F}$,
      then $\nondeg{F}\subseteq\mathscr{R}_2^+$ by  
Lemma~\ref{lem:prelim:vanishing_form_in_Z_basis}.
%Definition \ref{def:Rt}.
      Together with the fact just proved that all non-degenerate binary in $\mathcal{F}$ are in $\mathscr{R}_2^+$,
      Case~\ref{case:main_tractable:vanishing_and_unary} is satisfied.
      
      Otherwise there exists $g\in\nondeg{F}$ such that $\rd^+(g)\ge 2$.
Then $g \in \mathscr{V}^+$
by the first item above.
      If $\mathcal{F}$ contains any degenerate signature $v= u^{\otimes m}$
      for $m \ge 1$ and some unary $u$ that is not a multiple of $[1,i]$,
      then by Lemma~\ref{lem:van:deg}, $\PlHolant(g,v)$ is \numP-hard.
      Hence all degenerate signatures are multiples of tensor powers of $[1,i]$,
      which are in $\mathscr{V}^+$.
      It implies that $\mathcal{F}$ is in Case~\ref{case:main_tractable:vanishing_and_binary}.
  \end{enumerate}
  Now we have that $\emptyset\neq\nondeg{F} \subseteq \mathscr{P}_2 \cup \mathscr{M}_4$.
  We handle this in three cases.
  \begin{enumerate}
    \item Suppose $\nondeg{F} \subseteq \mathscr{M}_4$.
      First suppose $\nondeg{F}\subseteq\mathscr{M}_4^\sigma$ for some $\sigma\in\{+,-\}$.
      Assume $\sigma=+$ as $\sigma=-$ is similar.
      Then $\nondeg{F}\subseteq\mathscr{R}_2^+$
by Lemma~\ref{lem:trans:M4} and \ref{lem:prelim:vanishing_form_in_Z_basis}.
      If all non-degenerate binary signatures are in $\mathscr{R}_2^+$ as well,
      then this is Case \ref{case:main_tractable:vanishing_and_unary} and tractable.
      Let $h$ be a non-degenerate binary signature in $\mathcal{F}$ that is not in $\mathscr{R}_2^+$.
      We apply Lemma~\ref{lem:van:M4:bin},
      and $\PlHolant(\mathcal{F})$ is \numP-hard unless $h=Z^{\otimes 2}[a,0,1]$ up to a nonzero factor, where $a\neq 0$.
      In this case we apply a $Z$ transformation, and get \plholant{\neq_2}{[a,0,1],Z^{-1}\mathcal{F}}.
      Then we do a diagonal transformation $D=\trans{a^{1/2}}{0}{0}{1}$.
%%% JYC: i editd D to inverse, as that is the canonical way to say apply "D".
      Note that this only changes $\neq_2$ on the left hand side to
      a nonzero multiple of $\neq_2$.
      Hence we have the reduction chain:
      \begin{align*}
        \PlHolant(\mathcal{F})
        &\equiv \plholant{{\neq}_2}{[a,0,1],Z^{-1}\mathcal{F}}\\
        &\equiv \plholant{{\neq}_2}{[1,0,1],D^{-1}Z^{-1}\mathcal{F}}\\
        &\ge_T \PlHolant(D^{-1}Z^{-1}\mathcal{F})
      \end{align*}
      Notice that $D^{-1}Z^{-1}\mathcal{F}$ contains $\exactone{k}$ with $k\ge 3$
      that is in $\mathscr{M}_3$ with $I_2$.
      Then by Lemma~\ref{lem:dichotomy:M3},
      $\PlHolant(\mathcal{F})$ is \numP-hard unless
      $D^{-1}Z^{-1}\mathcal{F}\subseteq I_2\mathscr{M}=\mathscr{M}$,
i.e., $\mathcal{F}\subseteq Z D \mathscr{M} = Z \mathscr{M}$.
      The exceptional case implies that $\mathcal{F}$ is $\mathscr{M}$-transformable via $Z$, and we are in the tractable
 Case~\ref{case:main_tractable:CSP:M}.
 
      Otherwise $\nondeg{F}$ contains both $f\in\mathscr{M}_4^+$ and $g\in\mathscr{M}_4^-$.
      Similarly as above, by Lemma~\ref{lem:van:M4:bin},
      any non-degenerate binary signature in $\mathcal{F}$ has to be 
      in $\mathscr{R}_2^+ \cap\mathscr{R}_2^-=\{Z^{\otimes 2}(\neq_2)\}$ 
      (cf.\ Lemma \ref{lem:prelim:vanishing_form_in_Z_basis}),
      or is a nonzero constant multiple of $Z^{\otimes 2}[a,0,1]$ 
where $a\neq 0$,
      as otherwise $\PlHolant(\mathcal{F})$ is \numP-hard.
      Moreover, by Lemma \ref{lem:M4+:M4-:deg},
      $\PlHolant(\mathcal{F})$ is \numP-hard, unless
      all degenerate signatures in $\mathcal{F}$ are of the form
      $[1,\pm i]^{\otimes m}$. Note that $[1,i]=Z[1,0]$ and $[1,-i]=Z[0,1]$.
      When this is the case, $\mathcal{F}$ is $\mathscr{M}$-transformable via $Z$.
    \item Suppose $\nondeg{F} \subseteq \mathscr{P}_2$.
      If $\mathcal{F}$ contains a non-degenerate binary signature $h$,
      then we apply Lemma~\ref{lem:P2:binary} and
      $\PlHolant(\mathcal{F})$ is \numP-hard unless $h\in Z\mathscr{P}$,
      or $\PlCSP^2(DZ^{-1}\mathcal{F})\le_T\PlHolant(\mathcal{F})$ for some diagonal transformation $D$.
      If it is the latter case,
      then by Theorem~\ref{thm:PlCSP2},
      either $\PlHolant(\mathcal{F})$ is \numP-hard,
      or $DZ^{-1}\mathcal{F}$ is a subset of $T\mathscr{A}$, $\mathscr{P}$, or $T\trans{1}{1}{1}{-1}\mathscr{M}$,
      for some diagonal matrix $T$.
      We claim that in any of these cases $\PlHolant(\mathcal{F})$ is tractable.
      In fact,
      \begin{enumerate}
        \item if $DZ^{-1}\mathcal{F}\subseteq T\mathscr{A}$,
          then $\mathcal{F}$ is $\mathscr{A}$-transformable as 
          $\mathcal{F}\subseteq Z D^{-1} T\mathscr{A}$
          and $[1,0,1]$ (as a row vector) is transformed into
$[1,0,1] (Z D^{-1} T)^{\otimes 2}$, which is $[0,1,0]\in\mathscr{A}$ 
          up to a nonzero constant;
        \item if $DZ^{-1}\mathcal{F}\subseteq \mathscr{P}$,
          then $\mathcal{F}$ is $\mathscr{P}$-transformable as
          $\mathcal{F}\subseteq Z D^{-1}\mathscr{P}$
          and $[1,0,1](ZD^{-1})^{\otimes 2}$ is $[0,1,0]\in\mathscr{P}$ up to a nonzero constant;          
        \item if $DZ^{-1}\mathcal{F}\subseteq T\trans{1}{1}{1}{-1}\mathscr{M}$,
          then $\mathcal{F}$ is $\mathscr{M}$-transformable as 
$\mathcal{F}\subseteq Z D^{-1} T \trans{1}{1}{1}{-1} \mathscr{M}$
          and $[1,0,1]$  is transformed
to $[1,0,1](Z D^{-1} T \trans{1}{1}{1}{-1})^{\otimes 2}$,
which is $[1,0,-1]\in\mathscr{M}$ 
          up to a nonzero constant.
      \end{enumerate}
      Hence we may assume that every non-degenerate binary in $\mathcal{F}$ is in $Z\mathscr{P}$.
      Notice that degenerate signatures are always in $\mathscr{P}$ under any transformation.
      Also $\nondeg{F}$ is a subset of $Z\mathscr{P}$
      because $\nondeg{F} \subseteq \mathscr{P}_2$
      and $\mathscr{P}_2$ is just weighted equalities under $Z$-transformation.
      It implies that $\mathcal{F}$ is $\mathscr{P}$-transformable
      under the $Z$ transformation. Hence we are in Case~\ref{case:main_tractable:CSP:P}.
    \item Finally, suppose neither of the above is the case.
      Then there are $f, g \in \nondeg{F}$ with $f \in \mathscr{M}_4$ and $g \in \mathscr{P}_2$.
      If $\nondeg{F}$ contains both $f\in\mathscr{M}_4^+$ and $f'\in\mathscr{M}_4^-$,
      then $\PlHolant(\mathcal{F})$ is \numP-hard by Lemma \ref{lem:P2:M4+-}.
      Otherwise $\nondeg{F}\cap\mathscr{M}_4\subseteq\mathscr{M}_4^+$ or $\mathscr{M}_4^-$.
      Let $\mathcal{G} = \mathcal{F}^* \cap \mathscr{P}_2$,
      and let $d$ be the gcd of the arities of the signatures in $\mathcal{G}$.
      Then $\mathcal{G}$ contains at least one non-degenerate signature $g$ of arity $\ge 3$.
      If $d\le 4$,
      then $\PlHolant(\mathcal{F})$ is \numP-hard by Lemma~\ref{lem:P2:M4:hardness}.
      Otherwise $d \ge 5$.
      If $\mathcal{F}$ contains a non-degenerate binary signature $h$,
      then we apply Lemma~\ref{lem:P2:binary} and 
      by a similar analysis as in the case of ``$\nondeg{F}\subseteq\mathscr{P}_2$'' above,
      we are done unless every such $h$ is in $Z \mathscr{P}$.
      Ignoring a nonzero factor, it implies that either $h=Z^{\otimes 2}[1,0,a]$ where $a\neq 0$ or $h=Z^{\otimes 2}(\neq_2)$.
      If $h=Z^{\otimes 2}[1,0,a]$, then $h\in \mathcal{F}^* \cap \mathscr{P}_2$,
      and it contradicts $d\ge 5$.
      Hence $h=Z^{\otimes 2}(\neq_2)$.
      If there is any degenerate $v=\left(Z[a,b]\right)^{\otimes m}$ in $\mathcal{F}$ with $ab\neq 0$,
      then $Z[a,b]\in \mathcal{F}^* \cap \mathscr{P}_2$ and it also contradicts $d\ge 5$.
      
      In summary, $\PlHolant(\mathcal{F})$ is \numP-hard
      unless $\nondeg{F}\subseteq\mathscr{P}_2\cup\mathscr{M}_4$,
      $\nondeg{F}\cap\mathscr{M}_4\subseteq\mathscr{M}_4^\sigma$ for some $\sigma\in\{+,-\}$,
      the greatest common divisor of the arities of the signatures in $\mathcal{F}^*\cap\mathscr{P}_2$ is at least~$5$.
      Every non-degenerate binary in $\mathcal{F}$ is of the form $Z^{\otimes 2}(\neq_2)$,
      and every degenerate in $\mathcal{F}$ is of the form $(Z[1,0])^{\otimes m}$ or $(Z[0,1])^{\otimes m}$.
      Notice that $\mathscr{P}_2$, $Z^{\otimes 2}(\neq_2)$, $(Z[1,0])^{\otimes m}$, and $(Z[0,1])^{\otimes m}$
      are all in $Z\mathscr{P}$.
      Hence the exceptional case implies that $\mathcal{F}\subseteq Z\mathscr{P}\cup\mathscr{M}_4^\sigma$ for some $\sigma\in\{+,-\}$
      and the greatest common divisor of the arities of the signatures in $\mathcal{F}^*\cap\mathscr{P}_2$ is at least~$5$.
      This is tractable Case~\ref{case:Eq:PM}.
  \end{enumerate}
  The tractability of $\Holant(\mathcal{F})$ 
in Case \ref{case:main_tractable:trivial}, 
  Case \ref{case:main_tractable:CSP:A}, Case \ref{case:main_tractable:CSP:P}, 
  Case \ref{case:main_tractable:vanishing_and_binary}, 
  and Case \ref{case:main_tractable:vanishing_and_unary} follows from the Holant dichotomy Theorem \ref{thm:Holant:set},  which also implies that
$\Holant(\mathcal{F})$ is \numP-hard otherwise.
  The tractability of $\PlHolant(\mathcal{F})$
in 
Case~\ref{case:main_tractable:CSP:M} follows from Theorem \ref{thm:tractable:M}. 
  The tractability of $\PlHolant(\mathcal{F})$ 
in Case~\ref{case:Eq:PM} follows from Lemma~\ref{lem:P2:M4:tractability}. 
  This completes the proof.
\end{proof}

\appendix

\title{A Holant Dichotomy: Is the FKT Algorithm Universal?\\
Part II: Planar \#CSP$^2$ Dichotomy}
\author{\emptythanks}
\maketitle

In Part~II of this paper,
we prove Theorem~\ref{dichotomy-pl-csp2},
which is the complexity dichotomy theorem of $\PlCSP^2(\mathcal{F})$,
where $\mathcal{F}$ is a set of complex-valued symmetric signatures on Boolean variables.
After we define some relevant notions,
we give an outline of the proof of Theorem~\ref{dichotomy-pl-csp2}.
Throughout Part~II,
we denote by $\alpha$
(respectively $\rho$)
any quantity that satisfies $\alpha^4 = -1$
(respectively $\rho^4 = 1$).
 
\section{Preliminaries} \label{partII:sec:prelim}

We will first define some tractable families of signatures that are expressible under a holographic transformation,
specific to the $\PlCSP^2$ framework.
\begin{definition}
Let
$\mathcal{T}_k =
\left\{
\trans{1}{0}{0}{\omega}
%\mid \omega = e^{\pi i j /k},  j = 1, 3, \ldots, 2k-1
\mid \omega^k = 1
\right\}$ be a set of diagonal matrices of order dividing $k$
and $\mathscr{T}_k = \mathcal{T}_{2k} \setminus \mathcal{T}_k
= \left\{
\trans{1}{0}{0}{\omega}
%\mid \omega = e^{\pi i j /k},  j = 1, 3, \ldots, 2k-1
\mid \omega^k = -1
\right\}$.
Let $\mathscr{A}^\dagger = \mathscr{T}_4 \mathscr{A}$
and
$\widehat{\mathscr{M}}^\dagger = \mathscr{T}_2 \widehat{\mathscr{M}}$
be the sets of signatures transformed by $\mathscr{T}_4$ from the
Affine family $\mathscr{A}$ and transformed by $\mathscr{T}_2$ from $\widehat{\mathscr{M}}$,
%the Matchgate family under the Hadamard basis,
respectively,
where for a class of signatures $\mathscr{C}$,
we denote
\[
 \mathscr{T}_k \mathscr{C}
 =
 \{
  T^{\otimes \arity(f)}f \mid T \in \mathscr{T}_k \text{ and } f \in \mathscr{C}
 \}.
\]
Let
\[
 \widetilde{\mathscr{A}} = \mathscr{A} \cup \mathscr{A}^\dagger
 \qquad
 \text{and}
 \qquad
 \widetilde{\mathscr{M}} = \widehat{\mathscr{M}} \cup \widehat{\mathscr{M}}^\dagger
\]
be the $\mathscr{A}$-transformable and ${\mathscr{M}}$-transformable
signatures for $\PlCSP^2$.
\end{definition}

Recall that
$\widehat{\mathscr{M}} = H \mathscr{M}$ is the set of Matchgate signatures
 $\mathscr{M}$ transformed by the Hadamard basis
$H=
\left[\begin{smallmatrix} 1 & 1 \\
1&-1 \end{smallmatrix}\right]$.
Note that $\mathscr{A}$ is unchanged under the transformation
by $H$,
%=
%\left[\begin{smallmatrix} 1 & 1 \\
%1&-1 \end{smallmatrix}\right]$,
and thus there is no  need to define
$\widehat{\mathscr{A}}$.
Also note that $\mathscr{P}$ is unchanged under
any diagonal matrix.
Thus there is no need to define $\mathscr{P}^\dagger$.
For $T = \left[\begin{smallmatrix} 1 & 0 \\
0 & \omega \end{smallmatrix}\right] \in \mathcal{T}_4$ with $\omega^4 = 1$,
$T \mathscr{A} =  \mathscr{A}$. Thus
$\widetilde{\mathscr{A}}$ is $\mathscr{A}$
under transformations by
$T  
= \left[\begin{smallmatrix} 1 & 0 \\
0 & \omega \end{smallmatrix}\right]
%\mid \omega^8 = 1
%\right\} 
\in \mathcal{T}_8$. For such $T$, we have
$(=_{2n}) T^{\otimes 2n} \in \mathscr{A}$.
Hence
$\widetilde{\mathscr{A}}$ is
$\mathscr{A}$-transformable for $\PlCSP^2$.
Similarly, for $T = \left[\begin{smallmatrix} 1 & 0 \\
0 & \pm 1 \end{smallmatrix}\right]$,
% with $(\pm 1)^2 = 1$,
$TH =
\left[\begin{smallmatrix} 1 & 1\\
\pm 1 & \mp 1 \end{smallmatrix}\right]
=$ either $H$ or $H
\left[\begin{smallmatrix} 0 & 1\\
1 & 0 \end{smallmatrix}\right]$,
and
$\left[\begin{smallmatrix} 0 & 1\\
1 & 0 \end{smallmatrix}\right] \mathscr{M} = \mathscr{M}$.
Thus $T \widehat{\mathscr{M}} = \widehat{\mathscr{M}}$, and
$\widetilde{\mathscr{M}}$ is $\mathscr{M}$ transformed under
$TH$
for all $T  \in
\mathcal{T}_4$.
% \left\{
%\left[\begin{smallmatrix} 1 & 0 \\
%0 & \omega \end{smallmatrix}\right]
%\mid \omega^4 = 1
%\right\}$.
Also note that for all such $T$,
we have $(=_{2n}) (TH)^{\otimes 2n} \in \mathscr{M}$.
%%% =_{2n} [1 0 \\ 0 \pm i]
%%% =   [1,0]^{2n} + (-1)^n [0, 1]^{2n}
%%% then by H, gets to M.
Hence
$\widetilde{\mathscr{M}}$
is ${\mathscr{M}}$-transformable
for $\PlCSP^2$.

In the proof of No-Mixing of different tractable sets,
because of a particular order in which we carry out the proof,
to make an overall logical structure more apparent
 we introduce the following notations
 %for convenience we denote 
\[
 S_1 = \widehat{\mathscr{M}},
 \qquad
 S_2 = \widehat{\mathscr{M}}^\dagger,
 \qquad
 S_3 = \mathscr{A}^\dagger,
 \qquad
 S_4 = \mathscr{A},
 \qquad \text{and} \qquad
 S_5 = \mathscr{P}.
\]

We will prove the following Main Theorem of Part~II.
It is not hard to see that this is a rephrase of Theorem \ref{thm:PlCSP2}
from Part~I.
It follows from Theorem \ref{odd-arity-dichotomy}, Theorem~\ref{general-single-dichotomy} and
Theorem \ref{mixing-theorem},
which will be shown in later sections.
It follows from the definition of $\mathscr{P}$-transformability,
$\mathscr{A}$-transformability and $\mathscr{M}$-transformability
that if $\mathcal{F}\subseteq S_k$ for any $1 \le k \le 5$, then
$\PlCSP^2(\mathcal{F})$ is tractable.

\begin{theorem} \label{dichotomy-pl-csp2}
For any set of complex-valued symmetric signatures $\mathcal{F}$
on Boolean variables,
if $\mathcal{F}\subseteq \mathscr{P}$, or
$\mathscr{A}$, or $\mathscr{A}^\dagger$,
or $\widehat{\mathscr{M}}$, or $\widehat{\mathscr{M}}^\dagger$,
then $\PlCSP^2(\mathcal{F})$ is tractable.
Otherwise, $\PlCSP^2(\mathcal{F})$ is \numP-hard.
\end{theorem}

%%%%%%%%%%%%%%%%%%%%%%%%%
\begin{proof}[Proof Outline]
We now give an outline of the proof of Theorem~\ref{dichotomy-pl-csp2}.
The overall plan is to break the proof into two main steps.

The first step is to prove the dichotomy theorem for $\PlCSP^2(\mathcal{F})$
when there is at least one nonzero signature of \emph{odd} arity in $\mathcal{F}$.
In this case we can make use of Lemma~\ref{[1,a]XXX}
that shows that we can simulate
$\PlCSP(\mathcal{F})$ by $\PlCSP^2(\mathcal{F})$
if $\mathcal{F}$ includes a unary signature $[a, b]$ with $ab\neq 0$.
Then we can apply the known dichotomy Theorem~\ref{pl-dicho-1} for $\PlCSP^1$.
However this strategy (provably) \emph{cannot} work
in the case when every signature in $\mathcal{F}$  satisfies
 the \emph{parity}
 constraint.
In that case we employ other means.
This first step of the proof is relatively uncomplicated.

The second step is to deal with the case when
all nonzero signatures in $\mathcal{F}$ have
even arity. This is where the real difficulties lie. In this case it is
impossible to directly construct \emph{any} unary signature.
So we cannot use  Lemma~\ref{[1,a]XXX} in this case.
But Lemma~\ref{mixing-P-global} provides a way to simulate
$\PlCSP(\mathcal{F})$ by $\PlCSP^2(\mathcal{F})$
in a \emph{global} fashion,
\emph{if} $\mathcal{F}$ includes some tensor power
 of the form $[a, b]^{\otimes 2}$ where  $ab\neq 0$.
Moreover, we have a lucky break (for the complexity of the proof) if
$\mathcal{F}$ includes
a signature that
is in $\widehat{\mathscr{M}}\setminus(\mathscr{P}\cup\widetilde{\mathscr{A}})$.
In this case, we can construct a special binary signature,
%%% [1,a,1], with a^4 \not = 0, 1
and then use Lemma~\ref{[1,a,1]-interpolation}
to obtain $[1, 1]^{\otimes 2}$ by interpolation.
This proof uses the theory of \emph{cyclotomic fields}.
This simplifies the proof greatly.
For all other cases (when $\mathcal{F}$ has only
even arity signatures), the proof gets going in earnest---we
 will attempt an induction on the arity of signatures.

The lowest arity of this induction
will be 2. We will try to reduce the arity
to 2 whenever possible; however for many cases an arity reduction
to 2
destroys the \#P-hardness at hand. Therefore the true basis
of this induction proof of Pl-\#{\rm CSP}$^2$ starts with
arity 4. Consequently we will first prove a dichotomy theorem
for $\PlCSP^2(f)$, where $f$ is a signature of arity 4.
This proof is presented in Section~\ref{Arity-4}.
Several
tools will be used. These include the rank criterion for redundant
signatures, Theorem~\ref{thm:parII:k-reg_homomorphism} for arity 2 signatures,
and a trick we call the \emph{Three Stooges} by domain pairing.

However in the next step we do not attempt a general
Pl-\#{\rm CSP}$^2$ dichotomy for a {\it single} signature of
even arity. This  would have been natural at this point, but
it would have been too difficult.
We will need some additional leverage by proving a conditional
No-Mixing Lemma for pairs of signatures of even arity.
So, seemingly taking a detour,
 we prove that for two signatures $f$ and $g$ both of even arity,
that individually belong to some tractable class, but  do not
belong to a single tractable class in the conjectured Pl-\#{\rm CSP}$^2$
 dichotomy
(that is yet to be proved),
the problem  $\PlCSP^2(f, g)$ is \#P-hard.
We prove this No-Mixing Lemma for any pair of
signatures $f$ and $g$ both of even arity, not restricted to arity 4.
Even though at this point we only have a dichotomy for a
single signature of arity 4, we prove this No-Mixing Lemma
for higher even arity pairs $f$ and $g$ by simulating
two signatures $f'$ and $g'$
of arity 4  that belong to different tractable sets,
from that of $\PlCSP^2(f, g)$.
After this arity reduction (within the No-Mixing Lemma),
we prove that $\PlCSP^2(f', g')$
is \#P-hard by the dichotomy for a {\it single} signature of arity 4.
After this, we prove a No-Mixing Lemma for a \emph{set} of
signatures $\mathcal{F}$ of even arities, which states
that \emph{if}  $\mathcal{F}$
is contained in the union of all tractable classes, then
it is still \#P-hard unless it is \emph{entirely} contained in a
single tractable class.  Note that at this point we still only have
a \emph{conditional} No-Mixing Lemma in the sense that we have to assume
every signature in $\mathcal{F}$ belongs to some tractable set.

We then attempt the proof of a
Pl-\#{\rm CSP}$^2$ dichotomy for a {\it single} signature of
arbitrary even arity. This uses all the previous lemmas, in
particular the (conditional) No-Mixing Lemma for a set
of signatures. However, after completing the proof of
this Pl-\#{\rm CSP}$^2$ dichotomy for a single signature
of even arity, the No-Mixing Lemma becomes absolute.

Finally the dichotomy for
a single signature
of even arity is logically extended
 to a dichotomy theorem for $\PlCSP^2(\mathcal{F})$
where all signatures in
$\mathcal{F}$ have even arity.
Together with the first main step when $\mathcal{F}$
contains some nonzero signature of odd arity,
this completes the proof of Theorem~\ref{dichotomy-pl-csp2}.
\end{proof}
%%%%%%%%%%%%%%%%%%%%%%%
\vspace{.2in}

In the rest of this Section~\ref{partII:sec:prelim}, we will introduce
the operators $\partial$ and $\int$, and give some characterizations
of the tractable classes. We will also introduce some preliminary
lemmas, including one using the domain pairing technique, and list some
known dichotomies.  In Section~\ref{partII-sec2}, we discuss a technique
to simulate $\PlCSP$ by $\PlCSP^2$.
Section~\ref{PartII.secC.Odd-arity} proves Theorem~\ref{dichotomy-pl-csp2}
in the case when $\mathcal{F}$ contains at least one nonzero signature of
odd arity. Section~\ref{Arity-4} proves the base case of 
the even arity case of
Theorem~\ref{dichotomy-pl-csp2}
when $\mathcal{F}$ consists of a single signature of arity~$4$.
Section~\ref{PartII.secE.cyclcomic} gives an application
of cyclotomic field which simplifies the proof of 
Theorem~\ref{dichotomy-pl-csp2}
when $\mathcal{F}$ contains a signature 
in 
$\widetilde{\mathscr{M}} \setminus(\mathscr{P} \cup \widetilde{\mathscr{A}})$.
Section~\ref{PartII.secF.No-Mixing-of-pairs} proves the conditional
No Mixing lemmas for a pair of signatures of even arity.
Section~\ref{PartII.secG.No-Mixing-of-sets} generalizes
the No Mixing lemmas to a set of signatures of even arity.
Section~\ref{PartII.secH.csp2-dichotomy} finishes the proof
of Theorem~\ref{dichotomy-pl-csp2}.

\begin{remark}
 We occasionally make some remarks
 (such as Remark~\ref{rmk:E2:1} and Remark~\ref{rmk:E2:2} in Subsection~\ref{subsection-E.2})
 to explain the complications forced upon the proof by various reasons,
 and why another more straightforward approach would not succeed.
 These remarks are not logically necessary to the proof,
 but hopefully they provide some insight and point out pitfalls in the proof.
\end{remark}

The next lemma is a simple fact that is used many times.
It essentially says that the set $\{0, 1, i, -1, -i, \infty\}$ is closed set-wise under the mapping $z \mapsto \frac{z+1}{z-1}$.
The proof is straightforward,
so we omit it.

\begin{lemma} \label{pre-4-power}
 Let $x \neq y$ and $\lambda = \frac{x+y}{x-y}$.
 Then $\lambda^4 \not\in \{0,1\}$ iff $x^4 \neq y^4$ and $x y \neq 0$.
\end{lemma}

%In the proof of dichotomy theorem of Pl-$\#$CSP$^2$, we often use the following computation.
\begin{definition}[Derivative] \label{deriviative}
Let $f$ and
 $g$ be two symmetric signatures of arities $n$ and $m$ respectively,
and $n > m$.
By connecting all $m$ input edges of  $g$ to
 $f$, we get a planar $\{f, g\}$-gate with
 a signature of arity $n-m$. This \emph{derivative signature} will be denoted
by $\partial_g(f)$.
If $kn<m$ and we connect $k$ copies of $g$ to $f$,
which is the same as forming $\partial_g(f)$ sequentially
$k$ times, the resulting \emph{repeated derivative signature} is
 denoted by $\partial^k_g(f)$.
If $g=[1, 0, 1]$, we  denote $\partial_g(f)$ simply by $\partial(f)$.
\end{definition}

%\begin{remark}
\noindent
{\bf Calculus}:
 Our proof will make substantial use of a {\it calculus}
using this notion of derivatives.
This calculus is essentially a systematic way to
calculate the signatures of some gadget constructions.
In a  {\rm Pl-Holant} problem $\PlHolant(\mathcal{G}\mid\mathcal{F})$,
if $g \in \mathcal{G}$ and $f \in \mathcal{F}$,
then we say that $g$ is from the LHS and $f$ is from the RHS.
If $f$ has arity $n$ and $g$ has arity $m$,
and $n >m$, then
we can form the signature $\partial_g(f)$ and 
$\PlHolant(\mathcal{G}\mid \mathcal{F} \cup\{\partial_g(f)\})
\leq_T \PlHolant(\mathcal{G}\mid\mathcal{F})$.
If $m >n$ we can form $\partial_f(g)$ and
$\PlHolant(\mathcal{G} \cup\{\partial_f(g)\}\mid \mathcal{F})
\leq_T \PlHolant(\mathcal{G}\mid\mathcal{F})$.
In particular, for $\PlCSP^2(\mathcal{F}) \equiv
\PlHolant(\mathcal{EQ}_2 \mid\mathcal{F})$
we consider all $(=_{2k})$ as from the LHS. In this case
if $h \in \mathcal{F}$ with arity $< n$
then we can also form $\partial_h(f)$, by first moving
$h$ to LHS via $(=_2) \in \mathcal{EQ}_2$, and then
$\PlCSP^2(\mathcal{F} \cup \{\partial_h(f)\})
\leq_T \PlCSP^2(\mathcal{F})$.
Note that if we discuss $\PlCSP^4(\mathcal{F})$ then
this operation $\partial_h(f)$ is in general not permissible,
and has to be justified in each individual case, e.g. when 
$h$ has even arity and one can construct
$[1,0,1]^{\otimes 2}$ in the LHS.

   To familiarize the readers with
this calculus, we list some simple calculations below,
which we will use often in our proofs freely without comment.

For any $g$, the operator $\partial_g(\cdot)$ is a linear operator.
It also depends on $g$ linearly.

By definition $\partial([f_0, f_1, \ldots, f_n])
= [f_0+ f_2, f_1 + f_3, \ldots, f_{n-2} + f_n]$ has arity $n-2$.

\begin{enumerate}
%\noindent
\item
If $f=[s, t]^{\otimes n}$, then
%
%Let $f=[s, t]^{\otimes n}+[u, v]^{\otimes n}$, we often use the following computation:
\begin{itemize}
\item $\partial^k_{[a, b]}(f)=(as+bt)^k[s, t]^{\otimes n-k}$ if $n >k$.
\item $\partial^k_{[a, b, c]}(f)=(as^2+2bst+ct^2)^k[s, t]^{\otimes n-2k}$
if $n >2k$;
\\
in particular, $\partial^k(f)=(s^2+t^2)^k[s, t]^{\otimes n-2k}$.
\item $\partial^k_{=_4}(f)=(s^4+t^4)^k[s, t]^{\otimes n-4k}$, if $n > 4k$.
\item
For $g=[g_0, g_1, \ldots, g_m]$, we have $\partial_{g}(=_n)=[g_0, 0, \ldots, 0, g_m]$ of arity $n-m$, where $n>m$.
\end{itemize}
%
%%% JYC should give a name for this Perf Match type.
%

%\noindent
\item
Let $f$ be of arity $n$ and $f_k=(\pm 1)^k(n-2k)$ ($0 \le k \le n$), then
\begin{itemize}
\item $\partial(f)$ has arity $n' =
n-2$ and $(\partial(f))_k=2(\pm 1)^k(n'-2k)$.
If $n$ is odd, then
$\partial^{\frac{n-1}{2}}(f) = 2^{\frac{n-1}{2}} [ 1, \mp 1]$.
\item $\partial_{=_4}(f)$ has arity $n''= n-4$ and $(\partial_{=_4}(f))_k=2(\pm 1)^k(n''-2k)$.

If $n \equiv 1 \pmod 4$, then
$\partial^{\frac{n-1}{4}}_{=_4}(f) = 2^{\frac{n-1}{4}} [ 1, \mp 1]$.
\\
If $n \equiv 3 \pmod 4$, then
$\partial(\partial^{\frac{n-3}{4}}_{=_4}(f))
 = 2^{\frac{n+1}{4}} [ 1, \mp 1]$.
\end{itemize}

%\noindent
\item
Let $f$ be of arity $n$ and $f_k=(\pm i)^k(n-2k)$ ($0 \le k \le n$), then
\begin{itemize}
\item $\partial(f)=4[1, \pm i]^{\otimes n-2}$.
% has arity $n-2$.
%% JYC already stated with  ^{\otimes n-2}.
\item $\partial_{=_4}(f)$ has arity $m=n-4$ and $(\partial_{=_4}(f))_k=2(\pm i)^k(m-2k)$.

If $n \equiv 1 \pmod 4$, then
$\partial^{\frac{n-1}{4}}_{=_4}(f) = 2^{\frac{n-1}{4}} [ 1, \mp i]$.
\\
If $n \equiv 3 \pmod 4$, then
$\partial(\partial^{\frac{n-3}{4}}_{=_4}(f))
 = 2^{\frac{n+5}{4}} [ 1, \pm i]$.
\end{itemize}
%\end{remark}
\end{enumerate}

\vspace{.2in}

Now we define an inverse operator $\int(\cdot)$ to $\partial$.
Just like the usual calculus there is a certain non-uniqueness in
the expression in an \emph{indefinite} integral; this non-uniqueness
is addressed in Lemma~\ref{general-f-construction}.
One reasonable definition for $\int([f_0, f_1, \ldots, f_n])$
is $F = [F_0, F_1, \ldots, F_{n+2}]$ such that
\[F_k = \sum_{s \ge 0} (-1)^s f_{k+2s} = f_k - f_{k+2} + f_{k+4} -  \ldots\]
where we define $f_k=0$ for all $k > n$.
Clearly $\partial(F) =f$.

 \begin{lemma}\label{general-f-construction}
 Let $F$ and $G$ be symmetric signatures of  arity $n\geq 3$ and
suppose $\partial(F)= \partial(G)$.
Then $F - G = x[1, i]^{\otimes n}+y[1, -i]^{\otimes n}$, for
some constants $x$ and $y$.
 \end{lemma}
\begin{proof}
The signature $H = F - G$ satisfies $\partial(H) =0$,
and thus satisfies the second order recurrence relation 
$H_k+H_{k+2}=0$ for $0\leq k\leq n-2$.
Hence there exist constants $x$ and $y$ such that $H=x[1, i]^{\otimes n}+y[1, -i]^{\otimes n}$.
\end{proof}

Thus $\int(\cdot)$ 
is well-defined  up to
an additive  term $x[1, i]^{\otimes n}+y[1, -i]^{\otimes n}$.
In this paper,
we choose to write the expression $\int(f)$ by the following
definition  when 
a certain special expression of $f$ exists. This is more convenient for
our proofs.
%%% JYC. i basically decided to define F in terms of that expression of f.
%%% but reserve the notation \int(f) in the following way.
%% the reason is that officially if we say we define \int in the following
%%% and extend by linearity.... it is not clear it is well defined
%%% as a linear operator. the def'n must not depend on the particular
%%% expression.
\begin{definition}\label{defn-integral}
For $n\geq 3$,
 %let $C=x[1, i]^{\otimes n}+y[1, -i]^{\otimes n}$, where $x, y$ are constants.
% let $\int(\cdot)$ be a linear operator such that
 \begin{itemize}
% \item If the signature $g$ has arity $n-2$ and is identically zero, then $\int(g)$ has arity $n$ and is identically zero.
\item $\int(0) = 0$.
 \item For $a^2+b^2\neq 0$, $\int([a, b]^{\otimes n-2})=\frac{1}{a^2+b^2}[a, b]^{\otimes n}$.
 \item $\int([1, \pm i]^{\otimes n-2})$ has arity $n$ and $[\int([1, \pm i]^{\otimes n-2})]_k=\frac{1}{4}(\pm i)^k(n-2k)$.
 \item If the signature $g$ has arity $n-2$ and
 $g_k=(\pm 1)^k(n-2-2k)$, then $\int(g)$ has arity $n$ and
 $[\int(g)]_k=\frac{1}{2}(\pm 1)^k(n-2k)$.
\item If the signature $g$ has arity $n-2$ and
 $g_k=(\pm i)^k(n-2-2k)$, then $\int(g)$ has arity $n$ and
$[\int(g)]_k=(-\frac{n}{2}k+\frac{1}{2}k^2)(\pm i)^k$.
 \end{itemize}
\end{definition}

%The operator $\int$ is an inverse of $\partial$: For all $f$
%when $\int(f)$ is defined,
Clearly for all $f$
where $\int(f)$ is given in the above definition, $\partial[\int(f)]=f$.

When we prove the dichotomy theorem for $\PlCSP^2(f)$,
where $f$ has arity $n$,
we can get a signature $f'$ of arity $n-2$ by taking a self loop with $f$,
i.e., $f' = \partial(f)$.
Clearly $\PlCSP^2(f') \le_T \PlCSP^2(f)$.
If $f'\notin \mathscr{P} \cup \widetilde{\mathscr{A}} \cup \widetilde{\mathscr{M}}$,
then by induction $\PlCSP^2(f')$ is \numP-hard.
Thus $\PlCSP^2(f)$ is also \numP-hard.
Definition~\ref{defn-integral} allows us
to write down an explicit expression for $\int(f')$ for all
cases when $f'\in \mathscr{P}\cup \widetilde{\mathscr{A}}\cup
\widetilde{\mathscr{M}}$.
%\widetilde{M}$,
%the next lemma  shows that
 %$f$ can be expressed explicitly.
 %\begin{lemma}\label{general-f-construction}
 %Let $f$ be a signature of  arity $n\geq 3$ and $f'=\partial(f)=\sum_{i=1}^{m}g_i$,
 %where $g_i=[a_i, b_i]^{\otimes n-2}$, or $g_i$ has arity $n-2$ and $(g_i)_k=(\pm 1)^k(n-2-2k)$ or $(g_i)_k=(\pm i)^k(n-2-2k)$.
 %then $f=\int(f')+x[1, i]^{\otimes n}+y[1, -i]^{\otimes n}$, where $x$, $y$ are constants.
 %\end{lemma}
 %\begin{proof}
%Note that $\partial[\int(f')]=f'$.
%Let $h=f-\int(f')$. Then
%$\partial(h)=0$, i.e.,
%$h_k$ satisfies the second recurrence relation $h_k+h_{k+2}=0$ for $0\leq k\leq n-2$.
%Thus there exist constants $x$, $y$ such that $h=x[1, i]^{\otimes n}+y[1, -i]^{\otimes n}$.
%So $f=\int(f')+h=\int(f')+x[1, i]^{\otimes n}+y[1, -i]^{\otimes n}$.
%\end{proof}

The following is an explicit list of $\int(f')$ for $f' = \partial(f) \in \mathscr{P}\cup \widetilde{\mathscr{A}}\cup \widetilde{\mathscr{M}}$.
We can recover $f$ up to the constants $x, y$ from $\partial(f)$ by Lemma~\ref{general-f-construction}.
This list is for the convenience of the readers.

\begin{proposition}[Explicit List for $\int(f')$] \label{prop:explicit_list}
 \ 
 
 \begin{itemize}
  \item $\int(f')\equiv 0$ if $f'\equiv 0$.
  
  %\item $\int(f')=\frac{1}{a^2+b^2}[a, b]^{\otimes n}$ if $f'=[a, b]^{\otimes n-2}, a^2+b^2\neq 0$.
  
  \item $\int([1, 0]^{\otimes n-2}+a[0, 1]^{\otimes n-2})=[1, 0]^{\otimes n}+a[0, 1]^{\otimes n}$.
  
  \item $\int([1, \gamma]^{\otimes n-2}+i^r[1, -\gamma]^{\otimes n-2}) = \frac{1}{1+\gamma^2}[1, \gamma]^{\otimes n}+\frac{i^r}{1+\gamma^2}[1, -\gamma]^{\otimes n}$
  where $\gamma^2\neq -1, \gamma^8=1$.
  
  \item $\int([s, t\rho]^{\otimes n-2}\pm[t, s\rho]^{\otimes n-2}) = \frac{1}{s^2+\rho^2 t^2}[s, \rho t]^{\otimes n}\pm\frac{1}{\rho^2 s^2+t^2}[t, \rho s]^{\otimes n}$,
  where $\rho^4=1, st\neq 0, s^4\neq t^4$.
  
  \item $[\int(f')]_k = \frac{1}{2}(\pm 1)^k(n-2k)$ if $f'$ has arity $n-2$ and  $f'_k = (\pm 1)^k (n-2-2k)$.
  % $\rho^2=1$.
  
  \item $[\int(f')]_k = \frac{1}{4}(\pm i)^k(n-2k)$ if $f'$ has arity $n-2$ and $f' = [1, \pm i]^{\otimes n-2}$.
  
  \item $[\int(f')]_k = \frac{1}{4} [i^k+i^r(-i)^k] (n-2k)$ if $f'$ has arity $n-2$ and $f' = [1, i]^{\otimes n-2} + i^r[1, -i]^{\otimes n-2}$.
  
  \item $[\int(f')]_k = (-\frac{n}{2}k+\frac{1}{2}k^2)(\pm i)^k$ if $f'$ has arity $n-2$ and $f'_k = (\pm i)^k (n-2-2k)$.
 \end{itemize}
\end{proposition}

The following lemma is used to determine whether a binary signature belongs to various tractable sets
$\mathscr{P}$,
$\mathscr{A}$,
$\mathscr{A}^\dagger$,
$\widehat{\mathscr{M}}$,
and $\widehat{\mathscr{M}}^\dagger$.
It can be proved directly by the definition.

\begin{lemma} \label{binary}
 For any binary symmetric signature $f$,
 \begin{itemize}
  \item $f\in\mathscr{P}$ iff $f = [a, 0, c]$ or $f = [0, b, 0]$ or $f = [a, b]^{\otimes 2}$.

  \item $f\in\mathscr{A}$ iff up to a scalar,
   $f = [1, \rho, -\rho^2]$ where $\rho^4=1$, or
   $[0, 1, 0]$, or
   $[1, 0, \rho]$ where $\rho^4=1$,
%    or $[x, y]^{\otimes 2}$ where {\rm [}$x^4=y^4\neq 0$ or $xy=0${\rm ]}.
   or $[x, y]^{\otimes 2}$ where $(x^4=y^4\neq 0 \text{ or } xy=0)$.

 \item $f\in\mathscr{A}^\dagger$ iff up to a scalar,
  $f=[1, \alpha, -\alpha^2]$ where $\alpha^4=-1$, or
%%%% this alpha = \pm i u
  $[0, 1, 0]$, or
  $[1, 0, \rho]$ where $\rho^4=1$, or
  $[x, y]^{\otimes 2}$ where $(x^4=-y^4 \neq 0 \text{ or }xy=0)$.
%%%% this u is the case of u in A, times -1. so ()^4 = same 1
%%% JYC:
%%% this u is the case of u in A, times i or -i. so ()^4 = same 1
 \item $f\in\widehat{\mathscr{M}}$ iff
  $f=[a, b, a]$ or $[a, 0, -a]$.
 \item $f\in\widehat{\mathscr{M}}^\dagger$ iff
  $f=[a, b, -a]$ or $[a, 0, a]$.
 \end{itemize}
 \end{lemma}

\noindent
 Corollary~\ref{binary-necessary} gives some necessary conditions for a binary signature to belong to a tractable set.

 \begin{corollary} \label{binary-necessary}
  For any binary signature $f = [a, b, c]$,
 \begin{itemize}
\item  $f\in\mathscr{P} \implies f$ satisfies  either
the parity constraint or $b^2=ac$.
\item  $f\in\mathscr{A} \implies a^2 = c^2$ or $b =0$.
If $f\in\mathscr{A}\setminus\mathscr{P}$, then $f=[1, \rho, -\rho^2]$, $\rho^4=1$.
\item  $f\in\mathscr{A}^\dagger \implies a^2 = -c^2$ or $b =0$.
If $f\in\mathscr{A}^\dagger\setminus\mathscr{P}$, then $f=[1, \alpha, -\alpha^2]$, $\alpha^4=-1$.
\item  $f\in\widetilde{\mathscr{A}}$ $\Longrightarrow$ the norms
 of all nonzero entries are equal.
\item  $f\in\widetilde{\mathscr{M}} \implies a^2 = c^2$.
\end{itemize}
\end{corollary}

Furthermore,
all signatures in each tractable set satisfy a second order recurrence relation.

\begin{definition} \label{second-recurrence-relation-definition}
Let $f=[f_0, f_1, \ldots, f_n]$.
If there exist constants $a, b$ and $c$, not all zero, such that
 $af_k-bf_{k+1}+cf_{k+2}=0$  for $1\leq k\leq n-2$,
% up to a nonzero scalar,
 then we say $f$ has type $\langle a, b, c\rangle$,
and it is denoted by $f\in \langle a, b, c\rangle$.
\end{definition}
For a non-degenerate symmetric signature $f$ of arity at least 3,
if $f$ has type $\langle a, b, c\rangle$, its type is uniquely
determined up to a nonzero multiple.
%
%
%Furthermore, if the signature is non-degenerate and has arity $\geq 3$,
%then the second recurrence relation is unique up to to a scalar.
%%% already said so in the comment after def'n.
The next lemma states this type information for the various
tractable sets.
We can use the lemma to check whether
 a symmetric signature can possibly be
 in a tractable set.
\begin{lemma}\label{second-recurrence-relation}
Let  $f\in\mathscr{P}\cup\widetilde{\mathscr{A}}\cup\widetilde{\mathscr{M}}$ be non-degenerate and have arity $\geq 3$.
\begin{itemize}
\item If $f\in \mathscr{P}$ then $f\in \langle0, 1, 0\rangle$.
%\item If $f\in \mathscr{P}\cap\mathscr{A}$, or $f\in \mathscr{P}\cap\mathscr{A}^\dagger$,
%or $f\in \mathscr{P}\cap\widehat{\mathscr{M}}$,
%$f\in \mathscr{P}\cap\widehat{\mathscr{M}}^\dagger$, then $f\in {\rm SRR}[0, 1, 0]$
%$f\in \mathscr{P}\cap\mathscr{A}$,
\item If $f\in \mathscr{A}$ then $f\in \langle0, 1, 0\rangle$ or $f\in\langle1, 0, \pm 1\rangle$.
If $f\in \mathscr{A}\setminus\mathscr{P}$ then $f\in \langle1, 0, \pm 1\rangle$.
\item If $f\in \mathscr{A}^\dagger$ then $f\in \langle 0, 1, 0 \rangle$ or $f\in\langle1, 0, \pm i\rangle$.
If $f\in \mathscr{A}^\dagger\setminus\mathscr{P}$ then $f\in \langle1, 0, \pm i\rangle$.
\item If $f\in \widehat{\mathscr{M}}$ then $f\in \langle 0, 1, 0 \rangle$ or $f\in \langle1, c, 1\rangle$.
If $f\in \widehat{\mathscr{M}}\setminus(\mathscr{P}\cup\widetilde{\mathscr{A}})$ then $f\in \langle1, c, 1\rangle$ with $c\neq 0$.
\item If $f\in \widehat{\mathscr{M}}^\dagger$ then $f\in\langle 0, 1, 0 \rangle$ or $f\in\langle1, c, -1\rangle$.
If $f\in \widehat{\mathscr{M}}^\dagger\setminus(\mathscr{P}\cup\widetilde{\mathscr{A}})$ then $f\in \langle1, c, -1\rangle$ with $c\neq 0$.
\end{itemize}
\end{lemma}

The following two corollaries follow from Lemma~\ref{binary}
for the binary case, and Lemma~\ref{second-recurrence-relation}
for arity $n \ge 3$.

\begin{corollary}\label{A-2-A-minus-P}
If $f\in\mathscr{A}\setminus\mathscr{P}$, then $f\notin\mathscr{A}^\dagger$.
Similarly, If $f\in\mathscr{A}^\dagger \setminus \mathscr{P}$, then $f\notin\mathscr{A}$.
\end{corollary}

\begin{corollary}\label{M-2-M-A-P}
If $f\in\widehat{\mathscr{M}}\setminus(\mathscr{P}\cup\widetilde{\mathscr{A}})$, then $f\notin\widehat{\mathscr{M}}^\dagger$.
Similarly, if $f\in\widehat{\mathscr{M}}^\dagger \setminus (\mathscr{P}\cup\widetilde{\mathscr{A}})$, then $f\notin\widehat{\mathscr{M}}$.
\end{corollary}

The following lemma gives a characterization for
$\widetilde{\mathscr{M}}\setminus(\mathscr{P}\cup\widetilde{\mathscr{A}})$.

\begin{lemma} \label{M-2-M-NOT-IN-A-AND-P}
 Let $f = [f_0, \dotsc, f_n]$ be a symmetric signature of arity $n$.
 Then $f \in \widehat{\mathscr{M}} \setminus(\mathscr{P} \cup \widetilde{\mathscr{A}})$
 iff
 \begin{itemize}
  \item $n = 2$ and $f = \lambda[1, a, 1]$,
  where $a^4 \not\in \{0,1\}$ and $\lambda \neq 0$; or

  \item $n \geq 3$ and $f = [s, t]^{\otimes n} \pm [t, s]^{\otimes n}$,
  where $s t \neq 0$ and $s^4 \neq t^4$; or

  \item $n \geq 3$ and $f_k = \lambda (\pm 1)^k (n-2k)$,
  where $\lambda \neq 0$.
 \end{itemize}
 Similarly,
 $f \in \widehat{\mathscr{M}}^\dagger \setminus (\mathscr{P} \cup \widetilde{\mathscr{A}})$
 iff
 \begin{itemize}
  \item $n = 2$ and $f = \lambda[1, b, -1]$,
  where $b^4 \not\in \{0,1\}$ and $\lambda \neq 0$; or

  \item $n \geq 3$ and $f = [s, ti]^{\otimes n} \pm [t, si]^{\otimes n}$
  where $s t \neq 0$ and $s^4 \neq t^4$; or

  \item $n \geq 3$ and $f_k = \lambda (\pm i)^k (n-2k)$,
  where $\lambda \neq 0$.
 \end{itemize}
\end{lemma}

\begin{proof}
 We prove the lemma for $\widehat{\mathscr{M}}$.
 The proof for $\widehat{\mathscr{M}}^\dagger$ follows from a holographic transformation by $\trans{1}{0}{0}{i}$.

 By Lemma~\ref{binary},
 a binary symmetric signature $f \in \widehat{\mathscr{M}}$ 
has the form $[a,b,a]$ or $[a, 0, -a]$.
 Since $[a, 0, -a] \in {\mathscr{A}}$ as a multiple of $[1,0,-1]$,
 we exclude it.
 For $[a,b,a]$,
 if $a b = 0$,
 then $f \in \mathscr{P}$.
 Also if $a^4 = b^4$,
 then $[a, b, a] \in  {\mathscr{A}}$, being a multiple of
$[1, \pm 1]^{\otimes 2}$ or $[1, \pm i, 1]$.
%%% including [1 1 1] or [1 -1 1] in A.
 This gives the form $f = \lambda [1, b, 1]$  with $b^4 \not\in \{0,1\}$ and $\lambda \neq 0$.
 Conversely,
 any $f$ of this form belongs to
 $\widehat{\mathscr{M}} \setminus(\mathscr{P} \cup \widetilde{\mathscr{A}})$.

%A binary symmetric signature $f\in\widehat{\mathscr{M}}
%\setminus \mathscr{P}$ iff
%$f=\lambda[1, a, 1]$ or $\lambda[1, 0, -1]$, with $\lambda \not =0$,
%by Lemma~\ref{binary}.
%Note that $[1, 0, -1]\in\mathscr{A}$ and $[1, a, 1]\in\mathscr{A}$ for $a^4=0, 1$.
%Thus if $f\in\widehat{\mathscr{M}}\setminus(\mathscr{P}\cup\widetilde{\mathscr{A}})$,
%then $f$ must be of form $[1, a, 1]$ with $a^4\neq 0, 1$ up to a scalar.
%Conversely, if $a^4\neq 0, 1$,
%$f=[1, a, 1]\in\widehat{\mathscr{M}}\setminus(\mathscr{P}\cup\widetilde{\mathscr{A}})$ by Lemma~\ref{binary}.
%
For arity $n\geq 3$, $f\in\widehat{\mathscr{M}}$ iff
$f$ takes the form $[s, t]^{\otimes n}\pm [t, s]^{\otimes n}$
or $f_k=\lambda(\pm 1)^k(n-2k)$.
% JYC can't say \lambda \not = 0 yet. as we only said in hat-M
% where $\lambda$ is a nonzero constant.
For the latter case
$f\in\widehat{\mathscr{M}}\setminus(\mathscr{P}\cup\widetilde{\mathscr{A}})$
%for $f_k=\lambda(\pm 1)^k(n-2k)$
follows from its type $\langle 1, \pm 2, 1\rangle$.
%%% Zhiguo, pl check.

For $f=[s, t]^{\otimes n}\pm [t, s]^{\otimes n}$,
if $st=0$, then $f\in\mathscr{P}$.
If $s^2=t^2$, then $f$ is degenerate, thus $f\in\mathscr{P}$.
If $s^2=-t^2$, then $f\in\mathscr{A}$.
Conversely, if $st\neq 0$ and $s^4\neq t^4$,
then $f$ is non-degenerate and $f_k$ has type
 $\langle 1, \frac{s}{t}+\frac{t}{s}, 1\rangle$.
Note that $\frac{s}{t}+\frac{t}{s}\neq 0$ by $s^4\neq t^4$.
Thus
$f\in\widehat{\mathscr{M}}\setminus(\mathscr{P}\cup\widetilde{\mathscr{A}})$ by Lemma~\ref{second-recurrence-relation}.
\end{proof}

By the second recurrence relation of the signatures in $\widetilde{\mathscr{M}}\setminus(\mathscr{P}\cup\widetilde{\mathscr{A}})$,
we have the following lemma that will be used in the proof of Theorem~\ref{odd-mixing-theorem}.
%\textcolor{blue}{
\begin{corollary}\label{widetilde-M-minus-P-and-A-parity}
If $f\in\widetilde{\mathscr{M}}\setminus(\mathscr{P}\cup\widetilde{\mathscr{A}})$,
then $f$ does not satisfy parity constraints.
\end{corollary}

\begin{proof}
For $f\in\widehat{\mathscr{M}}\setminus(\mathscr{P}\cup\widetilde{\mathscr{A}})$,
if $f$ has arity 2, then $f=\lambda [1, a, 1]$ for
some $\lambda \neq 0$, $a^4\neq 0, 1$ by Lemma~\ref{M-2-M-NOT-IN-A-AND-P}.
Thus it does not satisfy parity constraints.

For $f$ with arity $n\geq 3$, by Lemma~\ref{second-recurrence-relation}, there exists a constants $c\neq 0$
such that $f\in\langle 1, c, 1\rangle$.
Note that there exists $f_k\neq 0$, where $1\leq k\leq n-1$ by $f\notin \mathscr{P}$.
If $f$ satisfies parity constraints, then $f_{k-1}=f_{k+1}=0$.
Moreover, by $f_{k-1}-cf_k+f_{k+1}=0$, we have $c=0$. This is a contradiction.

The proof for
$f\in\widehat{\mathscr{M}}^\dagger\setminus(\mathscr{P}\cup\widetilde{\mathscr{A}})$
follows from a holographic transformation by
$\left[\begin{smallmatrix} 1 & 0 \\
0 & i \end{smallmatrix}\right]$.
\end{proof}

%\begin{corollary}\label{arity-4-binary-affine-norm}
%If $[a, b, c]\in\widetilde{\mathscr{A}}$ and $abc\neq 0$, then $|a|=|b|=|c|$.
%\end{corollary}

The following lemma gives a characterization
% necessary and sufficient condition
%of non-degenerate signatures in $\widehat{\mathscr{M}}$.
of nonzero signatures in $\widehat{\mathscr{M}}$.
A \textsc{Gen-Eq} is a signature of the form $f = [a, 0, \ldots, 0, b]$,
called a generalized equality (with $a=0$ or $b=0$ allowed.)
%%% JYC possibly a or b =0
\begin{lemma}\label{arity-hat-M}
A \textsc{Gen-Eq} signature $f$ is in $\widehat{\mathscr{M}}$ iff $f
= \lambda [1, 0, \ldots, 0, \pm 1]$, for some $\lambda$.

Suppose $f$ is a symmetric signature that is not a \textsc{Gen-Eq}.
%Suppose $f$ is not a \textsc{Gen-Eq}.
%Assume that the signature $f$ is non-degenerate and is not general equality.
Then $f \in \widehat{\mathscr{M}}$ iff $f$ satisfies a second order recurrence
$f_{k} - c f_{k+1} + f_{k+2} = 0$
(for $0 \leq k \leq \arity(f) - 2$)
and the following conditions hold.

If $f$ has arity $2n$, then
\begin{itemize}
\item $f_{n-k}=f_{n+k}$ (for $0 \le k \le n$),
 $f_n\neq 0$,
 $c=
%\frac{f_{n-1}+f_{n+1}}{f_n} =
\frac{2f_{n-1}}{f_n}$; or
\item$f_{n-k}=-f_{n+k}$ (for $0 \le k \le n$),
%$f_n=0,
$f_{n-1}\neq 0$,
 $c=\frac{f_{n-2}}{f_{n-1}}$.
\end{itemize}

If $f$ has arity $2n+1$, then
\begin{itemize}
\item $f_{n-k}=f_{n+1+k}$ (for $0 \le k \le n$),
$f_n\neq 0$,
 $c
%=\frac{f_{n-1}+f_{n+1}}{f_n}
= \frac{f_{n-1}}{f_n} + 1$; or
\item $f_{n-k}=-f_{n+1+k}$ (for $0 \le k \le n$),
$f_n\neq 0$,
$c
%=\frac{f_{n-1}+f_{n+1}}{f_{n}} =
=\frac{f_{n-1}}{f_n} -1$.
\end{itemize}
\end{lemma}

\begin{proof}
%Non-degenerate signatures in $\widehat{\mathscr{M}}$ have the following forms,
Symmetric signatures in $\widehat{\mathscr{M}}$ have the following forms,
$f=[s, t]^{\otimes m} \pm [t, s]^{\otimes m}$,
% where $s \not = \pm t$,
or $f_k = \lambda (\pm 1)^k(m-2k)$ ($0 \le k \le m$).
%cite http://pages.cs.wisc.edu/~jyc/papers/matchgate-arts-to-sc.pdf or
%%% in article already stated this thm.
A \textsc{Gen-Eq} $f \in \widehat{\mathscr{M}}$ iff
it takes the first form with $st=0$.
Suppose $f$ is not a  \textsc{Gen-Eq}, then we have $st \not =0$
in the first form.
In particular $f$ is not identically zero.
In both forms, $f$
satisfies  a second order recurrence
$f_{k}-cf_{k+1}+f_{k+2}=0~(0\leq k\leq m-2)$,
for some $c$.
For example in the first form with a tensor sum,
the product of the
eigenvalues $s/t \cdot t/s =1$.
%%% in the \pm case, eigenvalue product =1. in second, direct check.

%%% that st \not =0 is used to claim two eigenvalues s/ and t/s, and prod=1
%%% this implies a recurrence of this form.

For even arity $m=2n$, and
$f=[s, t]^{\otimes 2n} + [t, s]^{\otimes 2n}$,
we have the symmetry $f_{n+k}=f_{n-k}$.
Thus $f_{n-1}=f_{n+1}$ and $cf_n=2f_{n-1}$.
If $f_n=0$, then $f$ is identically zero, a contradiction.
Therefore, we have $c=\frac{2f_{n-1}}{f_n}$.

For $f=[s, t]^{\otimes 2n}-[t, s]^{\otimes 2n}$,
or $f_k= \lambda (\pm 1)^k(2n-2k)$,
we have $f_{n+k}=-f_{n-k}$.
Thus we have $f_{n}=0$ and $cf_{n-1}=f_{n-2}$.
If $f_{n-1}=0$, then $f$ is identically zero, a contradiction.
Therefore, we have $c=\frac{f_{n-2}}{f_{n-1}}$.

%For $f$ has arity $2n$ and $f_k=(\pm 1)^k(2n-2k)$, we have
%$c=-\frac{f_{n-2}}{f_{n-1}}=\mp 2$ and $f_{k}+cf_{k+1}+f_{k+2}=0$.

Conversely, the second order recurrence
$f_{k}-cf_{k+1}+f_{k+2}=0$ gives
the expression
$f = c_1 [s, t]^{\otimes 2n} + c_2 [t, s]^{\otimes 2n}$,
or  in the double root case
%$\pm 1$
when $c=\pm 2$, we have the form $f_k = \lambda (\pm 1)^k(2n- \mu k)$.
If $f_{n+k}=-f_{n-k}$, then $f_n=0$, the double root case
must be $f_k = \lambda (\pm 1)^k(2n- 2k)$,
and the tensor sum takes the form
$f = [s, t]^{\otimes 2n} - [t, s]^{\otimes 2n}$.
If $f_{n+k}=f_{n-k}$, then we only have the form
$f = [s, t]^{\otimes 2n} + [t, s]^{\otimes 2n}$.

For odd arity, the proof is similar. We omit it here.
%%% this time, the double root case is combined with the + case of \pm
%% recurrence. i think. JYC
\end{proof}

\begin{corollary}\label{arity-hat-M-dagger}
 If $f \in \widehat{\mathscr{M}}^\dagger$ has even arity $2n$,
 then for all $0 \le k \le 2n$,
 \[
  f_k
  =
  f_{2n-k}
  \qquad \text{or} \qquad
  f_k=-f_{2n-k}
 \]
 and the signs strictly alternate.
\end{corollary}

\begin{proof}
By definition,  $\widehat{\mathscr{M}}^\dagger =
\left[\begin{smallmatrix} 1 & 0 \\
0&i \end{smallmatrix}\right]^{\otimes 2n} \widehat{\mathscr{M}}$.
By Lemma~\ref{arity-hat-M},
for some $\epsilon = \pm 1$,
we have $i^{n-k} f_{n-k} = \epsilon i^{n+k} f_{n+k}$ for all $k$.
The Corollary follows.
%we have $i^kf_k=i^{2n-k}f_{2n-k}$
%or $i^kf_k=-i^{2n-k}f_{2n-k}$.
%Thus we have $f_k=f_{2n-k}$
%or $f_k=-f_{2n-k}$.
\end{proof}

In the proof of
Pl-\#CSP$^2$ dichotomy,
we often use the following Corollary. It gives a characterization
of a signature of arity 4 in $\widetilde{\mathscr{M}}$.
It follows directly from
 Lemma~\ref{arity-hat-M} and the definition of $\widehat{\mathscr{M}}^\dagger$.

\begin{corollary}\label{mixing-M-arity-4}
An arity 4 signature $f\in\widehat{\mathscr{M}}$ has the following forms:
\begin{itemize}
\item $[u, v, w, v, u]$ and $(u+w)w=2v^2$; or
\item $[u, v, 0, -v, -u]$.
\end{itemize}

%%% the above also contains the case when w=0 in 1st ==\rangle v=0 and [10001]
%%% also in 2dn, v=0==\rangle[1,0,0,0,-1] or identically 0.
%%% all in M-hat

An arity 4 signature $f\in\widehat{\mathscr{M}}^\dagger$ has the following forms:
\begin{itemize}
\item $[u, v, w, -v, u]$ and $(u-w)w=2v^2$,
\item $[u, v, 0, v, -u]$.
\end{itemize}
\end{corollary}

The following lemma can be proved by domain pairing.
We can use it to derive \#P-hardness of Pl-\#{\rm CSP}$^2$
problems by applying the known dichotomy of Pl-\#{\rm CSP}.

\begin{lemma}\label{domain-pairing-expand}
 Suppose $f = [f_0, f_1, \dotsc, f_{2n}]$ is a symmetric signature of arity $2n$.
 Let $g = [f_0, f_2, \dotsc, f_{2n}]$ be a symmetric signature of arity $n$ consisting of all even indexed entries of $f$.
 Then \[\PlCSP(g) \leq \PlCSP^2(f).\]
\end{lemma}

\begin{proof}
For any instance of $\PlCSP(g)$,  we replace
%note that this is a bipartite graph.
each edge $e$ by two edges that connect the same incident nodes of $e$.
For each variable node that is connected to $k$ edges,
 we replace its label $=_k$ by $=_{2k}$.
We replace each occurrence of $g$ by $f$ as a constraint.
Then the new instance is a problem
in $\PlCSP^2(f)$
and has the same value as the given instance of
$\PlCSP(g)$, because $g_k=f_{2k}$.
Note that the
values $f_{2k+1}$ with an odd index
 contribute nothing to the partition function
in this instance.
\end{proof}

The case when $f = [1, i]^{\otimes 4}+a[1, -i]^{\otimes 4}$ poses some
special difficulty, mainly because $\partial(f)$ is identically 0.
The following lemma shows that in this case,
with $a \neq 0$,
we can construct $[1, 0, -1]^{\otimes 2}$ in the LHS
in a Pl-Holant problem with $f$ on the RHS.
Its utility is that after a holographic transformation by
$\left[\begin{smallmatrix} 1 & 0 \\
0 & i \end{smallmatrix}\right]$
%
%$\begin{bmatrix}1&0\\0&i\end{bmatrix}$
or by
$\left[\begin{smallmatrix}
1&1\\
i&-i
\end{smallmatrix}\right]
=
\left[\begin{smallmatrix} 1 & 0 \\
0 & i \end{smallmatrix}\right]
\left[\begin{smallmatrix} 1 & 1 \\
1 & -1 \end{smallmatrix}\right]$
% $\begin{bmatrix}1&i\\i&-i\end{bmatrix}$,
we have $[1, 0, 1]^{\otimes 2}$ on the LHS.

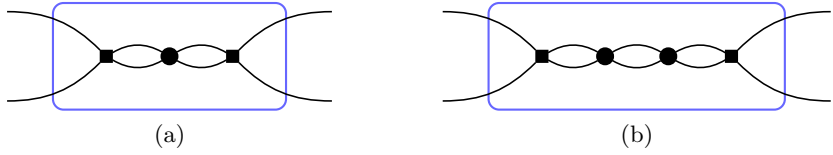
\begin{figure}[t]
 \centering
 \subfloat[]{
  \begin{tikzpicture}[scale=\scale,transform shape,node distance=\nodeDist,semithick]
   \node[square]    (0)                    {};
   \node[external]  (1) [above left  of=0] {};
   \node[external]  (2) [below left  of=0] {};
   \node[external]  (3) [left        of=1] {};
   \node[external]  (4) [left        of=2] {};
   \node[internal]  (5) [right       of=0] {};
   \node[square]    (6) [right       of=5] {};
   \node[external]  (7) [above right of=6] {};
   \node[external]  (8) [below right of=6] {};
   \node[external]  (9) [right       of=7] {};
   \node[external] (10) [right       of=8] {};
   \path (0) edge[out= 135, in=   0]  (3)
             edge[out=-135, in=   0]  (4)
             edge[bend left]          (5)
             edge[bend right]         (5)
         (5) edge[bend left]          (6)
             edge[bend right]         (6)
         (6) edge[out=  45, in= 180]  (9)
             edge[out= -45, in= 180] (10);
   \begin{pgfonlayer}{background}
    \node[draw=\borderColor,thick,rounded corners,fit = (1) (2) (7) (8),inner sep=0pt] {};
%     \node[draw=\borderColor,thick,rounded corners,fit = (1) (2) (7) (8),inner sep=0pt,transform shape=false] {};
   \end{pgfonlayer}
  \end{tikzpicture}
  \label{subfig:1-1}
 }
 \qquad
 \subfloat[]{
  \begin{tikzpicture}[scale=\scale,transform shape,node distance=\nodeDist,semithick]
   \node[square]    (0)                    {};
   \node[external]  (1) [above left  of=0] {};
   \node[external]  (2) [below left  of=0] {};
   \node[external]  (3) [left        of=1] {};
   \node[external]  (4) [left        of=2] {};
   \node[internal]  (5) [right       of=0] {};
   \node[internal]  (6) [right       of=5] {};
   \node[square]    (7) [right       of=6] {};
   \node[external]  (8) [above right of=7] {};
   \node[external]  (9) [below right of=7] {};
   \node[external] (10) [right       of=8] {};
   \node[external] (11) [right       of=9] {};
   \path (0) edge[out= 135, in=   0]  (3)
             edge[out=-135, in=   0]  (4)
             edge[bend left]          (5)
             edge[bend right]         (5)
         (5) edge[bend left]          (6)
             edge[bend right]         (6)
         (6) edge[bend left]          (7)
             edge[bend right]         (7)
         (7) edge[out=  45, in= 180] (10)
             edge[out= -45, in= 180] (11);
   \begin{pgfonlayer}{background}
    \node[draw=\borderColor,thick,rounded corners,fit = (1) (2) (8) (9),inner sep=0pt] {};
%     \node[draw=\borderColor,thick,rounded corners,fit = (1) (2) (8) (9),inner sep=0pt,transform shape=false] {};
   \end{pgfonlayer}
  \end{tikzpicture}
  \label{subfig:1-2}
 }
 \caption{Two gadgets used to obtain $[1,0,-1]^{\otimes 2}$.
 The circle vertices are assigned $f$ and the square vertices are assign $=_4$.}
 \label{fig:1}
\end{figure}

\begin{lemma} \label{construct-[1,0,1]-by-(1,i)-(1,-i)}
 Let $\mathcal{F}$ be a set of signatures containing  $f = [1, i]^{\otimes 4} + a[1, -i]^{\otimes 4}$.
 Then
 \[
  \PlHolant([1, 0, -1]^{\otimes 2} \cup \mathcal{EQ}_2 \mid \mathcal{F})
  %\leq_T
\equiv
%%% JYC i changed to \equiv. as it is how is quoted. clearly \equiv
%%% because EQ_2 are present on the LHS. just added something on the LHS
  \PlCSP^2(\mathcal{F}).
 \]
\end{lemma}

\begin{proof}
 Suppose $a \neq -1$ and consider the gadget in Figure~\ref{subfig:1-1}.
 We assign $f$ to the circle vertex and $=_4$ to the square vertices.
 This gives $(1+a) [1, 0, -1]^{\otimes 2}$ on the left as desired.

 Otherwise $a = -1$.
 Consider the gadget in Figure~\ref{subfig:1-2}.
 We assign $f$ to the circle vertices and $=_4$ to the square vertices.
 This gives $-8 [1, 0, -1]^{\otimes 2}$ on the left as desired.
\end{proof}

Coming up next are a couple of complexity dichotomy theorems 
that were previously shown.
They are  also quoted in Section \ref{sec:preliminaries} of Part~I.
Here we restate them for easier reference.
The first is a dichotomy theorem about counting complex weighted graph homomorphisms over degree prescribed graphs.
It includes $\PlCSP^2(f)$,
where $f$ is a symmetric \emph{binary} signature,
as a special case.
This is also quoted as Theorem~\ref{thm:k-reg_homomorphism} in Part~I.

\begin{theorem}[Theorem~3 in~\cite{CK12}]  \label{thm:parII:k-reg_homomorphism}
 Let $S \subseteq \mathbb{Z}^+$ contain $k \ge 3$,
 let $\mathcal{G} = \{=_k \st k \in S\}$,
 and let $d = \gcd(S)$.
 Further suppose that $f_0, f_1, f_2 \in \mathbb{C}$.
 Then $\plholant{[f_0, f_1, f_2]}{\mathcal{G}}$ is $\SHARPP$-hard unless
 one of the following conditions holds:
 \begin{enumerate}
  \item $f_0 f_2 = f_1^2$;
  \item $f_0 = f_2 = 0$;
  \item $f_1 = 0$;
  \item $f_0 f_2 = -f_1^2$ and $f_0^d = -f_2^d\neq0$;
  \item $f_0^d = f_2^d\neq0$.
 \end{enumerate}
 In any exceptional case,
 the problem is computable in polynomial time.
\end{theorem}

Theorem~\ref{thm:parII:k-reg_homomorphism} is the original statement as in~\cite{CK12}.
It is explicit and easy to apply.
Conceptually, it can be restated as Theorem~\ref{thm:parII:k-reg_homomorphism}$'$,
which supports the putative form of the $\PlCSP^d$ dichotomy.

%%% TDW: to cite this theorem, cite the previous theorem and add a prime in mathmode ($'$).

\begin{specialtheorem}[Theorem~3 in~\cite{CK12}]
 Let $S \subseteq \mathbb{Z}^+$ contain $k \ge 3$,
 let $\mathcal{G} = \{=_k \st k \in S\}$,
 and let $d = \gcd(S)$.
 Further suppose that $f$ is a non-degenerate, symmetric, complex-valued binary signature in Boolean variables.
 Then $\plholant{f}{\mathcal{G}}$ is \numP-hard unless
 $f$ satisfies one of the following conditions,
 in which case,
 the problem is computable in polynomial time:
 \begin{enumerate}
  \item there exists $T \in \mathcal{T}_{4d}$ such that $T^{\otimes 2} f \in \mathscr{A}$;
  \item                                                 $f               \in \mathscr{P}$;
  \item there exists $T \in \mathcal{T}_{2d}$ such that $T^{\otimes 2} f \in \widehat{\mathscr{M}}$.
 \end{enumerate}
\end{specialtheorem}

The following theorem is the dichotomy theorem of $\PlCSP(\mathcal{F})$,
where $\mathcal{F}$ is a set of symmetric signatures.
This is also quoted as Theorem \ref{thm:PlCSP} in Part~I.

\begin{theorem}[Theorem~19 in~\cite{GW13}] \label{pl-dicho-1}
 Let $\mathcal{F}$ be any set of symmetric, complex-valued signatures in Boolean variables.
 Then $\PlCSP(\mathcal{F})$ is $\numP$-hard unless
 $\mathcal{F} \subseteq \mathscr{A}$,
 $\mathcal{F} \subseteq \mathscr{P}$, or
 $\mathcal{F} \subseteq \widehat{\mathscr{M}}$,
 in which case the problem is computable in polynomial time.
\end{theorem}

We repeat the definition of redundant matrices in Section \ref{sec:redundant}.

\begin{definition}[Definition~6.1 in~\cite{CGW13}]
 A 4-by-4 matrix is \emph{redundant} if its middle two rows and middle two columns are the same.
\end{definition}

An example of a redundant matrix is the signature matrix of a symmetric arity~4 signature.

\begin{definition}[Definition~6.2 in~\cite{CGW13}]
 The \emph{signature matrix} of a symmetric arity~4 signature $f = [f_0, f_1, f_2, f_3, f_4]$ is
 \begin{align*}
  M_f =
  \begin{bmatrix}
   f_0 & f_1 & f_1 & f_2\\
   f_1 & f_2 & f_2 & f_3\\
   f_1 & f_2 & f_2 & f_3\\
   f_2 & f_3 & f_3 & f_4
  \end{bmatrix}.
 \end{align*}
 This definition extends to an asymmetric signature $g$ as
 \begin{align*}
  M_g =
  \begin{bmatrix}
   g^{0000} & g^{0010} & g^{0001} & g^{0011}\\
   g^{0100} & g^{0110} & g^{0101} & g^{0111}\\
   g^{1000} & g^{1010} & g^{1001} & g^{1011}\\
   g^{1100} & g^{1110} & g^{1101} & g^{1111}
  \end{bmatrix}.
 \end{align*}
 When we present $g$ as an $\mathcal{F}$-gate, we order the four external edges ABCD counterclockwise.
 In $M_g$,
 the row index bits are ordered AB and the column index bits are ordered DC,
 in reverse order.
 This is for convenience so that the signature matrix of the linking of two arity~4 $\mathcal{F}$-gates
 is the matrix product of the signature matrices of the two $\mathcal{F}$-gates.

 If $M_g$ is redundant, we also define the \emph{compressed signature matrix} of $g$ as
 \[
  \widetilde{M_g}
  =
  \begin{bmatrix}
   1 & 0 & 0 & 0\\
   0 & \frac{1}{2} & \frac{1}{2} & 0\\
   0 & 0 & 0 & 1
  \end{bmatrix}
  M_g
  \begin{bmatrix}
   1 & 0 & 0\\
   0 & \frac{1}{2} & 0\\
   0 & \frac{1}{2} & 0\\
   0 & 0 & 1
  \end{bmatrix}.
 \]
\end{definition}

The definition of \emph{compressed signature matrix}
is a slight change from~\cite{GW13} where $\widetilde{M_g} 
\left[
 \begin{smallmatrix}
 1 & 0 & 0\\
 0 & 2 & 0\\
 0 & 0 & 1
 \end{smallmatrix}
 \right]$ is called by that name.
%the compressed signature matrix.
It does not affect the following lemma.
We repeat the following lemma from \cite{GW13},
which is very convenient to apply.

\begin{lemma} [Corollary 3.8 in \cite{GW13}] \label{4-redundant}
 Let $f$ be an arity $4$ signature with complex weights.
 If $M_f$ is redundant and $\widetilde{M_f}$ is nonsingular,
 then $\PlHolant(f)$ is $\numP$-hard.
\end{lemma}

\section{Reduction from \texorpdfstring{$\PlCSP$}{Pl-\#CSP} to \texorpdfstring{$\PlCSP^2$}{Pl-\#CSP2}}\label{partII-sec2}

\begin{definition}
%For $k\equiv \ell\pmod 2$,
For $k \ge 1$, $\ell \ge 0$ and any $\omega$,
we define $E^{\ell}_k(\omega) = [1, 0, \ldots, 0, \omega^{\ell}]$
to be a signature of arity $k$,
and define $E(\omega)=\{E^{\ell}_k(\omega) \mid k\equiv \ell\pmod 2\}$.
We also write $E^{\ell}_k$ for $E^{\ell}_k(\omega)$ when
$\omega$ is clear from the context.
\end{definition}

The following lemma shows that if we have a unary $[1, \omega]\in\mathcal{F}$ with $\omega\neq 0$,
then either $\mathcal{F}$ is contained in one single tractable set or
$\PlCSP^2(\mathcal{F})$
is  $\#${\rm P-}hard.
We will use this lemma for the case that  $\mathcal{F}$ contains
 at least one nonzero signature of odd arity.
The proof of this lemma also
 demonstrates in a simple setting the idea that will be used in the proof
of
Lemma~\ref{mixing-P-global}.

\begin{lemma}\label{[1,a]XXX}
Let $\omega\neq 0$ and
let $\mathcal{F}$ be a set of symmetric signatures containing
 $[1, \omega]\in\mathcal{F}$.
If $\mathcal{F}\nsubseteq\mathscr{P}$,
$\mathcal{F}\nsubseteq\mathscr{A}$,
$\mathcal{F}\nsubseteq\mathscr{A}^{\dagger}$,
 $\mathcal{F}\nsubseteq\widehat{\mathscr{M}}$,
and
  $\mathcal{F}\nsubseteq\widehat{\mathscr{M}}^{\dagger}$,
then {\rm Pl}-$\#${\rm CSP}$^2(\mathcal{F})$ is  $\#${\rm P-}hard.
\end{lemma}
\begin{proof}
Firstly, we have $E^k_{k}(\omega)=\partial^{k}_{[1, \omega]}(=_{2k})$
of arity $k$
on the LHS in $\PlCSP^2(\mathcal{F})$, for all $k \ge 1$.
By a holographic transformation using $T^{-1}$, where
${T}=\left[\begin{smallmatrix} 1 & 0 \\
0&\omega \end{smallmatrix}\right]$,
we have $(E^k_{k}(\omega)) T^{-1} = (=_k)$ on the LHS, and
\[
\PlCSP({T}\mathcal{F})\leq_T
\PlHolant(\mathcal{EQ}\cup \mathcal{EQ}_2 {T}^{-1} \mid {T}\mathcal{F})
\leq_T
 \PlCSP^2(\mathcal{F}),
\]
where $\mathcal{EQ}$ on LHS of the Holant instance comes from $E^k_k(\omega)$
in  the second step of the reduction.
If ${T}\mathcal{F}\nsubseteq\mathscr{P}$,
${T}\mathcal{F}\nsubseteq\mathscr{A}$
and
${T}\mathcal{F}\nsubseteq\widehat{\mathscr{M}}$,
then
 Pl-$\#${\rm CSP}$({T}\mathcal{F})$ is  $\#${\rm P-}hard
 by Theorem~\ref{pl-dicho-1}.
Thus $\PlCSP^2(\mathcal{F})$ is  $\#${\rm P-}hard.

Otherwise,
${T}\mathcal{F}\subseteq\mathscr{P}$,
${T}\mathcal{F}\subseteq\mathscr{A}$ or
${T}\mathcal{F}\subseteq\widehat{\mathscr{M}}$.
If ${T}\mathcal{F}\subseteq\mathscr{P}$, then $\mathcal{F}\subseteq\mathscr{P}$.
In the following, assume that ${T}\mathcal{F}\nsubseteq\mathscr{P}$,
then ${T}\mathcal{F}\subseteq\mathscr{A}$ or
${T}\mathcal{F}\subseteq\widehat{\mathscr{M}}$.

Note that $[1, \omega^2]\in {T}\mathcal{F}$.
If  $\omega^8\neq 1$, then $[1, \omega^2]\notin\mathscr{A}\cup\widehat{\mathscr{M}}$.
%By Theorem~\ref{pl-dicho-1}, Pl-$\#${\rm CSP}$({T}\mathcal{F})$ is  $\#${\rm P-}hard.
This is a contradiction.

If $\omega^4=-1$, then $[1, \omega^2]\notin\widehat{\mathscr{M}}$.
Thus ${T}\mathcal{F}\subseteq\mathscr{A}$.
It follows that $\mathcal{F}\subseteq\mathscr{A}^{\dagger}$.

For $\omega^4=1$, if ${T}\mathcal{F}\subseteq\mathscr{A}$, then $\mathcal{F}\subseteq\mathscr{A}$.
If ${T}\mathcal{F}\subseteq\widehat{\mathscr{M}}$,
then
either  $\mathcal{F}\subseteq\widehat{\mathscr{M}}$ if $\omega^2=1$,
or  $\mathcal{F}\subseteq\widehat{\mathscr{M}}^{\dagger}$ if $\omega^2=-1$.
%
% So $\mathcal{F}\subseteq\mathscr{A}^{\dagger}$.
%
%
% If ${T}\mathcal{F}\nsubseteq\mathscr{A}$,
%by Theorem~\ref{pl-dicho-1}, Pl-$\#${\rm CSP}$([1, \omega^2] , {T}\mathcal{F})$ is  $\#${\rm P-}hard.
%Thus $\PlCSP^2([1, \omega] , \mathcal{F})$ is  $\#${\rm P-}hard.
%
%For $\omega^4=1$, if ${T}\mathcal{F}\subseteq\mathscr{A}$, then $\mathcal{F}\subseteq\mathscr{A}$.
%If ${T}\mathcal{F}\subseteq\widehat{\mathscr{M}}$, then $\mathcal{F}\subseteq\widehat{\mathscr{M}}^{\dagger}$
%for $\omega^2=-1$ and $\mathcal{F}\subseteq\widehat{\mathscr{M}}$
%for $\omega^2=1$.
%If ${T}\mathcal{F}\nsubseteq\mathscr{A}$ and ${T}\mathcal{F}\nsubseteq\widehat{\mathscr{M}}$,
%by Theorem~\ref{pl-dicho-1}, Pl-$\#${\rm CSP}$([1, \omega^2] , {T}\mathcal{F})$ is  $\#${\rm P-}hard.
%Thus $\PlCSP^2([1, \omega] , \mathcal{F})$ is  $\#${\rm P-}hard.
\end{proof}

%\begin{corollary}%
%Let $\mathcal{F}$ be a signature set and all of the signatures in $\mathcal{F}$ have even arities.
%If $\mathcal{F}$ does not satisfy any one of the following conditions:
%\begin{itemize}
%\item $\mathcal{F}\subseteq\mathscr{P}$,
%\item $\mathcal{F}\subseteq\mathscr{A}$,
%\item $\mathcal{F}\subseteq\widehat{\mathscr{M}}$,
%\end{itemize}
%then $\PlCSP^2([1, 1]^{\otimes 2k} , \mathcal{F})$ is $\#${\rm P-}hard for any $k\geq 1$.
%\end{corollary}
%\begin{proof}
%For $k=1$, we are done by Lemma~\ref{mixing-P-global} and Theorem~\ref{pl-dicho-1}.%
%
%For $k\geq 2$, we have $\partial^{k-1}([1, 1]^{\otimes 1k})=2^n[1, 1]^{\otimes 2}$.
%Thus
%\[
%\PlCSP^2([1, 1]^{\otimes 2} , \mathcal{F})\leq_T \PlCSP^2([1, 1]^{\otimes 2k} , \mathcal{F}),
%\]
%and we are done by Lemma~\ref{[1,a]XXX} and Theorem~\ref{pl-dicho-1}.
%\end{proof}

Lemma~\ref{[1,a]XXX} allows us to
transfer the complexity question of $\PlCSP^{2}$
to that of  $\PlCSP$, to which we can apply the known
dichotomy (Theorem~\ref{pl-dicho-1}).
However it requires a unary signature.
We observe that
 if all signatures in $\mathcal{F}$ have even arities, then
there is no way to construct a unary in $\PlCSP^{2}(\mathcal{F})$.
In this case, we use the
next lemma, which is similar to Lemma~\ref{[1,a]XXX}.
It shows that if we have $[1, \omega]^{\otimes 2}$ with $\omega\neq 0$
in $\mathcal{F}$,  then we can still
transfer the question of $\PlCSP^{2}$
to that of  $\PlCSP$.
It is proved using a global simulation of  $\PlCSP$
by $\PlCSP^2$.

\begin{lemma}\label{mixing-P-global}
Let $\mathcal{F}$ be a set of signatures  of
even arities. Suppose $[1, \omega]^{\otimes 2}\in\mathcal{F}$ for some
$\omega\neq 0$.
If $\mathcal{F}\nsubseteq\mathscr{P}$,
$\mathcal{F}\nsubseteq\mathscr{A}$,
$\mathcal{F}\nsubseteq\mathscr{A}^{\dagger}$,
$\mathcal{F}\nsubseteq\widehat{\mathscr{M}}$
and
$\mathcal{F}\nsubseteq\widehat{\mathscr{M}}^{\dagger}$,
%
%If $\mathcal{F}$ does not satisfy any one of the following conditions:
%\begin{itemize}
%\item $\mathcal{F}\subseteq\mathscr{P}$,
%\item $\mathcal{F}\subseteq\mathscr{A}$,
%\item $\mathcal{F}\subseteq\mathscr{A}^{\dagger}$,
%\item $\mathcal{F}\subseteq\widehat{\mathscr{M}}$,
%\item $\mathcal{F}\subseteq\widehat{\mathscr{M}}^{\dagger}$.
%\end{itemize}
then $\PlCSP^2(\mathcal{F})$ is  $\#${\rm P-}hard.
\end{lemma}
\begin{proof}
%For $k, \ell \geq 1$, $k\equiv\ell\pmod 2$,
%let
%$E_{\ell}^{k}$ be the signature of arity $\ell$ and
%$E_{\ell}^k=[1, 0, \ldots, 0, \omega^k]$,
%and $S_E=\{E_{\ell}^{k} \mid k, \ell \geq 1$, $k\equiv\ell\pmod 2\}$.
We first prove that $\PlHolant(E(\omega) \mid \mathcal{F})$
$\leq_T$ Pl-$\#${\rm CSP}$^{2}(\mathcal{F})$.

For $k \ge 1$ and $\ell\geq 0$,
we have all of $E^{2\ell}_{2k} = E^{2\ell}_{2k}(\omega)
=\partial_{[1, \omega]^{\otimes 2}}^{\ell}(=_{2k+2\ell})$ on LHS
in Pl-$\#${\rm CSP}$^{2}(\mathcal{F})$.
%(For the rest of this proof, we write $E^{\ell}_{k}$ for $E^{\ell}_{k}(\omega)$.)
Given any instance $\Omega$ of  $\PlHolant(E(\omega) \mid \mathcal{F})$,
since
all signatures in $\mathcal{F}$  have  even arities, the number of $E^{\ell}_k$
of odd arity must be even.
% in a instance of Pl-$\sharp$CSP$(f)$ is even.
In each connected component of $\Omega$,
we can connect all $E^{\ell}_k$ of odd arity  in  pairs,
by some copies of $[1, \omega]^{\otimes 2}$ in a planar way.
%This is illustrated in the following figures.
%\begin{center}
%\textcolor{red}{two figures:
%the first includes two odd $E_{\ell}^k$ and some even $E_{\ell}^k$.
%in the second figure, add one $[1, \omega]$ to the odd $E_{\ell}^k$ respectively
%and two $[1, \omega]$ to the even $E_{\ell}^k$ between the two odd $E_{\ell}^k$.
%Note that we should keep planarity.}
%\end{center}
Note that  when one input of $E^{\ell}_{k}$ is connected to
 a unary $[1, \omega]$, it  becomes $E^{\ell+1}_{k-1}$.
Hence a pair $E^{2u+1}_{2v-1}$ and $E^{2u'+1}_{2v'-1}$ can be functionally
replaced by  a pair $E^{2u}_{2v}$ and $E^{2u'}_{2v'}$ that are  connected by
$[1, \omega]^{\otimes 2}$.

Formally, we may assume the plane graph $\Omega$
is connected,
since the Holant value on $\Omega$ is the product
over its connected components,
and  the number of $E_k^{\ell} \in E(\omega)$ of odd arity is even
in each connected component of $\Omega$.
We will connect pairs of $E_k^{\ell}$ of odd arity
by copies of $[1, \omega]^{\otimes 2}$ within each
connected component.
%%% want to have connectedness since only for connected graph, dual is defined

Let $T$ be a spanning tree of the dual graph of $\Omega$, and pick
any node as the root of $T$.
%%% JYC added . for use in pf of cor 1.1
For definiteness we pick the node of $T$ that corresponds to
the external face of $\Omega$ as root.
If on a leaf node of $T$, i.e., a face of $\Omega$, there are an even number
of $E_k^{\ell}$ of odd arity, we can connect them in pairs
within the face by copies of $[1, \omega]^{\otimes 2}$, maintaining planarity.
If there are an odd number of them,
% {\sc Equalities} of odd arity,
we can pick any one, and
 still connect the others in pairs
within the face by copies of $[1, \omega]^{\otimes 2}$, maintaining planarity.
%%% this leaf is not the root, in this case of odd number. so parent exists
On the edge connecting the leaf to its parent in the tree $T$,
the corresponding edge in $\Omega$ has an $E_s^t$
 in one of the two
incident nodes of $\Omega$.
If $s$  is  odd, we pick this $E_s^t$.
If $s$  is even, we pick the first $E_k^\ell$
of odd arity in clockwise order in the face of $\Omega$, which
is the leaf node in $T$, and connect it to that $E_s^t$
 by one copy of $[1, \omega]^{\otimes 2}$.
This effectively transforms  that $E_s^t$ to $E_{s-1}^{t+1}$  of
odd arity. We then delete the leaf node from $T$.

The proof is completed by induction. Note that at the root of $T$,
there must be an even number of $E^{\ell}_k$ of odd arity,
including those which have been transformed by its children in $T$.
Thus we can simulate the  $\PlHolant(E(\omega) \mid \mathcal{F})$
problem $\Omega$ by
$\PlCSP^{2}(\mathcal{F})$.

Note that
$E^k_k \in E(\omega)$, for all $k \ge 1$.
Thus we have
\[
 \PlHolant(E_1^1, E^2_2, \dotsc, E^k_k, \dotsc \mid \mathcal{F})
 \leq_T
 \PlHolant(E(\omega) \mid \mathcal{F})
 \leq_T
 \PlCSP^2(\mathcal{F}).
\]
Then by a holographic transformation using
${T}=\left[\begin{smallmatrix} 1 & 0 \\
0&\omega^{-1} \end{smallmatrix}\right]$,
we have
\[
 \PlCSP({T}^{-1}\mathcal{F})
 \equiv \PlHolant(E_1^1, E^2_2, \ldots, E^k_k, \ldots \mid \mathcal{F})
 \leq_T \PlCSP^{2}(\mathcal{F}).
\]
The rest of the proof is the same
as the proof of Lemma~\ref{[1,a]XXX}. We omit it here.
\end{proof}

The next lemma shows that when we obtain $[1, 0, 1]^{\otimes 2}$,
we can reduce a  Pl-$\#$CSP$^2$ problem to
a  Pl-$\#$CSP$^4$
problem,
when all signatures in $\mathcal{F}$ have arity divisible by 4.
% $\equiv 0\pmod 4$.
%to which we can apply Theorem~\ref{pl-dicho-1}.

\begin{lemma}\label{mixing-P-global-binary}
{\rm Pl}-$\#${\rm CSP}$^2(\mathcal{F})$ $\leq_T$ {\rm Pl}-$\#${\rm CSP}$^{4}(\mathcal{F}, [1, 0, 1]^{\otimes 2})$,
if all signatures in $\mathcal{F}$ have arity $\equiv 0\pmod 4$.
\end{lemma}
%%% I think this is true for even non-symmetric signature,
%% and also true for signature sets.  ---JYC
\begin{proof}
Let $\Omega$ be an instance  of $\PlCSP^2(\mathcal{F})$. Since
all signatures in $\mathcal{F}$ have arity $\equiv 0\pmod 4$, the number of {\sc Equalities}
of  arity $\equiv 2 \pmod 4$ must be even.
% in a instance of Pl-$\sharp$CSP$(f)$ is even.
We can connect in pairs all {\sc Equalities} of arity $\equiv 2\pmod 4$
by some copies of $[1, 0, 1]^{\otimes 2}$ maintaining planarity
 similarly as in the  proof of Lemma~\ref{mixing-P-global}.
When two inputs of $=_{m+2}$ are connected
 to $[1, 0, 1]$ it becomes $\partial(=_{m+2}) = (=_{m})$.
Hence a pair $=_{4k-2}$ and $=_{4\ell-2}$ can be functionally
replaced by  a pair $=_{4k}$ and $=_{4\ell}$ that are  connected by
$[1, 0, 1]^{\otimes 2}$.
The rest of the proof is the same as  in Lemma~\ref{mixing-P-global} and we omit it here.
\end{proof}

The next corollary is used in the proof of the No-Mixing theorems.
%% JYC: one or more , any specific lemmas and theorems?
We present it here since the proof uses a global simulation that is similar to Lemma~\ref{mixing-P-global-binary}.
%%% JYC: more specifically lemma 1.2? or 1.3?

\begin{figure}[t]
 \centering
 \begin{tikzpicture}[scale=\scale,transform shape,node distance=\nodeDist,semithick]
  \node[internal]  (0)                     {};
  \node[external]  (1) [above left  of= 0] {};
  \node[external]  (2) [below left  of= 0] {};
  \node[external]  (3) [left        of= 1] {};
  \node[external]  (4) [left        of= 2] {};
  \node[external]  (5) [right       of= 0] {};
  \node[internal]  (6) [right       of= 5] {};
  \node[external]  (7) [right       of= 6] {};
  \node[internal]  (8) [right       of= 7] {};
  \node[external]  (9) [right       of= 8] {};
  \node[internal] (10) [right       of= 9] {};
  \node[external] (11) [above right of=10] {};
  \node[external] (12) [below right of=10] {};
  \node[external] (13) [right       of=11] {};
  \node[external] (14) [right       of=12] {};
  \path (0) edge[out= 135, in=   0]  (3)
            edge[out=-135, in=   0]  (4)
            edge[out=  45, in= 135] node[square] (s1) {}  (6)
            edge[out= -45, in=-135] node[square] (s2) {}  (6)
        (6) edge[out=  45, in= 135] node[square] (s3) {}  (8)
            edge[out= -45, in=-135] node[square] (s4) {}  (8)
        (8) edge[out=  45, in= 135] node[square] (s5) {} (10)
            edge[out= -45, in=-135] node[square] (s6) {} (10)
       (10) edge[out=  45, in= 180] (13)
            edge[out= -45, in= 180] (14);
  \begin{pgfonlayer}{background}
   \node[draw=\borderColor,thick,rounded corners,densely dashed,fit = (s1) (s2),inner sep=6pt] {};
%    \node[draw=\borderColor,thick,rounded corners,densely dashed,fit = (s1) (s2),inner sep=3pt,transform shape=false] {};
   \node[draw=\borderColor,thick,rounded corners,densely dashed,fit = (s3) (s4),inner sep=6pt] {};
%    \node[draw=\borderColor,thick,rounded corners,densely dashed,fit = (s3) (s4),inner sep=3pt,transform shape=false] {};
   \node[draw=\borderColor,thick,rounded corners,densely dashed,fit = (s5) (s6),inner sep=6pt] {};
%    \node[draw=\borderColor,thick,rounded corners,densely dashed,fit = (s5) (s6),inner sep=3pt,transform shape=false] {};
   \node[draw=\borderColor,thick,rounded corners,fit = (1) (2) (11) (12)] {};
%    \node[draw=\borderColor,thick,rounded corners,fit = (1) (2) (11) (12),inner sep=0pt,transform shape=false] {};
  \end{pgfonlayer}
 \end{tikzpicture}
 \caption{Gadget used to obtain $=_4$.
          The circle vertices are assigned $\hat{f}$ and the dashed subgadgets are assigned $[1,0,1]^{\otimes 2}$ aligned horizontally
          so that it is equivalent to assigning $[1,0,1]$ to the square vertices.}
 \label{fig:2}
\end{figure}
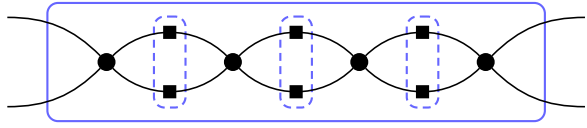

\begin{corollary} \label{2-M-double-roots}
 Suppose $f = [1, i]^{\otimes 4} + i^r [1, -i]^{\otimes 4}$ {\rm (}
$0 \le r \le 3$ {\rm )}
and $g = [g_0, \dotsc, g_{2n}]$ with $g_k = (\pm i)^{k} (2n-2k)$.
 Furthermore,
 let $\hat{g} = (Z^{-1})^{\otimes 2n} g$,
 where
 %$Z = \frac{1}{\sqrt{2}} \trans{1}{1}{i}{-i}$.
$Z = \trans{1}{1}{i}{-i}$.
 Then 
\[\PlCSP^2(\hat{g}) \leq_T \PlCSP^2(f, g).\]
\end{corollary}

\begin{proof}
 Clearly $\hat{g} = [0, 1, 0, \dotsc, 0]$ or $\hat{g} = [0, \dotsc, 0, 1, 0]$,
 the perfect matching signature or its reversal.
 By applying Lemma~\ref{construct-[1,0,1]-by-(1,i)-(1,-i)} to $f = [1, i]^{\otimes 4}+i^r[1, -i]^{\otimes 4}$,
 we get $[1, 0, -1]^{\otimes 2}$ on the left:
 %Thus
 \[
  \PlHolant([1, 0, -1]^{\otimes 2} \cup \mathcal{EQ}_2 \mid f, g)
  \leq_T
  \PlCSP^2(f, g).
 \]
 Under a holographic transformation by $Z$,
 we have
 \[
  \PlHolant([1, 0, 1]^{\otimes 2} \mid \hat{f}, \hat{g})
  \leq_T
  \PlHolant([1, 0, -1]^{\otimes 2} \cup \mathcal{EQ}_2 \mid f, g),
 \]
%
% JYC I corrected here. pl check.
%%%where $\hat{f}=(Z^{-1})^{\otimes 4}[1, 0, 0, 0, i^r]$.
 where $\hat{f} = (Z^{-1})^{\otimes 4} f = [1, 0, 0, 0, i^r]$.
Note that $[1,0,-1]Z^{\otimes 2} = 2 [1,0,1]$,
as $Z^{\rm T} \left[\begin{smallmatrix} 1 & 0 \\
0&-1 \end{smallmatrix}\right] Z = 
2
\left[\begin{smallmatrix} 1 & 0 \\
0& 1 \end{smallmatrix}\right]$.
 Consider the gadget in Figure~\ref{fig:2}.
 We assign $\hat{f}$ to the circle vertices and $[1,0,1]^{\otimes 2}$ the dashed subgadgets rotated appropriately so that it is equivalent to assigning $[1,0,1]$ to the square vertices.
 The signature of this gadget is $=_4$, for any $0 \le r \le 3$.
 Thus
 \[
  \PlHolant([1, 0, 1]^{\otimes 2} \mid  {=}_4, \hat{g})
  \leq_T
  \PlHolant([1, 0, 1]^{\otimes 2} \mid \hat{f}, \hat{g}).
 \]
 In $\PlHolant([1, 0, 1]^{\otimes 2} \mid {=}_4, \hat{g})$,
 by $[1, 0, 1]^{\otimes 2}$ and $=_4$,
 we can get all of $=_{4k}$ for $k\geq 1$ on RHS
%%% JYC need to draw picture here too.
 and then move them to LHS by $[1, 0, 1]^{\otimes 2}$.
 Moreover, we have $[1, 0, 1]^{\otimes 2}$ on RHS by connecting two copies of $=_4$
%%% JYC need to draw picture here too.
 by $[1, 0, 1]^{\otimes 2}$.
 Thus
 \[
  \PlHolant(\mathcal{EQ}_4 \mid [1, 0, 1]^{\otimes 2}, \hat{g})
  \leq_T
  \PlHolant([1, 0, 1]^{\otimes 2} \mid {=}_4, \hat{g}).
\]

Now we simulate $\PlCSP^2(\hat{g})$ by $\PlHolant(\mathcal{EQ}_4 \mid [1, 0, 1]^{\otimes 2}, \hat{g})$.
If $\hat{g}$ has arity $2n\equiv 0\pmod 4$, then we are done  by Lemma~\ref{mixing-P-global-binary}.

If $\hat{g}$ has arity $2n\equiv 2\pmod 4$, then
in an instance $\Omega$ of $\PlCSP^2(\hat{g})$,
the number of occurrences of  {\sc Equalities}
of  arity $\equiv 2 \pmod 4$ has the same parity
as the number of occurrences of $\hat{g}$,
which could be odd.
However, we observe that
all entries of signatures in $\PlCSP^2(\hat{g})$
  are nonnegative integers.
Thus the value of $\Omega$ is a nonnegative integer.
Let $\Omega  \uplus \Omega$ be the disjoint
union of two copies of $\Omega$ as a plane graph
with a common external face,
then the value of $\Omega  \uplus \Omega$ is the square of
the value of $\Omega$.
Thus computing the values on
 $\Omega  \uplus \Omega$ and $\Omega$ are equivalent.
In $\Omega  \uplus \Omega$, the number of {\sc Equalities}
of  arity $\equiv 2 \pmod 4$ is even.
Now we can use the same global simulation
as in Lemma~\ref{mixing-P-global-binary},
except that in the last step we may use one extra copy of
$[1,0,1]^{\otimes 2}$ to connect
two {\sc Equalities}
of  arity $\equiv 2 \pmod 4$ at the two root nodes
of the two spanning trees of the dual graphs of $\Omega$,
if the  number of occurrences of {\sc Equalities}
of  arity $\equiv 2 \pmod 4$ in $\Omega$ is odd.
 Thus we have
 \[
  \PlCSP^2(\hat{g})\leq \PlHolant(\mathcal{EQ}_4 \mid [1, 0, 1]^{\otimes 2}, \hat{g}).
  \qedhere
 \]
 \end{proof}

\section{Dichotomy Theorem when \texorpdfstring{$\mathcal{F}$}{F} Contains an Odd Arity Signature}\label{PartII.secC.Odd-arity}

In this section,
we give a dichotomy theorem for $\PlCSP^2(\mathcal{F})$,
where $\mathcal{F}$ includes at least one nonzero signature $f$ that has odd arity.

The next result is similar to Lemma~6.2 in~\cite{GW13}.

\begin{lemma}\label{odd-arity3-parity}
 Let $x, y \in \mathbb{C}$ and $f=[x,0,y,0]$.
 If $y \ne 0$ and $x^4 \ne y^4$,
 then
 {\rm Pl}-\#{\rm CSP}$^2(f)$ is \#P-hard.
\end{lemma}

\begin{proof}
 We reduce from $\PlCSP([x^2,y^2,y^2])$ to $\PlCSP^2(f)$.
 Since $\PlCSP([x^2,y^2,y^2])$ is $\numP$-hard when $y \ne 0$ and $x^4 \ne y^4$ by Theorem~\ref{thm:parII:k-reg_homomorphism},
 this shows that  $\PlCSP^2(f)$ is also $\numP$-hard.

 An instance of $\PlCSP([x^2,y^2,y^2])$ is a signature grid $\Omega$ with underlying graph $G = (U, V, E)$,
where
  $G$ is bipartite and planar, and every vertex in $U$ has degree~$2$.
 We replace every vertex in $V$ of degree $k$ (which is assigned ${=}_k \in \EQ$) with a vertex of degree $2k$,
 and bundle two adjacent variables to form $k$ bundles of~$2$ edges each.
 The $k$ bundles correspond to the $k$ incident edges of the original vertex with degree $k$.
 We assign $=_{2k}$ to the new vertices of degree $2k$.

 If the inputs to these equality signatures are restricted to $\{(0,0), (1,1)\}$ on each bundle,
 then these equality signatures take value~$1$ on $((0,0), \dotsc, (0,0))$ and $((1,1), \dotsc, (1,1))$ and take value~$0$ elsewhere.
 Thus, if we restrict the domain to $\{(0,0), (1,1)\}$, it is the equality signature $=_k$.

\begin{figure}[t]
 \centering
 \begin{tikzpicture}[scale=\scale,transform shape,node distance=\nodeDist,semithick,font=\LARGE]
  \node[external]  (0) [            label=left:$a_1$]  {};
  \node[external]  (1) [below of=0, label=left:$a_2$]  {};
  \node[internal]  (2) [right of=0]                    {};
  \node[internal]  (3) [right of=1]                    {};
  \node[external]  (4) [right of=2, label=right:$b_1$] {};
  \node[external]  (5) [right of=3, label=right:$b_2$] {};
  \path (0) edge node[very near end]   (n1) {} (2)
        (1) edge                               (3)
        (2) edge node[label=right:$c$]      {} (3)
            edge                               (4)
        (3) edge node[very near start] (n2) {} (5);
  \begin{pgfonlayer}{background}
   \node[draw=\borderColor,thick,rounded corners,fit = (n1) (n2),inner sep=6pt] {};
%    \node[draw=\borderColor,thick,rounded corners,fit = (n1) (n2),inner sep=3pt,transform shape=false] {};
  \end{pgfonlayer}
 \end{tikzpicture}
 \caption{Gadget designed for the paired domain. Both vertices are assigned $[x,0,y,0]$.}
 \label{fig:gadget:domain_pairing}
\end{figure}
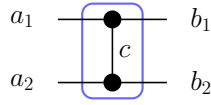

 To simulate $[x^2,y^2,y^2]$,
 we connect two copies of $f = [x,0,y,0]$ by a single edge as shown in Figure~\ref{fig:gadget:domain_pairing} to form a gadget with signature
 \[h(a_1, a_2, b_1, b_2) = \sum_{c=0,1} f(a_1, b_1, c) f(a_2, b_2, c).\]
 We replace every (degree~$2$) vertex in $U$ (which is assigned $[x^2,y^2,y^2]$) by a degree~$4$ vertex assigned $h$,
 where the variables of $h$ are bundled as $(a_1, a_2)$ and $(b_1, b_2)$.

 The vertices in this new graph $G'$ are connected as in the original graph $G$,
 except that every original edge is replaced by two edges that connect to the same side of the gadget in Figure~\ref{fig:gadget:domain_pairing}.
 Notice that $h$ is only connected by $(a_1, a_2)$ and $(b_1, b_2)$ to some bundle of two incident edges of an equality signature.
 Since this equality signature enforces that the value on each bundle is either $(0,0)$ or $(1,1)$,
 we only need to consider the restriction of $h$ to the domain $\{(0,0), (1,1)\}$.
 On this domain, $h = [x^2,y^2,y^2]$ is a \emph{symmetric} signature of arity~$2$.
 Therefore, the signature grid $\Omega'$ with underlying graph $G'$ has the same Holant value as the original signature grid $\Omega$.
\end{proof}

The following lemma is a dichotomy for $\PlCSP^2(f)$ where $f$ is a symmetric ternary signature.

\begin{lemma} \label{lem:PlCSP2:dichotomy:ternary}
 Let $f$ be a symmetric signature of arity~$3$,
 then $\PlCSP^2(f)$ is \#P-hard unless $f\in\mathscr{P}\cup\widetilde{\mathscr{A}}\cup\widetilde{\mathscr{M}}$.
\end{lemma}

\begin{proof}
 Let $f = [f_0, f_1, f_2, f_3]$.
 If $f$ satisfies parity constraints, then $f=[f_0, 0, f_2, 0]$ or $f=[0, f_1, 0, f_3]$.

 For $f=[f_0, 0, f_2, 0]$, if $f_2=0$, then $f\in\mathscr{P}$.
 If $f_0^2=f_2^2$, then $f\in\mathscr{A}$. If $f_0^2=-f_2^2$, then $f\in\mathscr{A}^{\dagger}$.
 Otherwise, we have $f_2\neq 0$ and $f_0^4\neq f_2^4$. Thus $\PlCSP^2(f)$ is \#P-hard by Lemma~\ref{odd-arity3-parity}.
 For $f=[0, f_1, 0, f_3]$, the proof follows
from a holographic transformation using
 $\left[
 \begin{smallmatrix}
 0 & 1\\
 1 & 0
 \end{smallmatrix}
 \right]$.

 In the following, assume that $f$ does not satisfy parity constraints.
 Firstly, we have $\partial(f)=[f_0+f_2, f_1+f_3]$.
\begin{itemize}
\item For $(f_0+f_2)(f_1+f_3)\neq 0$, we are done by Lemma~\ref{[1,a]XXX}.
\item For $f_0+f_2=f_1+f_3=0$, $f=[f_0, f_1, -f_0, -f_1]$.
Since $f$ does not satisfy parity constraints, we have $f_0f_1\neq 0$.
If $f_0^2=f_1^2$, then $f\in\mathscr{A}$.

Otherwise, we have $\partial_{f}(=_4)=[f_0, -f_1]$ on LHS and $\partial_{[f_0, -f_1]}(f)=[f_0^2-f_1^2, 2f_0f_1, f_1^2-f_0^2]$
on RHS. Moreover, we have $\partial_{[f_0^2-f_1^2, 2f_0f_1, f_1^2-f_0^2]}(=_4)=(f_0^2-f_1^2)[1, 0, -1]$ on LHS, where
$f_0^2-f_1^2\neq 0$. So we have $\partial_{[1, 0, -1]}(f)=2[f_0, f_1]$ on RHS.
Then we are done by Lemma~\ref{[1,a]XXX} and $f_0f_1\neq 0$.

\item For $f_0+f_2\neq 0$, $f_1+f_3=0$, we have
 $f_1=-f_3\neq 0$ since $f$ does not satisfy parity constraints.
Note that we have $\partial(f) = (f_0+f_2)[1, 0]$ in RHS, 
so we have $\partial^2_{[1, 0]}(f)=[f_0, f_1]$ in RHS.
If $f_0\neq 0$, then we are done by Lemma~\ref{[1,a]XXX}.
If $f_0=0$, then $f_2\neq 0$ since $f_0+f_2\neq 0$.
% $f$ does not satisfy parity constraints.
Note that we have $f_1[0, 1]$ and $f_2[1, 0]$ now.
Thus we have $\partial_{[1, 0]}[\partial_{[0, 1]}(f)]=[f_1, f_2]$.
Then we are done by Lemma~\ref{[1,a]XXX}.
\item For $f_0+f_2=0, f_1+f_3\neq 0$,
the proof follows from a holographic transformation using
 $\left[
 \begin{smallmatrix}
 0 & 1\\
 1 & 0
 \end{smallmatrix}
 \right]$.
\end{itemize}
\end{proof}

The next lemma shows that if we have an odd arity signature in
$\widetilde{\mathscr{M}} \setminus (\mathscr{P} \cup \widetilde{\mathscr{A}})$,
then we can prove Theorem~\ref{dichotomy-pl-csp2} directly.
The key point is that we can use such a signature to get a unary $[1, \omega]$ with $\omega \neq 0$.

\begin{lemma}\label{M-odd}
Let $\mathcal{F}$ be a symmetric
signature set and $f\in\mathcal{F}$ has odd arity.
\begin{itemize}
\item If $f\in\widehat{{\mathscr{M}}}\setminus(\mathscr{P}\cup\widetilde{\mathscr{A}})$,
then either $\mathcal{F}\subseteq\widehat{\mathscr{M}}$ or $\PlCSP^2(\mathcal{F})$ is $\#${\rm P-}hard.
\item If $f\in\widehat{\mathscr{M}}^\dagger\setminus(\mathscr{P}\cup\widetilde{\mathscr{A}})$,
then  either $\mathcal{F}\subseteq\widehat{\mathscr{M}}^\dagger$ or $\PlCSP^2(\mathcal{F})$ is $\#${\rm P-}hard.
\end{itemize}
\end{lemma}
\begin{proof}
We will use our calculus with the derivative operator $\partial$.
Firstly, we prove the lemma for $f\in\widehat{{\mathscr{M}}}\setminus(\mathscr{P}\cup\widetilde{\mathscr{A}})$.
We already have
$\mathcal{F} \nsubseteq\mathscr{P}$,
$\mathcal{F}\nsubseteq\mathscr{A}$,
$\mathcal{F}\nsubseteq\mathscr{A}^\dagger$ by the
presence of $f$,
and $\mathcal{F}\nsubseteq\widehat{\mathscr{M}}^\dagger$
 by Corollary~\ref{M-2-M-A-P}.
If we can construct
 a unary $[a, b]$ with $ab\neq 0$, then we can finish the
proof  by Lemma~\ref{[1,a]XXX}.

As $f \not \in \mathscr{P}$ and has odd arity,
its arity $n \ge 3$. By Lemma~\ref{M-2-M-NOT-IN-A-AND-P},
the signature $f \in \widehat{{\mathscr{M}}} \setminus(\mathscr{P}\cup\widetilde{\mathscr{A}})$ can take
one of the following two forms
(see the Calculus after Definition~\ref{deriviative}):
\begin{itemize}
\item
For $f=[s, t]^{\otimes n}\pm [t, s]^{\otimes n}$,
where $n \ge 3$ is odd, we have $st \not =0$ and $s^4 \not = t^4$.
% by
%Lemma~\ref{M-2-M-NOT-IN-A-AND-P}.
Thus we have $\partial^{\frac{n-1}{2}}(f)=
(s^2+t^2)^{\frac{n-1}{2}}([s,t] \pm [t, s])
=(s^2+t^2)^{\frac{n-1}{2}}(s\pm t)[1, \pm 1]$,
a nonzero multiple of $[1, \pm 1]$.
 So we are done by Lemma~\ref{[1,a]XXX}.

 \item For $f_k=\lambda(\pm 1)^k(n-2k)$, we
have $\partial^{\frac{n-1}{2}}(f)=  2^{\frac{n-1}{2}} \lambda [1, \mp 1]$
%\lambda[1, \pm 1]$
%%% JYC i think it is \mp. not \pm. Zhiguo , pl check.
%\lambda[1, \mp 1]$
and we are done by Lemma~\ref{[1,a]XXX}.
\end{itemize}

\noindent
Similarly, for $f\in\widehat{{\mathscr{M}}}^\dagger\setminus(\mathscr{P}\cup\widetilde{\mathscr{A}})$,
 we just need to construct a unary $[a, b]$ with $ab\neq 0$.
 %, then we finish the proof by Lemma~\ref{[1,a]XXX}.

\begin{itemize}
\item For $f=[s, ti]^{\otimes n}\pm [t, si]^{\otimes n}$,
we have $\partial^{\frac{n-1}{2}}(f)
= (s^2-t^2)^{\frac{n-1}{2}} [ s, ti] \pm
(t^2 - s^2)^{\frac{n-1}{2}} [t, si]
= (s^2-t^2)^{\frac{n-1}{2}} (s \pm t) [ 1, \pm i]$.
%%% these two \pm don't match up... not always the same
%%% but doesn't matter
 %(s^2-t^2)^{\frac{n-1}{2}}(s\pm(-1)^{\frac{n-1}{2}} t)[1, \pm i]$.
%%% JYC I think inside [1 \pm i] has a further factor ((-1)^{{n-1}/2})
%%% pl check
By Lemma~\ref{M-2-M-NOT-IN-A-AND-P}, we have
$st \not =0$ and
$s^4\neq t^4$,  and so
this is a nonzero multiple of $[1, \pm i]$.
%$s^2-t^2\neq 0$, $s\pm (-1)^{\frac{n-1}{2}}t\neq 0$ and
So we are done by Lemma~\ref{[1,a]XXX}.

\item For $f_k=\lambda(\pm i)^k(n-2k)$,
if $n\equiv 1\pmod 4$,  we have $\partial^{\frac{n-1}{4}}_{=_4}(f)
=  2^{\frac{n-1}{4}}\lambda[1, \mp i]$
%%% I think it is mp not pm. pl check. JYC
and we are done by Lemma~\ref{[1,a]XXX}.
If $n\equiv 3\pmod 4$, we have $\partial[\partial^{\frac{n-3}{4}}_{=_4}(f)]=
%\lambda 2^{\frac{n-3}{4}+2}[1, \pm i]$
2^{\frac{n+5}{4}} \lambda [1, \pm i]$
and we are done by Lemma~\ref{[1,a]XXX}.
\qedhere
\end{itemize}
\end{proof}

We remark that the use of $\partial_{=_4}$ instead of just
$\partial$ in this proof is necessary,
because $\partial^2(f) =0$ when $f_k=\lambda(\pm i)^k(n-2k)$, and $n \ge 5$.
One may also suppose that the case for 
$\widehat{{\mathscr{M}}}^\dagger\setminus(\mathscr{P}\cup\widetilde{\mathscr{A}})$
can be reduced to the case for 
$\widehat{{\mathscr{M}}}\setminus(\mathscr{P}\cup\widetilde{\mathscr{A}})$
by the transformation
$T = \left[
 \begin{smallmatrix}
 1 & 0\\
 0 & i
 \end{smallmatrix}
 \right]$.
While $T$ transforms $\widehat{{\mathscr{M}}}^\dagger$ to
 $\widehat{{\mathscr{M}}}$, and keeps
$\mathscr{P}$ and $\widetilde{\mathscr{A}}$ invariant,
this transformation does not keep $\mathcal{EQ}_2$ invariant.
In fact $[1,0,1]T^{\otimes 2} = [1,0,-1] \not \in \mathcal{EQ}_2$.
Therefore we need to handle the proof for
$\widehat{{\mathscr{M}}}^\dagger \setminus(\mathscr{P}\cup\widetilde{\mathscr{A}})$ separately.

By definitions of $\mathscr{P}$ and $\widetilde{\mathscr{A}}$,
%as summarized in Figure~\ref{fig:venn_diagram:A_Adagger_P},
we have the following simple lemma.

\begin{lemma} \label{widetilde-affine-and-P-parity}
 If $f\in\mathscr{P}\cup\widetilde{\mathscr{A}}$,
 then $f$ satisfies parity constraints iff
$f$ belongs to the following set, up to a scalar factor
%%% JYC: technically here the factor can be 0
\[\left\{[1, 0]^{\otimes n}, [0, 1]^{\otimes n},
 [1, 0]^{\otimes 2n}+t[0, 1]^{\otimes 2n},
 [1, \rho]^{\otimes n}\pm[1, -\rho]^{\otimes n},
 [1, \alpha]^{\otimes n}\pm[1, -\alpha]^{\otimes n}~\mid~t\neq 0, n\geq 1\right\}.\]
\end{lemma}

The next lemma shows that if we have a nonzero odd arity signature
$f\in\mathscr{P}\cup \widetilde{\mathscr{A}}$ that does not satisfy parity constraints,
then we can obtain a unary $[a, b]$ with $ab\neq 0$.
Note that if we have a unary $[a, b]$ with $ab\neq 0$, then we can 
apply Lemma~\ref{[1,a]XXX}.

\begin{lemma}\label{odd-nonzero-parity-unary-construction}
%Suppose that $f$ is nonzero and has odd arity $2n+1$.
%%% not satisfy parity ==> not identically 0
If $f\in\mathscr{P}\cup \widetilde{\mathscr{A}}$
has odd arity and does not satisfy parity constraints,
then we can
construct a unary $[a, b]$ with $ab\neq 0$ in  $\PlCSP^2(f)$.
\end{lemma}
\begin{proof}
Let $f$ have arity $2n+1$, $n\ge 0$. Not satisfying parity constraints
implies that $f$ is not identically 0.
Up to a nonzero factor, $f$ has the following forms.

If $f\in\mathscr{P}$, then
 $f=[a, b]^{\otimes 2n+1}$ with $ab\neq0$ or
 $f=[1, 0, \ldots, 0, x]$ with $x\neq 0$.
\begin{itemize}
\item If $f=[1, 0, \ldots, 0, x], x\neq 0$, then
$\partial^{n}(f)=[1, x]$.
\item If $f=[a, b]^{\otimes 2n+1}$, $a^2+b^2\neq 0$,
then $\partial^{n}(f)=(a^2+b^2)^{n}[a, b]$.
\item For $f=[1, \pm i]^{\otimes 2n+1}$,
% up to a scalar,
%if $2n+1\equiv 1\pmod 4$,
if $n$ is even, 
then $\partial^{\frac{n}{2}}_{=_4}(f)=2^{\frac{n}{2}}[1, \pm i]$.
%If $2n+1\equiv 3\pmod 4$,
If $n$ is odd, 
then we have $\partial_{f}(=_{2n+2})=[1, \mp i]$ on LHS
and we have $\partial^{2n}_{[1, \mp i]}(f)=2^{2n}[1, \pm i]$ on RHS.
\end{itemize}

For $f\in\widetilde{\mathscr{A}}\setminus\mathscr{P}$,
 we have $f=[1, \rho]^{\otimes 2n+1}\pm i[1, -\rho]^{\otimes 2n+1}$
or $f=[1, \alpha]^{\otimes 2n+1}\pm i[1, -\alpha]^{\otimes 2n+1}$.
\begin{itemize}
\item If $f=[1, \alpha]^{\otimes 2n+1}\pm i[1, -\alpha]^{\otimes 2n+1},$
then $\partial^{n}(f)=(1+\alpha^2)^{n}[1\pm i, (1\mp i)\alpha]$.
\item If
$f=[1, \rho]^{\otimes 2n+1}\pm i[1, -\rho]^{\otimes 2n+1}$ with $\rho^2=1$,
then $\partial^{n}(f)=(1+\rho^2)^{n}[1\pm i, (1\mp i)\rho]$.
\item For $f=[1, \rho]^{\otimes 2n+1}\pm i[1, -\rho]^{\otimes 2n+1}$
with $\rho^2=-1$,
and if $n$ is even, then 
we have $\partial^{\frac{n}{2}}_{=_4}(f)=2^{\frac{n}{2}}[1\pm i, (1\mp i)\rho]$.
If $n$ is odd, then $2n+1\equiv 3\pmod 4$,
and $(\pm \rho)^{2n+1} = \pm \rho^3 = \mp \rho$, by $\rho^2=-1$.
 Then we have $\partial_{f}(=_{2n+2})=
%%% 1st and last entry of f is 1 \pm i, \rho^3 \pm (-\rho^3) 
[1, \rho^{2n+1}]\pm i[1, (-\rho)^{2n+1}]
= [ 1, -\rho] \pm i[1, \rho] = (1\pm i)[1, \pm i\rho]$ on LHS.
%$[1, \rho^{2n+1}]\pm i[1, (-\rho)^{2n+1}]=[1\pm i, -\rho(1\mp i)]=(1\pm i)[1, \pm i\rho]$
%on LHS.
Note that $(\frac{1\mp i}{1\pm i})^{2n}=(\mp i)^{2n}=(-1)^n=-1$ since $n$ is 
odd.
 %by $2n+1\equiv 3\pmod 4$.

Then we have $\partial^{2n}_{[1, \pm i\rho]}(f)=(1\mp i)^{2n}[1, \rho]\pm i(1\pm i)^{2n}[1, -\rho]
=(1\pm i)^{2n}\{(\frac{1\mp i}{1\pm i})^{2n}[1, \rho]\pm i[1, -\rho]\}$$=(1\pm i)^{2n}\{-[1, \rho]\pm i[1, -\rho]\}
=(1\pm i)^{2n}[-1\pm i, -\rho(1\pm i)]$.
%\begin{itemize}
%\item If $r=1$ then we have $[1, i\rho]$ on LHS and
%$\partial^{2n}_{[1, i\rho]}(f)=(1-i)^{2n}[1-i, \rho(1+i)]$ on RHS;
%\item If $r=3$ then we have $[1, -i\rho]$ on LHS and
%$\partial^{2n}_{[1, -i\rho]}(f)=(1+i)^{2n}[1+i, \rho(1-i)]$ on RHS.
%\end{itemize}
\qedhere
\end{itemize}
\end{proof}

If a signature $f$ satisfies parity constraints, then
there is no way to construct $[a, b]$ with $ab\neq 0$ from $f$.
In fact in $\PlCSP^2$ using $f$, the signature of any $\{f\}$-gate will 
also satisfy the parity constraints, and in particular for unary
signature, it can only be a multiple of $[1, 0]$ or $[0, 1]$.
The next lemma shows that if we have a nonzero odd arity signature
$f\in\mathscr{P}\cup \widetilde{\mathscr{A}}$ that satisfies parity constraints,
then we can obtain $[1, 0]$ or $[0, 1]$.
We also remark that in $\PlCSP^2$ using signatures of even arity
one can only produce signatures of even arity, and thus no unary signatures.

\begin{lemma}\label{odd-[1,0]-[0,1]-construction}
If a nonzero $f\in\mathscr{P}\cup\widetilde{\mathscr{A}}$ 
has odd arity 
and satisfies parity constraints,
then we can
construct a unary $[1, 0]$ or $[0, 1]$ in  {\rm Pl}-$\#{\rm CSP}^{2}(f)$.
\end{lemma}

\begin{proof}
 By Lemma~\ref{widetilde-affine-and-P-parity},
 an nonzero $f$ of odd arity belongs to the following set, up to a nonzero
factor, 
 \[
  \left\{
         [1, 0]^{\otimes 2n+1},
         [0, 1]^{\otimes 2n+1},
         [1,   \rho]^{\otimes 2n+1} \pm [1, -  \rho]^{\otimes 2n+1}],
         [1, \alpha]^{\otimes 2n+1} \pm [1, -\alpha]^{\otimes 2n+1}]
~|~ n \ge 0
        \right\}.
 \]

For $f=[1, 0]^{\otimes 2n+1}$ or $f=[0, 1]^{\otimes 2n+1}$
we have $\partial^{n}(f)=[1, 0]$ or $[0, 1]$ respectively.

For $f=[1, \alpha]^{\otimes 2n+1}\pm [1, -\alpha]^{\otimes 2n+1},$
 $\partial^{n}(f)=(1+\alpha^2)^{n}[1\pm 1, (1\mp 1)\alpha]$, 
 a nonzero multiple of
% $2(1+\alpha^2)^n[1, 0]$ or $2\alpha(1+\alpha^2)^n[0, 1]$.
$[1, 0]$ or $[0, 1]$.

For $f=[1, \rho]^{\otimes 2n+1}\pm [1, -\rho]^{\otimes 2n+1}$ with $\rho^2=1$,
 $\partial^{n}(f)=(1+\rho^2)^{n}[1\pm 1, (1\mp 1)\rho]$, 
a nonzero multiple of $[1, 0]$ or $[0, 1]$.
% $2(1+\rho^2)^n[1, 0]$ or $2\rho(1+\rho^2)^n[0, 1]$.

For $f=[1, \rho]^{\otimes 2n+1}\pm [1, -\rho]^{\otimes 2n+1}$, with $\rho^2=-1$,
if $2n+1\equiv 1\pmod 4$, then $\partial^{\frac{n}{2}}_{=_4}(f)=2^{\frac{n}{2}}[1\pm 1, (1\mp 1)\rho]$,
a nonzero multiple of $[1, 0]$ or $[0, 1]$.
% $2[1, 0]$ or $2\rho[0, 1]$.
%
If  $2n+1\equiv 3\pmod 4$,
then $(\pm \rho)^{2n+1} = \mp \rho$.
If we write $f = [f_0, f_1, \ldots, f_{2n+1}]$, then
exactly one of $f_0$ and $f_{2n+1}$ is nonzero.
We have the unary $u = \partial_{f}(=_{2n+2}) = [f_0, f_{2n+1}]$ in LHS,
a nonzero multiple of $[1, 0]$ or $[0, 1]$.
Then we get $\partial^{2n}_{u}(f)$ in RHS,
also a nonzero multiple of $[1, 0]$ or $[0, 1]$.
%
%then we have $\partial_{f}(=_{2n+2})=2[1, 0]$ on LHS
%and we have $\partial^{2n}_{[1, 0]}(f)=2[1, 0]$.
%If $f=[1, \rho]^{\otimes 2n+1}-[1, -\rho]^{\otimes 2n+1}$,
%then we have $\partial_{f}(=_{2n+2})=-2\rho[0, 1]$ on LHS
%and we have $\partial^{2n}_{[0, 1]}(f)=2\rho^{2n+1}[0, 1]$.
\end{proof}

%\begin{lemma}\label{odd-unary-construction}
%Suppose that $f$ is nonzero and has odd arity $2n+1$.
%If $f\in\mathscr{P}\cup \widetilde{\mathscr{A}}$, then we can
%construct a nonzero unary in the instance $\PlCSP^2(f)$.
%\end{lemma}
%
%By the proof of Lemma
%\ref{odd-unary-construction}, we  have the following corollary.
%\begin{coro}\label{odd-[1,a]-construction}
%Suppose that $f$ is nonzero and has odd arity $2n+1$.
%If $f\in\mathscr{P}\cup \widetilde{\mathscr{A}}$ and
%$f\notin\{[1, 0]^{\otimes 2n+1}, [0, 1]^{\otimes 2n+1},$
%$[1, \alpha]^{\otimes 2n+1}\pm[1, -\alpha]^{\otimes 2n+1}],$
%$[1, \rho]^{\otimes 2n+1}\pm[1, -\rho]^{\otimes 2n+1}]$,
%i.e., $f$ does not satisfy parity constraints,
%then we can
%construct a unary $[a, b]$ with $ab\neq 0$ in the instance $\PlCSP^2(f)$.
%\end{coro}

The next lemma shows that if we already have $[1, 0]$ or $[0, 1]$
%and also a signature $f$, whether it has odd or even arity,
and also a signature $f$ of \emph{any} arity
 that does not satisfy the parity constraints,
then we can construct a unary $[a, b]$ with $ab\neq 0$.

\begin{lemma}\label{odd-[a,b]-parity}
If $f$ does not satisfy the parity constraints,
then we can construct a unary $[a, b]$ with $ab\neq 0$ in $\PlCSP^2([1, 0], f)$ or
$\PlCSP^2([0, 1], f)$.
\end{lemma}
\begin{proof}
We prove the lemma for $\PlCSP^2([1, 0], f)$.
The proof for the other case follows from a holographic transformation by
$\left[\begin{smallmatrix}0 & 1\\
1 & 0\end{smallmatrix}\right]$.

%Suppose  that $f$ has arity $n$ and
Let $f=[f_0, f_1, \ldots, f_n]$.
Since $f$ does not satisfy the parity constraints,
there exist $0 \le i < j \le n$ such that
$[f_i, f_{i+1}, \ldots, f_{j-1}, f_j]=[f_i, 0, \ldots, 0, f_j]$,
where $f_if_j\neq 0$ and $j-i$ is odd.
We can get both $f'=\partial^{n-i}_{[1, 0]}=[f_0, f_1, \dots, f_i]$  and
$f''=\partial^{n-j}_{[1, 0]}=[f_0, f_1, \dots, f_j]$ on RHS.
Either $i$ or $j$ is odd.
And so we have either $=_{i+1}$ or $=_{j+1}$,
and we can get either $\partial_{f'}(=_{i+1})=[f_0, f_i]$ 
or $\partial_{f''}(=_{j+1})=[f_0, f_j]$ on LHS.
Without loss of generality, assume that we have $[f_0, f_i]$ on LHS.
 
%If $i$ is odd,
%then  we have $f'=\partial^{n-i}_{[1, 0]}=[f_0, f_1, \dots, f_i]$ on RHS and
 %we have $\partial_{f'}(=_{i+1})=[f_0, f_i]$ on LHS.
 %If $j$ is odd, then  we have $f''=\partial^{n-j}_{[1, 0]}=[f_0, f_1, \dots, f_j]$ on RHS and
 %we have $\partial_{f''}(=_{j+1})=[f_0, f_j]$ on LHS.
%Without loss of generality, assume that we have $[f_0, f_i]$ on LHS.

If $f_0=0$, then we have $[0, 1]$ on LHS 
and $f'''=\partial^{\frac{j-i-1}{2}}(\partial^{i}_{[0,1]}(f''))
= \partial^{\frac{j-i-1}{2}}([f_i, 0, \ldots, 0, f_j])
=[f_i, f_j]$ on RHS, and we are done.
%Moreover,  we have $\partial^{\frac{j-1}{2}}(f''')=[f_i, f_j]$ and we are done.

If $f_0\neq 0$, let $m=\displaystyle\min_{1\leq k\leq n}\{k~\mid~f_k\neq 0\}$.
(As $j>0$ and $f_j \not =0$, this $m$ is well-defined.)
Then $f^{(4)}=\partial^{n-m}_{[1, 0]}(f)=[f_0, 0, \ldots, 0, f_m]$.
Moreover,  we have $\partial^{m-1}_{[f_0, f_i]}(f^{(4)})=[f^{m}_0, f^{m-1}_if_m]$.
\end{proof}

%Now it is ready to prove the No-Mixing lemma for $\mathcal{F}$.
The next lemma assumes the presence of a
non-degenerate binary {\sc Gen-Eq}. The conclusion is about
a transformed signature but still in the $\PlCSP^2$
setting.

%%%%%%%%%%%%%% Zhiguo, i did not continue here fully below.
%%%%%%%%%%%%% pl try to imitate like what I had above.
%%% pl note some changes are made below...
\begin{lemma}\label{mixing-P-binary}
%%%% Zhiguo, pl rewrite the statement like last lm.
For any $x\not = 0$ and any signature $f$ of arity $2n$,
let $\hat{f}=\left[\begin{smallmatrix} 1 & 0 \\
0 & x^{-\frac{1}{2}} \end{smallmatrix}\right]^{\otimes 2n}f$.
Then 
{\rm Pl}-$\#${\rm CSP}$^2(\hat{f})
\leq_T ${\rm Pl-}$\#${\rm CSP}$^{2}(f, [1, 0, x])$.
\end{lemma}
\begin{proof}
%$\PlCSP^2([1, 0, x^2], f)\equiv${\rm Pl-Holant}$(=_{2k} \mid [1, 0, x^2], f)$.
After a holographic transformation by $\left[\begin{smallmatrix} 1 & 0 \\
0 & x^{\frac{1}{2}} \end{smallmatrix}\right]$,
%$\begin{bmatrix}1&0\\0&x^{\frac{1}{2}}\end{bmatrix}$,
%%% changed x^{-1} to x.
we have
\[
 \PlCSP^2([1, 0, x], f)
 \equiv_T
 \PlHolant([1, 0, x], [1, 0, 0, 0, x^2], \cdots \mid [1, 0, 1], \hat{f}).
\]
If $x$ is a root of unity,
then there exists a $t \ge 1$ such that $x^t=1$.
Thus we have $=_{2kt}$ for all $k\geq 1$ on LHS.
Moreover, we have $=_{2k}$ by $\partial^{k(t-1)}(=_{2kt})$ on LHS
for all $k\geq 1$.
Thus we are done.

If $x$ is not a root of unity, then we have $\partial^{d-2}(E^{d}_{2d}(x))=[1, 0, 0, 0, x^{d}]$ of arity 4 on LHS
for all $d\geq 2$, where $E^{d}_{2d}(x)=[1, 0, \ldots, 0, x^d]$ has arity $2d$.
Thus we can get $[1, 0, 0, 0, 1]$ on LHS by interpolation.
Then we can get all of $=_{2k}$ on LHS since we have $[1, 0, 1]$ on RHS.
\end{proof}

\begin{lemma} \label{odd-binary-norm}
 Suppose either $f = [1, \rho]^{\otimes 3} \pm [1, -\rho]^{\otimes 3}$ or $f = [1, \alpha]^{\otimes 3} \pm [1, -\alpha]^{\otimes 3}$,
 and let $h = [1, 0, x]$.
 If $x^4 \not\in \{0,1\}$,
 then $\PlCSP^2(f, h)$ is $\#{\rm P}$-hard.
\end{lemma}

\begin{proof}
We prove the lemma for $f=[1, \rho]^{\otimes 3}\pm[1, -\rho]^{\otimes 3}$.
The proof for $f=[1, \alpha]^{\otimes 3}\pm[1, -\alpha]^{\otimes 3}$
is similar and we omit it here.

Let $\hat{f}=[1, x^{-\frac{1}{2}}\rho]^{\otimes 3}\pm[1, -x^{-\frac{1}{2}}\rho]^{\otimes 3}$,
then {\rm Pl}-$\#${\rm CSP}$^2(\hat{f})
\leq${\rm Pl-}$\#${\rm CSP}$^{2}(f, h)$
by Lemma~\ref{mixing-P-binary}.
%Note that $\hat{f}\in\langle 1, 0,  x^{-1}\rho^2\rangle$.
$\hat{f}$ satisfies a second order recurrence with
eigenvalues $\pm x^{-\frac{1}{2}}\rho$ with sum 0 and product
$- \rho^2/x$. Hence $\hat{f}$ has type $\langle - \rho^2/x, 0, 1 \rangle$.
Moreover, the  second recurrence relation is unique up to a scalar since
$\hat{f}$ is non-degenerate and has arity 3.
By $(x^{-1}\rho^2)^4\neq 1$,
we have
 $\hat{f}\notin\mathscr{P}\cup\widetilde{\mathscr{A}}\cup\widetilde{\mathscr{M}}$ by Lemma~\ref{second-recurrence-relation}.
So $\PlCSP^2(\hat{f})$ is $\#{\rm P}$-hard by Lemma~\ref{lem:PlCSP2:dichotomy:ternary}.
Thus $\PlCSP^2(f, h)$ is $\#{\rm P}$-hard.
\end{proof}

%\begin{lemma}\label{odd-binary-norm}
%Let $h=(a, 0, 0, b)$ or $(0, a, b, 0)$, where $ab\neq 0$ and $|a|\neq|b|$,
%and $f=[1, \alpha]^{\otimes 3}\pm[1, -\alpha]^{\otimes 3}$.
%Then $\PlCSP^2(f, h)$ is $\#{\rm P}$-hard.
%\end{lemma}
%\begin{proof}
%For $h=(a, 0, 0, b)$,
%We take 3 copies of $h$ and connect one input of each $h$ to an edge of $f$. The
%resulting signature is $f'=[a, b\alpha]^{\otimes 3}\pm[a, -b\alpha]^{\otimes 3}$.
%Note that $f'$ is non-degenerate and has arity 3, so $f'_k$ satisfy the unique second recurrence relation
%$\langle 1, 0,  -\frac{b^2}{a^2}\alpha^2\rangle$.
%Thus
% $f'\notin\mathscr{P}\cup\widetilde{\mathscr{A}}\cup\widetilde{\mathscr{A}}$ by Lemma~\ref{second-recurrence-relation}
%and $|a|\neq |b|$.
%So $\PlCSP^2(f')$ is $\#{\rm P}$-hard by Lemma~\ref{lem:PlCSP2:dichotomy:ternary}.
%Thus $\PlCSP^2(f, h)$ is $\#{\rm P}$-hard.%
%
%For $h=(0, a, b, 0)$,
%We take 3 copies of $h$ and connect one input of each $h$ to an edge of $f$. The
%resulting signature is $f=[a\alpha, b]^{\otimes 3}\pm[-a\alpha, b]^{\otimes 3}$.
%The rest of the proof is similar and we omit is here.
%\end{proof}

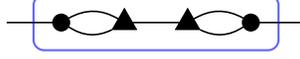
\begin{figure}[t]
 \centering
 \begin{tikzpicture}[scale=\scale,transform shape,node distance=\nodeDist,semithick]
  \node[external] (0)              {};
  \node[internal] (1) [right of=0] {};
  \node[triangle] (2) [right of=1] {};
  \node[triangle] (3) [right of=2] {};
  \node[internal] (4) [right of=3] {};
  \node[external] (5) [right of=4] {};
  \path (0) edge             (1)
        (1) edge[bend left]  (2)
            edge[bend right] (2)
        (2) edge             (3)
        (3) edge[bend left]  (4)
            edge[bend right] (4)
        (4) edge             (5);
  \begin{pgfonlayer}{background}
   \node[draw=\borderColor,thick,rounded corners,fit = (1) (4),inner sep=12pt] {};
%    \node[draw=\borderColor,thick,rounded corners,fit = (1) (4),inner sep=6pt,transform shape=false] {};
  \end{pgfonlayer}
 \end{tikzpicture}
 \caption{Gadget used to obtain a signature of the form $[a,0,b]$ with $|a| \ne |b|$.
 The circle vertices are assigned $f$ and the triangle vertices are assigned $g$.}
 \label{fig:3}
\end{figure}

\begin{lemma} \label{odd-arity-3-mixing}
 Let $f = [1, \rho]^{\otimes 3} \pm [1, -\rho]^{\otimes 3}$ and
 $g = [1, \alpha]^{\otimes 3} \pm [1, -\alpha]^{\otimes 3}$.
 Then $\PlCSP^2(f, g)$ is $\#{\rm P}$-hard.
\end{lemma}

\begin{proof}
 Consider the gadget in Figure~\ref{fig:3}.
 We assign $f$ to the circle vertices and $g$ to the triangle vertices.
 Let $h$ be the signature of this gadget.
 \begin{itemize}
  \item If
  $f = [1, \rho]^{\otimes 3}   + [1, -\rho]^{\otimes 3}$ and
  $g = [1, \alpha]^{\otimes 3} + [1, -\alpha]^{\otimes 3}$,
  then $h = 32 [\rho^2 \alpha^2, 0, -2]$.

  \item If
  $f = [1, \rho]^{\otimes 3}   - [1, -\rho]^{\otimes 3}$ and
  $g = [1, \alpha]^{\otimes 3} + [1, -\alpha]^{\otimes 3}$,
  then $h = 32 \rho^2 [-2, 0, \rho^2 \alpha^2]$.

  \item If
  $f = [1, \rho]^{\otimes 3}   + [1, -\rho]^{\otimes 3}$ and
  $g = [1, \alpha]^{\otimes 3} - [1, -\alpha]^{\otimes 3}$,
  then $h = 32 \alpha^2 [\rho^2 \alpha^2, 0, 2]$.

  \item If
  $f = [1, \rho]^{\otimes 3}   - [1, -\rho]^{\otimes 3}$ and
  $g = [1, \alpha]^{\otimes 3} - [1, -\alpha]^{\otimes 3}$,
  then $h = 32 \rho^2 \alpha^2 [2, 0, \rho^2 \alpha^2]$.
 \end{itemize}
Note that both $f$ and $g$ satisfy parity constraints,
and thus $h$ also satisfies that.
Hence, e.g., in the first case, $f=2[1,0,\rho^2,0]$
and $g=2[1,0, \alpha^2,0]$,  we only need to calculate
$h_0$ and $h_2$, since $h_1=0$ by parity.
In fact the left half of Figure~\ref{fig:3}, connecting $f$ to $g$,
also satisfies parity and has
the signature $4[1+ \rho^2 \alpha^2,0, 2\rho^2 \alpha^2]$,
and thus $h = 16[ (1+ \rho^2 \alpha^2)^2, 0, 4 (\rho^2 \alpha^2)^2]
= 32 [\rho^2 \alpha^2, 0, -2]$.

 Since $|\alpha \rho| = 1 \neq 2$,
 we are done by Lemma~\ref{odd-binary-norm}.
\end{proof}

Now we can prove a conditional No-Mixing theorem for $\PlCSP^2$
when a set of signatures $\mathcal{F}$ is assumed to consist
of only tractable signatures \emph{and} has 
at least one nonzero signature of odd arity.

\begin{theorem} \label{odd-mixing-theorem}
 Let $\mathcal{F} \subseteq \bigcup_{k=1}^5 S_k$ be a set of symmetric signatures that includes at least one nonzero signature of odd arity.
 If $\mathcal{F} \nsubseteq S_k$ for all $1 \leq k \leq 5$,
 then $\PlCSP^2(\mathcal{F})$ is $\#{\rm P}$-hard.
\end{theorem}

\begin{proof}
 If $\mathcal{F}$ contains a signature of odd arity in $\widetilde{\mathscr{M}} \setminus(\mathscr{P} \cup \widetilde{\mathscr{A}})$,
 then we are done by Lemma~\ref{M-odd}.
 Thus we can assume that $\mathcal{F}$ contains at least one
nonzero signature of odd arity 
$f \in \mathscr{P} \cup \widetilde{\mathscr{A}}$.
%Thus we can assume that $\mathcal{F}\in\mathscr{P}\cup\widetilde{\mathscr{A}}$ and includes at least one nonzero odd arity signature $f$.

By Lemma~\ref{odd-nonzero-parity-unary-construction}, if
$f$ does not satisfy parity constraints,
% $f\notin\{[1, 0]^{\otimes 2n+1}, [0, 1]^{\otimes 2n+1},$
%$[1, \rho]^{\otimes 2n+1}\pm[1, -\rho]^{\otimes 2n+1},$
%$[1, \alpha]^{\otimes 2n+1}\pm[1, -\alpha]^{\otimes 2n+1}]\}$,
then we have a unary $[a, b]$ with $ab\neq 0$ and we are done by Lemma~\ref{[1,a]XXX}.
Otherwise, we have $[1, 0]$ or $[0, 1]$ by Lemma~\ref{odd-[1,0]-[0,1]-construction}.
If there exists a signature in $\mathcal{F}$ that does not satisfy parity constraints, then we can
obtain a unary $[a, b]$ with $ab\neq 0$ by Lemma~\ref{odd-[a,b]-parity}.
Thus we are done by Lemma~\ref{[1,a]XXX}.

Now we can assume that $\mathcal{F}$ includes a nonzero odd arity signature $f \in \mathscr{P}\cup\widetilde{\mathscr{A}}$
and all signatures in $\mathcal{F}$ satisfy parity constraints.
Thus
$\mathcal{F} \bigcap \left(\widetilde{\mathscr{M}} \setminus (\mathscr{P} \cup \widetilde{\mathscr{A}})\right) = \emptyset$
by Corollary~\ref{widetilde-M-minus-P-and-A-parity}.
So we have $\mathcal{F} \subseteq \mathscr{P} \cup \widetilde{\mathscr{A}}$,
i.e., $\mathcal{F} \subseteq \bigcup_{k=3}^{5}S_k$.
Then by Lemma~\ref{widetilde-affine-and-P-parity},
we have,
up to scalar multiples,
\[
 \mathcal{F}
 \subseteq
 \left\{
  \begin{array}{ll}
   [1, 0]^{\otimes n}, \quad
   [0, 1]^{\otimes n}, \quad
   [1, 0]^{\otimes 2n} + t [0, 1]^{\otimes 2n},\\
   {[1, \rho]}^{\otimes n} \pm [1, -\rho]^{\otimes n}, \quad
   [1, \alpha]^{\otimes n} \pm [1, -\alpha]^{\otimes n}
  \end{array}
  \enspace \middle| \enspace
  t \neq 0 \text{ and } n \geq 1
 \right\}.
\]

Note that the following signatures are all in $\bigcap_{k=3}^{5}S_k$
(see Figure~\ref{fig:venn_diagram:A_Adagger_P}):
\begin{align*}
 &[1,0]^{\otimes n}
 \qquad \text{ and } \qquad
 [0,1]^{\otimes n},\\
 &[1,0]^{\otimes 2n} + t [0,1]^{\otimes 2n}
 \qquad \text{with} \qquad t^4=1,\\
 &[1, \rho]^{\otimes m} \pm [1, -\rho]^{\otimes m}
 \qquad \text{ and } \qquad
 [1, \alpha]^{\otimes \ell} \pm [1, -\alpha]^{\otimes \ell}
 \qquad \text{with} \qquad 1 \le m, \ell \leq 2.
\end{align*}
Let 
\[
 \mathcal{F'}
 =
 \mathcal{F}
 \; \bigcap \;
 \left\{
  \begin{array}{cc}
   [1,0]^{\otimes 2n} + t [0,1]^{\otimes 2n},\\
   {[1, \rho]}^{\otimes m} \pm [1, -\rho]^{\otimes m}, \quad
   [1, \alpha]^{\otimes \ell} \pm [1, -\alpha]^{\otimes \ell}
  \end{array}
  \enspace \middle| \enspace
  t^4 \not\in \{0,1\}
  \mbox{ and }
  m, \ell \geq 3
 \right\}.
\]
Then $\mathcal{F'}\nsubseteq S_k$ for $3 \leq k \leq 5$.
Indeed if $\mathcal{F'}\subseteq S_k$ for some $3\leq k\leq 5$,
then $\mathcal{F} \subseteq S_k$.
Let 
\[
 S
 =
 \mathcal{F'}
 \; \bigcap \;
 \left\{
  [1, \rho]^{\otimes m} \pm [1, -\rho]^{\otimes m}
  \mid
  m \geq 3
 \right\}
 \qquad \text{and} \qquad
 T
 =
 \mathcal{F'}
 \; \bigcap \;
 \left\{
  [1, \alpha]^{\otimes \ell} \pm [1, -\alpha]^{\otimes \ell}
  \mid
  \ell \geq 3
 \right\}.
\]

If $S \neq \emptyset$ and $T \neq \emptyset$,
then there exist $g, h \in \mathcal{F'}$ such that
$g = [1, \alpha]^{\otimes m} \pm [1, -\alpha]^{\otimes m}$
and
$h = [1, \rho]^{\otimes \ell} \pm [1, -\rho]^{\otimes \ell}$,
where $m, \ell \geq 3$.
By Lemma~\ref{odd-[1,0]-[0,1]-construction},
we can get $[1, 0]$ or $[0, 1]$ from $f$.
If we have $[1, 0]$,
then we have
$g' = \partial^{m-3}_{[1, 0]}(g)=[1, \alpha]^{\otimes 3}\pm[1, -\alpha]^{\otimes 3}$
and
$h' = \partial^{\ell-3}_{[1, 0]}(h)=[1, \rho]^{\otimes 3}\pm[1, -\rho]^{\otimes 3}$,
and are done by Lemma~\ref{odd-arity-3-mixing}.
If we have $[0, 1]$,
then the proof follows from a transformation by $\trans{0}{1}{1}{0}$.

If exactly one of $S$ and $T$ is not empty,
then there
 exists some $[1, 0]^{\otimes 2n}+t[0, 1]^{\otimes 2n}$ with $t^4 \notin \{0,1\}$ in $\mathcal{F'}$,
 since otherwise $\mathcal{F'}$ would be contained in either $\mathscr{A}$ or $\mathscr{A}^\dagger$.
This contradicts $\mathcal{F'} \nsubseteq S_k$ for $3\leq k\leq 5$.
 By taking $\partial^{n-1}$,
% doing $n-1$ loops to it, 
we have $[1, 0, t]$.
Moreover, we have
$g=[1, \alpha]^{\otimes m}\pm[1, -\alpha]^{\otimes m}$ or
$h=[1, \rho]^{\otimes \ell}\pm[1, -\rho]^{\otimes \ell}$ in $\mathcal{F'}$, where $m, \ell\geq 3$.
By a similar proof with the previous case, 
first getting $[0,1]$ or $[1,0]$ by Lemma~\ref{odd-[1,0]-[0,1]-construction},
 we can have
$g'=[1, \alpha]^{\otimes 3}\pm[1, -\alpha]^{\otimes 3}$ or
$h'=[1, \rho]^{\otimes 3}\pm[1, -\rho]^{\otimes 3}$ in $\mathcal{F'}$.
Thus $\PlCSP^2(\mathcal{F'})$ is \#P-hard by Lemma~\ref{odd-binary-norm}.
So $\PlCSP^2(\mathcal{F})$ is \#P-hard.

If $S=\emptyset$ and $T=\emptyset$,
then $\mathcal{F'} \subseteq \left\{[1, 0]^{\otimes 2n} + t [0, 1]^{\otimes 2n} \mid t^4 \notin \{0,1\}\right\} \subseteq \mathscr{P}$.
This contradicts that $\mathcal{F'}\nsubseteq S_k$ for $3\leq k\leq 5$.
\end{proof}

Now we can prove the dichotomy for $\PlCSP^2$
with a single symmetric signature of odd arity.

\begin{theorem}\label{odd-single}
If  $f$ is a symmetric signature of odd arity,
then either $\PlCSP^2(f)$ is $\#{\rm P}$-hard or
$f\in\mathscr{P}\cup\widetilde{\mathscr{A}}\cup\widetilde{\mathscr{M}}$.
\end{theorem}

\begin{proof}
Let $f$ have arity $2n+1$.
If $2n+1=1$, then $f\in\mathscr{P}$.
If $2n+1=3$, then
 we are done by Lemma~\ref{lem:PlCSP2:dichotomy:ternary}.
 In the following, assume that $2n+1\geq 5$.
 Let $f'=\partial(f)$.
If $f'\notin \mathscr{P}\cup \widetilde{\mathscr{A}}\cup \widetilde{\mathscr{M}}$,
then $\PlCSP^2(f')$ is $\#${\rm P-}hard by induction.
Thus $\PlCSP^2(f)$ is $\#${\rm P-}hard as well.
Otherwise,
$f'\in\mathscr{P}\cup \widetilde{\mathscr{A}}\cup\widetilde{\mathscr{M}}$.

If $f'\in\widetilde{\mathscr{M}}\setminus(\mathscr{P}\cup \widetilde{\mathscr{A}})$,
then we are done by Lemma~\ref{M-odd}.
So we can assume that $f'\in\mathscr{P}\cup \widetilde{\mathscr{A}}$.
Note that $f'$ has odd arity,
so if $f'$ 
%is not identically zero and 
%%% JYC: not parity ==> not 0
does not satisfy parity constraints,
then we have $[a, b]$ with $ab\neq 0$
by Lemma~\ref{odd-nonzero-parity-unary-construction} and we are done by Lemma~\ref{[1,a]XXX}.
Otherwise,
either $f'$ is identically zero or, as $f'$ has odd arity and satisfies parity,
by Lemma~\ref{widetilde-affine-and-P-parity}
\[f'\in\left\{[1, 0]^{\otimes 2n-1}, [0, 1]^{\otimes 2n-1},
[1, \rho]^{\otimes 2n-1}\pm[1, -\rho]^{\otimes 2n-1},
[1, \alpha]^{\otimes 2n-1}\pm[1, -\alpha]^{\otimes 2n-1}\right\}.\]
 %by Lemma~\ref{widetilde-affine-and-P-parity}.

%\begin{itemize}
%\item 
If $f'\equiv 0$, then $f=x[1, i]^{\otimes 2n+1}+y[1, -i]^{2n+1}$ by Lemma~\ref{general-f-construction}.
If $x=0$ or $y=0$ or $[xy\neq 0 \wedge x^4 = y^4]$, then $f\in\mathscr{A}$.
Otherwise, $xy\neq 0 \wedge x^4 \neq y^4$.
\begin{itemize}
\item For $2n+1\equiv 1\pmod 4$, we have
$\partial^{\frac{n}{2}}_{=_4}(f)=2^{\frac{n}{2}}\{x[1, i]+y[1, -i]\}=2^{\frac{n}{2}}[x+y, (x-y)i]$.
Note that $x+y\neq 0, x-y\neq 0$ by $x^4\neq y^4$.
Then we are done by Lemma~\ref{[1,a]XXX}.
\item For $2n+1\equiv 3\pmod 4$, we have
$f''=\partial^{\frac{n-1}{2}}_{=_4}(f)=2^{\frac{n-1}{2}}\{x[1, i]^{\otimes 3}+y[1, -i]^{\otimes 3}\}$.
Note that $xy\neq 0$ and $f$ is non-degenerate.
%%% then type is defined unique.
And by its second order recurrence,   $f\in \langle 1, 0,  1\rangle$.
%%% sum =0, prod =-1
it follows from Lemma~\ref{second-recurrence-relation} that 
 $f''\notin\mathscr{P}\cup\widetilde{\mathscr{A}}\cup\widetilde{\mathscr{M}}$
since $x^4\neq y^4$.
Thus  $\PlCSP^2(f'')$ is $\#${\rm P-}hard by Lemma~\ref{lem:PlCSP2:dichotomy:ternary}.
So  $\PlCSP^2(f)$ is $\#${\rm P-}hard.
\end{itemize}
%\end{itemize}

\vspace{.1in}
If $f'\in\{[1, 0]^{\otimes 2n-1}, [0, 1]^{\otimes 2n-1},$
$[1, \rho]^{\otimes 2n-1}\pm[1, -\rho]^{\otimes 2n-1}$,
$[1, \alpha]^{\otimes 2n-1}\pm[1, -\alpha]^{\otimes 2n-1} \}$,
then we have $[1, 0]$ or $[0, 1]$ by  Lemma~\ref{odd-[1,0]-[0,1]-construction}.
So if $f$ does not satisfy parity constraints, then we have $[a, b]$ with $ab\neq 0$ by  Lemma~\ref{odd-[a,b]-parity}
and we are done by
Lemma~\ref{[1,a]XXX}.
So we can assume that $f$ satisfies parity constraints in the following.

\begin{itemize}
\item For $f'=[1, 0]^{\otimes 2n-1}$, $f=x[1, i]^{\otimes 2n+1}+y[1, -i]^{2n+1}+[1, 0]^{\otimes 2n+1}$
by Lemma~\ref{general-f-construction}.
If $x=y=0$, then $f\in\mathscr{P}$.
Otherwise, $(x, y)\neq (0, 0)$.

Let $a=x+y$, $b=(x-y)i$, then $(a, b)\neq (0, 0)$.
Note that $f=[1+a, b, -a, -b, \ldots,\pm a,  \pm b]$.
Since $1+a$ and $-a$ cannot be both 0, by the parity constraints,
we have $b=0$.
And thus $a\neq 0$.
Moreover we have $\partial^{n-1}([1, 0]^{2n-1})=[1, 0]$ and $f'''=\partial^{2n-3}_{[1, 0]}(f)=[1+a, 0, -a, 0, a]$.
We note that $2n -3 \ge 1$ and so  $\partial^{2n-3}_{[1, 0]}$ is defined.
Note that $f'''$ is a redundant signature and its compressed signature matrix $\left[\begin{smallmatrix} 1+a & 0 & -a\\
0 & -a & 0\\
-a & 0 & a \end{smallmatrix}\right]$ is nonsingular,
so $\PlCSP^2(f''')$ is $\#{\rm P}$-hard by Lemma~\ref{4-redundant}.
Thus $\PlCSP^2(f)$ is $\#{\rm P}$-hard.

\item For $f'=[0, 1]^{\otimes 2n-1}$,
the proof follows from the previous case by a transformation using
%$f=x[1, i]^{\otimes 2n+1}+y[1, -i]^{2n+1}+[0, 1]^{\otimes 2n+1}$
%by Lemma~\ref{general-f-construction}.
%After the
%holographic transformation using
$\left[\begin{smallmatrix} 0 & 1\\
 1 & 0 \end{smallmatrix}\right]$.
 %we have
 %$\PlCSP^2(f, \partial(f))\equiv$ $\PlCSP^2(\widehat{f}, \widehat{\partial(f)})$,
 %where $\widehat{\partial(f)}=[1, 0]^{\otimes 2n-1}$,
 %and
 % $\widehat{f}=xi^{2n+1}[1, -i]^{\otimes 2n+1}+y(-i)^{2n+1}[1, i]^{2n+1}+[1, 0]^{\otimes 2n+1}$.
  %Thus we are done by the previous case.

\item For $f'=[1, \alpha]^{\otimes 2n-1}\pm[1, -\alpha]^{\otimes 2n-1}$,
$f=x[1, i]^{\otimes 2n+1}+y[1, -i]^{2n+1}+\frac{1}{1+\alpha^2}\{[1, \alpha]^{\otimes 2n+1}\pm[1, -\alpha]^{\otimes 2n+1}\}$
by Lemma~\ref{general-f-construction}.
If $x=y=0$, then $f\in\mathscr{A}^\dagger$.
Otherwise, $(x, y)\neq (0, 0)$.
%Let $a=x+y$, $b=(x-y)i$, then $(a, b)\neq 0$.
Firstly, we construct $[1, 0, \alpha^2]$ by $f$.
Note that we have $f^{(4)}=\partial^{n-1}(f)=(1+\alpha^2)^{n-2}\{[1, \alpha]^{\otimes 3}\pm[1, -\alpha]^{\otimes 3}\}$.

If $f^{(4)}=(1+\alpha^2)^{n-2}\{[1, \alpha]^{\otimes 3}+[1, -\alpha]^{\otimes 3}\}$ with a $+$ sign,
we have $\partial(f^{(4)})=2(1+\alpha^2)^{n-1}[1, 0]$
and $\partial_{[1, 0]}(f^{(4)})=2(1+\alpha^2)^{n-2}[1, 0, \alpha^2]$.

If $f^{(4)}=(1+\alpha^2)^{n-2}\{[1, \alpha]^{\otimes 3}-[1, -\alpha]^{\otimes 3}\}$ with a $-$ sign,
we have $\partial(f^{(4)})=2\alpha(1+\alpha^2)^{n-1}[0, 1]$
and $\partial_{[0, 1]}(f^{(4)})=2\alpha(1+\alpha^2)^{n-2}[1, 0, \alpha^2]$.

In either case, we have $[1, 0, \alpha^2]$.
Then we have $f^{(5)}=\partial^{n-1}_{[1, 0, \alpha^2]}(f)=(1-\alpha^2)^{n-1}\{x[1, i]^{\otimes 3}+y[1, -i]^{3}\}$.
If $x=0$ or $y=0$ or $[xy\neq 0 \wedge x^4=y^4]$, then $f^{(5)}\in\mathscr{A}\setminus\mathscr{A}^\dagger$.
By the eigenvalues, $f' \in \langle 1 ,0 \pm i\rangle$, hence 
 $f'\in\mathscr{A}^\dagger\setminus(\mathscr{P}\cup\mathscr{A}\cup\widetilde{\mathscr{M}})$ in this case.
So $\PlCSP^2(f^{(5)}, f')$ is $\#{\rm P}$-hard by Theorem~\ref{odd-mixing-theorem}.
Otherwise, $xy\neq 0$ and $x^4\neq y^4$. Then $f^{(5)}\notin\mathscr{P}\cup\widetilde{\mathscr{A}}\cup\widetilde{\mathscr{M}}$. Thus
$\PlCSP^2(f^{(5)})$ is $\#{\rm P}$-hard by Lemma~\ref{lem:PlCSP2:dichotomy:ternary}.
So $\PlCSP^2(f)$ is $\#{\rm P}$-hard.
\end{itemize}

The final case is $f'=[1, \rho]^{\otimes 2n-1}\pm[1, -\rho]^{\otimes 2n-1}$.

\begin{itemize}
\item
For $f'=[1, 1]^{\otimes 2n-1}+[1, -1]^{\otimes 2n-1}$,
\[f=x[1, i]^{\otimes 2n+1}+y[1, -i]^{2n+1} + \frac{1}{2} \left\{[1, 1]^{\otimes 2n+1}+[1, -1]^{\otimes 2n+1}\right\}.\]
If $x=y=0$, then $f\in\mathscr{A}$.
In the following, assume that $(x, y)\neq (0, 0)$.
Let $a=x+y$, $b=(x-y)i$, then $(a, b)\neq (0, 0)$.
Moreover, $f=
[a,b,-a,-b, \ldots, \pm b] + [1,0,1,0, \ldots, 0] =
[a+1, b, -a+1, -b, a+1, \ldots, \pm b]$.
Since $a+1$ and $-a+1$ cannot be both 0,
and $f$ satisfies parity, we have $b=0$. 
Thus $f=[a+1, 0, -a+1, 0, a+1, \ldots, \pm a+1, 0]$.
As $b=0$ we have $a \not =0$.
Note that we have $\partial^{n}(f)=2^{n}[1, 0]$.
Thus we have $f^{(6)}=\partial^{2n-3}_{[1, 0]}(f)=[a+1, 0, -a+1, 0, a+1]$.
The compressed signature matrix of $f^{(6)}$ is
$\left[
   \begin{smallmatrix}
    a+1 & 0 & -a+1 \\
    0 & -a+1 & 0 \\
    -a+1 & 0 & a+1
   \end{smallmatrix}
  \right]$ with determinant $4a(1-a)$.
If $a\neq 1$, then 
by $a \not =0$, this determinant is nonzero.
Thus the compressed signature matrix of $f^{(6)}$ is nonsingular and
 $\PlCSP^2(f^{(6)})$ is \numP-hard
 by Lemma~\ref{4-redundant}.
 So $\PlCSP^2(f)$ is \numP-hard.

\begin{figure}[t]
 \centering
 \begin{tikzpicture}[scale=\scale,transform shape,node distance=\nodeDist,semithick]
  \node[internal]  (0)                    {};
  \node[external]  (1) [above left  of=0] {};
  \node[external]  (2) [below left  of=0] {};
  \node[external]  (3) [left        of=1] {};
  \node[external]  (4) [left        of=2] {};
  \node[external]  (5) [right       of=0] {};
  \node[internal]  (6) [right       of=5] {};
  \node[external]  (7) [above right of=6] {};
  \node[external]  (8) [below right of=6] {};
  \node[external]  (9) [right       of=7] {};
  \node[external] (10) [right       of=8] {};
  \path (3) edge[in= 135, out=   0,postaction={decorate, decoration={
                                               markings,
                                               mark=at position 0.45 with {\arrow[>=diamond, white] {>}; },
                                               mark=at position 0.45 with {\arrow[>=open diamond]   {>}; } } }] (0)
        (0) edge[out=-135, in=   0]  (4)
            edge[bend left]          (6)
            edge                     (6)
            edge[bend right]         (6)
        (6) edge[out=  45, in= 180]  (9)
            edge[out= -45, in= 180] (10);
  \begin{pgfonlayer}{background}
   \node[draw=\borderColor,thick,rounded corners,fit = (1) (2) (7) (8),inner sep=0pt] {};
%    \node[draw=\borderColor,thick,rounded corners,fit = (1) (2) (7) (8),inner sep=0pt,transform shape=false] {};
  \end{pgfonlayer}
 \end{tikzpicture}
 \caption{Gadget used to obtain a signature whose signature matrix is redundant.}
 \label{fig:4}
\end{figure}
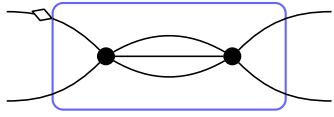

 If $a = 1$,
 then we have $f^{(7)} = \partial^{2n-4}_{[1, 0]}(f) = 2 [1,0,0,0,1,0]$
of arity 5 (note that $2n -4 \ge 0$).
 Consider the gadget in Figure~\ref{fig:4}.
 %We assign $f^{(7)}$ to both vertices.
We assign $[1,0,0,0,1,0]$ to both vertices.
 The signature of this gadget is redundant,
 and its compressed signature matrix is
 $\left[
   \begin{smallmatrix}
    1 & 0 & 0 \\
    0 & 1 & 0 \\
    0 & 0 & 3
   \end{smallmatrix}
  \right]$.
 Since this matrix is nonsingular,
 we are done by Lemma~\ref{4-redundant}.
%
%
%
%If $x=y=0$, then $f\in\mathscr{A}$.
%In the following, assume that $x\neq 0$ or $y\neq 0$.
%Note that $\partial(f)\in(\mathscr{A}\cap\widehat{\mathscr{M}})$ and
%$\partial(f)\notin(\mathscr{P}\cup\mathscr{A}^\dagger\cup\widehat{\mathscr{M}}^\dagger)$.
%By Lemma~\ref{odd-mixing} and induction, we are done by constructing a signature that
%is not in $\mathscr{A}\cup\widehat{\mathscr{M}}$ and has arity $<2n+1$.
%
%For $\partial(f)=[1, 1]^{\otimes 2n-1}+[1, -1]^{\otimes 2n-1}$.
%we have $\partial^{n}(f)=2^{n+1}[1, 0]$
%and $\partial^{2n-3}_{[1, 0]}(f)=x[1, i]^{\otimes 4}+y[1, -i]^{4}+\frac{1}{2}\{[1, 1]^{\otimes 4}\pm[1, -1]^{\otimes 4}\}$
%$=[a+1, b, -a+1, -b, a+1]$, where $a=x+y, b=(x-y)i$.
%Note that $[a+1, b, -a+1, -b, a+1]\notin\mathscr{A}\cup\widehat{\mathscr{M}}$ and we are done.
%
\item For $f'=[1, 1]^{\otimes 2n-1}-[1, -1]^{\otimes 2n-1}$,
\[f=x[1, i]^{\otimes 2n+1}+y[1, -i]^{2n+1} + \frac{1}{2} \left\{[1, 1]^{\otimes 2n+1}-[1, -1]^{\otimes 2n+1}\right\}.\]
After the
holographic transformation by $\left[\begin{smallmatrix} 0 & 1\\
 1 & 0 \end{smallmatrix}\right]$, we have
 $\PlCSP^2(f, f')\equiv$ $\PlCSP^2(\widehat{f}, \widehat{f'})$,
 where $\widehat{f'}=[1, 1]^{\otimes 2n-1}+[1, -1]^{\otimes 2n-1}$,
 and
  $\widehat{f}=xi^{2n+1}[1, -i]^{\otimes 2n+1}+y(-i)^{2n+1}[1, i]^{2n+1}+\frac{1}{2}\{[1, 1]^{\otimes 2n+1}+[1, -1]^{\otimes 2n+1}\}$.
  Thus we are done by the previous case.
%If $x=y=0$, then $f\in\mathscr{A}$.
%In the following, assume that $(x, y)\neq (0, 0)$.
%Let $a=x+y$, $b=(x-y)i$, then $(a, b)\neq (0, 0)$.
%Moreover, $f=[a, b+1, -a, -b+1, a, \ldots, \pm b+1]$.
%Note that $b+1\neq 0$
%or $-b+1\neq 0$. Thus if $a\neq 0$, then $f$ does not satisfy parity constraints.
%So we have $a=0$ and $f=[0, b+1, 0, -b+1, 0, \ldots, 0, \pm b+1]$.
%Note that we have $\partial^{n}(f)=2^{n}[0, 1]$.
%The rest of the proof is similar with the previous case and we omit it here.
%Thus we have $f''=\partial^{2n-3}_{[1, 0]}(f)=[a+1, 0, -a+1, 0, a+1]$.
%If $a\neq 1$, then the compressed signature matrix of $f''$ is nonsingular.
%Thus $\PlCSP^2(f'')$ is \numP-hard
% by Lemma~\ref{4-redundant}.
% So $\PlCSP^2(f)$ is \numP-hard.%
%
%If $a=1$, we have $f'''=\partial^{2n-4}_{[1, 0]}(f)=2[1, 0, 0, 0, 1, 0]$.
%By the following gadget
%\begin{center}
%two edges---f'''---three edges----f'''---two edges,
%\end{center}
%we get a redundant signature and its compressed signature matrix is
%$\left[\begin{smallmatrix} 1 & 0 & 0 \\
%0 & 1 & 0\\
%0 & 0 & 3\end{smallmatrix}\right]$ that is nonsingular. Then we are done by the similar statement.%
%
%the proof is similar.

\item For $f'=[1, i]^{\otimes 2n-1}+[1, -i]^{\otimes 2n-1}$,
$f$ has arity $2n+1$ and using
Proposition~\ref{prop:explicit_list}
(the Explicit List for $\int(f')$),
 $\int([1, \pm i]^{\otimes 2n-1})$ is 
a sum of   $\lambda [1, \pm i]^{\otimes 2n+1}$ with
a signature having the $k$-th term
of the form $-\frac{1}{2}k (\pm i)^k$. Thus, we can write 
$f_k=(x-\frac{1}{2}k)i^k+ (y-\frac{1}{2}k)(-i)^k$
 by Lemma~\ref{general-f-construction}.

We have $\partial_{f'}(=_{2n}) = \partial_{\partial(f)}(=_{2n})=2[1, 0]$ on LHS.

Let $a=x+y$, $b=(x-y)i$, then $f=[a, b, -a+2, -b, a-4, \ldots, \pm b]$.
Since $a$ and $-a+2$ cannot be both 0, and $f$ satisfies
parity, we have $b=0$.
Then we have $f^{(8)}=\partial^{2n-3}_{[1, 0]}(f)=[a, 0, -a+2, 0, a-4]$.
If $a\neq 2$, then the compressed signature matrix of $f^{(8)}$ 
is $\left[
     \begin{smallmatrix}
      a & 0 & -a+2 \\
      0 & -a+2 & 0  \\
      -a+2 & 0 & a-4
     \end{smallmatrix}
    \right]$,
and is nonsingular and we are done
%%% det = 4 (-a+2)
 by Lemma~\ref{4-redundant}.

   For $a = 2$,
   we have $\partial^{2n-4}_{[1, 0]}(f) = 2 [1,0,0,0,-1,0]$.
   Consider the gadget in Figure~\ref{fig:4}.
   We assign $[1,0,0,0,-1,0]$ to both vertices.
   The signature of this gadget is redundant,
   and its compressed signature matrix is
   $\left[
     \begin{smallmatrix}
      1 & 0 & 0 \\
      0 & 1 & 0 \\
      0 & 0 & 3
     \end{smallmatrix}
    \right]$.
   Since this matrix is nonsingular,
   we are done by Lemma~\ref{4-redundant}.
 % \end{itemize}

  \item For $f'=[1, i]^{\otimes 2n-1}-[1, -i]^{\otimes 2n-1}$,
  the proof follows from the previous case by a
  holographic transformation using $\left[\begin{smallmatrix} 0 & 1\\
  1 & 0 \end{smallmatrix}\right]$.
  \qedhere
 \end{itemize}
\end{proof}

By Theorem~\ref{odd-mixing-theorem} and Theorem~\ref{odd-single},
we have the following dichotomy theorem.

\begin{theorem}\label{odd-arity-dichotomy}
For any set of symmetric signatures $\mathcal{F}$
which contains at least one nonzero signature with odd arity,
if $\mathcal{F}\subseteq \mathscr{P}$, or
$\mathscr{A}$, or $\mathscr{A}^\dagger$,
or $\widehat{\mathscr{M}}$, or $\widehat{\mathscr{M}}^\dagger$,
then {\rm Pl}-$\#${\rm CSP}$^2(\mathcal{F})$ is tractable.
Otherwise, {\rm Pl}-$\#${\rm CSP}$^2(\mathcal{F})$ is $\#${\rm P-}hard.
\end{theorem}

\section{The Arity 4 Dichotomy} \label{Arity-4}

The goal of this section is a dichotomy theorem for $\PlCSP^2(f)$ when $f$ is a symmetric signature of arity~$4$.
Frequently our first test uses the determinantal criterion of a redundant signature of arity~$4$ based on Lemma~\ref{4-redundant}.

\begin{lemma} \label{redundant}
 Let $f$ be an arity~$4$ signature.
 If the signature matrix of $f$ is redundant,
 and its compressed form is nonsingular,
 then $\PlCSP^2(f)$ is \#P-hard.
\end{lemma}

\begin{proof}
 Since $\PlHolant(f) \leq_T \PlCSP^2(f)$,
 we are done by Lemma~\ref{4-redundant}.
\end{proof}

% Three Stooges superscripts
\newcommand{\timestimes}{\times\!\!\times}
\newcommand{\timesATOPtimes}{\begin{sideways}$\scriptstyle \timestimes$\end{sideways}}

Next we introduce a trick which we call the ``Three Stooges''.
For $f=[a,b,c,d,e]$,
define
\begin{align*}
 f^{\times} &= [a, c, e] \\
 f^{\timestimes} &= [a^2 +c^2 +2b^2, ac+ ce +2bd, c^2 +e^2 +2d^2],
 \qquad
 \text{and} \\
 f^{\timesATOPtimes} &= [a^2 +c^2 +2b^2, b^2 +d^2 +2c^2, c^2 +e^2 +2d^2].
\end{align*}

\begin{figure}[t]
 \centering
  \begin{tikzpicture}[scale=\scale,transform shape,node distance=\nodeDist,semithick]
   \node[internal] (0)                    {};
   \node[external] (1) [above left  of=0] {};
   \node[external] (2) [below left  of=0] {};
   \node[external] (3) [left        of=1] {};
   \node[external] (4) [left        of=2] {};
   \node[internal] (5) [right       of=0] {};
   \node[external] (6) [above right of=5] {};
   \node[external] (7) [below right of=5] {};
   \node[external] (8) [right       of=6] {};
   \node[external] (9) [right       of=7] {};
   \path (3) edge[in= 135, out=   0,postaction={decorate, decoration={
                                                markings,
                                                mark=at position 0.45 with {\arrow[>=diamond, white] {>}; },
                                                mark=at position 0.45 with {\arrow[>=open diamond]   {>}; } } }] (0)
         (0) edge[out=-135, in=   0] (4)
             edge[bend left]         (5)
             edge[bend right]        (5)
         (5) edge[out=  45, in= 180] (8)
             edge[out= -45, in= 180] (9);
   \begin{pgfonlayer}{background}
    \node[draw=\borderColor,thick,rounded corners,fit = (1) (2) (6) (7),inner sep=0pt] {};
%     \node[draw=\borderColor,thick,rounded corners,fit = (1) (2) (6) (7),inner sep=0pt,transform shape=false] {};
   \end{pgfonlayer}
  \end{tikzpicture}
 \caption{Gadget used in Lemma~\ref{domain-pairing-4}.
 Both vertices are assigned $f$.}
 \label{fig:domainpairing}
\end{figure}
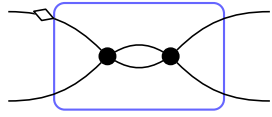

The following lemma is proved by the technique of domain pairing.

\begin{lemma} \label{domain-pairing-4}
 If $f=[a,b,c,d,e]$,
 then $\PlCSP(f^{\times}, f^{\timesATOPtimes}, f^{\timestimes}) \leq_T \PlCSP^2(f)$.
\end{lemma}

\begin{proof}
 Let $f'$ be the signature of the gadget in Figure~\ref{fig:domainpairing}
 and $f''$ be the signature of the gadget in Figure~\ref{fig:domainpairing} rotated 90$^{\circ}$.
 Then $f'$ has a signature matrix on the left,
 and $f''$ has a signature matrix on the right:
 \[
  \left[
  \begin{smallmatrix}
   a^2+c^2+2b^2 \enspace & \enspace ab+cd+2bc \enspace & \enspace ab+cd+2bc \enspace & \enspace ac+ce+2bd \\
   \\
   ab+cd+2bc             & b^2+d^2+2c^2                & b^2+d^2+2c^2                & bc+de+2cd \\
   \\
   ab+cd+2bc             & b^2+d^2+2c^2                & b^2+d^2+2c^2                & bc+de+2cd \\
   \\
   ac+ce+2bd             & bc+de+2cd                   & bc+de+2cd                   & c^2+e^2+2d^2
  \end{smallmatrix}
  \right];
  \enspace
  \left[
  \begin{smallmatrix}
   a^2+c^2+2b^2 \enspace & \enspace ab+cd+2bc \enspace & \enspace ab+cd+2bc \enspace & \enspace b^2+d^2+2c^2 \\
   \\
   ab+cd+2bc             & ac+ce+2bd                   & b^2+d^2+2c^2                & bc+de+2cd\\
   \\
   ab+cd+2bc             & b^2+d^2+2c^2                & ac+ce+2bd                   & bc+de+2cd\\
   \\
   b^2+d^2+2c^2          & bc+de+2cd                   & bc+de+2cd                   & c^2+e^2+2d^2
  \end{smallmatrix}
  \right].
 \]
 
 We highlight the relevant entries in the display below
 (in fact,
 readers should only focus on the entries highlighted;
 see Figure~\ref{fig:rotate_asymmetric_signature} in Part I for an illustration of the rotation operation):
 \[
  \left[
  \begin{smallmatrix}
   a^2+c^2+2b^2 \enspace & \enspace * \enspace & \enspace * \enspace & \enspace ac+ce+2bd \\
   \\
   *                     & b^2+d^2+2c^2        & *                   & *                  \\
   \\
   *                     & *                   & b^2+d^2+2c^2        & *                  \\
   \\
   ac+ce+2bd             & *                   & *                   & c^2+e^2+2d^2
  \end{smallmatrix}
  \right];
  \qquad
  \left[
  \begin{smallmatrix}
   a^2+c^2+2b^2 \enspace & \enspace * \enspace & \enspace * \enspace & \enspace b^2+d^2+2c^2 \\
   \\
   *                     & ac+ce+2bd           &          *          & *                     \\
   \\
   *                     &          *          & ac+ce+2bd           & *                     \\
   \\
   b^2+d^2+2c^2          &          *          &          *          & c^2+e^2+2d^2`
  \end{smallmatrix}
  \right].
 \]

 For any instance of $\PlCSP(f^{\times}, f^{\timesATOPtimes}, f^{\timestimes})$,
 we replace
%  note that this is a bipartite graph.
 each edge $e$ by two edges that connect the same incident nodes of $e$.
 For each variable node that is connected to $k$ edges,
 we replace its label $=_k$ by $=_{2k}$.
 We replace each occurrence of $f^{\times}, f^{\timesATOPtimes}, f^{\timestimes}$ by $f, f', f''$ as a constraint respectively.
 Then the new instance is a problem in $\PlCSP^2(f, f', f'')$
 and has the same value as the given instance of $\PlCSP(f^{\times}, f^{\timesATOPtimes}, f^{\timestimes})$.
 By $\PlCSP^2(f, f', f'')\equiv\PlCSP^2(f)$,
 we complete the proof.
\end{proof}

We demonstrate a simple use of the ``Three Stooges'' in the following lemma.

\begin{lemma} \label{rpower4neq1}
 If $a^4 \notin \{0,1\}$,
 then $\PlCSP^2([1, 0, a, 0, a^2])$ is \numP-hard.
\end{lemma}

\begin{proof}
 For $f= [1, 0, a, 0, a^2]$,
 we have $f^{\times} = [1, a, a^2]$
 and $f^{\timesATOPtimes} = [1+a^2, 2a^2, a^2(1+a^2)]$.
 By Lemma~\ref{binary},
 $f^{\times} \notin \mathscr{A} \cup \widehat{\mathscr{M}}$ since $a^4 \notin \{0,1\}$.
 By the same reason and Lemma~\ref{binary},
 the only possibility for $f^{\timesATOPtimes}\in\mathscr{P}$ is being degenerate.
 Thus $a^2 (1 + a^2)^2 = 4 a^4$.
 This implies that $a=0$ or $a = \pm 1$;
 a contradiction.
 This implies that $f^{\times}$ and $f^{\timesATOPtimes}$
 cannot both be in $\mathscr{P}, \mathscr{A}$, or $\widehat{\mathscr{M}}$.
 Thus $\PlCSP(f^{\times}, f^{\timesATOPtimes})$ is \#P-hard by Theorem~\ref{pl-dicho-1}.
 Then by Lemma~\ref{domain-pairing-4},
 $\PlCSP^2(f)$ is \#P-hard.
\end{proof}

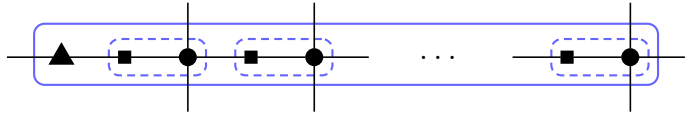
\begin{figure}[t]
 \centering
 \begin{tikzpicture}[scale=\scale,transform shape,node distance=\nodeDist,semithick]
  \node[external]  (0)              {};
  \node[triangle]  (1) [right of=0] {};
  \node[square]    (2) [right of=1] {};
  \node[internal]  (3) [right of=2] {};
  \node[external]  (4) [above of=3] {};
  \node[external]  (5) [below of=3] {};
  \node[square]    (6) [right of=3] {};
  \node[internal]  (7) [right of=6] {};
  \node[external]  (8) [above of=7] {};
  \node[external]  (9) [below of=7] {};
  \node[external] (10) [right of=7] {};
  \node[external] (11) [right of=10] {\huge $\dots$};
  \node[external] (12) [right of=11] {};
  \node[square]   (13) [right of=12] {};
  \node[internal] (14) [right of=13] {};
  \node[external] (15) [above of=14] {};
  \node[external] (16) [below of=14] {};
  \node[external] (17) [right of=14] {};
  \path (0) edge  (1)
        (1) edge  (2)
        (2) edge  (3)
        (3) edge  (4)
            edge  (5)
            edge  (6)
        (6) edge  (7)
        (7) edge  (8)
            edge  (9)
            edge (10)
       (12) edge (13)
       (13) edge (14)
       (14) edge (15)
            edge (16)
            edge (17);
  \begin{pgfonlayer}{background}
   \node[draw=\borderColor,thick,rounded corners,densely dashed,fit = (2) (3),inner sep=6pt] {};
%    \node[draw=\borderColor,thick,rounded corners,densely dashed,fit = (2) (3),inner sep=4pt,transform shape=false] {};
   \node[draw=\borderColor,thick,rounded corners,densely dashed,fit = (6) (7),inner sep=6pt] {};
%    \node[draw=\borderColor,thick,rounded corners,densely dashed,fit = (6) (7),inner sep=4pt,transform shape=false] {};
   \node[draw=\borderColor,thick,rounded corners,densely dashed,fit = (13) (14),inner sep=6pt] {};
%    \node[draw=\borderColor,thick,rounded corners,densely dashed,fit = (13) (14),inner sep=4pt,transform shape=false] {};
   \node[draw=\borderColor,thick,rounded corners,fit = (1) (14),inner sep=12pt] {};
%    \node[draw=\borderColor,thick,rounded corners,fit = (1) (14),inner sep=6pt,transform shape=false] {};
  \end{pgfonlayer}
 \end{tikzpicture}
 \caption{Gadget $\Gamma_k$, which has $k-1$ copies of the dashed box.
 Circle vertices are assigned $\hat{f}$,
 square vertices are assigned $=_2$,
 and the triangle vertex is assigned $[1,0,a]$.}
 \label{fig:5}
\end{figure}

\begin{lemma}\label{affine-with-wrong-scalar}
Let $f=[1, 1]^{\otimes 4}+a[1, -1]^{\otimes 4}$, where $a^4\neq 0, 1$.
Then $\PlCSP^2(f)$ is \#P-hard.
\end{lemma}
\begin{proof}
Under a holographic transformation by
$H=
\left[\begin{smallmatrix} 1 & 1 \\
1&-1 \end{smallmatrix}\right]$,
we have
\begin{align}\label{reduction-1}
\PlCSP^2(f)
&\equiv \PlHolant(\mathcal{EQ}_2 \mid f) \\
&\equiv  \PlHolant([1, 0, 1], [1, 0, 1, 0, 1], \ldots \mid \hat{f}),
\label{csp2-f-hat-form}
\end{align}
 where $\hat{f}= (H^{-1})^{\otimes 4} f = [1, 0, 0, 0, a]$.
%%% i use f-hat for Holo-transformed f.
%%%f' may confuse with \partial(f)
 By Lemma~\ref{rpower4neq1}, $\PlCSP^2([1, 0, a, 0, a^2])$ is \#P-hard,
and we have
\begin{align} \label{reduction-3}
\PlCSP^2([1, 0, a, 0, a^2])
&\equiv \PlHolant(\mathcal{EQ}_2 \mid [1,0,a,0,a^2])\\
&\equiv \PlHolant([1, 0, a], [1, 0, 0, 0, a^2], \cdots \mid [1, 0, 1, 0, 1])\\
&\leq \PlHolant([1, 0, a], [1, 0, 0, 0, a^2], \cdots \mid [1, 0, 1], [1, 0, 1, 0, 1], \cdots),
\label{intermediate-prob-lm-1.4}
\end{align}
where the second equivalence $\equiv$
 is by a holographic transformation with
$\left[\begin{smallmatrix} 1 & 0 \\
0& \sqrt{a} \end{smallmatrix}\right]$.

 The problem in~\eqref{csp2-f-hat-form} can simulate the problem in~\eqref{intermediate-prob-lm-1.4}.
 With $[1,0,1]$ on the left and $\hat{f}$ on the right in~\eqref{csp2-f-hat-form},
 we can get $\partial(\hat{f}) = [1, 0, a]$ on the right.
 Now consider the gadget in Figure~\ref{fig:5}.
 We assign $\hat{f}$ to the circle vertices,
 $=_2$ to the square vertices,
 and $[1,0,a]$ to the triangle vertex.
 If there are $k-1$ occurrences of the dashed subgadget,
 then the signature of this gadget is  $[1, 0, \dotsc, 0, a^{k}]$ of arity $2k$.
 Thus
 \begin{align*}
         &\PlHolant([1, 0, 1], [1, 0, 1, 0, 1], \cdots \mid [1, 0, a], [1, 0, 0, 0, a^2], \cdots)\\
  \leq {}&\PlHolant([1, 0, 1], [1, 0, 1, 0, 1], \cdots \mid \hat{f}).
 \end{align*}
Then combining three reductions,
we have $\PlCSP^2([1, 0, a, 0, a^2]) \leq \PlCSP^2(f)$,
where $a^4 \neq 0, 1$.
Thus $\PlCSP^2(f)$ is \#P-hard by Lemma~\ref{rpower4neq1}.
\end{proof}

Now we are ready to prove the following theorem.

\begin{theorem}\label{arity-4-dichotomy}
Let $f$ be a signature of arity 4, then $\PlCSP^2(f)$ is \numP-hard
or
$f\in\mathscr{P}\cup\widetilde{\mathscr{A}}\cup\widetilde{\mathscr{M}}$.
\end{theorem}

\begin{proof}
The first step is to apply Lemma~\ref{domain-pairing-4}
to $f^{\times}$.
For $f=[f_0, f_1, f_2, f_3, f_4]$ we have  $f^{\times}=[f_0, f_2, f_4]$.
If $\PlCSP(f^{\times})$ is \#P-hard, then $\PlCSP^2(f)$ is \#P-hard by Lemma~\ref{domain-pairing-4}.
In the following, assume that $\PlCSP(f^{\times})$ is 
not  \#P-hard, and hence tractable
by the dichotomy Theorem~\ref{thm:parII:k-reg_homomorphism}, and
$[f_0, f_2, f_4]$ takes the following form
\[
 [0, 0, 0],~
 [1, 0, 0],~
 [0, 0, 1],~
 [1, r, r^2],~
 [0, 1, 0],~
 [1, 0, a],~
 [1, 1, -1],~
 [1, -1, -1],
 ~~\mbox{or}~~
 [1, b, 1]
\]
up to a scalar,
where $r \neq 0$, $a \neq 0$, and $b^2 \notin \{0,1\}$.

\vspace{.2in}

\begin{description}
\item{Case 1: \texorpdfstring{$[f_0, f_2, f_4] = [0,0,0]$}{[f0,f2,f4]=[0,0,0]}}

In this case, $f=[0, x, 0, y, 0]$
and $f^{\timesATOPtimes}=[2x^2, x^2+y^2, 2y^2]$.
\begin{itemize}
 \item If $x^2 = y^2$,
 then $f=[0, x, 0, \pm x, 0]\in \mathscr{A}$.
 
 \item If $x^2 = -y^2$,
 then $f=[0, 1, 0, \pm i, 0] \in \mathscr{A}^{\dagger}$ since $\trans{1}{0}{0}{\sqrt{i}}^{\otimes 4}f\in\mathscr{A}$.
 
 \item If $x^4 \neq y^4$,
 then $\PlCSP(f^{\timesATOPtimes})$ is \#P-hard by Theorem~\ref{thm:parII:k-reg_homomorphism},
 so $\PlCSP^2(f)$ is \#P-hard by Lemma~\ref{domain-pairing-4}.
\end{itemize}

\vspace{.2in}
\item{Case 2: \texorpdfstring{$[f_0, f_2, f_4] = [1,0,0]$ or $[0,0,1]$}{[f0,f2,f4]=[1,0,0] or [0,0,1]}}

We prove the case for $[f_0, f_2, f_4]=[1, 0, 0]$ , i.e., $f=[1, x, 0, y, 0]$. The other case is similar.

Note that we have $\partial(f)=[1, x+y, 0]$.
If $x+y\neq 0$, then $\PlCSP^2([1, x+y, 0])$ is \#P-hard by Theorem~\ref{thm:parII:k-reg_homomorphism}. Thus $\PlCSP^2(f)$ is \#P-hard.

If $x=-y\neq 0$, then $\PlCSP^2(f)$  is \#P-hard by Lemma~\ref{redundant}.

If $x=-y=0$, then $f=[1, 0]^{\otimes 4}\in\mathscr{P}$.

\vspace{.2in}
\item{Case 3: \texorpdfstring{$[f_0, f_2, f_4] = [1,r,r^2]$ with $r \neq 0$}{[f0,f2,f4]=[1,r,r2] with r!=0}}

In this case, $f=[1, x, r, y, r^2]$.
If $rx\neq y$, then $\PlCSP^2(f)$  is \#P-hard by Lemma~\ref{redundant}.
Otherwise, $f=[1, x, r, xr, r^2]$. Then we have $\partial(f)=(1+r)[1, x, r]$.
If $r \not = -1$, then we have $[1, x, r]$. In the following we will
separate out the cases according to value of $r$.

\vspace{.1in}
\noindent
For $r^4\neq 1$ in $f = [1, x, r, xr, r^2]$.
\begin{itemize}
\item If $x=0$, then $f = [1, 0, r, 0, r^2]$, and  $\PlCSP^2(f)$ is \#P-hard by
Lemma~\ref{rpower4neq1}.
\item If $x^2=r$, then $f=[1, x]^{\otimes 4}\in\mathscr{P}$.
\item If $x^2\neq r$ and $x\neq 0$, then $\PlCSP^2([1, x, r])$  is \#P-hard by Theorem~\ref{thm:parII:k-reg_homomorphism}. Thus $\PlCSP^2(f)$ is \#P-hard.
\end{itemize}

\vspace{.1in}
\noindent
For $r=1$, then $f = [1, x, 1, x, 1]$.
\begin{itemize}
\item If $x^4=0$ or $1$, then $f\in\mathscr{A}$.
\item If $x^4\neq 0, 1$, then
let $a=\frac{1-x}{1+x}$ and we have $a^4\neq 0, 1$
by Lemma~\ref{pre-4-power}. Note that $f=\frac{1}{1+a}\left\{[1, 1]^{\otimes 4}+a[1, -1]^{\otimes 4}\right\}$.
By Lemma~\ref{affine-with-wrong-scalar}, Pl-$\#{\rm CSP}^{2}(f)$ is \#P-hard.
\end{itemize}

\vspace{.1in}
\noindent
For $r=-1$, then $f = [1, x, -1, -x, 1]$.
\begin{itemize}
\item If $x^4=0$ or $1$, then $f\in\mathscr{A}$.
\item If $x^4\neq 0, 1,$ then let
 $a=\frac{1+xi}{1-xi}$ and we have  $a^4\neq 0, 1$
by Lemma~\ref{pre-4-power}.
Note that $f=\frac{1}{a+1}\left\{
[1, i]^{\otimes 4}+ a[1, -i]^{\otimes 4} \right\}$.
Thus we have $[1, 0, -1]^{\otimes 2}$ on the left by Lemma~\ref{construct-[1,0,1]-by-(1,i)-(1,-i)}.
Under the holographic transformation by $\trans{1}{0}{0}{i}$,
this $[1, 0, -1]^{\otimes 2}$ is transformed to $[1, 0, 1]^{\otimes 2}$,
and we have
\begin{equation} \label{reduction-4}
 \PlCSP^{2}(f)
 \equiv_T \PlHolant(\mathcal{EQ}_4 \cup \{[1, 0, 1]^{\otimes 2}, [1, 0, -1], [1, 0, 0, 0, 0, 0, -1], \cdots\} \mid f'),
\end{equation}
where $f'=\frac{1}{1+a} \left\{[1, 1]^{\otimes 4}+a[1, -1]^{\otimes 4}\right \}$.
Now having $[1, 0, 1]^{\otimes 2}$ on the left, we can
form a pair of self loops in a planar way for a pair of adjacent $f'$
and get $(\partial(f'))^{\otimes 2}=
\left(\frac{2}{1+a}[1+a, 1-a, 1+a]\right)^{\otimes 2}$ on the right side.
%%% JYC pl check
Since we have $[1, 0, 1]^{\otimes 2}$ on the left side, we can obtain
$[1, 0, 1]^{\otimes 2}$ on the right side by interpolation using $[1+a, 1-a, 1+a]^{\otimes 2}$.
Note that the matrix 
$\left[\begin{smallmatrix} 1+a & 1-a \\
1-a & 1+a \end{smallmatrix}\right]$ 
can be diagonalized by 
$\left[\begin{smallmatrix} 1 & 1 \\
1 & -1 \end{smallmatrix}\right]$.
%has distinct nonzero eigenvalues,
%thus diagonalizable.
This implies that
\begin{multline*}
 \PlCSP^{4}(f', [1, 0, 1]^{\otimes 2})
 \leq_T\\
 \PlHolant(\mathcal{EQ}_4 \cup \{[1, 0, 1]^{\otimes 2}, [1, 0, -1], [1, 0, 0, 0, 0, 0, -1], \cdots\} \mid f').
\end{multline*}
Then by~\eqref{reduction-4} and Lemma~\ref{mixing-P-global-binary},
we have
\begin{equation}\label{reduction-6}
  \PlCSP^{2}(f')\leq \PlCSP^{4}(f', [1, 0, 1]^{\otimes 2})\leq
  \PlCSP^{2}(f).
\end{equation}
By Lemma~\ref{affine-with-wrong-scalar},
$\PlCSP(f')$ is \#P-hard.
Thus $\PlCSP(f)$ is \#P-hard.
\end{itemize}

\vspace{.1in}
\noindent
For $r^2=-1$, then $r = \pm i$ in $f= [1, x, r, xr, -1]$.
\begin{itemize}
\item If $x=0$, then $f=[1, 0, r, 0, -1]\in\mathscr{A}^{\dagger}$
since
$\left[\begin{smallmatrix} 1 & 0 \\
0& \sqrt{i} \end{smallmatrix}\right]^{\otimes 4}f
= [1,0,\pm 1, 0, 1] \in\mathscr{A}$.
\item If $x^2=r$, then $f=[1, x]^{\otimes 4}\in\mathscr{P}$.
\item If $x^2=-r$,
then $f=[1, x, -x^2, -x^3, -1]\in\mathscr{A}^{\dagger}$
since
$\left[\begin{smallmatrix} 1 & 0 \\
0& x^{-1} \end{smallmatrix}\right]^{\otimes 4}f\in\mathscr{A}$,
with $x^4 =-1$.
\item If $x^4\neq 0, -1$, thus $x^2 \not = \pm r$.
Then $\PlCSP^2([1, x, r])$  is \#P-hard by Theorem~\ref{thm:parII:k-reg_homomorphism}.
%This involves checking that $[1, x, r] \not \in
%\mathscr{P} \cup
%\mathscr{A}  \cup \mathscr{A}^{\dagger}
 %\cup \widehat{\mathscr{M}} \cup \widehat{\mathscr{M}}^{\dagger}$
%using Lemma~\ref{binary}.
%%% notes:
%%% \not \in P: all \not =s0 entries. and non-degen.
%%% \not \in A: as middle \not =0, only 1^2 = r^2 . no good.
%%% \not \in A^{\dagger}: with all non-0, can only be [1, al, -al^2] where
%%% al^4 = -1. translate to x^2 = \pm i. no good.  precisely we have x^4 not=-1
%%%  \not \in M or M^\dagger: as middle non-0, must a^2 = c^2. but our r^2 =-1
 Thus $\PlCSP^2(f)$ is \#P-hard.
\end{itemize}

%For $r=-i$,
%\begin{itemize}
%\item If $x=0$, then $f=[1, 0, -i, 0, -1]\in \mathscr{A}^{\dagger}$
%since $T^{\otimes 4}f\in\mathscr{A}$, where
%$T=\begin{bmatrix}
%1&0\\
%0&\sqrt{i}
%\end{bmatrix}$.
%\item If $x^2=-i$, then $f=[1, x]^{\otimes 4}\in\mathscr{P}$
%\item If $x^2=i$, then $f=[1, \sqrt{i}, -i, -i\sqrt{i}, -1]\in\mathscr{A}^{\dagger}$
%since $T^{\otimes 4}f\in\mathscr{A}$, where
%$T=\begin{bmatrix}
%1&0\\
%0&\sqrt{i}
%\end{bmatrix}$.
%\item If $x^4\neq 0, 1$, then $\PlCSP^2([1, x, r])$  is \#P-hard by Lemma~\ref{J-M-lemma}. Thus $\PlCSP^2(f)$ is \#P-hard.
%\end{itemize}

\vspace{.2in}
\item{Case 4: \texorpdfstring{$[f_0, f_2, f_4] = [0,1,0]$}{[f0,f2,f4]=[0,1,0]}}

In this case, $f=[0, x, 1, y, 0]$.
We first apply Lemma~\ref{redundant} and calculate the
determinant of the compressed matrix for $f$, which is $2xy-1$.
%a nonzero constant multiple of $2xy-1$.
%%% det = 2xy -1. just write [0 x 1\\ x 1 y \\ 1 y 0]
If $xy\neq \frac{1}{2}$, then Pl-$\#$CSP$^2(f)$
is \#P-hard by Lemma~\ref{redundant}.

If $xy=\frac{1}{2}$ and $x-y=0$,
then $f=[0, \frac{1}{\sqrt{2}}, 1, \frac{1}{\sqrt{2}}, 0]$
or $f=[0, -\frac{1}{\sqrt{2}}, 1, -\frac{1}{\sqrt{2}}, 0]$.
Both are in $\widehat{\mathscr{M}}$, by Lemma~\ref{mixing-M-arity-4}.

If $xy=\frac{1}{2}$ and $x+y=0$, then
$f=[0, \frac{i}{\sqrt{2}}, 1, -\frac{i}{\sqrt{2}}, 0]$
or $f=[0, -\frac{i}{\sqrt{2}}, 1, \frac{i}{\sqrt{2}}, 0]$.
Both are in $\widehat{\mathscr{M}}^{\dagger}$,
by Lemma~\ref{mixing-M-arity-4}.
In fact from the previous line 
with $[0, \pm\frac{1}{\sqrt{2}}, 1, \pm\frac{1}{\sqrt{2}}, 0]$,
 we can see directly
$\left[\begin{smallmatrix} 1 & 0 \\
0& i \end{smallmatrix}\right]^{\otimes 4}f\in\widehat{\mathscr{M}}$.

In the following we have $xy=\frac{1}{2}$ and $x^2\neq y^2$. Then
$f^{\timesATOPtimes}=[1+2x^2, 2+x^2+y^2, 1+2y^2]$
and $f^{\timestimes}=[1+2x^2, 1, 1+2y^2]$.
We will prove that $\PlCSP(f^{\timesATOPtimes}, f^{\timestimes})$ is \#P-hard by showing that $f^{\timesATOPtimes}, f^{\timestimes}$ cannot
be both in the same $\mathscr{P}, \mathscr{A}$, or $\widehat{\mathscr{M}}$.
\begin{itemize}
\item By $x^2\neq y^2$ and Lemma~\ref{binary}, we have $f^{\timestimes}\notin \widehat{\mathscr{M}}$.
\item Suppose $\{f^{\timesATOPtimes}, f^{\timestimes}\}
\subset \mathscr{P}$.
$f^{\timestimes}$ is not of the form
$[1, 0, a]$, and also not of the form
  $[0, 1, 0]$ since $1+2x^2\neq 1+2y^2$.
Thus $f^{\timestimes}$ is degenerate,
i.e., $(1+2x^2)(1+2y^2)=1$.
Note that $f^{\timesATOPtimes}$ is not of the form $[0, 1, 0]$ since $1+2x^2\neq 1+2y^2$.
If $f^{\timesATOPtimes}$ is of the form $[1, 0, a]$, then $x^2+y^2=-2$. Then together with $xy=\frac{1}{2}$
we obtain $(1+2x^2)(1+2y^2)=-2\neq 1$. This contradicts that $f^{\timestimes}$ is degenerate.
Thus $f^{\timesATOPtimes}$ and $f^{\timestimes}$ are both degenerate.
Then we have
\begin{align}
 (1+2x^2) (1+2y^2) &= (x^2+y^2+2)^2, \notag\\
 (1+2x^2) (1+2y^2) &= 1. \label{arity-4-[0,1,0]-2}
\end{align}
Together we have
$(x^2+y^2+2)^2 = 1$,
i.e., $x^2+y^2 = -3$ or $x^2+y^2 = -1$.
However both possibilities contradict~\eqref{arity-4-[0,1,0]-2} and $x y = \frac{1}{2}$.
Thus $f^{\timesATOPtimes}$ and $f^{\timestimes}$ cannot both belong to $\mathscr{P}$.
\item Suppose $\{f^{\timesATOPtimes}, f^{\timestimes}\} \subset \mathscr{A}$.
By $f^{\timestimes}\in\mathscr{A}$
and the middle term is nonzero, by
Corollary~\ref{binary-necessary} we have $1+2x^2=\pm (1+2y^2)$.
%%% I didn't go from Lemma\ref{binary}. but rather direct check on list of F123
Since $x^2\neq y^2$,
we have $1+2x^2 = -1-2y^2$.
This leads to $(x+y)^2=0$ by using $x y = \frac{1}{2}$.
This contradicts $x^2 \neq y^2$.
\end{itemize}
We have proved that
$f^{\timesATOPtimes}, f^{\timestimes}$ cannot be
both in $\mathscr{P}$, or $\mathscr{A}$, or $\widehat{\mathscr{M}}$.
Thus $\PlCSP(f^{\timesATOPtimes}, f^{\timestimes})$ is \#P-hard by Theorem~\ref{pl-dicho-1}.
So $\PlCSP^2(f)$ is \#P-hard by Lemma~\ref{domain-pairing-4}.

\vspace{.2in}
\item{Case 5: \texorpdfstring{$[f_0, f_2, f_4] = [1,0,a]$ with $a \neq 0$}{[f0,f2,f4]=[1,0,a] with a!=0}}

In this case, $f=[1, x, 0, y, a]$.
We first apply Lemma~\ref{redundant} and calculate the
determinant of the compressed matrix for $f$,
which is $-(ax^2 + y^2)$.
%a nonzero constant multiple of $ax^2 + y^2$.
If $ax^2+y^2\neq 0$, then $\PlCSP^2(f)$  is \#P-hard by Lemma~\ref{redundant}.
In the following we assume $ax^2+y^2=0$.

If $x=y=0$, then $f\in\mathscr{P}$.

If $x=y\neq 0$, then $a=-1$. So $f=[1, x, 0, x, -1]\in\widehat{\mathscr{M}}^{\dagger}$, by Corollary~\ref{mixing-M-arity-4}.

If $x=-y\neq 0$, then $a=-1$. So $f=[1, x, 0, -x, -1]\in\widehat{\mathscr{M}}$,
by Corollary~\ref{mixing-M-arity-4}.

Now we assume $ax^2+y^2=0$ and $x^2\neq y^2$.
Then $a\neq -1$ and $xy\neq 0$ by $a\neq 0$.
In this case,
the ``Three Stooges'' are
\[
 f^{\times} = [1, 0, a],
 \qquad
 f^{\timesATOPtimes} = [1+2x^2, x^2+y^2, a^2+2y^2],
 \qquad \text{and} \qquad
 f^{\timestimes}=[1+2x^2, 2xy, a^2+2y^2].
\]
By $ax^2+y^2=0$,
we have
\[
 f^{\timesATOPtimes} = [1+2x^2, (1-a)x^2, a^2-2ax^2]
 \qquad \text{and} \qquad
 f^{\timestimes}=[1+2x^2, 2xy, a^2-2ax^2].
\]
We will prove that $\PlCSP(f^{\times}, f^{\timesATOPtimes}, f^{\timestimes})$ is \#P-hard by showing that $f^{\times}, f^{\timesATOPtimes}$ and $f^{\timestimes}$ cannot be all in the same
$\mathscr{P}, \mathscr{A},$ or $\widehat{\mathscr{M}}$.
\begin{itemize}
\item Assume that $\{f^{\times}, f^{\timestimes}\}\subseteq \widehat{\mathscr{M}}$. Note that $a\neq -1$.
If $f^{\times}\in\widehat{\mathscr{M}}$, we have $a=1$ by Lemma~\ref{binary}.
Then by $f^{\timestimes}\in\widehat{\mathscr{M}}$ and Lemma~\ref{binary},
 we have $1+2x^2=1-2x^2$ or $2xy= 0$. This is a contradiction.
\item Assume that $\{f^{\timesATOPtimes}, f^{\timestimes}\}\subseteq \mathscr{P}$. If $1+2x^2$ and $a^2-2ax^2$ are both zero, then $a=0$ or $-1$.
This is a contradiction. Thus $f^{\timesATOPtimes}, f^{\timestimes}$
are not of the form $[0, 1, 0]$. By $xy\neq 0$, $f^{\timestimes}$
is not of the form $[1, 0, c]$ with $c \neq 0$.
Thus
$f^{\timestimes}$ is degenerate by Lemma~\ref{binary}, i.e.,
\begin{equation}\label{arity-4-[1,0,a]-prod}
(1+2x^2)(a^2-2ax^2)=4x^2y^2 = - 4 a x^4,
\end{equation}
where the last equality is by $ax^2 + y^2 =0$.

If $a=1$,
we have $1-4x^4=-4x^4$ by~\eqref{arity-4-[1,0,a]-prod}.
This is a contradiction.

If $a\neq 1$, then $f^{\timesATOPtimes}$ is not of the form $[1, 0, c]$ with $c\neq 0$.
Thus $f^{\timesATOPtimes}$ is degenerate by $f^{\timesATOPtimes}\in\mathscr{P}$,
i.e.,
\begin{equation*}
 (1+2x^2)(a^2-2ax^2)=(1-a)^2x^4.
 \end{equation*}
Then  by (\ref{arity-4-[1,0,a]-prod}), we have
$-4ax^4=(1-a)^2x^4$. This implies that $-4a=(1-a)^2$ by $x\neq 0$.
Then $(1 +a)^2 =0$, contradicting  $a\neq \pm 1$.
%Thus $\{f^{\timesATOPtimes}, f^{\timestimes}\}\subseteq\mathscr{P}$.
\item Suppose $\{f^{\times},
% f^{\timesATOPtimes},
f^{\timestimes}\}\subset \mathscr{A}$. By $f^{\times}\in\mathscr{A}$,
and $a \not =0$, 
%%% $a \not =0$ is from the heading of this Case. and is needed to be non-degen
we get $a^4=1$ from Lemma~\ref{binary}. It follows that $a=1$ or $a^2=-1$, as we have $a \not = -1$.

For $a=1$, the equation $ax^2 + y^2 =0$ gives us $y^2 = - x^2$. Then
from Corollary~\ref{binary-necessary} we have
\begin{equation*}
(1+2x^2)^2=(1-2x^2)^2
\end{equation*}
by $f^{\timestimes}\in \mathscr{A}$ and $2xy\neq 0$. Thus $x=0$.
%Thus $x^2=\pm\frac{i}{2}$ by $x\neq 0$.
%Then $f^{\timestimes}=(1+2x^2)[1, \frac{2xy}{1+2x^2}, \mp i]\notin \mathscr{A}$ by $2xy\neq 0$ and Lemma~\ref{binary}.
This is a contradiction.

For $a^2=-1$, by $f^{\timestimes}\in \mathscr{A}$
and $2xy\neq 0$, we have $(1+2x^2)^2=(-1-2ax^2)^2$ by Corollary~\ref{binary-necessary}.
% Lemma~\ref{binary}.
$1+2x^2 = 1 + 2ax^2$ leads to a contradiction $a=1$,
hence $1+2x^2 = -(1 + 2ax^2)$.
Then $x^2=-\frac{1}{a+1}$
and $f^{\timestimes}=[\frac{a-1}{a+1}, 2xy, \frac{a-1}{a+1}]$.
Note that $a+1\neq 0$.
We observe that the norm of $x^2$ is $\frac{1}{\sqrt{2}}$
and the norm of $x$ is equal to the norm of $y$ by $ax^2=-y^2$ and $a^2=-1$.
Thus the norm of $2xy$ is $\sqrt{2}$.
Moreover, the norm of $\frac{a-1}{a+1}$ is 1, as $a = \pm i$.
Thus the norm of $2xy$ is not equal to the norm of $\frac{a-1}{a+1}$,
and are nonzero.
So $f^{\timestimes}\notin\mathscr{A}$ by Corollary~\ref{binary-necessary}.
\end{itemize}
This implies that $f^{\times}, f^{\timesATOPtimes}$
and $f^{\timestimes}$ cannot be all in $\mathscr{P}$,
or all in $\mathscr{A}$,  or all in $\widehat{\mathscr{M}}$.
Thus the problem $\PlCSP(f^{\times}, f^{\timesATOPtimes}, f^{\timestimes})$ is \#P-hard by Theorem~\ref{pl-dicho-1}. So $\PlCSP^2(f)$ is \#P-hard.

\vspace{.2in}
\item{Case 6: \texorpdfstring{$[f_0, f_2, f_4] = [1, \pm 1, -1]$}{[f0,f2,f4]=[1,+-1,-1]}}

In this case, $f=[1, x, 1, y, -1]$ or $[1, x, -1, y, -1]$.
We consider the first case; the second case is similar.

We have $\partial(f)=[2, x+y, 0]$.
If $x+y\neq 0$, then $\PlCSP^2([2, x+y, 0])$
is \#P-hard by Theorem~\ref{thm:parII:k-reg_homomorphism}. Thus $\PlCSP^2(f)$ is \#P-hard.
Now we assume $x+y=0$.
Next we
apply Lemma~\ref{redundant} and calculate the
determinant of the compressed matrix for $f$,
which is a nonzero constant multiple of $x^2 + 1$.
If $x^2+1\neq 0$,
then $\PlCSP^2(f)$  is \#P-hard by Lemma~\ref{redundant}.

If $x+y=0$ and $x^2+1=0$, then $f=[1, \pm i, 1, \mp i, -1]$.
We have
\[
 \partial(f) = 2[1, 0, 0],
 \qquad
 \partial_{[1, 0, 0]}(f)=[1, \pm i, 1],
 \qquad \text{and} \qquad
 \partial_{[1, \pm i, 1]}(f) = [0, \pm 2i, 2].
\]
Then $\PlCSP^2([0, \pm 2i, 2])$ is \#P-hard by Theorem~\ref{thm:parII:k-reg_homomorphism}.
Thus $\PlCSP^2(f)$ is \#P-hard.

\vspace{.2in}
\item{Case 7: \texorpdfstring{$[f_0, f_2, f_4] = [1,b,1]$ with $b^2 \neq 0,1$}{[f0,f2,f4]=[1,b,1] with b2!=0,1}}

 In this last case of
Theorem~\ref{arity-4-dichotomy},
 $f=[1, x, b, y, 1]$ and the determinant of the compress signature matrix is
 \begin{equation}\label{deter-[1,b,1]}
D=b+2bxy-b^3-x^2-y^2.
 \end{equation}
If $D\neq 0$, then $\PlCSP^2(f)$  is \#P-hard by Lemma~\ref{redundant}.
In the following we assume that $D=0$.

 If $x=y=0$, then $b=0$ or $b^2=1$ by $D=b(1-b^2)=0$. This is a contradiction.

 If $x=y\neq 0$, then $D=(1-b)[b(1+b)-2x^2]=0$. By $b\neq 1$,
we have $b(1+b)=2x^2$.
 This
 implies that $f\in\widehat{\mathscr{M}}$ by Corollary~\ref{mixing-M-arity-4}.

  Similarly, if $x=-y\neq 0$, then $D=(1+b)[b(1-b)-2x^2]=0$.
  By $b\neq -1$, we have $b(1-b)=2x^2$.
 This
 implies that $f\in\widehat{\mathscr{M}}^{\dagger}$ by Corollary~\ref{mixing-M-arity-4}.

  In the following, assume that $x^2\neq y^2$ in addition to $D=0$.
In this case,
the ``Three Stooges'' are
\begin{align*}
 f^{\times} &= [1, b, 1],\\
 f^{\timesATOPtimes} &= [1+b^2+2x^2, 2b^2+x^2+y^2, 1+b^2+2y^2], \qquad \text{and}\\
 f^{\timestimes} &= [1+b^2+2x^2, 2b+2xy, 1+b^2+2y^2].
\end{align*}
We will prove that $\PlCSP(f^{\times}, f^{\timesATOPtimes}, f^{\timestimes})$ is \#P-hard by showing that $f^{\times}, f^{\timesATOPtimes}, f^{\timestimes}$ cannot all be in the same
$\mathscr{P}$, or $\mathscr{A}$, or $\widehat{\mathscr{M}}$.

By $b^2\neq 0, 1$, we have $f^{\times}\notin\mathscr{P}$
 by Lemma~\ref{binary}.
\begin{itemize}
\item Suppose $b^2\neq -1$. Then in addition to $b^2\neq 0,1$,
 we have $b^4 \not =  0,1$.
Then $f^{\times} \notin \mathscr{A}$ by  Lemma~\ref{binary}.
Moreover, if $f^{\timesATOPtimes}\in \widehat{\mathscr{M}}$,
then by Lemma~\ref{binary} and the fact that $x^2\neq y^2$, we must have
\begin{equation} \label{arity-4-[1,b,1]}
 1+b^2+2x^2 = -(1+b^2+2y^2)
 \qquad \text{and} \qquad
 2b^2+x^2+y^2 = 0.
\end{equation}
From~\eqref{arity-4-[1,b,1]},
we get $b^2 = 1$.
This is a contradiction.
This implies that $f^{\times}, f^{\timesATOPtimes}$ cannot be all in $\mathscr{P}$, or all in $\mathscr{A}$, or all in $\widehat{\mathscr{M}}$ when $b^2\neq -1$.
\item Now suppose
 $b^2=-1$. Then 
$f^{\timesATOPtimes}=[2x^2, x^2+y^2-2, 2y^2]$
and $f^{\timestimes}=2[x^2, b+xy, y^2]$.
If $f^{\timesATOPtimes}\in \widehat{\mathscr{M}}$,
then by $x^2\neq y^2$ and Lemma~\ref{binary},
we have
\[
 x^2 = -y^2
 \qquad \text{and} \qquad
 x^2+y^2-2=0
\]
This is a contradiction.

Finally suppose $\{f^{\timesATOPtimes}, f^{\timestimes}\}\subset\mathscr{A}$.
\begin{itemize}
\item If $x^2+y^2=0$, then $xy=-1$  by $b^2=-1$ and
%by $D=0$ and $b^2=-1$,
%\begin{eqnarray*}
\[D=b+2bxy-b^3-x^2-y^2=0.\]
%\end{eqnarray*}
Then $f^{\timesATOPtimes}=2[x^2, -1, y^2], f^{\timestimes}=2[x^2, b-1, y^2]$ both have all nonzero entries.
If they are both in $\mathscr{A}$, the norm of their entries
must be all the same $|b-1| = |x^2| = |-1| = 1$,
by Corollary~\ref{binary-necessary}.
However $b-1$ does not have norm 1 since $b^2=-1$.

\item If $x^2+y^2\neq 0$, then, since we also have $x^2\neq y^2$,
the first and the last entries of
both $f^{\timesATOPtimes}$ and $f^{\timestimes}$
are neither equal nor negative of each other.
  It follows from membership in $\mathscr{A}$
 that  $x^2+y^2-2=0$ and $b+xy=0$  by  Corollary~\ref{binary-necessary}.
% Lemma~\ref{binary}.
Then by $D=b+2bxy-b^3-x^2-y^2=0$ and $b^2=-1$,
we get a contradiction.
%%% my derivation seems easier:
%%% xy= -b , x^2+y^2 =2, and b^2 =-1, so D = b -2b^2 - b^3 -2 = b +2 +b -2 =2b
%%% D=0 --> b=0 contradiction.
%\begin{equation*}
%D=b+2bxy-b^3-x^2-y^2=0,
%\end{equation*}
%\begin{equation*}
%x^2+y^2-2=0,
%\end{equation*}
%and $b^2=-1$,
%we have $b+bxy-1=0$.
%Moreover, by
%\begin{equation*}
 %b+bxy-1=0,
 %\end{equation*}
%\begin{equation*}
 %b+xy=0,
 %\end{equation*}
  %we have $b-b^2-1=0$. This contradicts that $b^2=-1$.
  \end{itemize}
We have proved that $f^{\times}, f^{\timesATOPtimes}, f^{\timestimes}$
cannot be all in $\mathscr{P}$,
or all in $\mathscr{A}$,
or all in $\widehat{\mathscr{M}}$ when $b^2=-1$.
\end{itemize}
From above, $f^{\times}, f^{\timesATOPtimes}, f^{\timestimes}$
cannot be all in $\mathscr{P},$
or all in $\mathscr{A}$,
or all in $\widehat{\mathscr{M}}$ when $x^2\neq y^2$ and $D=0$.
Thus $\PlCSP(f^{\times}, f^{\timesATOPtimes}, f^{\timestimes})$ is \#P-hard by Theorem~\ref{pl-dicho-1}.
So $\PlCSP^2(f)$ is \#P-hard.
This completes the proof of Case~7.
\end{description}

This completes the proof of Theorem~\ref{arity-4-dichotomy}.
\end{proof}

%\section{Some Dichotomy for a set of Even-Arity Signatures}
\section{An Application of Cyclotomic Field} \label{PartII.secE.cyclcomic}

\subsection{Dichotomy Theorem with a Signature in \texorpdfstring{$\widehat{\mathscr{M}} \setminus (\mathscr{P} \cup \widetilde{\mathscr{A}})$}{Mtilde - (P union Atilde)}}

%%% once proved for hat-M, can do holo to do M-hat-dagger, as 10\\0i fixed
%%% P and \tilde{A}

The next three lemmas are crucial.
The purpose of these lemmas is to give a similar result as Lemma~\ref{M-odd} when the signature set $\mathcal{F}$ contains some
%$f \in \mathscr{M} \setminus(\mathscr{P}\cup\widetilde{\mathscr{A}})$,
$f \in \widehat{\mathscr{M}} \setminus (\mathscr{P} \cup \widetilde{\mathscr{A}})$,
%but
%%% I corrected to \hat{M}. not M. JYC, pl check
and all signatures in $\mathcal{F}$ have even arity.
The proof uses an argument involving the degree of extension of a \emph{cyclotomic field}.

We first prove that if we have an even arity signature in
$\widehat{\mathscr{M}} \setminus (\mathscr{P} \cup \widetilde{\mathscr{A}})$,
then we can construct a binary $[1,a,1]$ with $a^4 \notin \{0,1\}$.

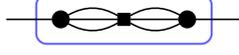
\begin{figure}[t]
 \centering
 \begin{tikzpicture}[scale=\scale,transform shape,node distance=\nodeDist,semithick]
  \node[external] (0)              {};
  \node[internal] (1) [right of=0] {};
  \node[square]   (2) [right of=1] {};
  \node[internal] (3) [right of=2] {};
  \node[external] (4) [right of=3] {};
  \path (0) edge             (1)
        (1) edge[bend left]  (2)
            edge             (2)
            edge[bend right] (2)
        (2) edge[bend left]  (3)
            edge             (3)
            edge[bend right] (3)
        (3) edge             (4);
  \begin{pgfonlayer}{background}
   \node[draw=\borderColor,thick,rounded corners,fit = (1) (3),inner sep=10pt] {};
%    \node[draw=\borderColor,thick,rounded corners,fit = (1) (3),inner sep=6pt,transform shape=false] {};
  \end{pgfonlayer}
 \end{tikzpicture}
 \caption{Gadget used in the proof of Lemma~\ref{[1,a,1]-construction}.}
 \label{fig:6}
\end{figure}

\begin{lemma} \label{[1,a,1]-construction}
 Let $\mathcal{F}$ be a set of symmetric signatures
 containing some
 $f \in \widehat{\mathscr{M}} \setminus(\mathscr{P} \cup \widetilde{\mathscr{A}})$, which has even arity.
 Then
 \[
  \PlCSP^2([1,a, 1], \mathcal{F})
  \leq_T
  \PlCSP^2(\mathcal{F})
 \]
 for some $a$ satisfying $a^4 \notin \{0,1\}$.
\end{lemma}

\begin{proof}
 If $f$ has arity~$2$,
 then we are done by Lemma~\ref{M-2-M-NOT-IN-A-AND-P}.
 Thus,
 we assume that $f$ has arity $2 n \geq 4$.
 By Lemma~\ref{M-2-M-NOT-IN-A-AND-P},
 we have either
 $f = [s, t]^{\otimes 2n} \pm [t, s]^{\otimes 2n}$ with $s^4 \neq t^4$ and $s t \neq 0$
 or $f_k = (\pm 1)^k (2n-2k)$ up to a scalar.

For $f=[s, t]^{\otimes 2n}+[t, s]^{\otimes 2n}$,
we have $\partial^{n-1}(f)
=(s^2+t^2)^{n-1} \{ [s, t]^{\otimes 2} + [t, s]^{\otimes 2} \}
=(s^2+t^2)^{n}[1, a, 1]$,
where $a = \frac{2st}{s^2+t^2}$.
Note that $s^2+t^2 \not =0$ and  $a \neq 0, \pm 1$.
%and  by $s\neq \pm t$.
If $a \neq \pm i$, then we are done.
Suppose $a=\pm i$.
Then $g=\partial^{n-2}(f)
=(s^2+t^2)^{n-2} \{ [s, t]^{\otimes 4} + [t, s]^{\otimes 4} \}$.
A simple calculation shows that $g =
-2s^2t^2(s^2+t^2)^{n-2} [3, \pm i, -1, \pm i, 3]$.
 Consider the gadget in Figure~\ref{fig:6}.
%%% this figure needs to be changed to f -- =_6 -- f
 We assign $[3, \pm i, -1, \pm i, 3]$ to the circle vertices
 and $=_6$ to the square vertex.
 %Although this gadget is not visually symmetric,
 %its signature is the symmetric signature $[3, \pm i, 3]$,
Its signature is $[8, \pm 6i, 8]$,
 so we are done.

For $f=[s, t]^{\otimes 2n}-[t, s]^{\otimes 2n}$,
we have $\partial^{n-1}(f)
=(s^2+t^2)^{n-1} \{ [s, t]^{\otimes 2} -  [t, s]^{\otimes 2} \}
=\lambda [1, 0, -1]$, where
$\lambda = (s^2+t^2)^{n-1} (s^2 -t^2) \not = 0$.
%= (s^2+t^2)^{n}[1, 0, -1]$.  JYC I don't think so...
For $2n\geq 6$, we have $\partial_{[1, 0, -1]}(f)
=(s^2-t^2)\{[s, t]^{\otimes 2n-2}+[t, s]^{\otimes 2n-2}\}$
and we are done by the proof of the previous case, as $2n-2 \ge 4$.
For $2n=4$, we have $\partial_{[1, 0, -1]}(f)
=(s^2-t^2) \{ [s, t]^{\otimes 2} +  [t, s]^{\otimes 2} \}
= (s^4 - t^4) [1, a, 1]$,
where $a = \frac{2st}{s^2+t^2} \not = 0, \pm 1$.
If $a \neq \pm i$, then we are done.
Suppose $a =\pm i$, then
a simple calculation shows that
%%% I think it is s^2-t^2 JYC
% $f=\mp st(s^2+t^2)[2i, \pm 1, 0, \mp 1, -2i]$.
%$f=\mp st(s^2-t^2)[2i, \mp 1, 0, \pm 1, -2i]$.
$f$ is a nonzero multiple of $[2 i, \mp 1, 0, \pm 1, -2 i]$.
(One can verify that $\frac{s^3t - st^3}{s^4 - t^4} = \frac{st}{s^2+t^2}
= \frac{a}{2} = \pm \frac{i}{2}$.)
%%% the nonzero multiple is s^4 - t^4, by lm statement \neq 0
 Consider the gadget in Figure~\ref{fig:6}.
 We assign $[2i, \mp 1, 0, \pm 1, -2i]$ to the circle vertices and
 $=_6$ to the square vertex..
 The signature of this gadget is $[-3, \mp 4i, -3]$,
 so we are done.

 For $f_k = (\pm 1)^k (2n-2k)$,
 we have $\partial^{n-2}(f) = 2^{n-1} [2, \pm 1, 0, \mp 1, -2]$.
 Consider the gadget in Figure~\ref{fig:6}.
 We assign $[2, \pm 1, 0, \mp 1, -2]$ to the circle vertices and 
 $=_6$ to the square vertex.
 The signature of this gadget is $ [5, \pm 4, 5]$,
 so we are done.
\end{proof}

The next lemma shows that if we have $[1, a, 1]$ with $a^4\neq 0, 1$,
%a signature in
%$\mathscr{M}^{\dagger}\setminus(\mathscr{P}\cup\widetilde{\mathscr{A}})$,
then we can obtain $[1, 1]^{\otimes 2}$ by interpolation.
%a binary in $\mathscr{M}^{\dagger}\setminus(\mathscr{P}\cup\widetilde{\mathscr{A}})$.

\begin{lemma} \label{[1,a,1]-interpolation}
 For any signature set $\mathcal{F}$ and any $a^4 \notin \{0,1\}$,
 \[
  \PlCSP^2(\{[1,1]^{\otimes 2}\} \cup \mathcal{F})
  \leq_T
  \PlCSP^2(\{[1,a,1]\} \cup \mathcal{F}).
 \]
\end{lemma}

\begin{proof}
 The eigenvalues of
 \(
  \left[
  \begin{smallmatrix}
   1 & b \\
   b & 1
  \end{smallmatrix}
  \right]
 \)
 are $1+b$ and $1-b$ respectively.
 If we have a signature $[1,b,1]$,
 for some $b \neq 1$,
 such that ratio $\frac{1+b}{1-b}$ of eigenvalues is not a root of unity,
 then we can interpolate any binary signature $[1,x,1]$ for $x \in \mathbb{C}$.
 In particular,
 we could interpolate the desired $[1,1]^{\otimes 2}$.

 Indeed,
 let $\Omega$ be an instance of $\PlCSP^2(\{[1,x,1]\} \cup \mathcal{F})$ in which $[1,x,1]$ occurs $n$ times.
 Write
 \(
  \left[
  \begin{smallmatrix}
   1 & x \\
   x & 1
  \end{smallmatrix}
  \right]
 \)
 as
 \(
  H
  \left[
  \begin{smallmatrix}
  1+x & 0 \\
  0 & 1-x
  \end{smallmatrix}
  \right]
  H,
 \)
 where
 \(
  H
  =
  \frac{1}{\sqrt{2}}
  \left[
  \begin{smallmatrix}
   1 &  1 \\
   1 & -1
  \end{smallmatrix}
  \right].
 \)
 We can stratify the partition function value on $\Omega$
 as $\operatorname{Z}(\Omega) = \sum_{\ell =1}^n c_\ell (1+x)^\ell (1-x)^{n-\ell}$,
 where $c_\ell$ is the sum,
 over all assignments that assign $00$ to $\ell$ copies of
 \(
  \left[
  \begin{smallmatrix}
   1+x & 0 \\
   0   & 1-x
  \end{smallmatrix}
  \right]
 \)
 and $11$ to the remaining $n-\ell$ copies,
 of the product of evaluations of all other signatures from $\mathcal{F}$ and those copies of $H$.
 If we construct a sequence $\Omega_k$ of instances of $\PlCSP^2(\{[1,b,1]\} \cup \mathcal{F})$,
 where we replace each occurrence of $[1,x,1]$ by a chain of $k$ linked copies of $[1,b,1]$,
 then since
 \(
  \left[
  \begin{smallmatrix}
   1 & b \\
   b & 1
  \end{smallmatrix}
  \right]^k
  =
  H
  \left[
  \begin{smallmatrix}
   (1+b)^k & 0 \\
   0       & (1-b)^k
  \end{smallmatrix}
  \right]
  H,
 \)
 we have $\operatorname{Z}(\Omega_k) = (1-b)^{kn} \sum_{\ell =1}^n c_\ell (\frac{1+b}{1-b})^{k \ell}$,
 for $0 \le k \le n$.
 This is a Vandermonde system of full rank,
 and we can solve for all $c_\ell$ and find the value $\operatorname{Z}(\Omega)$.

 The simple gadget with two copies of $=_{2k}$ connected by $2k-1$ parallel copies of $[1,a,1]$ has signature $[1,a^{2k-1},1]$.
 Our key claim is that there exists a $k \ge 1$,
 depending only on $a$,
 such that $\frac{1+a^{2k-1}}{1-a^{2k-1}}$ is \emph{not} a root of unity.
 Then we are done by the interpolation given above.
 
 For a contradiction,
 assume that $\frac{1+a^{2k-1}}{1-a^{2k-1}}$ is a root of unity for all $k \ge 1$.
 For $k = 1$,
 $\frac{1+a}{1-a}$ is some root of unity $e^{2 \pi i j / m}$,
 where $\gcd(j,m) = 1$.
 Then $a \in \Phi_m = \mathbb{Q}(e^{2 \pi i / m})$,
 the $m$-th cyclotomic field.
 Therefore $a^{2k-1} \in \Phi_m$ as well for all $k \ge 1$.
 Furthermore,
 $|\frac{1+a}{1-a}| = 1$,
 so $a$ is purely imaginary,
 i.e.~$a = ih$ for some real $h \notin \{0, \pm 1\}$ since $a^4 \notin \{0,1\}$.
 First we consider the case $0 < |h| < 1$.
 Then $a^{2k-1} = \pm i h^{2k-1}$ and $\lim_{k \to \infty} h^{2k-1} = 0$.

 For all $k \ge 1$,
 $\frac{1+a^{2k-1}}{1-a^{2k-1}}$ is some root of unity $e^{2 \pi i J / M}$
 (in which $J$ and $M$ depend on $k$),
 where $0 < |J| < M / 2$ with $\gcd(J,M) = 1$.
 Then $e^{2 \pi i / M} \in \Phi_m$ as well,
 so $\Phi_M \subseteq \Phi_m$.
 Note that $|\tan (\pi J/M)| = |h|^{2k-1}$.
%%% it is \pi J/M
%%% (1+ih)/(1-ih) forms angle 2 theta.
%%% so (1+(ih)^k) / (1-(ih)^k) forms angle 2 th_k which is 2 pi J/M
%%% so h^k = |(ih)^k| = tan(th_k) = tan(pi J/M)
 Hence $|h|^{2k-1} \ge \tan (\pi /M) \ge \pi /M$.
 Thus $M \ge \pi/|h|^{2k-1}$.

 However, the $M$-th cyclotomic field $\Phi_M$ has degree of extension $[\mathbb{Q}(e^{2 \pi i/M}) : \mathbb{Q}] = \varphi(M)$,
 where $\varphi$ is the Euler totient function.
 We have a crude estimate $(\varphi(M))^2 \ge M/2$,
 which is obvious by considering each prime dividing $M$.
 Then it follows that $\lim_{M \to \infty} \varphi(M) = \infty$,
 which contradicts $\varphi(M) \le \varphi(m) < \infty$.

 The remaining case $|h| > 1$ can be handled similarly.
 In fact,
 if $|h^{2k-1}|$ is large,
 then the angle
%of $\frac{1}{2i}{\log \left(\frac{1+a^{2k-1}}{1-a^{2k-1}}\right)} =
 $\tan^{-1}(|h^{2k-1}|)$ is close (but unequal) to $\pi/2$.
 Then the angle of $\left(\frac{1+a^{2k-1}}{1-a^{2k-1}}\right)^2$ is close (but unequal) to $0$ mod $2 \pi$.
\end{proof}

Combining Lemma~\ref{mixing-P-global}, Lemma~\ref{[1,a,1]-construction} and Lemma~\ref{[1,a,1]-interpolation},
we have proved the following.
%fact:
%For any symmetric signature set $\mathcal{F}$ containing
%some  $f\in\widehat{\mathscr{M}}\setminus(\mathscr{P}\cup\widetilde{\mathscr{A}})$
%of even arity,
%either
% $\mathcal{F}\subseteq\widehat{\mathscr{M}}$ or $\PlCSP^2(\mathcal{F})$ is $\#${\rm P-}hard.
%Together with Lemma~\ref{M-odd}, we can drop the condition that $f$ has
%an even arity.
% Thus if we have a signature $f\in\widehat{\mathscr{M}}\setminus(\mathscr{P}\cup\widetilde{\mathscr{A}})$,
% then we have Theorem~\ref{dichotomy-pl-csp2}.

\begin{lemma} \label{M-even}
 Let $\mathcal{F}$ be a set of even-arity signatures containing $f$.
 If $f \in \widehat{\mathscr{M}} \setminus (\mathscr{P} \cup \widetilde{\mathscr{A}})$,
 then $\PlCSP^2(\mathcal{F})$ is \numP-hard
 unless $\mathcal{F}\subseteq\widehat{\mathscr{M}}$.
%
%Let $\mathcal{F}$ be a signature set and all of the signatures in $\mathcal{F}$
%have even arities. For $f\in\mathcal{F}$,
 %if $f\in\hat{{\mathscr{M}}}\setminus(\mathscr{P}\cup\widetilde{\mathscr{A}})$,
%then $\mathcal{F}\subseteq\widehat{\mathscr{M}}$ or $\PlCSP^2(\mathcal{F})$ is $\#${\rm P-}hard.
\end{lemma}

\subsection{Dichotomy Theorem with a Signature in \texorpdfstring{$\widehat{\mathscr{M}}^\dagger \setminus (\mathscr{P} \cup \widetilde{\mathscr{A}})$}{Mhatdagger - (P union Atilde)}}\label{subsection-E.2}

We would like to prove a corresponding statement
 to Lemma~\ref{M-even} after replacing the condition 
$f\in\widehat{\mathscr{M}}\setminus(\mathscr{P}\cup\widetilde{\mathscr{A}})$
by
$f\in\widehat{\mathscr{M}}^\dagger\setminus(\mathscr{P}\cup\widetilde{\mathscr{A}})$.
This corresponding statement is indeed true
and is implied by Theorem~\ref{dichotomy-pl-csp2},
the final dichotomy theorem for $\PlCSP^2$.
However,
at this point leading up to the proof of Theorem~\ref{dichotomy-pl-csp2},
we are not able to prove it.
Instead,
we prove a weaker version,
Lemma~\ref{2-M-even},
in which $f$ is assisted by a binary signature other than a multiple of $[1,0,1]$.

\begin{remark} \label{rmk:E2:1}
 Here we explain some of the difficulties in the proof
 caused by structural complications of the signatures involved.

When we prove the No-Mixing statements for $\widehat{\mathscr{M}}$
the crucial step is the ability to construct $[1, \omega]^{\otimes 2}$
with $\omega\neq 0$ in the $\PlCSP^2$ setting
(cf.~Lemma~\ref{[1,a,1]-construction} and Lemma~\ref{[1,a,1]-interpolation}). 
This is the key, and the only known
method, for us to
leverage the existing dichotomy for $\PlCSP$
(cf.~Lemma~\ref{mixing-P-global}).
Then in a similar spirit,
to prove the No-Mixing statements for $\hat{\mathscr{M}}^{\dagger}$,
we would like to be able to construct $[1, \omega]^{\otimes 2}$ as well.

A signature $f = [f_0, \ldots, f_n]$ is called an \emph{odd} signature
if $f_{2k}=0$ for all $k \ge 0$,
and an \emph{even} signature if $f_{2k+1}=0$ for all $k \ge 0$.

In any $\mathcal{F}$-gate $H$, if every
signature in  $\mathcal{F}$ satisfies parity constraints,
then the signature of $H$ also
satisfies parity constraints. In fact the parity
of the signature of $H$ is the same as the parity
of the  number of occurrences of odd signatures of $\mathcal{F}$
 in $H$.
To see this, suppose   $\sigma$ 
 is a $\{0, 1\}$-assignment
 to all the edges of $H$, including
internal and external edges, that has a nonzero evaluation on $H$.
By parity constraints,
each odd (resp. even) signature appearing in $H$
% $\mathcal{F}$
 has an
odd (resp.  even) number of
incident edges assigned 1.
Adding up all these numbers $\bmod~2$,
noting that each internal edge of $H$ assigned 1 contributes 2 to the sum
while each external edge of $H$ assigned 1 contributes 1,
we get $N \equiv 2X + Y \equiv Y \pmod 2$,
where $N$ is the number of occurrences of odd signatures of $\mathcal{F}$
in $H$,  and $X$ (resp. $Y$)  is the number of internal 
(resp. external) edges assigned  to $1$ by  $\sigma$.
Hence $H$ has the same parity as $N$.

For any signature of the form $f=[s, ti]^{\otimes m} \pm [t, si]^{\otimes m}$,
or $f_k=(\pm i)^k(m-2k)$, for any arity $m$,
$(Z^{-1})^{\otimes m}f$ satisfies the parity constraints,
where $Z= \frac{1}{2}\left[
\begin{smallmatrix}
1 & 1\\
i & -i
\end{smallmatrix}
\right]$.
In fact for $f$ of the first type, $(Z^{-1})^{\otimes m}f
= [u, v]^{\otimes m} \pm [u, -v]^{\otimes m}$
for $u = s+t$ and $v= s-t$, and
for $f$ of the second type, $(Z^{-1})^{\otimes m}f
= 2^m [0,1,0, \ldots, 0]$ or $2^m [0, \ldots, 0,1,0]$.
Note that
\[
 \trans{1}{1}{i}{-1}^{\otimes m}
 [0,1,0, \ldots, 0]
 =
 \Sym_n^{n-1}( 
 \left[\begin{smallmatrix} 1 \\  i \end{smallmatrix}\right];
 \left[\begin{smallmatrix} 1 \\ -i \end{smallmatrix}\right])
 \]
has its $k$-th term $i^k (m-2k)$.
Similarly,
$\left[\begin{smallmatrix}
1 & 1\\
i & -i
\end{smallmatrix}
\right]^{\otimes m} [0,\ldots,0,1,0]$
has its $k$-th term $(-i)^k (m-2k)$.

Under the  holographic transformation $Z$,  we have
\begin{equation} \label{eqn-remark-E.2}
 \PlCSP^2(f)
 \equiv_T
 \PlHolant([0, 1, 0], [1, 0, 1, 0, 1], \ldots \mid \hat{f}),
\end{equation}
where $\hat{f}=(Z^{-1})^{\otimes m}f$,
and $\frac{1}{2} [0, 1, 0] = (=_2) Z^{\otimes 2}$, 
$\frac{1}{2^3}[1, 0, 1, 0, 1]
= (=_4) Z^{\otimes 4}$, etc.
Notice that for the signatures $(=_{2n}) Z^{\otimes 2n}$,
if  the arity $2n\equiv 2\pmod 4$ then the signature is odd,
and if $2n\equiv 0\pmod 4$ then the signature is even.

Every signature of the form
 $[s, ti]^{\otimes m}+[t, si]^{\otimes m}$ is even,
every signature of the form
 $[s, ti]^{\otimes m}-[t, si]^{\otimes m}$ is odd,
and for even arity $2n$ the signatures
$[0,1,0, \ldots, 0]$ and $[0,\ldots,0,1,0]$ are
both odd.

Thus, if we focus on signatures
$f=[s, ti]^{\otimes 2n}+[t, si]^{\otimes 2n}$ with arity $2n\equiv 0\pmod 4$,
or $f=[s, ti]^{\otimes 2n}-[t, si]^{\otimes 2n}$ with
arity $2n\equiv 2\pmod 4$,
or $f_k=(\pm i)^k(2n-2k)$ with arity $2n\equiv 2\pmod 4$,
then the following property holds
for \emph{all} the signatures in the bipartite $\PlHolant$ problem in~\eqref{eqn-remark-E.2}:
\begin{quote}
All signatures of arity $2n\equiv 2\pmod 4$
satisfy odd parity and all signatures of arity $2n\equiv 0\pmod 4$
satisfy even parity.
\end{quote}

It follows that, for such $f$, any gadget constructed from
(\ref{eqn-remark-E.2}) has the same parity as the number of  occurrences of
signatures of arity $2n\equiv 2\pmod 4$.

Furthermore, in a bipartite gadget construction in
$\PlHolant([0, 1, 0], [1, 0, 1, 0, 1], \ldots \mid \hat{f})$,
if the resulting signature of the gadget is \emph{binary},
 the number of  occurrences of  signatures of arity
$2n\equiv 2\pmod 4$ in this gadget \emph{must} be odd.
Indeed 
let $N_0$ (resp. $N_2$) denote the  number of  occurrences of
 signatures of arity
$2n\equiv 0 \pmod 4$ (resp. $2n\equiv 2\pmod 4$) in this 
bipartite gadget, and we add up the arities of all signatures modulo 4, 
we get $0 N_0 + 2 N_2 \equiv  2 N_I + 2 \pmod 4$,
where $N_I$ is the number of internal edges in the bipartite gadget,
and the additive term $2$ is because the gadget is a binary gadget.
Thus $N_2  \equiv  N_I + 1 \pmod 2$.
On the other hand, since the gadget  is bipartite, 
$N_I$ is the sum of all arities of signatures from RHS,  and
minus 2 if the external 2 edges come from the RHS.
As all signatures in this gadget 
have even arity,
$N_I \equiv 0 \pmod 2$.  Hence $N_2  \equiv 1 \pmod 2$.

This implies that any binary signature
constructed in
$\PlHolant([0, 1, 0], [1, 0, 1, 0, 1], \ldots \mid \hat{f})$
must have odd parity, i.e.,
they are all of the form $\lambda[0, 1, 0]$.
Thus, before the $Z$-transformation, one can only
construct binary signatures of the form  $\frac{\lambda}{2}[1, 0, 1] = 
\lambda Z^{\otimes 2} [0, 1, 0]$
in $\PlCSP(f)$ by gadget construction.
This can be verified as 
$\left[
\begin{smallmatrix}
1 & 1\\
i & -i
\end{smallmatrix}
\right]
\left[
\begin{smallmatrix}
0 & 1\\
1 & 0
\end{smallmatrix}
\right]
\left[
\begin{smallmatrix}
1 & i\\
1 & -i
\end{smallmatrix}
\right]
=
2 
\left[
\begin{smallmatrix}
1 & 0\\
0 & 1
\end{smallmatrix}
\right]$.

In particular one \emph{cannot} construct $[1, \omega]^{\otimes 2}$
in $\PlCSP^2(f)$ by gadget construction.
This explains the extra mile we have to travel in this proof.

As indicated, therefore,
 we prove a weaker version of  Lemma~\ref{M-even} in this subsection, namely
Lemma~\ref{2-M-even},
in which $f$ is assisted by a binary signature other 
than a multiple of $[1,0,1]$.
\end{remark}

\vspace{.1in}

We begin with the following lemma.

\begin{lemma}\label{2-M-even-X}
Let $\mathcal{F}$ be any set of symmetric signatures of even arities,
and suppose $\mathcal{F}$ contains signatures $f$ and $g$,
where
$f\in\widehat{\mathscr{M}}^\dagger
\setminus(\mathscr{P}\cup\widetilde{\mathscr{A}})$,
and $g=[g_0, g_1, \ldots, g_{2n}]$ and there exists a positive
 integer $s$
such that $g_0^{s}=-g_{2n}^{s}\neq 0$.
Then
either $\mathcal{F}\subseteq\widehat{\mathscr{M}}^\dagger$ or $\PlCSP^2(\mathcal{F})$ is \numP-hard.
\end{lemma}

\begin{proof}
Let $E_{2k}(-1)=[1, 0, \ldots, 0, -1]$ have
 arity $2k$ and $\mathcal{E}(-1) = \{E_{2k}(-1) \mid k\geq 1\}$.
Firstly, by our calculus
we have $\partial^{s}_{g}(=_{2ns+2k})=g^{s}_0E_{2k}(-1)$ on LHS for $k\geq 1$.
Thus we have
\[
 \PlHolant(\mathcal{E}(-1)\cup \mathcal{EQ}_2 \mid \mathcal{F})
 \leq_T
 \PlCSP^2(\mathcal{F}).
\]

Under a holographic transformation by
${T}^{-1} = \left[\begin{smallmatrix} 1 & 0 \\ 0 & i \end{smallmatrix}\right]$,
the set $\mathcal{E}(-1)\cup \mathcal{EQ}_2$ is set-wise invariant.
Indeed, for all $k \ge 1$, signatures of arity $4k$
in $\mathcal{E}(-1)\cup \mathcal{EQ}_2$  are pointwise fixed,
and signatures of arity $4k-2$ in $\mathcal{E}(-1)$ and in $\mathcal{EQ}_2$
are interchanged.
Thus,
\[
 \PlCSP^2({T}\mathcal{F})
 \leq_T
 \PlHolant(\mathcal{E}(-1) \cup \mathcal{EQ}_2 \mid T \mathcal{F})
 \equiv_T
 \PlHolant(\mathcal{E}(-1) \cup \mathcal{EQ}_2 \mid \mathcal{F}).
\]
Note that ${T}^{\otimes 2n}f\mathscr{f}\in{T}\mathcal{F}$
is in $\widehat{\mathscr{M}}\setminus(\mathscr{P}\cup\widetilde{\mathscr{A}})$.
Thus either
${T}\mathcal{F}\subseteq\widehat{\mathscr{M}}$ or
$\PlCSP^2({T}\mathcal{F})$ is $\#$P-hard
by Lemma~\ref{M-even}.
Note that ${T}\mathcal{F}\subseteq\widehat{\mathscr{M}}$ iff
$\mathcal{F}\subseteq\widehat{\mathscr{M}}^\dagger$.
Thus either $\mathcal{F}\subseteq\widehat{\mathscr{M}}^\dagger$ or
$\PlCSP^2(\mathcal{F})$ is $\#$P-hard.
\end{proof}

The next two lemmas show
that if we have a signature in
$\widehat{\mathscr{M}}^\dagger\setminus(\mathscr{P}\cup\widetilde{\mathscr{A}})$
and a binary signature that is not a multiple of
 $[1, 0, 1]$, then we have the same statement for
$\widehat{\mathscr{M}}^\dagger$,
 as Lemma~\ref{M-even} is
for $\widehat{\mathscr{M}}$.
This will be stated as Lemma~\ref{2-M-even}.
Note that if
$f\in\widehat{\mathscr{M}}^\dagger
\setminus(\mathscr{P}\cup\widetilde{\mathscr{A}})$
is a binary signature, then
$f$ takes the form $[1, b, -1]$ by Lemma~\ref{M-2-M-NOT-IN-A-AND-P},
and this case is covered by
Lemma~\ref{2-M-even-X}, where $f$ also plays the role of $g$.
Thus
% in Lemma~\ref{2-M-even-XX} and~\ref{2-M-even-XXX}, 
we assume
$f\in\widehat{\mathscr{M}}^\dagger
\setminus(\mathscr{P}\cup\widetilde{\mathscr{A}})$
has arity $\ge 4$.
By Lemma~\ref{M-2-M-NOT-IN-A-AND-P}, such a signature $f$ has two forms.
Lemma~\ref{2-M-even-XX} and~\ref{2-M-even-XXX} handle these two cases
respectively.

\begin{lemma}\label{2-M-even-XX}
Let $\mathcal{F}$ be any set of symmetric signatures of even arities,
and suppose $\mathcal{F}$ contains signatures $f$ and $h$,
where
$f=[s, ti]^{\otimes 2n}\pm[t, si]^{\otimes 2n}$ with $2n\geq 4$, $s^4\neq t^4$ and $st\neq 0$,
and $h$ is any nonzero binary signature other than $\lambda[1, 0, 1]$.
%$f\in\widehat{\mathscr{M}}^\dagger\setminus(\mathscr{P}\cup\widetilde{\mathscr{A}})$.
Then either
$\mathcal{F}\subseteq\widehat{\mathscr{M}}^\dagger$ or $\PlCSP^2(\mathcal{F})$ is \numP-hard.
\end{lemma}
\begin{proof}
Firstly, by our calculus, ignoring the nonzero factor
$(s^2-t^2)^{n-2}$ in $\partial^{n-2}(f)$,
we have $g=[s, ti]^{\otimes 4}\pm(-1)^{n-2}[t, si]^{\otimes 4}$.
If $g=[s, ti]^{\otimes 4}-[t, si]^{\otimes 4}$,
then we have $\partial(g)=(s^2-t^2)\{[s, ti]^{\otimes 2}+[t, si]^{\otimes 2}\}
=(s^2-t^2)[s^2+t^2, 2sti, -(s^2+t^2)]$
and we are done by Lemma~\ref{2-M-even-X}.

Suppose $g=[s, ti]^{\otimes 4}+[t, si]^{\otimes 4}$, and we also have
$h\neq\lambda[1, 0, 1]$.
If $h \notin \mathscr{P} \cup \widetilde{\mathscr{A}} \cup \widetilde{\mathscr{M}}$,
then $\PlCSP^2(\mathcal{F})$ is \numP-hard by Theorem~\ref{thm:parII:k-reg_homomorphism}$'$.
Otherwise,
by Lemma~\ref{binary},
the possibilities for $h$,
after normalizing,
are
\[
 [a, b]^{\otimes 2},~~
 [1, 0, x],~~
 [0, 1, 0],~~
 [1, \rho, -\rho^2],~~
 [1, \alpha, -\alpha^2],~~
 [1, u, 1],
 \quad \text{and} \quad
 [1, v, -1],
\]
where $x \notin \{0,1\}$, $\rho^4 = 1$, $\alpha^4 = -1$, $u^4 \notin \{0,1\}$, and $v^4 \notin \{0,1\}$.
\begin{itemize}
\item If $h=[a, b]^{\otimes 2}$ with $ab\neq 0$, then we are done by
Lemma~\ref{mixing-P-global}.
\item If $h\in\{[1, 0, -1], [1, 0, \pm i], [1, \pm 1, -1], [1, \alpha, -\alpha^2], [1, v, -1]\}$, then we
are done by Lemma~\ref{2-M-even-X}.
\item If $h=[1, u, 1]$ with $u^4\neq 0, 1$, then $h\in\widehat{\mathscr{M}}\setminus(\mathscr{P}\cup\widetilde{\mathscr{A}})$ by
Lemma~\ref{M-2-M-NOT-IN-A-AND-P}.
Thus we are done by Lemma~\ref{M-even}.
\end{itemize}
The remaining cases are $h= [1, 0]^{\otimes 2}$,
$[0, 1]^{\otimes 2}$, $[1, 0, x]$,
$[0, 1, 0]$ or $[1, \pm i, 1]$, 
%$[1, 0]^{\otimes 2}$ or $[0, 1]^{\otimes 2}$,
where $x^4\neq 0, 1$.
%by normalizing.
\begin{itemize}
\item If $h=[1, 0, x]$ with $x^4\neq 0, 1$,
then
by taking 4 copies of $h$ and connecting one input of $h$ to each edge of $g$,
 we have $\hat{g}=\left[\begin{smallmatrix} 1 & 0\\
0 & x \end{smallmatrix}\right]^{\otimes 2n}g=[s, xti]^{\otimes 4}+[t, xsi]^{\otimes 4}$.
The signature $\hat{g}$ is non-degenerate, has arity 4, and
satisfies a second recurrence relation.
 The eigenvalues of the recurrence relation are
$\frac{xti}{s}$ and $\frac{xsi}{t}$. By the trace and product,
$\hat{g}$  has type
$\langle -x^2, \frac{xti}{s}+\frac{xsi}{t}, 1\rangle$.
Thus
$\hat{g}\notin\mathscr{P}\cup\widetilde{\mathscr{A}}\cup\widetilde{\mathscr{M}}$ by Lemma~\ref{second-recurrence-relation}, since
$(-x^2)^2 \not = 0,1$ and $\frac{t}{s}+\frac{s}{t} \not = 0$.
So $\PlCSP^2(\hat{g})$ is \numP-hard by Theorem~\ref{arity-4-dichotomy}.
Thus $\PlCSP^2(\mathcal{F})$ is \numP-hard.
\item If $h=[0, 1, 0]$, then $\partial_{h}(g)=2sti\{[s, ti]^{\otimes 2}+[t, si]^{\otimes 2}\}=2sti[s^2+t^2, 2sti, -(s^2+t^2)]$.
Then we are done by Lemma~\ref{2-M-even-X}.
\item If $h=[1, \pm i, 1]$, by connecting two copies of $[1, \pm i, 1]$ we have $\pm 2i[0, 1, 0]$, as
$\left[\begin{smallmatrix} 1 & \pm i\\
\pm i  & 1 \end{smallmatrix}\right]^2
=
\left[\begin{smallmatrix} 0 & \pm 2i\\
\pm 2i & 0 \end{smallmatrix}\right]$.
Then we are done by the previous case.
\item If $h=[1, 0]^{\otimes 2}$, then
we have $g'=\partial_{h}(g)=s^2[s, ti]^{\otimes 2}+t^2[t, si]^{\otimes 2}=[s^4+t^4, (s^2+t^2)sti, -2s^2t^2]$.
We claim that $g'\notin\mathscr{P}\cup\widetilde{\mathscr{A}}\cup\widetilde{\mathscr{M}}$.
%Note that all of the entries of $h$ is nonzero by $s^4\neq t^4$ and $st\neq 0$.
\begin{itemize}
\item If $g'\in\mathscr{P}$, then $g'$ is degenerate by $(s^2+t^2)sti\neq 0$
and $-2s^2t^2\neq 0$.
So $-2s^2t^2(s^4+t^4)=-(s^2+t^2)^2s^2t^2$. Thus
$st =0$ or $(s^2-t^2)^2=0$. This is a contradiction.
\item If $g'\in\mathscr{A}\setminus\mathscr{P}$, then $g'=[1, \rho, -\rho^2]$
up to a scalar by Corollary~\ref{binary-necessary}, where $\rho^4 =1$.
By $\rho^2=\pm 1$, we have
$s^4+t^4=\pm2s^2t^2$. This contradicts that $s^4\neq t^4$.
%By  and $(s^2+t^2)sti\neq 0$, $g'\notin\mathscr{A}$ by Lemma~\ref{binary}.
\item If $g'\in\mathscr{A}^\dagger\setminus\mathscr{P}$,
then $g'=[1, \alpha, -\alpha^2]$ up to a scalar
by Corollary~\ref{binary-necessary}, where $\alpha^4 = -1$.
 Thus $2s^2t^2(s^4+t^4)=-(s^2+t^2)^2s^2t^2$.
Then, by $st \not =0$, we have
 $3(s^4+t^4)=-2s^2t^2 \not =0$, and so $|s^4+t^4|\neq|-2s^2t^2|$.
This implies that the norms of two nonzero entries of $g'$ are not equal.
%By Corollary~\ref{binary-necessary}, this contradicts that
This contradicts the form $g' = \lambda [1, \alpha, -\alpha^2]$.  
%$\lambda \not =0$. 
%Hence
%$g'\in\mathscr{A}^\dagger\setminus\mathscr{P}$.
\item Since $s^4 \neq t^4$ we have $s^4+t^4\neq\pm2s^2t^2$.
Hence  $g' \notin \widetilde{\mathscr{M}}$ by Corollary~\ref{binary-necessary}.
\end{itemize}
Then by Theorem~\ref{thm:parII:k-reg_homomorphism}$'$,
$\PlCSP^2(g')$ is \numP-hard.
Thus $\PlCSP^2(\mathcal{F})$ is \numP-hard.
\item If $h=[0, 1]^{\otimes 2}$,
then we apply the transformation $\trans{0}{1}{1}{0}$
and are done by the previous case.
\qedhere
\end{itemize}
\end{proof}

\begin{lemma}\label{2-M-even-XXX}
Let $\mathcal{F}$ be any set of symmetric signatures of even arities,
and suppose $\mathcal{F}$ contains signatures $f$ and $h$,
where
$f$ has arity  $2n\geq 4$ and $f_k=(\pm i)^k(2n-2k)$,
and $h$ is any nonzero binary signature other than $\lambda[1, 0, 1]$.
%$f\in\widehat{\mathscr{M}}^\dagger\setminus(\mathscr{P}\cup\widetilde{\mathscr{A}})$.
Then either $\mathcal{F}\subseteq\widehat{\mathscr{M}}^\dagger$ or $\PlCSP^2(\mathcal{F})$ is \numP-hard.
\end{lemma}

\begin{proof}
If $2n\equiv 0\pmod 4$,
then $f_0=-f_{2n}=2n$.
Thus we are done by Lemma~\ref{2-M-even-X}.

Suppose $2n\equiv 2\pmod 4$.
Thus $n \ge 3$ and  we have $g = \partial^{\frac{n-3}{2}}_{=_4}(f)$
of  arity 6. Ignoring the nonzero factor $2^{\frac{n-3}{2}}$,
we have $g_k = (\pm i)^k(6-2k)$. Removing another factor 2, we have
\[g=[3, \pm 2i, -1, 0, -1, \mp 2i, 3].\]
We also have a nonzero binary signature $h\neq\lambda[1, 0, 1]$.
If $h \notin \mathscr{P} \cup \widetilde{\mathscr{A}} \cup \widetilde{\mathscr{M}}$,
then $\PlCSP^2(\mathcal{F})$ is \numP-hard by Theorem~\ref{thm:parII:k-reg_homomorphism}$'$.
Otherwise
(similar to the proof of Lemma~\ref{2-M-even-XX}),
by Lemma~\ref{binary},
the possibilities for $h$,
after normalizing,
are
\[
 [a, b]^{\otimes 2},~~
 [1, 0, x],~~
 [0, 1, 0],~~
 [1, \rho, -\rho^2],~~
 [1, \alpha, -\alpha^2],~~
 [1, u, 1],
 \quad \text{and} \quad
 [1, v, -1],
\]
where $x \notin \{0,1\}$, $\rho^4 = 1$, $\alpha^4 = -1$, $u^4 \notin \{0,1\}$, and $v^4 \notin \{0,1\}$.
If $h = [1, 0, -1]$, $[1, 0, \pm i]$, $[1, \pm 1, -1]$, $[1, \alpha, -\alpha^2]$, $[1, v, -1]$, $[1, u, 1]$, or $[a, b]^{\otimes 2}$ with $ab\neq 0$,
then we are done with the same proof as in Lemma~\ref{2-M-even-XX}.

The remaining cases are
$h = [1, 0]^{\otimes 2}$,
$[0, 1]^{\otimes 2}$,
$[1, 0, x]$,
$[0, 1, 0]$, or
$[1, \pm i, 1]$, 
where $x^4 \notin \{0,1\}$.
\begin{itemize}
\item For $h=[1, 0, x]$ with $x^4 \notin \{0,1\}$,
we have $g'=\partial_{h}(g)=[3-x, \pm 2i, -1-x, \mp 2xi, -1+3x]$.
The signature $g'$ is non-degenerate
because $(\pm 2i)(\mp 2xi)\neq (-1-x)^2$ by $x\neq 1$.
Moreover,  $g'$ satisfies
the second recurrence relation with type $\langle 1, \mp 2i, -1\rangle$.
Thus $g'\notin\mathscr{P}\cup\widetilde{\mathscr{A}}\cup\widehat{\mathscr{M}}$ by Lemma~\ref{second-recurrence-relation}.
Moreover, by $x\neq \pm 1$, we have $3-x\neq \pm(-1+3x)$, so $g'\notin\widehat{\mathscr{M}}^\dagger$ by Corollary~\ref{arity-hat-M-dagger}.
So $\PlCSP^2(g')$ is \numP-hard by Theorem~\ref{arity-4-dichotomy}.
Thus $\PlCSP^2(\mathcal{F})$ is \numP-hard.
\item If $h=[0, 1, 0]$, then $\partial_{h}(g)=[\pm 2i, -1, 0, -1, \mp 2i]$.
Then we are done by Lemma~\ref{2-M-even-X}.
\item If $h=[1, \pm i, 1]$, by connecting two copies of $[1, \pm i, 1]$ we have $\pm 2i[0, 1, 0]$.
Then we are done by the proof of the previous case.
\item If $h=[1, 0]^{\otimes 2}$, then
we have $g''=\partial^2_{h}(g)=[3, \pm 2i, -1]$.
By Corollary~\ref{binary-necessary},
we have $g''\notin\mathscr{P}\cup\widetilde{\mathscr{A}}\cup\widetilde{\mathscr{M}}$.
Then by Theorem~\ref{thm:parII:k-reg_homomorphism}$'$,
$\PlCSP^2(g'')$ is \numP-hard.
Thus $\PlCSP^2(\mathcal{F})$ is \numP-hard.
\item If $h=[0, 1]^{\otimes 2}$,
we apply the transformation
$\left[\begin{smallmatrix} 0 & 1\\
 1 & 0 \end{smallmatrix}\right]$
 and it follows from the previous case. \qedhere
\end{itemize}
\end{proof}

\begin{lemma} \label{2-M-even}
Let $\mathcal{F}$ be any set of symmetric signatures of even arities,
and suppose $\mathcal{F}$ contains signatures $f$ and $h$,
where
$f\in\widehat{\mathscr{M}}^\dagger
\setminus(\mathscr{P}\cup\widetilde{\mathscr{A}})$,
and $h$ is any nonzero binary signature other than $\lambda[1, 0, 1]$.
%
%For $f\in\mathcal{F}$, if $f$ is of one of the following forms:
%\begin{itemize}
%\item $f=[s, ti]^{\otimes 2n}+[t, si]^{\otimes 2n}$ with $2n\equiv 2\pmod 4$,
%where $st\neq 0$, $s^4\neq t^4$ and $2n\geq 4$.
%\item $f=[s, ti]^{\otimes 2n}-[t, si]^{\otimes 2n}$ with $2n\equiv 0\pmod 4$, where
%$st\neq 0$, $s^4\neq t^4$ and $2n\geq 4$.
%\item $f$ has arity $2n\equiv 0\pmod 4$ and $f_k=\lambda(\pm i)^k(2n-2k)$.
%\item $f=[1, b, -1]$ with $b^4\neq 0, 1$.
%\end{itemize}
Then either $\mathcal{F}\subseteq\widehat{\mathscr{M}}^\dagger$ or $\PlCSP^2(\mathcal{F})$ is \numP-hard.
\end{lemma}

\begin{proof}
If $f$ has arity 2, then $f=[1, b, -1]$ by Lemma~\ref{M-2-M-NOT-IN-A-AND-P}.
Then we are done by Lemma~\ref{2-M-even-X}.

If $f$  has arity $2n\geq 4$, then by Lemma~\ref{M-2-M-NOT-IN-A-AND-P},
we have $f=[s, ti]^{\otimes 2n}\pm[t, si]^{\otimes 2n}$
 with $st\neq 0$, $s^4\neq t^4$, or $f_k=(\pm i)^k(2n-2k)$ up to a scalar.
These two cases are proved in Lemma~\ref{2-M-even-XX},
and~\ref{2-M-even-XXX} respectively.
\end{proof}

\begin{remark} \label{rmk:E2:2}
Lemma~\ref{M-even} and Lemma~\ref{2-M-even} will substantially simplify the
succeeding proof for No-Mixing Lemmas concerning $\widehat{\mathscr{M}}$
and $\widehat{\mathscr{M}}^\dagger$.
Thus it is natural that we wish to do the same for $\mathscr{A}$,
and that means we would like to
construct $[1, \omega]^{\otimes 2}$ with $\omega\neq 0$ in 
$\PlCSP^2(f)$ for $f\in\mathscr{A}$.
Unfortunately, for  most cases of $f \in\mathscr{A}$
this is impossible.

First, for a signature $f\in\mathscr{A}$, if $f$ satisfies parity constraints,
then all signatures constructed in $\PlCSP^2(f)$ satisfy parity constraints,
since all ${\cal EQ}_2$ also satisfy parity constraints.
So it is impossible to construct $[1, \omega]^{\otimes 2}$ with $\omega\neq 0$ in $\PlCSP^2(f)$.

If a signature $f\in\mathscr{A}$ is degenerate and does not satisfy parity constraints,
then $f=[1, \pm 1]^{\otimes 2n}$ or $f=[1, \pm i]^{\otimes 2n}$ up to a scalar.
For $f=[1, \pm 1]^{\otimes 2n}$, we have $\partial^{n-1}(f)=2^{n-1}[1, \pm 1]^{\otimes 2}$.
For $f=[1, \pm i]^{\otimes 2n}$ and $2n\equiv 2\pmod 4$, we have $\partial^{\frac{n-1}{2}}_{=_4}(f)=2^{\frac{n-1}{2}}[1, \pm i]^{\otimes 2}$.
Thus in these two particular cases we \emph{can}
 get $[1, \omega]^{\otimes 2}$ with $\omega\neq 0$.
We will show that these are the only cases that this is possible.

Let $f=[1, \pm i]^{\otimes 2n}$ and $2n\equiv 0\pmod 4$.
After a holographic transformation by
$Z = \trans{1}{1}{i}{-i}$,
we have
\[
 \PlCSP^2(f)
 \equiv_T
 \PlHolant([0, 1, 0], [1, 0, 1, 0, 1], \dotsc \mid \hat{f}),
\]
where $\hat{f}=(Z^{-1})^{\otimes 2n}f$, i.e., $\hat{f}=[1, 0]^{\otimes 2n}$ or $\hat{f}=[0, 1]^{\otimes 2n}$.
In $\PlHolant([0, 1, 0], [1, 0, 1, 0, 1], \dotsc \mid \hat{f})$,
all signatures of arity $\equiv 0\pmod 4$ have even parity and
all signatures of arity $\equiv 2\pmod 4$ have odd parity.
By the same proof in Remark~\ref{rmk:E2:2},
all nonzero binary signatures that can be constructed
in $\PlHolant([0, 1, 0], [1, 0, 1, 0, 1], \dotsc \mid \hat{f})$ are
multiples of $[0, 1, 0]$.
In terms of signatures that can be constructed before the
$Z$-transformation, this is equivalent to say that
all nonzero binary signatures that can be constructed  in $\PlCSP^2(f)$
must be multiples of $[1, 0, 1]$.
In particular, one \emph{cannot} construct 
 $[1, \omega]^{\otimes 2}$ with $\omega\neq 0$ in $\PlCSP^2(f)$. 

If $f\in\mathscr{A}$ is non-degenerate and does not satisfy parity constraints, then
$f=[1, i]^{\otimes 2n}\pm i[1, -i]^{\otimes 2n}$ or $f=[1, 1]^{\otimes 2n}\pm i[1, -1]^{\otimes 2n}$.
If we can construct $[1, \omega]^{\otimes 2}$ with $\omega\neq 0$ in $\PlCSP^2(f)$,
then $[1, \omega]^{\otimes 2}$ must
 be in $\mathscr{A}$. Thus $[1, \omega]^{\otimes 2}$=$[1, \pm 1]^{\otimes 2}$
or $[1, \pm i]^{\otimes 2}$.

For $f=[1, i]^{\otimes 2n}\pm i[1, -i]^{\otimes 2n}$,
 $f=[1, \pm 1, -1, \mp 1, \ldots, (-1)^n]$ up to the scalar $1\pm i$.
In any construction in  $\PlCSP^2(f)$, if we 
ignore a global scalar factor which is a power of $1\pm i$,
all entries of the constructed signature are real numbers.
Thus the ratio of any two nonzero entries is a real number.
But this is not the case with $[1, \pm i]^{\otimes 2}$.
This  implies that we \emph{cannot} construct $[1, \pm i]^{\otimes 2}$
in  $\PlCSP^2(f)$ by gadget construction.

Moreover, we claim that it is impossible to get
$[1, \pm 1]^{\otimes 2}$ in $\PlCSP^2(f)$ by gadget construction.
After a holographic transformation by $Z=\left[
\begin{smallmatrix}
1 & 1\\
i & -i
\end{smallmatrix}
\right]$,
we have
\[
 \PlCSP^2(f)
 \equiv_T
 \PlHolant([0, 1, 0], [1, 0, 1, 0, 1], \dotsc \mid \hat{f}),
\]
where $\hat{f}=(Z^{-1})^{\otimes 2n}f=[1, 0, \ldots, 0, \pm i]$.
All signatures in $\PlHolant([0, 1, 0], [1, 0, 1, 0, 1], \dotsc \mid \hat{f})$
satisfy parity constraints. Thus we cannot construct
$(Z^{-1})^{\otimes 2}[1, \pm 1]^{\otimes 2} = 
\mp \frac{i}{2}[1, \pm i]^{\otimes 2}$, which does not satisfy parity constraints,
by gadget construction.
Thus  we  \emph{cannot}
 get $[1, \pm 1]^{\otimes 2}$ in $\PlCSP^2(f)$ by gadget construction.

For $f=[1, 1]^{\otimes 2n}\pm i[1, -1]^{\otimes 2n}$, after
 a holographic transformation by $\trans{1}{0}{0}{i}$,
%%% psychologically easier to think: -i -->> i. it 's also correct.
we can use the same argument as the previous case for $[1, \pm i]^{\otimes 2}$
 to prove that
%%% similar: complex vs real
we cannot get $[1, \pm 1]^{\otimes 2}$ 
in $\PlCSP^2(f)$ by gadget construction.
Moreover,  it is also impossible to get
$[1, \pm i]^{\otimes 2}$ in $\PlCSP^2(f)$ by gadget construction.
After a holographic transformation by $H=\left[
\begin{smallmatrix}
1 & 1\\
1 & -1
\end{smallmatrix}
\right]$,
we have
\[
 \PlCSP^2(f)
 \equiv_T
 \PlHolant([1, 0, 1], [1, 0, 1, 0, 1], \dotsc \mid \hat{f}),
\]
where $\hat{f}=(H^{-1})^{\otimes 2n}f=[1, 0, \ldots, 0, \pm i]$.
All signatures in $\PlHolant([1, 0, 1], [1, 0, 1, 0, 1], \dotsc \mid \hat{f})$
satisfy parity constraints. Thus we cannot construct
$(H^{-1})^{\otimes 2}[1, \pm i]^{\otimes 2} =
\pm \frac{i}{2}[1, \mp i]^{\otimes 2}$
by gadget construction.
This implies that we cannot 
get $[1, \pm i]^{\otimes 2}$ in $\PlCSP^2(f)$ by gadget construction.
\end{remark}

\section{No-Mixing of a Pair of Signatures of Even Arity}\label{PartII.secF.No-Mixing-of-pairs}

The general theme of this section and the next is that,
for planar $\PlCSP^2$ problems,
various tractable signatures of different types cannot mix.
In these two sections,
all signatures are of even arity.
In this section we prove a No-Mixing theorem
for a pair of signatures. This will be extended to
a set of signatures in the next section.  

The general form of the No-Mixing theorem to be proved in this
section is as follows: Let $f$ and $g$ be two symmetric signatures
of even arity. Suppose for some $1 \le j < i \le 5$,  $f \in S_i 
\setminus S_j$ and $g \in S_j \setminus S_i$,
and for all $1 \le k \le 5$, $\{f, g \} \not \subseteq S_k$. Then 
 $\PlCSP^2(f, g)$ is \#P-hard.
We will call such a statement No-Mixing-$(i,j)$.

It is easy to see that,
with possibly switching the names $f$ and $g$,
the condition stated above is equivalent to the following assumption:
\[
 \{f, g \} \subset \bigcup_{k=1}^5 S_k
 \text{ but for any }
 1 \le k \le 5,
 \text{ we have }
 \{f, g \} \not\subseteq S_k.
\]
However under this assumption,
we make the following observation that any index $i$ for which $f \in S_i$ can be chosen as the distinguishing index:
\begin{quote}
 If $f \in S_i$ for some $i$,
 then there exists some $j \not = i$
 such that $g \in S_j \setminus S_i$
 and $f \in  S_i \setminus S_j$.
\end{quote}

In particular,
neither $f$ nor $g$ can be identically~$0$.

We will prove the No-Mixing theorem-$(i,j)$ 
in a reverse lexicographic order of $(i,j)$:
We order the statements as
$(5,4), (5,3), (5,2), (5,1), (4,3), (4,2), (4,1), (3,2), (3,1), (2,1)$.
After having proved all No-Mixing theorem-$(i',j')$    
preceding $(i,j)$ in this order,
we assume there are two signatures $f$ and $g$ 
such that $f \in S_i \setminus S_j$ and $g \in S_j \setminus S_i$.
Now we \emph{may} make the following additional assumption: 
\[
 f, g \notin \bigcup_{i<k \le 5} S_k
 \qquad \text{ and } \qquad
 g \notin \bigcup_{j<k \le i} S_k.
\]
Indeed, if $f$ or $g$ belongs to $S_k$ for some $k > i$,
then let $k$ be the maximum index such that $S_k$ contains either $f$ or $g$.
Then by the observation above,
there exists some $j \neq k$ such that one signature belongs to $S_j \setminus S_k$,
and the other one belongs to $S_k \setminus S_j$.
By the maximality of $k$,
we have $k > j$.
Since $k > i$ and No-Mixing theorem-$(k,j)$ has already been proved,
we have $\PlCSP^2(f, g)$ is \numP-hard.
Moreover, if 
$g \in \bigcup_{j< \ell \le i} S_\ell$,
then $g \in S_\ell$ for some $j<\ell < i$,  as $g \notin S_i$.
Then  $f \in S_i \setminus S_\ell$ since $\{f, g \} \not \subseteq S_\ell$, 
and also $g \in S_\ell \setminus S_i$.
Hence $\PlCSP^2(f, g)$ is \numP-hard
by No-Mixing-$(i,k)$ already proved.

\vspace{.2in}

% in Theorem \ref{mixing-theorem-2} and Theorem \ref{mixing-theorem}.

%We also denote by
%$S_1=\mathscr{P}, S_2=\mathscr{A}, S_3=2$-$\mathscr{A}$, $S_4=\hat{M}$,
%and  $S_5=2$-$\hat{M}$.
%%% To Zhiguo: we will figure out a better notation later. --JYC
%%% I suggest \mathscr{P}, \mathscr{A}, \mathscr{A}^{\dagger},
%%% \widehat{\mathscr{M}}, \widehat{\mathscr{M}}^{\dagger}
%%% \widetilde{\mathscr{A}} = \mathscr{A} \cup \mathscr{A}^{\dagger}
% \widetilde{\mathscr{M}} = \widehat{\mathscr{M}} \cup \widehat{\mathscr{M}}^{\dagger}
%\subsection{No-Mixing of two signatures}
%In this subsection, we give the following No-Mixing Theorem of two signatures of even arity.

We now proceed with this plan. We first prove a preliminary result,
which allows us to construct signatures of arbitrarily high even arities from a given binary signature.

\begin{lemma} \label{mixing-arity-expand}
 For any binary signature $[a, b, c]$, any integer $k \ge 1$,
 and any signature set $\mathcal{F}$,
 \[
  \PlCSP^2([a, b]^{\otimes 2k}+[b, c]^{\otimes 2k},\mathcal{F})
  \leq_T
  \PlCSP^2([a, b, c],\mathcal{F}).
 \]
\end{lemma}

\begin{proof}
We take $2k$ copies of $[a, b, c]$
and connect one input of each $[a, b, c]$ to an edge of $=_{2k}$.
The resulting signature is
$[a, b]^{\otimes 2k}+[b, c]^{\otimes 2k}$,
since
$\left[\begin{smallmatrix} a & b \\
b & c \end{smallmatrix}\right]^{\otimes n}
\left( \left[\begin{smallmatrix} 1 \\ 0 \end{smallmatrix}\right]^{\otimes n}
+ \left[\begin{smallmatrix} 0 \\ 1 \end{smallmatrix}\right]^{\otimes n}
\right) = \left[\begin{smallmatrix} a \\ b \end{smallmatrix}\right]^{\otimes n}
+ \left[\begin{smallmatrix} b \\ c \end{smallmatrix}\right]^{\otimes n}$.
\end{proof}

In the next lemma, we will prove that for any
symmetric signature $f\in \mathscr{A}\setminus\mathscr{P}$ of even arity,
we can construct an arity 4 signature $g\in \mathscr{A}\setminus\mathscr{P}$
in Pl-$\#${\rm CSP}$^2(\{f\}\cup\mathcal{F})$.
Thus we can assume that we have an arity 4 signature $g\in \mathscr{A}\setminus\mathscr{P}$
  in the proof of the No-Mixing lemma of $\mathscr{P}$ versus $\mathscr{A}$,
namely  No-Mixing-$(5,4)$.
We can prove a similar result for $\mathscr{A}^{\dagger}\setminus\mathscr{P}$.
This is for the proof of No-Mixing-$(5,3)$.

\begin{lemma}\label{mixing-arity-reduction-A}
%%% i think we use \widetilde{scriptA} as union of S2 and S3.
For any symmetric signature
 $f\in \mathscr{A}\setminus\mathscr{P}$ 
(respectively, 
$f\in \mathscr{A}^{\dagger}\setminus\mathscr{P}$)
of even arity $2n \ge 2$,
there exists a symmetric signature
$g \in \mathscr{A}\setminus\mathscr{P}$
(respectively,
$g \in \mathscr{A}^{\dagger}\setminus\mathscr{P}$)
 of arity 4, such that for any set  $\mathcal{F}$,
\[\PlCSP^2(\{g\}\cup\mathcal{F})\leq_T
\PlCSP^2(\{f\}\cup\mathcal{F}).\]
%for any set  $\mathcal{F}$.
%
%Similarly,
%for any symmetric
%$f\in \mathscr{A}^{\dagger}\setminus\mathscr{P}$ of even arity $2n \ge 2$,
%there exists a symmetric
%$g \in \mathscr{A}^{\dagger}\setminus\mathscr{P}$ of arity 4, such that
%$\PlCSP^2( \{g\}\cup\mathcal{F})\leq$
%$\PlCSP^2(\{f\}\cup\mathcal{F})$,
%for any set of signatures $\mathcal{F}$.
%%
%%Moreover, for $g\in \widetilde{\mathscr{A}}\setminus\mathscr{P}$,
%%Pl-$\#${\rm CSP}$^2( \{g'\}\cup\mathcal{F})\leq$
%%Pl-$\#${\rm CSP}$^2(\{g\}\cup\mathcal{F})$,
%%where $g'\in \widetilde{\mathscr{A}}\setminus\mathscr{P}$
%has arity 4.
%%%% JYC: why add this \widetilde{\mathscr{A}}?
%%%% isn't the above two statements for about A and A^\dagger just so?
%%%%%
%%% added
\end{lemma}
\begin{proof}
If $f$ has arity $2n= 4$, then there is nothing to prove. 
Suppose  $2n \not =4$.
For $f\in \mathscr{A}^{\dagger}\setminus\mathscr{P}$,
if $2n=2$, then $f=[1, \alpha, -\alpha^2]$ by Corollary~\ref{binary-necessary}.
%Lemma \ref{binary}.
%%% more specifically cor 0.1
By Lemma~\ref{mixing-arity-expand}, we have
$g=[1, \alpha]^{\otimes 4}-[1, -\alpha]^{\otimes 4}$, since $\alpha^4 = -1$.
Clearly $g \in \mathscr{A}^{\dagger}$ and is non-degenerate.
Note that $g$ satisfies a second recurrence relation of type
 $\langle  -\alpha^2 ,0, 1 \rangle$, since the eigenvalues
of the recurrence are $\pm \alpha$ with trace 0 and product $-\alpha^2$.
Thus
$g \notin \mathscr{P}$ by Lemma~\ref{second-recurrence-relation}.
For $2n\geq 6$,
we have $f=[1, \alpha]^{\otimes 2n}+i^{r}[1, -\alpha]^{\otimes 2n}$ by
definitions (see Figure~\ref{fig:venn_diagram:A_Adagger_P}).
Then by our calculus,
we have $\partial^{n-2}(f)=(1+\alpha^2)^{n-2}\{[1, \alpha]^{\otimes 4}+i^{r}[1, -\alpha]^{\otimes 4}\}$.
Clearly it is in $\mathscr{A}^{\dagger}$ and is non-degenerate.
%has arity 4.
It also has type $\langle -\alpha^2, 0,  1 \rangle$
and therefore
 it is not in $\mathscr{P}$.

For $f\in \mathscr{A}\setminus\mathscr{P}$,
if $2n=2$, then $f=[1, \rho, -\rho^2]$
by Corollary~\ref{binary-necessary}.
% Lemma \ref{binary}.
By Lemma~\ref{mixing-arity-expand}, we have $g=[1, \rho]^{\otimes 4}+[1, -\rho]^{\otimes 4}$,
since $\rho^4 = 1$.
Clearly $g \in \mathscr{A}$ and is non-degenerate.
Note that $g$ has type
$\langle  -\rho^2, 0, 1 \rangle$, since the eigenvalues
of its second recurrence relation
 are $\pm \rho$ with trace 0 and product $-\rho^2$.
Thus $g \notin \mathscr{P}$ by Lemma~\ref{second-recurrence-relation}.

For $2n\geq 6$, we have $f=[1, \rho]^{\otimes 2n}+i^r[1, -\rho]^{\otimes 2n}$ by definitions (see Figure~\ref{fig:venn_diagram:A_Adagger_P}).
If $2n\equiv 0\pmod 4$, then
$n$ is even, and we have
$\partial^{\frac{n-2}{2}}_{=_4}(f)=2^{\frac{n-2}{2}}\{[1, \rho]^{\otimes 4}+i^r[1, -\rho]^{\otimes 4}\}$
that is in $\mathscr{A}$, and not in $\mathscr{P}$ by its type
$\langle -\rho^2, 0, 1 \rangle$.
For $2n\equiv 2\pmod 4$,  we have $h=\partial^{\frac{n-1}{2}}_{=_4}(f)=2^{\frac{n-1}{2}}\{[1, \rho]^{\otimes 2}+i^r[1, -\rho]^{\otimes 2}\}$.
\begin{itemize}
\item If $r=2$, then we have $h=2^{\frac{n-1}{2}}[0, 2\rho, 0]$.
Thus we have $[0, 1, 0]$ up to a nonzero scalar and
$\partial^{n-2}_{[0, 1, 0]}(f)
=(2\rho)^{n-2}\{[1, \rho]^{\otimes 4}+i^r(-1)^{n-2}[1, -\rho]^{\otimes 4}\}$
that is in $\mathscr{A}$, and not in $\mathscr{P}$ by its type
$\langle -\rho^2, 0, 1 \rangle$.
\item If $r\neq 2$, then $h=2^{\frac{n-1}{2}}(1+i^r)[1, \frac{1-i^r}{1+i^r}\rho, \rho^2]$.
Then we have $\partial_{[1, \frac{1-i^r}{1+i^r}\rho, \rho^2]}(=_4)=[1, 0, \rho^2]$ on LHS
and $\partial^{n-2}_{[1, 0, \rho^2]}(f)=2^{n-2}\{[1, \rho]^{\otimes 4}+i^r[1, -\rho]^{\otimes 4}\}$ by $\rho^4 =1$,
that is in $\mathscr{A}\setminus\mathscr{P}$ by the same reason.\qedhere
\end{itemize}
\end{proof}
We note that the complication for the
case $f \in \mathscr{A}\setminus\mathscr{P}$ is unavoidable since
if $\rho = \pm i$, then $\partial(f) =0$, therefore we need to use
$\partial_{=_4}(f)$.

\subsection{Mixing with \texorpdfstring{$\mathscr{P}$}{P}}

In this subsection, we prove 
No-Mixing-$(5,j)$, for $1 \le j \le 4$,
namely the No-Mixing of one
signature in $\mathscr{P}$ and another signature 
in a different tractable set.
Thus we assume there is some $f \in S_5 = \mathscr{P}$,
and some $g \in S_k$ for some $1 \le k \le 4$, and for no $1  \le k \le 5$,
$\{f, g\} \subset S_k$. Under this assumption we  show that
${\rm Pl}$-$\#{\rm CSP}^2(f, g)$ is \#P-hard.
%Thus when we prove the No-Mixing statement
%of  $\mathscr{P}$ and $\widetilde{\mathscr{A}}$, we can assume
%that  $f \in \mathscr{P}\setminus\widetilde{\mathscr{A}}$ and
%$g \in \widetilde{\mathscr{A}}\setminus\mathscr{P}$,  as well as
%$\{f, g\}\nsubseteq\widehat{\mathscr{M}}$,
%and $\{f, g\}\nsubseteq\widehat{\mathscr{M}}^\dagger$.
%By Lemma~\ref{mixing-arity-reduction-A}, we
%can obtain an arity 4 signature $g'$ that is in $\widetilde{\mathscr{A}}
%\setminus\mathscr{P}$ from $g$.
%From that we will construct, using a mixture of $f$ and $g$,
%some arity 4 signature that is not in any of the tractable
%classes.
%Then we apply Theorem~\ref{arity-4-dichotomy} to prove
%that
%${\rm Pl}$-$\#{\rm CSP}^2(f, g)$ is \#P-hard.
As explained earlier, for $j<k <5$,
when  we prove  No-Mixing-$(5,j)$,
we can make logical use of  No-Mixing-$(5,k)$.
\begin{lemma}\label{mixing-P}
Let $\{f, g\}\subseteq\bigcup_{k=1}^5 S_k$ and $\{f, g\}\nsubseteq S_j$ for
every $1\leq j\leq 5$.
Assume that $f\in\mathscr{P}$, then $\PlCSP^2(f, g)$ is $\#{\rm P}$-hard.
\end{lemma}
\begin{proof}
%By $f\in\mathscr{P}$, and
%$\{f, g\} \nsubseteq \mathscr{P}$,
 %we have $g\in(\widetilde{\mathscr{A}}\cup\widetilde{\mathscr{M}})\setminus \mathscr{P}$.
%Moreover, if $g\in S_k\setminus\mathscr{P}$ for
%some $1\leq k\leq 4$, then $f\in \mathscr{P}\setminus S_k$.
As explained earlier, since $f \in\mathscr{P}$,
there exists some $1\leq k\leq 4$, such that $g \in S_k\setminus\mathscr{P}$
and $f \in \mathscr{P}\setminus S_k$.
Since $[0, 1, 0] \in \bigcap_{k=1}^5 S_{k}$,
we know that $f$ is not a multiple of $[0, 1, 0]$.
%\neq$ [0, 1, 0]$. 
Then by $f \in \mathscr{P}$ 
 (see Figure~\ref{fig:venn_diagram:A_Adagger_P}), we have
%%%%
%%% JYC
%%%% this "By FIG1" has to be justified by lm. after all just because
%% it is stated in a diagram is no proof.
%%%%%
%$\mathscr{P}\cap\mathscr{A}\cap\mathscr{A}^\dagger\cap\widehat{\mathscr{M}}\cap\widehat{\mathscr{M}}^\dagger$.
%Thus
$f=[a, b]^{\otimes 2n}$ with $a$ and $b$ not both 0 (because $f$ is
not identically $0$),
 or $f = [1, 0, \ldots, 0, x]$ with $x\neq 0$.

We first consider the case $f=[a, b]^{\otimes 2n}$, with $(a,b) \not = (0,0)$.
It has three subcases.
\begin{itemize}
\item If $ab\neq 0$ (i.e., $a$ and $b$ both nonzero)
 and $a^2+b^2\neq 0$, then we have
$\partial^{n-1}(f)=(a^2+b^2)^{n-1}[a, b]^{\otimes 2}$. We are done by Lemma~\ref{mixing-P-global}.

\item If $ab\neq 0$ and  $a^2+b^2=0$,
then $f=[1, \pm i]^{\otimes 2n}$ up to a nonzero scalar.
Note that $f\in\mathscr{P}\cap\mathscr{A}\cap\widehat{\mathscr{M}}^\dagger$.
%Thus $g\in(\mathscr{A}^\dagger\cup\widehat{\mathscr{M}})\setminus(\mathscr{P}\cup\mathscr{A}\cup\widehat{\mathscr{M}}^\dagger)$.
Hence  $g\in\mathscr{A}^\dagger\setminus(\mathscr{P}\cup\mathscr{A}\cup\widehat{\mathscr{M}}^\dagger)$
or
$g\in\widehat{\mathscr{M}}\setminus(\mathscr{P}\cup\mathscr{A}\cup\widehat{\mathscr{M}}^\dagger)$.
%or $g\in\mathscr{A}^\dagger\setminus(\mathscr{P}\cup\mathscr{A}\cup\widehat{\mathscr{M}}^\dagger)$

%%%%%%% JYC I commented out the next paragraph. seems unneeded.
%% and switched the order of the two cases below
%%%%%%%%%%%%%%%%
%Note that $\mathscr{A}^\dagger\setminus\mathscr{P}=\{[1, \alpha, -\alpha^2], [1, \alpha]^{\otimes 2n}+i^r[1, -\alpha]^{\otimes 2n} \mid 2n\geq 4\}$
%by Figure~\ref{fig:venn_diagram:A_Adagger_P}.
%Thus $(\mathscr{A}^\dagger\setminus\mathscr{P})\cap(\mathscr{A}\cup\widehat{\mathscr{M}}^\dagger)=\emptyset$.
%So $\mathscr{A}^\dagger\setminus(\mathscr{P}\cup\mathscr{A}\cup\widehat{\mathscr{M}}^\dagger)=
%\mathscr{A}^\dagger\setminus\mathscr{P}$.
%This implies that $g\in\mathscr{A}^\dagger\setminus(\mathscr{P}\cup\mathscr{A}\cup\widehat{\mathscr{M}}^\dagger)$
%iff $g\in\mathscr{A}^\dagger\setminus\mathscr{P}$.

 If $g\in\widehat{\mathscr{M}}\setminus(\mathscr{P}\cup\widetilde{\mathscr{A}}\cup\widehat{\mathscr{M}}^\dagger)$,
then a fortiori,
$g\in\widehat{\mathscr{M}}\setminus(\mathscr{P}\cup\widetilde{\mathscr{A}})$. 
Therefore we
are done by Lemma~\ref{M-even}.

The other case is $g\in\mathscr{A}^\dagger\setminus(\mathscr{P}\cup\mathscr{A}\cup\widehat{\mathscr{M}}^\dagger)$,
then a fortiori, $g\in\mathscr{A}^\dagger\setminus\mathscr{P}$,
and
by Lemma \ref{mixing-arity-reduction-A},
 we have an arity 4 signature
 $g'\in\mathscr{A}^\dagger\setminus\mathscr{P}$.
 By definition (see Figure~\ref{fig:venn_diagram:A_Adagger_P}),
%we have $g'\in\mathscr{A}^\dagger\setminus\mathscr{P}$ and \textcolor{red}{FIG1},
$g'=[1, \alpha]^{\otimes 4}+i^r[1, -\alpha]^{\otimes 4}$.
For $r=2$, we have $\partial(g')=2\alpha(1+\alpha^2)[0, 1, 0]$ and $\partial^{n-1}_{[0, 1, 0]}(f)=(\pm 2i)^{n-1}[1, \pm i]^{\otimes 2}$.
Then we are done by Lemma \ref{mixing-P-global}.
For $r\neq 2$, we have on LHS
\[\partial_{g'}(=_6)
= \partial_{[1, \alpha]^{\otimes 4}}(=_6) + i^r
\partial_{[1, -\alpha]^{\otimes 4}}(=_6)
= [1, 0, \alpha^4] + i^r [1, 0, (-\alpha)^4]
=(1+i^r)[1, 0, -1]\]
 and
$\partial^{n-1}_{[1, 0, -1]}(f)=2^{n-1}[1, \pm i]^{\otimes 2}$.
Then again we are done by Lemma \ref{mixing-P-global}.

\item For $f=[1, 0]^{\otimes 2n}$ or $[0, 1]^{\otimes 2n}$,
we have $\partial^{n-1}(f) = [1, 0]^{\otimes 2}$ or $[0, 1]^{\otimes 2}$.
Note that $f\in\mathscr{P}\cap\mathscr{A}\cap\mathscr{A}^\dagger$.
Thus $g\in\widetilde{\mathscr{M}}\setminus(\mathscr{P}\cup\widetilde{\mathscr{A}})$.
%or $g\in\widehat{\mathscr{M}}^\dagger\setminus(\mathscr{P}\cup\widetilde{\mathscr{A}})$.
%
If $g\in\widehat{\mathscr{M}}\setminus(\mathscr{P}\cup\widetilde{\mathscr{A}})$,
then we are done by Lemma \ref{M-even}.
If $g\in\widehat{\mathscr{M}}^\dagger\setminus(\mathscr{P}\cup\widetilde{\mathscr{A}})$,
then we are done by Lemma \ref{2-M-even}, where the binary signature
is supplied by $\partial^{n-1}(f) = [1, 0]^{\otimes 2}$ or $[0, 1]^{\otimes 2}$.
%For other cases, we are done by Lemma \ref{M-even} and Lemma \ref{2-M-even}.
\end{itemize}

The remaining case is
$f=[1, 0, \ldots, 0, x]$ with $x\neq 0$. We have $\partial^{n-1}(f)=[1, 0, x]$.
%\begin{itemize}
%\item

Suppose $g\in\widetilde{\mathscr{A}}$.
As $f \in \mathscr{P}$, we have $g \not \in \mathscr{P}$.
Then we have
an arity 4 signature $g'\in\widetilde{\mathscr{A}}\setminus \mathscr{P}$ by Lemma \ref{mixing-arity-reduction-A}.
Moreover, by definition (see Figure~\ref{fig:venn_diagram:A_Adagger_P}), we have
$g'=[1, \gamma]^{\otimes 4}+i^r[1, -\gamma]^{\otimes 4}$
where $\gamma^8=1$.
Depending on whether $g \in \mathscr{A}$
or $\mathscr{A}^\dagger$, we have
either
$f\in \mathscr{P}\setminus\mathscr{A}$, or $f\in \mathscr{P}\setminus\mathscr{A}^\dagger$.
Then we claim that $x^4\neq 1$.
Note that $f$ has even arity $2n$. If $x^4 =1$, then
$f = [1, 0, \ldots, 0, x] \in \mathscr{A}$ as well as
$\left[\begin{smallmatrix} 1 & 0 \\
0 & \alpha \end{smallmatrix}\right]^{\otimes 2n} f
%[1, 0, \ldots, 0, x]
%= [1, 0, \ldots, 0, x\alpha^{2n}]
= [1, 0, \ldots, 0, xi^n] \in \mathscr{A}$ thus $f \in \mathscr{A}^\dagger$.
This is a contradiction.
 Thus we have $x^4\neq 0, 1$.
% by Figure~\ref{fig:venn_diagram:A_Adagger_P} (note that
%$f$ has even arity).
%
%and $g=[1, \rho]^{\otimes 4}+i^r[1, -\rho]^{\otimes 4}$ by Lemma \ref{mixing-arity-reduction-A}.
Let $\widehat{g'}=[1, x^{-\frac{1}{2}}\gamma]^{\otimes 4}+i^r[1, -x^{-\frac{1}{2}}\gamma]^{\otimes 4}$.
Then by Lemma \ref{mixing-P-binary},
$\PlCSP^2(\widehat{g'})\leq$$\PlCSP^2(f, g)$.
Note that
%$(x^{-\frac{1}{2}}\gamma)^8\neq 1$ and
$\widehat{g'}$
%satisfies
% the second recurrence relation
has type
$\langle -x^{-1}\gamma^2, 0,  1 \rangle$
by calculating the trace and product of the eigenvalues of the
second recurrence relation.
Note that $(-x^{-1}\gamma^2)^4=x^{-4}\neq 0, 1$.
Thus
%it differs the types of the signatures in
$\widehat{g'} \not \in
\mathscr{P}\cup\widetilde{\mathscr{A}}\cup\widetilde{\mathscr{M}}$
by Lemma~\ref{second-recurrence-relation}.
This implies that
%$\hat{g}\notin\mathscr{P}\cup\widetilde{\mathscr{A}}\cup\widetilde{\mathscr{M}}$ and
$\PlCSP^2(\widehat{g'})$ is $\#{\rm P}$-hard by Theorem \ref{arity-4-dichotomy}.
So $\PlCSP^2(f, g)$ is $\#{\rm P}$-hard.
%
%\item $f\in \mathscr{P}\setminus\mathscr{A}^\dagger$,
%$g\in\mathscr{A}^\dagger\setminus \mathscr{P}$,
%then $x^4\neq 0, 1$ and $g=[1, \alpha]^{\otimes 4}+i^r[1, -\alpha]^{\otimes 4}$ by Lemma \ref{mixing-arity-reduction-A}.
%Let $g'=[1, x^{-\frac{1}{2}}\alpha]^{\otimes 4}+i^r[1, -x^{\frac{1}{2}}\alpha]^{\otimes 4}$.
%By Lemma \ref{mixing-P-binary},
%$\PlCSP^2(g')\leq$$\PlCSP^2(f, g)$.
%Note that $(x^{-\frac{1}{2}}\alpha)^8\neq 1$.
% Thus $g'\notin(\mathscr{P}\cup\widetilde{\mathscr{A}}\cup\widetilde{\mathscr{M}})$ and
%$\PlCSP^2(g')$ is $\#{\rm P}$-hard by Theorem \ref{arity-4-dichotomy}.
%So $\PlCSP^2(f, g)$ is $\#{\rm P}$-hard.
%
%Pl-$\#$CSP$(g')$ is hard, where $g'=[u, x^{-\frac{1}{2}}v, x^{-1}w, x^{-\frac{3}{2}}v, x^{-2}u]$
%or $[u,  x^{-\frac{1}{2}}v, 0, -x^{-\frac{3}{2}}v, -x^{-2}u]$. Thus Pl-$\#$CSP$^2(f, g)$ is hard by Lemma \ref{mixing-P-binary}.
%
%
%\item

Now we may assume that $g\notin\widetilde{\mathscr{A}}$.
Thus $g\in\widetilde{\mathscr{M}}\setminus(\mathscr{P}\cup\widetilde{\mathscr{A}})$.
If
$g\in\widehat{\mathscr{M}}\setminus (\mathscr{P}\cup\widetilde{\mathscr{A}})$,
then we are done by Lemma \ref{M-even}.
If
$g\in\widehat{\mathscr{M}}^\dagger\setminus (\mathscr{P}\cup\widetilde{\mathscr{A}})$,
then $f\in \mathscr{P}\setminus\widehat{\mathscr{M}}^\dagger$. 
In this case we claim that $x \neq 1$.
Suppose for a contradiction that $x=1$, then we show that
$f \in \widehat{\mathscr{M}}^\dagger$. Notice that
$f = [1, 0, \ldots, 0, 1] = (=_{2n})$ and
$\widehat{\mathscr{M}}^\dagger =  Z \mathscr{M}$,
where $Z = \left[\begin{smallmatrix} 1 & 1 \\
i & -i \end{smallmatrix}\right]$.
Crucially recall that $f$
has even arity. 
Then, up to a nonzero
scalar, $(Z^{-1})^{\otimes 2n} f = 
[1,0,1, \ldots, 0,1] \in \mathscr{M}$ of arity $2n$ (if $n$ is even) or
$(Z^{-1})^{\otimes 2n} f =
[0,1,0, \ldots,1,0] \in \mathscr{M}$ of arity $2n$ (if $n$ is odd).
%%% 1/2 ( [1\\1]^2n + (-i)^{2n} [ 1\\-1]^{2n} )
%%% = 1/2 ( [1 1 1...] + (-1)^n [1 -1 1 -1... 1]  end in +1, as even arity
%%%
%%% n even  = 1/2 ( [1 1 1...] + [1 -1 1 -1... 1] ) = [1,0,1,0,...,1] \in M
%%% n odd   = 1/2 ( [1 1 1...] - [1 -1 1 -1... 1]) = [0 1 0 1 ... 1 0 ] \in M
%
%%% JYC I think only x \not = 1 is needed.
Hence $x \neq 1$. Then we are done by Lemma \ref{2-M-even},
with 
$g\in\widehat{\mathscr{M}}^\dagger\setminus (\mathscr{P}\cup\widetilde{\mathscr{A}})$, and the help of $\partial^{n-1}(f)=[1, 0, x]$.
\end{proof}

\subsection{Mixing with \texorpdfstring{$\mathscr{A}$}{A}}

In this subsection, we prove the No-Mixing lemma of $\mathscr{A}$ with
 other tractable sets.
Because we have
already proved Lemma~\ref{mixing-P},
the No-Mixing lemma for $S_5 = \mathscr{P}$,
 we only need to consider No-Mixing-$(4,j)$
 of $S_4 = \mathscr{A}$
with $S_j$ for $1 \le j \le 3$.

There is a  particular case involving
$\mathscr{A}$  and  $\mathscr{A}^{\dagger}$
 that requires some special care. This is when two signatures
 $f\in\mathscr{A}$ and $g\in\mathscr{A}^{\dagger}$
both satisfy the parity constraint.
We deal with this case first.
Furthermore, by Lemma~\ref{mixing-arity-reduction-A},
for two signatures $f \in \mathscr{A} \setminus \mathscr{P}$ and
$g \in \mathscr{A}^{\dagger}\setminus\mathscr{P}$
 we may assume the
signatures $f$ and $g$ have arity 4.
Hence the next lemma considers signatures $f$ and $g$ of arity 4.
%
%In this case, the binary it is difficult to
%construct a binary  signature  from $f$ and $g$
%together with signatures from ${\cal EQ}_2$
%% $=_{2k}$
%will also satisfy the parity constraint, and therefore it
%is always in $\mathscr{A}\cap\mathscr{A}^{\dagger}$.
%%%% JYC why? i don't see tis form in [1, 0, \rho] must have \rho^4=1, say.
%This implies that we have to prove the lemma for the signatures of arity at least 4.

\begin{figure}[t]
 \centering
 \subfloat[]{
  \begin{tikzpicture}[scale=\scale,transform shape,node distance=\nodeDist,semithick]
   \node[internal] (0)                    {};
   \node[external] (1) [above left  of=0] {};
   \node[external] (2) [below left  of=0] {};
   \node[external] (3) [left        of=1] {};
   \node[external] (4) [left        of=2] {};
   \node[triangle] (5) [right       of=0] {};
   \node[external] (6) [above right of=5] {};
   \node[external] (7) [below right of=5] {};
   \node[external] (8) [right       of=6] {};
   \node[external] (9) [right       of=7] {};
   \path (3) edge[in= 135, out=   0,postaction={decorate, decoration={
                                               markings,
                                               mark=at position 0.45 with {\arrow[>=diamond, white] {>}; },
                                               mark=at position 0.45 with {\arrow[>=open diamond]   {>}; } } }] (0)
         (0) edge[out=-135, in=   0] (4)
             edge[bend left]         (5)
             edge[bend right]        (5)
         (5) edge[out=  45, in= 180] (8)
             edge[out= -45, in= 180] (9);
   \begin{pgfonlayer}{background}
    \node[draw=\borderColor,thick,rounded corners,fit = (1) (2) (6) (7),inner sep=0pt] {};
%     \node[draw=\borderColor,thick,rounded corners,fit = (1) (2) (6) (7),inner sep=0pt,transform shape=false] {};
   \end{pgfonlayer}
  \end{tikzpicture}
  \label{subfig:7-1}
 }
 \qquad
 \subfloat[]{
  \begin{tikzpicture}[scale=\scale,transform shape,node distance=\nodeDist,semithick]
   \node[internal]  (0)                    {};
   \node[external]  (1) [above left  of=0] {};
   \node[external]  (2) [below left  of=0] {};
   \node[external]  (3) [left        of=1] {};
   \node[external]  (4) [left        of=2] {};
   \node[triangle]  (5) [right       of=0] {};
   \node[internal]  (6) [right       of=5] {};
   \node[external]  (7) [above right of=6] {};
   \node[external]  (8) [below right of=6] {};
   \node[external]  (9) [right       of=7] {};
   \node[external] (10) [right       of=8] {};
   \path (3) edge[in= 135, out=   0,postaction={decorate, decoration={
                                               markings,
                                               mark=at position 0.45 with {\arrow[>=diamond, white] {>}; },
                                               mark=at position 0.45 with {\arrow[>=open diamond]   {>}; } } }] (0)
         (0) edge[out=-135, in=   0]  (4)
             edge[bend left]          (5)
             edge[bend right]         (5)
         (5) edge[bend left]          (6)
             edge[bend right]         (6)
         (6) edge[out=  45, in= 180]  (9)
             edge[out= -45, in= 180] (10);
   \begin{pgfonlayer}{background}
    \node[draw=\borderColor,thick,rounded corners,fit = (1) (2) (7) (8),inner sep=0pt] {};
%     \node[draw=\borderColor,thick,rounded corners,fit = (1) (2) (7) (8),inner sep=0pt,transform shape=false] {};
   \end{pgfonlayer}
  \end{tikzpicture}
  \label{subfig:7-2}
 }
 \caption{Two gadgets used in the proof of Lemma~\ref{mixing-A-2-A-parity}.}
 \label{fig:7}
\end{figure}
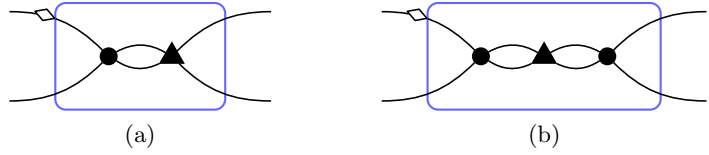

\begin{lemma} \label{mixing-A-2-A-parity}
 Let
 $f = [1, \rho]^{\otimes 4}   \pm [1, -\rho]^{\otimes 4} \in \mathscr{A}$ and
 $g = [1, \alpha]^{\otimes 4} \pm [1, -\alpha]^{\otimes 4} 
\in \mathscr{A}^{\dagger}$.
Then $\PlCSP^2(f, g)$ is \numP-hard.
\end{lemma}

\begin{proof}
There are four cases depending on the combination of the two $\pm$ signs.
 Suppose
 $f = [1, \rho]^{\otimes 4}   + [1, -\rho]^{\otimes 4}$ and
 $g = [1, \alpha]^{\otimes 4} + [1, -\alpha]^{\otimes 4}$.
 Consider the gadget in Figure~\ref{subfig:7-1}.
 We assign $g$ to the circle vertex and $f$ to the triangle vertex.
 Since both $f = 2 [1, 0, \rho^2, 0, 1]$ and $g = 2[1, 0, \alpha^2, 0, -1]$ have even parity,
 the signature of this gadget also has even parity.
 It is also clearly a redundant signature by design.
 Hence there are only five signature entries we need to compute.
 E.g., the entry of Hamming weight~0 is $g_0f_0 + g_2 f_2 = 4 (1 + \alpha^2 \rho^2)$.
 Up to a factor of~$4$,
 the signature of this gadget has signature matrix
 \[
  \begin{bmatrix}
   \alpha^2 \rho^2 + 1 & 0 & 0 & \alpha^2 + \rho^2 \\
   0 & 2 \alpha^2 \rho^2 & 2 \alpha^2 \rho^2 & 0 \\
   0 & 2 \alpha^2 \rho^2 & 2 \alpha^2 \rho^2 & 0 \\
   \alpha^2 - \rho^2 & 0 & 0 & \alpha^2 \rho^2 - 1
  \end{bmatrix},
  \text{ which becomes }
  \begin{bmatrix}
   \alpha^2 \rho^2 + 1 & 0  & 0 & 2 \alpha^2 \rho^2 \\
   0 & \alpha^2 - \rho^2 & 2 \alpha^2 \rho^2 & 0 \\
   0 & 2 \alpha^2 \rho^2 & \alpha^2 + \rho^2 & 0 \\
   2 \alpha^2 \rho^2 & 0 & 0 & \alpha^2 \rho^2 - 1
  \end{bmatrix}
 \]
 after a $90^\circ$ counterclockwise rotation of the gadget.
 (See Figure~\ref{fig:rotate_asymmetric_signature} in Part~I for an illustration of the rotation operation.)
 Taking the four corner entries,
 we define the binary signature
 $h = [\alpha^2 \rho^2 + 1, 2 \alpha^2 \rho^2, \alpha^2 \rho^2 - 1]$.
By domain pairing, $\PlCSP(h) \leq_T \PlCSP^2(f, g)$.
(Domain pairing is the following reduction: 
In an instance of  $\PlCSP(h)$
replace every occurrence of $h$ by a copy of the  $90^\circ$ counterclockwise
 rotated gadget, and replace both edges of $h$
 by two parallel edges each,
and replace every $(=_k)$ in the $\PlCSP(h)$ instance
by $(=_{2k})$ in $\PlCSP^2(f, g)$.
Note that the rotation is necessary to create a \emph{symmetric}
binary signature $h$ in the paired domain.)

Note that $\alpha^2 = \pm i$ and $\rho^2 = \pm 1$,
so $\alpha^2\rho^2  \pm 1$ has norm $\sqrt{2}$, while $2\alpha^2\rho^2$
has norm 2.
Also $\alpha^2\rho^2 +1 \not = \alpha^2\rho^2 -1$.
Hence $h\notin\mathscr{P}\cup \mathscr{A}$ by
Corollary~\ref{binary-necessary} and also $h\notin \widehat{\mathscr{M}}$ by
 Lemma~\ref{binary}.
%%% Corollary item 1 2 for P. A. Lemma item 4 for M-hat
Thus $\PlCSP(h)$ is $\#{\rm P}$-hard by Theorem~\ref{pl-dicho-1}.
So $\PlCSP^2(f, g)$ is $\#{\rm P}$-hard.

\vspace{.1in}
 Suppose
 $f = [1, \rho]^{\otimes 4}   - [1, -\rho]^{\otimes 4}$ and
 $g = [1, \alpha]^{\otimes 4} - [1, -\alpha]^{\otimes 4}$.
 Consider the same construction.
 Up to a nonzero factor of $4 \alpha \rho$,
 the signature of this gadget has the signature matrix
 \[
  \begin{bmatrix}
   2 & 0 & 0 & 2 \rho^2 \\
   0 & 1 + \alpha^2 \rho^2 & 1 + \alpha^2 \rho^2 & 0 \\
   0 & 1 + \alpha^2 \rho^2 & 1 + \alpha^2 \rho^2 & 0 \\
   2 \alpha^2 & 0 & 0 & 2 \alpha^2 \rho^2
  \end{bmatrix},
  \text{ which becomes }
  \begin{bmatrix}
   2 & 0 & 0 & 1 + \alpha^2 \rho^2 \\
   0 & 2 \alpha^2 & 1 + \alpha^2 \rho^2 & 0 \\
   0 & 1 + \alpha^2 \rho^2 & 2 \rho^2 & 0 \\
   1 + \alpha^2 \rho^2 & 0 & 0 & 2 \alpha^2 \rho^2
  \end{bmatrix}
 \]
 after a $90^\circ$ counterclockwise rotation of the gadget.
 Let $h = [2, 1 + \alpha^2 \rho^2, 2 \alpha^2 \rho^2]$.
 By domain pairing,
 we have $\PlCSP(h) \leq_T \PlCSP^2(f, g)$.
 Note that $1 + \alpha^2 \rho^2 = 1 \pm i$ has norm $\sqrt{2}$
 %while $2 \alpha^2 \rho^2 \neq \pm 2$ but has norm~$2$.
%%% i don't need to say \neq -2
%%% in lm A.7 for M-hat, since here middle entry is nonzero
while $2 \alpha^2 \rho^2 \neq 2$ but has norm~$2$.
 Hence $h \notin \mathscr{P} \cup \mathscr{A} \cup \widehat{\mathscr{M}}$
 by Corollary~\ref{binary-necessary} and Lemma~\ref{binary}.
 Thus we are done by Theorem~\ref{pl-dicho-1}.

\vspace{.1in}
 Suppose
 $f = [1, \rho]^{\otimes 4}   - [1, -\rho]^{\otimes 4}$ and
 $g = [1, \alpha]^{\otimes 4} + [1, -\alpha]^{\otimes 4}$.
 Consider the gadget in Figure~\ref{subfig:7-2}.
 We assign $f$ to the circle vertices and $g$ to the triangle vertex.
 Up to a nonzero factor of $16 \alpha^2 \rho^2$,
 the signature of this gadget has the signature matrix
 \[
  \begin{bmatrix}
   2 & 0 & 0 & 2 \rho^2 \\
   0 & \rho^2 & \rho^2 & 0 \\
   0 & \rho^2 & \rho^2 & 0 \\
   2 \rho^2 & 0 & 0 & 2
  \end{bmatrix},
  \qquad
  \text{which becomes}
  \qquad
  \begin{bmatrix}
   2 & 0 & 0 & \rho^2 \\
   0 & 2 \rho^2 & \rho^2 & 0 \\
   0 & \rho^2 & 2 \rho^2 & 0 \\
   \rho^2 & 0 & 0 & 2
  \end{bmatrix}
 \]
 after a $90^\circ$ rotation of the gadget.
 Let $h = [2, \rho^2, 2]$.
 We also have $g^{\times} = 2 [1, \alpha^2]^{\otimes 2}$ 
by domain pairing with $g$ (see Lemma~\ref{domain-pairing-expand}).
%%% I say "see" oinstead of "By" because technically
%%% this is a combination argument. not directly by the lm.
%%% why doesn't lm A18 directly apply? because A18 only apply to symm sig.
 Then $\PlCSP(g^{\times}, h) \leq_T \PlCSP^2(f, g)$.
 Note that $|\rho^2| = 1 \neq 2$,
 so by Lemma~\ref{binary} and Corollary~\ref{binary-necessary},
 $h \in \widehat{\mathscr{M}} \setminus (\mathscr{P} \cup \mathscr{A})$.
Also by Lemma~\ref{binary} and $(\alpha^2)^2 = -1 \not =1$ 
 we have $g^{\times} \notin \widehat{\mathscr{M}}$.
 Thus we are done by Theorem~\ref{pl-dicho-1}.
Note that in this case, the rotation is 
necessary to create a \emph{non-degenerate} binary signature $h$ in
the paired domain.

\vspace{.1in}
 Finally,
 suppose
 $f = [1, \rho]^{\otimes 4}   + [1, -\rho]^{\otimes 4}$ and
 $g = [1, \alpha]^{\otimes 4} - [1, -\alpha]^{\otimes 4}$.
 Consider the gadget in Figure~\ref{subfig:7-2}.
 We assign $g$ to the circle vertices and $f$ to the triangle vertex.
 Up to a nonzero factor of $16 \alpha^2 \rho^2$,
 the signature of this gadget has the signature matrix
 \[
  \begin{bmatrix}
   2 & 0 & 0 & 2 \alpha^2 \\
   0 & \alpha^2 & \alpha^2 & 0 \\
   0 & \alpha^2 & \alpha^2 & 0 \\
   2 \alpha^2 & 0 & 0 & -2
  \end{bmatrix},
  \qquad
  \text{which becomes}
  \qquad
  \begin{bmatrix}
   2 & 0 & 0 & \alpha^2 \\
   0 & \alpha^2 & 2 \alpha^2 & 0 \\
   0 & 2 \alpha^2 & \alpha^2 & 0 \\
   \alpha^2 & 0 & 0 & -2
  \end{bmatrix}
 \]
 after a $90^\circ$ rotation of the gadget.
 Let $h = [2, \alpha^2, -2]$,
 then $\PlCSP(h)\leq_T \PlCSP^2(f, g)$ by domain pairing.
 Since $|\alpha^2| = 1 \neq 2$,
 we have $h \notin \mathscr{P} \cup \mathscr{A}$ 
by Corollary~\ref{binary-necessary} and also
$h  \notin  \cup \widehat{\mathscr{M}}$ by Lemma~\ref{binary}.
 Thus we are done by Theorem~\ref{pl-dicho-1}.
Note that in this case, the rotation is also
necessary to create a \emph{non-degenerate} binary signature $h$ in
the paired domain.
\end{proof}

\begin{remark}
The use of a more complicated construction in the third case is necessary.
Notice that 
$g= [1, \alpha]^{\otimes 4}+[1, -\alpha]^{\otimes 4}
= 2 [1, 0, \alpha^2, 0, -1]$
 has an even parity, while $f=[1, \rho]^{\otimes 4}-[1, -\rho]^{\otimes 4}
= 2 \rho [0, 1, 0, \rho^2, 0]$ has an odd parity.
Then
in any construction  of a signature using $f$ and $g$,
%together with signatures from ${\cal EQ}_2$,
if the number of occurrences  $N_f$ of $f$ is odd (resp. even), 
then the resulting
signature also has an odd (resp. even) parity.
To see this,
let $H$ be an arbitrary $\{f, g\}$-gate with $N_f$  occurrences of $f$.
Suppose  $\sigma$  is a $\{0, 1\}$-assignment 
 to all the edges of $H$, including
internal and external edges, that has a nonzero evaluation on $H$.
  Then each copy of $f$ 
 has an odd  number of
incident edges assigned  to $1$.
Summing these numbers $\pmod 2$ over all copies of $f$ 
we get a number $\equiv N_f \pmod 2$, 
since each of these numbers is $\equiv 1 \pmod 2$.
Similarly each copy of $g$ has an even number of
incident edges assigned  to $1$.
Summing these numbers $\pmod 2$ over all copies of $g$
we get a number  $\equiv 0 \pmod 2$.
On the other hand,
if we add these two sums together we get 
$2 X + Y$ where $X$ is the number of internal edges
and $Y$ is the number of external edges assigned  to $1$ by  $\sigma$.
This is because each internal edge assigned  to $1$
 appears exactly twice in the sum.
Hence this number is $\equiv Y \pmod 2$.
We conclude that $N_f \equiv Y \pmod 2$, the Hamming weight
of $\sigma$ on the external edges. 

 If $N_f$ is odd, from any constructed signature of arity 4,
by rotation and domain pairing we can only get the identically
 zero binary signature.
Thus we must use  $f$ an even number of times.
Using $g$ alone will not get out of $\mathscr{A}^{\dagger}$,
which is a tractable set.
Thus we must use  $f$ at least twice.
Also using  $g$ alone will not get out of $\mathscr{A}$,
another tractable set. Therefore  we must use $g$ at least once.
 Therefore the construction
we give is  the simplest possible.

The same consideration applies for the
construction in the fourth case.
\end{remark}

\vspace{.1in}

%Except the cases in Lemma~\ref{mixing-A-2-A-parity},
%we always can construct a binary that is in $\mathscr{A}^{\dagger}\setminus\mathscr{P}$
%by an arity 4 signature in $\mathscr{A}\setminus\mathscr{P}$
%and an arity 4 signature in $\mathscr{A}^{\dagger}\setminus\mathscr{P}$.
The next Lemma deals with the
situation when we have a binary
signature in $\mathscr{A}\setminus\mathscr{P}$
and an arity 4 signature in $\mathscr{A}^{\dagger}\setminus\mathscr{P}$.

\begin{lemma}\label{mixing-A-binary-2-A}
Let
$f=[1, \rho, -\rho^2]$ and
$g=[1, \alpha]^{\otimes 4}+i^r[1, -\alpha]^{\otimes 4}$.
Then $\PlCSP^2(f, g)$ is $\#{\rm P}$-hard.
\end{lemma}
\begin{proof}
By our calculus, we have
$\partial_{[1, \rho, -\rho^2]}(g)
=\lambda [1, \alpha]^{\otimes 2}+i^r \mu [1, -\alpha]^{\otimes 2}$,
where $\lambda = 1-\rho^2\alpha^2+2\rho\alpha$ and $\mu
= 1-\rho^2\alpha^2-2\rho\alpha$.
Note that $1-\rho^2\alpha^2 = 1 \pm i$ has norm  $\sqrt{2}$
and $|2\rho\alpha| = 2$, we have $\lambda \neq 0$.
Let $x = i^r \mu/\lambda$, then
$\partial_{[1, \rho, -\rho^2]}(g) =
 \lambda (1+x) [1, \frac{1-x}{1+x} \alpha, \alpha^2]$.
By norm,
 $(1-\rho^2\alpha^2)^{4}\neq(2\rho\alpha)^4$ and
 $(1-\rho^2\alpha^2)(2\rho\alpha)\neq 0$,
we have $x^4\neq 0, 1$ by Lemma~\ref{pre-4-power}.
By  Lemma~\ref{pre-4-power} again, we have $(\frac{1-x}{1+x})^4\neq 0, 1$.
Thus $[1, \frac{1-x}{1+x}\alpha, \alpha^2]\notin\mathscr{P}\cup\widetilde{\mathscr{A}}\cup\widetilde{\mathscr{M}}$
by Corollary~\ref{binary-necessary}.
%by Lemma~\ref{binary}.
%%% for not in P: all non 0, and det (1)(\al)^  not equal to middle -squared.
% for not \tilde{A}: use 2, 3 bullet of Cor 0.1 middle sq = -last term
% for not \tilde{M}: last sq  is -1 not 1
This implies that  $\PlCSP^2([1, \frac{1-x}{1+x}\alpha, \alpha^2])$ is \numP-hard
by Theorem~\ref{thm:parII:k-reg_homomorphism}$'$.
Thus $\PlCSP^2(f, g)$ is $\#{\rm P}$-hard.
\end{proof}

\vspace{.1in}

The next lemma is the No-Mixing lemma of $\mathscr{A}$
with the other tractable sets, namely  the statements No-Mixing-$(4,j)$
for $1 \le j \le 3$.
Having already proved  Lemma~\ref{mixing-P}, we can assume that
both $f$ and $g$ are not in $S_5 = \mathscr{P}$.

\begin{lemma}\label{mixing-A}
Let $\{f, g\}\subseteq \left(\bigcup_{k=1}^4 S_k \right)\setminus S_5$ and $\{f, g\}\nsubseteq S_j$ for every $1\leq j\leq 4$.
Assume that $f\in\mathscr{A}$, then $\PlCSP^2(f, g)$ is $\#{\rm P}$-hard.
\end{lemma}
%If $f$, $g$ satisfy one of the following conditions:
%\begin{itemize}
%\item $f\in \mathscr{A}\setminus (\mathscr{P}\cup\mathscr{A}^{\dagger})$,
%$g\in \mathscr{A}^{\dagger}\setminus (\mathscr{P}\cup\mathscr{A})$,
%and $\{f, g\}\nsubseteq\widehat{\mathscr{M}}$, $\{f, g\}\nsubseteq\widehat{\mathscr{M}}^{\dagger}$.
%\item $f\in \mathscr{A}\setminus (\mathscr{P}\cup\widehat{\mathscr{M}})$,
%$g\in \widehat{\mathscr{M}}\setminus (\mathscr{P}\cup\mathscr{A})$,
%and $\{f, g\}\nsubseteq\mathscr{A}^{\dagger}$, $\{f, g\}\nsubseteq\widehat{\mathscr{M}}^{\dagger}$.
%\item $f\in \mathscr{A}\setminus (\mathscr{P}\cup\widehat{\mathscr{M}}^{\dagger})$,
%$g\in \widehat{\mathscr{M}}^{\dagger}\setminus (\mathscr{P}\cup\mathscr{A})$,
%and $\{f, g\}\nsubseteq\mathscr{A}^{\dagger}$, $\{f, g\}\nsubseteq\widehat{\mathscr{M}}$,
%\end{itemize}
%then Pl-$\shar${\rm CSP}$^2(f, g)$ is $\#{\rm P}$-hard.
%\end{lemma}

\begin{proof}
%By $f\in\mathscr{A}\setminus\mathscr{P}$,
By $f\in\mathscr{A}$,
%%% JYC
%%% the above statement does not use not in P.
%we have $g\in(\mathscr{A}^{\dagger}\cup\widetilde{\mathscr{M}})\setminus(\mathscr{P}\cup\mathscr{A})$,
we have $g \notin \mathscr{A}$. Thus,
 %$g\in\mathscr{A}^{\dagger}\setminus(\mathscr{P}\cup\mathscr{A})$,
%or $g\in\widehat{\mathscr{M}}\setminus(\mathscr{P}\cup\mathscr{A})$,
%or $g\in\widehat{\mathscr{M}}^{\dagger}\setminus(\mathscr{P}\cup\mathscr{A})$.
%%% JYC: I write more succinctly.
$g\in (\mathscr{A}^{\dagger} \cup \widehat{\mathscr{M}} \cup 
\widehat{\mathscr{M}}^{\dagger} )
\setminus(\mathscr{P}\cup\mathscr{A})$.
\begin{enumerate}
\item Suppose $g\in\mathscr{A}^{\dagger}\setminus(\mathscr{P}\cup\mathscr{A})$.
%%% JYC I don't need to say this here.
%iff $g\in\mathscr{A}^{\dagger}\setminus\mathscr{P}$ by
%Corollary~\ref{binary-necessary} and Lemma~\ref{second-recurrence-relation}.
%%% JYC probably you will have a dedicated lm after cor 0.1 & lm 0.6
%\textcolor{red}{FIG1}.
%Thus
%for $g\in\mathscr{A}^{\dagger}\setminus\mathscr{P}\cup\mathscr{A}$,
%we can assume that
%$g\in\mathscr{A}^{\dagger}\setminus\mathscr{P}$ equivalently.
Then  a fortiori, $g\in\mathscr{A}^{\dagger}\setminus\mathscr{P}$.
As $f\in\mathscr{A}\setminus\mathscr{P}$,
by Lemma~\ref{mixing-arity-reduction-A}, we have some
$f'\in\mathscr{A}\setminus\mathscr{P}$ and
 $g'\in\mathscr{A}^{\dagger}\setminus\mathscr{P}$,
both of arity 4.
Without loss of generality,
 we will assume the given $f$ and $g$ are of arity 4.
By definition
(see Figure~\ref{fig:venn_diagram:A_Adagger_P}), we can assume that
%%% this is actually easy. for A is clearly the F2F3 form,
%as not F1 which is GEnEQ and so in P. Then for A^+, do [1 0 \\ 0 \alpha].
\[f=[1, \rho]^{\otimes 4}+i^{r}[1, -\rho]^{\otimes 4}~~~~
 \mbox{and}~~~~
 g=[1, \alpha]^{\otimes 4}+i^{s}[1, -\alpha]^{\otimes 4} ~~~~\mbox{where}
~r, s=0, 1, 2, 3.\]
\begin{itemize}
\item If both $r, s \equiv 0 \pmod 2$, then
 $f=[1, \rho]^{\otimes 4}\pm[1, -\rho]^{\otimes 4}$ and
$g=[1, \alpha]^{\otimes 4}\pm [1, -\alpha]^{\otimes 4}$. This is the
case where both $f$ and $g$ satisfy the parity constraint, and it
is proved in Lemma~\ref{mixing-A-2-A-parity}.

\item If  $r \equiv 1 \pmod 2$ then
 $f=[1, \rho]^{\otimes 4}\pm i[1, -\rho]^{\otimes 4}$.
For $\rho^2=1$,
by our calculus we have
\[\partial(f)
= 2 \{ [1, \rho]^{\otimes 2} \pm i[1, -\rho]^{\otimes 2}\}
= 2 (1 \pm i) [1, \mp i \rho, \rho^2]
=2(1\pm i)[1, \rho', -\rho'^2],\]
 where $\rho'=\mp i\rho$, and
 $\rho'^4=1$.
Thus ${\rm Pl-}\#{\rm CSP}^2([1, \rho', -\rho'^2], g)$ is $\#{\rm P}$-hard
 by Lemma~\ref{mixing-A-binary-2-A}.
 So ${\rm Pl-}\#{\rm CSP}^2(f, g)$ is $\#{\rm P}$-hard.

 For $\rho^2=-1$, we cannot use $[1, 0, 1]$ to reduce the arity of $f$,
because $\partial(f) =0$ in this case.
 Instead we construct a suitable binary signature from $g$.
If $s\neq 2$, then we
have $g_0 = 1 + i^s \not =0$ and $g_4 = \alpha^4 + i^s (-\alpha)^4
=-(1 + i^s) = -g_0$,
 and therefore $\partial_{g}(=_6)=g_0[1, 0, -1]$ on the LHS.
 Then we have  $\partial_{[1, 0, -1]}(f)
=2 \{[1, \rho]^{\otimes 2} \pm i[1, -\rho]^{\otimes 2}\}
=2(1\pm i)[1, \mp i\rho, \rho^2]=2(1\pm i)[1, \rho', -\rho'^2]$,
where $\rho'=\mp i\rho$ and $\rho'^4=1$.
 Then we are done by Lemma~\ref{mixing-A-binary-2-A}.
 If $s=2$, then $\partial(g)
=(1+ \alpha^2) \{ [1, \alpha]^{\otimes 2} - [1, -\alpha]^{\otimes 2}\}$,
a nonzero multiple of $[0, 1, 0]$.
%%% JYC commented out the details. since it is easier to see [1 0 1]
%%% and that is all we need. don't need the exact value of alpha.
%=2\alpha(1 \pm i)[0, 1, 0]$, as $\alpha^2 = \pm i$.
%
 Thus we have  $\partial_{[0, 1, 0]}(f)
=2\rho \{[1, \rho]^{\otimes 2} \mp i[1, -\rho]^{\otimes 2}\}
=2\rho (1 \mp i)
[1, \pm i\rho, \rho^2]=2\rho(1\mp i)[1, \rho', -\rho'^2]$,
where $\rho'=\pm i\rho$ and $\rho'^4=1$.
 Then we are done by Lemma~\ref{mixing-A-binary-2-A} again.

\item If  $r \equiv 0 \pmod 2$ and  $s \equiv 1 \pmod 2$,
i.e., $f=[1, \rho]^{\otimes 4}\pm[1, -\rho]^{\otimes 4}$ and
$g=[1, \alpha]^{\otimes 4}\pm i[1, -\alpha]^{\otimes 4}$,
then
we will construct a binary signature
 $h= [1, b, \pm 1]$. Note that $h \in \widetilde{\mathscr{M}}$
by Lemma~\ref{binary}.
 Furthermore, we will ensure that $b^4\neq 0, 1$,
thus $h \notin \mathscr{P}\cup\widetilde{\mathscr{A}}$
by Corollary~\ref{binary-necessary}.
%$h\in \widetilde{\mathscr{M}}\setminus(\mathscr{P}\cup\widetilde{\mathscr{A}})$,
%and a binary $h'\in\mathscr{P}\setminus\widetilde{\mathscr{M}}$, then
%$\PlCSP^2(h, h')$ is $\#{\rm P}$-hard by Lemma~\ref{mixing-P}.
Then we are done by Lemma~\ref{M-even} and Lemma~\ref{2-M-even}.

We have $\partial(g)
=(1 + \alpha^2) \{[1, \alpha]^{\otimes 2} \pm i [1, -\alpha]^{\otimes 2}\}
=(1 + \alpha^2) (1\pm i) [1, \mp i\alpha, \alpha^2]$,
%%% I don't say it is (1\pm i)^2, since that pm in \alpha^2 is independent of \pm
% in the g's expression   JYC
a nonzero multiple  of $[1, \alpha', -\alpha'^2]$,
 where  $\alpha'=\mp i\alpha$ and
 $\alpha'^{4}=-1$.
Moreover, we have $h= \partial_{[1, \alpha', -\alpha'^2]}(f)
= \lambda [1, \rho]^{\otimes 2} \pm \mu [1, -\rho]^{\otimes 2}$,
where $\lambda = 1-\rho^2\alpha'^2+2\rho\alpha'$
and $\mu = 1-\rho^2\alpha'^2-2\rho\alpha'$.
Then $h= \lambda (1 \pm x) [1, a \rho, \rho^2]$,
where $x =  \mu/\lambda$ and $a =
\frac{1\mp x}{1\pm x}$.
Note that $1-\rho^2\alpha'^2 = 1 \pm i$ has norm $\sqrt{2}$
and $|2\rho\alpha'|=2$, thus $\lambda \not =0$
and $(1-\rho^2\alpha'^2)^4 \not = (2\rho\alpha')^4$ by norm,
therefore $x^4 \not = 0,1$ by Lemma~\ref{pre-4-power}.
Then by Lemma~\ref{pre-4-power} again, $a^4\neq 0, 1$,
and so $(a\rho)^4\neq 0, 1$ as well.
As $\lambda \neq 0$, $1 \pm x \neq 0$, $\rho^2 = \pm 1$,
we have a nonzero multiple of $[1, a\rho, \pm 1]$,
our desired binary signature, and
we are done by Lemma~\ref{M-even} and Lemma~\ref{2-M-even}.
%Thus $h\in \widetilde{\mathscr{M}}\setminus(\mathscr{P}\cup\widetilde{\mathscr{A}})$.
%
%Secondly, note that we have $\partial_{[1, \alpha', -\alpha'^2]}(=_4)=[1, 0, -\alpha'^2]\in\mathscr{P}\setminus\widetilde{\mathscr{M}}$
%and $\partial^2_{[1, \alpha', -\alpha'^2]}(=_6)=[1, 0, -1]$ on {\rm LHS}.
%For $\rho^2=1$, we have $\partial(f)=\lambda[1, 0, 1]$ or $\lambda[0, 1, 0]$ with $\lambda\neq 0$.
%For $\rho^2=-1$, we have $\partial_{[1, 0, -1]}(f)=\lambda[1, 0, 1]$ or $\lambda[0, 1, 0]$ with $\lambda\neq 0$.
%So we have $[1, 0, 1]$ or $[0, 1, 0]$ on {\rm RHS}.
%Thus we have $h'=[1, 0, -\alpha'^2]$ or $h'=[\alpha'^2, 0, 1]$ on the right side by
%moving $[1, 0, -\alpha'^2]$ to {\rm RHS}.
%Note that $h'\in\mathscr{P}\setminus\widetilde{\mathscr{M}}$ by Lemma~\ref{binary}.
%Thus we are done.
%
\end{itemize}

In the following we may assume $g \notin \mathscr{A}^{\dagger}$.
\item Suppose
$g\in\widehat{\mathscr{M}}\setminus(\mathscr{P}\cup\mathscr{A})$,
then
$g\in\widehat{\mathscr{M}}\setminus(\mathscr{P}\cup\widetilde{\mathscr{A}})$.
We also have $f \notin \widehat{\mathscr{M}}$, lest $\{f, g\}
\subseteq \widehat{\mathscr{M}}$,
and we are done by Lemma~\ref{M-even}.

\item Suppose $g\in\widehat{\mathscr{M}}^{\dagger}\setminus(\mathscr{P}\cup\mathscr{A})$,
then $g\in\widehat{\mathscr{M}}^{\dagger}\setminus(\mathscr{P}\cup\widetilde{\mathscr{A}})$.
Now
$f\in \mathscr{A}\setminus(\mathscr{P}\cup\widehat{\mathscr{M}}^{\dagger})$.
%and $g\in\widehat{\mathscr{M}}^{\dagger}\setminus(\mathscr{P}\cup\widetilde{\mathscr{A}})$.
%Here $g\notin \mathscr{A}^{\dagger}$ since we have proven the No-Mixing of $\mathscr{A}$ and $\mathscr{A}^{\dagger}$.
Note that $[1, 1]^{\otimes 2n}\pm [1, -1]^{\otimes 2n}
%
%= \left[\begin{smallmatrix} 1 & 1 \\
%1 & -1 \end{smallmatrix}\right]^{\otimes 2n}
%\left\{ \left[\begin{smallmatrix} 1 \\ 0\end{smallmatrix}\right]^{\otimes 2n}
%\pm \left[\begin{smallmatrix} 0 \\ 1 \end{smallmatrix}\right]^{\otimes 2n}
%\right\}
%% [ 1,0]^{\otimes 2n}\pm [0,1]^{\otimes 2n} \right\}
\in\widehat{\mathscr{M}}^{\dagger}$.
This can be verified as follows: 
Let 
$Z= 
%\frac{1}{\sqrt{2}}
\left[\begin{smallmatrix} 1 & 1 \\
i & -i \end{smallmatrix}\right]$,  then $\widehat{\mathscr{M}}^{\dagger}
= Z \mathscr{M}$, and
$Z^{-1}
=
\frac{1}{2}
\left[\begin{smallmatrix} 1 & -i\\
1 & i \end{smallmatrix}\right]$.
We first verify that
$[ 1,0]^{\otimes 2n}\pm [0,1]^{\otimes 2n} \in 
\widehat{\mathscr{M}}^{\dagger}$, by
\begin{align*}
 \trans{1}{-i}{1}{i}^{\otimes 2n}
 \left\{ 
  \left[\begin{matrix} 1 \\ 0 \end{matrix} \right]^{\otimes 2n}
  \pm 
  \left[\begin{matrix} 0 \\ 1 \end{matrix} \right]^{\otimes 2n}
 \right\}
 &=
 \left[\begin{matrix} 1 \\  1 \end{matrix} \right]^{\otimes 2n}
 \pm (-i)^{2n}
 \left[\begin{matrix} 1 \\ -1 \end{matrix} \right]^{\otimes 2n} \\
 &=
 \left[\begin{matrix} 1 \\  1 \end{matrix} \right]^{\otimes 2n}
 \pm (-1)^{n}
 \left[\begin{matrix} 1 \\ -1 \end{matrix} \right]^{\otimes 2n}
 \in \mathscr{M}.
\end{align*}
Then notice that
\[\left[\begin{matrix} 1 \\1 \end{matrix} \right]^{\otimes 2n}
\pm 
\left[\begin{matrix} 1 \\ -1 \end{matrix} \right]^{\otimes 2n}
=
\left[\begin{matrix} 1 & 1 \\ 1 &-1 \end{matrix} \right]^{\otimes 2n}
\left\{
\left[\begin{matrix} 1 \\0 \end{matrix} \right]^{\otimes 2n}
\pm
\left[\begin{matrix} 0 \\ 1 \end{matrix} \right]^{\otimes 2n}
\right\}
\in
\left[\begin{matrix} 1 & 1 \\ 1 &-1 \end{matrix} \right]
\widehat{\mathscr{M}}^{\dagger}.\]
However
\[\left[\begin{matrix} 1 & 1 \\ 1 &-1 \end{matrix} \right]
\left[\begin{matrix} 1 & 1 \\ i &-i \end{matrix} \right]
=
\left[\begin{matrix} 1+i & 1-i \\ 1-i & 1+i \end{matrix} \right]
=
\left[\begin{matrix} 1 & 1 \\ i &-i \end{matrix} \right]
\left[\begin{matrix} 0 & 1-i \\ 1+i & 0 \end{matrix} \right] 
\]
and
$\left[\begin{smallmatrix} 0 & 1-i \\ 1+i & 0 \end{smallmatrix} \right]
\mathscr{M} = \mathscr{M}$, therefore
$\left[\begin{smallmatrix} 1 & 1 \\ 1 & -1 \end{smallmatrix} \right]
\widehat{\mathscr{M}}^{\dagger}
= \left[\begin{smallmatrix} 1 & 1 \\ 1 & -1 \end{smallmatrix} \right]
Z \mathscr{M}
= Z \left[\begin{smallmatrix} 0 & 1-i \\ 1+i & 0 \end{smallmatrix} \right]
\mathscr{M}
= Z \mathscr{M}
= \widehat{\mathscr{M}}^{\dagger}$.
(Also see Figure~\ref{fig:venn_diagram:M_Mdagger_Atilde_P}).

Since $[1, 1]^{\otimes 2n}\pm [1, -1]^{\otimes 2n} \in \widehat{\mathscr{M}}^{\dagger}$ and $f \notin \widehat{\mathscr{M}}^{\dagger}$,
 $f$ cannot take the form
$[1, 1]^{\otimes 2n}\pm [1, -1]^{\otimes 2n}$.
Then by definition (see
 Figure~\ref{fig:venn_diagram:A_Adagger_P}) $f$ takes the form
\[[1, \rho, -\rho^2],~~\mbox{or}~~
[1, 1]^{\otimes 2n}\pm i[1, -1]^{\otimes 2n},~~\mbox{or}~~
[1, i]^{\otimes 2n}+i^r [1, -i]^{\otimes 2n},~~\mbox{where}~~2n\geq 4.\]
The following three cases are immediately done by Lemma~\ref{2-M-even}:
\begin{itemize}
\item $f=[1, \rho, -\rho^2]$.
\item $f=[1, 1]^{\otimes 2n}\pm i[1, -1]^{\otimes 2n}$ with $2n\geq 4$, then we have
$\partial^{n-1}(f)=2^{n-1}[1\pm i, 1\mp i, 1\pm i]$ which is not $\lambda[1, 0, 1]$.
\item If $f=[1, i]^{\otimes 2n}+i^r[1, -i]^{\otimes 2n}$ with $2n\equiv 2\pmod 4$,
then we have $\partial^{\frac{n-1}{2}}_{=_4}(f)=2^{\frac{n-1}{2}}[1+i^r, (1-i^r)i, -(1+i^r)]$
which is not $\lambda[1, 0, 1]$, no matter what value $r$ takes.
\end{itemize}

The remaining case is that $f=[1, i]^{\otimes 2n}+i^r[1, -i]^{\otimes 2n}$ with $2n\equiv 0\pmod 4$.
In this case, we have
\[\partial^{\frac{n-2}{2}}_{=_4}(f)=2^{\frac{n-2}{2}}\{[1, i]^{\otimes 4}+i^r[1, -i]^{\otimes 4}\}.\]
We will denote by $f' = [1, i]^{\otimes 4}+i^r[1, -i]^{\otimes 4}$.
%By Lemma~\ref{mixing-arity-reduction-A}, %and Lemma~\ref{2-M-arity-reduction},
%we can assume that
%and $f=[1, \rho]^{\otimes 4}+i^{r}[1, -\rho]^{\otimes 4}$,
 %$g=[1, \alpha]^{\otimes 4}+i^{r'}[1, -\alpha]^{\otimes 4}$.
%
%For $f\in\mathscr{A}\setminus(\mathscr{P}\cup\widehat{\mathscr{M}})$,
%$g\in\widehat{\mathscr{M}}\setminus(\mathscr{P}\cup\mathscr{A})$,
   %$f=[1, 1]^{\otimes 4}\pm i[1, -1]^{\otimes 4}$, or $f=[1, i]^{\otimes 4}+i^r[1, -i]^{\otimes 4}$
   %and $g\in\widehat{\mathscr{M}}^{\dagger}\setminus(\mathscr{P}\cup\widetilde{\mathscr{A}})$.
   If $g$ has arity 2, then up to a nonzero scalar, $g=[1, b, -1]$ 
with $b^4\neq 0, 1$ by Lemma~\ref{M-2-M-NOT-IN-A-AND-P},
   and we are done by Lemma~\ref{2-M-even}.
   In the following, assume that $g$ has arity $2m\geq 4$.
   By Lemma~\ref{M-2-M-NOT-IN-A-AND-P},
either $g=[s, ti]^{\otimes 2m}\pm[t, si]^{\otimes 2m}$ with $s^4\neq t^4$ and $st\neq 0$,
   or $g$ has arity $2m$ and $g_k=(\pm i)^k(2m-2k)$.
%and
%$g=[s, ti]^{\otimes 4}+[t, si]^{\otimes 4}$, $st\neq 0$, $s^4\neq t^4$ or
%$g$ has arity 6 and $g_k=(\pm i)^k(6-2k)$.
%
%For $f\in\mathscr{A}\setminus(\mathscr{P}\cup\widehat{\mathscr{M}}^{\dagger})$,
%$g\in\widehat{\mathscr{M}}^{\dagger}\setminus(\mathscr{P}\cup\mathscr{A})$, we have
%$f=(1, i)^{\otimes 4}+i^r(1, -i)^{\otimes 4}$, or $(1, 1)^{\otimes 4}\pm i(1, -1)^{\otimes 4}$.
\begin{itemize}
%\item For $f=[1, 1]^{\otimes 4}\pm i[1, -1]^{\otimes 4}$,
% we have $\partial(f)=4[1, \mp i, 1]\neq \lambda[1, 0, 1]$.
% By connecting two copies of $[1, \mp i, 1]$, we have $[0, 1, 0]$.
% Then we are done by Lemma~\ref{2-M-even}.
%
 \item
 If $g$ has arity $2m\geq 4$ and $g_k=(\pm i)^k(2m-2k)$ up to a nonzero scalar,
then let $\hat{g}= (Z^{-1})^{\otimes 2m} g$,
where $Z=
\left[\begin{smallmatrix} 1 & 1 \\
i & -i \end{smallmatrix}\right]$. Then
$\hat{g}=[0, 1, 0, \ldots, 0]$ or $\hat{g}=[0, \ldots, 0, 1, 0]$
of arity $2m$.
 By Corollary~\ref{2-M-double-roots}, we have
 \[\PlCSP^2(\hat{g})\leq\PlCSP^2(f', g).\]
 Let $\hat{g}'=\partial^{m-2}(\hat{g})=[0, 1, 0, 0, 0]$ or $[0, 0, 0, 1, 0]$.
%We verify that $\hat{g}' \notin 
%\mathscr{P}\cup\widetilde{\mathscr{A}}\cup\widetilde{\mathscr{M}}$.
Clearly $\hat{g}'$ is non-degenerate.
It also has a second order recurrence of type $\langle 0, 0, 1\rangle$
or $\langle 1, 0, 0 \rangle$. By Lemma~\ref{second-recurrence-relation},
$\hat{g}' \notin \mathscr{P}\cup\widetilde{\mathscr{A}}\cup\widetilde{\mathscr{M}}$.
Then $\operatorname{Pl-}\#{\rm CSP}^2(\hat{g}')$
is $\#$P-hard by Theorem~\ref{arity-4-dichotomy}
 and we are done.

\item If $g=[s, ti]^{\otimes 2m}\pm [t, si]^{\otimes 2m}$,
we have
\[
 g'=\partial^{m-2}(g)=(s^2-t^2)^{m-2} \left\{[s, ti]^{\otimes 4}\pm (-1)^{m-2} [t, si]^{\otimes 4}\right\}
\]
and from $f'$ we get
 $[1, 0, -1]^{\otimes 2}$ on LHS by Lemma~\ref{construct-[1,0,1]-by-(1,i)-(1,-i)}, thus
\[
  \PlHolant([1, 0, -1]^{\otimes 2}\cup{\cal EQ}_2 \mid f', g')\leq\PlCSP^2(f, g).
\]
After a holographic transformation using
${T}=\left[\begin{smallmatrix} 1 & 0 \\
0 & -i \end{smallmatrix}\right]$,
%$\begin{bmatrix}1&0\\0&-i\end{bmatrix}$,
we have
\begin{multline*}
 \PlHolant([1, 0, 1]^{\otimes 2}, [1, 0, -1], [1, 0, 0, 0, 1], \ldots \mid \hat{f'}, \hat{g'}) \\
 \equiv
 \PlHolant([1, 0, -1]^{\otimes 2}\cup{\cal EQ}_2 \mid f', g'),
\end{multline*}
where $\hat{f'}=({T}^{-1})^{\otimes 4}f'=[1, 1]^{\otimes 4}+i^r[1, -1]^{\otimes 4}$ and $\hat{g'}=({T}^{-1})^{\otimes 4}g'$.
Note that $\hat{f'}$ satisfies a second order
recurrence of type $\langle -1, 0, 1\rangle$.
Thus $\hat{f'} \notin \widehat{\mathscr{M}}$ by 
Lemma~\ref{second-recurrence-relation}.
%since $f' \notin \widehat{\mathscr{M}}^{\dagger}$,
%%% JYC by recurrnce, more directly...
Also note that  $\mathscr{P}$ and $\widetilde{\mathscr{A}}$ are
invariant under $T$, and since $g' \in \widehat{\mathscr{M}}^{\dagger}
\setminus(\mathscr{P}\cup \widetilde{\mathscr{A}})$,
we have $\hat{g'} \in \widehat{\mathscr{M}}
\setminus(\mathscr{P}\cup\widetilde{\mathscr{A}})$.
In the following, we will construct $[1, 0, 1]^{\otimes 2}$ on {\rm RHS}
for
\[\PlHolant([1, 0, 1]^{\otimes 2}, [1, 0, -1], [1, 0, 0, 0, 1], \ldots \mid \hat{f'}, \hat{g'}).\]
Since we have $[1, 0, 1]^{\otimes 2}$ on LHS, we can get
$[\partial(\hat{f'})]^{\otimes 2}=4[1+i^r, 1-i^r, 1+i^r]^{\otimes 2}$
on {\rm RHS}.
\begin{itemize}
\item If $r=0$, then we directly have $[1, 0, 1]^{\otimes 2}$ on {\rm RHS}.
\item If $r=2$, then we have $[0, 1, 0]^{\otimes 2}$ on {\rm RHS}.
Thus we can move $[1, 0, 1]^{\otimes 2}$ on {\rm LHS} to {\rm RHS}.
\item If $r=1$ or $3$, then we have $[1, \pm i, 1]^{\otimes 2}$ on {\rm RHS}.
By connecting two copies of $[1, \pm i, 1]^{\otimes 2}$ by 
$[1, 0, 1]^{\otimes 2}$ of LHS,
we have a nonzero multiple of $[0, 1, 0]^{\otimes 2}$ on {\rm RHS}.
Then we can move $[1, 0, 1]^{\otimes 2}$ on {\rm LHS} to {\rm RHS}.
\end{itemize}
From the above, we have
\begin{align*}
        &\PlHolant([1, 0, 1]^{\otimes 2}, [1, 0, -1], [1, 0, 0, 0, 1], \ldots \mid \hat{f'}, \hat{g'}, [1, 0, 1]^{\otimes 2})\\
 \leq {}&\PlHolant([1, 0, 1]^{\otimes 2}, [1, 0, -1], [1, 0, 0, 0, 1], \ldots \mid \hat{f'}, \hat{g'}).
\end{align*}
Note that we have all of $=_{4k}$ on the {\rm LHS}.
Thus by Lemma~\ref{mixing-P-global-binary},
\[
  \PlCSP^2(\hat{f'}, \hat{g'})\leq 
  \PlHolant([1, 0, 1]^{\otimes 2}, [1, 0, -1], [1, 0, 0, 0, 1], \ldots \mid \hat{f'}, \hat{g'}, [1, 0, 1]^{\otimes 2}).
\]
%by Lemma~\ref{mixing-P-global-binary}.
Recall that
 %$\hat{f}\in\mathscr{A}\setminus(\mathscr{P}\cup\widehat{\mathscr{M}})$,
$\hat{g'}\in\widehat{\mathscr{M}}\setminus(\mathscr{P}\cup\widetilde{\mathscr{A}})$
and $\hat{f'} \notin \widehat{\mathscr{M}}$.
 Thus we are done by Lemma~\ref{M-even}.
 \qedhere
%Thus ${\rm Pl-}\#{\rm CSP}^2(f', g')$ is $\#{\rm P}$-hard by the proof the second case.
%So ${\rm Pl-}\#{\rm CSP}^2(f, g)$ is $\#{\rm P}$-hard.
% \[
%{\rm Pl-Hol}([1, 0, -1], [1, 0, 0, 0, 1], \ldots \mid f', g')\leq
%{\rm Pl-}\#{\rm CSP}^2(f', g', [1, 0, 1]^{\otimes 2}).
%\]
%then we are done
 %For $g=[s, ti]^{\otimes 4}+[t, si]^{\otimes 4}$,
 %we have $\partial_{[0, 1, 0]}(g)=(s^2-t^2)^2[1, \frac{2sti}{s^2-t^2}, -1]$
 %and we have $\partial_{[1, \frac{2sti}{s^2-t^2}, -1]}(=_{2k+2})=[1, 0, \ldots, 0, -1]$ that has arity $2k$ for $k\geq 1$.
%
% Similarly,
% for $g$ has arity 6 and $g_k=(\pm i)^k(6-2k)$.
% we have $\partial_{[0, 1, 0]}(g)=[4i, -2, 0, 2, -4i]$
% and we have $\partial_{[4i, -2, 0, 2, -4i]}(=_{2k+4})=4i[1, 0, \ldots, 0, -1]$ that has arity $2k$ for $k\geq 1$.
%
 %Denote $\{[1, 0, -1], [1, 0, 0, 0, -1], \ldots\}$ by {\cal EQ'}$_2$.
 %Then we have
 %\[
%{%\rm Pl-Hol}({\cal EQ}_2\cup{\cal EQ'}_2 \mid f, g)\leq{\rm Pl-}\#{\rm CSP}^2(f, g).
%\%]
%After the holographic transformation using $A=\begin{bmatrix}1&0\\0&-i\end{bmatrix}$,
%we have
%\[
%{\rm Pl-Hol}({\cal EQ}\cup{\cal EQ'} \mid f', g')\equiv{\rm Pl-Hol}({\cal EQ}\cup{\cal EQ'} \mid f, g),
%\]
%where $f'=A^{\otimes 4}f=[1, 1]^{\otimes 4}+i^r[1, -1]^{\otimes 4}$ and $g'=A^{\otimes 4}g$.
%Thus we have
%\%[
%{\rm Pl-}\#{\rm CSP}^2(f', g')\leq{\rm Pl-}\#{\rm CSP}^2(f, g).
%\]
%Note that  %$f'\in\mathscr{A}\setminus(\mathscr{P}\cup\widehat{\mathscr{M}})$,
%$g\in\widehat{\mathscr{M}}\setminus(\mathscr{P}\cup\mathscr{A})$. Thus we are done by Lemma~\ref{M-even}.
\end{itemize}

%For $f=(1, i)^{\otimes 4}+i^r(1, -i)^{\otimes 4}$,
%\begin{itemize}
%\item If $g=[u, v, 0, v, -u]$,
%where $f'=(1, 1)^{\otimes 4}+i^r(1, -1)^{\otimes 4}$, $g'=A^{\otimes 4}g\in\widehat{\mathscr{M}}=[u', v', 0, -v', -u']$, $u'v'\neq 0$.
%Then we have $\partial_{g'}([1, 0, 0, 0, 0, 0, -1])=u'[1, 0, 1]$ on the left side and
%$\partial(g')=u'[1, 0, -1]$ on the right side.
%Since we have $[1, 0, 1]$ on the left side, we can get $[1, 0, 1]$ on the right side by connecting two copies of $[1, 0, -1]$.
%Then we can get all of $=_{2k}$ for $k\geq 1$ since we have $=_2, =_4$ on the left side and $=_2$ on the right side.
%This implies that
%\[
%{\rm Pl-}\#{\rm CSP}^2(f', g')\leq {\rm Pl-Hol}([1, 0, -1], [1, 0, 0, 0, 1], \ldots \mid f', g')\equiv{\rm Pl-}\#{\rm CSP}^2(f, g).
%\]
%Note that $f'\in\mathscr{A}\setminus(\mathscr{P}\cup\widehat{\mathscr{M}}^{\dagger})$,
%$g'\in\widehat{\mathscr{M}}^{\dagger}\setminus(\mathscr{P}\cup\mathscr{A})$.
%Thus ${\rm Pl-}\#{\rm CSP}^2(f', g')$ is hard and we are done.

%\item For $g=[u, v, w, -v, u]$, there exist $s, t$ such that
%$g=(s, t)^{\otimes 4}+(t, si)^{\otimes 4}$, $st\neq 0$, $s^4\neq t^4$.
%after the holographic transformation using $Z=\begin{bmatrix}1&1\\i&-i\end{bmatrix}$,
%we have
%\[
%{\rm Pl-}\#{\rm CSP}^2(f, g)\equiv {\rm Pl-Hol}([0, 1, 0], [1, 0, 1, 0, 1], \ldots \mid f', g'),
%\]
%where $f'=[1, 0, 0, 0, i^r]$, $g'=(s+t)^4[1, 0, r^2, 0, r^4]$, where $r=\frac{s-t}{s+t}$.
%Note that $r^4\neq 0, 1$.
%By the following gadget,
%\begin{center}
%\textcolor{red}{ADDING A FIGURE}
%\end{center}
%we have the signature%
%
%\end{itemize}
\end{enumerate}
\end{proof}

\subsection{Mixing with \texorpdfstring{$\mathscr{A}^\dagger$}{Adagger}}

In this subsection,
we prove the No-Mixing lemma for $\mathscr{A}^\dagger$ with other tractable sets,
namely the statements No-Mixing-$(3,j)$, for $1 \le j \le 2$.
Because we have already proved Lemma~\ref{mixing-P} and Lemma~\ref{mixing-A},
the No-Mixing lemmas for $S_5 = \mathscr{P}$ and $S_4 = \mathscr{A}$ respectively,
we only need to consider the mixing of $S_3 = \mathscr{A}^\dagger$ with $S_j$ for $1 \le j \le 2$.
Thus we may assume $f \in \mathscr{A}^\dagger$ and $g \in \widetilde{\mathscr{M}} \setminus \mathscr{A}^\dagger$.
Moreover,
we can assume that
$f,g \notin \mathscr{P} \cup \mathscr{A}$.
%Moreover, by Lemma \ref{mixing-arity-reduction-A} and Lemma \ref{mixing-arity-reduction-M},
%we can assume that $f$ has arity 4 and $g$ has arity 4 or
%$g$ has arity 6 and $g_k=(\pm i)^{k}(6-2k)$.

\begin{lemma} \label{mixing-2-A}
Let $\{f, g\}\subseteq \left( \bigcup_{k=1}^3 S_k \right)
\setminus (S_4\cup S_5)$ and $\{f, g\}\nsubseteq S_j$ for $1\leq j\leq 3$.
Assume that $f\in\mathscr{A}^\dagger$, then $\PlCSP^2(f, g)$ is $\#{\rm P}$-hard.
%
%
%If $f$, $g$ satisfies one the following conditions:
%\begin{itemize}
%\item $f\in \mathscr{A}^\dagger\setminus (\mathscr{P}\cup\mathscr{A}\cup\widehat{\mathscr{M}})$,
%$g\in \widehat{\mathscr{M}}\setminus (\mathscr{P}\cup\mathscr{A}\cup\mathscr{A}^\dagger)$
%and $\{f, g\}\nsubseteq\widehat{\mathscr{M}}^\dagger$.
%\item $f\in \mathscr{A}^\dagger\setminus (\mathscr{P}\cup\mathscr{A}\cup\widehat{\mathscr{M}}^\dagger)$,
%$g\in \widehat{\mathscr{M}}^\dagger\setminus (\mathscr{P}\cup\mathscr{A}\cup\mathscr{A}^\dagger)$
%and $\{f, g\}\nsubseteq\widehat{\mathscr{M}}$.
%\end{itemize}
% Then Pl-$\sharp${\rm CSP}$^2(f, g)$ is $\#{\rm P}$-hard.
\end{lemma}

\begin{proof}
Firstly, we have $f\in\mathscr{A}^\dagger\setminus\mathscr{P}$,
thus $f\in\{[1, \alpha, -\alpha^2], [1, \alpha]^{\otimes 2n}+i^r[1, -\alpha]^{\otimes 2n} \mid 2n\geq 4\}$
(see Figure~\ref{fig:venn_diagram:A_Adagger_P}).
Clearly $[1, \alpha, -\alpha^2]$ is not $\lambda[1, 0, 1]$.
If $f=[1, \alpha]^{\otimes 2n}+i^r[1, -\alpha]^{\otimes 2n}$, then we have
$\partial^{n-1}(f)=(1+\alpha^2)^{n-1}\{[1, \alpha]^{\otimes 2}+i^r[1, -\alpha]^{\otimes 2}\}=(1+\alpha^2)^{n-1}[1+i^r, (1-i^r)\alpha, (1+i^r)\alpha^2]$
which is not $\lambda[1, 0, 1]$.
Hence we can always
obtain a nonzero binary signature that is not $\lambda[1, 0, 1]$ from $f$.

Note that $g\in\widetilde{\mathscr{M}}\setminus(\mathscr{P}\cup\widetilde{\mathscr{A}})$.
If $g\in\widehat{\mathscr{M}}\setminus(\mathscr{P}\cup\widetilde{\mathscr{A}})$, we are done by Lemma \ref{M-even}.
For $g\in\widehat{\mathscr{M}}^\dagger\setminus(\mathscr{P}\cup\widetilde{\mathscr{A}})$,
since we  have a nonzero binary signature that is not $\lambda[1, 0, 1]$, we are done by
Lemma \ref{2-M-even}.
\end{proof}

\subsection{Mixing with \texorpdfstring{$\widehat{\mathscr{M}}$}{Mhat}}

In this subsection,
we prove the No-Mixing lemma for $\widehat{\mathscr{M}}$ with other tractable sets.
Because we have already proved Lemma~\ref{mixing-P}, Lemma~\ref{mixing-A},  and Lemma~\ref{mixing-2-A},
the No-Mixing lemmas for $S_5 = \mathscr{P}$, $S_4 = \mathscr{A}$, and $S_3 = \mathscr{A}^\dagger$ respectively,
we only need to consider the No-Mixing of $S_2 = \widehat{\mathscr{M}}$ with $S_1 = \widehat{\mathscr{M}}^\dagger$.
%Thus we may assume $f \in \widehat{\mathscr{M}}$ and $g \in \widehat{\mathscr{M}}^\dagger \setminus \widehat{\mathscr{M}}$.

\begin{lemma}\label{mixing-M-hat}
Let $\{f, g\}\subseteq \left(\bigcup_{k=1}^2 S_k \right)
\setminus (S_3\cup S_4\cup S_5)$ and $\{f, g\}\nsubseteq S_j$ for $1\leq j\leq 2$.
%Assume that $f\in\widehat{\mathscr{M}}$, then 
Then $\PlCSP^2(f, g)$ is $\#{\rm P}$-hard.
%
%If $f\in\widehat{\mathscr{M}}\setminus(\mathscr{P}\cup\widetilde{\mathscr{A}}\cup\widehat{\mathscr{M}}^\dagger)$,
%$g\in\widehat{\mathscr{M}}^\dagger\setminus(\mathscr{P}\cup\widetilde{\mathscr{A}}\cup\widehat{\mathscr{M}})$,
% then Pl-$\sharp${\rm CSP}$^2(f, g)$ is $\#{\rm P}$-hard.
\end{lemma}

\begin{proof}
Either $f$ or $g \in \widehat{\mathscr{M}}$, otherwise 
$\{f, g\}\subseteq \widehat{\mathscr{M}}^\dagger$.
As they do not belong to $S_3\cup S_4\cup S_5 = 
\mathscr{P}\cup\widetilde{\mathscr{A}}$,
we have a signature in 
$\widehat{\mathscr{M}}\setminus(\mathscr{P}\cup\widetilde{\mathscr{A}})$.
Thus we are done by Lemma \ref{M-even}.
\end{proof}

By Lemma \ref{mixing-P}, Lemma \ref{mixing-A}, Lemma \ref{mixing-2-A}
and Lemma \ref{mixing-M-hat},
we have the following No-Mixing theorem for  two signatures with even arities.
\begin{theorem}\label{mixing-theorem-2}
Let $f$ and $g$ be two symmetric signatures of even arity.
%%If $f\in S_i$, $g\in S_j$, where $1 \le i\neq j \le 5$,
If $\{f,g\} \subseteq \bigcup_{k=1}^5 S_k$
and $\{f, g\}\nsubseteq S_j$ for $1\leq j\leq 5$,
then $\PlCSP^2(f, g)$ is $\#{\rm P}$-hard.
\end{theorem}

%The next corollary shows that the No-Mixing of the signature in
%$\hat{{\mathscr{M}}}^\dagger\setminus(\mathscr{P}\cup\mathscr{A})$
%and the binary signatures that are not $[1, 0, 1]$.
%This corollary will be used in the proof of the dichotomy of the signatures of general even arities.
%\begin{coro}\label{2-M-even-binary}
%Let $\mathscr{F}$ be a signature set that all of the signatures in $\mathscr{F}$
%has even arities and
%and $f\in\mathscr{F}$, $h\in\mathscr{F}$,
%where $h$ is a non-zero binary and $h\neq \lambda[1, 0, 1]$ and
%$f\in\hat{{\mathscr{M}}}^\dagger\setminus(\mathscr{P}\cup\mathscr{A})$.
%then $\mathscr{F}\subseteq\widehat{\mathscr{M}}^\dagger$ or Pl-$\#${\rm CSP}$^2(\mathscr{F})$ is $\#${\rm P-}hard.
%\end{coro}
%\begin{proof}
%Note that if $f\in\hat{{\mathscr{M}}}^\dagger\setminus(\mathscr{P}\cup\mathscr{A})$,
%then $f\notin\mathscr{A}^\dagger\cup\widehat{\mathscr{M}}$.%
%
%If $h=[1, b, -1]$ or $[0, 1, 0]$, we are done by Lemma \ref{2-M-even}.
%Otherwise, $h\notin\widehat{\mathscr{M}}^\dagger$, then we are done by  Theorem \ref{mixing-theorem}.
%\end{proof}%

\section{No-Mixing of Even Arity Signature Set} \label{PartII.secG.No-Mixing-of-sets}

In this section,
we extend Theorem~\ref{mixing-theorem-2},
the No-Mixing theorem for a pair of two sigatures of even arity, 
to Theorem~\ref{mixing-theorem},
the No-Mixing theorem for a set of signatures of even arity.
For convenience,
we explicitly list some signature sets that are be used in the proof of Theorem~\ref{mixing-theorem}.

\begin{lemma} \label{mixing-clarifying-signature-set}
For nonzero even arity signatures,
ignoring a nonzero factor, we have
\begin{enumerate}
%%% JYC i don't feel right to write "=".
\item $\mathscr{A}^\dagger\cap(\mathscr{P}\cup\mathscr{A})$ is
the set 
\[\{ [1, \alpha]^{\otimes 2n},
[1, 0]^{\otimes 2n}, [0, 1]^{\otimes 2n}, [0, 1, 0],
 [1, 0, \ldots, 0, i^r] ~\mid~ n\geq 1,~ 0 \le r \le 3 \}.\]
\item $\widehat{\mathscr{M}}\cap(\mathscr{P}\cup\widetilde{\mathscr{A}})$
is the set 
\[\{[1, \pm 1]^{\otimes 2m}, [0, 1, 0],
[1, \pm i, 1], [1, 0, \ldots, 0, \pm 1], [1, i]^{\otimes 2n}\pm[1, -i]^{\otimes 2n} ~\mid~ m\geq 1,~ n\geq 2\}.\]
\item $\widehat{\mathscr{M}}^\dagger\cap(\mathscr{P}\cup\widetilde{\mathscr{A}})$ is the set 
\[
\{
[1, \pm i]^{\otimes 2m}, [0, 1, 0],
[1, \pm 1, -1],  [1, 0, \ldots, 0, \pm 1], [1, 1]^{\otimes 2n}\pm[1, -1]^{\otimes 2n} ~\mid~  m\geq 1,~ n\geq 2\}.\]
\item ${\displaystyle \bigcap_{3 \le k \le 5}} S_k$ is
the set
\(\{[1, 0]^{\otimes 2n}, [0, 1]^{\otimes 2n}, [0, 1, 0], [1, 0, \ldots, 0, i^r] ~\mid~  n\geq 1, ~ 0 \le r \le 3 \}\).
%where $m, n\geq 1$.
\item ${\displaystyle \bigcap_{1 \le k \le 5}} S_k
= {\displaystyle \bigcap_{2 \le k \le 5}} S_k$ is the set
\(\{[0, 1, 0], [1, 0, \ldots, 0, \pm 1]\}\).
\end{enumerate}
\end{lemma}

\begin{proof}
For all five cases, it is easy to show that the listed signatures
 in the displayed set are indeed members of the respective
stated intersection,
bear in mind that the signatures all have even arity.
E.g., the signature $f = [1, 0, \ldots, 0, i^r]$ is clearly in $\mathscr{P}$
(as well as $\mathscr{A}$), and it has  even arity $2n$, and thus
under the transformation
$T= \left[
\begin{smallmatrix}
1 & 0\\
0 & \alpha
\end{smallmatrix}
\right]$,
$(T^{-1})^{2n} f = [1, 0, \ldots, 0, i^s]
\in \mathscr{A}$, for some $0 \le s \le 3$.
 Thus $f \in \mathscr{A}^\dagger$.

%%% clear? for 1st item, [1, 0,..., 0, i^r] is in A^+ is by even arity.
In the following, we prove that if $f$ has even arity and is
in the stated intersection then it is among the listed types.
\begin{enumerate}
\item
\subitem{(a.)}
Firstly, suppose that $f\in\mathscr{A}^\dagger\cap(\mathscr{P}\cup\mathscr{A})$ is degenerate, i.e., $f = [a,b]^{2n}$.
If $f = [1,0]^{2n}$ or $[0,1]^{2n}$ up to a nonzero  scalar,
 then $f$ is among the listed.
Suppose $ab \not =0$. Then up to a nonzero scalar, $f = [1, \omega]^{2n}$,
for some $\omega \neq 0$.
By $f\in\mathscr{A}^\dagger$, we have
$\left[
\begin{smallmatrix}
1 & 0\\
0 & \alpha
\end{smallmatrix}
\right]^{2n} f = [1, \alpha \omega]^{2n} \in \mathscr{A}$.
Thus
 $(\alpha \omega)^4 =1$, i.e., $\omega^4 = -1$.
 So $f$ is among the listed types.

\subitem{(b.)}
If $f\in\mathscr{A}^\dagger\cap(\mathscr{P}\cup\mathscr{A})$ is a non-degenerate binary signature,
 by $f\in\mathscr{A}^{\dagger}$ and Lemma~\ref{binary}, we have
$f=[1, \alpha, -\alpha^2]$, or
 $[0, 1, 0]$, or $[1, 0, \rho]$ up to a scalar, where
$\alpha^4 =-1, \rho^4 =1$.
 Note that $[1, \alpha, -\alpha^2]\notin\mathscr{P}\cup\mathscr{A}$
by Corollary~\ref{binary-necessary}.
 Thus $f=[0, 1, 0]$ or $[1, 0, \rho]$; these are among the listed types.

\subitem{(c.)}
 If $f\in\mathscr{A}^\dagger\cap(\mathscr{P}\cup\mathscr{A})$
  is non-degenerate and and has arity $2n \ge 4$, by
$f\in\mathscr{A}^{\dagger}$ and  Lemma~\ref{second-recurrence-relation},
 $f$ has type $\langle 0, 1, 0\rangle$ or $\langle 1, 0, \pm i\rangle$
 and the second order recurrence relation is unique up to a scalar.
 If $f$ has type $\langle 1, 0, \pm i\rangle$, then $f\notin\mathscr{P}\cup\mathscr{A}$ by Lemma~\ref{second-recurrence-relation}.
 This contradicts that $f\in\mathscr{A}^\dagger\cap(\mathscr{P}\cup\mathscr{A})$.
 If $f$ has type $\langle 0, 1, 0\rangle$, then $f=[1, 0, \ldots, 0, x]$
 with $x\neq 0$ up to a nonzero scalar, because $f$ is non-degenerate.
 Moreover, if $x^4\neq 1$, 
bear in mind that $f$ has even arity,
then $f\notin\mathscr{A}^{\dagger}$ and
this 
 contradicts that $f\in\mathscr{A}^\dagger\cap(\mathscr{P}\cup\mathscr{A})$.
 Hence $x^4=1$ and $f=[1, 0, \ldots, 0, i^r]$, for some $0 \le r \le 3$;
this is among the listed types.

 Summarizing, we proved that if $f\in\mathscr{A}^\dagger\cap(\mathscr{P}\cup\mathscr{A})$  then $f$
 is among the listed types.

\item
\subitem{(a.)}
Suppose $f\in\widehat{\mathscr{M}}\cap(\mathscr{P}\cup\widetilde{\mathscr{A}})$ is a nonzero degenerate signature, i.e.,
$f=[a, b]^{\otimes 2n}$.
By $f \in \widehat{\mathscr{M}}$ we have
$\left[\begin{smallmatrix} 1 & 1 \\
1 & -1 \end{smallmatrix}\right]^{\otimes 2n} f
= [a+b, a-b]^{\otimes 2n} \in \mathscr{M}$,
which must satisty the parity constraints. Thus $a = \pm b$
and $f= [1, \pm 1]^{\otimes 2n}$
up to a nonzero scalar. 

\subitem{(b.)}
If $f\in\widehat{\mathscr{M}}\cap(\mathscr{P}\cup\widetilde{\mathscr{A}})$ is a non-degenerate binary signature, by
 $f\in\widehat{\mathscr{M}}$ and  Lemma~\ref{binary}, we have
$f=[0, 1, 0]$, or $[1, b, 1]$,  or $[1, 0, -1]$ up to a
nonzero scalar.
If  $f=[1, b, 1]$ and $b^4\neq 0, 1$, 
then $f\notin\mathscr{P}\cup\widetilde{\mathscr{A}}$,
by Corollary~\ref{binary-necessary}.
This contradicts that $f\in\widehat{\mathscr{M}}\cap(\mathscr{P}\cup\widetilde{\mathscr{A}})$.
 Thus,  $f=[0, 1, 0]$, $[1,0,1]$,  $[1, i^r, 1]$ or $[1, 0, -1]$,
where $0 \le r \le 3$.
Note that  if $r = 0$ or $2$, then $[1, i^r, 1]=[1, \pm 1, 1]
=[1, \pm 1]^{\otimes 2}$.
Thus all these binary signatures are in the listed types. 

\subitem{(c.)}
 If $f\in\widehat{\mathscr{M}}\cap(\mathscr{P}\cup\widetilde{\mathscr{A}})$
  is non-degenerate and has arity 
$2n \ge 4$, by $f\in\widehat{\mathscr{M}}$ and 
Lemma~\ref{second-recurrence-relation},
 $f$ has type $\langle 0, 1, 0\rangle$ or $\langle 1, c, 1\rangle$,
 and the second order recurrence relation is unique up to a scalar.
 If $f$ has type $\langle 1, c, 1\rangle$ with $c\neq 0$,
 then $f\notin\mathscr{P}\cup\widetilde{\mathscr{A}}$ by Lemma~\ref{second-recurrence-relation}
 and this
contradicts that $f\in\widehat{\mathscr{M}}\cap(\mathscr{P}\cup\widetilde{\mathscr{A}})$.
 If $f$ has type $\langle 1, 0, 1\rangle$, then there exist constants $x$
 and $y$ such that
 $f=x[1, i]^{\otimes 2n}+y[1, -i]^{\otimes 2n}$.
By non-degeneracy, we get $xy\neq 0$,
and by its type $\langle 1, 0, 1\rangle$, $f \not \in \mathscr{P}$
by Lemma~\ref{second-recurrence-relation}.
Thus $f \in \widetilde{\mathscr{A}}$.
% Then $x^4 = y^4$, bear in mind that $f$ has even arity.
%%% I don't know why i said foir even arity
In fact by Lemma~\ref{second-recurrence-relation} and its type
$\langle 1, 0, 1\rangle$, $f \not \in 
\mathscr{A}^\dagger\setminus\mathscr{P}$, thus it follows that
$f \in \mathscr{A}\setminus\mathscr{P}$.
Then there are two possiblities:
Either $f = 
\left[\begin{smallmatrix} 1 & 1 \\
1 & -1 \end{smallmatrix}\right]^{\otimes 2n}
\left\{[1, 0]^{\otimes 2n}+i^r [0, 1]^{\otimes 2n} \right\}$,
or
$f = \left[\begin{smallmatrix} 1 & 1 \\
i & -i \end{smallmatrix}\right]^{\otimes 2n}
\left\{[1, 0]^{\otimes 2n}+i^r [0, 1]^{\otimes 2n} \right\}$,
up to a nonzero scalar,
where $0 \le r \le 3$.
By $\left[\begin{smallmatrix} 1 & 1 \\
1 & -1 \end{smallmatrix}\right]^{-1} Z =  \frac{1}{2}Z 
\left[\begin{smallmatrix} 0 & {1-i} \\
{1+i} & 0 \end{smallmatrix}\right]$
the first possiblity quickly reaches a contradiction.
%%% contradiction is some 0 entry on RHS is nonzero on LHS
Thus $f=[1, i]^{\otimes 2n} + i^r [1, -i]^{\otimes 2n}$
up to a nonzero scalar, for some $0 \le r \le 3$.
 If $f=[1, i]^{\otimes 2n}\pm i[1, -i]^{\otimes 2n}$, 
then $\left[\begin{smallmatrix} 1 & 1 \\
1 & -1 \end{smallmatrix}\right]^{\otimes 2n} f$
is a nonzero multiple of the form
$[1, i]^{\otimes 2n}\pm i[1, -i]^{\otimes 2n}$, which does not
satisfy parity, and hence not in $\mathscr{M}$.
 So $f$ is not in $\widehat{\mathscr{M}}$.
 Hence $f=[1, i]^{\otimes 2n}\pm [1, -i]^{\otimes 2n}$,
which is among the listed types.

 If $f$ has type $\langle 0, 1, 0\rangle$, then $f=[1, 0, \ldots, 0, x]$ with $x\neq 0$, up to a nonzero scalar.
By $f \in \widehat{\mathscr{M}}$ and Lemma~\ref{arity-hat-M}, we have
$x^2 =1$.
 Thus $f=[1, 0, \ldots, 0, \pm 1]$, which is among the listed types.

Summarizing,  we proved
 that if $f\in\widehat{\mathscr{M}}\cap(\mathscr{P}\cup\widetilde{\mathscr{A}})$,
 then $f$ is among the listed types.

\item Note that $\mathscr{P}\cup\widetilde{\mathscr{A}}$ is unchanged under the transformation by
$\left[
\begin{smallmatrix}
1 & 0\\
0 & i
\end{smallmatrix}
\right]$.
Thus
\[
 \widehat{\mathscr{M}}^{\dagger} \cap (\mathscr{P} \cup \widetilde{\mathscr{A}})
 =
 \trans{1}{0}{0}{i}
 \left\{
  \widehat{\mathscr{M}} \cap (\mathscr{P} \cup \widetilde{\mathscr{A}})
 \right\}.
\]
Then the proof of this case follows 
from the previous case by a transformation using
$\left[
\begin{smallmatrix}
1 & 0\\
0 & i
\end{smallmatrix}
\right]$.

\item
If $f\in\bigcap_{k=3}^5 S_k$, then a fortiori, $f\in\mathscr{A}^\dagger\cap(\mathscr{P}\cup\mathscr{A})$.
This implies that \[f\in\{[1, \alpha]^{\otimes 2n},
[1, 0]^{\otimes 2n}, [0, 1]^{\otimes 2n}, [0, 1, 0],
 [1, 0, \ldots, 0, i^r]~\mid~ n \ge 1,~ 0 \le r \le 3\}.\]
 Note that  $[1, \alpha]^{\otimes 2m}\notin\mathscr{A}$.
 Thus $f=[1, 0]^{\otimes 2n}$, or $[0, 1]^{\otimes 2n}$, or $[0, 1, 0]$,
 or $[1, 0, \ldots, 0, i^r]$. All of these four types are among the listed.

\item
We already have
\[ \left\{[0, 1, 0], [1, 0, \ldots, 0, \pm 1]\right\}
\subseteq \bigcap_{1 \le k \le 5} S_k \subseteq   
 \bigcap_{2 \le k \le 5} S_k.\]

 If $f\in\bigcap_{k=2}^5 S_k$, then $f\in\bigcap_{k=3}^5 S_k$.
This implies that
\[f\in\{
[1, 0]^{\otimes 2n}, [0, 1]^{\otimes 2n}, [0, 1, 0],
 [1, 0, \ldots, 0, i^r]~\mid~ n \ge 1,~ 0 \le r \le 3\}.\]
 Moreover,  if $f = [1, 0]^{\otimes 2n}$,
or $[0, 1]^{\otimes 2n}$ or $[1, 0, \dots, 0, \pm i]$,
then $f \notin\widehat{\mathscr{M}}$, because
$\left[\begin{smallmatrix} 1 & 1 \\
1 & -1 \end{smallmatrix}\right]^{\otimes 2n} f$
does not satisfy parity constraints.
  Hence $f=
 [0, 1, 0]$, or
 $[1, 0, \ldots, 0, \pm 1]$,  and both types are among the listed.
\end{enumerate}
\end{proof}

%%%%%%%%%% edited below by JYC

%We directly give the next lemma since the proof is trivial.

We state the following simple lemma which allows us
to replace a signature set $\mathcal{F}$ in the proof of 
the No-Mixing Theorem
by a smaller set $\mathcal{F}'$ that subtracts from $\mathcal{F}$
those signatures that belong to all common tractable signature sets.

%It shows when we prove the No-Mixing theorem, we can delete the signatures in the intersection
%of the tractable sets in $\mathcal{F}$.
%\begin{lemma}\label{mixing-replacing-by}
%For $1\leq k\leq 5$, let $\mathcal{F}$ be a symmetric signature set and
%$\mathcal{F}\cap(\bigcup_{i=1}^{k-1}S_i)=\emptyset$ and
%$\mathcal{F}\cap S_k\neq \emptyset$, $\mathcal{F}\cap(\bigcup_{i=k+1}^{5}S_i)\neq\emptyset$.
%Let $\mathcal{F}'=
%\mathcal{F}\setminus(\bigcap_{i=k}^{5}S_i)$,
%then we have
%$\mathcal{F}'\cap(\bigcup_{i=1}^{k-1}S_i)=\emptyset$ and
%$\mathcal{F}'\cap S_k\neq \emptyset$, $\mathcal{F}'\cap(\bigcup_{i=k+1}^{5}S_i)\neq\emptyset$.
%
%Moreover, $\PlCSP^2(\mathcal{F}')\leq$
%$\PlCSP^2(\mathcal{F})$.
%\end{lemma}

\begin{lemma}\label{mixing-replacing-by}
Let  $\mathcal{F}$ be a set of symmetric signatures such that
for all $1 \le k \le 5$, $\mathcal{F} \not \subseteq S_k$.
Let $\mathcal{F}'= \mathcal{F} \setminus (\bigcap_{k=1}^{5}S_k)$.
Then for all $1 \le k \le 5$, $\mathcal{F}' \not \subseteq S_k$
and $\PlCSP^2(\mathcal{F}')\leq$
$\PlCSP^2(\mathcal{F})$.
\end{lemma}

\begin{proof}
Suppose for some $1 \le k \le 5$, $\mathcal{F}' \subseteq S_k$,
then clearly $\mathcal{F} \subseteq S_k$.
The reduction is trivial since $\mathcal{F}' \subseteq \mathcal{F}$.
\end{proof}

Suppose $\mathcal{F}$ is as given in
Lemma~\ref{mixing-replacing-by},
and $\mathcal{F} \cap(\bigcup_{k=1}^{5}S_k) \not =\emptyset$.
Let $j = \min \{k \mid \mathcal{F} \cap S_k \not = \emptyset, 1 \le k \le 5\}$.
Then $j$ is well defined.
The same proof shows that $\mathcal{F}'=
\mathcal{F} \setminus  (\bigcap_{k=j}^{5}S_k)$
also has the property that $\mathcal{F}'
\not \subseteq S_k$, for $j  \le k \le 5$,
and $\mathcal{F}' \cap S_k = \emptyset$ for $1 \le k < j$.
%In the following we will use this fact.

\begin{corollary}\label{Corollary-to-F11}
Let  $\mathcal{F}$ be a set of symmetric signatures such that
for all $1 \le k \le 5$, $\mathcal{F} \not \subseteq S_k$.
Furthermore suppose 
$\mathcal{F} \cap(\bigcup_{k=1}^{5}S_k) \not =\emptyset$ and
let
$j = \min \{k \mid \mathcal{F} \cap S_k \not = \emptyset, 1 \le k \le 5\}$.
Let $\mathcal{F}'=
\mathcal{F} \setminus  (\bigcap_{k=j}^{5}S_k)$.
Then for all $1 \le k \le 5$, $\mathcal{F}' \not \subseteq S_k$
and $\PlCSP^2(\mathcal{F}')\leq$
$\PlCSP^2(\mathcal{F})$.
\end{corollary}

Recall that
$S_1=\widehat{\mathscr{M}}$,
$S_2=\widehat{\mathscr{M}}^\dagger$,
$S_3=\mathscr{A}^\dagger$,  $S_4=\mathscr{A}$ and
$S_5=\mathscr{P}$.

\begin{theorem}\label{mixing-theorem}
Let $\mathcal{F}\subseteq\bigcup_{k=1}^5 S_k$ be a
set of symmetric signatures of  even arities.
If $\mathcal{F} \subseteq S_k$ for some $1\leq k\leq 5$,
then $\PlCSP^2(\mathcal{F})$ is tractable.  Otherwise,
 %$\mathcal{F}\nsubseteq S_k$ for all $1\leq k\leq 5$,
%then
%%% no need to say. more succint this way
 $\PlCSP^2(\mathcal{F})$ is $\#{\rm P}$-hard.
\end{theorem}

\begin{proof}
If $\mathcal{F} \subseteq S_k$ for some $1\leq k\leq 5$,
then tractability follows from the definition of
$\mathscr{P}$-transformability,
$\mathscr{A}$-transformability and $\mathscr{M}$-transformability.

Now suppose  $\mathcal{F}\nsubseteq S_k$ for all $1\leq k\leq 5$.
We first replace $\mathcal{F}$ by
$\mathcal{F}'= \mathcal{F} \setminus (\bigcap_{k=1}^{5}S_k)$.
This also excludes the identically 0 signature.
By Lemma~\ref{mixing-replacing-by}, we still have
$\mathcal{F}'\nsubseteq S_k$ for $1\leq k\leq 5$,
and we only need to prove
$\PlCSP^2(\mathcal{F}')$ is $\#{\rm P}$-hard.

We will treat the tractable sets in the order $S_1, S_2, \ldots, S_5$,
starting with $S_1 = \widehat{\mathscr{M}}$.
%Firstly, we prove the NO-Mixing of $\widehat{\mathscr{M}}$ and other tractable sets.
%
\begin{enumerate}
\item
%1.
Suppose that $\mathcal{F}'\cap S_1\neq \emptyset$.

Let $\mathcal{G}_1=\mathcal{F}'\cap S_1$,
and $\mathcal{H}_1=\mathcal{F}'\setminus S_1$.
%%% i decided to call it G ad H. not to confuse with ' and ''.
Then $\mathcal{G}_1 \neq \emptyset$, and since $\mathcal{F}' \nsubseteq S_1$
we also have $\mathcal{H}_1\neq \emptyset$.
If there exists $g\in\mathcal{G}_1$ such that $g\in\widehat{\mathscr{M}}\setminus(\mathscr{P}\cup\widetilde{\mathscr{A}})$,
then we are done by Lemma~\ref{M-even}.
Otherwise,
 $\mathcal{G}_1\subseteq\widehat{\mathscr{M}}\cap(\mathscr{P}\cup\widetilde{\mathscr{A}})$.
%%% I think the statement of Lemma~\ref{mixing-clarifying-signature-set}
%%% ignored some signatures in all S_k, namely those [1, 0.,,.0, \pm 1]
%%% [0,1,0] any other???  needs to be fixed...
%%% JYC
Then by the forms given in Lemma~\ref{mixing-clarifying-signature-set},
ignoring nonzero scalars,
 $\mathcal{G}_1\subseteq
\{[1, \pm 1]^{\otimes 2m}, [1, \pm i, 1],
[1, i]^{\otimes 2n}\pm[1, -i]^{\otimes 2n} \mid m\ge 1, n\ge 2\}$.
%Note that we have excluded $\bigcup_{i=1}^{5}S_i$ in $\mathcal{F}'$,
%%% why do you change it to cup?  it is cap
Note that we have excluded $\bigcap_{k=1}^{5}S_k$ in $\mathcal{F}'$,
hence also in $\mathcal{G}_1$.
By Lemma~\ref{mixing-clarifying-signature-set},
 $[1, 0, \ldots, 0, \pm 1], [0, 1, 0] \not \in \mathcal{F}'$.

%If $[1, \pm 1]^{\otimes 2m} \in \mathcal{F}'$, then we
%%% to say G1 is more direct
If $[1, \pm 1]^{\otimes 2m} \in \mathcal{G}_1$ for some $m \ge 1$, then we
can construct $\partial^{m-1}([1, \pm 1]^{\otimes 2m})=2^{m-1}[1, \pm 1]^{\otimes 2}$,
and we are done by Lemma~\ref{mixing-P-global}.

Otherwise, by the forms in
\begin{equation}\label{form-two-sig-in-G1}
\mathcal{G}_1\subseteq
\{[1, \pm i, 1],
[1, i]^{\otimes 2n}\pm[1, -i]^{\otimes 2n} \mid n\ge 2\},
\end{equation}
we have
%$\mathcal{G}_1\subseteq\mathscr{A}\cap\widehat{\mathscr{M}}$.
%%% only used in A
%%% it's easier to see.
%%% of course also in \widehat{\mathscr{M}}, as by def of G1
%%% but the form [1, i]^{\otimes 2n}\pm[1, -i]^{\otimes 2n} in 
%%%  \widehat{\mathscr{M}} is harder to see.  true:Due to arity 2n even
$\mathcal{G}_1\subseteq\mathscr{A}$.
If $\mathcal{H}_1 \subseteq \mathscr{A}$, then we would have
$\mathcal{F}' \subseteq \mathscr{A}$, a contradiction.
Thus $\mathcal{H}_1 \nsubseteq \mathscr{A}$.
Thus there exists $h\in\mathcal{H}_1 \setminus \mathscr{A}$.
By definition of $\mathcal{H}_1$, $h \not \in \widehat{\mathscr{M}}$.
Also $\mathcal{H}_1 \subseteq \bigcup_{k=1}^{5}S_k$, thus
$h\in(\mathscr{P}\cup\mathscr{A}^\dagger\cup\widehat{\mathscr{M}}^\dagger)\setminus(\mathscr{A}\cup\widehat{\mathscr{M}})$.
By the  forms of signatures
in  the nonempty set $\mathcal{G}_1$ in
(\ref{form-two-sig-in-G1})
we have $\mathcal{G}_1 \cap (\mathscr{P}\cup\mathscr{A}^\dagger\cup\widehat{\mathscr{M}}^\dagger)=\emptyset$.
To check this: for the binary $[1, \pm i, 1]$, we apply Lemma~\ref{binary};
for $[1, i]^{\otimes 2n}\pm[1, -i]^{\otimes 2n}$ we use its
second order recurrence of type  $\langle 1,0,1\rangle$  and  then we apply
Lemma~\ref{second-recurrence-relation}.
Thus $\PlCSP^2(\mathcal{F}')$ is $\#{\rm P}$-hard by
Theorem~\ref{mixing-theorem-2}.

\item
%2.
We have $\mathcal{F}'\cap S_1=\emptyset$.
We replace $\mathcal{F}'$ by
$\mathcal{F}''=\mathcal{F}'\setminus(\bigcap_{k=2}^5 S_k)$.
%= \mathcal{F}'\setminus(\bigcap_{k=1}^5 S_k)$.
%By Lemma~\ref{mixing-replacing-by}, we still have
%%% Lemma~\ref{mixing-replacing-by} is technically for exclude all S_i
By Corollary~\ref{Corollary-to-F11}, we still have 
$\mathcal{F}''\nsubseteq S_k$ for $2\leq k\leq 5$,
 $\mathcal{F}'' \cap S_1=\emptyset$,
and we only need to prove
$\PlCSP^2(\mathcal{F}'')$ is $\#{\rm P}$-hard.

Now suppose that $\mathcal{F}''\cap S_2\neq \emptyset$.

By Lemma~\ref{mixing-clarifying-signature-set},
$[1, 0, \ldots, 0, \pm 1], [0, 1, 0] \not \in \mathcal{F}''$.

Let $\mathcal{G}_2=\mathcal{F}''\cap S_2$
and $\mathcal{H}_2=\mathcal{F}''\setminus S_2$.
Both $\mathcal{G}_2$, $\mathcal{H}_2 \neq \emptyset$ and
by definition $\mathcal{H}_2\cap\widetilde{\mathscr{M}}=\emptyset$.
Thus there exists $h\in\mathcal{H}_2 \setminus \widetilde{\mathscr{M}}$.
If there exists $g\in\mathcal{G}_2$ such that $g\in\widehat{\mathscr{M}}^\dagger\setminus(\mathscr{P}\cup\widetilde{\mathscr{A}})$,
then $\PlCSP^2(f, g)$ is $\#{\rm P}$-hard by Theorem~\ref{mixing-theorem-2}.

Otherwise, $\mathcal{G}_2\subseteq\widehat{\mathscr{M}}^\dagger\cap(\mathscr{P}\cup\widetilde{\mathscr{A}})$.
Then $\mathcal{G}_2\subseteq\{[1, \pm i]^{2m}, [1, \pm 1, -1], [1, 1]^{\otimes 2n}\pm[1, -1]^{\otimes 2n} \mid m\geq 1, n\geq 2\}$
by Lemma~\ref{mixing-clarifying-signature-set}.
By its form
%$\mathcal{G}_2\subseteq\mathscr{A}\cap\widehat{\mathscr{M}}^\dagger$.
%%% easier to  to see A. all we need is A. although \widehat{\mathscr{M}}^\d
%%% is also true. in fact by def of G2
$\mathcal{G}_2\subseteq\mathscr{A}$.
If $\mathcal{H}_2 \subseteq \mathscr{A}$, then we would have
$\mathcal{F}'' \subseteq \mathscr{A}$, a contradiction.
Thus $\mathcal{H}_2 \nsubseteq \mathscr{A}$.
Hence there exists $h'\in\mathcal{H}_2 \setminus \mathscr{A}$.
By definition of $\mathcal{H}_2$, $h' \not \in \widetilde{\mathscr{M}}$.
As $\mathcal{F}'' \subseteq \bigcup_{k=2}^5 S_k$,
$h'\in(\mathscr{P}
\cup\mathscr{A}^\dagger)\setminus(\mathscr{A}\cup\widetilde{\mathscr{M}})$.
%or $g'\in\mathscr{A}^\dagger\setminus(\mathscr{A}\cup\widetilde{\mathscr{M}})$.
If $\mathcal{G}_2$ includes
either $[1, \pm 1, -1]$ or
 $[1, 1]^{\otimes 2n}\pm[1, -1]^{\otimes 2n}$ for some $n \ge 2$,
%%%
%%% JYC: Zhiguo, pl check i think the arity is >=4. ie n >=2
%%%
both  are not in $\mathscr{P}\cup\mathscr{A}^\dagger$.
To see this, we apply Corollary~\ref{binary-necessary} to 
the binary $[1, \pm 1, -1]$. For 
$[1, 1]^{\otimes 2n}\pm[1, -1]^{\otimes 2n}$ with $n \ge 2$,
we note its recurrence type $\langle -1, 0, 1\rangle$ and then
apply Lemma~\ref{second-recurrence-relation}.
%by $[1, 1]^{\otimes 2n}\pm[1, -1]^{\otimes 2n}\notin(\mathscr{P}\cup\mathscr{A}^\dagger)$,
 Then $\PlCSP^2(\mathcal{F}'')$ is $\#{\rm P}$-hard by
Theorem~\ref{mixing-theorem-2}.

%If there are no $[1, 1]^{\otimes 2n}\pm[1, -1]^{\otimes 2n}$ in $\mathcal{F}_1$,
We are left with the case where
the nonempty set $\mathcal{G}_2\subseteq\{[1, \pm i]^{2m} \mid m \ge 1\}$.
By its form
$\mathcal{G}_2\subseteq\mathscr{P}\cap\mathscr{A}\cap\widehat{\mathscr{M}}^\dagger$ and
$\mathcal{G}_2\cap\mathscr{A}^\dagger=\emptyset$.
%%%
If there exists $h''\in \mathcal{H}_2 \setminus(\mathscr{A}\cup\mathscr{P})$,
then by definition of $\mathcal{H}_2$
this $h'' \not \in \widetilde{\mathscr{M}}$ as well,
and we conclude that
 $\PlCSP^2(\mathcal{F}'')$ is $\#{\rm P}$-hard by
Theorem~\ref{mixing-theorem-2}.

So we may assume  $\mathcal{H}_2\subseteq \mathscr{A}\cup\mathscr{P}$.
If $\mathcal{H}_2\subseteq \mathscr{A}$, then we would have
$\mathcal{F}'' \subseteq \mathscr{A}$, a contradiction.
Thus there exists $h''' \in (\mathcal{H}_2 \cap
\mathscr{P}) \setminus \mathscr{A}$.
Considering the forms of signatures in $\mathscr{P} \setminus \mathscr{A}$,
it takes the form $h''' = [a,b]^{\otimes 2n}$ with $a^4 \not = b^4$,
 $ab \not =0$, or $h''' = [1, 0, \ldots, 0, x]$ of arity $2n$,
 with $x^4 \not = 0,1$, for some $n \ge 1$.
Taking $h^{(4)}= \partial^{n-1}(h''')$, we get
a nonzero multiple of either $[a,b]^{\otimes 2}$ or $[1, 0, x]$.
Then taking $\partial^{m-1}_{h^{(4)}}([1, \pm i]^{\otimes 2m})$, for some $m
\ge 1$,
where $[1, \pm i]^{\otimes 2m} \in \mathcal{G}_2$ which is nonempty,
we get   a nonzero multiple of $[1, \pm i]^{\otimes 2}$,
and we are done by Lemma~\ref{mixing-P-global}.

\item
%3.
Now we have $\mathcal{F}''\cap S_2=\emptyset$.

We replace $\mathcal{F}''$ by
$\mathcal{F}'''=\mathcal{F}''\setminus(\bigcap_{k=3}^5 S_k)$.
By Lemma~\ref{mixing-replacing-by}, we still have
$\mathcal{F}'''\nsubseteq S_k$ for $3\leq k\leq 5$,
$\mathcal{F}''' \cap (S_1 \cup S_2) =\emptyset$,
and we only need to prove
$\PlCSP^2(\mathcal{F}''')$ is $\#{\rm P}$-hard.

Suppose that $\mathcal{F}'''\cap S_3\neq \emptyset$.

By Lemma~\ref{mixing-clarifying-signature-set},
the following signatures
$[1, 0, \ldots, 0, i^r]$ of arity $2n$, $[0, 1, 0], [1, 0]^{\otimes 2n},
[0, 1]^{\otimes 2n}$ are all out of $\mathcal{F}'''$,
%\not \in \mathcal{F}'''$,
for any  $0 \le r \le 3$ and any $n\geq 1$.

Let $\mathcal{G}_3=\mathcal{F}'''\cap S_3$,
 $\mathcal{H}_3=\mathcal{F}'''\setminus S_3$.
Both $\mathcal{G}_3$, $\mathcal{H}_3\neq \emptyset$.
Thus there exists $h\in\mathcal{H}_3$ such that $h\in(\mathscr{P}\cup\mathscr{A})\setminus(\mathscr{A}^\dagger\cup\widetilde{\mathscr{M}})$.
If there exists $g\in\mathcal{G}_3$ such that $g\in\mathscr{A}^\dagger\setminus\mathscr{P}$,
then by Corollary~\ref{A-2-A-minus-P}, $g \not \in \mathscr{A}$.
%$f=[1, \alpha, -\alpha^2]$ or $f=[1, \alpha]^{\otimes 2n}+i^r[1, \alpha]^{\otimes 2n}$ with $n\geq 2$.
%%% easier to argue by Cor A11
%Note that by its form $f \not \in \mathscr{A}$.
% % $f\in\mathscr{A}^\dagger\setminus(\mathscr{P}\cup\mathscr{A})$.
Thus $\PlCSP^2(g, h)$ is $\#{\rm P}$-hard by Theorem~\ref{mixing-theorem-2}.

Otherwise, we have $\mathcal{G}_3\subseteq\mathscr{A}^\dagger\cap\mathscr{P}$.
Thus we have
$\mathcal{G}_3\subseteq\{[1, \alpha]^{\otimes 2m} \mid m \ge 1\}$.
Note that by Lemma~\ref{mixing-clarifying-signature-set},
we have excluded 
$[1, 0, \ldots, 0, i^r]$ of arity $2n$, $[0, 1, 0], [1, 0]^{\otimes 2n},
[0, 1]^{\otimes 2n}$ which are all in $\bigcap_{k=3}^5 S_k$.
(See Figure~\ref{fig:venn_diagram:A_Adagger_P}.)

Then we have $\partial^{m-1}([1, \alpha]^{\otimes 2m})=(1+\alpha^2)[1, \alpha]^{\otimes 2}$
and we are done by Lemma~\ref{mixing-P-global}.

\item
%4.
Finally we have $\mathcal{F}'''\cap S_3 = \emptyset$.

We have
$\mathcal{F}'''\nsubseteq S_k$ for $4\leq k\leq 5$,
$\mathcal{F}''' \cap (S_1 \cup S_2 \cup S_3) =\emptyset$,
and thus $\mathcal{F}''' \subseteq S_4 \cup S_5$.
Then we are done directly by Theorem~\ref{mixing-theorem-2}.
\qedhere
\end{enumerate}
\end{proof}

\section{Dichotomy Theorem for an Even-Arity Signature}\label{PartII.secH.csp2-dichotomy}

In this section,
we prove the dichotomy theorem for $\PlCSP^2(f)$,
where $f$ has a general even arity $2n$.
If $2n =2$ or $4$,
then it has been proved in Theorem~\ref{thm:parII:k-reg_homomorphism}$'$ and Theorem~\ref{arity-4-dichotomy} respectively.
Thus we will assume $2n \geq 6$.

The next simple lemma is to determine if a symmetric signature satisfies a second order recurrence relation.
In the following proof, we often argue that a signature $f$ does not belong to
$\mathscr{P}\cup\widetilde{\mathscr{A}}\cup\widetilde{\mathscr{M}}$ by
Lemma~\ref{second-recurrence-relation},
and by showing that $f$ does not satisfy any second order recurrence relation.

\begin{lemma} \label{general-second-recurrence}
For  a symmetric signature $f=[f_0, f_1, \ldots, f_n]$,
let $M_f=\left[\begin{smallmatrix} f_0 & f_1 & f_2 \\
f_1 & f_2 & f_3 \\
\vdots & \vdots & \vdots \\
f_{n-2} & f_{n-1} & f_{n}
 \end{smallmatrix}\right]$,
 then $f$ satisfies a second order recurrence relation iff
 $\rank(M_f)\leq 2$.
\end{lemma}

\begin{proof}
 The signature $f$ satisfies a second order recurrence relation
$af_k+bf_{k+1}+cf_{k+2}=0$ for $0\leq k\leq n-2$
%, i.e.,
%$f$ has type $\langle a, -b, c\rangle$, where $(a, b, c)\neq (0, 0, 0)$,
iff
the linear system $M_fX=0$ has a nonzero solution $(a, b, c)^{T}$
iff
$\rank(M_f)\leq 2$.
%
%Conversely, if $r(M_f)\leq 2$, then the linear system $M_fX=0$ has a nonzero solution $(a, b, c)^{T}$.
%This implies that $af_k+bf_{k+1}+cf_{k+2}=0$ for $0\leq k\leq n-2$.
\end{proof}

We often use the following argument to prove hardness:
 Firstly, we prove $f\notin\mathscr{P}\cup\widetilde{\mathscr{A}}\cup\widetilde{\mathscr{M}}$
 using Lemma~\ref{general-second-recurrence}.
 Moreover,
  if we can get
$[1, \omega]^{\otimes 2}$ in $\PlCSP^2(f)$ for some $\omega \neq 0$,
then $\PlCSP^2(f, [1, \omega]^{\otimes 2})$ is $\#$P-hard by  Lemma~\ref{mixing-P-global}.
Or if we can get a signature $g \in
\widehat{\mathscr{M}}\setminus(\mathscr{P}\cup\widetilde{\mathscr{A}})$
in $\PlCSP^2(f)$,
then $\PlCSP^2(f, g)$ is $\#$P-hard by Lemma~\ref{M-even}.
% etc.
%%% JYC what is this etc. if there is anything it should be spelled out.

The next three lemmas are some special cases of Theorem~\ref{general-single-dichotomy} which is the main result of this section.
We prove these lemmas separately to facilitate the presentation
of the proof of Theorem~\ref{general-single-dichotomy}.

\begin{lemma} \label{general-[1,a,0,a,0,a,b]}
Suppose $ab\neq 0$ and $f=[1,a,0,-a,0,a,b]$,
then $\PlCSP^2(f)$ is \numP-hard.
\end{lemma}

\begin{proof}
Note that  $M_f=\left[\begin{smallmatrix} f_0 & f_1 & f_2 \\
f_1 & f_2 & f_3 \\
f_2 & f_3 & f_4 \\
f_3 & f_4 & f_5 \\
f_{4} & f_{5} & f_{6}
 \end{smallmatrix}\right]=
 \left[\begin{smallmatrix} 1 & a & 0 \\
a & 0 & -a \\
0 & -a & 0 \\
-a & 0 & a \\
0 & a & b
 \end{smallmatrix}\right]$ has rank 3.
 Thus  $f$  does not satisfy any second order recurrence relation by Lemma~\ref{general-second-recurrence}.
So $f\notin\mathscr{P}\cup\widetilde{\mathscr{A}}\cup\widetilde{\mathscr{M}}$
by Lemma~\ref{second-recurrence-relation}.

Moreover, we have $\partial_{=_4}(f)=[1, 2a, b]$.
If $[1, 2a, b]$ is degenerate, then $[1, 2a, b]=[1, 2a]^{\otimes 2}$.
We are done since $\PlCSP^2(f, [1, 2a]^{\otimes 2})$ is $\#$P-hard by Lemma~\ref{mixing-P-global}.
Otherwise,
\begin{itemize}
 \item For $b^4\neq 1$,
 we have $[1, 2a, b] \notin \mathscr{P} \cup \widetilde{\mathscr{A}} \cup \widetilde{\mathscr{M}}$
 by Corollary~\ref{binary-necessary} and Lemma~\ref{binary}.
%%% hat-M and hat-M-dagger by Lemma A7. not by Cor A8. for CorA8
%%% do separatly for A and A^+. not by the \tilde{A} statement of Cor A8.
 Thus $\PlCSP^2([1, 2a, b])$ is \numP-hard by Theorem~\ref{thm:parII:k-reg_homomorphism}$'$ and we are done.
%%% JYC: I think to cite A20' is easier here.
%we have $\partial_{[1, 2b', x']}(=_4)=[1, 0, x']$ on the left
%and $f'''=\partial_{[1, 0, x']}(f'')=[1, (1-x')b', 0, -(1-x')b', (x')^2]$.
%Note that $f'''$ is redundant and the determinant of its compressed signature matrix
%is $[1+(x')^2](1-x')b'\neq 0$.
%Thus $\PlCSP^2(f''')$ is $\#$P-hard by Lemma~\ref{4-redundant} and we are done.

\item  For $b^2=-1$,
 we have $\partial^2_{[1, 2a, b]}(=_6)=[1, 0, -1]$ on the left
 and we have $f'=\partial_{[1, 0, -1]}(f)=[1, 2a, 0, -2a, -b]$.
 Note that $f'$ is redundant and the determinant of its compressed signature matrix
is $4(b-1)a^2\neq 0$.
Thus $\PlCSP^2(f')$ is $\#$P-hard by Lemma~\ref{4-redundant} and we are done.

\item For $b^2=1$,
if $(2a)^4\neq 1$,
then we have $[1, 2a, b] \in \widetilde{\mathscr{M}} \setminus (\mathscr{P} \cup \widetilde{\mathscr{A}})$
by Lemma~\ref{M-2-M-NOT-IN-A-AND-P}. Thus
$\PlCSP^2(f, [1, 2a, b])$ is $\#$P-hard by
 Lemma~\ref{M-even} and Lemma~\ref{2-M-even} and we are done.

Otherwise, we have $(2a)^4=1$. This implies that $(2a)^2=\pm b$.
Since  $[1, 2a, b]$ is non-degenerate,
we have $(2a)^2 \neq b$, thus $(2a)^2=-b$.
 Moreover, we have $f''=\partial_{[1, 2a, b]}(f)=[1+(2a)^2, (1-b)a, -(2a)^2, -(1-b)a, b^2+(2a)^2]$.
Note that $f''=[0, 0, 1, 0, 0]$ for $b=1$ and $f''=[2, \pm 1, -1, \mp 1, 2]$ for $b=-1$.
Both of $[0, 0, 1, 0, 0]$  and $[2, \pm 1, -1, \mp 1, 2]$ are redundant and their compressed signature matrices  are
nonsingular.
Thus $\PlCSP^2(f'')$ is $\#$P-hard by Lemma~\ref{4-redundant} and we are done. \qedhere
\end{itemize}
\end{proof}

The next lemma shows that if $\partial(f)=[1, 0]^{\otimes 2n-2}+t[0, 1]^{\otimes 2n-2}$ with $t\neq 0$, then either
 $f=[1, 0]^{\otimes 2n}+t[0, 1]^{\otimes 2n}$
or {\rm Pl}-$\#{\rm CSP}(f)$ is \numP-hard.
We will use this lemma in Theorem~\ref{general-single-dichotomy} for the cases
where $\partial(f)$ is a non-degenerate
 generalized equality {\sc Gen-Eq}.
 %$\partial(f)=[1, 0]^{\otimes 2n-2}+t[0, 1]^{\otimes 2n-2}$
%and $\partial(f)=[1, 1]^{\otimes 2n-2}+i^r[1, -1]^{\otimes 2n-2}$.

\begin{figure}[t]
 \centering
 \begin{tikzpicture}[scale=\scale,transform shape,node distance=\nodeDist,semithick]
   \node[internal] (0)                    {};
   \node[external] (1) [above left  of=0] {};
   \node[external] (2) [below left  of=0] {};
   \node[external] (3) [left        of=1] {};
   \node[external] (4) [left        of=2] {};
   \node[internal] (5) [right       of=0] {};
   \node[external] (6) [above right of=5] {};
   \node[external] (7) [below right of=5] {};
   \node[external] (8) [right       of=6] {};
   \node[external] (9) [right       of=7] {};
  \path (3) edge[in= 135, out=   0,postaction={decorate, decoration={
                                               markings,
                                               mark=at position 0.45 with {\arrow[>=diamond, white] {>}; },
                                               mark=at position 0.45 with {\arrow[>=open diamond]   {>}; } } }] (0)
        (0) edge[out=-135, in=   0]  (4)
             edge[out=  25, in= 155] (5)
             edge[out=  65, in= 115] (5)
             edge[out= -65, in=-115] (5)
             edge[out= -25, in=-155] (5)
         (5) edge[out=  45, in= 180] (8)
             edge[out= -45, in= 180] (9);
   \begin{pgfonlayer}{background}
    \node[draw=\borderColor,thick,rounded corners,fit = (1) (2) (6) (7),inner sep=0pt] {};
%     \node[draw=\borderColor,thick,rounded corners,fit = (1) (2) (6) (7),inner sep=0pt,transform shape=false] {};
   \end{pgfonlayer}
 \end{tikzpicture}
 \caption{Gadget used to obtain a signature whose signature matrix is redundant.
 Both vertices are assigned $f$.}
 \label{fig:8}
\end{figure}
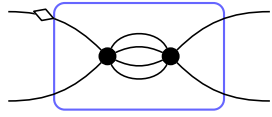

For $f=[a, b]^{\otimes 2n} =
%+[c, d]^{\otimes 2n} =
[f_0, f_1, \ldots, f_{2n}]$ we have $f_k = a^{n-k} b^k$.
% + c^{n-k} d^k$.
Then it is easy to see that
$\bar{f}= [a^2, b^2]^{\otimes n}
%+[c^2, d^2]^{\otimes n}
= [f_0, f_2, \ldots, f_{2n}]$, consisting of even indexed entries of $f$.
This observation also extends to a sum of tensor powers by linearity.
We will use this simple fact in the next lemma.

\begin{lemma} \label{deriviation-is-equality}
Suppose that $(x, y)\neq (0, 0)$ and $f=x[1, i]^{\otimes 2n}+y[1, -i]^{\otimes 2n}+[1, 0]^{\otimes 2n}+t[0, 1]^{\otimes 2n}$,
where $2n\geq 6$ and $t\neq 0$,
then $\PlCSP(f)$ is \numP-hard.
 \end{lemma}

\begin{proof}
Let $a=x+y$, $b=(x-y)i$,
then $(a, b)\neq (0, 0)$.
Note that
\[f=  [a, b, -a, -b, \ldots, \pm b, \mp a] +
[1, 0, \ldots, 0, t]
= [a+1, b, -a, -b, a, \ldots, \pm b, \mp a+t].\]

Since $M_f$ has a rank 3 submatrix
$\left[\begin{smallmatrix} f_0 & f_1 & f_2 \\
f_1 & f_2 & f_3 \\
f_2 & f_3 & f_4 \\
f_{2n-2} & f_{2n-1} & f_{2n}
 \end{smallmatrix}\right]=\left[\begin{smallmatrix} a+1 & b & -a \\
b & -a & -b \\
-a & -b & a \\
\pm a & \pm b & \mp a + t
%\mp a & \mp b & \pm a+t
 \end{smallmatrix}\right]$, $M_f$ has rank 3.
 By Lemma~\ref{general-second-recurrence}, $f$ does not satisfy any second
order recurrence relation.
So $f\notin\mathscr{P}\cup\widetilde{\mathscr{A}}\cup\widetilde{\mathscr{M}}$
by Lemma~\ref{second-recurrence-relation}.
%\begin{itemize}
%\item If $f$ is degenerate, then $f_1^2=f_0f_2$ and $f_2^2=f_1f_3$, i.e.,
%$-a^2-a=b^2$ and $a^2=-b^2$. This implies that $a=b=0$.
%It contradicts that $(a, b)\neq (0, 0)$.
%Thus $f$ is non-degenerate.
%\item $f$ does not satisfy the second order recurrence relation $\langle 0, 1, 0\rangle$ since $f_1\neq 0$ or $f_2\neq 0$.
%\item If $f$ satisfies the second order recurrence relation $\langle 1, 0, i^r\rangle$,
%then by $f_1+i^rf_3=(1-i^r)b$ and $f_2+i^rf_4=(1-i^r)a$, we have $i^r=1$.
%But $f_0+f_2=1\neq 0$. This is a contradiction. Thus $f$ does not satisfy the second recurrence relation $\langle 1, 0, i^r\rangle$.
%\item  If $f$ satisfies the second recurrence relation $\langle 1, c, 1\rangle$ with $c\neq 0$,
% $f_1-cf_2+f_3=ac=0$ and $f_2-cf_3+f_4=bc=0$.
% By $(a, b)\neq (0, 0)$, we have $c=0$. This is a contradiction.
% Thus $f$ does not satisfy the second recurrence relation $\langle 1, c, 1\rangle$.
% \item If $f$ satisfies the second recurrence relation $\langle 1, c, -1\rangle$ with $c\neq 0$,
% $f_0-cf_1+f_2=1+2a-bc=0$ and $f_2-cf_3+f_4=-2a+bc=0$.
% This is a contradiction.
% Thus $f$ does not satisfy the second recurrence relation $\langle 1, c, 1\rangle$.
%\end{itemize}
%Then by Lemma~\ref{second-recurrence-relation}, we have $f\notin\mathscr{P}\cup\widetilde{\mathscr{A}}\cup\widetilde{\mathscr{M}}$.

\begin{enumerate}
\item For $a\neq 0$,
let $\bar{f}=[f_0, f_2, \ldots, f_{2n}]$,
then $\bar{f}=a[1, -1]^{\otimes n}+[1, 0]^{\otimes n}+t[0, 1]^{\otimes n}$
and $\PlCSP(\bar{f})\leq\PlCSP^2(f)$ by Lemma~\ref{domain-pairing-expand}.
Note that $\bar{f}=[a+1, -a, a, \ldots, \pm a, \mp a+t]$
has arity $n\geq 3$.
\begin{itemize}
\item For $2n\geq 8$ or [$2n=6$ and $t \not =-1$],
we claim that $\bar{f}\notin\mathscr{P}\cup\mathscr{A}\cup\widehat{\mathscr{M}}$.
%Thus $\PlCSP(\bar{f})$ is \#P-hard by Theorem~\ref{pl-dicho-1},
%and therefore $\PlCSP^2(f)$ is also \#P-hard.

For $2n\geq 8$,
Since $M_{\bar{f}}$ has a rank 3 submatrix
$\left[\begin{smallmatrix} \bar{f}_0 & \bar{f}_1 & \bar{f}_2 \\
\bar{f}_1 & \bar{f}_2 & \bar{f}_3 \\
\bar{f}_{n-2} & \bar{f}_{n-1} & \bar{f}_{n}
 \end{smallmatrix}\right]=\left[\begin{smallmatrix} a+1 & -a & a \\
-a & a & -a \\
\mp a & \pm a & \mp a +t
%\pm a & \mp a & \pm a+t
 \end{smallmatrix}\right]$,
 $M_{\bar{f}}$ has rank 3.
 Thus $\bar{f}$ does not satisfy any second order recurrence relation
 by Lemma \ref{general-second-recurrence}.
So $\bar{f}\notin\mathscr{P}\cup\mathscr{A}\cup\widehat{\mathscr{M}}$
by Lemma~\ref{second-recurrence-relation}.

For $2n=6$ and $t\neq -1$,
$M_{\bar{f}}$ is a $2\times 3$ matrix and has rank less than $3$.
So it always satisfies a second order recurrence relation.
But we still show  that $\bar{f}\notin\mathscr{P}\cup\mathscr{A}\cup\widehat{\mathscr{M}}$.

Note that $\bar{f}=[a+1, -a, a, -a+t]$ when $n=3$.
%for $2n=6$.
% and $t\neq -1$.
\begin{itemize}
\item $\bar{f}$ is non-degenerate by $(\bar{f}_1)^2\neq \bar{f}_0\bar{f}_2$
and $\bar{f}$ is not {\sc Gen-Eq} since $\bar{f}_1\neq 0$, so $\bar{f}\notin\mathscr{P}$.

\item If $\bar{f}\in\mathscr{A}\setminus\mathscr{P}$, then
$\bar{f}$ has type $\langle 1, 0, \pm 1\rangle$
by Lemma~\ref{second-recurrence-relation}.
By $\bar{f}_0-\bar{f}_2\neq 0$, $\bar{f}$ does not
have type $\langle 1, 0, -1\rangle$.
% satisfy the second order recurrence relation
%$\langle 1, 0, -1\rangle$.
If $\bar{f}$ has type
$\langle 1, 0, 1\rangle$,
then $\bar{f}_0+\bar{f}_2=0$, $\bar{f}_1+\bar{f}_3=0$.
% and $\bar{f}_1=-\bar{f}_2=-a$,
%we have $\bar{f}_0=-\bar{f}_3$.
This implies $t=-1$. It is a contradiction.
Thus $\bar{f}\notin\mathscr{A}\setminus\mathscr{P}$.

\item By $\bar{f}_1=-\bar{f}_2\neq 0$ and Lemma~\ref{arity-hat-M}, if $\bar{f}\in\widehat{\mathscr{M}}$,
then $\bar{f}_0=-\bar{f}_3$. This contradicts that $t\neq -1$.
Thus $\bar{f}\notin\widehat{\mathscr{M}}$.
%\setminus\mathscr{P}$.
\end{itemize}

To summarize, $\bar{f}\notin\mathscr{P}\cup\mathscr{A}\cup\widehat{\mathscr{M}}$
for  $2n\geq 8$, or [$2n=6$ and $t\neq -1$].
Thus $\PlCSP(\bar{f})$  is $\#$P-hard by Theorem~\ref{pl-dicho-1}.
So $\PlCSP^2(f)$  is $\#$P-hard.

\item For $2n=6$ and $t=-1$, we have $f=[a+1, b, -a, -b, a, b, -a-1]$.
Firstly, we have $\partial^2(f)=[1, 0, -1]$ and $f'=\partial_{[1, 0, -1]}(f)=[1+2a, 2b, -2a, -2b, 1+2a]$.
The compressed signature matrix of $f'$ is
$\left[\begin{smallmatrix} 1+2a & 2b & -2a \\
2b & -2a & -2b\\
-2a & -2b & 1+2a \end{smallmatrix}\right]$
and its determinant is $-2(4 a^2 + 4 b^2 + a)$.
 If $4 a^2 + 4 b^2 + a \neq 0$,
 then it is nonsingular,
 and we are done by Lemma~\ref{4-redundant}.

 Otherwise we have $4 a^2 + 4 b^2 + a = 0$.
 Consider the gadget in Figure~\ref{fig:8}.
 We assign $f$ to both vertices.
 The signature of this gadget is redundant,
 and its compressed signature matrix is
 \[
  \begin{bmatrix}
   1+2a+8a^2+8b^2 & b & -2a-8a^2-8b^2 \\
   b & 8a^2+8b^2 & -b\\
   -2a-8a^2-8b^2 & -b & 1+2a+8a^2+8b^2
  \end{bmatrix}
  =
  \begin{bmatrix}
   1 & b & 0 \\
   b & -2a & -b\\
   0 & -b & 1
  \end{bmatrix}.
 \]
 If $a + b^2 \neq 0$,
 then this matrix is nonsingular,
 so we are done by Lemma~\ref{4-redundant}.

 Otherwise we have $4 a^2 + 4 b^2 + a = 0$ and $a + b^2 = 0$.
Also we have $a \not =0$.
 By solving these two equations, $a = \frac{3}{4}$ and $b  =\pm \frac{\sqrt{3}}{2} i$.
 Moreover, we have $\partial_{=_4}(f)=[1+2a, 2b, -1-2a]=[\frac{5}{2}, \pm \sqrt{3}i, -\frac{5}{2}]$.
 By Lemma~\ref{M-2-M-NOT-IN-A-AND-P}, $\partial_{=_4}(f)
\in \widehat{\mathscr{M}}^\dagger\setminus(\mathscr{P}\cup\mathscr{A})$.
Recall that $f \notin \widehat{\mathscr{M}}^\dagger$.
%Note that $f$ does not satisfy the second recurrence relation $\langle 0, 1, 0\rangle$ and $\langle 1, c, -1\rangle$.
%Thus $f\notin\widehat{\mathscr{M}}^\dagger$.
Thus $\PlCSP^2(f, [1+2a, 2b, -1-2a])$ is $\#$P-hard by Lemma~\ref{2-M-even} and we are done.
\end{itemize}

\item For $a=0$, then $b\neq 0$ by $(a, b)\neq (0, 0)$.
\begin{itemize}
\item
if $2n\equiv 0\pmod 4$ and $t\neq -1$, then
\[f''=\partial^{\frac{n-2}2}_{=_4}(f)=2^{\frac{n-2}2}x[1, i]^{\otimes 4}+2^{\frac{n-2}2}y[1, -i]^{\otimes 4}+[1, 0]^{\otimes 4}+t[0, 1]^{\otimes 4},\]
i.e., $f''=[1, 2^{\frac{n-2}2}b, 0, -2^{\frac{n-2}2}b, t]$.
Note that $f''$ is redundant and the determinant of its compressed signature matrix
is $-2^{n-2}b^2(t+1)$.
By $t\neq -1$ and $b\neq 0$, the compressed signature matrix
 is nonsingular.
So $\PlCSP^2(f'')$ is \numP-hard by Lemma~\ref{4-redundant}.
Thus $\PlCSP^2(f)$ is \numP-hard.

\item
if $2n\equiv 0\pmod 4$ and $t=-1$, we have $\partial^{n-1}(f)=[1, 0, -1]$ and
\[f'''=\partial^{n-3}_{[1, 0, -1]}(f)=2^{n-3}x[1, i]^{\otimes 6}+2^{n-3}y[1, -i]^{\otimes 6}+[1, 0]^{\otimes 6}+(-1)^{n-2}[0, 1]^{\otimes 6},\]
i.e., $f'''=[1, 2^{n-3}b, 0, -2^{n-3}b, 0, 2^{n-3}b, (-1)^{n-2}]$.
By Lemma~\ref{general-[1,a,0,a,0,a,b]}, $\PlCSP^2(f''')$ is \#P-hard and we are done.

\item if $2n\equiv 2\pmod 4$, we have
\[f^{(4)}=\partial^{\frac{n-3}2}_{=_4}(f)=2^{\frac{n-3}2}x[1, i]^{\otimes 6}+2^{\frac{n-3}2}y[1, -i]^{\otimes 6}+[1, 0]^{\otimes 6}+t[0, 1]^{\otimes 6}.\]
Note that $f^{(4)}=[1, 2^{\frac{n-3}2}b, 0, -2^{\frac{n-3}2}b, 0, 2^{\frac{n-3}2}b, t]$.
By Lemma~\ref{general-[1,a,0,a,0,a,b]}, $\PlCSP^2(f^{(4)})$ is \#P-hard and we are done. \qedhere
\end{itemize}
\end{enumerate}
\end{proof}

We will use the next lemma in the proof of Theorem~\ref{general-single-dichotomy} for
the case that $\partial(f)=[1, i]^{\otimes 2n-2}+i^r[1, -i]^{\otimes 2n-2}$.
In this case, we will transform $\PlCSP^2$ to $\PlCSP^4$
by holographic transformation and gadget construction.
This is why we have to deal with $\PlCSP^4$ problems in the next lemma.

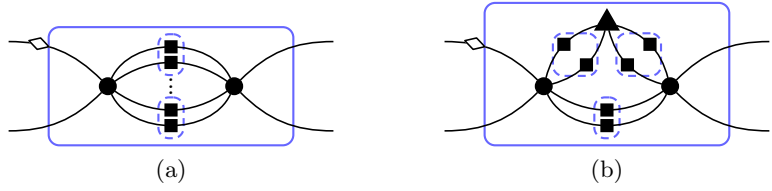
\begin{figure}[t]
 \centering
 \subfloat[]{
  \begin{tikzpicture}[scale=\scale,transform shape,node distance=\nodeDist,semithick]
   \node[internal]  (0)                    {};
   \node[external]  (1) [above left  of=0] {};
   \node[external]  (2) [below left  of=0] {};
   \node[external]  (3) [left        of=1] {};
   \node[external]  (4) [left        of=2] {};
   \node[external]  (5) [right       of=0] {\huge \raisebox{6pt}{$\vdots$}};
   \node[internal]  (6) [right       of=5] {};
   \node[external]  (7) [above right of=6] {};
   \node[external]  (8) [below right of=6] {};
   \node[external]  (9) [right       of=7] {};
   \node[external] (10) [right       of=8] {};
  \path (3) edge[in= 135, out=   0,postaction={decorate, decoration={
                                               markings,
                                               mark=at position 0.4 with {\arrow[>=diamond, white] {>}; },
                                               mark=at position 0.4 with {\arrow[>=open diamond]   {>}; } } }] (0)
        (0) edge[out=-135, in=   0]  (4)
             edge[out=  65, in= 115] node[square] (s1) {} (6)
             edge[out=  35, in= 145] node[square] (s2) {} (6)
             edge[out= -35, in=-145] node[square] (s3) {} (6)
             edge[out= -65, in=-115] node[square] (s4) {} (6)
         (6) edge[out=  45, in= 180]  (9)
             edge[out= -45, in= 180] (10);
   \begin{pgfonlayer}{background}
    \node[draw=\borderColor,thick,rounded corners,densely dashed,fit = (s1) (s2),inner sep=4pt] {};
%     \node[draw=\borderColor,thick,rounded corners,densely dashed,fit = (s1) (s2),inner sep=2pt,transform shape=false] {};
    \node[draw=\borderColor,thick,rounded corners,densely dashed,fit = (s3) (s4),inner sep=4pt] {};
%     \node[draw=\borderColor,thick,rounded corners,densely dashed,fit = (s3) (s4),inner sep=2pt,transform shape=false] {};
    \node[draw=\borderColor,thick,rounded corners,fit = (1) (2) (7) (8)] {};
%     \node[draw=\borderColor,thick,rounded corners,fit = (1) (2) (7) (8),inner sep=2pt,transform shape=false] {};
   \end{pgfonlayer}
  \end{tikzpicture}
  \label{subfig:9-1}
 }
 \qquad
 \subfloat[]{
  \begin{tikzpicture}[scale=\scale,transform shape,node distance=\nodeDist,semithick]
   \node[internal]  (0)                    {};
   \node[external]  (1) [above left  of=0] {};
   \node[external]  (2) [below left  of=0] {};
   \node[external]  (3) [left        of=1] {};
   \node[external]  (4) [left        of=2] {};
   \node[external]  (5) [right       of=0] {};
   \node[triangle]  (6) [above       of=5] {};
   \node[internal]  (7) [right       of=5] {};
   \node[external]  (8) [above right of=7] {};
   \node[external]  (9) [below right of=7] {};
   \node[external] (10) [right       of=8] {};
   \node[external] (11) [right       of=9] {};
  \path (3) edge[in= 135, out=   0,postaction={decorate, decoration={
                                               markings,
                                               mark=at position 0.4 with {\arrow[>=diamond, white] {>}; },
                                               mark=at position 0.4 with {\arrow[>=open diamond]   {>}; } } }] (0)
        (0) edge[out=-135, in=   0]  (4)
             edge[bend left]         node[square] (s1) {} (6)
             edge[bend right]        node[square] (s2) {} (6)
             edge[out= -35, in=-145] node[square] (s3) {} (7)
             edge[out= -65, in=-115] node[square] (s4) {} (7)
         (6) edge[bend left]         node[square] (s5) {} (7)
             edge[bend right]        node[square] (s6) {} (7)
         (7) edge[out=  45, in= 180] (10)
             edge[out= -45, in= 180] (11);
   \begin{pgfonlayer}{background}
    \node[draw=\borderColor,thick,rounded corners,densely dashed,fit = (s1) (s2),inner sep=3pt] {};
%     \node[draw=\borderColor,thick,rounded corners,densely dashed,fit = (s1) (s2),inner sep=2pt,transform shape=false] {};
    \node[draw=\borderColor,thick,rounded corners,densely dashed,fit = (s3) (s4),inner sep=4pt] {};
%     \node[draw=\borderColor,thick,rounded corners,densely dashed,fit = (s3) (s4),inner sep=2pt,transform shape=false] {};
    \node[draw=\borderColor,thick,rounded corners,densely dashed,fit = (s5) (s6),inner sep=3pt] {};
%     \node[draw=\borderColor,thick,rounded corners,densely dashed,fit = (s5) (s6),inner sep=2pt,transform shape=false] {};
    \node[draw=\borderColor,thick,rounded corners,fit = (1) (2) (6) (8) (9)] {};
%     \node[draw=\borderColor,thick,rounded corners,fit = (1) (2) (6) (8) (9),inner sep=2pt,transform shape=false] {};
   \end{pgfonlayer}
  \end{tikzpicture}
  \label{subfig:9-2}
 }
 \caption{Two gadgets used to obtain a signature whose signature matrix is redundant.
 The dashed subgadgets are assigned $[1,0,1]^{\otimes 2}$ rotated so that it is equivalent to assigning $[1,0,1]$ to the square vertices.}
 \label{fig:9}
\end{figure}

\begin{lemma} \label{general-x=y=0-[1,i]-[1,-i]}
 Suppose $f = [0, 1, 0, \dotsc, 0, a, 0]$ has arity $2n \geq 6$.
 If $a^4 = 1$,
 then the problem $\PlCSP^4(f,$ $[1, 0, 1]^{\otimes 2}, [1, 0, 1, 0, 1])$ is \numP-hard.
\end{lemma}

\begin{proof}
In $\PlCSP^4(f,$ $[1, 0, 1]^{\otimes 2}, [1, 0, 1, 0, 1])$, we do not have
$=_2$ on the left, so we cannot connect the two edges on the right freely.
 But we do have $[1, 0, 1]^{\otimes 2}$ on the right and $=_4$ on the left,
so we can do a loop to a pair of $=_4$ on the left respectively and we get $[1, 0, 1]^{\otimes 2}$ on the left.

 Suppose $a^2 = 1$.
 Consider the gadget in Figure~\ref{subfig:9-1}.
 We assign $f$ to the circle vertices and $[1, 0, 1]^{\otimes 2}$ to the dashed subgadgets
 rotated so that it is equivalent to assigning $[1,0,1]$ to the square vertices,
where there are $2n-2$ parallel edges connecting
the 2 copies of $f$ with $2n-2$ square vertices.
 The signature of this gadget is redundant,
 and its compressed signature matrix is
 $\left[
  \begin{smallmatrix}
   2n-2 & 0 & 0 \\
   0    & 1+a^2 & 0\\
   0    & 0 & (2n-2)a^2
  \end{smallmatrix}
 \right]$,
 which is nonsingular, by $a^2=1$.
%%% just need a not 0 and a^2 not -1
 Thus
 we have
 \[
  \PlCSP^2(f')
  \leq_T
  \PlCSP^4(f', [1, 0, 1]^{\otimes 2})
  \leq_T
  \PlCSP^4(f,[1, 0, 1]^{\otimes 2}, [1, 0, 1, 0, 1]),
 \]
 where the first $\leq_T$ is by Lemma~\ref{mixing-P-global-binary}.
 Then we are done by Lemma~\ref{4-redundant}.

 For $a^2 = -1$,
 the gadget in Figure~\ref{subfig:9-1} cannot work since the compressed signature matrix
 of its resulting signature is $\left[
  \begin{smallmatrix}
   2n-2 & 0 & 0 \\
   0    & 0 & 0\\
   0    & 0 & -2n+2
  \end{smallmatrix}
 \right]$ which is singular.

 We consider two cases.
 \begin{itemize}
  \item Suppose $2n \equiv 0 \pmod 4$.
  Then by Lemma~\ref{mixing-P-global-binary},
  we have
  \begin{equation}\label{general-4-2}
   \PlCSP^2(f, [1, 0, 1, 0, 1])
   \leq_T
   \PlCSP^4(f, [1, 0, 1, 0, 1], [1, 0, 1]^{\otimes 2}).
  \end{equation}
  In $\PlCSP^2(f, [1, 0, 1, 0, 1])$,
  we have $f' = \partial^{n-2}(f) = [0, 1, 0, \pm i, 0]$.
%and $\partial_{[0, 1, 0, \pm i, 0]}(\widehat{f^{(4)}})=4[0, 1, 0, -1, 0]$.
  Note that $f' \in \mathscr{A}^\dagger$
by considering
 $\left[\begin{smallmatrix} 1 & 0 \\
0 & \alpha \end{smallmatrix}\right]^{\otimes 4} f'$,
and also $f' \notin \mathscr{P} \cup \mathscr{A} \cup \widetilde{\mathscr{M}}$
by considering its type  $\langle 1, 0, \pm i\rangle$,
and by Lemma~\ref{second-recurrence-relation}.
  Furthermore we have
$[1,0,1,0,1] \in \mathscr{A}$,
and also $[1,0,1,0,1]  \notin \mathscr{A}^\dagger$ by its type
 $\langle 1, 0, -1 \rangle$,
and by Lemma~\ref{second-recurrence-relation}.
%Moreover, we have $\partial_{=_4}(\widehat{f^{(4)}})=(1\pm i)[0, 1, 0]$,
%$\partial_{[0, 1, 0]}(\widehat{f^{(4)}})=[1, 0, 0, 0, \pm i]$, $\partial([1, 0, 0, 0, \pm i])=[1, 0, \pm i]$,
%and $\widehat{f^{(7)}}=\partial([1, 0, \pm i])=[0, 1, 0, -1, 0]$.
%Note that $\widehat{f^{(7)}}\in\mathscr{A}\setminus\mathscr{A}^\dagger$.
  Thus $\PlCSP^2(f', [1,0,1,0,1])$ is \numP-hard by Theorem~\ref{mixing-theorem} and we are done.

  \item For  $2n \equiv 2 \pmod 4$, we cannot use Lemma~\ref{mixing-P-global-binary}
  to get the reduction in (\ref{general-4-2})
  since Lemma~\ref{mixing-P-global-binary} requires
that  all signatures on the right have arity $\equiv 0\pmod 4$.
  But we have $f' = \partial^{\frac{n-3}2}_{=_4}(f) = [0, 1, 0, 0, 0, \pm i, 0]$
  as well as $\partial_{=_4}(f') = (1\pm i)[0, 1, 0]$.
%%% i added some explanation to move [0, 1, 0] to LHS to use for derivative
We may use $[1, 0, 1]^{\otimes 2}$ of the LHS to transport
this $[0, 1, 0]$ from the RHS to the LHS as follows:
Let $f(x_1, y_1, x_2, y_2)$ be the function $[1, 0, 1]^{\otimes 2}$
which is $1$ iff $x_1 = y_1$ and $x_2 = y_2$, and $0$ otherwise.
Then we connect $x_1$ and $x_2$ with the two edges of $[0, 1, 0]$ from the RHS.
This creates $[0, 1, 0]$ on the LHS, with which we can take derivative
of $f'$ from the RHS.
  Then we have $\partial_{[0, 1, 0]}(f') = [1, 0, 0, 0, \pm i]$.
  Consider the gadget in Figure~\ref{subfig:9-2}.
  We assign $f'$ to the circle vertices,
  $[1, 0, 0, 0,\pm i]$ to the triangle vertex,
  and $[1, 0, 1]^{\otimes 2}$ to the dashed subgadgets
  rotated so that it is equivalent to assigning $[1,0,1]$ to the square vertices.
  The signature $f''$ of this gadget is redundant,
  and its compressed signature matrix is
  $\left[
   \begin{smallmatrix}
    2 &       0 &      0 \\
    0 & 1 \mp i &      0 \\
    0 &       0 & \mp 2i
   \end{smallmatrix}
  \right]$,
  which is nonsingular.
  Thus $\PlCSP^2(f'')$ is $\#$P-hard by Lemma~\ref{4-redundant}.
  Moreover,
  we have
  \[
   \PlCSP^4(f'', [1, 0, 1]^{\otimes 2})
   \leq_T
   \PlCSP^4(f, [1, 0, 1, 0, 1], [1, 0, 1]^{\otimes 2})
  \]
  and
  \[
   \PlCSP^2(f'')
   \leq_T
   \PlCSP^4(f'', [1, 0, 1]^{\otimes 2})
  \]
  by Lemma~\ref{mixing-P-global-binary} and we are done.
  Now Lemma~\ref{mixing-P-global-binary} can work since $f''$ has arity 4.
  \qedhere
 \end{itemize}
\end{proof}

Now we are ready to prove the main theorem of this section,
the dichotomy of $\PlCSP^2(f)$, where $f$ has a general even arity $2n$.
We will prove the theorem by induction on the arity $2n$.
The base cases $2n=2$ and $2n=4$ are already done in
Theorem~\ref{thm:parII:k-reg_homomorphism}$'$ and Theorem~\ref{arity-4-dichotomy}, respectively.
We always have $f'=\partial(f)$ in $\PlCSP^2(f)$ which has arity $2n-2$.
If $f'\notin \mathscr{P} \cup \widetilde{\mathscr{A}} \cup \widetilde{\mathscr{M}}$,
then $\PlCSP^2(f')$ is \numP-hard by induction and
$\PlCSP^2(f)$ is
\numP-hard.
Otherwise, for $f'\in \mathscr{P} \cup \widetilde{\mathscr{A}} \cup \widetilde{\mathscr{M}}$,
we can explicitly express $f$ by the integral operator
$\int (f')$.
% Lemma~\ref{general-f-construction}.
We will prove the theorem in the following order:\\
(1) $f'\in\mathscr{P}$, (2) $f'\in\mathscr{A}^{\dagger}\setminus\mathscr{P}$,
(3) $f'\in\mathscr{A}\setminus\mathscr{P}$,
(4) $f'\in\widehat{\mathscr{M}}\setminus(\mathscr{P}\cup\widetilde{\mathscr{A}})$,
and (5)
$f'\in\widehat{\mathscr{M}}^{\dagger}\setminus(\mathscr{P}
\cup\widetilde{\mathscr{A}})$.

Note that by Corollary~\ref{M-2-M-A-P},
Case (4) is equivalent to
$f'\in\widehat{\mathscr{M}}\setminus(\mathscr{P}\cup\widetilde{\mathscr{A}}
\cup\widehat{\mathscr{M}}^{\dagger})$,
and Case (5) is equivalent to
$f'\in\widehat{\mathscr{M}}^{\dagger}\setminus(\mathscr{P}
\cup\widetilde{\mathscr{A}}\cup\widehat{\mathscr{M}})$.

In the proof,
to use Theorem~\ref{arity-4-dichotomy},
we often construct arity 4 signatures by our Calculus with binary signatures or $=_4$.

\begin{theorem} \label{general-single-dichotomy}
 Let $f$ be a symmetric signature of even arity $2n$.
If $f \in \mathscr{P} \cup \widetilde{\mathscr{A}}
\cup \widetilde{\mathscr{M}}$,
 then $\PlCSP^2(f)$ is tractable.  Otherwise,
 $\PlCSP^2(f)$ is \numP-hard.
\end{theorem}

\begin{proof}
If $f \in \mathscr{P} \cup \widetilde{\mathscr{A}} \cup \widetilde{\mathscr{M}}$,
then tractability follows from the definition of
$\mathscr{P}$-transformability,
$\mathscr{A}$-transformability,
and $\mathscr{M}$-transformability.
Now suppose $f \notin \mathscr{P} \cup \widetilde{\mathscr{A}} \cup \widetilde{\mathscr{M}}$.
If $2n \in \{2, 4\}$,
then we are done by Theorem~\ref{thm:parII:k-reg_homomorphism}$'$ and Theorem~\ref{arity-4-dichotomy} respectively.

For $2n\geq 6$,
we will prove the theorem by induction on arity $2n$.
If $f'=\partial(f)\notin\mathscr{P}\cup\widetilde{\mathscr{A}}\cup\widetilde{\mathscr{M}}$,
then $\PlCSP^2(f')$ is $\#$P-hard by induction.
Thus $\PlCSP^2(f)$ is $\#$P-hard.
Otherwise, $f'\in\mathscr{P}\cup\widetilde{\mathscr{A}}\cup\widetilde{\mathscr{M}}$.

\begin{enumerate}
\item For $f'\in\mathscr{P}$, we have
$f'\equiv 0$ or $f'=[a, b]^{\otimes 2n-2}$ (where $(a,b) \not = (0,0)$)
 or $f'=[1, 0]^{\otimes 2n-2}+t[0, 1]^{\otimes 2n-2}$ with $t\neq 0$
by definition. Note that $2n -2 \ge 4$.
% (See Figure \ref{fig:venn_diagram:A_Adagger_P}).
%%% I don't think at his point the fig 35 is very helpful.
%%% for P, it 's easy, just degenerate or gen-eq, can't be [0,1,0]  as 2n-2>=4
\begin{enumerate}
\item $f'\equiv 0$. Then $f=x[1, i]^{\otimes 2n}+y[1, -i]^{\otimes 2n}$
by Proposition~\ref{prop:explicit_list}
(the Explicit List for $\int(f')$).

If $x=0$ or $y=0$, then $f\in\mathscr{P}$. If $xy\neq 0$ and $x^4=y^4$,
then $f\in\mathscr{A}$. In the following, assume that $xy\neq 0$ and $x^4\neq y^4$.

\begin{itemize}
\item For $2n\equiv 0\pmod 4$, we have
$f''=\partial^{\frac{n-2}2}_{=_4}(f)=2^{\frac{n-2}2}x[1, i]^{\otimes 4}+2^{\frac{n-2}2}y[1, -i]^{\otimes 4}$.
By $xy\neq 0$, $f''$ is non-degenerate,
and has the unique recurrence type $\langle 1,0, 1\rangle$.
Therefore $f''\notin\mathscr{P}\cup \mathscr{A}^\dagger \cup\widetilde{\mathscr{M}}$
by Lemma~\ref{second-recurrence-relation}. By $x^4\neq y^4$
it is also not in $\mathscr{A}$. Thus
$f''\notin\mathscr{P}\cup\widetilde{\mathscr{A}}\cup\widetilde{\mathscr{M}}$.
Therefore {\rm Pl-}$\#${\rm CSP}$^2(f'')$ is $\#$P-hard  by Theorem~\ref{arity-4-dichotomy}.
So  $\PlCSP^2(f)$ is $\#$P-hard.

\item For $2n\equiv 2\pmod 4$,
we cannot reduce the arity of $f$ to 4 by $=_4$ directly as
in the previous case.
We will construct a binary signature that is not $\lambda[1, 0, 1]$ to reduce the arity of $f$.
Firstly, we have $f'''=\partial^{\frac{n-1}2}_{=_4}(f)
=2^{\frac{n-1}2}x[1, i]^{\otimes 2}+2^{\frac{n-1}2}y[1, -i]^{\otimes 2}=2^{\frac{n-1}2}[a, b, -a]$, where $a=x+y, b=(x-y)i$.
We remark that $[a, b, -a]$ can reduce the arity of $f$, but it involves a case analysis of $a$ and $b$.
Instead  we use $[a, b, -a]$ to construct a simpler binary signature.

Note that $a\neq 0$ by $x^4\neq y^4$. Then we have $\partial_{[a, b, -a]}(=_4)=a[1, 0, -1]$ on the left.
Thus we have $f^{(4)}=\partial^{n-2}_{[1, 0, -1]}(f)=2^{n-2}x[1, i]^{\otimes 4}+2^{n-2}y[1, -i]^{\otimes 4}$.
With the same reason as in the previous case,
 $f^{(4)}\notin\mathscr{P}\cup\widetilde{\mathscr{A}}\cup\widetilde{\mathscr{M}}$ by its type, and by $xy\neq 0, x^4\neq y^4$.
Thus  $\PlCSP^2(f^{(4)})$ is $\#$P-hard  by Theorem~\ref{arity-4-dichotomy}.
So  {\rm Pl}-$\#{\rm CSP}^2(f)$ is $\#$P-hard.
\end{itemize}

\item $f'=[a, b]^{\otimes 2n-2}$ with $ab\neq 0$.
If $a^2+b^2\neq 0$, we have $\partial^{n-2}(f)=(a^2+b^2)^{n-2}[a, b]^{\otimes 2}$ and
we are done by Lemma~\ref{mixing-P-global}.

Suppose $a^2+b^2=0$, i.e., $f'=[1, \pm i]^{\otimes 2n-2}$ up to a scalar.
\begin{itemize}
\item For $2n\equiv 0\pmod 4$,
we have $\partial^{\frac{n-2}2}_{=_4}(f')=2^{\frac{n-2}2}[1, \pm i]^{\otimes 2}$
and are done by Lemma~\ref{mixing-P-global}.

\item For $2n \equiv 2\pmod 4$,
we cannot get $[1, \pm i]^{\otimes 2}$ in $\PlCSP^2(f')$ by Remark~\ref{rmk:E2:2}
(note that the arity of $f'$ is
$2n-2 \equiv  0 \pmod 4$).
% statement of \textcolor{red}{remark 2}.
To get $[1, \pm i]^{\otimes 2}$,
we need the help of $f$.
By Proposition~\ref{prop:explicit_list}
(the Explicit List for $\int(f')$),
 $f=x[1, i]^{\otimes 2n}+y[1, -i]^{\otimes 2n}+g$,
where $g$ has arity $2n$ and $g_k=\frac{1}{4}(\pm i)^k(2n-2k)$.
If $x=y=0$, then $f\in\widehat{\mathscr{M}}^\dagger$.
Otherwise, let $a=x+y$, $b=(x-y)i$, then $(a, b)\neq (0, 0)$.
We have $\partial^{\frac{n-1}2}_{=_4}(f)=
2^{\frac{n-1}2}x[1, i]^{\otimes 2}+2^{\frac{n-1}2}y[1, -i]^{\otimes 2}+2^{\frac{n-3}2}[1, 0, 1]$,
i.e., $\partial^{\frac{n-1}2}_{=_4}(f)=2^{\frac{n-3}2}[2a+1, 2b, -2a+1]$.

If $a\neq 0$,
then we have $\partial_{[2a+1, 2b, -2a+1]}(=_4)=[2a+1, 0, -2a+1]$ on the left
and $\partial^{n-2}_{[2a+1, 0, -2a+1]}(f')=(4a)^{n-2}[1, \pm i]^{\otimes 2}$.
Then we are done by Lemma~\ref{mixing-P-global}.

If $a=0$, then $b\neq 0$ and we have $[1, 2b, 1]$ and $\partial^{n-2}_{[1, 2b, 1]}(f')=(\pm 4bi)^{n-2}[1, \pm i]^{\otimes 2}$.
Then we are done by Lemma~\ref{mixing-P-global} again.
%
%
%If $y\neq 0$,
%we have $\partial_{f'}(f)=2^{2n-2}y[1, -i]^{\otimes 2}$, then we are done by Lemma~\ref{mixing-P-global}.
%
%For $x\neq 0$, $y=0$, we have $h=\partial^{\frac{n-1}2}_{=_4}(f)=2^{\frac{n-1}2}x[1, i]^{\otimes 2}+g'$,
%where $g'$ has arity 2 and $g'_k=2^{\frac{n-3}2}(\pm i)^k(2-2k)$, i.e., $h=2^{\frac{n-3}2}[1+2x, 2ix, 1-2x]$.
%Then we have $\partial_{h}(=_4)=[1+2x, 0, 1-2x]$
%on the left and $\partial^{n-2}_{[1+2x, 0, 1-2x]}(f')=(4x)^{n-2}[1, \pm i]^{\otimes 2}$.
%Thus we are done by Lemma~\ref{mixing-P-global}.
\end{itemize}

%we consider the case that $f'=[1, i]^{\otimes 2n-2}$.
%By Lemma~\ref{general-f-construction}, $f=x[1, i]^{\otimes 2n}+y[1, -i]^{\otimes 2n}+g$,
%where $g$ has arity $2n$ and $g_k=\frac{1}2i^k(2n-2k)$.
%If $x=y=0$, then $f\in\widehat{\mathscr{M}}^\dagger$.
%For $y\neq 0$,
 %
 %
%For $y=0$, we have $x\neq 0$ and $\partial_{=_4}(f)=2x[1, i]^{\otimes 2n-4}+g'$,
%where $g'$ has arity $2n-4$ and $g_k=i^k(2n-2k-4)$.
%$g'=2x[1, i, -1, \ldots, i^{2n-4}]+[2n-4, (2n-6)i, \ldots, -(2n-4)i^{2n-4}]$.
%$g'$ is non-degenerate since $(g'_1)^2\neq g'_0 g'_2$.
%\begin{itemize}
%\item For $2n=6$, $g'=2[1+x, xi, 1-x]$, %
%%
%
%\item For $2n\geq 8$,
%$g'_k$ satisfies the unique recurrence relation $<1, 2i, -1>$ up to a scalar since $g'$
%is non-degenerate and has arity $\geq 4$.
%Thus $g'\notin\mathscr{P}\cup\widetilde{\mathscr{A}}\cup\widehat{\mathscr{M}}$.
%Moreover, $g'\notin\widehat{\mathscr{M}}^\dagger$ since $g'_0\neq \pm g'_{2n-4}$.
%Thus {\rm Pl-}$\#${\rm CSP}$^2(f, g)$ is $\#{\rm P}$-hard by induction.
%\end{itemize}
%
%For  $\partial(f)=[1, i]^{\otimes 2n-2}$, the proof is similar.

\item $f'=[1, 0]^{\otimes 2n-2}$.
Then $f=x[1, i]^{\otimes 2n}+y[1, -i]^{\otimes 2n}+[1, 0]^{\otimes 2n}$ by
Proposition~\ref{prop:explicit_list}
(the Explicit List for $\int(f')$).
If $x=y=0$, then $f\in\mathscr{P}$.
In the following, assume that $(x, y)\neq (0, 0)$. Let
$a=x+y$, $b=(x-y)i$, then $(a, b)\neq (0, 0)$.

We have $\partial^{n-1}(f)=[1, 0]^{\otimes 2}$
and $f''=\partial^{n-2}_{[1, 0]^{\otimes 2}}(f)=
x[1, i]^{\otimes 4}+y[1, -i]^{\otimes 4}+[1, 0]^{\otimes 4}$, i.e.,
$f''=[1+a, b, -a, -b, a]$.
Note that $f''$ is redundant.
If $a^2+b^2\neq 0$, then the compressed signature matrix of $f''$ is nonsingular
and we are done by Lemma~\ref{4-redundant}.

Otherwise, we have
 $a=\pm i b$.
We claim that
$f''\notin\mathscr{P}\cup\widetilde{\mathscr{A}}\cup\widetilde{\mathscr{M}}$.
Note that  $ab\neq 0$ by $(a, b)\neq (0, 0)$ and $a=\pm i b$.
If $f''$ is degenerate, then by
$(f''_1)^2=f''_{0}f''_2$, we have
 $-a-a^2=b^2$. This implies that $a=0$. It is a contradiction.
 Moreover, note that $f'' = [1+a, \mp ia, -a, \pm ia, a]$
and has type $\langle 0, 1 ,\pm i\rangle$.
 Since $f''$ is non-degenerate and has arity $\geq 3$, the second
order recurrence relation $\langle 0, 1 ,\pm i\rangle$ is unique up to a scalar.
 Thus $f''\notin\mathscr{P}\cup\widetilde{\mathscr{A}}\cup\widetilde{\mathscr{M}}$ by Lemma~\ref{second-recurrence-relation}.
  So $\PlCSP^2(f'')$ is $\#$P-hard  by Theorem~\ref{arity-4-dichotomy} and we are done.

\item
$f'=[0, 1]^{\otimes 2n-2}$.
The proof follows from the previous case by a holographic transformation using
$\left[\begin{smallmatrix}
0 & 1\\
1 & 0
\end{smallmatrix}\right]$.

\item $f'=[1, 0]^{\otimes 2n-2}+t[0, 1]^{\otimes 2n-2}$ with $t\neq 0$.
Then $f=x[1, i]^{\otimes 2n}+y[1, -i]^{\otimes 2n}+[1, 0]^{\otimes 2n}+t[0, 1]^{\otimes 2n}$ by
Proposition~\ref{prop:explicit_list}
(the Explicit List for $\int(f')$).
If $x=y=0$, then $f\in\mathscr{P}$.
Otherwise, we have $(x, y)\neq (0, 0)$ and we are done by Lemma~\ref{deriviation-is-equality}.
%\end{itemize}
\end{enumerate}

\item For $f'\in\mathscr{A}^\dagger\setminus\mathscr{P}$,
we have $f'=[1, \alpha]^{\otimes 2n-2}+i^r[1, -\alpha]^{\otimes 2n-2}$ by definition (See Figure \ref{fig:venn_diagram:A_Adagger_P}).
Then $f=x[1, i]^{\otimes 2n}+y[1, -i]^{\otimes 2n}+\frac{1}{1+\alpha^2}\{[1, \alpha]^{\otimes 2n}+i^r[1, -\alpha]^{\otimes 2n}\}$
by Proposition~\ref{prop:explicit_list}
(the Explicit List for $\int(f')$).
If $x=y=0$, then $f\in\mathscr{A}^\dagger$.
In the following, assume that $(x, y)\neq (0, 0)$.

Note that $f'$ has type $\langle 1, 0, \pm i\rangle$ up to a scalar.
And this second order recurrence relation is unique up to a scalar.
Thus $f'\in\mathscr{A}^\dagger\setminus(\mathscr{P}\cup\mathscr{A}\cup\widetilde{\mathscr{M}})$
by Lemma \ref{second-recurrence-relation}.
In the following, we complete the proof by
constructing a signature of even arity
 in $(\mathscr{P}\cup\mathscr{A}\cup\widetilde{\mathscr{M}})\setminus\mathscr{A}^\dagger$
and
apply Theorem~\ref{mixing-theorem}, or constructing
 an arity 4 signature that is not in $\mathscr{P}\cup\widetilde{\mathscr{A}}\cup\widetilde{\mathscr{M}}$
 and apply Theorem~\ref{arity-4-dichotomy}.

Firstly, we have $f''=\partial^{n-3}(f')=(1+\alpha^2)^{n-3}\{[1, \alpha]^{\otimes 4}+i^r[1, -\alpha]^{\otimes 4}\}$.
We will discard the nonzero factor that are powers of $1+\alpha^2$.
If $r\neq 2$, we have $\partial(f'')=(1+i^r)[1, \frac{1-i^r}{1+i^r}\alpha, \alpha^2]$ and
we have $\partial_{[1, \frac{1-i^r}{1+i^r}\alpha, \alpha^2]}(=_4)=[1, 0, \alpha^2]$ on the left.
For $r=2$,  $\partial(f'')$ is a nonzero multiple of
$[0, 1, 0]$ and we have $\partial_{[0, 1, 0]}(f'')=2\alpha[1, 0, \alpha^2]$ on the right.
%%% lots of factors (1+\alpha^2)^{n-3}, (1+\alpha^2), 2 \alpha
Either way, we can take the derivative (for $[1, 0, \alpha^2]$
in RHS we connect it via $(=_2)$ of LHS to $f$)
\[f'''=\partial^{n-2}_{[1, 0, \alpha^2]}(f)=(1-\alpha^2)^{n-2}\{x[1, i]^{\otimes 4}+y[1, -i]^{\otimes 4}\}.\]
Note that $\partial_{[1, 0, \alpha^2]}([1, \pm \alpha]^{2n})$ is the identically zero signature, since $\alpha^4 =-1$.

If $xy=0$, or [$xy\neq 0$ and $x^4=y^4$], then $f'''\in\mathscr{A}\setminus\mathscr{A}^\dagger$.
So $\PlCSP^2(f', f''')$ is $\#$P-hard by Theorem~\ref{mixing-theorem}.
Thus $\PlCSP^2(f)$ is $\#$P-hard.

Otherwise, $xy\neq 0$ and $x^4\neq y^4$, so $f'''\notin\mathscr{P}\cup\widetilde{\mathscr{A}}\cup\widetilde{\mathscr{M}}$
(by the same reason as before: first by its type  $\langle 1, 0, 1\rangle$
it could only possibly be in $\mathscr{A}$ among the five classes
by Lemma~\ref{second-recurrence-relation}; but
$x^4\neq y^4$ rules that out too).
Thus $\PlCSP^2(f''')$ is $\#$P-hard by Theorem~\ref{arity-4-dichotomy}.
So $\PlCSP^2(f)$ is $\#$P-hard.

%%%%%%%%%%%%%%%here

\item For $f'\in\mathscr{A}\setminus\mathscr{P}$,
we have
 $f'=[1, \rho]^{\otimes 2n-2}+i^r[1, -\rho]^{\otimes 2n-2}$ by definition (See Figure~\ref{fig:venn_diagram:A_Adagger_P}).
\begin{itemize}
\item If $f'=[1, 1]^{\otimes 2n-2}+i^r[1, -1]^{\otimes 2n-2}$, then
$f=x[1, i]^{\otimes 2n}+y[1, -i]^{\otimes 2n}+\frac{1}2\{[1, 1]^{\otimes 2n}+i^r[1, -1]^{\otimes 2n}\}$
by Proposition~\ref{prop:explicit_list}
(the Explicit List for $\int(f')$).
If $x=y=0$, then $f\in\mathscr{A}$.
In the following, assume that $(x, y)\neq (0, 0)$.

By a holographic transformation using $H=\left[\begin{smallmatrix} 1 & 1 \\
1 & -1 \end{smallmatrix}\right]$, we have
\[
 \PlHolant([1, 0, 1], [1, 0, 1, 0, 1], \ldots \mid \hat{f'}, \hat{f})
 \equiv
 \PlCSP^2(f', f),
\]
where $\hat{f'}=(H^{-1})^{\otimes 2n-2}f'=[1, 0]^{2n-2}+i^r[0, 1]^{2n-2}
= [1, 0, \ldots, 0, i^r]$,
$\hat{f}=(H^{-1})^{\otimes 2n}f=x'[1, -i]^{\otimes 2n}+y'[1, i]^{\otimes 2n}+\frac{1}2\{[1, 0]^{\otimes 2n}+i^r[0, 1]^{\otimes 2n}\}$,
where $x'=\frac{(1+i)^{2n}}{2^{2n}}x$, $y'=\frac{(1-i)^{2n}}{2^{2n}}y$.
Note that $(x', y')\neq (0, 0)$.

Since we have $[1, 0, 1]$ on the left and $[1, 0, \ldots, 0, i^r]$ of arity $2n-2 \ge 4$
on the right
in
\[
 \PlHolant([1, 0, 1], [1, 0, 1, 0, 1], \ldots \mid \hat{f'}, \hat{f}),
\]
we can construct $=_{2k}$ on the right for $k\geq 1$
in the following way:
Firstly, connect four copies of $[1, 0, \ldots, 0, i^r]$ by three copies of
 $[1, 0, 1]$ in a planar fashion, to form an equality
$[1, 0, \ldots, 0, 1]$ of arity $4(2n-2) - 6 =8n -14$.
Then use $4n -9$ copies of $[1, 0, 1]$  to form loops on $(=_{8n -14})$,
and we get $(=_{4})$. From this, and $(=_2) = [1, 0, 1]$ on the left,
we can get all $(=_{2k})$ on the right for $k\geq 1$.
Then by $=_2$ on the left, we can construct all of $=_{2k}$
on the left.
Thus
\[
 \PlCSP^2(\hat{f'}, \hat{f})
 \leq
 \PlHolant([1, 0, 1], [1, 0, 1, 0, 1], \ldots \mid \hat{f'}, \hat{f}).
\]
By Lemma~\ref{deriviation-is-equality}
$\PlCSP^2(\hat{f})$ is \numP-hard.
Thus $\PlCSP^2(f)$ is \numP-hard.

\item If $f'=[1, i]^{\otimes 2n-2}+i^r[1, -i]^{\otimes 2n-2}$, then
$f=x[1, i]^{\otimes 2n}+y[1, -i]^{\otimes 2n}+\tilde{f}$,
%up to the scalar $\frac{1}{4}$,
%%% it's complicated enough here, i don' want to introduce this.
 where $\tilde{f}$ has arity $2n$ and
$\tilde{f}_k=\frac{1}{4} \left\{ i^k(2n-2k)+i^r(-i)^k(2n-2k) \right\}$
by Proposition~\ref{prop:explicit_list}
(the Explicit List for $\int(f')$).
Under the holographic transformation by  $Z=\left[\begin{smallmatrix} 1 & 1 \\
i & -i \end{smallmatrix}\right]$,
the expressions are more revealing:
$f = Z^{\otimes 2n} [x, 1, 0, \ldots, 0, i^r, y]$,
and $f' = \partial(f) = Z^{\otimes (2n-2)}[1, 0, \ldots, 0, i^r]$.
However, if we apply the holographic transformation $Z$ to $\PlCSP^2(f,f')$,
we have
\[
 \PlHolant([0, 1, 0], [1, 0, 1, 0, 1], \ldots \mid \widehat{f}, \widehat{f'})
 \equiv
 \PlHolant(\mathcal{EQ}_2 \mid f, f'),
\]
where $\widehat{f}=(Z^{-1})^{\otimes 2n}f=[x, 1, 0, \ldots, 0, i^r, y]$,
% up to a scalar 2,
%%% I don't see a scalar 2
and
$\widehat{f'}=(Z^{-1})^{\otimes 4}f=[1, 0, \ldots, 0, i^r]$.
Note that now we do not have $=_2$ on the left
in $\PlHolant([0, 1, 0], [1, 0, 1, 0, 1], \ldots |\widehat{f}, \widehat{f'})$.
This is inconvenient to construct gadget.
So, in the following steps we first try to
 construct $[1, 0, -1]^{\otimes 2}$ on the LHS
of $\PlCSP^2(f)$ to get
$\PlHolant([1, 0, -1]^{\otimes 2} \cup \mathcal{EQ}_2 \mid f)$.
  This will be done with the help of Lemma~\ref{construct-[1,0,1]-by-(1,i)-(1,-i)}.
Then after
the holographic transformation by  $Z$,
we have
$[1, 0, -1]^{\otimes 2} Z^{\otimes 2} = 4 [1, 0, 1]^{\otimes 2}$ on the left.

To apply Lemma~\ref{construct-[1,0,1]-by-(1,i)-(1,-i)}, we
construct $[1, i]^{\otimes 4}+i^s[1, -i]^{\otimes 4}$ in $\PlCSP^2(f, f')$
for some $0 \le s \le 3$ as follows.
\begin{itemize}
\item If $2n\equiv 2\pmod 4$, then we have
$\partial^{\frac{n-3}2}_{=_4}(f')=2^{\frac{n-3}2}\{[1, i]^{\otimes 4}+i^r[1, -i]^{\otimes 4}\}$.
\item If $2n\equiv 0\pmod 4$, then we have
$\partial^{\frac{n-2}2}_{=_4}(f')=2^{\frac{n-2}2}\{[1, i]^{\otimes 2}+i^r[1, -i]^{\otimes 2}\}=
2^{\frac{n-2}2}[1+i^r, (1-i^r)i, -(1+i^r)]$.
This is a nonzero multiple of $[1, \pm 1,-1]$ if $r\neq 2$,
and a nonzero multiple of $[0,1,0]$ if $r=2$.

If $r\neq 2$,
%then we have $\partial_{[1+i^r, (1-i^r)i, -(1+i^r)]}(=_4)=(1+i^r)[1, 0, -1]$ on the left
then we have $\partial_{[1,\pm 1,-1]}(=_4)=[1, 0, -1]$ on the left
and
\[\partial^{n-2}_{[1, 0, -1]}(f')=2^{n-2}\{[1, i]^{\otimes 4}+i^r[1, -i]^{\otimes 4}\}.\]

If $r=2$, we have  $\partial^{n-2}_{[0, 1, 0]}(f')=(2i)^{n-2}\{[1, i]^{\otimes 4}+i^r(-1)^{n-2}[1, -i]^{\otimes 4}\}$.
\end{itemize}
Thus we have $f''=[1, i]^{\otimes 4}+ i^s[1, -i]^{\otimes 4}$, for some $0 \le s \le 3$,
in $\PlCSP^2(f, f')$.
Then by Lemma~\ref{construct-[1,0,1]-by-(1,i)-(1,-i)},
we have $[1, 0, -1]^{\otimes 2}$ on the left,
i.e., we have
\[
 \PlHolant(\mathcal{EQ}_2, [1, 0, -1]^{\otimes 2} \mid f, f'')
 \equiv
 \PlCSP^2(f).
\]

By a holographic transformation using $Z=\left[\begin{smallmatrix} 1 & 1 \\
i & -i \end{smallmatrix}\right]$,
we have
\[\PlHolant([1, 0, 1]^{\otimes 2}, [0, 1, 0], [1, 0, 1, 0, 1], \ldots |\widehat{f}, \widehat{f''})
\equiv \PlHolant(\mathcal{EQ}_2, [1, 0, -1]^{\otimes 2}|f, f''),\]
where $\widehat{f}=(Z^{-1})^{\otimes 2n}f=[x, 1, 0, \ldots, 0, i^r, y]$, and
%up to a scalar 2,
%% JYC I don't see a factor 2
$\widehat{f''}=(Z^{-1})^{\otimes 4}f=[1, 0, 0, 0, i^s]$.

In $\PlHolant([1, 0, 1]^{\otimes 2}, [0, 1, 0], [1, 0, 1, 0, 1], \ldots |\widehat{f}, \widehat{f''})$,
by $[1, 0, 1]^{\otimes 2}$ on the left and
%$[1, 0, 0, 0, i^s]$, we have $=_4$ on the right
$\widehat{f''}$ on the right, we get $=_4$ on the right
as follows:
Use 4 copies of $\widehat{f''}$, connected together by 3
copies of $[1, 0, 1]^{\otimes 2}$ in a planar way. Each copy of
$[1, 0, 1]^{\otimes 2}$ connects two edges of one copy of $\widehat{f''}$
to another copy of $\widehat{f''}$ in such a way that the effect is
equivalent to connecting them by two copies of $(=_2) = [1, 0, 1]$.
 This way
we get an arity $16 - 12 =4$ signature $(=_4) = [1, 0, 0, 0, (i^s)^4]$.
Moreover, we have $=_{4k}$ for $k\geq 1$ on the right by $[1, 0, 1]^{\otimes 2}$ on the left
 and $=_4$ on the right in a similar way.
Then we can move $\widehat{f}$ to LHS by $[1, 0, 1]^{\otimes 2}$
because $\widehat{f}$ has even arity.
Thus we have
\[\PlHolant([1, 0, 1]^{\otimes 2}, [1, 0, 1, 0, 1],
 \widehat{f}|\mathcal{EQ}_4)\leq\PlHolant([1, 0, 1]^{\otimes 2}, [1, 0, 1, 0, 1], \ldots |\widehat{f}, \widehat{f''}).\]
Note that
\[
 \PlCSP^4(\widehat{f}, [1, 0, 1, 0, 1], [1, 0, 1]^{\otimes 2})
 \equiv
 \PlHolant([1, 0, 1]^{\otimes 2}, [1, 0, 1, 0, 1], \widehat{f} \mid \mathcal{EQ}_4).
\]
We will prove that $\PlCSP^4(\widehat{f}, [1, 0, 1, 0, 1], [1, 0, 1]^{\otimes 2})$ is $\#$P-hard
to complete the proof of this case.

Note that $\left[\begin{smallmatrix} \widehat{f}_0 & \widehat{f}_1 & \widehat{f}_2 \\
\widehat{f}_1 & \widehat{f}_2 & \widehat{f}_3 \\
\widehat{f}_{2n-3} & \widehat{f}_{2n-2} & \widehat{f}_{2n-1}
 \end{smallmatrix}\right]=\left[\begin{smallmatrix} x & 1 & 0 \\
1 & 0 & 0 \\
0 & 0 & i^r
 \end{smallmatrix}\right]$ has rank~$3$.
 Thus $\widehat{f}$ does not satisfy any second order recurrence relation
 by Lemma \ref{general-second-recurrence}.
So $\widehat{f}\notin\mathscr{P}\cup\widetilde{\mathscr{A}}\cup\widetilde{\mathscr{M}}$ by Lemma~\ref{second-recurrence-relation}.

If $(x, y)=(0, 0)$, we are done by Lemma~\ref{general-x=y=0-[1,i]-[1,-i]}.
In the following, assume that
$(x, y)\neq (0, 0)$.

\begin{itemize}
\item If $2n\equiv 0\pmod 4$, then
\begin{equation}\label{general-4-2-1}
\PlCSP^2(\widehat{f})\leq \PlCSP^4(\widehat{f}, [1, 0, 1]^{\otimes 2})
\end{equation}
by Lemma~\ref{mixing-P-global-binary}.

For $\PlCSP^2(\widehat{f})$,
 we have
$\widehat{f'''}=\partial^{n-2}(\widehat{f})=[x, 1, 0, i^r, y]$.
Note that $\widehat{f'''}$ is redundant.
If $(-1)^rx+y\neq 0$, then the
 compressed signature matrix  of $\widehat{f'''}$ is
nonsingular and we are done by Lemma~\ref{4-redundant}.

Otherwise, we have $x=\pm y$, and thus both $x, y  \not =0$.
% then we have $r=1, 3$ and $\widehat{f'''}=[x, 1, 0, \pm i, x]$.
It is easy to see that $\widehat{f'''}$ does not satisfy the second
order recurrence relations
$\langle 0, 1, 0\rangle$, $\langle 1, 0, \pm 1\rangle$, $\langle 1, 0, \pm i\rangle$.
Thus $\widehat{f'''}\notin\mathscr{P}\cup\widetilde{\mathscr{A}}$ by Lemma~\ref{second-recurrence-relation}.
%Similarly, there are at least three 0's in the expression
%$\widehat{f} = [x, 1, 0, \ldots, 0, i^r, x]$ of arity $2n \ge 6$,
%and therefore it is not of the types $\langle 0, 1, 0\rangle$,
%  $\langle 1, c, \pm 1 \rangle$, and thus $\widehat{f}
%\notin \widetilde{\mathscr{M}}$, by
%Lemma~\ref{second-recurrence-relation}.
%%% JYC
% i don't need to say this, because already said by det \neq 0
%%% \widehat{f} is not in any of 5 classes.

We consider three possibilities for $\widehat{f'''}$.

%$\diamondsuit$
%$\spadesuit$
{\small $\bullet$}
If $\widehat{f'''}\in\widehat{\mathscr{M}}\setminus(\mathscr{P}\cup\widetilde{\mathscr{A}})$,
then $\PlCSP^2(\widehat{f}, \widehat{f'''})$ is \#P-hard by Lemma~\ref{M-even},
where we have $\widehat{f} \notin \widehat{\mathscr{M}}$
because we have noted earlier that
$\widehat{f}\notin\mathscr{P}\cup\widetilde{\mathscr{A}}\cup\widetilde{\mathscr{M}}$.
Thus $\PlCSP^4(\widehat{f}, [1, 0, 1]^{\otimes 2})$ is \#P-hard
by (\ref{general-4-2-1}) and we are done.

%$\diamondsuit$
{\small $\bullet$}
If $\widehat{f'''}\in\widehat{\mathscr{M}}^\dagger\setminus(\mathscr{P}\cup\widetilde{\mathscr{A}})$,
then $\widehat{f'''}=[x, 1, 0, 1, -x]$ by Corollary~\ref{mixing-M-arity-4}
(the other form $[u, v, w, v, u]$ with
$(u+w)w=2v^2$  in Corollary~\ref{mixing-M-arity-4}
is impossible because
$w=0$ here and $(u+w)w=2v^2$ would force $v=0$.)
Then  we are done by Lemma~\ref{2-M-even-X},
because $\widehat{f'''}$ plays the role of $g$ in
Lemma~\ref{2-M-even-X}, and $\widehat{f} \notin
\widehat{\mathscr{M}}^\dagger$ by
$\widehat{f}\notin\mathscr{P}\cup\widetilde{\mathscr{A}}\cup\widetilde{\mathscr{M}}$.

%$\diamondsuit$
{\small $\bullet$}
If $\widehat{f'''}\notin\mathscr{P}\cup\widetilde{\mathscr{A}}\cup\widetilde{\mathscr{M}}$,
then $\PlCSP^2(\widehat{f'''})$ is $\#$P-hard by Theorem~\ref{arity-4-dichotomy} and we are done.

\item For $2n\equiv 2\pmod 4$, we cannot use the reduction in (\ref{general-4-2-1})
since Lemma~\ref{mixing-P-global-binary}
requires that all signatures on the right have arity $\equiv 0\pmod 4$.
We get around this difficulty by
constructing some arity 4 signatures in $\PlCSP^4(\widehat{f})$, and then
use Lemma~\ref{mixing-P-global-binary} for these arity 4 signatures.

%Note that we can get $[1, 0, 1]^{\otimes 2}$ on the left by $=_4$ on the left and $[1, 0, 1]^{\otimes 2}$ on the right
%in $\PlCSP^4(\widehat{f}, [1, 0, 1, 0, 1], [1, 0, 1]^{\otimes 2})$.
%So we can move any signature with even arity $2k$ on the right to LHS.

Firstly,
%we have $\widehat{g}=\partial^{\frac{n-1}2}_{=_4}(\hat{f})=[x, 1+i^r, y]$.
we have $\widehat{g}=\partial^{\frac{n-3}2}_{=_4}(\hat{f})=
[x, 1, 0, 0, 0, i^r, y]$.
We also have $\partial_{=_4}( \widehat{g}) = [x, 1+i^r, y]$.
They are both on the right.
Then we have $\partial_{[x, 1+i^r, y]}(=_4)=[x, 0, y]$ on the left.
We also connect $[x, 0, y]$ and $[x, 1+i^r, y]$ and then $[x, 0, y]$
in a chain, to get another binary signature $h=[x^3, (1+i^r)xy, y^3]$ on the left.
This can be verified by
\[\left[\begin{matrix} x &  0 \\
 0 & y
 \end{matrix}\right]
\left[\begin{matrix} x &  1+i^r \\
 1+i^r & y
 \end{matrix}\right]
\left[\begin{matrix} x &  0 \\
 0 & y
 \end{matrix}\right]
=
\left[\begin{matrix} x^3 &  (1+i^r)xy \\
 (1+i^r)xy & y^3
 \end{matrix}\right].\]

From these we produce two arity 4 signatures on the right:
\begin{alignat*}{3}
 & \widehat{g'}  &&= \partial_{[x, 0, y]}(\widehat{g}) &&= [x^2, x, 0, i^ry, y^2]\\
 & \widehat{g''} &&= \partial_{h}(\widehat{g}) &&= [x^4+2(1+i^r)xy, x^3, 0, i^ry^3, y^4+2i^r(1+i^r)xy].
\end{alignat*}
 Thus
\[\PlCSP^4(\widehat{g'}, \widehat{g''}, [1, 0, 1]^{\otimes 2}, [1, 0, 1, 0, 1])\leq
 \PlCSP^4(\widehat{f}, [1, 0, 1]^{\otimes 2}, [1, 0, 1, 0, 1]).\]

Moreover,
 note that all signatures in $\{\widehat{g'}, \widehat{g''}, [1, 0, 1, 0, 1])\}$
 have arity~$4$.
 Then by Lemma~\ref{mixing-P-global-binary}, we have
\[\PlCSP^2(\widehat{g'}, \widehat{g''}, [1, 0, 1, 0, 1])\leq
\PlCSP^4(\widehat{g''}, \widehat{g''}, [1, 0, 1]^{\otimes 2}, [1, 0, 1, 0, 1]).\]
It is easy to see that $\widehat{g'}$ is non-degenerate and
does not satisfy the second order recurrence relations
$\langle 0, 1, 0\rangle$, $\langle 1, 0, \pm 1\rangle$, $\langle 1, 0, \pm i\rangle$,
because $(x,y) \neq (0,0)$.
Thus $\widehat{g'}\notin\mathscr{P}\cup\widetilde{\mathscr{A}}$ by Lemma~\ref{second-recurrence-relation}.
If $\widehat{g'}\notin\widetilde{\mathscr{M}}$,
then $\PlCSP^2(\widehat{g'})$ is $\#$P-hard by Theorem~\ref{arity-4-dichotomy} and we are done.

Otherwise, $\widehat{g'}\in\widehat{\mathscr{M}}\setminus(\mathscr{P}\cup\widetilde{\mathscr{A}})$
or $\widehat{g'}\in\widehat{\mathscr{M}}^\dagger\setminus(\mathscr{P}\cup\widetilde{\mathscr{A}})$.

Note that $[1, 0, 1, 0, 1]$
has type $\langle 1, 0, -1\rangle$
and the second order recurrence relation
is unique up to a scalar.
Thus $[1, 0, 1, 0, 1]\notin\widehat{\mathscr{M}}$
 by Lemma~\ref{second-recurrence-relation}.
If $\widehat{g'}\in\widehat{\mathscr{M}}\setminus(\mathscr{P}\cup\widetilde{\mathscr{A}})$,
then $\PlCSP^2(\widehat{g'}, \widehat{g''}, [1, 0, 1, 0, 1])$ is \#P-hard
 by Lemma~\ref{M-even} and we are done.

Therefore we may assume
$\widehat{g'}\in\widehat{\mathscr{M}}^\dagger\setminus(\mathscr{P}\cup\widetilde{\mathscr{A}}\cup\widehat{\mathscr{M}})$.

By Corollary~\ref{mixing-M-arity-4},
for
$\widehat{g'}\in\widehat{\mathscr{M}}^\dagger\setminus(\mathscr{P}\cup\widetilde{\mathscr{A}}\cup\widehat{\mathscr{M}})$,
it cannot be of the form $[u, v, w, -v, u]$ with $(u-w)w=2v^2$;
for if it were so, then by $w=0$ in this case, we would have $v=0$,
 and this would  imply that $x=i^ry=0$ in $\widehat{g'}$.
It contradicts that $(x, y)\neq (0, 0)$.
So $\widehat{g'}$ must be of the form $[u, v, 0, v, -u]$, i.e.,
$x^2=-y^2$, $x=i^ry$.
Thus we have $x= \epsilon iy$ and $i^r= \epsilon i$,
for some $\epsilon =\pm 1$, and $x \not =0$.
Hence both $x, y \neq 0$ and $1+i^r \neq 0$.
It follows that
$x^3  = - \epsilon i y^3 \not = \epsilon i y^3 = i^r y^3$.
%$x^4+2(1+i^r)xy \not =  y^4+2i^r(1+i^r)xy$.
%%% below this is not used ... i changed argument.
%%% x^4 = y^4.
%%% 1+i^r \neq 0.
%%% xy \neq 0
%$ 2 \neq 2i^r = \pm 2i

Moreover,
% by the similar statement with $\widehat{g'}$,
 if
$\widehat{g''}\in\widehat{\mathscr{M}}^\dagger$, it cannot take the form
$[u, v, w, -v, u]$ with $(u-w)w=2v^2$ in Corollary~\ref{mixing-M-arity-4}
because if so then $w=0$ would force $v= 0$ and that would force
both $x =y=0$.
Then $\widehat{g''}$ must be of the form $[u, v, 0, v, -u]$.
But this would force $x^3 = i^r y^3$, a  contradiction.
Thus
 $\widehat{g''}\notin\widehat{\mathscr{M}}^\dagger$.

 If $\widehat{g''}\notin\mathscr{P}\cup\widetilde{\mathscr{A}}\cup\widetilde{\mathscr{M}}$, then
% $\PlCSP^2(\widehat{g'}, \widehat{g''}, [1, 0, 1, 0, 1])$ is \#P-hard by Theorem~\ref{arity-4-dichotomy} and we are done.
$\PlCSP^2( \widehat{g''})$ is \#P-hard by Theorem~\ref{arity-4-dichotomy} and we are done.
Otherwise,  $\widehat{g''}\in(\mathscr{P}\cup\widetilde{\mathscr{A}}\cup\widehat{\mathscr{M}})\setminus\widehat{\mathscr{M}}^\dagger$,
 $\PlCSP^2(\widehat{g'}, \widehat{g''}, [1, 0, 1, 0, 1])$ is \#P-hard by Lemma~\ref{2-M-even-X}
and we are done.
\end{itemize}

\end{itemize}

%\begin{lemma}\label{general-M-4}
%Let $f=x(1, i)^{\otimes 4}+y(1, -i)^{\otimes 4}+(s, t)^{\otimes 4}+(s, t)^{\otimes 4}$, $st\neq 0$, $s^4\neq t^4$,
%$x\neq 0$, or $y\neq 0$.
%Then $f\notin \widehat{\mathscr{M}}$.
%\end{lemma}
\item For $f'\in\widehat{\mathscr{M}}\setminus(\mathscr{P}\cup\widetilde{\mathscr{A}})$,
%with $f'\in\widehat{\mathscr{M}}\setminus(\mathscr{P}\cup\widetilde{\mathscr{A}})$ in hand,
we are done by Lemma~\ref{M-even}.

 \item\label{last-item-5} For
$f'\in\widehat{\mathscr{M}}^\dagger\setminus(\mathscr{P}\cup\widetilde{\mathscr{A}})$, or equivalently,
 $f'\in\widehat{\mathscr{M}}^\dagger\setminus(\mathscr{P}\cup\widetilde{\mathscr{A}}\cup\widehat{\mathscr{M}})$,
 $f'=[s, ti]^{\otimes 2n-2}\pm[t, si]^{\otimes 2n-2}$, $st\neq 0,$ $s^4\neq t^4$,
or $f'$ has arity $2n-2$ and $f'_k=(\pm i)^k(2n-2-2k)$ by Lemma~\ref{M-2-M-NOT-IN-A-AND-P}.
%Note that
%$\partial(f)\in\widehat{\mathscr{M}}^\dagger\setminus(\mathscr{P}\cup\widetilde{\mathscr{A}}\cup\widehat{\mathscr{M}})$.
Note that we are done if we have a nonzero binary signature that is not $\lambda[1, 0, 1]$ by Lemma~\ref{2-M-even}.
Moreover, if we have an arity 4 signature $h$ that is not in $\widehat{\mathscr{M}}^\dagger$ then we are done by the following
argument: if $h\in(\mathscr{P}\cup\widetilde{\mathscr{A}}\cup\widehat{\mathscr{M}})\setminus\widehat{\mathscr{M}}^{\dagger}$,
then $\PlCSP^2(h, f')$ is $\#$P-hard by Theorem~\ref{mixing-theorem} since
$f'\in\widehat{\mathscr{M}}^\dagger\setminus(\mathscr{P}\cup\widetilde{\mathscr{A}}\cup\widehat{\mathscr{M}})$;
if $h\notin\mathscr{P}\cup\widetilde{\mathscr{A}}\cup\widetilde{\mathscr{M}}$,
then $\PlCSP^2(h)$ is $\#$P-hard
by Theorem~\ref{arity-4-dichotomy}.
%or
%we have a signature is not in $\widehat{\mathscr{M}}^\dagger$ and its arity is less that $2n-2$.

\begin{itemize}
\item For $f'=[s, ti]^{\otimes 2n-2}+[t, si]^{\otimes 2n-2}$ with $2n\equiv 0\pmod 4$ or
$f'=[s, ti]^{\otimes 2n-2}-[t, si]^{\otimes 2n-2}$ with $2n\equiv 2\pmod 4$,
we have $\partial^{n-1}(f)=(s^2+t^2)(s^2-t^2)^{n-1}[1, \frac{2sti}{s^2+t^2}, -1]\neq \lambda[1, 0, 1]$.

%\item For the case that $f'$ has arity $2n-2$ and $f'_k=(\pm i)^k(2n-2-2k)$ with $2n\equiv 2\pmod 4$,
%we have $f''=\partial^{\frac{n-3}2}_{=_4}(f')$
%which has arity 4 and $f''_k=2^{\frac{n-3}2}(\pm i)^k(4-2k)$.
%Moreover, we have $\partial(f'')=2^{\frac{n+1}2}[1, \pm i, -1]\neq \lambda[1, 0, 1]$.
%We remark that it is necessary to use $=_4$ many times,
%% $\frac{n-3}2$ times loops by $=_4$, then follows a loop by $=_2$,
%since $f$ with two loops by $=_2$ is already identically zero.
%%%%
%%%%%%%%%%  i moved below. easier to follow both PM together.
%
%*********************************
%
%For $\partial(f)=[s, ti]^{\otimes 2n-2}+[t, si]^{\otimes 2n-2}$ with $2n\equiv 2\pmod 4$,
%$\partial(f)=[s, ti]^{\otimes 2n-2}-[t, si]^{\otimes 2n-2}$ with $2n\equiv 0\pmod 4$,
%note that $\partial(f)\in\widehat{\mathscr{M}}^\dagger\setminus(\mathscr{P}\cup\widetilde{\mathscr{A}}\cup\widehat{\mathscr{M}})$.
%Thus if we can construct a signature is not in $\widehat{\mathscr{M}}^\dagger$ and its arity is less that $2n-2$,
%then we are done by induction and Theorem~\ref{mixing-theorem}.
%
 %
\item For $f'=[s, ti]^{\otimes 2n-2}+[t, si]^{\otimes 2n-2}$ with $2n\equiv 2\pmod 4$,
$f=x[1, i]^{\otimes 2n}+y[1, -i]^{\otimes 2n}+\frac{1}{s^2-t^2}\{[s, ti]^{\otimes 2n}-[t, si]^{\otimes 2n}\}$
by Proposition~\ref{prop:explicit_list}
(the Explicit List for $\int(f')$).
If $x=y=0$, then $f\in\widehat{\mathscr{M}}^\dagger$.
Otherwise, we have
\begin{align*}
 f'''
 =
 \partial^{\frac{n-1}2}_{=_4}(f)
 &=
   2^{\frac{n-1}2} x [1,  i]^{\otimes 2}
 + 2^{\frac{n-1}2} y [1, -i]^{\otimes 2}
 + \frac{(s^4+t^4)^{\frac{n-1}2}}{s^2-t^2} \left\{[s, ti]^{\otimes 2}-[t, si]^{\otimes 2}\right\} \\
 &=
   2^{\frac{n-1}2} x [1,  i]^{\otimes 2}
 + 2^{\frac{n-1}2} y [1, -i]^{\otimes 2}
 + (s^4+t^4)^{\frac{n-1}2} [1, 0, 1]
\end{align*}
Let $a=2^{\frac{n-1}2}(x+y)$, $b=2^{\frac{n-1}2}(x-y)i$ and
 $c=(s^4+t^4)^{\frac{n-1}2}$,
then $f'''=[c+a, b, c-a]$. Note that $(a, b)\neq (0, 0)$.
If $b\neq 0$, it is obvious that $f'''\neq \lambda[1, 0, 1]$.
If $b=0$, then $a\neq 0$. Then  $f'''\neq \lambda[1, 0, 1]$ by $c+a\neq c-a$.

\item For the case that $f'$ has arity $2n-2$ and $f'_k=(\pm i)^k(2n-2-2k)$
with $2n
\equiv 2\pmod 4$,
we have $f''=\partial^{\frac{n-3}2}_{=_4}(f')$ which has arity 4 and $f''_k=2^{\frac{n-3}2}(\pm i)^k(4-2k)$.
Moreover, we have $\partial(f'')=2^{\frac{n+1}2}[1, \pm i, -1]\neq \lambda[1, 0, 1]$.
We remark that it is necessary to use $=_4$ that many times,
% $\frac{n-3}2$ times loops by $=_4$, then follows a loop by $=_2$,
since $f$ with two loops by $=_2$ is already identically zero.

 \item For the case that $f'$ has arity $2n-2$ and $f'_k=(\pm i)^k(2n-2-2k)$ with $2n\equiv 0\pmod 4$, we may
consider only  the case where the sign $\pm$ is $+$.
Indeed under $Z= \left[\begin{smallmatrix} 1 & 1 \\
i & -i \end{smallmatrix}\right]$, for the $+$ sign
$f' = Z^{\otimes (2n-2)} [0, 1, 0, \ldots, 0]$
and for the $-$ sign $f' = Z^{\otimes (2n-2)} [0, \ldots, 0, 1, 0]$,
a reversal under the $Z$-transformation.
If we take a holographic transformation by $T = \left[\begin{smallmatrix} 1 & 0\\
0 & -1 \end{smallmatrix}\right]$,
we have
%$\left[\begin{smallmatrix} 1 & 0\\0 & -1 \end{smallmatrix}\right] Z
$TZ
=  \left[\begin{smallmatrix} 1 & 1 \\-i & i \end{smallmatrix}\right]
= Z \left[\begin{smallmatrix} 0 & 1\\1 & 0 \end{smallmatrix}\right]$,
and so $(TZ)^{\otimes (2n-2)} [0, \ldots, 0, 1, 0]
= Z^{\otimes (2n-2)} [0, 1, 0, \ldots, 0]$. Meanwhile,
$\mathcal{EQ}_2$ is invariant under $T$.

Thus we consider $f'$ of arity $2n-2$ where $f'_k= i^k(2n-2-2k)$ with $2n\equiv 0\pmod 4$.
Let $\hat{f'}= (Z^{-1})^{\otimes (2n-2)}f' = [0, 1, 0, \ldots, 0]$
%then $(Z^{-1})^{\otimes (2n-2)}f'=\hat{f'}$.
%%%up to a nonzero scalar,
%%% I don't see any nonzero scalar factor
% 2,
 and let $\hat{f} = (Z^{-1})^{\otimes (2n-2)}f$.
 Then we have $(Z^{-1})^{\otimes (2n-2)}(\partial(f))=\partial_{[0, 1, 0]}(\hat{f})$
up to a scalar.
This implies $\partial_{[0, 1, 0]}(\hat{f})=[0, 1, 0, \ldots, 0]$.
Thus there exist constants $x$ and $y$ such that
 $\hat{f}=[x, 0, 1, 0, \ldots, 0, y]$.
 By the holographic transformation using $Z$, we have
\[
 \PlCSP^2(f)
 \equiv
 \PlHolant([0, 1, 0], [1, 0, 1, 0, 1], \ldots \mid \hat{f}).
\]

%Constructing a binary signature
%that is not $\lambda [1, 0, 1]$ in $\PlCSP(f)$
%is equivalent to constructing a binary signature
%that is not $\lambda[0, 1, 0]$ in $\PlHolant([0, 1, 0], [1, 0, 1, 0, 1], \dotsc \mid \hat{f})$.
%We claim that this is impossible.

We remark that,
in $\PlHolant([0, 1, 0], [1, 0, 1, 0, 1], \ldots \mid \hat{f})$,
all signatures have even arities.
And all signatures of arity $2m\equiv 2\pmod 4$
satisfy odd parity and all signatures of arity $2m\equiv 0\mod 4$
satisfy even parity.
Then by the statement of Remark~\ref{rmk:E2:1},
any binary signature constructed in $\PlHolant([0, 1, 0], [1, 0, 1, 0, 1], \ldots \mid \hat{f})$ can only
 be of the form $\lambda[0, 1, 0]$.
This implies that the binary signature constructed in $\PlCSP^2(f)$
can only be of the form $\lambda[1, 0, 1]$ before the $Z$-transformation.
This forces us to construct signatures of arity at least 4 to prove hardness.
%
%
%In a gadget, the parity of the resulting signature depends on the number of the odd parity signatures:
%if  the number of the odd parity signatures is odd, then the resulting signature satisfies odd parity;
%if  the number of the odd parity signatures is even, then the resulting signature satisfies even parity.
%On the other hand,
%if the signature of a gadget is binary, the number of the signatures of arity
%$2n\equiv 2\pmod 4$ that this gadget employs is odd.
%This implies that the binary signatures constructed in $\PlHolant
%([0, 1, 0], [1, 0, 1, 0, 1], \ldots|\hat{f})$ have to satisfy odd parity, i.e.,
%the binary signatures are of form $\lambda[0, 1, 0]$.
%This implies that we cannot get a binary signature that is not $\lambda[1, 0, 1]$
%in $\PlCSP(f)$ by gadget construction.

In $\PlHolant([0, 1, 0], [1, 0, 1, 0, 1], \ldots|\hat{f})$,
note that by $2n\equiv 0\mod 4$ we have $2n \ge 8$,
and $\hat{g}=\partial^{\frac{n-2}2}_{[1, 0, 1, 0, 1]}(\hat{f})=[x+\frac{n-2}2\cdot 6, 0, 1, 0, y]$. It has arity 4.
%Note that $\hat{g}$ has arity $2n-4$.
%If $2n\geq 10$ or $(x+4)y\neq 1$, then $\hat{g}\notin\mathscr{M}$. Thus we have $Z^{\otimes 2n-4}\hat{g}\notin\widehat{\mathscr{M}}^\dagger$
%in $\PlCSP^2(f)$.
% and we finish the proof.
%
If $(x+\frac{n-2}2\cdot 6)y\neq 1$, then $\hat{g}\notin\mathscr{M}$,
because symmetric matchgate signatures must form
 geometric series in alternate terms.
Thus we have $Z^{\otimes 4}(\hat{g})\notin\widehat{\mathscr{M}}^\dagger$
in $\PlCSP^2(f, f')$ and we are done.

 If $(x+\frac{n-2}2\cdot 6)y=1$,
 then $y\neq 0$. Firstly, we have
an arity 8 signature
 \[\widehat{g'}=\partial^{\frac{n-4}2}_{[1, 0, 1, 0, 1]}(\hat{f})
 =[x+\frac{n-4}2\cdot 6, 0, 1, 0, 0, 0, 0, 0,  y]\]
(note that $n \ge 4$ when $2n \equiv 0\mod 4$),
 and we have $\partial^2_{[0, 1, 0]}(\widehat{g'})=[1, 0]^{\otimes 4}$ on the right.
 So we have $[0, 1]^{\otimes 4}$ on the left.
 Moreover, we have $\partial_{[0, 1]^{\otimes 4}}(\widehat{g'})=y[0, 1]^{\otimes 4}$ on the right.
 So we  have $[1, 0]^{\otimes 4}$ on the left.
Then we have $\widehat{g''}=\partial^{\frac{n-2}2}_{[1, 0]^{\otimes 4}}(\hat{f})=[x, 0, 1, 0, 0]$ on the right.
Note that $\widehat{g''}\notin\mathscr{M}$. Thus we have $Z^{\otimes 4}(\widehat{g''})\notin\widehat{\mathscr{M}}^\dagger$
in $\PlCSP^2(f, f')$
 and we are done.

\item
For the last case of Case 5,
 $f'=[s, ti]^{\otimes 2n-2}-[t, si]^{\otimes 2n-2}$ with $2n\equiv 0\pmod 4$,
we let $u=\frac{s-t}{s+t}$, then $u^4\neq 0, 1$ by Lemma~\ref{pre-4-power}.
Let $Z= \left[\begin{smallmatrix} 1 & 1 \\
i & -i \end{smallmatrix}\right]$, then
\begin{align*}
 \widehat{f'}
 &= (Z^{-1})^{\otimes 2n-2}(f') \\
 &= \frac{1}{2^{2n-2}}\left\{[s+t, s-t]^{\otimes 2n-2}-[s+t, t-s]^{\otimes 2n-2}\right\} \\
 &= \frac{(s+t)^{2n-2}}{2^{2n-2}}\left\{[1, u]^{\otimes 2n-2}-[1, -u]^{\otimes 2n-2}\right\} \\
 &= \lambda[0, u^2, 0, u^4, \ldots, u^{2n-2}, 0],
\end{align*}
 %where $\lambda=\frac{(s+t)^{2n-2}}{2^{2n-2}u}\neq 0$.
where $\lambda=\frac{(s+t)^{2n-2}}{2^{2n-3}u}\neq 0$.
%%% there is this extra 2 on top from the difference in tensor power
Let $(Z^{-1})^{\otimes 2n}f=\widehat{f}$, then $(Z^{-1})^{\otimes 2n-2}(\partial(f))=\partial_{[0, 1, 0]}(\widehat{f})$ up to a scalar.
This implies that $\partial_{[0, 1, 0]}(\widehat{f})=\lambda[0, u^2, 0, u^4, \ldots, u^{2n-2}, 0]$.
Thus there exist constants $x$ and $y$ such that
 $\widehat{f}=(Z^{-1})^{\otimes 2n}f=\lambda[1+x, 0, u^2, 0, u^4, \ldots, u^{2n-2}, 0, u^{2n}+y]$,
where we append the terms $1$ and $u^{2n}$ for future convenience.
(This can be accommodated by naming  different $x$ and $y$.)
If $x=y=0$, then $\widehat{f}\in\mathscr{M}$ and $f\in\widehat{\mathscr{M}}^\dagger$.
In the following, assume that $(x, y)\neq (0, 0)$.
By the holographic transformation using $Z$, we have
\[
 \PlCSP^2(f)
 \equiv
 \PlHolant([0, 1, 0], [1, 0, 1, 0, 1], \ldots \mid \widehat{f}).
\]

By the same argument as the previous case, it is impossible to construct a ``good'' binary signature
in this case.
So we have to construct signatures of arity at least 4 to prove hardness.

%We often use the following computation in the remainder of this proof:
We will repeatedly use the following computation in the remainder of this proof:
Let $\bar{g}=\partial_{[1, 0, v, 0, v^2]}(g)$ for some $v$, then
arity($\bar{g}$) $=$ arity$(g)-4$ and $\bar{g}_k=g_{k}+6vg_{k+2}+v^2g_{k+4}$.

We will
complete the proof by
constructing some arity 4 signatures $\widehat{h}$
in the setting after the $Z$-transformation
$\PlHolant([0, 1, 0], [1, 0, 1, 0, 1], \ldots \mid \widehat{f})$ that
cannot all belong to $\mathscr{M}$.
We note that if $\widehat{h} \notin \mathscr{M}$
then $h = Z^{\otimes 4} \widehat{h} \notin \widehat{\mathscr{M}}^\dagger$.
This will imply $\PlCSP^2(h, f')$ is $\#$P-hard
as noted earlier, thus complete the proof of this Case~\ref{last-item-5}.

In $\PlHolant([0, 1, 0], [1, 0, 1, 0, 1], \ldots \mid \widehat{f})$,
we have
 $\partial^{n-2}_{[0, 1, 0]}(\widehat{f})$
which is a nonzero multiple of $[1, 0, u^2, 0, u^4]$.
%=u^{n-2}[1, 0, u^2, 0, u^4]$.
%%% there was also lambda
 Then we have $[u^4, 0, u^2, 0, 1] = u^4[1, 0, u^{-2}, 0, u^{-4}]$ on the left.
Ignoring $\lambda \neq 0$, we write
\[\widehat{f} = [1, 0, u^2, 0, u^4, 0, \ldots, 0, u^{2n}]
+ [x, 0, 0, 0, 0, \ldots, 0, y]\]
which has arity $2n \ge 8$,
%\frac{1}{2}
%\left\{ [1, u]^{\otimes 2n} + [1, -u]^{\otimes 2n} \right\}
%+ [x, 0, \ldots, 0, y]
and we have
\begin{align*}
 \widehat{f^{(4)}}
 &= \partial^{\frac{n-4}2}_{[1, 0, u^{-2}, 0, u^{-4}]}(\widehat{f}) \\
 &= 8^{\frac{n-4}2} [ 1, 0, u^2, 0, u^4, 0, u^6, 0, u^8] +
                    [x, 0, 0, 0, 0, 0, 0, 0, yu^{-2(n-4)}] \\
 &= [x+8^{\frac{n-4}2}, 0, 8^{\frac{n-4}2}u^2, 0, 8^{\frac{n-4}2}u^4, 0,  8^{\frac{n-4}2}u^6, 0, 8^{\frac{n-4}2}u^8+ yu^{-2(n-4)}].
\end{align*}
Let $x'=\frac{x}{8^{\frac{n-4}2}}$, $y'=\frac{yu^{-2(n-4)}}{8^{\frac{n-4}2}}$,
then $\widehat{f^{(4)}}
 =[x'+1, 0, u^2, 0, u^4, 0,  u^6, 0, u^8+ y']$ up to the scalar $8^{\frac{n-4}2}$.
 Further, we have $\widehat{f^{(5)}}=\partial_{[1, 0, u^{-2}, 0, u^{-4}]}(\widehat{f^{(4)}})=[x'+8, 0, 8u^2, 0, 8u^4+y'u^{-4}]$.
 If $x'=0$ or $y'=0$, then
 $\widehat{f^{(5)}}\notin\mathscr{M}$ by $(x'+8)(8u^4+y'u^{-4})\neq (8u^2)^2$ and we are done.
 So we can assume that $x'y'\neq 0$ in the following.

 In the following,
 if we have the signature $[1, 0, v, 0, v^2]$ with $v \neq 0$ on the left,
 then we have $\partial_{[1, 0, v, 0, v^2]}(\widehat{f^{(4)}}) = [x'+c, 0, cu^2, 0, y'v^2+cu^4]$,
 where $c=1+6u^2v+u^4v^2$.
If $c=0$, then we have $[x', 0, 0, 0, y'v^2]\notin\mathscr{M}$ and we are done.
So in the following, we always suppose that $c=1+6u^2v+u^4v^2\neq 0$.
Moreover, if $(x'+c)(y'v^2+cu^4)\neq(cu^2)^2$, then $[x'+c, 0, cu^2, 0, y'v^2+cu^4]\notin\mathscr{M}$ and we are done.
So we assume that $(x'+c)(y'v^2+cu^4)=(cu^2)^2$. This implies that $x'+c\neq 0$ and $x'y'v^2+(x'u^4+y'v^2)c=0$.
To summerize, in the following if we have $[1, 0, v, 0, v^2]$ with $v\neq 0$ on the left,
then we have
\begin{equation}\label{general-condition}
\begin{split}
&c=1+6u^2v+u^4v^2\neq 0,\\
&x'+c\neq 0,\\
&x'y'v^2+(x'u^4+y'v^2)c=0.
\end{split}
\end{equation}

Firstly, by $\widehat{f^{(5)}}=\partial_{[1, 0, u^{-2}, 0, u^{-4}]}(\widehat{f^{(4)}})=[x'+8, 0, 8u^2, 0, 8u^4+y'u^{-4}]$
%If $(x'+8)(8u^4+y'u^{-4})\neq(8u^2)^2,$
%then $\widehat{f^{(5)}}\notin\mathscr{M}$ since its nonzero entries are not a geometric array.
% So we have $Z^{\otimes 4}(\widehat{f^{(5)}})\notin\widehat{\mathscr{M}}^\dagger$
% in $\PlCSP^2(f)$ and we are done.
%In the following, assume that
%\begin{equation}\label{2-M-gen-1-equ}
and (\ref{general-condition}),
% we have $(x'+8)(8u^4+y'u^{-4})=(8u^2)^2.$
 %\end{equation}
 we have
 \begin{equation}\label{2-M-gen-1-equ}
 \begin{split}
 &x'+8\neq 0,\\
 &x'y'u^{-4}+8(x'u^4+y'u^{-4})=0.
 \end{split}
 \end{equation}

Note that we have $[1, 0, 1, 0, 1]$ on the left, so we have
$\widehat{f^{(6)}}=\partial_{[1, 0, 1, 0, 1]}(\widehat{f^{(4)}})=[x'+c_1, 0, c_1u^2, 0, y'+c_1u^4]$,
where $c_1=1+6u^2+u^4$.
Then by (\ref{general-condition}),
we have $c_1 \not = 0$ and
\begin{equation}\label{2-M-gen-2-equ}
\begin{split}
& x'+c_1\neq 0,\\
&x'y'+(x'u^4+y')c_1=0.
\end{split}
\end{equation}

By (\ref{2-M-gen-1-equ}), (\ref{2-M-gen-2-equ}), and $x'y'\neq 0$, we have
\begin{align*}
 1+(\frac{u^8}{y'}+\frac{1}{x'})8   &= 0,\\
 1+(\frac{u^4}{y'}+\frac{1}{x'})c_1 &= 0.
\end{align*}
Then we have
\begin{alignat*}{3}
 & \frac{1}{x'} &&= \frac{c_1-8u^4}{8c_1(u^4-1)} &&= -\frac{1+7u^2}{8(1+u^2)(1+6u^2+u^4)},\\
 & \frac{1}{y'} &&= \frac{8-c_1}{8c_1(u^8-u^4)}  &&= -\frac{7+u^2}{8u^4(1+u^2)(1+6u^2+u^4)}.
\end{alignat*}
Since $x'\neq 0$,
we have $1+7u^2\neq 0$.

For $\widehat{f^{(5)}}$, $\widehat{f^{(6)}}$, let $v_2=\frac{x'+8}{8u^2}$ and $v_3=\frac{x'+c_1}{c_1u^2}$,
then $v_2\neq 0$, $v_3\neq 0$ by $x'+8\neq 0$ and $x'+c_1\neq 0$,
  and $\widehat{f^{(5)}}=[1, 0, v^{-1}_2 ,0, v^{-2}_2]$, $\widehat{f^{(6)}}=[1, 0, v^{-1}_3 ,0, v^{-2}_3]$
up to the scalars $x'+8$, $x'+c_1$ respectively.
So we have $[1, 0, v_2 ,0, v^2_2]$, $[1, 0, v_3 ,0, v^2_3]$ on the left.
Moreover, let $c_2=1+6u^2v_2+u^4v^2_2$, $c_3=1+6u^2v_3+u^4v^2_3$,
then we have by (\ref{general-condition})
\begin{equation}\label{2-M-gen-3-equ}
\begin{split}
&x'y'v_2^2+(x'u^4+y'v_2^2)c_2=0,\\
&x'y'v_3^2+(x'u^4+y'v_3^2)c_3=0.
\end{split}
\end{equation}
In (\ref{2-M-gen-3-equ}), we have 
\begin{align*}
 c_1 &= 1+6u^2+u^4,\\
 \frac{1}{x'} &= \frac{c_1-8u^4}{8c_1(u^4-1)}=-\frac{7u^2+1}{8(u^2+1)(u^4+6u^2+1)},\\
 \frac{1}{y'} &= \frac{8-c_1}{8c_1(u^8-u^4)}=-\frac{u^2+7}{8u^4(u^2+1)(u^4+6u^2+1)},\\
 v_2 &= \frac{x'+8}{8u^2}=\frac{-7u^2-u^4}{7u^2+1},\\
 c_2 &= 1+6u^2v_2+u^4v_2^2=\frac{u^{12}+14u^{10}+7u^8-300u^6+7u^4+14u^2+1}{(7u^2+1)^2},\\
 v_3 &= \frac{x'+c_1}{c_1u^2}=-\frac{7+u^2}{u^2(1+7u^2)},\\
 c_3 &= 1+6u^2v_3+u^4v_3^2=\frac{8u^4-272u^2+8}{(1+7u^2)^2}.
\end{align*}
Note that all of them are functions of $u$.
Thus
(\ref{2-M-gen-3-equ}) gives two equations of $u$ as following:
\begin{equation}\label{2-M-gen-4-equ}
\begin{aligned}
\frac{8u^4c_1^2(1+u^2)^2 \cdot p_1(u)}{(1+7u^2)^4}  &= 0,\\
\frac{3072u^2(1+u^2)^2c_1 \cdot p_2(u)}{(1+7u^2)^4} &= 0,
\end{aligned}
\end{equation}
where $p_1(u)=u^{12}+14u^{10}-49u^8-700u^6-49u^4+14u^2+1$, $p_2(u)=7u^4+2u^2+7$.
Note that $q_1(u)p_1(u)+q_2(u)p_2(u)=244224$, where
$q_1(u)=-188-315u^2, $ $q_2(u)=34916-9555u^2-32872u^4-2058u^6+644u^8+45u^{10}$, thus $\gcd(p_1(u), p_2(u))=1$.
Then by
$u^4\neq 0, 1$, $c_1\neq 0$,
the two equations in (\ref{2-M-gen-3-equ}) have no common solution in $u$.
This is a contradiction and we finish the proof.
\qedhere
 \end{itemize}
\end{enumerate}
\end{proof}

We hereby finish the proof of Theorem~\ref{general-single-dichotomy},
and hence we complete the proof of the
main theorem of Part II---Theorem~\ref{dichotomy-pl-csp2}
is a straightforward combination of Theorem~\ref{odd-arity-dichotomy},
Theorem~\ref{general-single-dichotomy}
and Theorem~\ref{mixing-theorem}.

\clearpage

\bibliographystyle{plain}
\bibliography{bib}

\begin{thebibliography}{10}

\bibitem{Bax82}
Rodney~J. Baxter.
\newblock {\em Exactly solved models in statistical mechanics}.
\newblock Academic press London, 1982.

\bibitem{CCLL10}
Jin-Yi Cai, Xi~Chen, Richard~J. Lipton, and Pinyan Lu.
\newblock On tractable exponential sums.
\newblock In {\em FAW}, pages 148--159. Springer Berlin Heidelberg, 2010.

\bibitem{CC07b}
Jin-Yi Cai and Vinay Choudhary.
\newblock Some results on matchgates and holographic algorithms.
\newblock {\em Int. J. Software and Informatics}, 1(1):3--36, 2007.

\bibitem{CCL09}
Jin-Yi Cai, Vinay Choudhary, and Pinyan Lu.
\newblock On the theory of matchgate computations.
\newblock {\em Theory Comput. Syst.}, 45(1):108--132, 2009.

\bibitem{CG14}
Jin-Yi Cai and Aaron Gorenstein.
\newblock Matchgates revisited.
\newblock {\em Theory Comput.}, 10(7):167--197, 2014.

\bibitem{CGW13}
Jin-Yi Cai, Heng Guo, and Tyson Williams.
\newblock A complete dichotomy rises from the capture of vanishing signatures
  (extended abstract).
\newblock In {\em STOC}, pages 635--644. ACM, 2013.
\newblock \textit{CoRR}, \href{http://arxiv.org/abs/1307.7430}{abs/1204.6445}.

\bibitem{CGW14a}
Jin-Yi Cai, Heng Guo, and Tyson Williams.
\newblock Holographic algorithms beyond matchgates.
\newblock In {\em ICALP}, pages 271--282. Springer Berlin Heidelberg, 2014.
\newblock \textit{CoRR}, \href{http://arxiv.org/abs/1307.7430}{abs/1307.7430}.

\bibitem{CK12}
Jin-Yi Cai and Michael Kowalczyk.
\newblock Spin systems on $k$-regular graphs with complex edge functions.
\newblock {\em Theoretical Computer Science}, 2012.
\newblock DOI:10.1016/j.tcs.2012.01.021.

\bibitem{CKW12}
Jin-Yi Cai, Michael Kowalczyk, and Tyson Williams.
\newblock Gadgets and anti-gadgets leading to a complexity dichotomy.
\newblock In {\em ITCS}, pages 452--467. ACM, 2012.

\bibitem{CL10}
Jin-Yi Cai and Pinyan Lu.
\newblock On symmetric signatures in holographic algorithms.
\newblock {\em Theory Comput. Syst.}, 46(3):398--415, 2010.

\bibitem{CL11a}
Jin-Yi Cai and Pinyan Lu.
\newblock Holographic algorithms: {F}rom art to science.
\newblock {\em J. Comput. Syst. Sci.}, 77(1):41--61, 2011.

\bibitem{CLX10}
Jin-Yi Cai, Pinyan Lu, and Mingji Xia.
\newblock Holographic algorithms with matchgates capture precisely tractable
  planar \#{CSP}.
\newblock In {\em FOCS}, pages 427--436. IEEE Computer Society, 2010.

\bibitem{CLX11d}
Jin-Yi Cai, Pinyan Lu, and Mingji Xia.
\newblock Computational complexity of {H}olant problems.
\newblock {\em SIAM J. Comput.}, 40(4):1101--1132, 2011.

\bibitem{CLX11a}
Jin-Yi Cai, Pinyan Lu, and Mingji Xia.
\newblock Dichotomy for {H}olant* problems of {B}oolean domain.
\newblock In {\em SODA}, pages 1714--1728, 2011.

\bibitem{CLX14}
Jin-Yi Cai, Pinyan Lu, and Mingji Xia.
\newblock The complexity of complex weighted {B}oolean \#{CSP}.
\newblock {\em J. Comput. System Sci.}, 80(1):217--236, 2014.

\bibitem{DP91}
C.~T.~J. Dodson and T.~Poston.
\newblock {\em Tensor Geometry}.
\newblock Graduate Texts in Mathematics. Springer-Verlag, second edition, 1991.

\bibitem{DGLRS12}
Jan Draisma, Dion~C. Gijswijt, L{\'a}szl{\'o} Lov{\'a}sz, Guus Regts, and
  Alexander Schrijver.
\newblock Characterizing partition functions of the vertex model.
\newblock {\em J. Algebra}, 350:197--206, 2012.

\bibitem{FLS07}
Michael Freedman, L{\'a}szl{\'o} Lov{\'a}sz, and Alexander Schrijver.
\newblock Reflection positivity, rank connectivity, and homomorphism of graphs.
\newblock {\em J. Amer. Math. Soc.}, 20(1):37--51, 2007.

\bibitem{GLV13}
Heng Guo, Pinyan Lu, and Leslie~G. Valiant.
\newblock The complexity of symmetric {B}oolean parity {H}olant problems.
\newblock {\em SIAM J. Comput.}, 42(1):324--356, 2013.

\bibitem{GW13}
Heng Guo and Tyson Williams.
\newblock The complexity of planar {B}oolean \#{CSP} with complex weights.
\newblock In {\em ICALP}, pages 516--527. Springer Berlin Heidelberg, 2013.
\newblock \textit{CoRR}, \href{http://arxiv.org/abs/1307.7430}{abs/1212.2284}.

\bibitem{HL12}
Sangxia Huang and Pinyan Lu.
\newblock A dichotomy for real weighted {H}olant problems.
\newblock In {\em CCC}, pages 96--106. IEEE Computer Society, 2012.
\newblock Full version available at
  \url{http://www.csc.kth.se/~sangxia/papers/2012-ccc.pdf}.

\bibitem{Isi25}
Ernst Ising.
\newblock Beitrag z\"{u}r theorie des ferromagnetismus.
\newblock {\em Zeitschrift f\"{u}r Physik}, 31(1):253--258, 1925.

\bibitem{Kas61}
P.~W. Kasteleyn.
\newblock The statistics of dimers on a lattice.
\newblock {\em Physica}, 27:1209--1225, 1961.

\bibitem{Kas67}
P.~W. Kasteleyn.
\newblock Graph theory and crystal physics.
\newblock In F.~Harary, editor, {\em Graph Theory and Theoretical Physics},
  pages 43--110. Academic Press, London, 1967.

\bibitem{Kow10}
Michael Kowalczyk.
\newblock {\em Dichotomy theorems for {H}olant problems}.
\newblock PhD thesis, University of Wisconsin---Madison, 2010.
\newblock \url{http://cs.nmu.edu/~mkowalcz/research/main.pdf}.

\bibitem{LMN13}
J.~M. Landsberg, Jason Morton, and Serguei Norine.
\newblock Holographic algorithms without matchgates.
\newblock {\em Linear Algebra Appl.}, 438(2):782--795, 2013.

\bibitem{LY52}
T.~D. Lee and C.~N. Yang.
\newblock Statistical theory of equations of state and phase transitions. {II}.
  {L}attice gas and {I}sing model.
\newblock {\em Phys. Rev.}, 87(3):410--419, 1952.

\bibitem{Lie67}
Elliott~H. Lieb.
\newblock Residual entropy of square ice.
\newblock {\em Phys. Rev.}, 162(1):162--172, 1967.

\bibitem{LS81}
Elliott~H. Lieb and Alan~D. Sokal.
\newblock A general {L}ee-{Y}ang theorem for one-component and multicomponent
  ferromagnets.
\newblock {\em Comm. Math. Phys.}, 80(2):153--179, 1981.

\bibitem{Mor11}
Jason Morton.
\newblock Pfaffian circuits.
\newblock {\em CoRR}, abs/1101.0129, 2011.

\bibitem{MM13}
Jason Morton and Susan Margulies.
\newblock Polynomial-time solvable {\#}{CSP} problems via algebraic models and
  {P}faffian circuits.
\newblock {\em CoRR}, abs/1311.4066, 2013.
\newblock To appear in Journal of Symbolic Computation.

\bibitem{Ons44}
Lars Onsager.
\newblock Crystal statistics. {I}. {A} two-dimensional model with an
  order-disorder transition.
\newblock {\em Phys. Rev.}, 65(3-4):117--149, 1944.

\bibitem{Sch13}
Alexander Schrijver.
\newblock Characterizing partition functions of the spin model by rank growth.
\newblock {\em Indag. Math. (N.S.)}, 24(4):1018--1023, 2013.

\bibitem{TF61}
H.~N.~V. Temperley and M.~E. Fisher.
\newblock Dimer problem in statistical mechanics---an exact result.
\newblock {\em Philosophical Magazine}, 6:1061--1063, 1961.

\bibitem{Val02b}
Leslie~G. Valiant.
\newblock Expressiveness of matchgates.
\newblock {\em Theor. Comput. Sci.}, 289(1):457--471, 2002.

\bibitem{Val02a}
Leslie~G. Valiant.
\newblock Quantum circuits that can be simulated classically in polynomial
  time.
\newblock {\em SIAM J. Comput.}, 31(4):1229--1254, 2002.

\bibitem{Val06}
Leslie~G. Valiant.
\newblock Accidental algorthims.
\newblock In {\em FOCS}, pages 509--517. IEEE Computer Society, 2006.

\bibitem{Val08}
Leslie~G. Valiant.
\newblock Holographic algorithms.
\newblock {\em SIAM J. Comput.}, 37(5):1565--1594, 2008.

\bibitem{Val10}
Leslie~G. Valiant.
\newblock Some observations on holographic algorithms.
\newblock In {\em LATIN}, pages 577--590. Springer Berlin Heidelberg, 2010.

\bibitem{Ver05}
Dirk Vertigan.
\newblock The computational complexity of {T}utte invariants for planar graphs.
\newblock {\em SIAM Journal on Computing}, 35(3):690--712, 2005.

\bibitem{Ver91}
Dirk~Llewellyn Vertigan.
\newblock {\em On the computational complexity of Tutte, Jones, Homfly and
  Kauffman invariants}.
\newblock PhD thesis, University of Oxford, 1991.

\bibitem{Wel93}
Dominic Welsh.
\newblock {\em Complexity: Knots, Colourings and Countings}.
\newblock London Mathematical Society Lecture Note Series. Cambridge University
  Press, 1993.

\bibitem{Yan52}
C.~N. Yang.
\newblock The spontaneous magnetization of a two-dimensional {I}sing model.
\newblock {\em Phys. Rev.}, 85(5):808--816, 1952.

\bibitem{YL52}
C.~N. Yang and T.~D. Lee.
\newblock Statistical theory of equations of state and phase transitions. {I}.
  {T}heory of condensation.
\newblock {\em Phys. Rev.}, 87(3):404--409, 1952.

\end{thebibliography}

\newpage

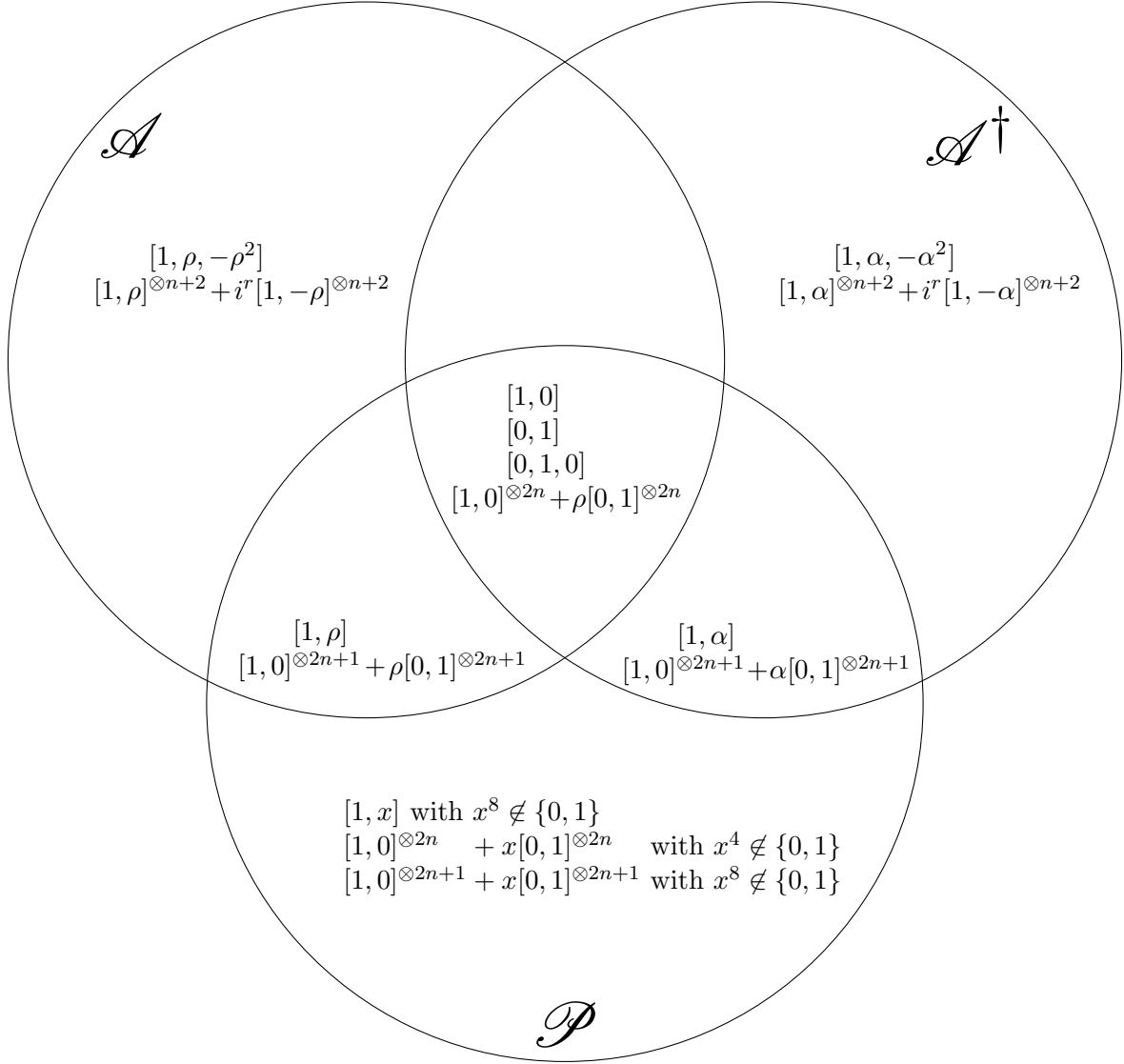
\begin{figure}[p]
 \centering
 \def\radius{5cm}
 \def\nodeDistVD{3.2cm}
 \begin{tikzpicture}
  \draw (150:\nodeDistVD) circle (\radius) node {};
  \draw  (30:\nodeDistVD) circle (\radius) node {};
  \draw (-90:\nodeDistVD) circle (\radius) node {};
  
  \draw node at (142:2.4 * \nodeDistVD) {\Huge $\mathscr{A}$};
  \draw node at  (40:2.3 * \nodeDistVD) {\Huge $\mathscr{A}^{\dagger}$};
  \draw node at (-90:2.4 * \nodeDistVD) {\Huge $\mathscr{P}$};
  
  \draw node[text width=4.2cm] at (135:0.2 * \nodeDistVD)
   {\begin{itemize}[label=]
     \item \qquad $[1,0]$
     \item \qquad $[0,1]$
     \item \qquad $[0,1,0]$
     \item $[1,0]^{\otimes 2n} + \rho [0,1]^{\otimes 2n}$
    \end{itemize}};
  \draw node[text width=5cm] at (314:1.05 * \nodeDistVD)
   {\begin{itemize}[label=]
     \item \qquad $[1,\alpha]$
     \item $[1,0]^{\otimes 2n+1} + \alpha [0,1]^{\otimes 2n+1}$
    \end{itemize}};
  \draw node[text width=5cm] at (218:1.2 * \nodeDistVD)
   {\begin{itemize}[label=]
     \item \qquad $[1,\rho]$
     \item $[1,0]^{\otimes 2n+1} + \rho [0,1]^{\otimes 2n+1}$
    \end{itemize}};
  \draw node[text width=5.1cm] at (150:1.8 * \nodeDistVD)
   {\begin{itemize}[label=]
     \item \qquad $[1, \rho, -\rho^2]$
     \item $[1,\rho]^{\otimes n+2} + i^r [1,-\rho]^{\otimes n+2}$
    \end{itemize}};
  \draw node[text width=5.2cm] at (32:1.7 * \nodeDistVD)
   {\begin{itemize}[label=]
     \item \qquad $[1, \alpha, -\alpha^2]$
     \item $[1,\alpha]^{\otimes n+2} + i^r [1,-\alpha]^{\otimes n+2}$
    \end{itemize}};
  \draw node[text width=8.1cm] at (-90:1.6 * \nodeDistVD)
   {\begin{itemize}[label=]
     \item $[1,x]$ with $x^8 \not\in \{0,1\}$
     \item $[1,0]^{\otimes 2n\phantom{+1}} + x [0,1]^{\otimes 2n\phantom{+1}}$ with $x^4 \not\in \{0,1\}$
     \item $[1,0]^{\otimes 2n+1}           + x [0,1]^{\otimes 2n+1}$           with $x^8 \not\in \{0,1\}$
    \end{itemize}};
 \end{tikzpicture}
 \caption{A Venn diagram of the \#CSP$^2$ tractable sets $\mathscr{A}$, $\mathscr{A}^{\dagger}$, and $\mathscr{P}$.
 Note that $\rho^4 = 1$, $\alpha^4 = -1$, and $n \ge 1$.
 Excluded are tensor products of unary signatures.
 }
 \label{fig:venn_diagram:A_Adagger_P}
\end{figure}

\newpage

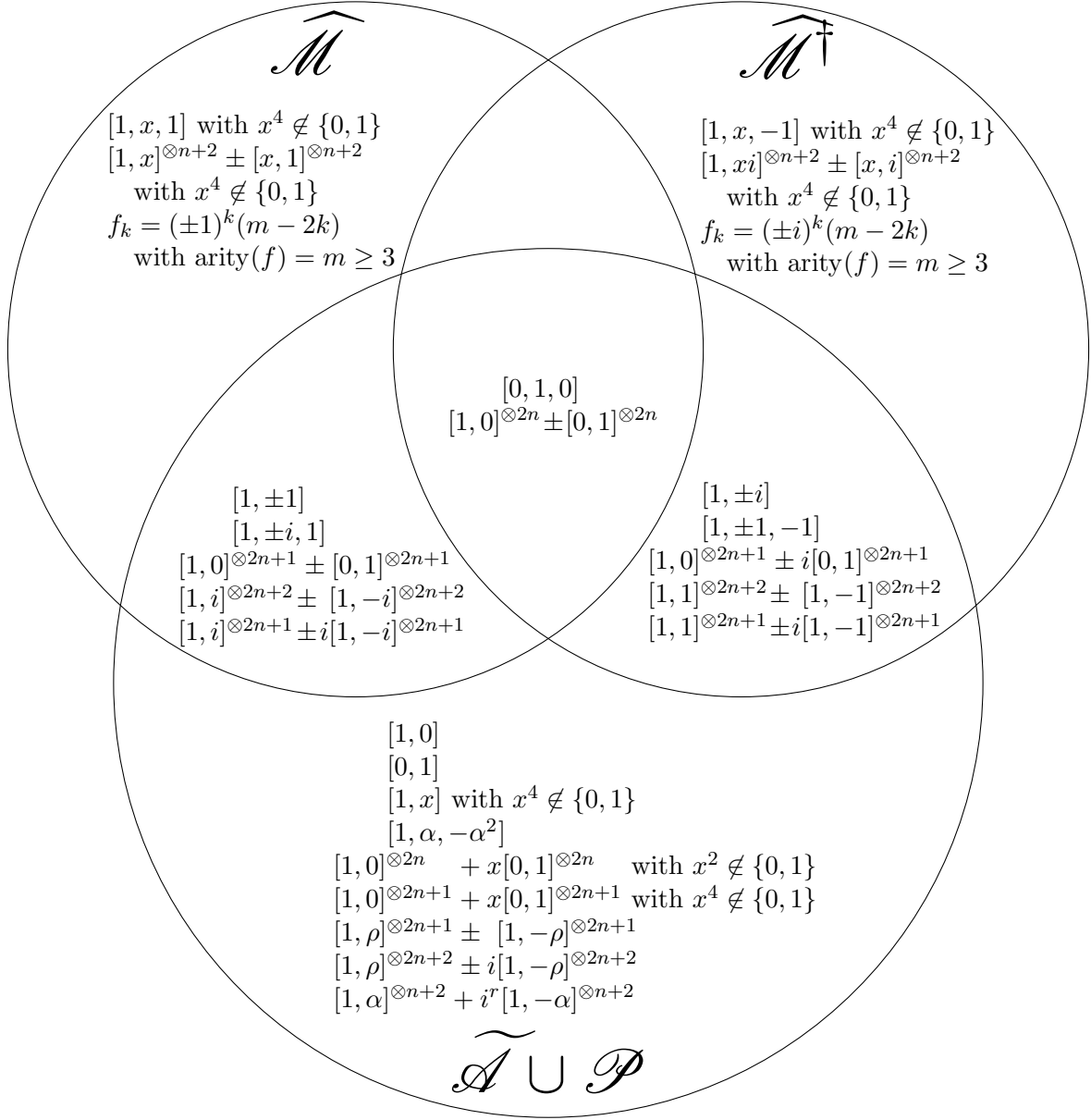
\begin{figure}[p]
 \centering
 \def\radius{5cm}
 \def\nodeDistVD{3.2cm}
 \begin{tikzpicture}
  \draw (150:\nodeDistVD) circle        (\radius) node {};
  \draw  (30:\nodeDistVD) circle        (\radius) node {};
  \draw (-90:\nodeDistVD) circle (1.25 * \radius) node {};
  
  \draw node at (120:2.17 * \nodeDistVD) {\Huge $\widehat{\mathscr{M}}$};
  \draw node at  (60:2.15 * \nodeDistVD) {\Huge $\widehat{\mathscr{M}}^{\dagger}$};
  \draw node at (-90:2.7  * \nodeDistVD) {\Huge $\widetilde{\mathscr{A}} \cup \mathscr{P}$};
  
  \draw node[text width=4.0cm] at (115:0.3 * \nodeDistVD)
   {\begin{itemize}[label=]
     \item \qquad $[0,1,0]$
     \item $[1,0]^{\otimes 2n} \pm [0,1]^{\otimes 2n}$
    \end{itemize}};
  \draw node[text width=5.2cm] at (336:1.05 * \nodeDistVD)
   {\begin{itemize}[label=]
     \item \qquad $[1, \pm i]$
     \item \qquad $[1,\pm 1,-1]$
     \item $[1,0]^{\otimes 2n+1} \pm          i  [0, 1]^{\otimes 2n+1}$
     \item $[1,1]^{\otimes 2n+2} \pm \phantom{i} [1,-1]^{\otimes 2n+2}$
     \item $[1,1]^{\otimes 2n+1} \pm          i  [1,-1]^{\otimes 2n+1}$
    \end{itemize}};
  \draw node[text width=5.1cm] at (201:1.25 * \nodeDistVD)
   {\begin{itemize}[label=]
     \item \qquad $[1, \pm 1]$
     \item \qquad $[1,\pm i,1]$
     \item $[1,0]^{\otimes 2n+1} \pm             [0, 1]^{\otimes 2n+1}$
     \item $[1,i]^{\otimes 2n+2} \pm \phantom{i} [1,-i]^{\otimes 2n+2}$
     \item $[1,i]^{\otimes 2n+1} \pm          i  [1,-i]^{\otimes 2n+1}$
    \end{itemize}};
  \draw node[text width=5.2cm] at (140:1.92 * \nodeDistVD)
   {\begin{itemize}[label=]
     \item $[1,x,1]$ with $x^4 \not\in \{0,1\}$
     \item $[1,x]^{\otimes n+2} \pm [x,1]^{\otimes n+2}$\\\quad with $x^4 \not\in \{0,1\}$
     \item $f_k = (\pm 1)^k (m - 2 k)$\\\quad with $\operatorname{arity}(f) = m \ge 3$
    \end{itemize}};
  \draw node[text width=5.3cm] at (45:1.71 * \nodeDistVD)
   {\begin{itemize}[label=]
     \item $[1,x,-1]$ with $x^4 \not\in \{0,1\}$
     \item $[1,x i]^{\otimes n+2} \pm [x,i]^{\otimes n+2}$\\\quad with $x^4 \not\in \{0,1\}$
     \item $f_k = (\pm i)^k (m - 2 k)$\\\quad with $\operatorname{arity}(f) = m \ge 3$
    \end{itemize}};
  \draw node[text width=8.1cm] at (-90:1.8 * \nodeDistVD)
   {\begin{itemize}[label=]
     \item \qquad $[1,0]$
     \item \qquad $[0,1]$
     \item \qquad $[1,x]$ with $x^4 \not\in \{0,1\}$
     \item \qquad $[1, \alpha, -\alpha^2]$
     \item $[1,0]^{\otimes 2n\phantom{+1}} + x [0,1]^{\otimes 2n\phantom{+1}}$ with $x^2 \not\in \{0,1\}$
     \item $[1,0]^{\otimes 2n+1}           + x [0,1]^{\otimes 2n+1}$           with $x^4 \not\in \{0,1\}$
     \item $[1,\rho]^{\otimes 2n+1} \pm \phantom{i} [1,-\rho]^{\otimes 2n+1}$
     \item $[1,\rho]^{\otimes 2n+2} \pm          i  [1,-\rho]^{\otimes 2n+2}$
     \item $[1,\alpha]^{\otimes n+2} + i^r [1,-\alpha]^{\otimes n+2}$
    \end{itemize}};
 \end{tikzpicture}
 \caption{A Venn diagram of the Pl-\#CSP$^2$ tractable sets $\widehat{\mathscr{M}}$ and $\widehat{\mathscr{M}}^{\dagger}$
 along with the set $\widetilde{\mathscr{A}} \cup \mathscr{P}$ of all tractable \#CSP$^2$ signatures.
 Note that $\rho^4 = 1$, $\alpha^4 = -1$, and $n \ge 1$.
 Excluded are tensor products of unary signatures.
 }
 \label{fig:venn_diagram:M_Mdagger_Atilde_P}
\end{figure}

\end{document}